\documentclass{thesis_template}
\author{Sumeet Khatri}
\date{\today}

\usepackage[square,numbers,sort&compress]{natbib}

\usepackage{anyfontsize}  % Can help to remove certain font size warnings

\newcommand{\I}{\text{i}}

\numberwithin{equation}{chapter}

\begin{document}

\hypersetup{pageanchor=false}
%\maketitle % make title page
%\thispagestyle{plain}

\begin{titlepage}

	\begin{center}
		{\fontsize{16}{19}\selectfont\bfseries TOWARDS A GENERAL FRAMEWORK FOR PRACTICAL QUANTUM NETWORK PROTOCOLS}\vspace*{\fill}
		
		{A Dissertation\\[0.5cm]Submitted to the Graduate Faculty of the\\Louisiana State University and\\Agricultural and Mechanical College\\in partial fulfillment of the\\requirements for the degree of\\Doctor of Philosophy\\[0.5cm]in\\[0.5cm] The Department of Physics and Astronomy}\vspace*{\fill}
		
		{by\\Sumeet Khatri\\B.Sc., University of Waterloo, 2014\\M.Sc., University of Waterloo, 2016\\May 2021}
		
	\end{center}
	
\end{titlepage}

\hypersetup{pageanchor=true}

\begin{frontmatter} % automatically formats page numbers using roman numerals
\addtocounter{page}{1}
\pagestyle{plain}

\chapter*{ACKNOWLEDGMENTS}
%{\noindent\large\textbf{Acknowledgments}}
\addcontentsline{toc}{chapter}{ACKNOWLEDGMENTS}

		I start by thanking Jonathan Dowling, under whose guidance and encouragement a large portion of the research that this thesis is based on was conducted. He got me interested in quantum networks, and he also provided me with several opportunities to travel to conferences. By working with Jon, I learned that the serious endeavor of research can be fun and productive at the same time. 

	I also thank Mark Wilde for his patience, support, and guidance throughout my time here.
Mark showed great confidence in me and gave me the opportunity to discover and pursue my
interests, and also provided me with many opportunities to travel to conferences and to collaborate
with researchers outside of LSU.
	
	The Quantum Science and Technologies group at LSU has been a great place to work for the past four years, and I thank all members of the group, past and present, for making it so. I have learned a lot from the friends I have made here, and they have all enriched my life in innumerable ways. I am especially thankful to the following individuals for their friendship, valuable discussions about work and life, and fun times: Sushovit Adhikari, Anthony Brady, Siddhartha Das, Kevin Valson Jacob, Vishal Katariya, Eneet Kaur, Kunal Sharma, Elisha Siddiqui, and Chenglong You. I also had the greatest pleasure of working with Aliza Siddiqui, Ren\'{e}e Desporte, Manon Bart, and Corey Matyas. I also thank the other graduate students in the Physics and Astronomy Department at LSU for their friendship, particularly Anshuman Bhardwaj, Sahil Saini, Karunya Shirali, Siddharth Soni, and Eklavya Thareja.
	
	During the summer of 2018, I had the opportunity to attend the Quantum Computing Summer School at Los Alamos National Laboratory (LANL). I thank the organizers Patrick Coles, Lukasz Cincio, and Carelton Coffrin, as well as all of the students at the summer school for a memorable experience. I especially thank Ryan LaRose and Alexander Poremba for a great research collaboration. I also thank Patrick for having me back at LANL for two research visits in 2019.

	In November 2019, I went to Brazil to deliver lectures on quantum information theory. I thank Mark for encouraging me to pursue this opportunity. I also thank B\'{a}rbara Amaral and Rafael Chaves for hosting me at the International Institute of Physics in Natal, and B\'{a}rbara especially for her hospitality and for showing me around Natal outside working hours. I acknowledge the American Physical Society for providing funding for the trip through the Brazil--U.S. Student \& Postdoc Visitation Program.
	
	I also thank all of the researchers outside of LSU that I had the pleasure of collaborating with during my PhD: Gerardo Adesso, David Ding, Marco Cerezo, Lukasz Cincio, Patrick Coles, Saikat Guha, Ludovico Lami, Ryan LaRose, Alexander Poremba, Yihui Quek, Peter Shor, George Siopsis, Andrew Sornborger, and Xin Wang.
	
	For financial support during my PhD, I acknowledge the National Science Foundation, Grant No. 1714215. I also acknowledge funding over the past three years from the Natural Sciences and Engineering Research Council of Canada Postgraduate Scholarship. 
	
	%Although learning never stops, the PhD marks the end of my formal academic training. I would thus like to take this opportunity to thank the many teachers I have had during my formative elementary and high school years, who nurtured my curiosity and fed my appetite for material beyond the curriculum, many times by staying with me after regular school hours. They have all inspired me to carry their selflessness and generosity forward and to show the same level of devotion to my own students. I especially thank Angie Somanlall, Navin Kanagasooriam, Daniel Nini, ...Gary Hophan, Anne Dale, Naomi Wittlin, and my high school physics and math teacher Hamid Sharbaf Ebrahimi, who ultimately inspired me to pursue a degree in physics.
	
	Finally, I thank my family for their constant love and support, and for allowing me to pursue my interests. None of this would be possible without them.

\newpage

%%%%%%%%%%%%%%%%%%%%%%%%
%Dedication

%	\vspace*{\fill}
%	\begin{center}
%		{\itshape Dedication.}
%	\end{center}
%	\vspace*{\fill}

%\newpage

%%%%%%%%%%%%%%%%%%%%%%%%

\tableofcontents % make TOC
	\cleardoublepage\phantomsection % for proper hyperlinking
	
\newpage

%\listoffigures % make LOF
	%\addcontentsline{toc}{chapter}{List of Figures} % add to TOC
%	\cleardoublepage\phantomsection % for proper hyperlinking
%\newpage

%\listoftables % make LOT
	%\addcontentsline{toc}{chapter}{List of Tables} % add to TOC
%	\cleardoublepage\phantomsection % for proper hyperlinking
%\newpage

%\renewcommand{\listtheoremname}{List of Definitions}
%\listoftheorems[numwidth=1cm,ignoreall,show={defn}]
	%\addcontentsline{toc}{chapter}{List of Theorems}
%	\cleardoublepage\phantomsection
%\newpage

\chapter*{ABSTRACT}
%{\noindent\large\textbf{Abstract}}
\addcontentsline{toc}{chapter}{ABSTRACT}

	The quantum internet is one of the frontiers of quantum information science. It will revolutionize the way we communicate and do other tasks, and it will allow for tasks that are not possible using the current, classical internet. The backbone of a quantum internet is entanglement distributed globally in order to allow for such novel applications to be performed over long distances. Experimental progress is currently being made to realize quantum networks on a small scale, but much theoretical work is still needed in order to understand how best to distribute entanglement and to guide the realization of large-scale quantum networks, and eventually the quantum internet, especially with the limitations of near-term quantum technologies. This work provides an initial step towards this goal. The main contribution of this thesis is a mathematical framework for entanglement distribution protocols in a quantum network, which allows for discovering optimal protocols using reinforcement learning. We start with a general development of quantum decision processes, which is the theoretical backdrop of reinforcement learning. Then, we define the general task of entanglement distribution in a quantum network, and we present ground- and satellite-based quantum network architectures that incorporate practical aspects of entanglement distribution. We combine the theory of decision processes and the practical quantum network architectures into an overall entanglement distribution protocol. We also define practical figures of merit to evaluate entanglement distribution protocols, which help to guide experimental implementations.

\newpage

\end{frontmatter}

\cleardoublepage % for proper hyperlinking

\begin{mainmatter}

\chapter{INTRODUCTION}\label{chap-introduction}

	Quantum mechanics is the study of physical phenomena that occur at the atomic and molecular levels. It is well known by now that superposition, entanglement, and tunneling are phenomena that arise at such microscopic levels and are not predicted by the theory of classical mechanics. The existence of these phenomena have profound consequences for science, technology, and humanity as a whole, one of which is the possibility of a \textit{quantum internet}, which is the subject of this thesis.
	
	The quantum internet \cite{Kim08,Sim17,Cast18,Wehn+18,Dowling_book2} is envisioned to be a global-scale interconnected network of devices that exploit these uniquely quantum-mechanical phenomena, particularly superposition and entanglement. By operating in tandem with today's internet, it will allow people all over the world to perform \textit{quantum communication} tasks.
	
	``Quantum communication'', broadly speaking, refers to the task of sending information that is encoded into the states of a quantum-mechanical system. The study of quantum communication arises from the merging of the fields of information theory and quantum mechanics, often called quantum information theory, or quantum information science more broadly. The main idea is to think of quantum systems as carriers of information. Then, just as the ``bit'' is the basic unit of information in classical information theory, the ``qubit'' is the basic unit of information in quantum information theory. From an information-theoretic perspective, a qubit is \textit{any} two-level quantum system, and the details of how the qubit is realized in practice are typically not important. Examples of physical realizations of qubits include the ground state and first excited state of an atom, or a single photon in one of two modes of light.
	
	One of the most remarkable practical applications in quantum communication, and one of the primary use cases of the quantum internet in the near term, is quantum key distribution (QKD) \cite{BB84,Eke91,GRG+02,SBPC+09,XXQ+20,PAB+19}. QKD is a quantum communication protocol that uses qubits to enable secure (classical) communication. QKD offers, in principle, unconditional (information-theoretic) security, even against adversaries with a quantum computer. This level of security is important because most current global communication systems rely on encryption protocols that are secure only under computational hardness assumptions, and the discovery of Shor \cite{Shor94,Shor97,MVZJ18} tells us that a quantum computer is capable of breaking them. %It is arguably this discovery that gave practical importance to QKD and to quantum information science in general.
	
	We are now at the point where we can begin to realize a global-scale quantum internet. Indeed, with several metropolitan-scale QKD systems already in place \cite{PPA+09,CWL+10,MP10,SLB+11,SFI+11,WCY+14,BLL+18,ZXCPP18}, and with the development of quantum computers proceeding at a steady pace \cite{LSB+19,BCMS19,arute2019quantum}, the time is right to begin transitioning to quantum communication systems before quantum computers render current communication systems defenseless \cite{Mos15,GM18,MP19}. In addition to QKD, a quantum internet will allow for the execution of other quantum-information-processing tasks, such as quantum teleportation \cite{BBC+93,Vaidman94,BFK00}, quantum clock synchronization \cite{JADW00,Preskill00,UD02,DB+18}, distributed quantum computation \cite{CEHM99}, distributed quantum metrology and sensing \cite{DRC17,ZZS18,XZCZ19}. A quantum internet will also allow for exploring fundamental physics \cite{Bruschi+14}, and for forming an international time standard \cite{KKBJ+14}.
	
	Building the quantum internet is a major challenge. All of the aforementioned tasks make use of shared entanglement between distant locations on the earth, which is typically distributed using single-photonic qubits sent through either the atmosphere or optical fibers. These schemes require reliable single-photon sources, quantum memories with high coherence times, and quantum gate operations with low error. From the theoretical side, it is well known that optical signals transmitted through either the atmosphere or optical fibers undergo an exponential decrease in the transmission success probability with distance \cite{SveltoBook,KJK_book,KGMS88_book}, limiting direct transmission distances to roughly hundreds of kilometers. Therefore, one of the central research questions in the theory of quantum networks is how to overcome this exponential loss and thus to distribute entanglement over long distances efficiently and at high rates. 
	
	The initial theoretical proposal \cite{BDC98,DBC99,SSR+11} for long-distance entanglement distribution involves placing devices called \textit{quantum repeaters} at intermediate points between the end points of a communication line, whose task is to mitigate the effects of loss and noise along the communication line, thereby making the quantum information transmission more reliable. A vast body of literature now exists on a variety of theoretical quantum repeater schemes \cite{BDC98,DBC99,DLCZ01,CJKK07,SRM+07,SSM+07,BPv11,ZDB12,KKL15,ZBD16,EKB16,WZM+16,LZH+17,VK17,MMG19,ZPD+18,PWD18,DKD18,WPZD19,PD19,GI19b,GI19c,HBE20,GKL+03,RHG05,JTN+09,FWH+10,MSD+12,MKL+14,NJKL16,MLK+16,MEL17}. (See also Refs.~\cite{SSR+11,MATN15,VanMeter_book} and the references therein.). All of these proposals deal almost exclusively with a single transmission line connecting a sender and a receiver. However, for a quantum internet, we need to go beyond a single transmission line, and we need to consider multiple transmission lines operating in parallel. A proper theoretical framework needs to be established in order to guide real-world implementations. The first step in this direction is to consider small-scale quantum networks.
	
	\begin{figure}
		\centering
		\includegraphics[width=0.8\textwidth]{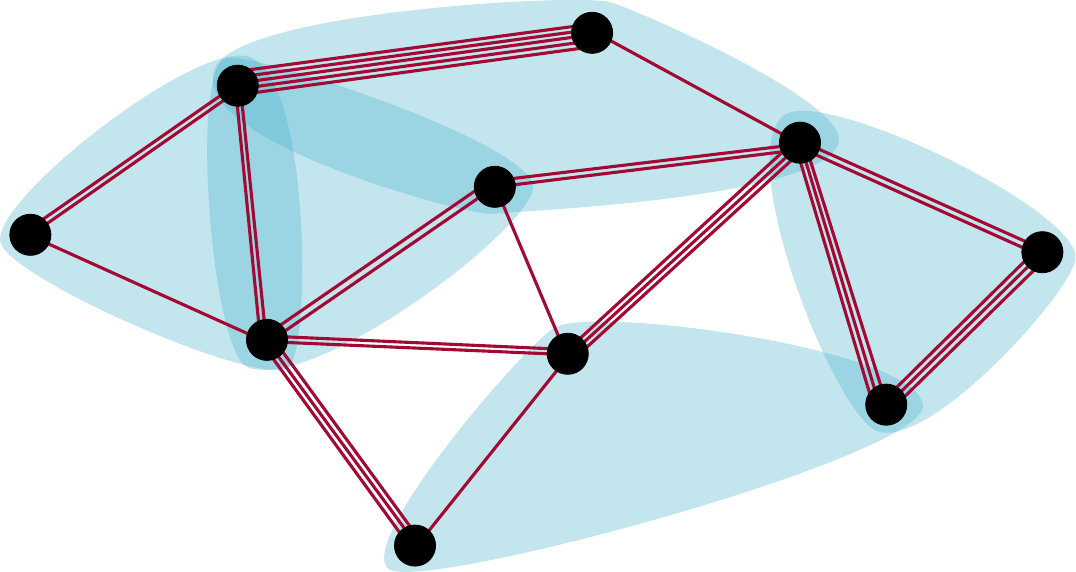}
		\caption{Representation of a quantum network as a hypergraph. The vertices represent the nodes (senders and receivers) of the network, and the edges represent entangled states shared by the corresponding nodes. (We provide a brief overview of graph theory in Section~\ref{sec-graph_theory}.) Edges between two nodes (shown in red) represent bipartite entanglement, while hyperedges (consisting of more than two nodes and indicated by a blue bubble) represent multipartite entanglement. Nodes can be connected by multiple edges, indicating that they can share multiple entangled states simultaneously.}\label{fig-network_model}
	\end{figure}
	
	As shown in Figure~\ref{fig-network_model}, a quantum network can be modeled as a graph. The vertices of the graph represent the senders and receivers in the network, and the edges represent \textit{elementary links}, which are entangled states shared by the corresponding nodes. The edges can be between two nodes only, as indicated by the red lines, or they can be hyperedges connecting three or more nodes, as indicated by the blue bubbles. Groups of nodes can be connected by more than one edge, and in this case the graph is called a multigraph. Multiple edges between nodes are shown explicitly in Figure~\ref{fig-network_model} for two-node edges, although we can also have multiple hyperedges between a set of adjacent nodes. Each of these edges is regarded as a distinct edge in the graph. The goal of entanglement distribution in a quantum network is to use the elementary links to form \textit{virtual links}, i.e., entangled states shared by nodes that are not physically connected to each other by elementary links.
	
	Considering a quantum network such as the one in Figure~\ref{fig-network_model}, as opposed to just one line between a sender and a receiver, is a much more complicated setting that leads to questions about, e.g., routing \cite{GI17,SMI+17,GI18,PKT+19,HPE19,CRDW19,Pir19,Pir19b,LLLC20,LBD+20} and multicast communication (simultaneous communication between several senders and receivers). Consequently, protocols in a general quantum network can be much more varied than protocols along a linear chain of nodes. As done in classical networking, it is possible to develop a so-called ``quantum network stack'', which divides the various steps of a quantum network protocol into distinct layers of functionality \cite{AME11,PWD18,PD19,DSM+19}. Along these lines, quantum network protocols have been described from an information-theoretic perspective in Refs.~\cite{AML16,AK17,BA17,RKB+18,Pir19,Pir19b,DBWH19}, and limits on communication in quantum networks have been explored in Refs.~\cite{BCHW15,AML16,STW16,TSW17,LP17,AK17,BA17,CM17,RKB+18,BAKE20,Pir19,Pir19b,DBWH19}. Linear programs, and other techniques for obtaining optimal entanglement distribution rates in a quantum network, have been explored in \cite{BAKE20,DPW20,CERW20,GEW20}.
	
	In order to physically realize quantum networks, and the quantum internet more generally, the continual challenge to is to bridge theoretical statements about what can be achieved to statements that are directly useful for the purpose of implementation. This link between theory and reality should also take into account the limitations of current and near-term quantum technologies, which include imperfect sources of entanglement, quantum memories with relatively low coherence times, and imperfect measurements and gate operations. Many of the aforementioned theoretical works do not explicitly take these practical limitations into account. What is currently lacking is a formal theoretical framework for quantum network protocols that incorporates both the limitations of near-term quantum technologies and is general enough to allow for optimization of protocol parameters. The purpose of this thesis is to begin such a development.
	
	The core conceptual contribution of this thesis is to propose viewing entanglement distribution protocols in quantum networks from the lens of decision processes \cite{Put14_book}. In a decision process, an agent interacts with its environment through a sequence of actions, and it receives rewards from the environment based on these actions. The goal of the agent is to devise a policy that maximizes its expected total reward. In this thesis, we consider a particular quantum generalization of a decision process given in Ref.~\cite{BBA14} (see also Ref.~\cite{YY18}), called a \textit{quantum partially observable Markov decision process}. In such a decision process, the agent's action at each time step results in a transformation of the quantum state of the environment, and the agent receives both partial (classical) information about the new quantum state of the environment along with a reward.
	
	\begin{figure}
		\centering
		\includegraphics[width=0.95\textwidth]{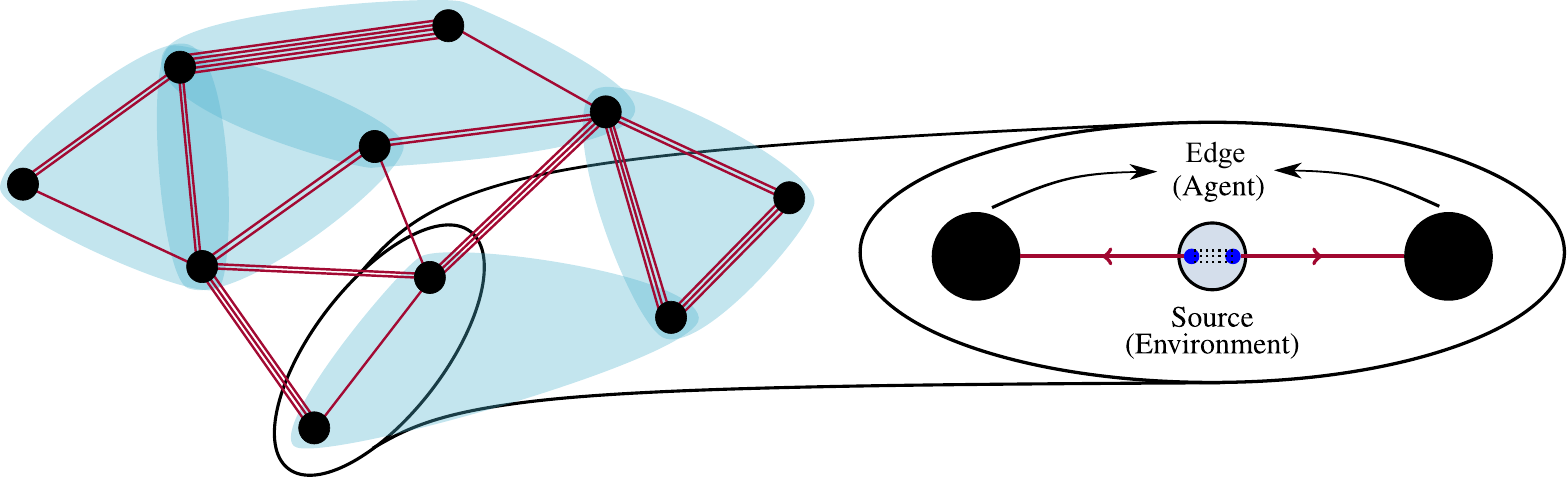}
		\caption{Given a quantum network represented as a graph, in this thesis we develop a quantum network protocol based on decision processes in which we associate to every edge an agent. At every time step, the agent decides to either request an entangled state from a source station, or to keep the entangled state it currently has in its quantum memory. The collection of policies for all of the agents then forms an element of the overall quantum network protocol. The protocol is outlined in detail in Chapter~\ref{chap-network_QDP}.}\label{fig-network_agent_env_intro}
	\end{figure}
	
	Using the language and formalism of quantum decision processes, in this thesis we propose a general framework for entanglement distribution protocols in quantum networks. The basic idea is illustrated in Figure~\ref{fig-network_agent_env_intro}. Given a quantum network represented as a graph, we associate to every edge an agent. At every time step, the agent decides to either request an entangled state from a source station, or to keep the entangled state it currently has in its quantum memory. A quantum network protocol then corresponds to a collection of policies for the agents.
	
	One advantage of the approach taken in this thesis is based on the fact that, in the classical setting, decision processes form the theoretical foundation for reinforcement learning \cite{Sut18_book} and artificial intelligence \cite{RN09_book}. Therefore, due to this connection between decision processes and reinforcement learning in the classical case, the main message of this thesis is that reinforcement learning can be used to discover optimal quantum network protocols, and the developments of this thesis provide the tools needed to do so. (See Ref.~\cite{WMDB19} for related work on machine learning for quantum communication.) Another advantage of our approach is that, even though reinforcement learning techniques cannot always be applied efficiently to large-scale problems, decision processes provide us with a systematic framework for combining optimal small-scale protocols in order to create large-scale protocols; see Ref.~\cite{JTKL07} for similar ideas. Our framework also allows for a systematic consideration of agents with local and global knowledge of the network, as well as agents that are independent and/or cooperate with each other. We do not go through all of these aspects in detail in this thesis; however, we provide the starting point for the future exploration of some of these ideas in Appendix~\ref{sec-future_work}.
	
	The remainder of this thesis is organized as follows.
	\begin{itemize}
		\item \textit{Chapter~\ref{chap-prelim}.} We provide a brief introduction to quantum mechanics from an information-theoretic perspective. In particular, the concept of LOCC (``local operations and classical communication'') quantum channels plays a crucial role in the developments of this thesis, and we outline several examples of such channels as they arise in quantum communication and quantum networks. We also provide brief introductions to quantum key distribution and graph theory.
		 
		\item \textit{Chapter~\ref{chap-QDP}.} We develop the theory of quantum decision processes in which both the agent and the environment are described by quantum systems. This development extends prior work \cite{BBA14,YY18} on quantum partially observable Markov decision processes. Due to the wide applicability of quantum decision processes in quantum information processing tasks (as explained in Section~\ref{sec-QDP_q_info_tasks}), the material of this chapter is expected to be of independent interest outside of the quantum network context.
		
		\item \textit{Chapter~\ref{chap-network_setup}.} We present a mathematical model for describing quantum networks, as well as quantum states and channels in a quantum network. We also present a theoretical framework for entanglement distribution in a quantum network based on graph transformations, as in Refs.~\cite{SMI+17,CRDW19}, and we present practical ground- and satellite-based network architectures that take into account the limitations of near-term quantum technologies. 
		
			The content of this chapter is based on the work of Refs.~\cite{DKD18,KMSD19,KBD+19}.
		
		\item \textit{Chapter~\ref{chap-network_QDP}.} Using the formalism of quantum decision processes developed in Chapter~\ref{chap-QDP}, and based on the task of entanglement distribution as outlined in Chapter~\ref{chap-network_setup}, in this chapter we provide an explicit quantum network protocol using quantum decision processes to model elementary link generation under practical settings. We also define practical figures of merit for evaluating policies.
		
			The content of this chapter is based on the work of Ref.~\cite{Kha20,KMSD19}.
		
		\item \textit{Chapter~\ref{chap-mem_cutoff}.} We consider an explicit example of a policy for elementary link generation, called the memory-cutoff policy. This is a natural policy to consider for near-term quantum networks, and it has been considered extensively in prior work, although not explicitly in the language of quantum decision processes. For this policy, we provide expressions for the key elements of the corresponding quantum decision process, and we also evaluate the figures of merit defined in Chapter~\ref{chap-network_QDP}. 
		
			The content of this chapter is based on the work of Refs.~\cite{KMSD19,Kha20}.
		
		\item \textit{Chapter~\ref{chap-sats}.} In this chapter, we bring together all of the elements considered in this thesis and provide a complete example of elementary link generation using satellite-to-ground photon transmission. We first provide estimates for secret key rates for quantum key distribution as a function of satellite altitude and ground distance separation. Then, we consider the memory-cutoff policy for elementary link generation. Finally, we perform policy optimization using the algorithms described in Chapter~\ref{chap-QDP}.
		
			The content of this chapter is based on the work of Ref.~\cite{KBD+19}.
	\end{itemize}

\chapter{PRELIMINARIES}\label{chap-prelim}

	We start with a brief exposition of some of the background material needed for this thesis. We start in Section~\ref{sec-states_channels} with a summary of the basics of quantum mechanics from an information-theoretic perspective, with definitions of quantum states, measurements, and channels. The concept of local operations and classical communication (LOCC) plays an important role in quantum communication and quantum networks, so we devote Section~\ref{sec-LOCC_channels} to a detailed explanation of LOCC channels, with three specific examples that are relevant for communication in a quantum network. Another important element of quantum networks, particularly from a practical perspective, is quantum key distribution (QKD), which is a protocol for secure quantum communication (see, e.g., Refs.~\cite{GRG+02,SBPC+09,XXQ+20,PAB+19} for reviews). In Section~\ref{sec-QKD}, we provide a brief overview of the basics of both device-dependent and device-independent QKD protocols. Finally, because quantum networks can be visualized using (mathematical) graphs, in Section~\ref{sec-graph_theory}, we provide a summary of the basic definitions and concepts of graph theory as they pertain to the contents of this thesis.

\section{Quantum states, measurements, and channels}\label{sec-states_channels}
	
	The formal, mathematical treatment of the theory of quantum mechanics was initiated by von Neumann \cite{Neu32_book,Neu18_book_new} and Dirac \cite{Dirac_book}, who presented an axiomatic approach to quantum theory. An advantage of such an axiomatic approach, especially for theoretical work, is that the specific details pertaining to the physical systems involved are not important. This approach has proved to be very fruitful in quantum information and quantum computing, and it is now the standard approach to teaching these subjects; see Refs.~\cite{NC00_book,Hol12_book,Wilde17_book,Wat18_book,KW20_book}, which we also refer to for detailed treatments of all of the topics presented in this section.

	The foundation on which the mathematics of quantum mechanics is built is the Hilbert space. A Hilbert space is an inner product space (vector space equipped with an inner product) that is complete (see, e.g., Ref.~\cite{Reed_Simon_book} for details). Throughout this thesis, we are concerned only with finite-dimensional complex Hilbert spaces, which means that every $d$-dimensional Hilbert space, $d\geq 1$, can be thought of as the vector space $\mathbb{C}^d$ equipped with the Euclidean inner product\footnote{All finite-dimensional inner product spaces are trivially complete, and thus are Hilbert spaces.}. The theory of linear algebra and matrix analysis (see, e.g., Refs.~\cite{Axler_book,Horn_Johnson_book}) are therefore sufficient for our purposes.
	
	We denote elements of $\mathbb{C}^d$ by $\ket{\psi}$, and every such $\ket{\psi}\in\mathbb{C}^d$ can be written as
	\begin{equation}
		\ket{\psi}=\sum_{j=0}^{d-1}\alpha_j\ket{j},
	\end{equation}
	where $\alpha_j\in\mathbb{C}$ and $\{\ket{j}\}_{j=0}^{d-1}$ is the (standard) orthonormal basis for $\mathbb{C}^d$, whose elements can be represented as columns vectors as follows:
	\begin{equation}
		\ket{0}=\begin{pmatrix} 1\\ 0\\0\\ \vdots\\ 0\end{pmatrix},\quad\ket{1}=\begin{pmatrix}0\\1\\0\\\vdots\\0\end{pmatrix},\quad\dotsb,\quad\ket{d-1}=\begin{pmatrix}0\\0\\0\\\vdots\\1\end{pmatrix}.
	\end{equation}
	The Euclidean inner product of $\ket{\psi}$ and another vector $\mathbb{C}^d\ni\ket{\phi}=\sum_{j=0}^{d-1}\beta_j\ket{j}$ is
	\begin{equation}
		\braket{\phi}{\psi}=\sum_{j=0}^{d-1}\conj{\alpha_j}\beta_j,
	\end{equation}
	where $\conj{\alpha_j}$ denotes the complex conjugate of $\alpha_j$.

\subsubsection*{Quantum states}

	To every quantum system, call it $A$, we associate a Hilbert space, which we denote by $\mathcal{H}_A$, and we denote its (finite) dimension by $d_A$. A \textit{qudit} is a quantum system with $d\geq 2$ levels and associated Hilbert space $\mathbb{C}^d$. When $d=2$, we use the term \textit{qubit}, and when $d=3$, we use the term \textit{qutrit}.

	Given a quantum system $A$, the wavefunctions that describe the system are unit-norm vectors (often called \textit{state vectors}) in the Hilbert space, and we denote them by $\ket{\psi}_A$. Note that the system label $A$ is placed in the subscript of the vector, which is helpful when dealing with composite quantum systems. Every state vector has the form
	\begin{equation}\label{eq-state_vector}
		\ket{\psi}_A=\sum_{k=0}^{d-1}\alpha_k\ket{k},\quad \alpha_k\in\mathbb{C},\quad\sum_{k=0}^{d-1}\abs{\alpha_k}^2=1.
	\end{equation}
	
	More generally, the state of a quantum system is given by a \textit{density operator}, which is a linear operator that is positive semi-definite and has unit trace.
	
	\begin{definition}[Quantum state]
		Given a quantum system $A$, the \textit{quantum state} of $A$ is given by a density operator acting on $\mathcal{H}_A$: a linear operator $\rho_A:\mathcal{H}_A\to\mathcal{H}_A$ that is positive semi-definite and has unit trace, i.e., $\rho_A\geq 0$ and $\Tr[\rho_A]=1$. We let $\Density(\mathcal{H}_A)$ denote the set of all quantum states of the system $A$.~\defqed
	\end{definition}
	
	Unit-rank density operators of the form $\rho_A=\ket{\psi}\bra{\psi}_A$, where $\ket{\psi}_A$ is a state vector as in \eqref{eq-state_vector}, are called \textit{pure states}. We often use the abbreviation $\psi_A=\ket{\psi}\bra{\psi}_A$. In general, every density operator $\rho_A$ can be decomposed as a convex combination of pure states, i.e.,
	\begin{equation}\label{eq-state_decomp_pure}
		\rho_A=\sum_{x\in\mathcal{X}} p(x)\ket{\psi^x}\bra{\psi^x}_A,
	\end{equation}
	where $\mathcal{X}$ is some finite set, $\{\ket{\psi^x}\}_{x\in\mathcal{X}}$ is a set of state vectors, and $p:\mathcal{X}\to[0,1]$ is a probability mass function, meaning that $0\leq p(x)\leq 1$ for all $x\in\mathcal{X}$ and $\sum_{x\in\mathcal{X}}p(x)=1$. The decomposition in \eqref{eq-state_decomp_pure} can be thought of as arising from a probabilistic preparation of the quantum system $A$, such that the system $A$ is prepared in the pure state $\psi_A^x$ with probability $p(x)$, for some $x\in\mathcal{X}$, but the person in possession of the quantum system does not know $x$. The state of the quantum system is therefore the average (expectation) with respect to all $x\in\mathcal{X}$.

	Bipartite quantum systems are quantum systems with two constituent subsystems. Multipartite quantum systems are quantum systems with two or more constituent subsystems. The Hilbert space of a multipartite quantum system is the tensor product of the Hilbert spaces of the constituent subsystems. Specifically, if $A_1,A_2,\dotsc,A_k$ are $k\geq 2$ distinct quantum systems, then the composite $k$-partite quantum system containing all $k$ systems is denoted by $A_1A_2\dotsb A_k$, and its associated Hilbert space is
	\begin{equation}
		\mathcal{H}_{A_1A_2\dotsb A_k}=\mathcal{H}_{A_1}\otimes\mathcal{H}_{A_2}\otimes\dotsb\otimes\mathcal{H}_{A_k}.
	\end{equation}
	If $\{\ket{j_1}_{A_1}\}_{j_1=0}^{d_{A_1}-1}$, $\{\ket{j_2}_{A_2}\}_{j_2=0}^{d_{A_2}-1}$, \ldots $\{\ket{j_k}_{A_k}\}_{j_k=0}^{d_{A_k}-1}$ are orthonormal bases for $\mathcal{H}_{A_1},\mathcal{H}_{A_2},\dotsc,\mathcal{H}_{A_k}$, then
	\begin{equation}
		\left\{\ket{j_1}_{A_1}\otimes\ket{j_2}_{A_2}\otimes\dotsb\otimes\ket{j_k}_{A_k}:0\leq j_1\leq d_{A_1}-1,\,0\leq j_2\leq d_{A_2}-1,\dotsc,0\leq j_k\leq d_{A_k}-1\right\}
	\end{equation}
	is an orthonormal basis for $\mathcal{H}_{A_1A_2\dotsb A_k}$. For brevity, we often write
	\begin{equation}
		\ket{j_1,j_2,\dotsc,j_k}_{A_1A_2\dotsb A_k}\equiv\ket{j_1}_{A_1}\otimes\ket{j_2}_{A_2}\otimes\dotsb\otimes\ket{j_k}_{A_k}.
	\end{equation}
	
	\begin{definition}[Fidelity~\cite{Uhl76}]
		Given two quantum states $\rho$ and $\sigma$, the \textit{fidelity} between them is defined to be
		\begin{equation}
			F(\rho,\sigma)\coloneqq\left(\Tr\left[\sqrt{\sqrt{\rho}\sigma\sqrt{\rho}}\right]\right)^2.~\defqedspec
		\end{equation}
	\end{definition}
	
	The fidelity quantifies the closeness of two quantum states. In particular, $F(\rho,\sigma)=1$ if and only if $\rho=\sigma$, and $F(\rho,\sigma)=0$ if and only if $\rho$ and $\sigma$ are supported on orthogonal subspaces. If one of the states, say $\sigma$, is pure, then it is straightforward to show that
	\begin{equation}
		F(\rho,\ket{\psi}\bra{\psi})=\bra{\psi}\rho\ket{\psi}.
	\end{equation}
	The fidelity is also \textit{multiplicative}, meaning that
	\begin{equation}\label{eq-fidelity_multiplicative}
		F(\rho_1\otimes\rho_2,\sigma_1\otimes\sigma_2)=F(\rho_1,\sigma_1)F(\rho_2,\sigma_2),
	\end{equation}
	for all states $\rho_1,\rho_2,\sigma_1,\sigma_2$.
	
	\begin{definition}[Separable and entangled states]
		Let $A$ and $B$ be two quantum systems. We say that a state $\rho_{AB}\in\Density(\mathcal{H}_{AB})$ is \textit{separable} if it can be written in the form
		\begin{equation}\label{eq-separable_state}
			\rho_{AB}=\sum_{x\in\mathcal{X}}p(x) \sigma_A^x\otimes\tau_B^x,
		\end{equation}
		where $\mathcal{X}$ is some finite set, $\{\sigma_A^x\}_{x\in\mathcal{X}}$ and $\{\tau_B^x\}_{x\in\mathcal{X}}$ are sets of quantum states, and $p:\mathcal{X}\to[0,1]$ is a probability mass function. If $\rho_{AB}$ does not have a decomposition as in \eqref{eq-separable_state}, then $\rho_{AB}$ is \textit{entangled}.~\defqed
	\end{definition}
	
	An important example of a set of entangled states is the four two-qubit \textit{Bell states} $\Phi_{AB}^{\pm}=\ket{\Phi^{\pm}}\bra{\Phi^{\pm}}_{AB}$ and $\Psi_{AB}^{\pm}=\ket{\Psi^{\pm}}\bra{\Psi^{\pm}}_{AB}$, where
	\begin{align}
		\ket{\Phi^{\pm}}_{AB}&\coloneqq\frac{1}{\sqrt{2}}\left(\ket{0,0}_{AB}\pm\ket{1,1}_{AB}\right),\\
		\ket{\Psi^{\pm}}_{AB}&\coloneqq\frac{1}{\sqrt{2}}\left(\ket{0,1}_{AB}\pm\ket{1,0}_{AB}\right).
	\end{align}
	
	In the case of two-qubit states (as well as qubit-qutrit states), it is well known that a quantum state $\rho_{AB}$ is separable if and only if its \textit{partial transpose} $\rho_{AB}^{\t_B}$ is positive semi-definite \cite{Peres96,HHH96}, where the partial transpose is defined as follows:
	\begin{equation}
		\rho_{AB}=\sum_{i,j,k,\ell=0}^1\alpha_{\substack{i,j\\k,\ell}}\ket{i,j}\bra{k,\ell}_{AB}\quad\Longrightarrow\quad\rho_{AB}^{\t_B}=\sum_{i,j,k,\ell=0}^1\alpha_{\substack{i,j\\k,\ell}}\ket{i,\ell}\bra{k,j}_{AB}.
	\end{equation}
	In other words, $\rho_{AB}$ is entangled if and only if the partial transpose $\rho_{AB}^{\t_B}$ contains a negative eigenvalue\footnote{The same statement holds if one instead considers the partial transpose $\rho_{AB}^{\t_A}$ with respect to the system $A$.}. We emphasize that this is known to be true only for qubit-qubit and qubit-qutrit states. In general, if a quantum state $\rho_{AB}$ is separable, then its partial transpose is positive semi-definite, but the converse statement does not necessarily hold.
	 
	The \textit{log-negativity} \cite{ZHS+98,VW02,Plen05} is an entanglement measure that can be used to quantify the extent to which the partial transpose of a quantum state $\rho_{AB}$ has negative eigenvalues. It is defined as
	\begin{equation}\label{eq-log_negativity}
		E_N(\rho_{AB})=\log_2\left(\norm{\rho_{AB}^{\t_B}}_1\right),
	\end{equation}
	where $\norm{X}_1\coloneqq\Tr[\sqrt{X^\dagger X}]$ is the \textit{trace norm} of a linear operator $X$, which is equivalent to the sum of the singular values of $X$. As a result of the previous paragraph, the log-negativity is a faithful entanglement measure for all qubit-qubit and qubit-qutrit states: if $A$ and $B$ are both qubit systems, or if one is a qubit and the other a qutrit, then a state $\rho_{AB}$ is entangled if and only if $E_N(\rho_{AB})>0$.
		
	A simple, yet important set of quantum states that we consider in Chapter~\ref{chap-sats} of this thesis is the set of two-qubit Bell-diagonal states:
	\begin{equation}\label{eq-Bell_diag_state}
		\rho_{AB}=p_1\Phi_{AB}^++p_2\Phi_{AB}^-+p_3\Psi_{AB}^++p_4\Psi_{AB}^-,
	\end{equation}
	where $0\leq p_1,p_2,p_3,p_4\leq 1$ and $p_1+p_2+p_3+p_4=1$. In this case, it is straightforward to show that the log-negativity is
	\begin{equation}\label{eq-log_neg_Bell_diag}
		E_N(\rho_{AB})=\log_2\left(\frac{1}{2}\abs{1-2p_1}+\frac{1}{2}\abs{1-2p_2}+\frac{1}{2}\abs{1-2p_3}+\frac{1}{2}\abs{1-2(p_1+p_2+p_3)}\right).
	\end{equation}
	In the case that
	\begin{equation}\label{eq-Bell_diag_spec}
		p_1=\alpha+\beta,\quad p_2=\alpha-\beta,\quad p_3=p_4=\gamma,
	\end{equation}
	the log-negativity simplifies to
	\begin{equation}\label{eq-log_neg_Bell_diag_spec}
		E_N(\rho_{AB})=\log_2\left(\abs{\frac{1}{2}-(\alpha+\beta)}+\frac{1}{2}+\alpha+\beta\right).
	\end{equation}
	This means that, if the coefficients of the Bell-diagonal state in \eqref{eq-Bell_diag_state} are constrained as in \eqref{eq-Bell_diag_spec}, then $\rho_{AB}$ is entangled if and only if $\alpha+\beta>\frac{1}{2}$. Observing that $\alpha+\beta=p_1=\bra{\Phi^+}\rho_{AB}\ket{\Phi^+}$, we see that the entanglement of $\rho_{AB}$ is given simply by its fidelity to the state $\Phi^+$; i.e., $\rho_{AB}$ is entangled if and only if $F(\rho_{AB},\Phi_{AB}^+)>\frac{1}{2}$.
	
	A type of quantum state that we frequently encounter in this thesis is a classical-quantum state.
	
	\begin{definition}[Classical-quantum state]
		A quantum state $\rho_{XB}$ is called a \textit{classical-quantum state} if it is of the form
		\begin{equation}\label{eq-classical_quantum_state}
			\rho_{XA}=\sum_{x\in\mathcal{X}}p(x)\ket{x}\bra{x}_X\otimes\rho_A^x,
		\end{equation}
		where $\mathcal{X}$ is a finite set, $p:\mathcal{X}\to[0,1]$ is a probability mass function, and $\{\rho_A^x\}_{x\in\mathcal{X}}$ is a set of quantum states.~\defqed
	\end{definition}
	
	Classical-quantum states can be used to model scenarios in which classical information accompanies the state of a quantum system. Specifically, if a quantum system $A$ is prepared in a state from the set $\{\rho_A^x\}_{x\in\mathcal{X}}$ according to the probability distribution given by $p:\mathcal{X}\to[0,1]$, then knowledge of the label $x\in\mathcal{X}$ is stored in the classical register $X$.

\subsubsection*{Quantum measurements}

	A measurement of a quantum system is a procedure that is used to extract classical information from the system. Mathematically, measurements in quantum mechanics are defined as follows.
	
	\begin{definition}[Quantum measurement]\label{def-quantum_measurement}
		Let $A$ be a quantum system with associated Hilbert space $\mathcal{H}_A$. A \textit{measurement} of $A$ is defined by a finite set $\{M_A^x\}_{x\in\mathcal{X}}$ of linear operators, called a \textit{positive operator-valued measure (POVM)}, that satisfies the following two properties.
		\begin{itemize}
			\item $M_A^x \geq 0$ for all $x\in\mathcal{X}$;
			\item $\displaystyle \sum_{x\in\mathcal{X}} M_A^x=\mathbbm{1}_A$.
		\end{itemize}
		Elements of the set $\mathcal{X}$ label the possible outcomes of the measurement. Given a state $\rho_A$, the probability of obtaining the outcome $x\in\mathcal{X}$ is given by the Born rule as $\Tr[M_A^x\rho_A]$.~\defqed
	\end{definition}
	
	A special case of a measurement as defined above is the measurement of an observable (i.e., a Hermitian operator). Let $R$ be a Hermitian operator. From linear algebra, it is known that $R$ can be diagonalized, meaning that it has a \textit{spectral decomposition}:
	\begin{equation}
		R=\sum_{x\in\mathcal{X}}\lambda_x\Pi^x,
	\end{equation}
	where $\mathcal{X}$ is some finite set, $\{\lambda_x\}_{x\in\mathcal{X}}$ are the (distinct) eigenvalues of $R$, and $\Pi^x$ are the corresponding spectral projections, which satisfy $\sum_{x\in\mathcal{X}}\Pi^x=\mathbbm{1}$ and $\Pi^x\Pi^{x'}=\delta_{x,x'}\Pi^x$ for all $x,x'\in\mathcal{X}$. The measurement of $R$ is then described mathematically by the POVM $\{\Pi^x\}_{x\in\mathcal{X}}$.

\subsubsection*{Quantum channels}

	A quantum channel is a mathematical description of the evolution of a quantum system. Let $\Lin(\mathcal{H}_A)$ denote the vector space of linear operators acting on the Hilbert space $\mathcal{H}_A$. A linear map $\mathcal{T}:\Lin(\mathcal{H}_A)\to\Lin(\mathcal{H}_B)$ is often called a \textit{superoperator}, and it is such that
	\begin{equation}
		\mathcal{T}(\alpha X+\beta Y)=\alpha\mathcal{T}(X)+\beta\mathcal{T}(Y)
	\end{equation}
	for all $\alpha,\beta\in\mathbb{C}$ and all $X,Y\in\Lin(\mathcal{H}_A)$. It is often helpful to explicitly indicate the input and output Hilbert spaces of a superoperator $\mathcal{T}:\Lin(\mathcal{H}_A)\to\Lin(\mathcal{H}_B)$ by writing $\mathcal{T}_{A\to B}$. The identity superoperator is denoted by $\id_A$, and it satisfies $\id_A(X)=X$ for all $X\in\Lin(\mathcal{H}_A)$. A quantum channel is a particular kind of superoperator.
	
	\begin{definition}[Quantum channel]
		A quantum channel $\mathcal{N}_{A\to B}$ is a linear, \textit{completely positive}, and \textit{trace-preserving} superoperator acting on the vector space $\Lin(\mathcal{H}_A)$ of linear operators of the Hilbert space $\mathcal{H}_A$ of the quantum system $A$. Given an input state $\rho_A$ of the system $A$, the output is the state of a new quantum system $B$ given by $\mathcal{N}_{A\to B}(\rho_A)$.
		\begin{itemize}
			\item A superoperator $\mathcal{N}$ is \textit{completely positive} if the map $\id_{k}\otimes\mathcal{N}$ is positive for all $k\in\mathbb{N}$, where $\id_k:\Lin(\mathbb{C}^k)\to\Lin(\mathbb{C}^k)$ is the identity superoperator. In other words, $(\id_k\otimes\mathcal{N})(X)\geq 0$ for all linear operators $X\geq 0$.
			\item A superoperator $\mathcal{N}$ is \textit{trace preserving} if $\Tr[\mathcal{N}(X)]=\Tr[X]$ for all linear operators $X$.~\defqed
		\end{itemize}
	\end{definition}
	\medskip
	\begin{theorem}[Choi~\cite{Choi75}, Kraus~\cite{Kraus_book}, and Stinespring~\cite{Stinespring55}]
		Given two Hilbert spaces $\mathcal{H}_A$ and $\mathcal{H}_B$, a superoperator $\mathcal{N}:\Lin(\mathcal{H}_A)\to\Lin(\mathcal{H}_B)$ is a quantum channel if and only if:
		\begin{itemize}
			\item (Choi) The operator
				\begin{equation}
					\Gamma_{AB}^{\mathcal{N}}\coloneqq (\id_A\otimes\mathcal{N}_{A'\to B})(\Gamma_{AA'})
				\end{equation}
				is positive semi-definite and $\Tr_{B}[\Gamma_{AB}^{\mathcal{N}}]=\mathbbm{1}_A$, where
				\begin{equation}
					\Gamma_{AA'}\coloneqq\sum_{i,i'=0}^{d_A-1}\ket{i,i}\bra{i',i'}_{AA'},
				\end{equation}
				and $A'$ is a system with the same dimension as $A$.
				
			\item (Kraus) There exists a set $\{K_i:\Lin(\mathcal{H}_A)\to\Lin(\mathcal{H}_B)\}_{i=1}^r$ such that
				\begin{equation}
					\mathcal{N}_{A\to B}(X_A)=\sum_{i=1}^r K_iX_AK_i^\dagger
				\end{equation}
				for all linear operators $X_A$, where $r\geq\text{rank}(\Gamma_{AB}^{\mathcal{N}})$ and $\sum_{i=1}^r K_i^\dagger K_i=\mathbbm{1}_A$.
				
			\item (Stinespring) There exists a linear operator $V:\mathcal{H}_A\to\mathcal{H}_B\otimes\mathcal{H}_E$ satisfying $V^\dagger V=\mathbbm{1}_A$, with $d_E\geq\text{rank}(\Gamma_{AB}^{\mathcal{N}})$, such that
				\begin{equation}
					\mathcal{N}_{A\to B}(X_A)=\Tr_E[VX_AV^\dagger]
				\end{equation}
				for all linear operators $X_A$.
		\end{itemize}
	\end{theorem}
	\medskip
	\begin{remark}
		Our definition of a quantum channel is in the spirit of the axiomatic treatment of quantum mechanics. Another approach is to start with the Schr\"{o}dinger equation for (closed) quantum system evolution:
		\begin{equation}\label{eq-Schrodinger_eq}
			\I\hbar\frac{\text{d}}{\text{d}t}\ket{\psi(t)}=H\ket{\psi(t)},
		\end{equation}
		where $H$ is the Hamiltonian of the system and $\ket{\psi(t)}$ is the state vector of the system at time $t\in\mathbb{R}$. The corresponding equation for density operators is often known as the \textit{von Neumann equation}:
		\begin{equation}\label{eq-von_Neumann_eq}
			\I\hbar\frac{\text{d}\rho(t)}{\text{d}t}=[H,\rho(t)],
		\end{equation}
		where $[H,\rho(t)]=H\rho(t)-\rho(t) H$ is the commutator of $H$ and $\rho(t)$; see, e.g., Ref.~\cite{Sakurai_book}. With \eqref{eq-von_Neumann_eq} as the starting point, one can obtain both the Kraus and Stinespring forms of a quantum channel; see Refs.~\cite{Nakajima58,Zwanzig60,GKS76,Lin76}, as well as Refs.~\cite{Kok15,Lidar19,MM20,BP_book} for pedagogical reviews.~\defqed
	\end{remark}
	
	A simple example of a quantum channel is the \textit{partial trace}, which corresponds to the physical operation of discarding a quantum system. Given a linear operator $X_{AB}$, the partial trace over $B$ is a quantum channel denoted by $\Tr_B$, and it has the following Kraus representation:
	\begin{equation}
		\Tr_B[X_{AB}]\coloneqq\sum_{j=0}^{d_B-1} (\mathbbm{1}_A\otimes\bra{j}_B)X_{AB}(\mathbbm{1}_A\otimes\ket{j}_B).
	\end{equation}
	Similarly, the partial trace over $A$ is a quantum channel denoted by $\Tr_A$, and it has the following Kraus representation:
	\begin{equation}
		\Tr_A[X_{AB}]\coloneqq\sum_{j=0}^{d_A-1}(\bra{j}_A\otimes\mathbbm{1}_B)X_{AB}(\ket{j}_A\otimes\mathbbm{1}_B).
	\end{equation}
	We often write $X_A\equiv\Tr_B[X_{AB}]$ and $X_B\equiv\Tr_A[X_{AB}]$ to denote the operators at the output of the partial trace channels over $B$ and $A$, respectively.
	
	Another example of a quantum channel, which arises very frequently in this thesis, is the quantum instrument channel.
	
	\begin{definition}[Quantum instrument]\label{def-quantum_instrument}
		A \textit{quantum instrument} is a finite set $\{\mathcal{M}^x\}_{x\in\mathcal{X}}$ of completely positive trace non-increasing maps\footnote{A trace non-increasing map $\mathcal{N}_{A\to B}$ satisfies $\Tr[\mathcal{N}_{A\to B}(X_A)]\leq\Tr[X_A]$ for all positive semi-definite linear operators $X_A$.} such that the sum $\sum_{x\in\mathcal{X}}\mathcal{M}^x$ is a trace-preserving map, and thus a quantum channel. The \textit{quantum instrument channel} $\mathcal{M}$ associated to the quantum instrument $\{\mathcal{M}^x\}_{x\in\mathcal{X}}$ is defined as
		\begin{equation}\label{eq-quantum_instrument_channel}
			\mathcal{M}(\cdot)\coloneqq\sum_{x\in\mathcal{X}} \ket{x}\bra{x}\otimes\mathcal{M}^x(\cdot).~\defqedspec
		\end{equation}
	\end{definition}
	
	A quantum instrument $\{\mathcal{M}^x\}_{x\in\mathcal{X}}$ can be thought of as a generalized form of a measurement, in which the completely positive maps $\mathcal{M}^x$ represent the evolution of the quantum system conditioned on the outcome $x$. The trace non-increasing property of the maps $\mathcal{M}^x$ represents the fact that the outcome $x$ occurs probabilistically. Specifically, the probability of obtaining the outcome $x$ is equal to $\Tr[\mathcal{M}^x(\rho)]$, which can be thought of as a generalized form of the Born rule given in Definition~\ref{def-quantum_measurement}. The quantum instrument channel in \eqref{eq-quantum_instrument_channel} can be thought of as an operation that stores both the outcome $x$ of the instrument in the classical register as well as the corresponding output state.

\subsection{LOCC channels}\label{sec-LOCC_channels}

	Consider two parties, Alice and Bob, who are spatially separated. Suppose that they have the ability to perform arbitrary quantum operations (quantum channels, measurements, instruments) in their respective labs and that they are connected by a classical communication channel. It is often the case that Alice and Bob are also connected by a quantum channel and/or share an entangled quantum state, and their task is to make use of these \textit{resources} as sparingly as possible in order to accomplish their desired goal. Their local operations and classical communication (abbreviated ``LOCC'') can be used freely to help with achieving the goal\footnote{From a resource-theoretic perspective, LOCC channels are regarded as the free operations of the resource theory of entanglement; see, e.g., Refs.~\cite{HHHH09,CG19} for more information about the resource theory of entanglement and resources theories more generally.}, which could be feedback-assisted quantum communication \cite{BDSW96} (see also Ref.~\cite{KW20_book}), which includes quantum teleportation and entanglement swapping \cite{BBC+93,Vaidman94,BFK00,ZZH93}, or it could be entanglement distillation \cite{BDSW96}. In the network setting, the task is \textit{repeater-assisted quantum communication}, and we explain this in detail in Chapter~\ref{chap-network_setup}. For now, we focus on the basic mathematical definition of an LOCC channel and provide some examples of LOCC channels. For more mathematical details about LOCC channels, we refer to Ref.~\cite{CLM+14}.
	
	\begin{definition}[LOCC channel]
		An \textit{LOCC channel} $\mathcal{L}_{AB\to\hat{A}\hat{B}}$ is a quantum channel with input quantum systems $A$ and $B$ and output quantum systems $\hat{A}$ and $\hat{B}$ that has the form
		\begin{equation}\label{eq-LOCC_channel}
			\mathcal{L}_{AB\to\hat{A}\hat{B}}(\rho_{AB})\coloneqq\sum_{x\in\mathcal{X}}(\mathcal{S}_{A\to\hat{A}}^x\otimes\mathcal{T}_{B\to\hat{B}}^x)(\rho_{AB})
		\end{equation}
		for all quantum states $\rho_{AB}$. Here, $\{\mathcal{S}^x\}_{x\in\mathcal{X}}$ and $\{\mathcal{T}^x\}_{x\in\mathcal{X}}$ are completely positive trace non-increasing maps such that the sum $\sum_{x\in\mathcal{X}}\mathcal{S}^x\otimes\mathcal{T}^x$ is a trace-preserving map, and thus a quantum channel.~\defqed
	\end{definition}
	
	\begin{remark}
		Quantum channels with the form shown in \eqref{eq-LOCC_channel} are known more generally as \textit{separable channels}. It is important to remark that not all separable channels are LOCC channels. Only those separable channels that can be realized by an LOCC protocol, as described below, are LOCC channels. See Ref.~\cite[Section 3.2.11]{KW20_book} for an example of a separable channel that is not an LOCC channel.~\defqed
	\end{remark}
	
	In order to understand why LOCC channels are defined in this way, consider the scenario shown in Figure~\ref{fig-LOCC}, which is an \textit{LOCC protocol} over $t$ rounds. The corresponding channel from $A_0B_0$ to $A_tB_t$ can then be derived as follows.
	
	\begin{figure}
		\centering
		\includegraphics[width=0.95\textwidth]{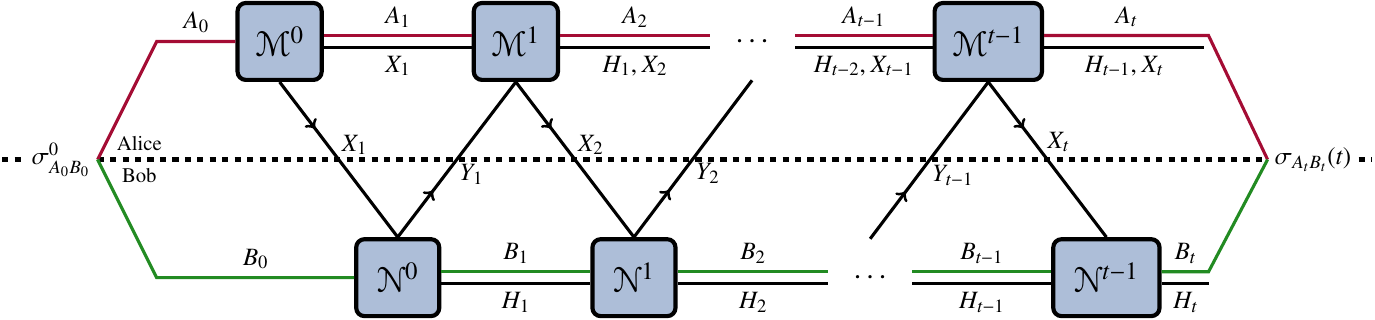}
		\caption{Depiction of a $t$-round LOCC protocol between Alice and Bob. The channels $\mathcal{M}^0,\mathcal{M}^1,\dotsc,\mathcal{M}^{t-1}$ represent Alice's local operations, the channels $\mathcal{N}^0,\mathcal{N}^1,\dotsc,\mathcal{N}^{t-1}$ represent Bob's local operations, and the registers $H_1,H_2,\dotsc,H_t$ represent the classical communication to and from Alice and Bob. The final quantum state $\sigma_{A_tB_t}(t)$ shared by Alice and Bob has the form shown in \eqref{eq-LOCC_protocol_final_state}.}\label{fig-LOCC}
	\end{figure}
	
	We start with the initial state $\sigma_{A_0B_0}^0$ shared by Alice and Bob. Then, in the first round, Alice acts on her system $A_0$ with a quantum instrument channel $\mathcal{M}_{A_0\to X_1A_1}^0$, which leads to the following output:
	\begin{equation}
		\sigma_{A_0B_0}^0\mapsto \mathcal{M}_{A_0\to X_1A_1}(\sigma_{A_0B_0}^0)=\sum_{x_1\in\mathcal{X}_1}\ket{x_1}\bra{x_1}_{X_1}\otimes\mathcal{M}_{A_0\to A_1}^{0;x_1}(\sigma_{A_0B_0}^0),
	\end{equation}
	where $\mathcal{X}_1$ is a finite set and $\{\mathcal{M}_{A_0\to A_1}^{0;x_1}\}_{x_1\in\mathcal{X}_1}$ is a quantum instrument. The classical register $X_1$ is then communicated to Bob, who applies the conditional quantum instrument channel given by
	\begin{equation}
		\mathcal{N}_{X_1B_0\to X_1Y_1B_1}^0(\ket{x_1}\bra{x_1}_{X_1}\otimes \tau_{B_0})=\ket{x_1}\bra{x_1}_{X_1}\otimes\mathcal{N}_{B_0\to Y_1B_1}^{0;x_1}(\tau_{B_0}),
	\end{equation}
	where $\mathcal{N}_{B_0\to Y_1B_1}^{0;x_1}$ is a quantum instrument channel, i.e.,
	\begin{equation}
		\mathcal{N}_{B_0\to Y_1B_1}^{0;x_1}(\tau_{B_0})=\sum_{y_1\in\mathcal{Y}_1}\ket{y_1}\bra{y_1}_{Y_1}\otimes\mathcal{N}_{B_0\to B_1}^{0;x_1,y_1}(\tau_{B_0}),
	\end{equation}
	with $\{\mathcal{N}_{B_0\to B_1}^{0;x_1,y_1}\}_{y_1\in\mathcal{Y}_1}$ being a quantum instrument, i.e., every $\mathcal{N}_{B_0\to B_1}^{0;x_1,y_1}$ is a completely positive trace non-increasing map and $\sum_{y_1\in\mathcal{Y}_1}\mathcal{N}_{B_0\to B_1}^{0;x_1,y_1}$ is a trace-preserving map. Bob sends the outcome $y_1\in\mathcal{Y}_1$ of the quantum instrument to Alice. This completes the first round, and the quantum state shared by Alice and Bob is
	\begin{align}
		\widehat{\sigma}_{X_1Y_1A_1B_1}(1)&\coloneqq\left(\mathcal{M}_{A_0\to X_1A_1}^0\otimes\mathcal{N}_{X_1B_0\to X_1Y_1B_1}^0\right)(\sigma_{A_0B_0}^0)\\
		&=\sum_{\substack{x_1\in\mathcal{X}_1\\y_1\in\mathcal{Y}_1}}\ket{x_1,y_1}\bra{x_1,y_1}_{X_1Y_1}\otimes\left(\mathcal{M}_{A_0\to A_1}^{0;x_1}\otimes\mathcal{N}_{B_0\to B_1}^{0;x_1,y_1}\right)(\sigma_{A_0B_0}^0)\\
		&=\sum_{\substack{x_1\in\mathcal{X}_1\\y_1\in\mathcal{Y}_1}}\ket{x_1,y_1}\bra{x_1,y_1}_{X_1Y_1}\otimes\widetilde{\sigma}_{A_1B_1}(1;x_1,y_1),
	\end{align}
	where in the last line we let
	\begin{equation}
		\widetilde{\sigma}_{A_1B_1}(1;x_1,y_1)\coloneqq\left(\mathcal{M}_{A_0\to A_1}^{0;x_1}\otimes\mathcal{N}_{B_0\to B_1}^{0;x_1,y_1}\right)(\sigma_{A_0B_0}^0).
	\end{equation}
	
	Now, in the second round, Alice applies the conditional quantum instrument channel given by
	\begin{equation}
		\mathcal{M}_{X_1Y_1A_1\to X_1Y_1X_2A_2}^1(\ket{x_1,y_1}\bra{x_1,y_1}_{X_1Y_1}\otimes\rho_{A_1})\coloneqq\ket{x_1,y_1}\bra{x_1,y_1}_{X_1Y_1}\otimes\mathcal{M}_{A_1\to X_2A_2}^{1;x_1,y_1}(\rho_{A_1}),
	\end{equation}
	where $\mathcal{M}_{A_1\to X_2A_2}^{1;x_1,y_1}$ is a quantum instrument channel in which the underlying quantum instrument $\{\mathcal{M}_{A_1\to A_2}^{1;x_1,y_1,x_2}\}_{x_2\in\mathcal{X}_2}$ is conditioned on the histories $\{(x_1,y_1):x_1\in\mathcal{X}_1,\,y_1\in\mathcal{Y}_1\}$ of hers and Bob's prior outcomes. She sends the outcome $x_2\in\mathcal{X}_2$ of the quantum instrument to Bob, who then applies the conditional quantum instrument channel given by
	\begin{multline}
		\mathcal{N}_{X_1Y_1X_2B_1\to X_1Y_1X_2Y_2B_2}^1(\ket{x_1,y_1,x_2}\bra{x_1,y_1,x_2}_{X_1Y_1X_2}\otimes\tau_{B_1})\\\coloneqq\ket{x_1,y_1,x_2}\bra{x_1,y_1,x_2}_{X_1Y_1X_2}\otimes\mathcal{N}_{B_1\to Y_2B_2}^{1;x_1,y_1,x_2}(\tau_{B_1}),
	\end{multline}
	where $\mathcal{N}_{B_1\to Y_2B_2}^{1;x_1,y_1,x_2}$ is a quantum instrument channel in which the underlying quantum instrument $\{\mathcal{N}_{B_1\to B_2}^{1;x_1,y_1,x_2,y_2}\}_{y_2\in\mathcal{Y}_2}$ depends on the prior outcomes $x_1\in\mathcal{X}_1,\,y_1\in\mathcal{Y}_1,\,x_2\in\mathcal{X}_2$. The outcome of the instrument is $y_2\in\mathcal{Y}_2$, so that, at the end of the second round, the state shared by Alice and Bob is
	\begin{align}
		\widehat{\sigma}_{X_1Y_1X_2Y_2A_2B_2}(2)&\coloneqq\sum_{\substack{x_1\in\mathcal{X}_1\\y_1\in\mathcal{Y}_1}}\sum_{\substack{x_2\in\mathcal{X}_2\\y_2\in\mathcal{Y}_2}}\ket{x_1,y_1,x_2,y_2}\bra{x_1,y_1,x_2,y_2}_{X_1Y_1X_2Y_2}\nonumber\\
		&\qquad\qquad\qquad\otimes\left(\mathcal{M}_{A_1\to A_2}^{1;x_1,y_1,x_2}\circ\mathcal{M}_{A_0\to A_1}^{0;x_1}\otimes\mathcal{N}_{B_1\to B_2}^{1;x_1,y_1,x_2,y_2}\circ\mathcal{N}_{B_0\to B_1}^{0;x_1,y_1}\right)(\sigma_{A_0B_0}^0)\\
		&=\sum_{h^2}\ket{h^2}\bra{h^2}_{H_2}\otimes\widetilde{\sigma}_{A_2B_2}(2;h^2),
	\end{align}
	where in the last line we used the abbreviations
	\begin{align}
		H_2&\equiv X_1Y_1X_2Y_2,\\
		h^2&\equiv (x_1,y_1,x_2,y_2),\quad x_1\in\mathcal{X}_1,\,x_2\in\mathcal{X}_2,\,y_1\in\mathcal{Y}_1,\,y_2\in\mathcal{Y}_2.
	\end{align}
	
	Proceeding in the manner presented above, at each step $j\geq 1$ of the protocol, Alice and Bob apply conditional quantum instrument channels of the form
	\begin{align}
		&\mathcal{M}_{H_{j}A_j\to H_jX_{j+1}A_{j+1}}^j\left(\ket{h^j}\bra{h^j}_{H_j}\otimes\rho_{A_j}\right)=\ket{h^j}\bra{h^j}_{H_j}\otimes\mathcal{M}_{A_j\to X_{j+1}A_{j+1}}^{j;h^j}(\rho_{A_j}),\\[0.2cm]
		&\mathcal{N}_{H_jX_{j+1}B_j\to H_{j+1}B_{j+1}}^{j}\left(\ket{h^j,x_{j+1}}\bra{h^j,x_{j+1}}_{H_jX_{j+1}}\otimes\sigma_{B_j}\right)=\ket{h^j,x_{j+1}}\bra{h^j,x_{j+1}}_{H_jX_{j+1}}\otimes\mathcal{N}_{B_j\to Y_{j+1}B_{j+1}}^{j;h^j,x_{j+1}}(\sigma_{B_j}),
	\end{align}
	where
	\begin{align}
		\mathcal{M}_{A_j\to X_{j+1}A_{j+1}}^{j;h^j}(\rho_{A_j})&=\sum_{x_{j+1}\in\mathcal{X}_{j+1}}\ket{x_{j+1}}\bra{x_{j+1}}_{X_{j+1}}\otimes\mathcal{M}_{A_j\to A_{j+1}}^{j;h^j,x_{j+1}}(\rho_{A_j}),\\[0.2cm]
		\mathcal{N}_{B_j\to Y_{j+1}B_{j+1}}^{j;h^j,x_{j+1}}(\sigma_{B_j})&=\sum_{y_{j+1}\in\mathcal{Y}_{j+1}}\ket{y_{j+1}}\bra{y_{j+1}}_{Y_{j+1}}\otimes\mathcal{N}_{B_j\to B_{j+1}}^{j;h^j,x_{j+1},y_{j+1}}(\sigma_{B_j}).
	\end{align}
	Therefore, at the end of the $t^{\text{th}}$ round, the classical-quantum state shared by Alice and Bob is
	\begin{align}
		\widehat{\sigma}_{H_tA_tB_t}(t)&=\left(\mathcal{N}_{H_{t-1}X_tB_{t-1}\to H_tB_t}^{t-1}\circ\mathcal{M}_{H_{t-1}A_{t-1}\to H_{t-1}X_tA_t}^{t-1}\circ\dotsb\right.\nonumber\\
		&\qquad\qquad\left.\circ\mathcal{N}_{H_1X_2B_1\to H_2B_2}^1\circ\mathcal{M}_{H_1A_1\to H_1X_2A_2}^1\circ\mathcal{N}_{X_1B_0\to H_1B_1}^0\circ\mathcal{M}_{A_0\to X_1A_1}^0\right)(\sigma_{A_0B_0}^0)\\
		&=\sum_{h^t}\ket{h^t}\bra{h^t}_{H_t}\otimes\widetilde{\sigma}_{A_tB_t}(t;h^t),\label{eq-cq_state_LOCC}
	\end{align}
	where $h^t=(x_1,y_1,x_2,y_2,\dotsc,x_t,y_t)$ and
	\begin{align}
		\widetilde{\sigma}(t;h^t)&=\left(\left(\mathcal{M}_{A_{t-1}\to A_1}^{t-1;h_{t-1}^t,x_t}\circ\dotsb\circ\mathcal{M}_{A_1\to A_2}^{1;h_1^t,x_2}\circ\mathcal{M}_{A_0\to A_1}^{0;x_1}\right)\right.\nonumber\\&\qquad\qquad\qquad\qquad\left.\otimes\left(\mathcal{N}_{B_{t-1}\to B_t}^{t-1;h^t}\circ\dotsb\circ\mathcal{N}_{B_1\to B_2}^{1;h_2^t}\circ\mathcal{N}_{B_0\to B_1}^{0;h_1^t}\right)\right)(\sigma_{A_0B_0}^0)\label{eq-cq_state_LOCC_2}\\[0.2cm]
		&=\left(\mathcal{S}_{A_0\to A_t}^{t;h^t}\otimes\mathcal{T}_{B_0\to B_t}^{t;h^t}\right)(\sigma_{A_0B_0}^0),\label{eq-cq_state_LOCC_3}
	\end{align}
	and we have defined
	\begin{equation}
		h_j^t\coloneqq (x_1,y_1,x_2,y_2,\dotsc,x_j,y_j),\quad 1\leq j\leq t-1.
	\end{equation}
	Also, in \eqref{eq-cq_state_LOCC_3}, we have defined 
	\begin{align}
		\mathcal{S}_{A_0\to A_t}^{t;h^t}&\coloneqq\mathcal{M}_{A_{t-1}\to A_t}^{t-1;h_{t-1}^t,x_t}\circ\dotsb\circ\mathcal{M}_{A_1\to A_2}^{1;h_1^t,x_2}\circ\mathcal{M}_{A_0\to A_1}^{0;x_1},\\[0.2cm]
		\mathcal{T}_{B_0\to B_t}^{t;h^t}&\coloneqq\mathcal{N}_{B_{t-1}\to B_t}^{t-1;h^t}\circ\dotsb\circ\mathcal{N}_{B_1\to B_2}^{1;h_2^t}\circ\mathcal{N}_{B_0\to B_1}^{0;h_1^t},
	\end{align}
	
	Now, if Alice and Bob forget the history of their outcomes, then this corresponds to discarding the classical history register $H_t$, and it results in the state
	\begin{equation}\label{eq-LOCC_protocol_final_state}
		\sigma_{A_tB_t}(t)\coloneqq\Tr_{H_t}[\widehat{\sigma}_{H_tA_tB_t}(t)]=\sum_{h^t}\widetilde{\sigma}_{A_tB_t}(t;h^t)=\sum_{h^t}(\mathcal{S}_{A_0\to A_t}^{t;h^t}\otimes\mathcal{T}_{B_0\to B_t}^{t;h^t})(\sigma_{A_0B_0}^0),
	\end{equation}
	which is precisely of the form in \eqref{eq-LOCC_channel}. In particular, observe that the sum map $\sum_{h^t}\mathcal{S}^{t;h^t}\otimes\mathcal{T}^{t;h^t}$ is indeed trace preserving, because for all $1\leq j\leq t-1$, all histories $h^{j-1}$ up to time $j-1$, and all linear operators $X_{A_{j-1}B_{j-1}}$,
	\begin{align}
		&\sum_{x_j,y_j}\Tr\left[\left(\mathcal{M}_{A_{j-1}\to A_j}^{j-1;h^{j-1},x_j}\otimes\mathcal{N}_{B_{j-1}\to B_j}^{h^{j-1},x_j,y_j}\right)(X_{A_{j-1}B_{j-1}})\right]\nonumber\\
		&\qquad=\sum_{x_j}\Tr\left[\left(\mathcal{M}_{A_{j-1}\to A_j}^{j-1;h^{j-1},x_j}\otimes\sum_{y_j}\mathcal{N}_{B_{j-1}\to B_j}^{h^{j-1},x_j,y_j}\right)(X_{A_{j-1}B_{j-1}})\right]\\
		&\qquad=\sum_{x_j}\Tr_{A_j}\left[\mathcal{M}_{A_{j-1}\to A_j}^{j-1;h^{j-1},x_j}(\Tr_{B_{j-1}}[X_{A_{j-1}B_{j-1}}])\right]\\
		&\qquad=\Tr_{A_j}\left[\sum_{x_j}\mathcal{M}_{A_{j-1}\to A_j}^{j-1;h^{j-1},x_j}(\Tr_{B_{j-1}}[X_{A_{j-1}B_{j-1}}])\right]\\
		&\qquad=\Tr[X_{A_{j-1}B_{j-1}}],
	\end{align}
	where we have used the fact that $\sum_{y_j}\mathcal{N}_{B_{j-1}\to B_j}^{h^{j-1},x_j,y_j}$ and $\sum_{x_j}\mathcal{M}_{A_{j-1}\to A_j}^{j-1;h^{j-1},x_j}$ are trace-preserving maps. Therefore, for all linear operators $X_{A_0B_0}$,
	\begin{align}
		&\Tr\left[\sum_{h^t}(\mathcal{S}_{A_0\to A_t}^{t;h^t}\otimes\mathcal{T}_{A_0\to A_t}^{t;h^t})(X_{A_0B_0})\right]\nonumber\\
		&\qquad=\sum_{h^t_{t-1}}\sum_{x_t,y_t}\Tr\left[\left(\mathcal{M}_{A_{t-1}\to A_t}^{t-1;h^t_{t-1},x_t}\otimes\mathcal{N}_{B_{t-1}\to B_t}^{t-1;h_{t-1}^t,x_t,y_t}\right)\left(\mathcal{S}^{t-1;h_{t-1}^t}_{A_0\to A_{t-1}}\otimes\mathcal{T}_{B_0\to B_{t-1}}^{t-1;h_{t-1}^t}\right)(X_{A_0B_0})\right]\\
		&\qquad=\sum_{h_{t-1}^{t}}\Tr\left[\left(\mathcal{S}_{A_0\to A_{t-1}}^{t-1;h_{t-1}^t}\otimes\mathcal{T}_{A_0\to A_{t-1}}^{t-1;h_{t-1}^t}\right)(X_{A_0B_0})\right]\\
		&\qquad=\sum_{h_{t-2}^t}\sum_{x_{t-1},y_{t-1}}\Tr\left[\left(\mathcal{M}_{A_{t-2}\to A_{t-1}}^{t-2;h^t_{t-2},x_{t-1}}\otimes\mathcal{N}_{B_{t-2}\to B_{t-1}}^{t-2;h_{t-2}^t,x_{t-1},y_{t-1}}\right)\left(\mathcal{S}^{t-2;h_{t-2}^t}_{A_0\to A_{t-2}}\otimes\mathcal{T}_{B_0\to B_{t-2}}^{t-2;h_{t-2}^t}\right)(X_{A_0B_0})\right]\\
		&\qquad=\sum_{h_{t-2}^t}\Tr\left[\left(\mathcal{S}^{t-2;h_{t-2}^t}_{A_0\to A_{t-2}}\otimes\mathcal{T}_{B_0\to B_{t-2}}^{t-2;h_{t-2}^t}\right)(X_{A_0B_0})\right]\\
		&\quad\quad\,\,\,\vdots\\
		&\qquad=\sum_{x_1,y_1}\Tr\left[\left(\mathcal{M}_{A_0\to A_1}^{0;x_1}\otimes\mathcal{N}_{B_0\to B_1}^{0;x_1,y_1}\right)(X_{A_0B_0})\right]\\
		&\qquad=\Tr[X_{A_0B_0}].
	\end{align}
	So we conclude that the sum map $\sum_{h^t}\mathcal{S}^{t;h^t}\otimes\mathcal{T}^{t;h^t}$ is trace preserving.
	
	\begin{remark}[LOCC instruments]
		From \eqref{eq-cq_state_LOCC} and \eqref{eq-cq_state_LOCC_3}, the classical-quantum state $\widehat{\sigma}_{H_tA_tB_t}(t)$ after $t$ rounds of an LOCC protocol is
		\begin{equation}
			\widehat{\sigma}_{H_tA_tB_t}(t)=\sum_{h^t}\ket{h^t}\bra{h^t}_{H_t}\otimes\left(\mathcal{S}_{A_0\to A_t}^{t;h^t}\otimes\mathcal{T}_{B_0\to B_t}^{t;h^t}\right)(\sigma_{A_0B_0}^0).
		\end{equation}
		Observe that this state has exactly the form of the output state of a quantum instrument channel. In particular, letting
		\begin{equation}
			\mathcal{L}_{A_0B_0\to A_tB_t}^{t;h^t}\coloneqq\mathcal{S}_{A_0\to A_t}^{t;h^t}\otimes\mathcal{T}_{B_0\to B_t}^{t;h^t},
		\end{equation}
		we see that the state $\widehat{\sigma}_{H_tA_tB_t}(t)$ can be regarded as the output state of an \textit{LOCC instrument}, i.e., a finite set $\left\{\mathcal{L}_{AB\to\hat{A}\hat{B}}^{x}\right\}_{x\in\mathcal{X}}$ of completely positive trace non-increasing LOCC maps such that sum $\sum_{x\in\mathcal{X}}\mathcal{L}_{AB\to \hat{A}\hat{B}}^{x}$ is a trace-preserving map, and thus an LOCC quantum channel.~\defqed
	\end{remark}
	
	Let us now consider some examples of LOCC channels.
	
	\begin{example}[Entanglement swapping protocol]\label{ex-ent_swap}
		Let $\rho_{A\vec{R}_1\vec{R}_2\dotsb\vec{R}_nB}$ be a multipartite quantum state, where $n\geq 1$ and $\vec{R}_j\equiv R_j^1R_j^2$ is an abbreviation for two the quantum systems $R_j^1$ and $R_j^2$. The entanglement swapping protocol with $n$ intermediate nodes is defined by a Bell-basis measurement of the systems $\vec{R}_j$, i.e., a measurement described by the POVM $\{\Phi^{z,x}:0\leq z,x\leq d-1\}$, where $\Phi^{z,x}=\ket{\Phi^{z,x}}\bra{\Phi^{z,x}}$ and
		\begin{equation}
			\ket{\Phi^{z,x}}\coloneqq(Z^zX^x\otimes\mathbbm{1})\ket{\Phi^+}
		\end{equation}
		are the qudit Bell states. The operators $Z$ and $X$ are the discrete Weyl operators \cite{Wat18_book}, which are defined as
		\begin{equation}
			Z\coloneqq\sum_{k=0}^{d-1}\e^{\frac{2\pi\I k}{d}}\ket{k}\bra{k},\quad X\coloneqq\sum_{k=0}^{d-1}\ket{k+1}\bra{k}.
		\end{equation}
		Conditioned on the outcomes $(z_j,x_j)$ of the Bell measurement on $\vec{R}_j$, the unitary $Z_B^{z_1+\dotsb+z_n}X_B^{x_1+\dotsb+x_n}$ is applied to the system $B$, where the addition is performed modulo $d$. Let $[d]=\{0,1,\dotsc,d-1\}$, and let $\vec{z},\vec{x}\in[d]^{\times n}$. Also, let
		\begin{align}
			M_{\vec{R}_1\vec{R}_2\dotsb\vec{R}_n}^{\vec{z},\vec{x}}&\coloneqq \Phi_{\vec{R}_1}^{z_1,x_1}\otimes\Phi_{\vec{R}_2}^{z_2,x_2}\otimes\dotsb\otimes\Phi_{\vec{R}_n}^{z_n,x_n},\\[0.2cm]
			W_B^{\vec{z},\vec{x}}&\coloneqq Z_B^{z_1+\dotsb+z_n}X_B^{x_1+\dotsb+x_n},
		\end{align}
		where the addition in the second line is performed modulo $d$. Then, the LOCC quantum channel corresponding to the entanglement swapping protocol with $n\geq 1$ intermediate nodes is
		\begin{equation}\label{eq-ent_swap_channel}
			\mathcal{L}_{A\vec{R}_1\vec{R}_2\dotsb\vec{R}_nB\to AB}^{\text{ES};n}\left(\rho_{A\vec{R}_1\vec{R}_2\dotsb\vec{R}_nB}\right)\coloneqq\sum_{\vec{z},\vec{x}\in[d]^n}\Tr_{\vec{R}_1\vec{R}_2\dotsb\vec{R}_n}\left[M_{\vec{R}_1\vec{R}_2\dotsb\vec{R}_n}^{\vec{z},\vec{x}}W_B^{\vec{z},\vec{x}}\left(\rho_{A\vec{R}_1\vec{R}_2\dotsb\vec{R}_nB}\right)\left(W_B^{\vec{z},\vec{x}}\right)^\dagger\right].
		\end{equation}
		
		\begin{figure}
			\centering
			\includegraphics[scale=1]{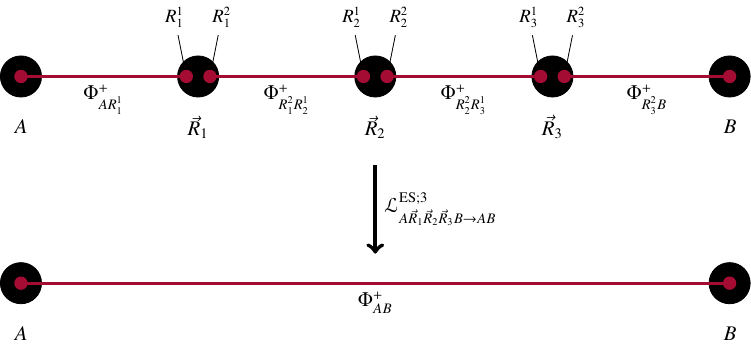}
			\caption{A chain of five nodes corresponding to the entanglement swapping protocol with $n=3$ intermediate nodes. The red lines represent maximally entangled states. The goal of the entanglement swapping protocol is to establish entanglement between $A$ and $B$. The protocol proceeds by first performing a Bell-basis measurement on the systems at the nodes $\vec{R}_j$, $1\leq j\leq n$, and communicating the results of the measurement to $B$, who applies a correction operation based on the outcomes.}\label{fig-ent_swap_chain}
		\end{figure}
		
		The standard entanglement swapping protocol \cite{ZZH93} corresponds to the input state
		\begin{equation}
			\rho_{A\vec{R}_1\vec{R}_2\dotsb\vec{R}_nB}=\Phi_{AR_1^1}^+\otimes \Phi_{R_1^2R_2^1}^+\otimes\dotsb\otimes\Phi_{R_{n-1}^2R_n^1}^+\otimes\Phi_{R_n^2B}^+.
		\end{equation}
		This scenario is shown in Figure~\ref{fig-ent_swap_chain}. Indeed, it can be shown that 
		\begin{equation}
			\mathcal{L}_{A\vec{R}_1\vec{R}_2\dotsb\vec{R}_nB\to AB}^{\text{ES};n}\left(\Phi_{AR_1^1}^+\otimes \Phi_{R_1^2R_2^1}^+\otimes\dotsb\otimes\Phi_{R_{n-1}^2R_n^1}^+\otimes\Phi_{R_n^2B}^+\right)=\Phi_{AB}^+,
		\end{equation}
		Furthermore, the standard teleportation protocol \cite{BBC+93} corresponds to $n=1$ and the input state
		\begin{equation}
			\rho_{A\vec{R}_1B}=\sigma_{R_1^1}\otimes\Phi_{R_1^2B}^+,
		\end{equation}
		where $A=\varnothing$ is a trivial (one-dimensional) system and $\sigma_{R_1^1}$ is an arbitrary $d$-dimensional quantum state, so that
		\begin{equation}
			\mathcal{L}_{\vec{R}_1\to B}^{\text{ES};1}(\sigma_{R_1^1}\otimes\Phi_{R_1^2B}^+)=\sigma_B,
		\end{equation}
		as expected.~\defqed.
	\end{example}
	\medskip
	\begin{example}[GHZ entanglement swapping protocol]\label{ex-GHZ_ent_swap}
		The previous example takes a chain of Bell states and transforms them into a Bell state shared by the end nodes of the chain. In this example, we look at a protocol that takes the same chain of Bell states and transforms them instead to a multi-qubit GHZ state, which is defined as \cite{GHZ89}
		\begin{equation}\label{eq-GHZ_state}
			\ket{\text{GHZ}_n}\coloneqq\frac{1}{\sqrt{2}}(\ket{\underbrace{0,0,\dotsc,0}_{n\text{ times}}}+\ket{1,1,\dotsc,1}).
		\end{equation}
		We call this protocol the \textit{GHZ entanglement swapping protocol}.
		
		The protocol for taking a chain of two Bell states to a three-party GHZ state is shown in Figure~\ref{fig-ent_swap_GHZ}. First, the two qubits $R_1^1$ and $R_1^2$ in the central node are entangled with a CNOT gate, followed by a measurement of $R_1^2$ in the standard basis (with corresponding POVM $\{\ket{0}\bra{0},\ket{1}\bra{1}\}$. The result $x\in\{0,1\}$ is communicated to $B$, where the correction operation $X_B^x$ is applied. The LOCC channel corresponding to this protocol is
		\begin{equation}\label{eq-GHZ_ent_swap_channel_1}
			\mathcal{L}_{A\vec{R}_1B}^{\text{GHZ};1}\left(\rho_{A\vec{R}_1B}\right)=\sum_{x=0}^1 \left(K_{\vec{R}}^x\otimes X_B^x\right)\rho_{A\vec{R}_1B}\left(K_{\vec{R}}^x\otimes X_B^x\right)^\dagger,
		\end{equation}
		where
		\begin{equation}
			K_{\vec{R}}^x\coloneqq\bra{x}_{R_1^2}\text{CNOT}_{\vec{R}_1},\quad \text{CNOT}_{\vec{R}_1}\coloneqq\ket{0}\bra{0}_{R_1^1}\otimes\mathbbm{1}_{R_1^2}+\ket{1}\bra{1}_{R_1^1}\otimes X_{R_1^2}.
		\end{equation}
		
		\begin{figure}
			\centering
			\includegraphics[scale=1]{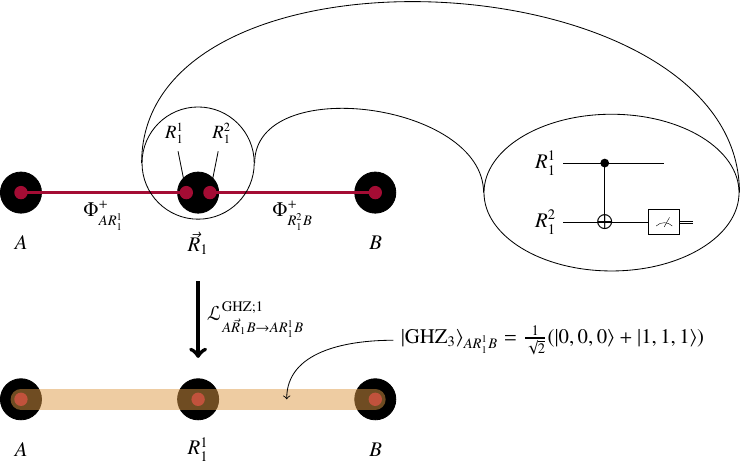}
			\caption{The GHZ entanglement swapping protocol with one intermediate node. The two qubits in the central node are entangled using the CNOT gate, after which the qubit $R_1^2$ is measured in the standard basis. The result $x\in\{0,1\}$ of the measurement is communicated to $B$, where the gate $X_B^x$ is applied; i.e., if the result of the measurement is $x=0$, then no correction is needed, and if the result is $x=1$, then $X_B$ (the Pauli-$x$ gate) is applied.}\label{fig-ent_swap_GHZ}
		\end{figure}
		
		The protocol shown in Figure~\ref{fig-ent_swap_GHZ}, with corresponding LOCC quantum channel in \eqref{eq-GHZ_ent_swap_channel_1}, can be easily extended to a scenario with $n>1$ intermediate nodes. In this case, the node $\vec{R}_1$ starts by applying the gate $\text{CNOT}_{\vec{R}_1}$ to its qubits and then measuring the qubit $R_1^2$ in the standard basis. The outcome of this measurement is sent to the node $\vec{R}_2$, and the corresponding correction operation is applied to the qubit $R_2^1$. Then, the gate $\text{CNOT}_{\vec{R}_2}$ is applied to the qubits at $\vec{R}_2$, followed by a standard-basis measurement of $R_2^2$ and communication of the outcome to $\vec{R}_3$ and a correction operation on $R_3^1$. This proceeds in sequence until the $n^{\text{th}}$ intermediate node $\vec{R}_n$, which sends its measurement outcome to $B$, which applies the appropriate correction operation. The LOCC channel for this protocol is
		\begin{equation}\label{eq-GHZ_ent_swap_channel}
			\mathcal{L}_{A\vec{R}_1\dotsb\vec{R}_nB\to AR_1^1\dotsb R_n^1B}^{\text{GHZ};n}\left(\rho_{A\vec{R}_1\dotsb\vec{R}_nB}\right)\coloneqq\sum_{\vec{x}\in\{0,1\}^n} P_{\vec{R}_1\dotsb\vec{R}_nB}^{\vec{x}}\left(\rho_{A\vec{R}_1\dotsb\vec{R}_nB}\right)P_{\vec{R}_1\dotsb\vec{R}_nB}^{\vec{x}~\dagger},
		\end{equation}
		where
		\begin{equation}
			P_{\vec{R}_1\dotsb\vec{R}_nB}^{\vec{x}}\coloneqq K_{\vec{R}_1}^{x_1}\otimes K_{\vec{R}_2}^{x_2}X_{R_2^1}^{x_1}\otimes\dotsb\otimes K_{\vec{R}_n}^{x_n}X_{R_n^1}^{x_{n-1}}\otimes X_B^{x_n}
		\end{equation}
		for all $\vec{x}\in\{0,1\}^n$. If the input state to this channel is
		\begin{equation}
			\rho_{A\vec{R}_1\dotsb\vec{R}_nB}=\Phi_{AR_1^1}^+\otimes \Phi_{R_1^2R_2^1}^+\otimes\dotsb\otimes\Phi_{R_{n-1}^2R_n^1}^+\otimes\Phi_{R_n^2B}^+,
		\end{equation}
		then the output is a $(n+2)$-party GHZ state given by the state vector $\ket{\text{GHZ}_{n+2}}_{AR_1^1\dotsb R_n^1B}$ as defined in \eqref{eq-GHZ_state}, i.e.,
		\begin{equation}
			\mathcal{L}_{A\vec{R}_1\dotsb\vec{R}_nB\to AR_1^1\dotsb R_n^1B}^{\text{GHZ};n}\left(\Phi_{AR_1^1}^+\otimes \Phi_{R_1^2R_2^1}^+\otimes\dotsb\otimes\Phi_{R_{n-1}^2R_n^1}^+\otimes\Phi_{R_n^2B}^+\right)=\ket{\text{GHZ}_{n+2}}\bra{\text{GHZ}_{n+2}}.~\defqedspec
		\end{equation}
	\end{example}
	\medskip
	\begin{example}[Graph state distribution protocol]\label{ex-graph_state_dist}
		We now consider an example of distributing an arbitrary graph state, which can be viewed as a special case of the procedure considered in Ref.~\cite{MMG19}. A graph state \cite{BR01,RB01,Briegel09} is a multi-qubit quantum state defined using graphs. (In Section~\ref{sec-graph_theory} below, we provide a brief review of graph theory.)
		
		Consider a graph $G=(V,E)$, which consists of a set $V$ of vertices and a set $E$ of edges. For the purposes of this example, $G$ is an undirected graph, and $E$ is a set of two-element subsets of $V$. The graph state $\ket{G}$ is an $n$-qubit quantum state $\ket{G}_{A_1\dotsb A_n}$, with $n=|V|$, that is defined as
		\begin{equation}
			 \ket{G}_{A_1\dotsb A_n}\coloneqq\sum_{\vec{\alpha}\in\{0,1\}^n}(-1)^{\frac{1}{2}\vec{\alpha}^{\t}A(G)\vec{\alpha}}\ket{\vec{\alpha}},
		\end{equation}
		where $A(G)$ is the adjacency matrix of $G$, which is defined as
		\begin{equation}
			A(G)_{i,j}=\left\{\begin{array}{l l} 1 & \text{if } \{v_i,v_j\}\in E,\\ 0 & \text{otherwise}, \end{array}\right.
		\end{equation}
		and $\vec{\alpha}$ is the column vector $(\alpha_1,\dotsc,\alpha_n)^{\t}$. It is easy to show that
		\begin{equation}
			\ket{G}_{A_1\dotsb A_n}=\text{CZ}(G)(\ket{+}_{A_1}\otimes\dotsb\otimes\ket{+}_{A_n}),
		\end{equation}
		where $\ket{+}\coloneqq\frac{1}{\sqrt{2}}(\ket{0}+\ket{1})$ and
		\begin{equation}
			\text{CZ}(G)\coloneqq\bigotimes_{\{v_i,v_j\}\in E}\text{CZ}_{A_iA_j},
		\end{equation}
		with $\text{CZ}_{A_iA_j}\coloneqq\ket{0}\bra{0}_{A_i}\otimes\mathbbm{1}_{A_j}+\ket{1}\bra{1}_{A_i}\otimes Z_{A_j}$ being the controlled-$Z$ gate.
		
		\begin{figure}
			\centering
			\includegraphics[width=0.97\textwidth]{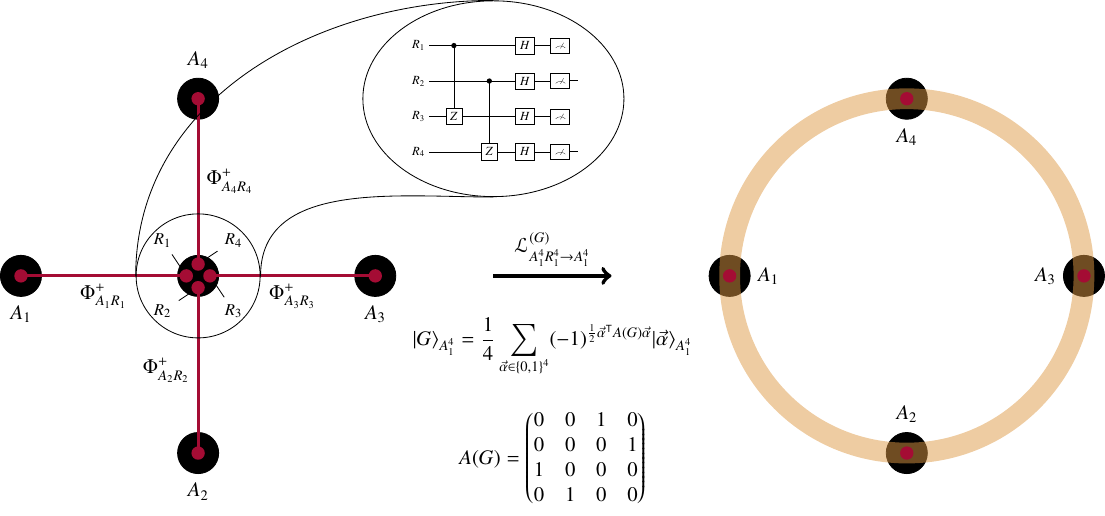}
			\caption{Depiction of a protocol for distributing a graph state among four nodes $A_1,A_2,A_3,A_4$, all of which initially share Bell states with the central node.}\label{fig-graph_state_dist}
		\end{figure}
		
		Now, consider the scenario depicted in Figure~\ref{fig-graph_state_dist}, in which $n=4$ nodes share Bell states with a central node. The task is for the central node to distribute the graph state $\ket{G}$ to the outer nodes. One possible procedure is for the central node to locally prepare the graph state and then to teleport the individual qubits using the Bell states. However, it is possible to perform a slightly simpler procedure that does not require the additional qubits needed to prepare the graph state locally. In fact, the following deterministic procedure produces the required graph state $\ket{G}$ shared by the nodes $A_1,\dotsc,A_n$.
		\begin{enumerate}
			\item The central node applies $\text{CZ}(G)$ to the qubits $R_1,\dotsc,R_n$.
			\item On each of the qubits $R_1,\dotsc,R_n$, the central node performs the measurement given by the POVM $\{\ket{+}\bra{+},\ket{-}\bra{-}\}$, where $\ket{\pm}=\frac{1}{\sqrt{2}}(\ket{0}\pm\ket{1})$. The outcome is an $n$-bit string $\vec{x}=(x_1,\dotsc,x_n)$, where $x_i=0$ corresponds to the ``$+$'' outcome and $x_i=1$ corresponds to the ``$-$'' outcome. The central node communicates outcome $x_i$ to the node $A_i$.
			\item The nodes $A_i$ apply $Z^{x_i}$ to their qubit. In other words, if $x_i=0$, then $A_i$ does nothing, and if $x_i=1$, then $A_i$ applies $Z$ to their qubit.
		\end{enumerate}
		Let us prove that this protocol achieves the desired outcome. First, observe that
		\begin{equation}
			\ket{\Phi^+}_{A_1R_1}\otimes\dotsb\otimes\ket{\Phi^+}_{A_nR_n}=\frac{1}{\sqrt{2^n}}\sum_{\vec{\alpha}\in\{0,1\}^n}\ket{\vec{\alpha}}_{A_1\dotsb A_n}\ket{\vec{\alpha}}_{R_1\dotsb R_n}.
		\end{equation}
		Then, after the first step, the state is
		\begin{equation}
			\frac{1}{\sqrt{2}^n}\sum_{\vec{\alpha}\in\{0,1\}^n}\ket{\vec{\alpha}}_{A_1\dotsb A_n}\text{CZ}(G)\ket{\vec{\alpha}}_{R_1\dotsb R_n}=\frac{1}{\sqrt{2^n}}\sum_{\vec{\alpha}\in\{0,1\}^n}(-1)^{\frac{1}{2}\vec{\alpha}^{\t}A(G)\vec{\alpha}}\ket{\vec{\alpha}}_{A_1\dotsb A_n}\ket{\vec{\alpha}}_{R_1\dotsb R_n},
		\end{equation}
		where we have used the fact that
		\begin{equation}
			\text{CZ}(G)\ket{\vec{\alpha}}=(-1)^{\sum_{i,j:\{v_i,v_j\}\in E}\alpha_i\alpha_j}=(-1)^{\frac{1}{2}\vec{\alpha}^{\t}A(G)\vec{\alpha}}.
		\end{equation}
		Then, using the fact that $\braket{+}{\alpha}=\frac{1}{\sqrt{2}}$ for all $\alpha\in\{0,1\}$ and $\braket{-}{\alpha}=\frac{1}{\sqrt{2}}(-1)^{\alpha}$ for all $\alpha\in\{0,1\}$, we find that for every outcome string $(x_1,\dotsc,x_n)$ of the measurement on the qubits $R_1,\dotsc,R_n$ the corresponding (unnormalized) post-measurement state is
		\begin{equation}
			\frac{1}{\sqrt{2^n}}\sum_{\vec{\alpha}\in\{0,1\}^n}(-1)^{\frac{1}{2}\vec{\alpha}^{\t}A(G)\vec{\alpha}}\frac{1}{\sqrt{2^n}}(-1)^{\alpha_1x_1+\dotsb+\alpha_nx_n}\ket{\vec{\alpha}}_{A_1\dotsb A_n}.
		\end{equation}
		Then, using the fact that $Z^x\ket{\alpha}=(-1)^{\alpha x}\ket{\alpha}$ for all $x,\alpha\in\{0,1\}$, we find that at the end of the second step the (unnormalized) state is
		\begin{equation}
			\frac{1}{\sqrt{2^n}}(\sigma_z^{x_1}\otimes\dotsb\otimes\sigma_z^{x_n})\frac{1}{\sqrt{2^n}}\sum_{\vec{\alpha}\in\{0,1\}^n}(-1)^{\frac{1}{2}\vec{\alpha}^{\t}A(G)\vec{\alpha}}\ket{\vec{\alpha}}_{A_1\dotsb A_n}=\frac{1}{\sqrt{2^n}}(Z_{A_1}^{x_1}\otimes\dotsb\otimes Z_{A_n}^{x_n})\ket{G}_{A_1\dotsb A_n}
		\end{equation}
		for all $(x_1,\dotsc,x_n)\in\{0,1\}^n$. From this, we see that, up to local Pauli-$z$ corrections, the post-measurement state is equal to the desired graph state $\ket{G}$ with probability $\frac{1}{2^n}$ for every measurement outcome string $(x_1,\dotsc,x_n)$. Once all of the nodes $A_i$ receive their corresponding outcome $x_i$ and apply the correction $Z_{A_i}^{x_i}$, the nodes $A_1,\dotsc,A_n$ share the graph state $\ket{G}$. As a result of the classical communication of the measurement outcomes and the subsequent correction operations, the protocol is deterministic.
		
		This protocol has the following representation as an LOCC channel:
		\begin{equation}
			\mathcal{L}_{A_1^nR_1^n\to A_1^n}^{(G)}\left(\rho_{A_1^nR_1^n}\right)\coloneqq\sum_{\vec{x}\in\{0,1\}^n}\left(Z_{A_1^n}^{\vec{x}}\otimes\bra{\vec{x}}_{R_1^n}H^{\otimes n}\text{CZ}(G)_{R_1^n}\right)\left(\rho_{A_1^nR_1^n}\right)\left(Z_{A_1^n}^{\vec{x}}\otimes \text{CZ}(G)_{R_1^n}^{\dagger}H^{\otimes n}\ket{\vec{x}}_{R_1^n}\right),
		\end{equation}
		for all states $\rho_{A_1^nR_1^n}$, where $H=\ket{+}\bra{0}+\ket{-}\bra{1}$ is the Hadamard operator, and we have let
		\begin{equation}
			Z_{A_1\dotsb A_n}^{\vec{x}}\coloneqq Z_{A_1}^{x_1}\otimes\dotsb\otimes Z_{A_n}^{x_n}.
		\end{equation}
		We have also used the abbreviation $A_1^n\equiv A_1A_2\dotsb A_n$, and similarly for $R_1^n$. Using the fact that
		\begin{equation}\label{eq-graph_state_x_0}
			\text{CZ}(G)H^{\otimes n}\ket{\vec{x}}=Z^{\vec{x}}\ket{G}
		\end{equation}
		for all $\vec{x}\in\{0,1\}^n$, and letting
		\begin{equation}\label{eq-graph_state_x}
			\ket{G^{\vec{x}}}\coloneqq Z^{\vec{x}}\ket{G},
		\end{equation}
		we can write the channel in the following simpler form:
		\begin{equation}\label{eq-graph_state_dist_channel}
			\mathcal{L}_{A_1^nR_1^n\to A_1^n}^{(G)}\left(\rho_{A_1^nR_1^n}\right)=\sum_{\vec{x}\in\{0,1\}^n} \left(Z_{A_1^n}^{\vec{x}}\otimes\bra{G^{\vec{x}}}_{R_1^n}\right)\left(\rho_{A_1^nR_1^n}\right)\left(Z_{A_1^n}^{\vec{x}}\otimes\ket{G^{\vec{x}}}_{R_1^n}\right).
		\end{equation}
		From this, we see that the protocol can be thought of as measuring the systems $R_1,\dotsc, R_n$ jointly with the POVM $\left\{\ket{G^{\vec{x}}}\bra{G^{\vec{x}}}\right\}_{\vec{x}\in\{0,1\}^n}$ and, conditioned on the outcome $\vec{x}$, applying the correction operation $Z^{\vec{x}}$ to the systems $A_1,\dotsc, A_n$. Note that $\left\{\ket{G^{\vec{x}}}\bra{G^{\vec{x}}}\right\}_{\vec{x}\in\{0,1\}^n}$ is indeed a POVM due to the fact that
		\begin{equation}
			\ket{G^{\vec{x}}}=\text{CZ}(G)H^{\otimes n}\ket{\vec{x}}
		\end{equation}
		for all $\vec{x}\in\{0,1\}^n$, which follows from \eqref{eq-graph_state_x_0} and \eqref{eq-graph_state_x}, so that
		\begin{equation}
			\sum_{\vec{x}\in\{0,1\}^n}\ket{G^{\vec{x}}}\bra{G^{\vec{x}}}=\text{CZ}(G)H^{\otimes n}\underbrace{\sum_{\vec{x}\in\{0,1\}^n}\ket{\vec{x}}\bra{\vec{x}}}_{\mathbbm{1}}H^{\otimes n}\text{CZ}(G)^{\dagger}=\mathbbm{1}. \quad \defqedspec
		\end{equation}
	\end{example}
	\medskip
	\begin{example}[Entanglement distillation]\label{ex-ent_distill}
		The term ``entanglement distillation'' refers to the task of taking many copies of a given quantum state $\rho_{AB}$ and transforming them, via an LOCC protocol, to several (fewer) copies of the maximally entangled state $\Phi_{AB}^+$. Typically, with only a finite number of copies of the initial state $\rho_{AB}$, it is not possible to perfectly obtain copies of the maximally entangled state, so we aim instead for a state $\sigma_{AB}$ whose fidelity $F(\Phi_{AB},\sigma_{AB})$ to the maximally entangled state is higher than the fidelity $F(\Phi_{AB},\rho_{AB})$ of the initial state. Mathematically, the task of entanglement distillation corresponds to the transformation
		\begin{equation}
			\rho_{AB}^{\otimes n}\mapsto\mathcal{L}_{A^nB^n\to A^m B^m}(\rho_{AB}^{\otimes n})=\sigma_{AB}^{\otimes m},
		\end{equation}
		where $n,m\in\mathbb{N}$, $m<n$, and $\mathcal{L}_{A^nB^n\to A^m B^m}$ is an LOCC channel.
		
		Typically, in practice, we have $n=2$ and $m=1$, with the task being to transform two two-qubit states $\rho_{A_1B_1}^1$ and $\rho_{A_2B_2}^2$ to a two-qubit state $\sigma_{A_1B_1}$ with a higher fidelity. Protocols achieving this aim are typically probabilistic in practice, meaning that the state $\sigma_{A_1B_1}$ with higher fidelity is obtained only with some non-unit probability.
		
		\begin{figure}
			\centering
			\includegraphics[scale=1]{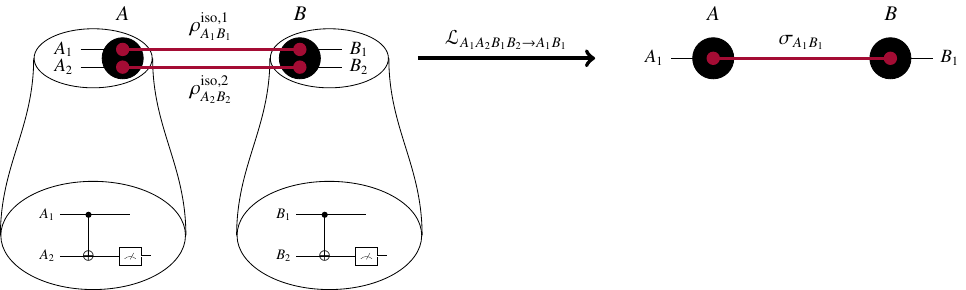}
			\caption{Depiction of the simple entanglement distillation protocol from Ref.~\cite{BBP96}. The protocol takes two isotropic states $\rho_{A_jB_j}^{\text{iso},j}$, $j\in\{1,2\}$ (see \eqref{eq-isotropic_state}), and transforms them probabilistically to a state $\sigma_{A_1B_1}$ with higher fidelity.}\label{fig-ent_distill_example}
		\end{figure}
		
		We are not concerned with any particular entanglement distillation protocol in this thesis. All we are concerned with is their mathematical structure. In particular, entanglement distillation protocols that are probabilistic can be described mathematically as an LOCC instrument, which we now demonstrate with a simple example, depicted in Figure~\ref{fig-ent_distill_example}, which comes from Ref.~\cite{BBP96}. In this protocol, Alice and Bob first apply the CNOT gate to their qubits and follow it with a measurement of their second qubit in the standard basis. They then communicate the results of their measurement to each other. The protocol is considered successful if they both obtain the same outcome, and a failure otherwise. This protocol has the following corresponding LOCC instrument channel:
		\begin{align}
			\mathcal{L}_{A_1A_2B_1B_2\to A_1B_1}\left(\rho_{A_1B_1}^1\otimes\rho_{A_2B_2}^2\right)&=\ket{0}\bra{0}\otimes\left((K_{A}^0\otimes K_{B}^1)(\rho_{A_1B_1}^{\text{iso},1}\otimes\rho_{A_2B_2}^{\text{iso},2})(K_{A}^0\otimes K_{B}^1)^{\dagger}\right.\nonumber\\
			&\qquad\qquad\qquad\left.+(K_{A}^1\otimes K_{B}^0)(\rho_{A_1B_1}^{\text{iso},1}\otimes\rho_{A_2B_2}^{\text{iso},2})(K_{A}^1\otimes K_{B}^0)^{\dagger}\right)\nonumber\\
			&+\ket{1}\bra{1}\otimes\left((K_{A}^0\otimes K_{B}^0)(\rho_{A_1B_1}^{\text{iso},1}\otimes\rho_{A_2B_2}^{\text{iso},2})(K_{A}^0\otimes K_{B}^0)^{\dagger}\right.\nonumber\\
			&\qquad\qquad\qquad\left.+(K_{A}^1\otimes K_{B}^1)(\rho_{A_1B_1}^{\text{iso},1}\otimes\rho_{A_2B_2}^{\text{iso},2})(K_{A}^1\otimes K_{B}^1)^{\dagger}\right)\label{eq-ent_distill_channel}
		\end{align}
		where
		\begin{align}
			K_A^{x}&\equiv K_{A_1A_2\to A_1}^x \coloneqq\bra{x}_{A_2}\text{CNOT}_{A_1A_2}\quad\forall~x\in\{0,1\},\\
			K_B^x&\equiv K_{B_1B_2\to B_1}^x\coloneqq\bra{x}_{B_2}\text{CNOT}_{B_1B_2}\quad\forall~x\in\{0,1\}.
		\end{align}
		Furthermore, the states $\rho_{A_iB_i}^{\text{iso},j}$, $j\in\{1,2\}$ are defined as
		\begin{equation}\label{eq-isotropic_state}
			\rho_{A_jB_j}^{\text{iso},j}\coloneqq\mathcal{T}_{A_jB_j}^U(\rho_{A_jB_j}^j)\coloneqq\int_U \left(U_{A_j}\otimes\conj{U}_{B_j}\right)(\rho_{A_jB_j}^j)\left(U_{A_j}\otimes\conj{U}_{B_j}\right)^{\dagger},
		\end{equation}
		where $\mathcal{T}^U$ is the \textit{isotropic twirling channel}; see, e.g., \cite[Example~7.25]{Wat18_book}.
		
		In general, the isotropic twirling channel $\mathcal{T}_{AB}^U$, with $d_A=d_B=d\geq 2$, is an LOCC channel that is defined by an LOCC protocol in which Alice (who holds system $A$) picks a unitary $U$ uniformly at random (i.e., from the Haar distribution), applies it to her system, and communicates her choice to Bob, who applies the complex conjugate $\conj{U}$ of the unitary to his system. The choice of the unitary is then forgotten. For every linear operator $X_{AB}$, the output of the isotropic twirling channel has the following form \cite[Example~7.25]{Wat18_book}:
		\begin{equation}
			\mathcal{T}_{AB}^U(X_{AB})=\widetilde{\alpha}(X_{AB})\Phi_{AB}^++\widetilde{\beta}(X_{AB})\frac{\mathbbm{1}_{AB}}{d^2},
		\end{equation}
		where
		\begin{align}
			\widetilde{\alpha}(X_{AB})&\coloneqq\frac{d^2\bra{\Phi^+}X_{AB}\ket{\Phi^+}-\Tr[X_{AB}]}{d^2-1},\\
			\widetilde{\beta}(X_{AB})&\coloneqq\frac{d^2\left(\Tr[X_{AB}]-\bra{\Phi^+}X_{AB}\ket{\Phi^+}\right)}{d^2-1}.
		\end{align}
		In other words, the isotropic twirling channel takes a linear operator and makes it ``isotropic'', i.e., invariant under the action of $U\otimes\conj{U}$ for all unitaries $U$.
		
		It is a straightforward calculation to show that if $f_1=F(\Phi_{A_1B_1}^+,\rho_{A_1B_1}^1)$ and $f_2=F(\Phi_{A_2B_2}^+,\rho_{A_2B_2}^2)$ are the fidelities of the initial states, then the protocol depicted in Figure~\ref{fig-ent_distill_example}, with corresponding LOCC channel given by \eqref{eq-ent_distill_channel}, succeeds with probability
		\begin{equation}
			p_{\text{succ}}=\frac{8}{9}f_1f_2-\frac{2}{9}(f_1+f_2)+\frac{5}{9},
		\end{equation}
		and the fidelity of the output state $\sigma_{A_1B_1}$ (conditioned on success) is
		\begin{equation}
			\frac{1}{p_{\text{succ}}}\left(\frac{10}{9}f_1f_2-\frac{1}{9}(f_1+f_2)+\frac{1}{9}\right).
		\end{equation}
		
		In general, every probabilistic entanglement distillation protocol has a corresponding LOCC quantum instrument channel with a form analogous to \eqref{eq-ent_distill_channel}, namely,
		\begin{equation}
			\mathcal{L}(\cdot)=\ket{0}\bra{0}\otimes\mathcal{L}^0(\cdot)+\ket{1}\bra{1}\otimes\mathcal{L}^1(\cdot),
		\end{equation}
		where the classical register indicates failure (`0') and success (`1'), and $\mathcal{L}^0$ and $\mathcal{L}^1$ are the corresponding completely positive trace non-increasing LOCC maps such that $\mathcal{L}^0+\mathcal{L}^1$ is an LOCC channel.~\defqed
	\end{example}

\section{Quantum key distribution}\label{sec-QKD}

	Quantum key distribution (QKD) is a method for creating a secret key using the principles of quantum mechanics, which can be used for private communication (e.g., via the one-time pad). A secret key is a random string of bits shared by the desired trusted parties that is not known to any eavesdropper or adversary. Intuitively, the underlying quantum-mechanical principles that make quantum key distribution secure are no-cloning, the Heisenberg uncertainty relation, and the fact that measurements disturb quantum systems. Long-distance QKD is one of the most prominent applications of quantum networks and of the quantum internet in general. The purpose of this section is to provide a brief summary of the basic concepts and language of QKD in order to use them later in Chapter~\ref{chap-sats} when we discuss satellite-based entanglement distribution. We refer to Refs.~\cite{GRG+02,SBPC+09,Lut14,MyhrThesis,kaur2020,XXQ+20,PAB+19} for pedagogical introductions and reviews of state-of-the-art QKD research.
	
	Let us consider the following scenario of so-called entanglement-based QKD. Suppose that Alice and Bob have access to a source that distributes entangled states $\rho_{AB}$ to them, and that their task is to use many copies of this quantum state to distill a secret key. The general strategy of Alice and Bob is to measure their quantum systems. Based on their measurement statistics, they decide whether or not to use their classical measurement data to distill a secret key. The measurement statistics are of the form
	\begin{equation}\label{eq-QKD_correlation}
		p_{AB}(x,y|a,b)\coloneqq\Tr\left[\left(\Pi_{A}^{a,x}\otimes \Lambda_B^{b,y}\right)\rho_{AB}\right],\quad x\in\mathcal{X},\,y\in\mathcal{Y},\,a\in\mathcal{A},\,b\in\mathcal{B},
	\end{equation}
	where $\mathcal{A}$ and $\mathcal{B}$ are finite sets of POVMs, such that $\{\Pi_A^{a,x}\}_{x\in\mathcal{X}}$ is a POVM for Alice's measurement for every $a\in\mathcal{A}$ and $\{\Lambda_B^{b,y}\}_{y\in\mathcal{Y}}$ is a POVM for Bob's measurement for every $b\in\mathcal{B}$. 
	
	From a cryptographic perspective, we would like to make as few assumptions as possible about what can be trusted when proving the security of QKD, i.e., when determining how much eavesdropping can be tolerated and what rates are achievable in principle. One basic assumption is that the source itself is untrusted, which means that Alice and Bob cannot base their protocol on any prior knowledge of the state $\rho_{AB}$. Beyond this basic assumption, one can ask whether the measurement devices being used by Alice and Bob in order to obtain the correlation in \eqref{eq-QKD_correlation} can be trusted. In this context, we consider two scenarios.
	
	\begin{itemize}
		\item \textit{Device-dependent scenario}: The devices implementing the measurement are known and trusted. Alice and Bob can thus use not only their measurement statistics, but their knowledge of the specific measurements being implemented, in order to distill a secret key.
		
		\item\textit{Device-independent scenario}: The devices implementing the measurement are assumed to be untrusted, and thus, from a cryptographic perspective, are part of the eavesdropper's domain. Furthermore, no assumption on the dimensions of the systems $A$ and $B$ are made (except that they are finite). This means that, in this scenario, Alice and Bob can use only their measurement statistics to distill a secret key. We refer to Ref.~\cite{AF20} for a general introduction to device-independent quantum information processing.
	\end{itemize}
	
	Proving security in the device-independent scenario means that Alice and Bob can distill a secret key even without trusting their devices, which is ideal from a practical perspective because the future quantum internet will be comprised of devices made by third-party manufacturers, which may not be trustworthy and might be susceptible to hacking by adversaries.

\subsubsection*{Device-dependent protocols}

	Two well-known device-dependent protocols that we discuss here are the BB84 \cite{BB84} and six-state \cite{B98,BG99} protocols. The original formulation of these protocols is as so-called prepare-and-measure protocols, which do not require Alice and Bob to share entanglement. However, these protocols can be viewed from an entanglement-based point of view, in which Alice and Bob possess an entangled state; see Ref.~\cite{MyhrThesis} for a discussion on the equivalence of entanglement-based and prepare-and-measure-based protocols, and Ref.~\cite{TL17} for a more general discussion of the security of prepare-and-measure-based and entanglement-based QKD protocols. In this device-dependent scenario, we explicitly assume that the state $\rho_{AB}$ is a two-qubit state, and the correlation in \eqref{eq-QKD_correlation} is given by measurement of the qubit Pauli observables $X$, $Z$, and $Y=\I XZ$. In other words, the sets $\mathcal{A}$ and $\mathcal{B}$ indicate which observable to be measured, and the sets $\mathcal{X}$ and $\mathcal{Y}$ contain the outcomes of the measurements. It can be shown via certain symmetrization procedures that, without loss of generality, $\rho_{AB}$ is a Bell-diagonal state; see Refs.~\cite{MyhrThesis,WKKG19} for details. It then suffices to estimate the following three quantities, called \textit{quantum bit-error rates (QBERs)}, in order to characterize the eavesdropper's knowledge:
	\begin{align}
		Q_x&\coloneqq\Tr[(\ket{+}\bra{+}_A\otimes\ket{-}\bra{-}_B)\rho_{AB}]+\Tr[(\ket{-}\bra{-}_A\otimes\ket{+}\bra{+}_B)\rho_{AB}]\label{eq-rhoAB_Qx}\\
		&=\frac{1}{2}(1-\Tr[(X\otimes X)\rho_{AB}]),\\
		Q_y&\coloneqq\Tr[(\ket{+\I}\bra{+\I}_A\otimes\ket{-\I}\bra{-\I}_B)\rho_{AB}]+\Tr[(\ket{-\I}\bra{-\I}_A\otimes\ket{+\I}\bra{+\I}_B)\rho_{AB}]\label{eq-rhoAB_Qy}\\
		&=\frac{1}{2}(1+\Tr[(Y\otimes Y)\rho_{AB}]),\\
		Q_z&\coloneqq\Tr[(\ket{0}\bra{0}_A\otimes\ket{1}\bra{1}_B)\rho_{AB}]+\Tr[(\ket{1}\bra{1}_A\otimes\ket{0}\bra{0}_B)\rho_{AB}]\label{eq-rhoAB_Qz}\\
		&=\frac{1}{2}(1-\Tr[(Z\otimes Z)\rho_{AB}]),
	\end{align}
	where $\ket{\pm}=\frac{1}{\sqrt{2}}(\ket{0}\pm\ket{1})$ and $\ket{\pm\I}=\frac{1}{\sqrt{2}}(\ket{0}\pm\I\ket{1})$. For example, $Q_x$ is simply the probability that Alice and Bob's measurement outcomes disagree when they both measure the observable $X$, and similarly for $Q_y$ and $Q_z$.
	
	A standard figure of merit for QKD protocols is the number of secret key bits obtained per copy of the source state; see, e.g., Ref.~\cite{WKKG19} for precise definitions. For the BB84 protocol, the asymptotic secret key rate is \cite{MY98,LC99,BBB+00,SP00,Mayers01,BBB+06}
	\begin{equation}\label{eq-key_rate_BB84}
		K_{\text{BB84}}(\rho_{AB})=1-2h_2(Q),
	\end{equation}
	where $Q=\frac{1}{2}(Q_x+Q_z)$ and
	\begin{equation}
		h_2(Q)\coloneqq-Q\log_2(Q)-(1-Q)\log_2(1-Q)
	\end{equation}
	is the binary entropy.
		
	For the six-state protocol, the asymptotic secret key rate is \cite{B98,Lo01}.
	\begin{equation}\label{eq-key_rate_6state}
		K_{\text{6-state}}(Q)=1+\left(1-\frac{3Q}{2}\right)\log_2\left(1-\frac{3Q}{2}\right)+\frac{3Q}{2}\log_2\left(\frac{Q}{2}\right)
	\end{equation}
	where $Q=\frac{1}{3}(Q_x+Q_y+Q_z)$.
	
	\begin{remark}
		The QBERs $Q_x,Q_y,Q_z$ in \eqref{eq-rhoAB_Qx}, \eqref{eq-rhoAB_Qy}, and \eqref{eq-rhoAB_Qz} have a useful interpretation in terms of the fidelity of an arbitrary two-qubit state $\rho_{AB}$ to the maximally entangled state $\Phi_{AB}^+$. In particular,
		\begin{equation}\label{eq-rhoAB_fid_Bell_QBERs}
			F(\rho_{AB},\Phi_{AB}^+)=1-\frac{1}{2}(Q_x+Q_y+Q_z)
		\end{equation}
		for all two-qubit states $\rho_{AB}$. It is easy to see this by noting that
		\begin{equation}
			\Phi_{AB}^+=\frac{1}{4}\left(\mathbbm{1}_A\otimes\mathbbm{1}_B+X_A\otimes X_B-Y_A\otimes Y_B+Z_A\otimes Z_B\right).
		\end{equation}
		Then, using the definitions in \eqref{eq-rhoAB_Qx}, \eqref{eq-rhoAB_Qy}, and \eqref{eq-rhoAB_Qz}, we obtain \eqref{eq-rhoAB_fid_Bell_QBERs}.~\defqed
	\end{remark}

\subsubsection*{Device-independent protocols}
	
	The device-independent protocol that we present here is the one introduced in Refs.~\cite{AMP06,ABG+07}, and the basic idea behind the protocol comes from the protocol in Ref.~\cite{Eke91}. The security of the protocol is based on violation of a Bell inequality, specifically the CHSH inequality \cite{CHSH69} (see Ref.~\cite{Scar13} for a pedagogical introduction). In this protocol, unlike the device-dependent protocols shown above, it is not required to assume that $\rho_{AB}$ is a two-qubit state. However, like the device-dependent protocols considered above, there are symmetrization procedures and other reductions from which it can be argued that $\rho_{AB}$ is a two-qubit Bell-diagonal state without loss of generality; see Refs.~\cite{AMP06,ABG+07} for details. The correlation in \eqref{eq-QKD_correlation} is given by measurement of observables $P_A^0,P_A^1,P_A^2$ for system $A$ and observables $Q_B^1,Q_B^2$ for system $B$, and we assume that they all have spectral decompositions of the form
	\begin{align}
		P_A^j&=\Pi_A^{j,0}-\Pi_A^{j,1},\quad j\in\{0,1,2\},\\
		T_B^k&=\Lambda_B^{k,0}-\Lambda_B^{k,1},\quad k\in\{1,2\}.
	\end{align}
	In other words, $\mathcal{A}=\{0,1,2\}$, $\mathcal{B}=\{1,2\}$, and $\mathcal{X}=\mathcal{Y}=\{0,1\}$.
	
	Two quantities in this case characterize the secret key rate:
	\begin{equation}
		S\coloneqq \Tr\left[\left(P_A^1\otimes T_B^1+P_A^1\otimes T_B^2+P_A^2\otimes T_B^1-P_A^2\otimes T_B^2\right)\rho_{AB}\right],
	\end{equation}
	and the quantum bit-error rate (QBER) $Q$, which is defined as
	\begin{equation}
		Q\coloneqq \Tr[(\Pi_A^{0,0}\otimes\Lambda_B^{1,1})\rho_{AB}]+\Tr[(\Pi_A^{0,1}\otimes\Lambda_B^{1,0})\rho_{AB}].
	\end{equation}
	As with the QBERs defined previously, the QBER here is the probability that the outcomes of Alice and Bob disagree when a measurement of $P_A^0$ is performed by Alice and a measurement of $T_B^1$ is performed by Bob. The asymptotic key rate is then \cite{ABG+07,PAB+09}
	\begin{equation}\label{eq-key_rate_DIQKD}
		K_{\text{DI}}(Q,S)=1-h_2(Q)-h_2\left(\frac{1+\sqrt{(S/2)^2-1}}{2}\right).
	\end{equation}
	
	In Figure~\ref{fig-key_rates}, we plot the key rates in \eqref{eq-key_rate_BB84} and \eqref{eq-key_rate_6state} for the BB84 and six-state protocols, as well as the rate in \eqref{eq-key_rate_DIQKD} for the device-independent protocol. As one might expect, due to the stronger trust assumptions on the device-independent protocol, its key rate is lower than both the device-dependent BB84 and six-state protocols.

	\begin{figure}
		\centering
		\includegraphics[scale=1]{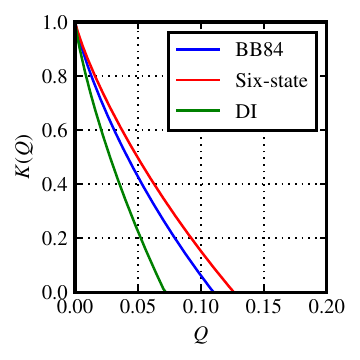}
		\caption{Plots of the secret key rates for device-dependent and device-independent (DI) QKD protocols. Shown is the secret key rate in \eqref{eq-key_rate_BB84} for the BB84 protocol, the secret key rate in \eqref{eq-key_rate_6state} for the six-state protocol, and the secret key rate in \eqref{eq-key_rate_DIQKD} for the DI protocol. For the DI protocol, we plot the secret key rate for the example considered in Ref.~\cite{ABG+07}, in which $P_A^0=T_B^1=Z$, $T_B^2=X$, $P_A^1=\frac{1}{\sqrt{2}}(Z+X)$, and $P_A^2=\frac{1}{\sqrt{2}}(Z-X)$, and $S=2\sqrt{2}(1-2Q)$. The BB84 secret key rate vanishes at $Q\approx 11\%$, while the six-state secret key rate vanishes at $Q\approx 12.7\%$ and the DI secret key rate vanishes at $Q\approx 7.1\%$.}\label{fig-key_rates}
	\end{figure}
	
	\begin{remark}
		We note that the key rate formulas presented here in \eqref{eq-key_rate_BB84}, \eqref{eq-key_rate_6state}, and \eqref{eq-key_rate_DIQKD} are achievable rates in the IID (i.e., collective attack) and asymptotic setting under certain assumptions on the nature of the protocol. Also, to be clear, let us state that these formulas represent the number of secret key bits that can be obtained per copy of $\rho_{AB}$ in the asymptotic and IID setting. More sophisticated procedures, such as noisy pre-processing and advantage distillation \cite{Maurer93}, can be used to increase the key rate threshold. See Refs.~\cite{KGR05,RGK05,KBR07,KR08} for noisy pre-processing and Refs.~\cite{GL03,Chau02,WMUK07,KBR07,BA07,MRDL09,KL17} for advantage distillation for the BB84 and six-state protocols; see Ref.~\cite{TSB+20} for noisy pre-processing and Ref.~\cite{TLR20} for advantage distillation for device-independent protocols.
	\end{remark}

\section{Graph theory}\label{sec-graph_theory}

	We end this chapter with a brief overview of graph theory. Graphs provide a visual representation of quantum networks (see Chapter~\ref{chap-network_setup}), and many graph-theoretic results are relevant for determining rates of entanglement distribution in a quantum network. In this section, we provide the basic definitions and theorems that are needed for this thesis. For further details, and proofs of the theorems presented here, please consult, e.g., Refs.~\cite{Bollobas98_book,BM08_book}. Throughout this thesis, we consider only undirected graphs.
	
	An undirected graph is a pair $G=(V,E)$, where $V$ is a set of objects, called \textit{vertices} or \textit{nodes}, and $E\subseteq\text{P}(V)\setminus\{\varnothing\}$ is a set of \textit{edges}, where $\text{P}(V)$ denotes the power set (set of subsets) of $V$. For an ordinary graph, each element $e\in E$ is a two-element set, i.e., $e=\{v_i,v_j\}$ for $v_i,v_j\in V$, while in the case of a \textit{hypergraph}, each element $e\in E$ is a set $\{v_{i_1},v_{i_2},\dotsc,v_{i_k}\}$ of $k\geq 2$ vertices and is called a \textit{hyperedge}. The \textit{degree} $d(v)$ of a vertex $v\in V$ is defined to be the number of edges that are incident to $v$, i.e., $d(v)=|\{e\in E:v\in e\}|$. The size of the graph $G$, denoted by $|G|$, is equal to the number of edges, i.e., $|G|\coloneqq |E|$. The term \textit{multigraph} is used if vertices can be connected to each other by multiple edges, often called parallel edges. In this case, the graph is specified by a triple, $G=(V,E,c)$, where $V$ and $E\subseteq\text{P}(V)\setminus\{\varnothing\}$ are defined as before, and $c:E\to\mathbb{N}$ is a function such that $ c(e)$ is equal to the number of parallel edges connecting the vertices in $e$. For a given $e\in E$, we let $e^1,e^2,\dotsc,e^{ c(e)}$ denote the $ c(e)$ parallel edges connecting the vertices in $e$, and we regard each of these parallel edges as distinct objects. The set of edges is then $\{e^j:1\leq j\leq  c(e),\,e\in E\}$ rather than $E$, which means that the size of a multigraph is $|G|=\sum_{e\in E} c(e)$, and the degree of a node $v\in V$ is $d(v)=\sum_{e\in E:v\in e} c(e)$. We often suppress the dependence of a multigraph $G$ on the function $c$ when there is no chance of confusion.
	
	Graphs are depicted visually by drawing the vertices as dots and the edges as lines connecting the appropriate vertices. Hyperedges are drawn as a bubble surrounding the appropriate vertices. See Figure~\ref{fig-graph} for a visual representation of a multigraph.
	
	\begin{figure}
		\centering
		\includegraphics[scale=1]{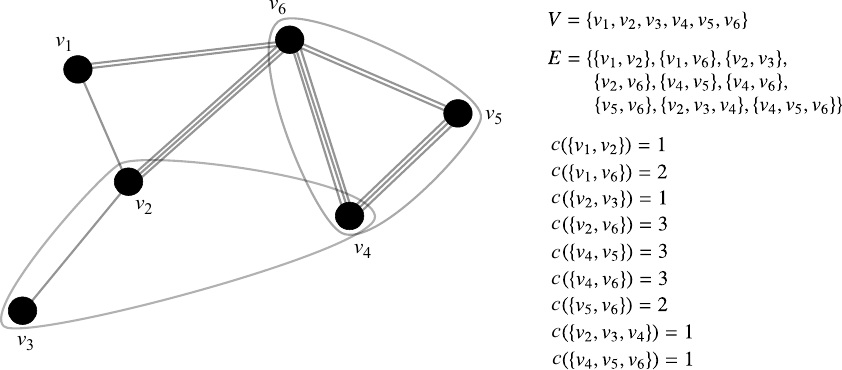}
		\caption{Visual representation of a multigraph $G=(V,E,c)$.}\label{fig-graph}
	\end{figure}
	
	For every graph $G=(V,E)$ consisting of only two-element edges, its \textit{adjacency matrix} is a $|V|\times|V|$ matrix denoted by $A(G)$, and it is defined as
	\begin{equation}
		A(G)_{i,j}\coloneqq\left\{\begin{array}{l l} 1 & \text{if }\{v_i,v_j\}\in E,\\0 & \text{otherwise},\end{array}\right.
	\end{equation}
	for all $1\leq i,j\leq |V|$.
	
	Given a graph $G=(V,E)$, for all $v,v'\in V$ we write $v\leftrightarrow v'$ if there exists a finite sequence $W=(v,e_1,w_1,e_2,w_2,\dotsc, w_{n-1},e_n,v')$ of edges $e_i\in E$ and vertices $w_i\in V$ such that $v,w_1\in e_1$, $w_1,w_2\in e_2$, etc., and $w_{n-1},v'\in e_n$. The sequence $W$ is called a \textit{walk} from $v$ to $v'$. A \textit{path} between two distinct vertices is a walk in which all vertices and edges are distinct.
	
	The relation ``$\leftrightarrow$'' defined in the previous paragraph is an equivalence relation. Indeed, it satisfies the following three properties:
	\begin{itemize}
		\item $v\leftrightarrow v$ for all $v\in V$ by taking $W=(v)$.
		\item If $v\leftrightarrow v'$, then $v'\leftrightarrow v$ by reversing the walk from $v$ to $v'$.
		\item If $v\leftrightarrow v'$ and $v'\leftrightarrow v''$, then $v\leftrightarrow v''$, which follows by concatenating the walks from $v$ to $v'$ and from $v'$ to $v''$.
	\end{itemize}
	We say that $v$ and $v'$ are \textit{connected} if $v\leftrightarrow v'$. Since ``$\leftrightarrow$'' is an equivalence relation, it can be used to partition the set $V$ of vertices. In particular, we let
	\begin{equation}
		[v]\coloneqq\{v'\in V: v'\leftrightarrow v\}
	\end{equation}
	denote the equivalence class corresponding to $v\in V$, i.e., the set of all vertices connected to $V$. Then, $V=\bigsqcup [v]$, i.e., $V$ is the disjoint union of the equivalence classes under ``$\leftrightarrow$''. We denote the set of edges connecting all of the vertices in $[v]$ by $E_{[v]}$, and we note that $E_{[v]}\subseteq E$ for all $v\in V$. Consequently, $C_v(G)\coloneqq ([v],E_{[v]})$ is a subgraph of $G$.
	
	Given a graph $G$, consider a subgraph $G'\coloneqq (V',E')$ of $G$, so that $V'\subseteq V$ and $E'\subseteq E$. We call $G'$ a \textit{connected subgraph} if $v\leftrightarrow w$ for all $v,w\in V'$. A \textit{maximal connected subgraph} is a connected subgraph such that no element of $V\setminus V'$ is connected to a vertex in $V'$. A maximal connected subgraph is also called a \textit{(connected) component}. The connected components partition a graph, and thus we can label every connected component by the equivalence classes under ``$\leftrightarrow$''. In particular, observe that the subgraph $C_v(G)$ is precisely a connected component of $G$. We let
	\begin{equation}
		\sfrac{G}{\leftrightarrow}\coloneqq\{C_v(G):v\in V\}
	\end{equation}
	be the set of all connected components of $G$, and we let $S^{\max}(G)$ denote the size of the largest connected component of $G$, i.e.,
	\begin{equation}\label{eq-largest_connected_component}
		S^{\max}(G)\coloneqq\max\{|C_v(G)|:v\in V\}=\max\{|C|:C\in\sfrac{G}{\leftrightarrow}\}.
	\end{equation}
	
	An important problem in graph theory, which is also relevant in the context of bipartite entanglement distribution in quantum networks, is to determine the number of edge-disjoint paths (i.e., paths that have no edges in common) between vertices in graphs and hypergraphs; see, e.g., Ref.~\cite{Nis90}. In the case of graphs, the solution is provided by Menger's theorem, and the analogous statement holds for hypergraphs; see, e.g., Ref.~\cite[Theorem~1.11]{Kir03}.
	\medskip
	\begin{theorem}[Menger~\cite{Men27}]\label{thm-Menger}
		Let $G$ be an undirected graph and let $v$ and $v'$ be distinct vertices of $G$. Then, the maximum number of edge-disjoint paths between $v$ and $v'$ is equal to the minimum number of edges that separates $v$ and $v'$.
	\end{theorem}
	
	A \textit{tree} (sometimes called a \textit{Steiner tree}) in a (hyper)graph $G$ is a connected acyclic sub-(hyper)graph. Given a subset $V'\subseteq V$, a \textit{spanning tree} for $V'$ is a tree in $G$ whose vertices contain the vertices in $V'$. See Ref.~\cite{Brazil2015} for more information on trees in graphs and hypergraphs. When $G$ is a graph, we have the following result pertaining to the number of edge-disjoint spanning trees of the entire graph $G$, and it is relevant for distributing multipartite entanglement in a quantum network.
	\medskip
	\begin{theorem}[Tutte~\cite{Tutte61} and Nash-Williams~\cite{NW61}]\label{thm-Tutte_NW}
		An undirected graph $G$ contains $k$ edge-disjoint spanning trees if and only if $E_G(\mathcal{P})\geq k(|\mathcal{P}|-1)$ holds for every partition $\mathcal{P}$ of $V$ into non-empty subsets, where $E_G(\mathcal{P})$ denotes the number of edges connecting distinct member of $\mathcal{P}$.
	\end{theorem}
	
	For more information on the general problem of finding edge-disjoint Steiner trees in (hyper)graphs (often called the \textit{Steiner tree packing problem}), we refer to Refs.~\cite{JMS03,Kaski04,Lau07}. Specifically, in Ref.~\cite{Kaski04}, it is shown that the Steiner tree packing problem is NP-complete, and Refs.~\cite{Kaski04,Lau07} provide algorithms for approximately determining the maximum number of edge-disjoint Steiner trees.

\chapter{QUANTUM DECISION PROCESSES}\label{chap-QDP}

	In this chapter, we detail the main mathematical tool that we use to develop quantum network protocols, namely, quantum decision processes. The notion of quantum decision process that we define in this chapter builds on the notion of a \textit{quantum partially observable Markov decision process} as defined in Ref.~\cite{BBA14} (see also Refs.~\cite{Cid16,YY18}). The name ``quantum partially observable Markov decision process'' is based on classical (partially observable) Markov decision processes \cite{KLC98,Put14_book}, which provide a simple mathematical model for agent-environment interactions and have been widely studied \cite{Bell54,Bell57,Bell57_book,Blackwell62,Aas65,WL95,Maha96,KLM96,KLC98,CYZ01}. Markov decision processes also provide the mathematical foundation for reinforcement learning \cite{Sut18_book} and artificial intelligence \cite{RN09_book}.
	
	A classical Markov decision process is depicted in the left panel of Figure~\ref{fig-MDP}. At each time step, the environment is assumed to be in one of a discrete set $\mathcal{X}$ of states, and the agent can perform an action from a discrete set $\mathcal{A}$ of actions. Starting at time $t=1$, the agent perceives a state $x_1$ of the environment. Based on this knowledge, the agent picks an action $a_1\in\mathcal{A}$, and sends it to the environment. The environment then transitions to the state $x_2\in\mathcal{X}$, and this transition is governed by a function $T_1(x_1,a_1,x_2)$, which gives the probability that the environment transitions to the state $x_2\in\mathcal{X}$ at time $t=2$ given that the state at time $t=1$ was $x_1\in\mathcal{X}$ and the action taken at that time was $a_1\in\mathcal{A}$. The agent also receives a reward $R(1)(x_1,a_1,x_2)$, which depends on the current state of the environment, the action taken, and the state that the environment transitions to. This process continues either indefinitely, until a given amount of time has elapsed, or until some stopping criterion is satisfied. The agent's goal is to devise a \textit{policy}, i.e., a mapping from states to actions, such that its expected reward is maximized. Markov decision problems can be seen as a special case of dynamic programming/control theory/scheduling problems in which the set of actions is discrete instead of continuous. Also, the term ``schedule'' is sometimes used instead of ``policy''.
	
	\begin{figure}
		\centering
		\includegraphics[width=0.46\textwidth]{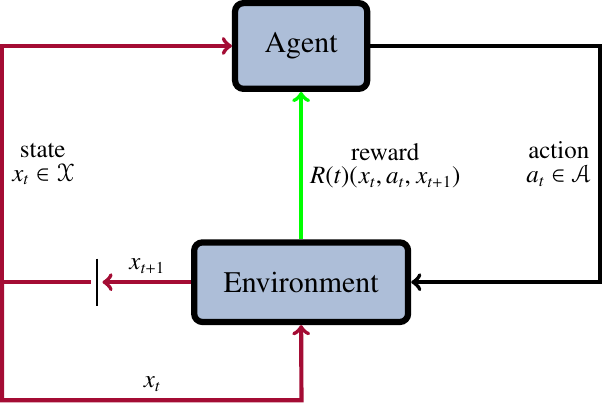}\qquad
		\includegraphics[width=0.46\textwidth]{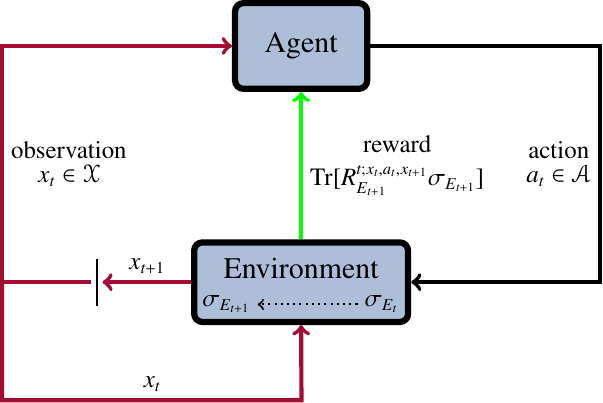}
		\caption{Schematic diagrams of classical (left) and quantum partially observable (right) Markov decision processes. See Refs.~\cite{KLC98,Put14_book} for details on classical Markov decision processes. The quantum generalization shown on the right has been defined previously in Ref.~\cite{BBA14} (see also Refs.~\cite{Cid16,YY18}), and it involves making the environment a quantum system and defining the reward as the expectation value of a Hermitian operator $R^{t;x_t,a_t,x_{t+1}}$. Note that the agent is still classical, in the sense that it sends the classical values in $\mathcal{A}$ to the environment.}\label{fig-MDP}
	\end{figure}
	
	A straightforward quantum generalization of a Markov decision process is shown in the right panel of Figure~\ref{fig-MDP}, and it is essentially a quantum partially observable Markov decision process as defined in Ref.~\cite{BBA14}. (See Refs.~\cite{Cid16,YY18} for similar quantum generalizations.) The environment is now a quantum system, which we call $E$. The state of the environment at time $t\geq 1$ is given by a density operator $\sigma_{E_t}$. What the agent receives at each time step is not this quantum state but some \textit{observations} pertaining to the environment's quantum state, which is classical data that contains some desired classical property of the environment and is captured by the value $x_t\in\mathcal{X}$. The agent is still classical, and it sends an action $a_t\in\mathcal{A}$ to the environment. Based on the agent's action, the environment transitions to a new quantum state $\sigma_{E_{t+1}}$, while outputting a new observation $x_{t+1}\in\mathcal{X}$. The agent's reward is given by $\Tr[R_{E_{t+1}}^{t;x_t,a_t,x_{t+1}}\sigma_{E_{t+1}}]$, where $R_{E_{t+1}}^{t;x_t,a_t,x_{t+1}}$ is a Hermitian operator that depends on $x_t,a_t,x_{t+1}$. As before, the agent's goal is to determine an optimal policy, i.e., a mapping from observations to actions such that its expected reward is maximized. This quantum generalization of a Markov decision process is called ``partially observable'' because the agent never receives full information about the quantum state of the environment. Strictly speaking, therefore, quantum partially observable Markov decision processes should be thought of as a quantum generalization of classical partially observable Markov decision processes as defined in Refs.~\cite{Aas65,KLC98}.
	
	\begin{figure}
		\centering
		\includegraphics[width=0.5\textwidth]{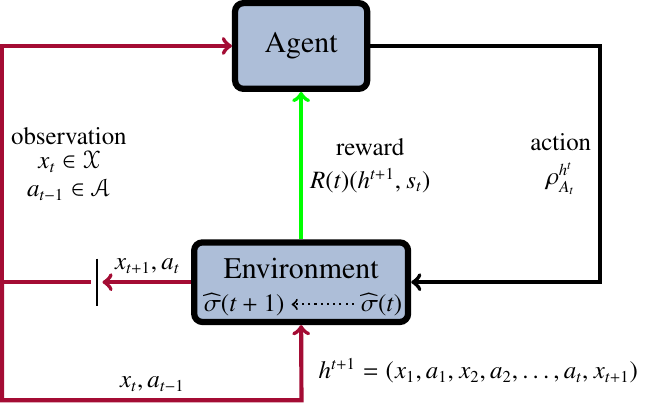}
		\caption{Schematic diagram of a quantum decision process as we define it in Definition~\ref{def-QDP}. They key difference between this formulation and the one in the right panel of Figure~\ref{fig-MDP} is that the agent is allowed to possess a quantum system, so that its action can correspond to a quantum state rather that just a classical value from the set $\mathcal{A}$. This so-called \textit{action state} is measured by the environment, and its outcome is sent to the agent as part of the observation. We also allow the reward to take a more general form, and we keep track of the full history of the agent-environment interaction via the classical quantum states $\widehat{\sigma}(t)$.}\label{fig-QDP_0}
	\end{figure}
	
	In this chapter, we define a slightly more general notion of a quantum partially observable Markov decision process, and we call it simply a \textit{quantum decision process}; see Figure~\ref{fig-QDP_0} and Definition~\ref{def-QDP} below. Roughly speaking, instead of sending one out of a set of possible (classical) actions to the environment at each time step, in a quantum decision process we allow the agent to send a quantum state, which we call an \textit{action state}, to the environment. The agent's observation of the environment remains classical, but it consists now of the pair $(x_{t+1},a_t)$, where $x_{t+1}\in\mathcal{X}$ as before, but the value $a_t\in\mathcal{A}$ results from the environment's measurement of the action state sent by the agent. In this case, a policy is a sequence of action states, and we focus specifically on the problem of obtaining an optimal policy with respect to the expected reward criterion. Unlike classical Markov decision processes, our results in Section~\ref{sec-opt_pol_analytical} tell us that the optimal policy cannot necessarily be described in terms of a classical probabilistic decision among the different possible actions. Instead, we find that, under certain conditions, the optimal decision at each time can involve a \textit{superposition} of the different possible actions. Under these conditions, an agent equipped with a quantum system is necessary to achieve the optimal expected reward. A purely classical agent, therefore, cannot in general achieve the optimal reward in the presence of a quantum environment.
	
	Just as classical Markov decision processes provide the mathematical framework for reinforcement learning, so too do quantum decision processes provide the mathematical framework for \textit{quantum reinforcement learning}. A general framework for fully-quantum reinforcement learning has been laid out in Refs.~\cite{DCLT08,DTB16} (see also Refs.~\cite{DTB15,DB17,DTB17}). In particular, in Ref.~\cite{DTB16}, the authors provide general conditions under which a quantum-enhanced learning agent that has access to a quantum-oracular variant of a classical environment can outperform a classical agent acting in the same (classical) environment, and these results have been used to demonstrate exponential speed-ups in learning efficiency for particular classes of Markov decision processes \cite{DLWT17} (see also Ref.~\cite{PDM+14}). Other examples of quantum reinforcement learning have been considered in Refs.~\cite{L17,CLRS17,SWG+18}. These prior works deal primarily with quantum strategies for solving \textit{classical} reinforcement problems. The developments of this chapter provide the theoretical framework for reinforcement learning in a purely quantum setting. We discuss this point further in Appendix~\ref{sec-future_work}. Specifically, in the context of this thesis, the developments of this chapter provide the tools needed in order to perform reinforcement learning of practical quantum network protocols.

\section{Definition and basic properties}\label{sec-QDP}

	Let us start by first formally defining a classical Markov decision process, as depicted in the left panel of Figure~\ref{fig-MDP}. We follow the definition presented in Ref.~\cite[Chapter~2]{Put14_book}. A Markov decision process is a sequence of interactions between an agent and its environment that is defined by the following elements.
	\begin{itemize}
		\item A finite set $\mathcal{X}$ of states of the environment, with associated random variables $X(t)$ for all $t\geq 1$ whose values are contained in $\mathcal{X}$. We also have a finite set $\mathcal{A}$ of actions of the agent, with associated random variables $A(t)$ for all $t\geq 1$ whose values are contained in $\mathcal{A}$. The sequence
			\begin{equation}
				H(t)\coloneqq (X(1),A(1),X(2),A(2),\dotsc,A(t-1),X(t))
			\end{equation}
			of state and action random variables tells us the \textit{history} of the agent-environment interaction up to some time $t\geq 1$. Every realization of the history is a sequence of the form
			\begin{equation}\label{eq-link_history}
				h^t\coloneqq (x_1,a_1,x_2,a_2,\dotsc,a_{t-1},x_t),
			\end{equation}
			where $x_j\in\mathcal{X}$ and $a_j\in\mathcal{A}$. Given a history $h^t$ of the form shown above, we let
			\begin{equation}
				h^t_j\coloneqq (x_1,a_1,x_2,a_2,\dotsc,a_{j-1},x_j)
			\end{equation}
			denote the history up to time $j\geq 2$. For $j=1$, we let $h^t_1=x_1$. Then, we can regard the state and action random variables as functions such that, for every history $h^t$ as in \eqref{eq-link_history},
			\begin{equation}
				X(j)(h^t)=x_j,\quad A(j)(h^t)=a_j
			\end{equation}
			for all $1\leq j\leq t$. We let $\Omega(t)\coloneqq\{\mathcal{X}\times\mathcal{A}\}^{t-1}\times\mathcal{X}$ denote the set of all histories up to time $t\geq 1$.
			
		\item Transition functions $T_t:\mathcal{X}\times\mathcal{A}\times\mathcal{X}\to[0,1]$ for all $t\geq 1$ such that $T_t(x_t,a_t,x_{t+1})=\Pr[X(t+1)=x_{t+1}|X(t)=x_t,A(t)=a_t]$. In other words, the transition function gives us the probability that, at time $t$, the environment transitions to a particular state at time $t+1$ given its state at time $t$ and the agent's action at time $t$.
		
		\item Reward functions $R(t):\mathcal{X}\times\mathcal{A}\times\mathcal{X}\to\mathbb{R}$ for all $t\geq 1$ such that $R(t)(x_{t},a_{t},x_{t+1})$ is the reward received by the agent at time $t$ based on the state $x_{t}$ of the environment at time $t$, the agent's action $a_{t}$ at time $t$, and the new state $x_{t+1}$ of the environment at time $t+1$ based on the agent's action.
		
		\item Decision functions $d_t:\Omega(t)\times\mathcal{A}\to[0,1]$ for all $t\geq 1$ such that
			\begin{equation}\label{eq-MDP_decision_functions}
				d_t(h^t)(a_t)\coloneqq\Pr[A(t)=a_t|H(t)=h^t].
			\end{equation}
			In other words, the decision function gives us the probability that, at time $t$, the agent takes the action $a_t$ conditioned on the history $h^t$ of the interaction up to time $t$. The sequence $\pi\coloneqq (d_1,d_2,\dotsc)$ is called a \textit{policy} for the agent, and it tells us how action decisions are made at each time step.
	\end{itemize}
	The agent's goal is to perform actions that maximize its long-term reward. Specifically, in the \textit{finite-horizon setting}, the agent's goal is to maximize the expected value of the quantity $\sum_{t=1}^T R(t)$ up to a given amount $T<\infty$ of time, called the \textit{horizon time}. In the \textit{infinite-horizon setting}, the agent's goal is to maximize the expected value of the quantity $\sum_{t=1}^{\infty}\gamma^{t-1} R(t)$, where $\gamma\in(0,1]$ is a \textit{discount factor}. A thorough introduction to classical Markov decision processes can be found in Refs.~\cite{KLC98,Put14_book}. Note that what makes a classical Markov decision process Markovian is the fact that the transition function and the reward function at each time step depend only on the state and action of the previous time step. However, the decision functions can in general depend on the entire history of the interaction, even in a Markov decision process.
	
	By the basic rules of probability, the probability of every history $h^t$ is given by
	\begin{align}
		\Pr[H(t)=h^t]&=\Pr[X(t)=x_t|H(t-1)=h_{t-1}^t,A(t-1)=a_{t-1}]\cdot\nonumber\\
		&\qquad\qquad \Pr[A(t-1)=a_{t-1}|H(t-1)=h_{t-1}^t]\cdot\Pr[H(t-1)=h_{t-1}^t]\\
		&=\Pr[X(1)=x_1]\prod_{j=2}^t \left(\Pr[X(j)=x_j|H(j-1)=h_{j-1}^t,A(j-1)=a_{j-1}]\right.\cdot\nonumber\\
		&\qquad\qquad\qquad\qquad\left.\Pr[A(j-1)=a_{j-1}|H(j-1)=h_{j-1}^t]\right)\label{eq-decision_process_prob}\\
		&=\Pr[X(1)=x_1]\prod_{j=2}^t\left(T_{j-1}(x_{j-1},a_{j-1},x_j)\cdot d_{j-1}(h_{j-1}^t)(a_{j-1})\right).
	\end{align}

	Having recalled the definition of a classical Markov decision process, let us now turn to the quantum case.
	%[system $E$ does not have to be a finite-dimensional quantum system...]
	
	\begin{definition}[Quantum Decision Process]\label{def-QDP}
		A \textit{quantum decision process (QDP)} is defined as a pair $\mathsf{Q}=(\mathsf{E},\mathsf{A})$, where
		\begin{equation}
			\mathsf{E}=\left(E,\sigma_{E_0},\mathcal{X},\mathcal{A},\{\mathcal{E}^0\}\cup\{\mathcal{E}^t\}_{t\geq 1},\{\mathcal{R}^t\}_{t\geq 1},\{R(t)\}_{t\geq 1}\right)
		\end{equation}
		is the \textit{environment} of the QDP, and it consists of the following elements.
		\begin{enumerate}
			\item A quantum system $E$. We let $E_t$ denote the quantum system at time $t\geq 0$. The state of the environment at time $t=0$ is $\sigma_{E_0}$.
			
			\item A finite set $\mathcal{X}$ of \textit{observations} (or classical states) of the quantum system, which correspond to some classical property of the quantum system. To the set $\mathcal{X}$ there corresponds a Hilbert space $\mathcal{H}_{\mathcal{X}}\coloneqq\text{span}\{\ket{x}:x\in\mathcal{X}\}$ defined by the orthonormal basis $\{\ket{x}\}_{x\in\mathcal{X}}$. We denote the corresponding classical register storing the observations at time $t\geq 1$ by $X_t$.
			
			\item A finite set $\mathcal{A}$ of \textit{actions} that can be performed by the agent. To this set corresponds a Hilbert space $\mathcal{H}_{\mathcal{A}}\coloneqq\text{span}\{\ket{a}:a\in\mathcal{A}\}$ defined by the orthonormal basis $\{\ket{a}\}_{a\in\mathcal{A}}$. We associate a quantum system $A_t$ to this Hilbert space and a classical register $\overline{A}_t$ for all $t\geq 1$.
			
				For all $t\geq 1$, a \textit{history} of the QDP is given by
				\begin{equation}
					h^t\coloneqq(x_1,a_1,x_2,a_2,\dotsc,a_{t-1},x_t)\in\{\mathcal{X}\times\mathcal{A}\}^{t-1}\times\mathcal{X}\eqqcolon\Omega(t),
				\end{equation}
				where $\Omega(t)$ is the set of all histories. We let $h_1^t\coloneqq x_1$, and
				\begin{equation}\label{eq-QDP_histories}
					h_j^t\coloneqq (x_1,a_1,x_2,a_2,\dotsc,a_{j-1},x_j),\quad 2\leq j\leq t-1
				\end{equation}
				denote ``slices'' of the history up to time $j$. Note that $h_j^t\in\Omega(j)$ for all $1\leq j\leq t-1$.

				For all $t\geq 1$, we define the Hilbert space
				\begin{equation}
					\mathcal{H}_{\Omega(t)}\coloneqq (\mathcal{H}_{\mathcal{X}}\otimes\mathcal{H}_{\mathcal{A}})^{\otimes t-1}\otimes\mathcal{H}_{\mathcal{X}},
				\end{equation}
				and an associated classical register $H_t$ corresponding to the history given by $H_t=(X_1,\overline{A}_1,\allowbreak\dotsc,\overline{A}_{t-1},X_t)$. We have that $\mathcal{H}_{\Omega(t)}=\text{span}\{\ket{h^t}:h^t\in\Omega(t)\}$, where
				\begin{equation}
					\ket{h^t}_{H_t}=\ket{x_1,a_1,\dotsc,a_{t-1},x_t}_{H_t}\equiv\ket{x_1}_{X_1}\otimes\ket{a_1}_{\overline{A}_1}\otimes\dotsb\otimes\ket{a_{t-1}}_{\overline{A}_{t-1}}\otimes\ket{x_t}_{X_t}
				\end{equation}
				for all $h^t\in\Omega(t)$, with $x_j\in\mathcal{X}$ for all $1\leq j\leq t$ and $a_j\in\mathcal{A}$ for all $1\leq j\leq t-1$.
				
			\item A set $\{\mathcal{E}^0\}\cup\{\mathcal{E}^t\}_{t\geq 1}$ of \textit{(environment) response channels}. The channel $\mathcal{E}_{E_0\to H_1E_1}^0$ is a quantum instrument channel defined by
				\begin{equation}
					\mathcal{E}_{E_0\to H_1E_1}^0(\sigma_{E_0})\coloneqq\sum_{x_1\in\mathcal{X}}\ket{x_1}\bra{x_1}_{H_1}\otimes\mathcal{T}_{E_0\to E_1}^{0;x_1}(\sigma_{E_0}),
				\end{equation}
				with $\left\{\mathcal{T}_{E_0\to E_1}^{0:x_1}\right\}_{x_1\in\mathcal{X}}$ being a quantum instrument, which we recall from Definition~\ref{def-quantum_instrument} is a set of completely positive trace non-increasing transition maps such that the sum $\sum_{x_1\in\mathcal{X}}\mathcal{T}_{E_0\to E_1}^{0;x_1}$ is a trace-preserving map. For $t\geq 1$, we have
				\begin{multline}\label{eq-env_response_chan}
					\mathcal{E}_{H_tA_tE_t\to H_{t+1}E_{t+1}}^t(\ket{h^t}\bra{h^t}_{H_t}\otimes\rho_{A_t}\otimes\sigma_{E_t})\\\coloneqq\sum_{a_t\in\mathcal{A}}\sum_{x_{t+1}\in\mathcal{S}}\Tr[M_{A_t}^{a_t}\rho_{A_t}]\ket{h^t}\bra{h^t}_{H_t}\otimes\ket{a_t,x_{t+1}}\bra{a_t,x_{t+1}}_{\overline{A}_tX_{t+1}}\otimes\mathcal{T}_{E_t\to E_{t+1}}^{t;h^t,a_t,x_{t+1}}(\sigma_{E_t}),
				\end{multline}
				where $\left\{\mathcal{T}_{E_t\to E_{t+1}}^{t;h^t,a_t,x_{t+1}}\right\}_{t\geq 1}$ is a set of \textit{transition maps} and $\{M_{A_t}^{t;a_t}:a_t\in\mathcal{A}\}$ is a POVM for every $t\geq 1$, i.e., $M_{A_t}^{t;a_t}\geq 0$ for all $a_t\in\mathcal{A}$ and $\sum_{a_t}M_{A_t}^{t;a_t}=\mathbbm{1}_{A_t}$.
				
				The transition maps $\left\{\mathcal{T}_{E_t\to E_{t+1}}^{t;h^t,a_t,x_{t+1}}\right\}_{x_{t+1}\in\mathcal{X}}$ constitute a quantum instrument for all $t\geq 1$, $h^t\in\Omega(t)$, and $a_t\in\mathcal{A}$.
				
			\item A set $\{\mathcal{R}^t\}_{t\geq 1}$ of \textit{reward channels}, which are defined as
				\begin{multline}\label{eq-QDP_reward_channel}
					\mathcal{R}_{H_{t+1}E_{t+1}\to H_{t+1}\overline{R}_tE_{t+1}}^t(\ket{h^{t+1}}\bra{h^{t+1}}_{H_{t+1}}\otimes\sigma_{E_{t+1}})\\=\ket{h^{t+1}}\bra{h^{t+1}}_{H_{t+1}}\otimes\sum_{s_t\in\mathcal{Y}_t}\ket{s_t}\bra{s_t}_{\overline{R}_t}\otimes\mathcal{R}^{t;h^{t+1},s_t}(\sigma_{E_{t+1}}),
				\end{multline}
				where $\mathcal{Y}_t$ is a finite set and $\left\{\mathcal{R}^{t;h^{t+1},s_t}\right\}_{s_t\in\mathcal{Y}_t}$ is a quantum instrument of \textit{reward maps}.
				
				Associated to the reward channels are functions (random variables) $R(t):\Omega(t+1)\times\mathcal{Y}_t\to\mathbb{R}$ such that $R(t)(h^{t+1},s_t)$ is the value of the reward at time $t$.
		\end{enumerate}
		
		\noindent Interacting with the environment is an \textit{agent}, which is defined as
		\begin{equation}
			\mathsf{A}\coloneqq\left(T,\pi_T\right),
		\end{equation}
		where the elements are defined as follows.
		\begin{itemize}
			\item $T\geq 1$ is the \textit{horizon time}. If $T<\infty$, then we refer to $\mathsf{Q}$ as a \textit{finite-horizon} QDP; otherwise, if $T=\infty$, then $\mathsf{Q}$ is called an \textit{infinite-horizon} QDP.
			
			\item $\pi_T\coloneqq(\rho_{A_t}^{h_t^T})_{t=1}^T$ is a \textit{$T$-step policy}, which is a sequence of \textit{action states} $\rho_{A_t}^{h_t^T}$. The action states are in one-to-one correspondence with \textit{decision channels} $\mathcal{D}_{H_t\to H_tA_t}^t$, which are defined as
				\begin{equation}
					\mathcal{D}_{H_t\to H_tA_t}^t(\ket{h^t}\bra{h^t}_{H_t})=\ket{h^t}\bra{h_t}_{H_t}\otimes\rho_{A_t}^{h^t}
				\end{equation}
				for all $t\geq 1$ and all histories $h^t\in\Omega(t)$.
		\end{itemize}
		
		\noindent The \textit{$t$-step QDP channel} is defined as
		\begin{equation}\label{eq-QDP_channel}
			\mathcal{P}_{E_0\to H_tE_t}^{(\mathsf{E},\mathsf{A});t}\coloneqq \mathcal{E}_{H_{t-1}A_{t-1}E_{t-1}\to H_tE_t}^{t-1}\circ\mathcal{D}_{H_{t-1}\to H_{t-1}A_{t-1}}^{t-1}\circ\dotsb\circ\mathcal{E}_{H_1A_1E_1\to H_2E_2}^1\circ\mathcal{D}_{H_1\to H_1A_1}^1\circ\mathcal{E}_{E_0\to H_1E_1}^0
		\end{equation}
		for all $t\geq 1$, and the \textit{classical-quantum state of the QDP} is defined as
		\begin{equation}\label{eq-QDP_cq_state_0}
			\widehat{\sigma}_{H_t E_t}^{(\mathsf{E},\mathsf{A})}(t)=\mathcal{P}_{E_0\to H_tE_t}^{(\mathsf{E},\mathsf{A});t}(\sigma_{E_0})
		\end{equation}
		for all $t\geq 1$.~\defqed
	\end{definition}
	
	\begin{remark}\label{rem-QDP_def}
		We make several remarks about our definition of a quantum decision process.
		\begin{itemize}
			\item The environment quantum system $E$, although specified ostensibly as a single quantum system, is allowed to be a multipartite quantum system. In particular, then, the transition maps and the reward channels can be defined such that they act either independently or jointly only on some subset of the subsystems comprising $E$, which can be used to model the situation in which the agent might have access only to a part of the environment.
	
			\item As stated earlier, our definition of a quantum decision process differs slightly from the definitions of quantum partially observable Markov decision processes provided in Refs.~\cite{BBA14,Cid16,YY18} due to the fact that the agent, in our definition, is equipped with a quantum system and can send decisions to the environment that are encoded in a quantum state. In other words, it is possible, for example, for the agent to send a \textit{superposition} of different action basis elements to the environment. Note that this goes beyond classical probabilistic policies in general (as we show explicitly later), because the measurements involved in the environment response channels need not be given by projective measurements onto the action basis elements, but instead could be described by a more general POVM. Furthermore, in our definition both the transition channels and the rewards can depend on the entire history of the interaction.
		
				On the other hand, our definition of a quantum decision process falls into the general framework of agent-environment interactions developed in Ref.~\cite{DTB16} (see also Ref.~\cite{DTB15}). These works consider in detail the interaction of a classical or quantum agent in a classical or quantum environment. In this framework, both the agent and environment are given by sequences of quantum channels. Specifically, it is pointed out that every agent-environment interaction is an example of a quantum strategy/quantum comb/quantum causal network \cite{CDP08,CDP09,GW07} (see also Ref.~\cite{VW16}). Furthermore, the concept of a history of the interaction is captured by a so-called ``tester'', which is a controlled unitary acting on the interface between the agent and the environment. Our definition of a quantum decision process does indeed fall into this general framework, because the agent and environment are specified by sequences of quantum channels, and the classical history is encoded in the individual completely positive maps constituting the quantum instruments of the environment. Also, because the action values of the history are given by measurements of the quantum systems $A_t$ that contain the action states, according to the classification in Ref.~\cite{DTB15}, the agent-environment interaction in a quantum decision process is purely classical.
				
				Furthermore, because the agent-environment interaction in a quantum decision process is simply a sequence of quantum channels, based on the observation in Ref.~\cite{DTB16} we can view every quantum decision process as a quantum causal network, and we illustrate this in Figure~\ref{fig-QDP}.
				
	\begin{figure}
		\centering
		\includegraphics[scale=1.25]{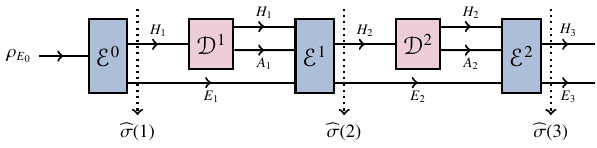}
		\caption{A quantum decision process (see Definition~\ref{def-QDP}) as a quantum causal network. The channels in red correspond to the actions of the agent, and the channels in blue correspond to the transitions of the environment. The classical registers $H_j$ contain the history of the interaction up to time $j$, the quantum registers $A_j$ contain the action of the agent at time $j$, and the quantum registers $E_j$ contain the state of the environment at time $j$. Shown is the agent-environment interaction up to time $t=3$.}\label{fig-QDP}
	\end{figure}
	
			\item Our notion of a quantum decision process essentially corresponds to a communication scenario between the agent and the environment in which the agent has the ability to send a quantum system to the environment through an ideal, noiseless quantum channel, but the environment provides only classical feedback to the agent. Furthermore, we assume that the agent does not make use of a quantum memory in order to determine future actions, only the classical memory that stores the history of observations from the environment.
		
			\item Important special cases of the reward channel are when they are given by POVMs and by Hermitian operators. We deal exclusively with the latter case in this thesis, so let us look at it more closely. Let $\{R^{t;h^{t+1}}_{E_{t+1}}:t\geq 1,\,h^{t+1}\in\Omega(t+1)\}$ be a set of Hermitian operators such that
				\begin{equation}
					R_{E_{t+1}}^{t;h^{t+1}}=\sum_{s_t\in\mathcal{Y}_t}\lambda_{s_t}^{t;h^{t+1}}\Pi_{E_{t+1}}^{t;h^{t+1},s_t}
				\end{equation}
				is a spectral decomposition of $R_{E_{t+1}}^{t;h^{t+1}}$. In this case, we define the completely positive maps $\mathcal{R}_{E_{t+1}}^{t;h^{t+1},s_t}$ in \eqref{eq-QDP_reward_channel} as
				\begin{equation}
					\mathcal{R}_{E_{t+1}}^{t;h^{t+1},s_t}(\cdot)\coloneqq \Pi_{E_{t+1}}^{t;h^{t+1},s_t}(\cdot)\Pi_{E_{t+1}}^{t;h^{t+1},s_t},
				\end{equation}
				and we define the functions $R(t)$ for $t\geq 1$ as
				\begin{equation}\label{eq-QDP_reward_Herm}
					R(t)(h^{t+1},s_t)\coloneqq\lambda_{s_t}^{t;h^{t+1}}.
				\end{equation}
				In other words, the reward values are simply the eigenvalues of the corresponding Hermitian operator.
				
				In addition to rewards being defined via Hermitian operators, another special case that arises later is one in which the reward should be non-zero if and only if the classical state at time $t+1$ is equal to a particular value $x_{t+1}\in\mathcal{X}$. In this case, the definition in \eqref{eq-QDP_reward_Herm} gets modified to
				\begin{equation}
					R(t)(h^{t+1},s_t)=\delta_{X(t+1)(h^{t+1}),x_{t+1}}\lambda_{s_t}^{t;h^{t+1}}.
				\end{equation}
			
			\item In our definition of the environment of a quantum decision process, the transition and reward maps can depend both on the entire history of the interaction as well as on time. The POVMs defining the environment response channels can also be time dependent. There are two special and simplified settings that are widely considered.
				\begin{itemize}
					\item\textit{Markovian setting}: In this case, both the transition and reward maps depend only on the previous classical state of the environment and on the previous action taken, i.e.,
						\begin{align}
							\mathcal{T}_{E_{t}\to E_{t+1}}^{t;h^{t},a_{t},s_{t+1}}&=\mathcal{T}_{E_{t}\to E_{t+1}}^{t;x_{t},a_{t},s_{t+1}}\quad\forall~t\geq 1,\\
							\mathcal{R}_{E_{t+1}}^{t;h^{t+1},s_t}&=\mathcal{R}_{E_{t+1}}^{t;x_{t},a_{t},x_{t+1},s_t}\quad\forall~t\geq 1.
						\end{align}
						
					\item\textit{Stationary setting}: The transition and reward maps, and the POVMs $\{M_{A_t}^{t;a_t}:a_t\in\mathcal{A},t\geq 1\}$ are time independent and Markovian, meaning that
						\begin{align}
							\mathcal{T}_{E_{t}\to E_{t+1}}^{t;h^{t},a_{t},x_{t+1}}&=\mathcal{T}_{E_{t}\to E_{t+1}}^{x_{t},a_{t},x_{t+1}}\quad\forall~t\geq 1,\\
							\mathcal{R}_{E_{t+1}}^{t;h^{t+1},s_t}&=\mathcal{R}_{E_{t+1}}^{x_{t},a_{t},x_{t+1},s_t}\quad\forall ~t\geq 1,\\
							M_{A_t}^{t;a_t}&=M_{A_t}^{a_t}\quad\forall~t\geq 1.
						\end{align}
				\end{itemize}
				Note that even when the transition and reward maps are Markovian, the evolution of the environment can still be non-Markovian in general because the quantum comb defined by the quantum channels $(\mathcal{E}^0,\mathcal{E}^1,\mathcal{E}^2,\dotsc)$ (see Figure~\ref{fig-QDP}) contains the memory systems $E_t$.
				
			\item A special case of the POVMs $\{M_{A_t}^{t;a_t}\}_{a_t\in\mathcal{A}}$ in the definition of the environment response channel is when they are equal simply to rank-one projections onto the action basis elements, i.e., $M_{A_t}^{t;a_t}=\ket{a_t}\bra{a_t}_{A_t}$ for all $t\geq 1$ and all $a_t\in\mathcal{A}$. In this case, it suffices to consider classical action states, i.e., action states that are diagonal in the action basis. Specifically, without loss of generality, every action state $\rho_{A_t}^{h^t}$ can be taken to be of the form
				\begin{equation}
					\rho_{A_t}^{h^t}=\sum_{a_t\in\mathcal{A}}d_t(h^t)(a_t)\ket{a_t}\bra{a_t}_{A_t},
				\end{equation}
				where we recall the decision functions $d_t$ that we defined in \eqref{eq-MDP_decision_functions} for classical Markov decision processes. This scenario then corresponds to a classical decision process in which the transition probabilities are non-Markovian. In particular,
				\begin{equation}\label{eq-QDP_transition_prob_gen}
					\Pr[X(t+1)=x_{t+1}|H(t)=h^t,A(t)=a_t]=\Tr\left[\mathcal{T}_{E_t\to E_{t+1}}^{t;h^t,a_t,x_{t+1}}(\sigma_{E_t}(t|h^t))\right],
				\end{equation}
				where the state $\sigma_{E_t}(t|h^t)$ is defined in \eqref{eq-QDP_cond_states} below.~\defqed
		\end{itemize}
	\end{remark}
	\bigskip
	\begin{proposition}\label{prop-QDP_cq_state}
		Let $\mathsf{Q}=(\mathsf{E},\mathsf{A})$ be a quantum decision process. The classical-quantum state of $\mathsf{Q}$, as defined in \eqref{eq-QDP_cq_state_0}, is given by
		\begin{equation}\label{eq-QDP_cq_state}
			\widehat{\sigma}^{(\mathsf{E},\mathsf{A})}_{H_tE_t}(t)\coloneqq\sum_{h^t\in\Omega(t)}\ket{h^t}\bra{h^t}_{H_t}\otimes\widetilde{\sigma}_{E_t}^{(\mathsf{E},\mathsf{A})}(t;h^t)
		\end{equation}
		for all $t\geq 1$, where
		\begin{align}
			\widetilde{\sigma}_{E_t}^{(\mathsf{E},\mathsf{A})}(t;h^t)&\coloneqq\left(\prod_{j=1}^{t-1}\Tr\left[M_{j;A_j}^{a_j}\rho_{A_j}^{h_j^t}\right]\right)\left(\mathcal{T}_{E_{t-1}\to E_t}^{t-1;h_{t-1}^t,a_{t-1},x_t}\circ\dotsb\circ\mathcal{T}_{E_2\to E_3}^{2;h_2^t,a_2,x_3}\circ\mathcal{T}_{E_1\to E_2}^{1;x_1,a_1,x_2}\right)(\widetilde{\sigma}_{E_1}(1;x_1))\label{eq-QDP_cond_state}\\
			&=\left(\prod_{j=1}^{t-1}\Tr\left[M_{A_j}^{j;a_j}\rho_{A_j}^{h_j^t}\right]\right)\sigma_{E_t}^{(\mathsf{E})}(t;h^t),
		\end{align}
		and
		\begin{align}
			\widetilde{\sigma}_{E_t}^{(\mathsf{E})}(t;h^t)&\coloneqq\left(\mathcal{T}_{E_{t-1}\to E_t}^{t-1;h_{t-1}^t,a_{t-1},x_t}\circ\dotsb\circ\mathcal{T}_{E_2\to E_3}^{2;h_2^t,a_2,x_3}\circ\mathcal{T}_{E_1\to E_2}^{1;x_1,a_1,x_2}\right)(\widetilde{\sigma}_{E_1}(1;x_1))\label{eq-QDP_cond_state_alt}\\
			&=\mathcal{P}_{E_0\to E_t}^{(\mathsf{E}),t;h^t}(\sigma_{E_0}),
		\end{align}
		where
		\begin{equation}
			\mathcal{P}_{E_0\to E_t}^{(\mathsf{E}),t;h^t}\coloneqq\mathcal{T}_{E_{t-1}\to E_t}^{t-1;h_{t-1}^t,a_{t-1},x_t}\circ\dotsb\circ\mathcal{T}_{E_2\to E_3}^{2;h_2^t,a_2,x_3}\circ\mathcal{T}_{E_1\to E_2}^{1;x_1,a_1,x_2}\circ\mathcal{T}_{E_0\to E_1}^{0;x_1}.
		\end{equation}
	\end{proposition}
	\medskip
	\begin{remark}\label{rem-QDP_cond_states}
		Note that the operators $\widetilde{\sigma}_{E_t}^{(\mathsf{E},\mathsf{A})}(t;h^t)$ are unnormalized. The (normalized) conditional states are given by
		\begin{equation}\label{eq-QDP_cond_states}
			\sigma_{E_t}^{(\mathsf{E},\mathsf{A})}(t|h^t)\coloneqq\frac{\widetilde{\sigma}_{E_t}^{(\mathsf{E},\mathsf{A})}(t;h^t)}{\Tr[\widetilde{\sigma}_{E_t}^{(\mathsf{E},\mathsf{A})}(t;h^t)]}.
		\end{equation}
		The joint probability distribution of the history $H_t$ is given by
		\begin{equation}
			\Pr[H(t)=h^t]_{(\mathsf{E},\mathsf{A})}=\Tr[\widetilde{\sigma}_{E_t}^{(\mathsf{E},\mathsf{A})}(t;h^t)]
		\end{equation}
		for all histories $h^t\in\Omega(t)$. Also, the \textit{expected quantum state} of the QDP is defined as
		\begin{equation}\label{eq-QDP_exp_state}
			\sigma_{E_t}^{(\mathsf{E},\mathsf{A})}(t)\coloneqq\Tr_{H_t}\left[\widehat{\sigma}_{H_tE_t}^{(\mathsf{E},\mathsf{A})}(t)\right]=\sum_{h^t\in\Omega(t)}\widetilde{\sigma}_{E_t}^{(\mathsf{E},\mathsf{A})}(t;h^t).~\defqedspec
		\end{equation}
	\end{remark}
	\medskip
	\begin{proof}[Proof of Proposition~\ref{prop-QDP_cq_state}]
		It is relatively simple to see this by making use of Figure~\ref{fig-QDP} and the definitions of the decision channels and environment response channels in Definition~\ref{def-QDP}.
		
		With $\sigma_{E_0}$ being the starting state of the environment, the state $\widehat{\sigma}_{H_1E_1}^{(\mathsf{E},\mathsf{A})}(1)$ at time $t=1$ is
		\begin{equation}
			\widehat{\sigma}_{H_1E_1}^{(\mathsf{E},\mathsf{A})}(1)=\mathcal{E}_{E_0\to H_1E_1}^0(\sigma_{E_0})=\sum_{x_1\in\mathcal{X}}\ket{x_1}\bra{x_1}_{H_1}\otimes\widetilde{\sigma}_{E_1}(1;x_1),
		\end{equation}
		where
		\begin{equation}\label{eq-QDP_cq_state_pf1}
			\widetilde{\sigma}_{E_1}(1;x_1)\coloneqq\mathcal{T}_{E_0\to E_1}^{0;x_1}(\sigma_{E_0}).
		\end{equation}
		This proves the desired result for $t=1$.
		
		Now, suppose that the agent's policy is given by the actions states $\rho_{A_t}^{h^t}$ for all $t\geq 1$, $h^t\in\Omega(t)$. Then, the agent's first decision, given by the action states $\rho_{A_1}^{x_1}$, $x_1\in\mathcal{X}$, and the corresponding channel $\mathcal{D}_{H_1\to H_1A_1}^1$, leads to
		\begin{equation}
			\widehat{\sigma}_{H_1E_1}^{(\mathsf{E},\mathsf{A})}(1)\mapsto \sum_{x_1\in\mathcal{X}}\ket{x_1}\bra{x_1}_{H_1}\otimes\rho_{A_1}^{x_1}\otimes\widetilde{\sigma}_{E_1}(1;x_1).
		\end{equation}
		Then, using \eqref{eq-env_response_chan}, we obtain
		\begin{equation}
			\widehat{\sigma}_{H_2E_2}^{(\mathsf{E},\mathsf{A})}(2)=\sum_{x_1,x_2\in\mathcal{X}}\sum_{a_1\in\mathcal{A}}\ket{x_1,a_1,x_2}\bra{x_1,a_1,x_2}_{H_2}\otimes\Tr\left[M_{A_1}^{1;a_1}\rho_{A_1}^{x_1}\right]\mathcal{T}_{E_1\to E_2}^{1;x_1,a_1,x_2}(\widetilde{\sigma}_{E_1}(1;x_1)).
		\end{equation}
		Using \eqref{eq-QDP_histories} and \eqref{eq-QDP_cq_state_pf1}, we can write this as
		\begin{equation}
			\widehat{\sigma}_{H_2E_2}^{(\mathsf{E},\mathsf{A})}(2)=\sum_{h^2\in\Omega(2)}\ket{h^2}\bra{h^2}_{H_2}\otimes\widetilde{\sigma}_{E_2}(2;h^2),
		\end{equation}
		where
		\begin{align}
			\widetilde{\sigma}_{E_2}^{(\mathsf{E},\mathsf{A})}(2;h^2)&=\Tr\left[M_{A_1}^{1;a_1}\rho_{A_1}^{x_1}\right]\mathcal{T}_{E_1\to E_2}^{1;x_1,a_1,x_2}(\widetilde{\sigma}_{E_1}(1;s_1))\\
			&=\Tr\left[M_{A_1}^{1;a_1}\rho_{A_1}^{x_1}\right]\left(\mathcal{T}_{E_1\to E_2}^{1;x_1,a_1,x_2}\circ\mathcal{T}_{E_0\to E_1}^{0;x_1}\right)(\sigma_{E_0}),
		\end{align}
		which proves the desired result for $t=2$.
		
		The agent's next decision, given by the action states $\rho_{A_2}^{h^2}$, $h^2\in\Omega(2)$, and the corresponding channel $\mathcal{D}_{H_2\to H_2A_2}^2$, then leads to
		\begin{equation}
			\widehat{\sigma}_{H_2E_2}^{(\mathsf{E},\mathsf{A})}(2;h^2)\mapsto\sum_{h^2\in\Omega(2)}\ket{h^2}\bra{h^2}_{H_2}\otimes\rho_{A_2}^{h^2}\otimes\widetilde{\sigma}_{E_2}^{(\mathsf{E},\mathsf{A})}(2;h^2),
		\end{equation}
		so that
		\begin{equation}
			\widehat{\sigma}_{H_3E_3}^{(\mathsf{E},\mathsf{A})}(3)=\sum_{h^3\in\Omega(3)}\ket{h^3}\bra{h^3}_{H_3}\otimes\widetilde{\sigma}_{E_3}^{(\mathsf{E},\mathsf{A})}(3;h^3),
		\end{equation}
		where
		\begin{equation}
			\widetilde{\sigma}_{E_3}^{(\mathsf{E},\mathsf{A})}(3;h^3)\coloneqq\Tr\left[M_{A_2}^{2;a_2}\rho_{A_2}^{h^2}\right]\Tr\left[M_{A_1}^{1;a_1}\rho_{A_1}^{x_1}\right]\left(\mathcal{T}_{E_2\to E_3}^{h^2,a_2,x_3}\circ\mathcal{T}_{E_1\to E_2}^{1;x_1,a_1,x_2}\circ\mathcal{T}_{E_0\to E_1}^{0;x_1}\right)(\sigma_{E_0}),
		\end{equation}
		proving the result for $t=3$. Proceeding in this manner for $t>3$, it is clear that the expression in \eqref{eq-QDP_cq_state} holds.
	\end{proof}
	
	By observing that the classical-quantum state in \eqref{eq-QDP_cq_state} has the form of the classical-quantum state in \eqref{eq-cq_state_LOCC} that is obtained at the output of an arbitrary LOCC protocol, we immediately obtain the following result.
	
	\begin{corollary}\label{lem-QDP_LOCC}
		For all $t\geq 1$, the quantum channel $\mathcal{P}_{E_0\to H_tE_t}^{(\mathsf{E},\mathsf{A}),t}$ defined in \eqref{eq-QDP_channel} is an LOCC quantum instrument channel for all quantum decision processes $\mathsf{Q}=(\mathsf{E},\mathsf{A})$.
	\end{corollary}
	
	\begin{proof}
		We have that
		\begin{align}
			\mathcal{P}_{E_0\to H_tE_t}^{(\mathsf{E},\mathsf{A}),t}(\sigma_{E_0})&=\sum_{h^t\in\Omega(t)}\ket{h^t}\bra{h^t}_{H_t}\otimes\widetilde{\sigma}_{E_t}^{(\mathsf{E},\mathsf{A})}(t;h^t)\\
			&=\sum_{h^t\in\Omega(t)}\ket{h^t}\bra{h^t}_{H_t}\otimes\left(\prod_{j=1}^{t-1}\Tr\left[M_{A_j}^{j;a_j}\rho_{A_j}^{h_j^t}\right]\right)\mathcal{P}_{E_0\to E_t}^{(\mathsf{E}),t;h^t}(\sigma_{E_0}).
		\end{align}
		Now, we can write the factors $\prod_{j=1}^{t-1}\Tr\left[M_{A_j}^{j;a_j}\rho_{A_j}^{h_j^t}\right]$ as follows:
		\begin{equation}
			\prod_{j=1}^{t-1}\Tr\left[M_{A_j}^{j;a_j}\rho_{A_j}^{h_j^t}\right]=\mathcal{M}_{A_{t-1}\to\varnothing}^{t-1;a_{t-1}}(\rho_{A_{t-1}}^{h_{t-1}^t})\circ\dotsb\circ\mathcal{M}_{A_1\to\varnothing}(\rho_{A_1}^{h_1^t}),
		\end{equation}
		where
		\begin{equation}\label{eq-QDP_agent_meas_map}
			\mathcal{M}_{A_j\to\varnothing}^{j;a_j}(\cdot)\coloneqq\Tr_{A_j}[M_{A_j}^{j;a_j}(\cdot)]
		\end{equation}
		is a completely positive map for all $1\leq j\leq t-1$ such that the sum $\sum_{a_j\in\mathcal{A}_j}\mathcal{M}_{A_j\to\varnothing}^{t;a_j}$ is a trace-preserving map. Therefore, letting
		\begin{equation}
			\mathcal{Q}^{(\mathsf{A}),t;h^t}\coloneqq\mathcal{M}_{A_{t-1}\to\varnothing}^{t-1;a_{t-1}}(\rho_{A_{t-1}}^{h_{t-1}^t})\circ\dotsb\circ\mathcal{M}_{A_1\to\varnothing}^{1;a_1}(\rho_{A_1}^{h_1^t}),
		\end{equation}
		we have that
		\begin{equation}
			\mathcal{P}_{E_0\to H_tE_t}^{(\mathsf{E},\mathsf{A});t}(\sigma_{E_0})=\sum_{h^t\in\Omega(t)}\ket{h^t}\bra{h^t}_{H_t}\otimes\left(\mathcal{Q}^{(\mathsf{A}),t;h^t}\otimes\mathcal{P}_{E_0\to E_t}^{(\mathsf{E}),t;h^t}\right)(\sigma_{E_0}),
		\end{equation}
		which has exactly the form of the output state in \eqref{eq-cq_state_LOCC} and \eqref{eq-cq_state_LOCC_2} of a $t$-round LOCC protocol, because the maps $\mathcal{Q}^{(\mathsf{A}),t;h^t}$ and $\mathcal{P}_{E_0\to E_t}^{(\mathsf{E}),t;h^t}$ act only the agent's and environment's systems, respectively. Therefore, $\mathcal{P}_{E_0\to H_tE_t}^{(\mathsf{E},\mathsf{A});t}$ is an LOCC quantum instrument channel.
	\end{proof}

\subsubsection*{The expected reward}

	The figure of merit that we use to evaluate policies and to perform policy optimization is the expected reward at the horizon time $T$. For simplicity, we consider policy optimization only in the finite-horizon setting, so that $T<\infty$ throughout.
	
	Consider an arbitrary QDP $\mathsf{Q}=(\mathsf{E},\mathsf{A})$, with $\mathsf{A}=(T,\pi_T)$. Then, the expected reward at time $T$ is
	\begin{align}
		F(\mathsf{E},(T,\pi_T))&\coloneqq\mathbb{E}[R(T)]_{\pi_T}\label{eq-exp_reward_def}\\
		&=\sum_{h^{T+1}\in\Omega(T+1)}\sum_{s\in\mathcal{Y}_T}R(T)(h^{T+1},s)\Pr\left[H(T+1)=h^{T+1},s\right]_{(\mathsf{E},(T,\pi_T))}\\
		&=\sum_{h^{T+1}\in\Omega(T+1)}\sum_{s\in\mathcal{Y}_T}R(T)(h^{T+1},s) \Tr\left[\mathcal{R}_{E_{T+1}}^{T;h^{T+1},s}\left(\widetilde{\sigma}_{E_{T+1}}^{(\mathsf{E},(T,\pi_T))}(T+1;h^{T+1})\right)\right]\\
		&=\sum_{h^{T+1}\in\Omega(T+1)}\Tr\left[\widetilde{\mathcal{R}}_{E_{T+1}}^{T;h^{T+1}}\left(\widetilde{\sigma}_{E_{T+1}}^{(\mathsf{E},(T,\pi_T))}(T+1;h^{T+1})\right)\right]\label{eq-QDP_exp_reward_formula_1}\\
		&=\Tr\left[\widehat{\mathcal{R}}_{H_{T+1}E_{T+1}\to E_{T+1}}^{(\mathsf{E})}(T)\left(\widehat{\sigma}_{H_{T+1}E_{T+1}}^{(\mathsf{E},(T,\pi_T))}(T+1)\right)\right],\label{eq-QDP_exp_reward_formula_2}
	\end{align}
	where in the second-last line we defined the map
	\begin{equation}\label{eq-QDP_reward_map_tilde}
		\widetilde{\mathcal{R}}_{E_{T+1}}^{T;h^{T+1}}\coloneqq\sum_{s\in\mathcal{Y}_T}R(T)(h^{T+1},s)\mathcal{R}_{E_{T+1}}^{T;h^{T+1},s},
	\end{equation}
	and in the last line we defined the map $\widehat{\mathcal{R}}_{H_{T+1}E_{T+1}\to E_{T+1}}^{(\mathsf{E})}(T)$ given by
	\begin{equation}\label{eq-QDP_reward_map_hat}
		\widehat{\mathcal{R}}_{H_{T+1}E_{T+1}\to E_{T+1}}^{(\mathsf{E})}(T)\left(\ket{h^{T+1}}\bra{h^{T+1}}_{H_{T+1}}\otimes\sigma_{E_{T+1}}\right)\coloneqq\widetilde{\mathcal{R}}_{E_{T+1}}^{T;h^{T+1}}(\sigma_{E_{T+1}})
	\end{equation}
	for all $h^{T+1}\in\Omega(T+1)$ and all linear operators $\sigma_{E_{t+1}}$. A straightforward consequence of this formula, along with \eqref{eq-QDP_cond_state} and \eqref{eq-QDP_cond_state_alt} is the following algorithm for policy evaluation.
	
	\begin{proposition}[Finite-horizon policy evaluation]\label{prop-QDP_policy_eval}
		Let $\mathsf{Q}=(\mathsf{E},\mathsf{A})$ be a quantum decision process, where $\mathsf{A}=(T,\pi_T)$, $T<\infty$, and $\pi_T=(\rho_{A_t}^{h^t}:1\leq t\leq T,\,h^t\in\Omega(t))$. Then,
		\begin{equation}
			F(\mathsf{E},\mathsf{A})=\sum_{\substack{x_1\in\mathcal{X}\\a_1\in\mathcal{A}}}\Tr\left[M_{A_1}^{1;a_1}\rho_{A_1}^{x_1}\right]v_2^{(\mathsf{E},\mathsf{A})}(x_1,a_1),
		\end{equation}
		where
		\begin{equation}\label{eq-QDP_policy_eval_1}
			v_t^{(\mathsf{E},\mathsf{A})}(h^{t-1},a_{t-1})\coloneqq\sum_{\substack{x_t\in\mathcal{X}\\a_t\in\mathcal{A}}}\Tr\left[M_{A_t}^{t;a_t}\rho_{A_t}^{h^{t-1},a_{t-1},x_t}\right]v_{t+1}^{(\mathsf{E},\mathsf{A})}(h^{t-1},a_{t-1},x_t,a_t)
		\end{equation}
		for all $2\leq t\leq T$, $h^{t-1}\in\Omega(t-1)$, and $a_{t-1}\in\mathcal{A}$, and
		\begin{equation}\label{eq-QDP_policy_eval_2}
			v_{T+1}^{(\mathsf{E},\mathsf{A})}(h^T,a_T)\coloneqq\sum_{x_{T+1}\in\mathcal{X}}\Tr\left[\widetilde{\mathcal{R}}_{E_{T+1}}^{T;h^T,a_T,x_{T+1}}\left(\widetilde{\sigma}^{(\mathsf{E})}_{E_{T+1}}(T+1;h^T,a_T,x_{T+1})\right)\right]
		\end{equation}
		for all $h^T\in\Omega(T)$ and $a_T\in\mathcal{A}$.
	\end{proposition}
	
	\begin{remark}	
		The functions $v_t^{(\mathsf{E},\mathsf{A})}$ that we have defined in the statement of the proposition can be thought of as analogous to \textit{action-value functions} in classical Markov decision processes; see, e.g., Refs.~\cite{Put14_book,Sut18_book}. Also, as in the classical case, observe that \eqref{eq-QDP_policy_eval_1} and \eqref{eq-QDP_policy_eval_2} specify a \textit{backward recursion algorithm} for evaluating a given policy. For every history $h^T\in\Omega(T)$ and action $a_T\in\mathcal{A}$, the algorithm proceeds by first evaluating the function $v_{T+1}^{(\mathsf{E},\mathsf{A})}$, then proceeding backwards, calculating $v_t^{(\mathsf{E},\mathsf{A})}$ for all $T\geq t\geq 2$ in order to finally obtain $F(\mathsf{E},\mathsf{A})$.~\defqed
	\end{remark}
	
	\begin{proof}[Proof of Proposition~\ref{prop-QDP_policy_eval}]
		Using \eqref{eq-QDP_cond_state} and \eqref{eq-QDP_cond_state_alt}, we have that
		\begin{align}
			F(\mathsf{E},(T,\pi_T))&=\sum_{h^{T+1}\in\Omega(T+1)}\Tr\left[\widetilde{\mathcal{R}}_{E_{T+1}}^{T;h^{T+1}}\left(\widetilde{\sigma}_{E_{T+1}}^{(\mathsf{E},\mathsf{A})}(T+1;h^{T+1})\right)\right]\\
			&=\sum_{\substack{x_1,\dotsc,x_T,x_{T+1}\in\mathcal{X}\\a_1,\dotsc,a_T\in\mathcal{A}}}\left(\prod_{t=1}^T\Tr\left[M_{A_t}^{t;a_t}\rho_{A_t}^{h^{T+1}_t}\right]\right)\Tr\left[\widetilde{\mathcal{R}}^{T;h^{T+1}}_{E_{T+1}}\left(\widetilde{\sigma}_{E_{T+1}}^{(\mathsf{E})}(T+1;h^{T+1})\right)\right].\label{eq-QDP_policy_eval_pf1}
		\end{align}
		From this, we see that
		\begin{equation}
			F(\mathsf{E},\mathsf{A})=\sum_{\substack{x_1\in\mathcal{X}\\a_1\in\mathcal{A}}}\Tr\left[M_{A_1}^{1;a_1}\rho_{A_1}^{x_1}\right]v_2^{(\mathsf{E},\mathsf{A})}(x_1,a_1),
		\end{equation}
		where
		\begin{equation}
			v_2(T;x_1,a_1)=\sum_{\substack{x_2,\dotsc,x_T,x_{T+1}\in\mathcal{X}\\a_2,\dotsc,a_T\in\mathcal{A}}}\left(\prod_{t=2}^T\Tr\left[M_{A_t}^{t;a_t}\rho_{A_t}^{h^{T+1}_t}\right]\right)\Tr\left[\widetilde{\mathcal{R}}^{T;h^{T+1}}_{E_{T+1}}\left(\widetilde{\sigma}_{E_{T+1}}^{(\mathsf{E})}(T+1;h^{T+1})\right)\right].
		\end{equation}
		Then, separating the sum over $x_2\in\mathcal{X}$ and $a_2\in\mathcal{A}$ in the above equation leads to
		\begin{equation}
			v_2^{(\mathsf{E},\mathsf{A})}(x_1,a_1)=\sum_{\substack{x_2\in\mathcal{X}\\a_2\in\mathcal{A}}}\Tr\left[M_{A_2}^{2;a_2}\rho_{A_2}^{h^{T+1}_2}\right]v_3^{(\mathsf{E},\mathsf{A})}(\underbrace{x_1,a_1,x_2}_{h^{T+1}_2},a_2),
		\end{equation}
		where
		\begin{equation}
			v_3^{(\mathsf{E},\mathsf{A})}(h^{T+1}_2,a_2)=\sum_{\substack{x_3,\dotsc,x_T,x_{T+1}\in\mathcal{X}\\a_3,\dotsc,a_T\in\mathcal{A}}}\left(\prod_{t=3}^T\Tr\left[M_{A_t}^{t;a_t}\rho_{A_t}^{h^{T+1}_t}\right]\right)\Tr\left[\widetilde{\mathcal{R}}^{T;h^{T+1}}_{E_{T+1}}\left(\widetilde{\sigma}_{E_{T+1}}^{(\mathsf{E})}(T+1;h^{T+1})\right)\right].
		\end{equation}
		Proceeding in this manner, we define functions $v_t^{(\mathsf{E},\mathsf{A})}(h^{t-1},a_{t-1})$ for $2\leq t\leq T$, and we observe that
		\begin{equation}
			v_t^{(\mathsf{E},\mathsf{A})}(h^{t-1},a_{t-1})\coloneqq\sum_{\substack{x_t\in\mathcal{X}\\a_t\in\mathcal{A}}}\Tr\left[M_{A_t}^{t;a_t}\rho_{A_t}^{h^{t-1},a_{t-1},x_t}\right]v_{t+1}^{(\mathsf{E},\mathsf{A})}(h^{t-1},a_{t-1},x_t,a_t).
		\end{equation}
		In particular, for $t=T$, we have
		\begin{align}
			v_T^{(\mathsf{E},\mathsf{A})}(h^{T-1},a_{T-1})&=\sum_{\substack{x_T,x_{T+1}\in\mathcal{X}\\a_T\in\mathcal{X}}}\Tr\left[M_{A_T}^{T;a_T}\rho_{A_T}^{h^{T-1},a_{T-1},x_T}\right]\Tr\left[\widetilde{\mathcal{R}}^{T;h^{T+1}}_{E_{T+1}}\left(\widetilde{\sigma}_{E_{T+1}}^{(\mathsf{E})}(T+1;h^{T+1})\right)\right]\\
			&=\sum_{\substack{x_T\in\mathcal{X}\\a_T\in\mathcal{A}}}\Tr\left[M_{A_T}^{T;a_T}\rho_{A_T}^{h^{T-1},a_{T-1},x_T}\right] v_{T+1}^{(\mathsf{E},\mathsf{A})}(h^T,a_T,x_{T+1}),
		\end{align}
		where
		\begin{equation}
			v_{T+1}^{(\mathsf{E},\mathsf{A})}(h^T,a_T)=\sum_{x_{T+1}\in\mathcal{X}}\Tr\left[\widetilde{\mathcal{R}}_{E_{T+1}}^{T;h^T,a_T,x_{T+1}}\left(\widetilde{\sigma}_{E_{T+1}}^{(\mathsf{E})}(h^T,a_T,x_{T+1})\right)\right].
		\end{equation}
		This completes the proof.
	\end{proof}
	
	\begin{remark}\label{rem-QDP_policy_eval_2}
		There is an advantage to using the backward recursion algorithm (as specified by Proposition~\ref{prop-QDP_policy_eval}) to evaluate a policy, rather than simply using the definition of the expected reward in \eqref{eq-QDP_exp_reward_formula_1} or \eqref{eq-QDP_exp_reward_formula_2}. This advantage comes from the fact that the function $v_{T+1}^{(\mathsf{E},\mathsf{A})}$ defined in \eqref{eq-QDP_policy_eval_2} is independent of the policy---it depends only on the elements of the environment and on the horizon time. Therefore, for a given environment $\mathsf{E}$ and a given horizon time $T$, the function values $v_{T+1}^{(\mathsf{E},\mathsf{A})}(h^T,a_T)$ can be computed once and need never be computed again. Then, given any $T$-step policy $\pi_T$, the backward recursion algorithm can be used to quickly evaluate the expected reward.~\defqed
	\end{remark}
	
	\begin{remark}[Episodic and non-episodic decision processes]
		We have assumed in our definition of a quantum decision process that the reward is given entirely at the end of the agent-environment interaction, i.e., at the horizon time $T$. Such decision processes are called \textit{episodic}. We consider only episodic processes in this thesis.
						
		In general, rewards can be given at intermediate times as well. In this case, we define the agent more generally as follows:
		\begin{equation}
			\mathsf{A}\coloneqq\left(T,\mathcal{S}\subseteq[T],\pi_T\right),
		\end{equation}
		where $\mathcal{S}\subseteq[T]\coloneqq\{1,2,\dotsc,T\}$ is a countable subset of \textit{rewarding times}, and we let $r_{\mathcal{S}}\coloneqq\{s_t:t\in\mathcal{S},\,s_t\in\mathcal{Y}_t\}$ be the set of reward labels for the rewarding times and $\overline{R}_{\mathcal{S}}=(\overline{R}_t:t\in\mathcal{S})$ the corresponding classical registers. An episodic QDP corresponds to taking $\mathcal{S}=\{T\}$.
		
		Allowing for rewards at intermediate times leads to a change in the definition of the $t$-step QDP channel given in \eqref{eq-QDP_channel}. To see how, suppose that $\mathcal{S}=\{1,3,5\}$, so that rewards are given at times $t=1,3,5$. Then, the new definition of the $t$-step QDP channel is
		\begin{multline}
			\mathcal{P}^{(\mathsf{E},\mathsf{A});t}=\mathcal{E}^{t-1}\circ\mathcal{D}^{t-1}\circ\dotsb\circ\mathcal{E}^6\circ\mathcal{D}^6\circ\mathcal{R}^5\circ\mathcal{E}^5\circ\mathcal{D}^5\\\circ\mathcal{E}^4\circ\mathcal{D}^4\circ\mathcal{R}^3\circ\mathcal{E}^3\circ\mathcal{D}^3\circ\mathcal{E}^2\circ\mathcal{D}^2\circ\mathcal{R}^1\circ\mathcal{E}^1\circ\mathcal{D}^1\circ\mathcal{E}^0,
		\end{multline}
		where for brevity we have suppressed the system labels in the subscripts of the channels. In other words, we place a reward channel at every time step in $\mathcal{S}$. In general, if $\mathcal{S}=\{t_1,t_2,\dotsc,t_n\}$, with $t_1<t_2<\dotsb<t_n$ and $n\leq T$, then
		\begin{multline}
			\mathcal{P}^{(\mathsf{E},\mathsf{A});t}=\mathcal{E}^{t-1}\circ\mathcal{D}^{t-1}\circ\dotsb\circ\mathcal{R}^{t_n}\circ\left(\bigcirc_{j=t_{n-1}+1}^{t_n}\mathcal{E}^j\circ\mathcal{D}^j\right)\circ\dotsb\\\circ\mathcal{R}^{t_2}\circ\left(\bigcirc_{j=t_1+1}^{t_2}\mathcal{E}^j\circ\mathcal{D}^j\right)\circ\mathcal{R}^{t_1}\circ\left(\bigcirc_{j=1}^{t_1}\mathcal{E}^j\circ\mathcal{D}^j\right)\circ\mathcal{E}^0.
		\end{multline}
		The figure of merit for the agent, for the purpose of policy evaluation and optimization, is then the expected total reward at time $T$, i.e., $\mathbb{E}[R^{\text{tot}}_{(\mathsf{E},\mathsf{A})}]$, where $R^{\text{tot}}_{(\mathsf{E},\mathsf{A})}(T)=\sum_{t\in\mathcal{S}}R(t)$. It can be shown that a backward recursion algorithm analogous to the one given in Proposition~\ref{prop-QDP_policy_eval} can be used to evaluate policies in this setting.~\defqed
	\end{remark}

\section{Policy optimization}\label{sec-QDP_pol_opt}

	We now discuss methods for obtaining optimal policies. As mentioned earlier, throughout this thesis, we focus our attention on finite-horizon episodic decision processes, i.e., decision processes in which the reward is given only at the horizon time $T<\infty$.
	
	Given an environment $\mathsf{E}$ and horizon time $T<\infty$, the task is to maximize the expected reward at time $T$ (as defined in \eqref{eq-exp_reward_def}--\eqref{eq-QDP_exp_reward_formula_2}) with respect to $T$-step policies $\pi_T$. In other words, the task is to solve the following optimization problem:	
	\begin{equation}\label{eq-QDP_problem}
		\begin{array}{l l} \text{maximize} & F(\mathsf{E},(T,\pi_T)) \\[0.2cm] \text{subject to} & \pi_T=\left(\rho_{A_t}^{h^t}:1\leq t\leq T,\,h^t\in\Omega(t)\right), \\[0.2cm] & \rho_{A_t}^{h^t}\geq 0,\quad \Tr\left[\rho_{A_t}^{h^t}\right]=1\quad\forall~1\leq t\leq T,\, h^t\in\Omega(t). \end{array}
	\end{equation}
	In general, we could add additional constraints to this optimization problem if the policies are required to be constrained. Throughout this thesis, we consider only the basic optimization problem in \eqref{eq-QDP_problem}.

\subsection{Backward recursion}\label{sec-opt_pol_analytical}

	Using the result of Proposition~\ref{prop-QDP_policy_eval} allows us to determine an optimal policy via the backward recursion algorithm.
	
	\begin{theorem}[Finite-horizon policy optimization]\label{thm-QDP_opt_policy}
		Let $\mathsf{E}$ be the environment corresponding to a quantum decision process, and let $T<\infty$. Then,
		\begin{equation}
			\max_{\pi_T}F(\mathsf{E},(T,\pi_T))=\sum_{x_1\in\mathcal{X}}\lambda_{\max}(\widetilde{M}_{A_1}^{1;x_1}),
		\end{equation}
		where
		\begin{align}
			& \forall~1\leq t\leq T-1:\left\{\begin{array}{l l} \displaystyle \widetilde{M}_{A_t}^{t;h^t}\coloneqq\sum_{a_t\in\mathcal{A}}M_{A_t}^{t;a_t}w_{t+1}^{(\mathsf{E},T)}(h^t,a_t),\\[0.8cm]
				\displaystyle w_{t+1}^{(\mathsf{E},T)}(h^t,a_t)=\sum_{x_{t+1}\in\mathcal{X}}\lambda_{\max}(\widetilde{M}_{A_{t+1}}^{t+1;h^t,a_t,x_{t+1}}),\end{array}\right.\label{eq-QDP_opt_pol_back_recur_1}\\[0.5cm]
			& \widetilde{M}_{A_T}^{T;h^T}\coloneqq\sum_{a_T\in\mathcal{A}}M_{A_T}^{T;a_T}w_{T+1}^{(T)}(h^T,a_T),\label{eq-QDP_opt_pol_back_recur_2}\\
			& w_{T+1}(h^T,a_T)=\sum_{x_{T+1}\in\mathcal{X}}\Tr\left[\widetilde{\mathcal{R}}_{E_{T+1}}^{T;h^T,a_T,x_{T+1}}\left(\widetilde{\sigma}_{E_{T+1}}^{(\mathsf{E})}(T+1;h^T,a_T,x_{T+1})\right)\right].\label{eq-QDP_opt_pol_back_recur_3}
		\end{align}
		Furthermore, the optimal policy is
		\begin{equation}
			\pi_T^{\ast}=\left(\ket{\lambda_{\max}(\widetilde{M}_{A_t}^{t;h^t})}\bra{\lambda_{\max}(\widetilde{M}_{A_t}^{t;h^t})}:h^t\in\Omega(t),\,1\leq t\leq T\right),
		\end{equation}
		where $\ket{\lambda_{\max}(H)}$ denotes an eigenvector corresponding the largest eigenvalue of the Hermitian operator $H$.
	\end{theorem}
	
	\begin{proof}
		We start with the backward recursion algorithm given in Proposition~\ref{prop-QDP_policy_eval}. Let
		\begin{equation}
			\pi_T=\left(\rho_{A_t}^{h^t}:1\leq t\leq T,\,h^t\in\Omega(t)\right)
		\end{equation}
		denote an arbitrary $T$-step policy, and let us define
		\begin{equation}
			\pi_T^{(t\to T)}\coloneqq\left(\rho_{A_j}^{h^j}:t\leq j\leq T,\,h^j\in\Omega(j)\right)
		\end{equation}
		to be the ``slices'' of $\pi_T$ from time $t$ onwards. By observing that $v_t^{(\mathsf{E},\mathsf{A})}$ depends only on the policy from time $t$ onwards, i.e., on $\pi_T^{t\to T}$, we find that
		\begin{equation}
			\max_{\pi_T}F(\mathsf{E},(T,\pi_T))=\max_{\left\{\rho_{A_1}^{x_1}\right\}_{x_1}}\sum_{\substack{x_1\in\mathcal{X}\\a_1\in\mathcal{A}}}\Tr\left[M_{A_1}^{1;a_1}\rho_{A_1}^{x_1}\right]\max_{\pi_T^{(2\to T)}}v_2^{(\mathsf{E},\mathsf{A})}(x_1,a_1),
		\end{equation}
		Then, using \eqref{eq-QDP_policy_eval_1}, we have that
		\begin{equation}
			\max_{\pi_T^{(t\to T)}}v_t^{(\mathsf{E},\mathsf{A})}(h^{t-1},a_{t-1})=\max_{\left\{\rho_{A_t}^{x_t}\right\}_{x_t}}\sum_{\substack{x_t\in\mathcal{X}\\a_t\in\mathcal{A}}}\Tr\left[M_{A_t}^{t;a_t}\rho_{A_t}^{x_t}\right]\max_{\pi_T^{(t+1\to T)}}v_{T+1}^{(\mathsf{E},\mathsf{A})}(h^{t-1},a_{t-1},x_t,a_t)
		\end{equation}
		for all $2\leq t\leq T$, $h^{t-1}\in\Omega(t-1)$, and $a_{t-1}\in\mathcal{A}$. By defining the quantities
		\begin{align}
			w_t^{(\mathsf{E},T)}&\coloneqq\max_{\pi_T^{(t\to T)}}v_t^{(\mathsf{E},\mathsf{A})}\quad\forall~2\leq t\leq T,\\
			w_{T+1}^{(\mathsf{E},T)}&\coloneqq v_{T+1}^{(\mathsf{E},\mathsf{A})},
		\end{align}
		(recall that $v_{T+1}^{(\mathsf{E},\mathsf{A})}$ does not depend on $\pi_T$; see Remark~\ref{rem-QDP_policy_eval_2}) we see that the optimization problem reduces to the following recursive procedure:
		\begin{multline}\label{eq-QDP_opt_pol_pf1}
			\max_{\pi_T}F(\mathsf{E},(T,\pi_T))\\=\max\left\{\sum_{\substack{x_1\in\mathcal{X}\\a_1\in\mathcal{A}}}\Tr\left[M_{A_1}^{1;a_1}\rho_{A_1}^{x_1}\right] w_2^{(\mathsf{E},T)}(x_1,a_1): \rho_{A_1}^{x_1}\geq 0,\,\Tr[\rho_{A_1}^{x_1}]=1,\,x_1\in\mathcal{X}\right\},
		\end{multline}
		where
		\begin{equation}\label{eq-QDP_opt_pol_pf2}
			w_t^{(\mathsf{E},T)}(h^{t-1},a_{t-1})\coloneqq\max\left\{\sum_{\substack{x_t\in\mathcal{X}\\a_t\in\mathcal{A}}}\Tr\left[M_{A_t}^{t;a_t}\rho_{A_t}^{x_t}\right]w_{t+1}^{(\mathsf{E},T)}(h^t,a_t):\rho_{A_t}^{x_t}\geq 0,\,\Tr[\rho_{A_t}^{x_t}]=1,\,x_t\in\mathcal{X}\right\}
		\end{equation}
		for all $2\leq t\leq T$, $h^{t-1}\in\Omega(t-1)$, and $a_{t-1}\in\mathcal{A}$.
		
		Now, observe that the objective function in \eqref{eq-QDP_opt_pol_pf1} can be written as
		\begin{align}
			\sum_{\substack{x_1\in\mathcal{X}\\a_1\in\mathcal{A}}}\Tr\left[M_{A_1}^{1;a_1}\rho_{A_1}^{x_1}\right]w_2^{(\mathsf{E},T)}(x_1,a_1)&=\sum_{x_1\in\mathcal{X}}\Tr\left[\left(\sum_{a_1\in\mathcal{A}}M_{A_1}^{1;a_1}w_2^{(\mathsf{E},T)}(x_1,a_1)\right)\rho_{A_1}^{x_1}\right]\\
			&=\sum_{x_1\in\mathcal{X}}\Tr\left[\widetilde{M}_{A_1}^{1;x_1}\rho_{A_1}^{x_1}\right],
		\end{align}
		where in the last line we defined $\widetilde{M}_{A_1}^{t;x_1}\coloneqq\sum_{a_1\in\mathcal{A}_1}M_{A_1}^{t;a_1}w_2^{(\mathsf{E},T)}(x_1,a_1)$. Similarly, in \eqref{eq-QDP_opt_pol_pf2}, the objective function can be written as
		\begin{align}
			\sum_{\substack{x_t\in\mathcal{X}\\a_t\in\mathcal{A}}}\Tr\left[M_{A_t}^{t;a_t}\rho_{A_t}^{x_t}\right]w_{t+1}^{(\mathsf{E},T)}(h^t,a_t)=\sum_{x_t\in\mathcal{X}}\Tr\left[\widetilde{M}_{A_t}^{t;h^t}\rho_{A_t}^{x_t}\right],
		\end{align}
		where $\widetilde{M}_{A_t}^{t;h^t}\coloneqq\sum_{a_t\in\mathcal{A}}M_{A_t}^{t;a_t}w_{t+1}^{(\mathsf{E},T)}(h^t,a_t)$. Therefore, we have
		\begin{align}
			\max_{\pi_T}F(\mathsf{E},(T,\pi_T))&=\max\left\{\sum_{x_1\in\mathcal{X}}\Tr\left[\widetilde{M}_{A_1}^{1;x_1}\rho_{A_1}^{x_1}\right]:\rho_{A_1}^{x_1}\geq 0,\,\Tr[\rho_{A_1}^{x_1}]=1,\,x_1\in\mathcal{X}_1\right\},\label{eq-QDP_opt_pol_pf3}\\
			w_t^{(\mathsf{E},T)}(h^{t-1},a_{t-1})&=\max\left\{\sum_{x_t\in\mathcal{X}}\Tr\left[\widetilde{M}_{A_t}^{t;h^t}\rho_{A_t}^{x_t}\right]:\rho_{A_t}^{x_t}\geq 0,\,\Tr[\rho_{A_t}^{x_t}]=1,\,x_t\in\mathcal{X}_t\right\},\label{eq-QDP_opt_pol_pf4}
		\end{align}
		for all $2\leq t\leq T$, $h^{t-1}\in\Omega(t-1)$, and $a_{t-1}\in\mathcal{A}$. Now, notice that at every iteration the optimization problem is of the form
		\begin{equation}\label{eq-opt_reward_generic}
			\begin{array}{l l} \text{maximize} & \displaystyle \sum_{s\in\mathcal{S}}\Tr[H^s\rho^s] \\[0.6cm] \text{subject to} & \rho^s\geq 0,\,\Tr[\rho^s]=1\quad\forall~s\in\mathcal{S}, \end{array}
		\end{equation}
		where $\mathcal{S}$ is some finite set and $\{H^s\}_{s\in\mathcal{S}}$ is some set of Hermitian operators. Because the optimization is with respect to the independent variables $\{\rho^s\}_{s\in\mathcal{S}}$, the maximum can be brought inside the sum, so that the solution to the optimization problem is simply
		\begin{equation}
			\sum_{s\in\mathcal{S}}\lambda_{\max}(H^s)
		\end{equation}
		where we have used to fact that, for all Hermitian operators $H$,
		\begin{equation}\label{eq-max_EV_SDP}
			\max_{\rho:\rho\geq 0,\Tr[\rho]=1}\Tr[H\rho]=\lambda_{\max}(H),
		\end{equation}
		where $\lambda_{\max}(H)$ denotes the largest eigenvalue of $H$ \cite{BV96}. A state $\rho$ achieving the maximum is given by an eigenvector $\ket{\lambda_{\max}(H)}$ corresponding to the largest eigenvalue, i.e., by the state $\ket{\lambda_{\max}(H)}\bra{\lambda_{\max}(H)}$. Applying this result to \eqref{eq-QDP_opt_pol_pf3} and \eqref{eq-QDP_opt_pol_pf4} leads to the desired result. 
		
		Note that the optimal solution at the $t^{\text{th}}$ iteration gives us the optimal decision states at time $t$. Specifically, the optimal decision state at time $t$ for the history $h^t=(h^{t-1},a_{t-1},x_t)\in\Omega(t)$ is given by the optimal solution corresponding to the function $w_t^{(\mathsf{E},T)}(h^{t-1},a_{t-1})$.	
	\end{proof}
	
	Theorem~\ref{thm-QDP_opt_policy} tells us that we can determine the optimal policy by going \textit{backwards} in time: we first optimize the actions for the horizon time $T$ by finding the largest eigenvalue of the operators $\widetilde{M}_{A_{T}}^{T;h^{T}}$, then optimize the actions for the time $T-1$ by finding the largest eigenvalue of the operators $\widetilde{M}_{A_{T-1}}^{T-1;h^{T-1}}$, and so on, until finally we optimize the actions for the first time step by finding the largest eigenvalue of the operators $\widetilde{M}_{A_1}^{1;x_1}$. This technique is well known from dynamic programming and is used also for obtaining optimal policies in classical Markov decision processes. Furthermore, this procedure is related to the density-matrix renormalization group (DMRG) algorithm \cite{White92,Sch05}, which is used for tensor network optimization. In fact, as pointed out in Ref.~\cite{GRG20}, which uses tensor network techniques to solve classical Markov decision problems, the backward recursion procedure is nothing but one reverse sweep in the DMRG algorithm, in which the policies are viewed as tensors. Intuitively, the main reason why optimization via a backward recursion procedure can be done is that each term in the objective function is a product of terms, with each term in the product corresponding to an independent action state. This allows the joint optimization over all action states to be split up into optimizations over the individual factors in a successive fashion starting from the actions states for the final time.

\subsection{Forward recursion}\label{sec-QDP_forward_recursion}

	Observe that the backward recursion algorithm presented in Theorem~\ref{thm-QDP_opt_policy} is exponentially slow in the horizon time because of the fact that the number of histories grows exponentially in time---the number of histories up to time $t$ is $|\Omega(t)|=|\mathcal{X}|^t|\mathcal{A}|^{t-1}$. Therefore, given a finite horizon time $T$, the functions $w_t^{(\mathsf{E},T)}$ used to determine the optimal action states have an exponentially increasing number of values. For this reason, it is useful to have efficient methods for estimating the maximum expected reward. One such method is the following.

	Instead of starting from the horizon time and finding the optimal actions by going backwards, we could instead find the optimal actions by going forwards, i.e., by selecting the action such that the immediate expected reward is maximized. Such a ``forward recursion'' approach is more natural from the perspective of a real-world learning agent, who has to make decisions in real time and does not necessarily have complete knowledge of the environment in order to perform the backward recursion algorithm. However, the forward recursion algorithm will not necessarily lead to a globally optimal policy. In fact, a globally optimal policy can be obtained using the backward recursion algorithm, which we proved in Theorem~\ref{thm-QDP_opt_policy}. Nevertheless, it is worthwhile to briefly discuss the forward recursion algorithm because many reinforcement learning algorithms are based on it, and they give efficiently computable lower bounds on the maximum expected reward. In Appendix~\ref{app-QDP_SDP}, we investigate upper bounds on the maximum expected reward that are based on semi-definite programming.
	
	Consider an arbitrary QDP $\mathsf{Q}=(\mathsf{E},\mathsf{A})$. The classical-quantum state of $\mathsf{Q}$ at time $t\geq 1$ is
	\begin{equation}
		\widehat{\sigma}_{H_tE_t}^{(\mathsf{E},\mathsf{A})}(t)=\sum_{h^t\in\Omega(t)}\ket{h^t}\bra{h^t}_{H_t}\otimes\widetilde{\sigma}_{E_t}^{(\mathsf{E},\mathsf{A})}(t;h^t).
	\end{equation}
	Now, the forward recursion algorithm is defined by the task of determining the action state such that the immediate expected reward is maximized. After one step of the agent-environment interaction, we have
	\begin{equation}
		\widehat{\sigma}_{H_{t+1}E_{t+1}}^{(\mathsf{E},\mathsf{A})}(t+1)=\sum_{\substack{h^t\in\Omega(t)\\a_t\in\mathcal{A}\\x_{t+1}\in\mathcal{X}}}\ket{h^t,a_t,x_{t+1}}\bra{h^t,a_t,x_{t+1}}_{H_{t+1}}\otimes\Tr\left[M_{A_t}^{t;a_t}\rho_{A_t}^{h^t}\right]\mathcal{T}_{E_t\to E_{t+1}}^{t;h^t,a_t,x_{t+1}}\left(\widetilde{\sigma}_{E_t}^{(\mathsf{E},\mathsf{A})}(t;h^t)\right).
	\end{equation}
	The expected reward is then
	\begin{equation}
		\sum_{\substack{h^t\in\Omega(t)\\a_t\in\mathcal{A}\\x_{t+1}\in\mathcal{X}}}\Tr\left[M_{A_t}^{t;a_t}\rho_{A_t}^{h^t}\right]\Tr\left[\left(\widetilde{\mathcal{R}}_{E_{t+1}}^{t;h^t,a_t,x_{t+1}}\circ\mathcal{T}_{E_t\to E_{t+1}}^{t;h^t,a_t,x_{t+1}}\right)\left(\widetilde{\sigma}_{E_t}^{(\mathsf{E},\mathsf{A})}(t;h^t)\right)\right]=\sum_{h^t\in\Omega(t)}\Tr\left[\widetilde{N}_{A_t}^{h^t}\rho_{A_t}^{h^t}\right],
	\end{equation}
	where
	\begin{equation}
		\widetilde{N}_{A_t}^{h^t}\coloneqq \sum_{\substack{a_t\in\mathcal{A}\\x_{t+1}\in\mathcal{X}}} M_{A_t}^{t;a_t}\Tr\left[\left(\widetilde{\mathcal{R}}_{E_{t+1}}^{t;h^t,a_t,x_{t+1}}\circ\mathcal{T}_{E_t\to E_{t+1}}^{t;h^t,a_t,x_{t+1}}\right)\left(\widetilde{\sigma}_{E_t}^{(\mathsf{E},\mathsf{A})}(t;h^t)\right)\right].
	\end{equation}
	Maximizing the expected reward is therefore an optimization problem of the form \eqref{eq-opt_reward_generic}, so that the maximum expected reward is simply
	\begin{equation}
		\sum_{h^t\in\Omega(t)}\lambda_{\max}(\widetilde{N}_{A_t}^{h^t}),
	\end{equation}
	with the optimal decision states being $\ket{\lambda_{\max}(\widetilde{N}^{h^t})}\bra{\lambda_{\max}(\widetilde{N}^{h^t})}$. In other words, given that the agent observes a particular history $h^t\in\Omega(t)$, the decision state $\rho_{A_t}^{h^t}$ that maximizes its reward is $\rho_{A_t}^{h^t}=\ket{\lambda_{\max}(\widetilde{N}^{h^t})}\bra{\lambda_{\max}(\widetilde{N}^{h^t})}$.
	
	If the POVM $\{M_{A_t}^{t;a_t}\}_{a_t\in\mathcal{A}_t}$ is a set of rank-one projections onto action states, i.e., $M_{A_t}^{t;a_t}=\ket{a_t}\bra{a_t}_{A_t}$ for all $a_t\in\mathcal{A}$, then
	\begin{equation}
		\widetilde{N}_{A_t}^{h^t}=\sum_{a_t\in\mathcal{A}}\ket{a_t}\bra{a_t}\left(\sum_{x_{t+1}\in\mathcal{X}_{t+1}}\Tr\left[\left(\widetilde{\mathcal{R}}_{E_{t+1}}^{t;h^t,a_t,x_{t+1}}\circ\mathcal{T}_{E_t\to E_{t+1}}^{t;h^t,a_t,x_{t+1}}\right)\left(\widetilde{\sigma}_{E_t}^{(\mathsf{E},\mathsf{A})}(t;h^t)\right)\right]\right),
	\end{equation}
	so that the maximum expected reward is
	\begin{equation}
		\sum_{h^t\in\Omega(t)} \max_{a_t\in\mathcal{A}}\left(\sum_{x_{t+1}\in\mathcal{X}}\Tr\left[\left(\widetilde{\mathcal{R}}_{E_{t+1}}^{t;h^t,a_t,x_{t+1}}\circ\mathcal{T}_{E_t\to E_{t+1}}^{t;h^t,a_t,x_{t+1}}\right)\left(\widetilde{\sigma}_{E_t}^{(\mathsf{E},\mathsf{A})}(t;h^t)\right)\right]\right),
	\end{equation}
	and the optimal (deterministic) action is
	\begin{equation}
		\argmax_{a_t\in\mathcal{A}}\left(\sum_{x_{t+1}\in\mathcal{X}}\Tr\left[\left(\widetilde{\mathcal{R}}_{E_{t+1}}^{t;h^t,a_t,x_{t+1}}\circ\mathcal{T}_{E_t\to E_{t+1}}^{t;h^t,a_t,x_{t+1}}\right)\left(\widetilde{\sigma}_{E_t}^{(\mathsf{E},\mathsf{A})}(t;h^t)\right)\right]\right).
	\end{equation}
	
	We note that an optimal policy obtained via the forward recursion algorithm is often called a ``greedy'' policy; see, e.g., Ref.~\cite{Sut18_book}.

\section{Classical vs. quantum agents}\label{sec-QDP_classical_v_quantum_agents}

	At the end of the previous subsection, we briefly discussed, in the context of the forward recursion algorithm, the case that the environment response channels are defined by POVMs $\{M_{A_t}^{t;a_t}\}_{a_t\in\mathcal{A}}$ such that $M_{A_t}^{t;a_t}=\ket{a_t}\bra{a_t}_{A_t}$ for all $a_t\in\mathcal{A}$. In other words, the POVMs consist of rank-one projections onto the action basis elements. Let us now look at this special case in the context of the backward recursion algorithm and use Theorem~\ref{thm-QDP_opt_policy} to explicitly write down that algorithm for this case.
	
	When $M_{A_t}^{t;a_t}=\ket{a_t}\bra{a_t}_{A_t}$ for all $1\leq t\leq T$ and $a_t\in\mathcal{A}$, from \eqref{eq-QDP_opt_pol_back_recur_1} and \eqref{eq-QDP_opt_pol_back_recur_2} we obtain
	\begin{equation}
		\widetilde{M}_{A_t}^{t;h^t}=\sum_{a_t\in\mathcal{A}}w_{t+1}^{(\mathsf{E},T)}(h^t,a_t)\ket{a_t}\bra{a_t}_{A_t}\quad\forall~1\leq t\leq T,\,h^t\in\Omega(t)
	\end{equation}
	These operators are diagonal in the action basis, so that
	\begin{equation}
		\lambda_{\max}(\widetilde{M}_{A_t}^{t;h^t})=\max_{a_t\in\mathcal{A}}w_{t+1}^{(\mathsf{E},T)}(h^t,a_t)\quad\forall~1\leq t\leq T,\,h^t\in\Omega(t).
	\end{equation}
	The result of Theorem~\ref{thm-QDP_opt_policy} therefore simplifies as follows:
	\begin{equation}\label{eq-QDP_opt_pol_classical_agent_1}
		\max_{\pi_T}F(\mathsf{E},(T,\pi_T))=\sum_{x_1\in\mathcal{X}}\max_{a_1\in\mathcal{A}}w_2^{(\mathsf{E},T)}(x_1,a_1),
	\end{equation}
	where
	\begin{align}
		w_t^{(\mathsf{E},T)}(h^{t-1},a_{t-1})&=\sum_{x_t\in\mathcal{X}}\max_{a_t\in\mathcal{A}}w_{t+1}^{(\mathsf{E},T)}(h^{t-1},a_{t-1},x_t,a_t)\quad\forall~2\leq t\leq T,\label{eq-QDP_opt_pol_classical_agent_2}\\[0.2cm]
		w_{T+1}(h^T,a_T)&=\sum_{x_{T+1}\in\mathcal{X}}\Tr\left[\widetilde{\mathcal{R}}_{E_{T+1}}^{T;h^T,a_T,x_{T+1}}\left(\widetilde{\sigma}_{E_{T+1}}^{(\mathsf{E})}(T+1;h^T,a_T,x_{T+1})\right)\right].\label{eq-QDP_opt_pol_classical_agent_3}
	\end{align}
	Furthermore, the optimal policy is
	\begin{equation}\label{eq-QDP_opt_pol_classical_agent_4}
		\pi_T^*=\left(\ket{d_t^*(h^t)}\bra{d_t^*(h^t)}_{A_t}:1\leq t\leq T,\,h^t\in\Omega(t)\right),
	\end{equation}
	where
	\begin{equation}\label{eq-QDP_opt_pol_classical_agent_5}
		d_t^*(h^t)\coloneqq\argmax_{a_t\in\mathcal{A}}w_{t+1}^{(\mathsf{E},T)}(h^t,a_t)\quad\forall~1\leq t\leq T.
	\end{equation}
	
	Notice that the optimal action states are nothing more than rank-one projections onto the action basis elements. This means that the optimal actions are purely classical and deterministic, meaning that when the POVMs of the environment response channels are rank-one projections onto the action basis elements, a classical agent is sufficient for achieving the optimal expected reward. A classical agent is similarly sufficient if the POVM elements have rank greater than one but are diagonal in the action basis.
	
	On the other hand, when the POVM elements $M_{A_j}^{j;a_j}$ are not diagonal in the action basis elements, the optimal action states need not be diagonal in the action basis. In particular, from Theorem~\ref{thm-QDP_opt_policy}, we see that, in general, the optimal action states will involve superpositions of action basis elements. In this case, therefore, a quantum agent is needed to obtain the optimal expected reward.

\section{Connections to quantum information tasks}\label{sec-QDP_q_info_tasks}

	We pointed out in Remark~\ref{rem-QDP_def}, and showed explicitly in Figure~\ref{fig-QDP}, that every quantum decision process can be viewed as a quantum causal network/quantum comb. Drawing this connection to quantum causal networks is particularly powerful because quantum causal networks are general enough that they arise in virtually all quantum information processing tasks that contain a causal structure. Thus, by viewing quantum decision processes as an example of a quantum causal network, we come to the realization that many quantum information processing tasks that involve sequential interaction between different parties can be thought of as a quantum decision process. Let us briefly consider some examples for which this is the case.
	\begin{itemize}
		\item \textit{Device-independent quantum information processing.} This includes device-independent quantum key distribution \cite{ABG+07,BP12,LCQ12}; see Ref.~\cite[Chapter~6]{AF20} for a broad introduction. In this case, the blue channels in Figure~\ref{fig-QDP} should be thought of as a device at different points in time, such that the agent interacting with the device has access only to the input and output systems $H_j,A_j$ and $H_{j+1}$, while the systems $E_j$ represent the device's internal memory, which is inaccessible to the agent. The agent's goal is then to interact optimally with the device relative to a given task.
			
			The CHSH game \cite{Bell64,CHSH69} (see Ref.~\cite{Scar13} for a pedagogical introduction) is an important subroutine in many device-independent quantum information processing tasks (see, e.g., Refs.~\cite{ADF+18,ARV19}). In this case, the environment is the underlying quantum system(s) describing the device, and the agents' actions represent measurement settings for the quantum system. The classical outputs from the device are labels corresponding to the outcome of the measurement. A policy in this setting is a mapping from prior measurement outcomes to future measurement setting choices, and the reward is the score of the CHSH game.
		
		\item \textit{Quantum games.} These are quantum generalizations of stochastic games \cite{Shap53,SV15,LS2017}, which themselves are multi-agent generalizations of (classical) Markov decision processes \cite{Tan93,HW98,BBS08,FV97_book}. The agents are the players in the game, who may or may not cooperate, and the environment can be thought of as the referee, which assigns a reward/score to the players in every round.
			
			See Refs.~\cite{GVW99,Meyer99,EWL99,Meyer00,EW00,BH01,BH01b,EP02,DLX+02,LJ03,WAN20} for examples of quantum games, and see Refs.~\cite{GZK08,KSBH18} for reviews. Although the presentation in this chapter has been restricted to quantum decision processes with a single agent, it is possible to generalize Definition~\ref{def-QDP} to multiple agents.
		
		\item \textit{Quantum control.} See, e.g., Refs.~\cite{MDS+18,ZWA+19}. In these works, the agent is an experimentalist whose task is to ``control'' a quantum system, i.e., to drive the quantum system to a particular state, and the environment is simply the quantum system itself. The actions of the agent correspond to the physical evolutions of the quantum system, and the reward is typically the fidelity (or some other distinguishability measure) with respect to the target state.
		
		\item \textit{Quantum-enhanced parameter estimation.} See, e.g., Refs.~\cite{HS10a,HS11,PWS16,PWS16b,PWS17}. In these works, in particular Refs.~\cite{PWS16,PWS17}, the task of the agent is to estimate an unknown phase shift applied to a quantum system (which represents the environment). In each round, the agent makes an estimate of this phase, and the classical outputs from the environment are particular measurement outcomes. The agent's policy is a function from prior measurement outcomes to future estimates, and the reward is the variance of the estimate; see Ref.~\cite{PWS17} for details.
		
		\item \textit{Quantum error correction.} See, e.g., Refs.~\cite{FTP+18,SKLE18,NDV+19}. In these works, roughly speaking, the environment is the quantum system to be protected by an error-correction code, and the agent's task is to perform actions (which correspond to physical evolutions of the quantum system) in order to keep the state of the quantum system in the codespace.
		
	\end{itemize}

\section{Summary}

	In this chapter, we developed the primary concept that we use throughout the rest of this thesis, namely, the concept of quantum decision processes. Our notion of a quantum decision process builds on the notion of a quantum partially observable Markov decision process as defined in Refs.~\cite{BBA14,YY18}, which itself generalizes the concept of a classical (partially observable) Markov decision process. In a quantum decision process, an agent (which could be described by a classical or a quantum system) interacts with its environment (which is described by a quantum system) via a sequence of actions. Every action of the agent changes the quantum state of the environment, and in return the agent receives (classical) observations about the environment, along with a reward. The agent uses the observations it receives in order to decide its next action, and its goal is to maximize the reward it receives after a pre-specified, finite horizon time. Several quantum information processing tasks can be thought of in the context of quantum decision processes, and we provide examples of some of these tasks in Section~\ref{sec-QDP_q_info_tasks}.
	
	We began the chapter with the formal definition of a quantum decision process in Section~\ref{sec-QDP}, and thereafter we determined some basic facts, such as the classical-quantum state and the expected quantum state of a quantum decision process. Then, in Section~\ref{sec-QDP_pol_opt}, we considered finite-horizon policy optimization, in which the goal is to determine the optimal sequence of actions that should be performed by an agent in order to maximize its expected reward at the horizon time. The backward recursion algorithm presented in Section~\ref{sec-opt_pol_analytical} gives us the optimal policy, but it is rarely used in practice because it is exponentially slow in the horizon time. The forward recursion algorithm, presented in Section~\ref{sec-QDP_forward_recursion}, is generally sub-optimal and provides only a lower bound on the maximum expected reward, but it is at the basis of algorithms that are more efficient than backward recursion. In Section~\ref{sec-QDP_classical_v_quantum_agents}, we showed how the backward recursion algorithm, and the resulting optimal policy, simplifies when the agent can be described classically.
	
	Decision processes form the theoretical foundation for reinforcement learning. To be precise, reinforcement learning algorithms take as their underlying model for the environment a (classical) decision process, and they generally provide lower bounds on the maximum expected reward of the agent. In fact, the forward recursion algorithm presented in Section~\ref{sec-QDP_forward_recursion} is at the basis of many reinforcement learning algorithms. We refer to Ref.~\cite{Sut18_book} for more information about reinforcement learning. Due to the connection between decision processes and reinforcement learning in the classical case, we can use quantum decision processes as a basis for a quantum generalization of reinforcement learning. Although reinforcement learning is not the primary subject of this thesis, the developments of this chapter, and of this thesis as a whole, provide the tools needed for reinforcement learning of practical network protocols; we explain this in more detail in Appendix~\ref{sec-future_work}.

\chapter{QUANTUM NETWORKS}\label{chap-network_setup}

	In this chapter, we begin the study of quantum networks. We discuss how to describe a quantum network mathematically. Then, we define the general task of entanglement distribution in a quantum network along with a figure of merit for entanglement distribution protocols in a quantum network. From these general considerations, we motivate a more practically-oriented model for entanglement distribution. The developments of this chapter prepare us for devising practical quantum network protocols using quantum decision processes in Chapter~\ref{chap-network_QDP}.

	What exactly do we mean by the term ``quantum network''? Generally speaking, the term ``network'' can be used to describe any collection of entities that interact with each other and whose behavior must be described collectively rather than with the individual components \cite{Barabasi16_book}. In a communication context, like the networks that comprise the current internet, a quantum network is simply a collection of spatially-separated quantum-mechanical devices, some of which are directly connected to each other via communication channels. Consequently, every quantum network (specifically, its physical layout) has a natural association to a (multi)graph $G=(V,E,c)$, in which the nodes of the network are associated to the vertices of the graph, and the nodes that are directly connected to each other in the network with a quantum channel are represented in the graph by an edge connecting the corresponding vertices. As we describe in Section~\ref{sec-network_architecture}, for the types of networks that we consider in this thesis the quantum channels are used to distribute entangled states to the nodes connected by the channels. Consequently, we can equivalently think of the edges of the graph as representing an entangled state shared by the corresponding nodes. We refer to these entangled states associated with the edges as \textit{elementary links}.
	
	\begin{figure}
		\centering
		\includegraphics[width=0.45\textwidth]{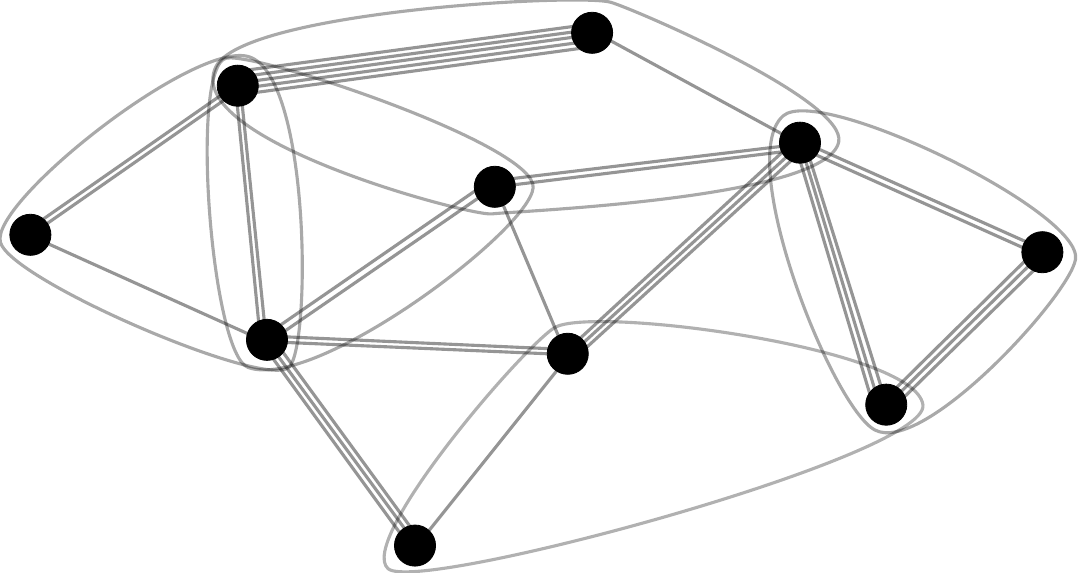}\quad
		\includegraphics[width=0.45\textwidth]{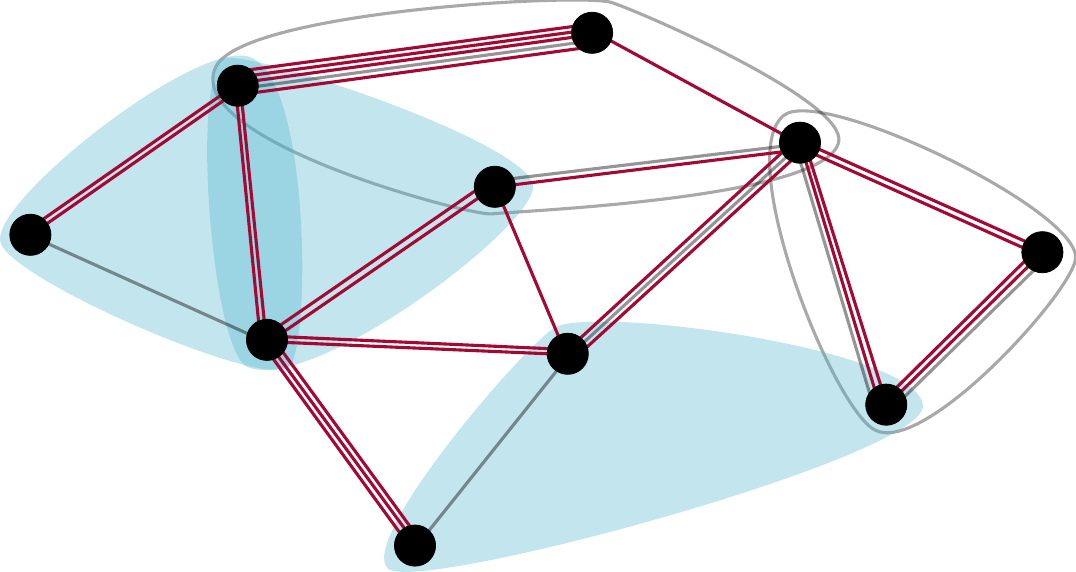}
		\caption{Graphical description of a quantum network. (Left) The physical layout of the quantum network is described by a hypergraph $G$, which should be thought of as fixed, in which the vertices represent the nodes (senders and receivers) in the network and the edges represent quantum channels that are used to distribute entangled states (elementary links) shared by the corresponding nodes. (Right) At any point in time only a certain number of elementary links in the network may be active. By ``active'', we mean that an entangled state has been distributed successfully to the nodes and the corresponding quantum systems stored in the respective quantum memories. Active bipartite links are indicated by a red line, and active $k$-partite elementary links, $k\geq 3$, corresponding to the hyperedges are indicated by a blue bubble. The active elementary links collectively constitute a subgraph of $G$. In Section~\ref{sec-network_architecture}, we describe how information transmission in a quantum network is typically probabilistic in practice, which means that the subgraphs are random.}\label{fig-network_physical}
	\end{figure}
	
	In Figure~\ref{fig-network_physical}, we illustrate an arbitrary quantum network using its corresponding graph $G$. The left panel of the figure depicts the physical layout of the network, which we assume to be fixed. Every gray edge indicates a quantum channel, and whenever the quantum channel is used to successfully distribute an entangled state to the corresponding nodes, we color the edge red (in the case of bipartite states) or blue (in the case of multipartite states), and we refer to all such edges as ``active elementary links''. The collection of active elementary links in the network forms a subgraph of $G$. In this way, we can think of $G$ as describing the ``physical (elementary) links'' in the network. The overall goal in a quantum network is to use the elementary links to create ``virtual links'', which is entanglement shared by nodes that are not physically connected, allowing them to accomplish a desired task, such as quantum teleportation or quantum key distribution. In Section~\ref{sec-ent_dist_general}, we formally define entanglement distribution in a quantum network in terms of graph transformations, and we define a figure of merit for evaluating entanglement distribution protocols.
	
	In Section~\ref{sec-network_architecture}, we explain why the generation of elementary links is typically probabilistic in practice, and how the finite coherence time of quantum memories limits the time duration of active elementary links. These facts necessitate thinking of the subgraphs of active elementary links as random variables, which we formalize in Section~\ref{sec-random_subgraphs}, and we define the corresponding quantum state of the network in Section~\ref{sec-network_q_state_practical}.
	
	\begin{remark}
		An alternative mathematical description of a quantum network as a graph is one in which every edge represents a vector space and every vertex represents a linear transformation of the vector spaces corresponding to the edges incident to the vertex. This type of object is typically called a \textit{tensor network} (see, e.g., Ref.~\cite{Orus14,BB17,BC17}), because the linear transformations at the vertices can be thought of as tensors. The usual diagrammatic representation of quantum circuits in quantum computing (see, e.g., Ref.~\cite{NC00_book}) can be thought of as a tensor network. More generally, quantum (causal) networks, as defined in Ref.~\cite{CDP09}, are tensor networks, because the edges represent Hilbert spaces and the vertices represent completely positive maps, and the latter can be thought of as a tensor via the Choi representation.~\defqed
	\end{remark}

\section{Quantum states and channels}\label{sec-network_q_state}
	
	Along with the abstraction of a quantum network as a graph, which we have put into place above and is summarized by Figure~\ref{fig-network_physical}, in order to have a formal theoretical study of quantum networks, we need a systematic way of referring to the quantum systems and channels that are involved in the physical functioning of the network. We can then talk about, for example, the overall quantum state of the network and communication between different nodes in the network via the underlying quantum channels.

	Given a graph $G=(V,E)$ corresponding to a quantum network, we label the quantum systems in the network as $A^v_e$, with $v\in V$ and $e\in E$, which tells us that the system $A^v_e$ is located at the vertex $v$ and is associated with the edge $e$ incident to $v$, i.e., $v\in e$; see Figure~\ref{fig-network_q_systems} for an example. We can label quantum systems in the same way for a multigraph, i.e., in the case that a pair of vertices is connected by multiple edges, because we regard each of the parallel edges as a distinct edge.
	
	\begin{figure}
		\centering
		\includegraphics[scale=1]{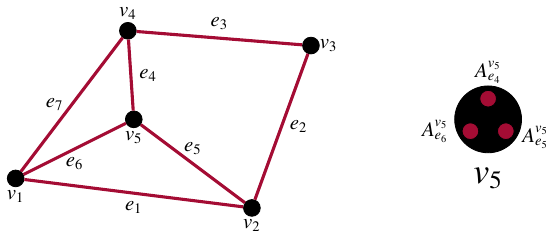}
		\caption{(Left) A graph with 5 vertices and 7 edges corresponding to a quantum network in which all of the adjacent nodes share bipartite entanglement. (Right) Focusing on the vertex $v_5$, we define at least three quantum systems located at that vertex, one system for each edge attached to the vertex.}\label{fig-network_q_systems}
	\end{figure}

	As described at the beginning of this chapter, every active elementary link of the network corresponds to a quantum state, which we denote by
	\begin{equation}
		\rho_e\equiv \rho_{\{A_e^v:v\in V,\,v\in e\}}.
	\end{equation}
	In other words, $\rho_e$ is the quantum state for the systems located at the nodes that are attached to $e$. Note that if $e$ is a hyperedge, with $k\geq 2$ vertices attached to it, then $\rho_e$ is a $k$-partite quantum state. As an example, let us refer to Figure~\ref{fig-network_q_systems}. We have
	\begin{equation}
		\rho_{e_1}\equiv \rho_{A_{e_1}^{v_1}A_{e_1}^{v_2}},\quad\rho_{e_2}\equiv\rho_{A_{e_2}^{v_2}A_{e_2}^{v_3}},
	\end{equation}
	and similarly for all of the other edges. Given any subgraph $G'=(V',E')$ of $G$, consisting of the active elementary links and their associated nodes, we define the quantum state of $G'$ as
	\begin{equation}\label{eq-network_total_state}
		\sigma_{G'}\equiv\sigma_{\{A_e^v:e\in E',\,v\in V'\}}\coloneqq\bigotimes_{e\in E'}\rho_e.
	\end{equation}
	Note that in this formulation, the quantum state of any subgraph is always a tensor-product state with respect to the different edges.
	
	We also define the following partial trace functions:
	\begin{align}
		\Tr_e\coloneqq\Tr_{\{A_e^v:v\in V,\,v\in e\}} \quad\forall~e\in E,\\
		\Tr_v\coloneqq\Tr_{\{A_e^v:e\in E,\,v\in e\}} \quad\forall~v\in V.
	\end{align}
	For example, referring again to Figure~\ref{fig-network_q_systems}, we have that
	\begin{align}
		\Tr_{v_2}&=\Tr_{A_{e_1}^{v_2}A_{e_5}^{v_2}A_{e_2}^{v_2}},\\
		\Tr_{e_2}&=\Tr_{A_{e_2}^{v_2}A_{e_2}^{v_3}}.
	\end{align}
	More generally, for subsets $V'\subseteq V$ and $E'\subseteq E$, we have
	\begin{align}
		\Tr_{V'}&\coloneqq\Tr_{\{A_e^v:v\in V',\,v\in e,\,e\in E\}},\\
		\Tr_{E'}&\coloneqq\Tr_{\{A_e^v:e\in E',\,v\in e,\,v\in V\}}.
	\end{align}
	Using these definitions, we can define reduced/marginal states of the overall quantum state of the network.
	
	All of the definitions above apply to hypergraphs and multigraphs. In the following, however, let us restrict ourselves to multigraphs with bipartite edges. In this case, a path in the graph can be specified by a sequence $(v_1,e_1,v_2,e_2,v_3,\dotsb,v_{n-1},e_{n-1},v_n)$, such that $v_i,v_{i+1}\in e_i$ for all $1\leq i\leq n$.
	
	Let us now discuss how to describe channels for joining elementary links, starting with the entanglement swapping channel defined in \eqref{eq-ent_swap_channel}. For $n>2$, given a particular path $(v_1,e_1,v_2,e_2,\allowbreak v_3,\dotsb,v_{n-1},e_{n-1},v_n)$ between the nodes $v_1$ and $v_n$, such that there are $n-2$ intermediate nodes, we define the channel $\mathcal{L}_{v_1e_1v_2e_2v_3\dotsb v_{n-1}e_{n-1}v_n\to \{v_1,v_n\}}^{\text{ES};{n-2}}$ to be the entanglement swapping channel given in \eqref{eq-ent_swap_channel} that connects the vertices $v_1$ and $v_n$ at the ends of the path, and thus creates a new edge $e'=\{v_1,v_n\}$, so that
	\begin{equation}\label{eq-ent_swap_path}
		\mathcal{L}_{v_1e_1v_2e_2v_3\dotsb v_{n-1}e_{n-1}v_n\to \{v_1,v_n\}}^{\text{ES};{n-2}} \coloneqq \mathcal{L}_{A_{e_1}^{v_1}A_{e_1}^{v_2}A_{e_2}^{v_2}\dotsb A_{e_{n-1}}^{v_{n-1}}A_{e_{n-1}}^{v_n}\to A_{e'}^{v_1}A_{e'}^{v_n}}^{\text{ES};{n-2}}.
	\end{equation}
	We sometimes refer to the quantum state corresponding to the edge $e'$ as a \textit{virtual link} when we want to distinguish it from an elementary link. Note that, instead of the output systems being $A_{e_1}^{v_1}A_{e_{n-1}}^{v_n}$, we have assigned new systems $A_{e'}^{v_1}A_{e'}^{v_n}$ for the output in order to signify the creation of a new (virtual) link. In other words, the output state of the channel is put into two entirely new quantum systems, which allows for the systems $A_{e_1}^{v_1}$ and $A_{e_{n-1}}^{v_n}$ to be used again as part of the elementary links. As an example, for just two elementary links $e_1$ and $e_2$ and the path $(v_1,e_1,v_2,e_2,v_3)$, we have
	\begin{equation}
		\mathcal{L}_{v_1e_1v_2e_2v_3\to \{v_1,v_3\}}^{\text{ES};1}\coloneqq \mathcal{L}_{A_{e_1}^{v_1}A_{e_1}^{v_2}A_{e_2}^{v_2}A_{e_2}^{v_3}\to A_{e'}^{v_1}A_{e'}^{v_3}}^{\text{ES};1},
	\end{equation}
	where $e'=\{v_1,v_3\}$.
	
	We make an analogous definition for the GHZ entanglement swapping channel defined in \eqref{eq-GHZ_ent_swap_channel}. Given $n>2$ and a path $(v_1,e_1,v_2,e_2,v_3,\dotsb,v_{n-1},e_{n-1},v_n)$ between the nodes $v_1$ and $v_n$, such that there are $n-2$ intermediate nodes, we define the channel $\mathcal{L}_{v_1e_1v_2e_2v_3\dotsb v_{n-1}e_{n-1}v_n\to \{v_1,v_2,\dotsc,v_n\}}^{\text{GHZ};n-2}$ to be the GHZ entanglement swapping channel that connects the vertices $v_1,v_2,\dotsc,v_n$ and thus creates a new edge $e'=\{v_1,v_2,\dotsc,v_n\}$, so that
	\begin{equation}\label{eq-ent_swap_GHZ_path}
		\mathcal{L}_{v_1e_1v_2e_2v_3\dotsb v_{n-1}e_{n-1}v_n\to \{v_1,v_2,\dotsc,v_n\}}^{\text{GHZ};n-2} \coloneqq \mathcal{L}_{A_{e_1}^{v_1}A_{e_1}^{v_2}A_{e_2}^{v_2}\dotsb A_{e_{n-1}}^{v_{n-1}}A_{e_{n-1}}^{v_n}\to A_{e'}^{v_1}A_{e'}^{v_2}\dotsb A_{e'}^{v_n}}^{\text{GHZ};{n-2}}.
	\end{equation}
	Like before, we assign the output of the channel to new quantum systems in order to signify the creation of a new (virtual) link, as well as to allow for the creation of new elementary links with the systems that were consumed by measurements in the GHZ entanglement swapping protocol. As an example, consider again just two elementary links $e_1$ and $e_2$ and the path $(v_1,e_1,v_2,e_2,v_3)$. Then,
	\begin{equation}
		\mathcal{L}_{v_1e_1v_2e_2v_3\to \{v_1,v_2,v_3\}}^{\text{GHZ};1}\coloneqq \mathcal{L}_{A_{e_1}^{v_1}A_{e_1}^{v_2}A_{e_2}^{v_2}A_{e_2}^{v_3}\to A_{e'}^{v_1}A_{e'}^{v_2}A_{e'}^{v_3}}^{\text{GHZ};1},
	\end{equation}
	where $e'=\{v_1,v_2,v_3\}$.

	We describe distillation channels in a manner similar to how we did for entanglement swapping channels. Given a hyperedge $e=\{v_1,\dotsc,v_k\}$, with $k\geq 2$, such that $ c(e)=n$ is the number of parallel edges connecting the vertices in $e$, we denote by
	\begin{equation}
		\mathcal{D}_{e^1\dotsb e^n\to e^1\dotsb e^{n'}}^{e}
	\end{equation}
	the entanglement distillation channel that takes all $n$ parallel edges and transforms them to a number $n'<n$ of parallel edges.

\section{The entanglement distribution task}\label{sec-ent_dist_general}

	The central task that we consider in this thesis is entanglement distribution, which we define to be the task of transforming elementary links, which are entangled states shared by the nodes that are physically connected to each other, to virtual links, which are entangled states shared by non-adjacent nodes in the network. This transformation takes places through the execution of an LOCC protocol; see Figure~\ref{fig-network_transformation} for a generic depiction. Importantly, in a real network, there are multiple simultaneous user requests that have to be accommodated, not just a request between one set (e.g., pair) of users, and protocols for entanglement distribution should accommodate such multi-user requests~\cite{CERW20}.
	
	\begin{figure}
		\centering
		\includegraphics[scale=1]{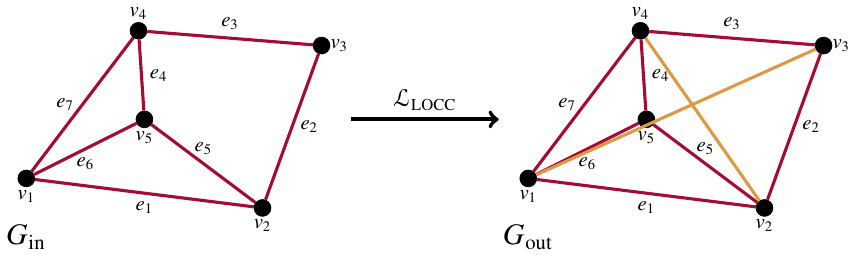}
		\caption{Depiction of entanglement distribution as a transformation of a graph $G_{\text{in}}$ of active elementary links of the physical graph to a new graph $G_{\text{out}}$ consisting of the edges of $G_{\text{in}}$ plus two additional virtual edges indicated in orange. The transformation occurs through an LOCC protocol executed by the nodes in the network, which is described by the LOCC channel $\mathcal{L}_{\text{LOCC}}$. This formulation of the task of entanglement distribution in terms of transformations of graphs is based on Refs.~\cite{SMI+17,CRDW19}.}\label{fig-network_transformation}
	\end{figure}
	
	Entanglement distribution protocols can be described in terms of graph transformations, as done in Ref.~\cite{SMI+17,CRDW19} and depicted in Figure~\ref{fig-network_transformation}. We start with the quantum state corresponding to a collection of elementary links in a network, as in \eqref{eq-network_total_state}. This collection of elementary links corresponds to a graph $G_{\text{in}}=(V_{\text{in}},E_{\text{in}})$. Note that the graph $G_{\text{in}}$ could be a subgraph of the full graph $G$ consisting of all of the physical links, as described at the beginning of this chapter and shown in right panel of Figure~\ref{fig-network_physical}. If $\sigma_{G_{\text{in}}}$ is the state of the graph $G_{\text{in}}$, as in \eqref{eq-network_total_state}, then a general entanglement distribution protocol is an LOCC protocol among the nodes $V_{\text{in}}$, with corresponding LOCC channel $\mathcal{L}_{\text{LOCC}}$, such that
	\begin{equation}\label{eq-ent_dist_protocol_graphs_trans}
		\sigma_{G_{\text{in}}}\mapsto \sigma_{G_{\text{out}}}\coloneqq\mathcal{L}_{\text{LOCC}}(\sigma_{G_{\text{in}}}^{\otimes n}).
	\end{equation}
	The output state $\sigma_{G_{\text{out}}}$ corresponds to a new graph $G_{\text{out}}=(V_{\text{in}},E_{\text{out}})$, which has the general form in \eqref{eq-network_total_state}, namely
	\begin{equation}\label{eq-ent_dist_protocol_output_state}
		\sigma_{G_{\text{out}}}=\bigotimes_{e\in E_{\text{out}}}\rho_e^{\text{out}},
	\end{equation}
	where $\{\rho_e^{\text{out}}:e\in E_{\text{out}}\}$ is a set of quantum states corresponding to the edges of $G_{\text{out}}$. Observe that the new graph $G_{\text{out}}$ has the same set $V_{\text{in}}$ of nodes as the input graph $G_{\text{in}}$. Also, in \eqref{eq-ent_dist_protocol_graphs_trans}, we allow for the protocol to take $n\in\mathbb{N}=\{1,2,\dotsc\}$ copies of the state $\sigma_{G_{\text{in}}}$.
	
	A basic example of an entanglement distribution protocol consists of the following three steps.
	\begin{enumerate}
		\item Generate elementary links in the physical graph.
		\item Perform entanglement distillation protocols, such as the one in Example~\ref{ex-ent_distill}, on (parallel) elementary links.
		\item Perform joining protocols, such as the ones in Examples~\ref{ex-ent_swap}, \ref{ex-GHZ_ent_swap}, and \ref{ex-graph_state_dist}, to create virtual links.
	\end{enumerate}
	
	Typically, the task of entanglement distribution in a quantum network is defined not only by the physical graph $G_{\text{in}}$ consisting of active elementary links, but also by a target graph $G_{\text{target}}$ and its associated quantum state $\sigma_{G_{\text{target}}}$, which has the general form
	\begin{equation}\label{eq-ent_dist_protocol_target}
		\sigma_{G_{\text{target}}}=\bigotimes_{e\in E_{\text{target}}}\rho^{\text{target}}_e.
	\end{equation}
	This is the case, for example, in a real-world scenario in which multiple user requests for entanglement are made, and thus the topology of the graph $G_{\text{target}}$ is given as part of the problem. The goal is then to have the output state $\sigma_{G_{\text{out}}}$ be close to the quantum state of the target graph. We discuss this further in Section~\ref{sec-figures_of_merit_general} below.
	
	\begin{figure}
		\centering
		\includegraphics[scale=1]{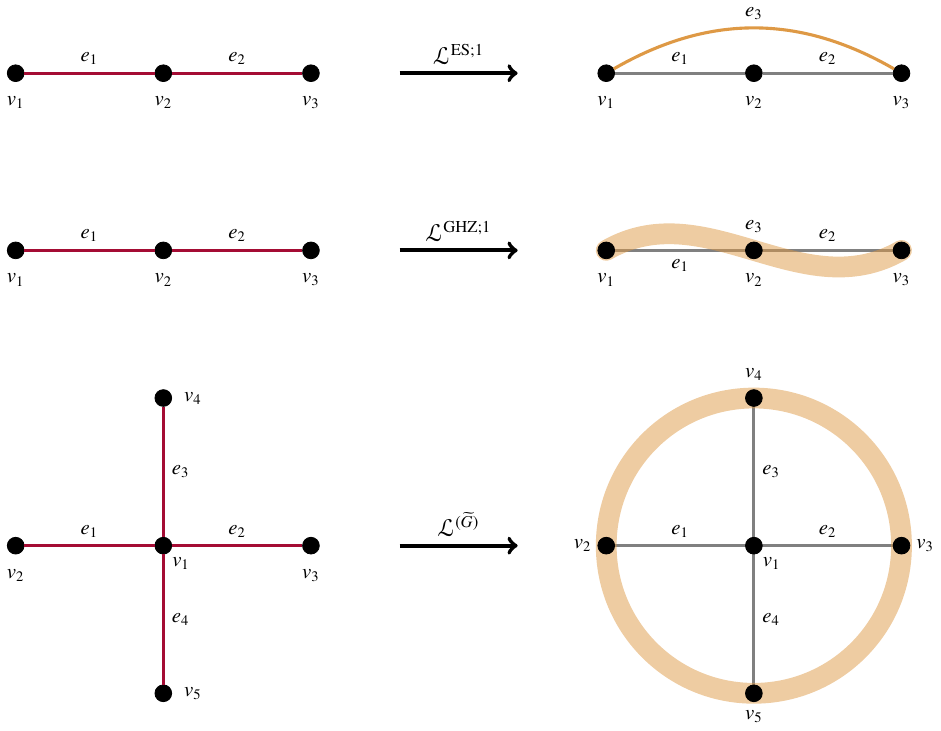}
		\caption{The three basic LOCC protocols defined in Section~\ref{sec-LOCC_channels} can be thought of as graph transformations. (Top) The physical graph of three vertices and two active elementary links can be transformed into a graph of just one edge $e_3=\{v_1,v_2\}$ via the entanglement swapping channel $\mathcal{L}^{\text{ES};1}$ defined in \eqref{eq-ent_swap_channel}. (Center) The same physical graph of three vertices and two active elementary links can be transformed into a graph of one hyperedge $e_3=\{v_1,v_2,v_3\}$ via the GHZ entanglement swapping channel $\mathcal{L}^{\text{GHZ};1}$ defined in \eqref{eq-GHZ_ent_swap_channel}. (Bottom) The five-node physical graph with four active elementary links can be transformed it into a graph with the single hyperedge $e_5=\{v_2,v_3,v_4,v_5\}$ via the graph state distribution channel $\mathcal{L}^{(\widetilde{G})}$ defined in \eqref{eq-graph_state_dist_channel}, where $\widetilde{G}$ is an arbitrary four-vertex graph.}\label{fig-network_simple_transformations}
	\end{figure}
	
	Note that all three of the LOCC protocols considered in Section~\ref{sec-LOCC_channels} (specifically, the ones in Examples~\ref{ex-ent_swap}, \ref{ex-GHZ_ent_swap}, and \ref{ex-graph_state_dist}) are nothing more than special cases of entanglement distribution protocols as we have described them here, and they can be thought of in terms of a graph transformation; see Figure~\ref{fig-network_simple_transformations}. Indeed, let us first consider the entanglement swapping channel defined in \eqref{eq-ent_swap_channel}. As shown in the top panel of Figure~\ref{fig-network_simple_transformations}, we start with the physical graph of two elementary links. The entanglement swapping channel $\mathcal{L}^{\text{ES};1}$ takes these two elementary links and transforms them to a virtual link between the end nodes. Similarly, in the central panel of Figure~\ref{fig-network_simple_transformations}, we see that the same physical graph of two elementary links can be transformed to a single hyperedge of the three nodes via the GHZ entanglement swapping channel $\mathcal{L}^{\text{GHZ};1}$ defined in \eqref{eq-GHZ_ent_swap_channel}. Finally, as shown in the bottom panel of Figure~\ref{fig-network_simple_transformations}, the graph state distribution channel $\mathcal{L}^{(\widetilde{G})}$ defined in \eqref{eq-graph_state_dist_channel} takes the graph of four physical, elementary links and transforms them to a hyperedge consisting of the four outer nodes in the network. Note that, unlike what is shown in Figure~\ref{fig-network_simple_transformations}, the elementary links used in the transformation do not have to be inactive after the transformation. We can incorporate the fresh preparation of those elementary links as part of the transformation, so that the graphs $G_{\text{out}}$ on the right-hand side have active elementary links in addition to the new virtual link.
	
	In Chapter~\ref{chap-network_QDP}, we use quantum decision processes to provide an explicit form for the LOCC channel in \eqref{eq-ent_dist_protocol_graphs_trans} corresponding to a quantum network protocol.
	
	\begin{remark}[Quantum repeaters]\label{rem-quantum_repeaters}
		In Chapter~\ref{chap-introduction}, we briefly mentioned quantum repeaters as devices, originally conceived in Refs.~\cite{BDC98,DBC99}, whose purpose is to mitigate the effects of loss and noise along a path connecting a sender and a receiver, thereby making the quantum information transmission more reliable. Specifically, quantum repeaters perform entanglement distillation \cite{BBP96,DAR96,BDSW96} (or some other form of quantum error correction) and entanglement swapping \cite{BBC+93,ZZH93} to iteratively extend the entanglement range to the desired distance. In the network setting, any node in the network that is not either a sender or a receiver can function as a quantum repeater, and it can be thought of simply as a helper node that functions in a manner similar to the original proposal, except that the repeater would in general have the additional task of routing incoming signals in the required directions using the appropriate measurement settings; see Figure~\ref{fig-grid_architecture} below for an example.~\defqed
	\end{remark}

\subsection{Figure of merit}\label{sec-figures_of_merit_general}

	How do we quantify the performance of a given quantum network protocol? Given a graph $G_{\text{in}}=(V_{\text{in}},E_{\text{in}})$ of elementary links, as described above, along with a protocol whose associated LOCC quantum channel is $\mathcal{L}_{\text{LOCC}}$, we would like the output state $\mathcal{L}_{\text{LOCC}}(\sigma_{G_{\text{in}}}^{\otimes n})$ of the protocol to be close to the state $\sigma_{G_{\text{target}}}$ corresponding to a given target graph $G_{\text{target}}=(V_{\text{in}},E_{\text{target}})$. A basic figure of merit is then simply the fidelity between $\mathcal{L}_{\text{LOCC}}(\sigma_{G_{\text{in}}}^{\otimes n})$ and $\sigma_{G_{\text{target}}}$, i.e.,
	\begin{equation}\label{eq-network_protocol_fidelity}
		F\left(\mathcal{L}_{\text{LOCC}}(\sigma_{G_{\text{in}}}^{\otimes n}),\sigma_{G_{\text{target}}}\right).
	\end{equation}
	%We can also optimize over all possible protocols and over all natural numbers $n$, so that
	%\begin{equation}
	%	F(G_{\text{in}},G_{\text{target}})\coloneqq\max_{\substack{n\in\mathbb{N}\\\mathcal{L}_{\text{LOCC}}}}F\left(\mathcal{L}_{\text{LOCC}}(\sigma_{G_{\text{in}}}^{\otimes n}),\sigma_{G_{\text{target}}}\right)
	%\end{equation}
	%is a figure of merit that depends only on the graphs $G_{\text{in}}$ and $G_{\text{target}}$ and their corresponding states $\sigma_{G_{\text{in}}}$ and $\sigma_{G_{\text{target}}}$.
	
	Now, the quantum network protocol should be such that the graph $G_{\text{out}}$ corresponding to the state $\mathcal{L}_{\text{LOCC}}(\sigma_{G_{\text{in}}}^{\otimes n})$ has the same edge set as the target graph $G_{\text{target}}$, i.e., we require $E_{\text{out}}=E_{\text{target}}$. The goal is then to calculate the fidelity between the states in \eqref{eq-ent_dist_protocol_output_state} and \eqref{eq-ent_dist_protocol_target}. Using multiplicativity of fidelity (see \eqref{eq-fidelity_multiplicative}), we immediately obtain
	\begin{align}
		F\left(\mathcal{L}_{\text{LOCC}}(\sigma_{G_{\text{in}}}^{\otimes n}),\sigma_{G_{\text{target}}}\right)&=F\left(\bigotimes_{e\in E_{\text{target}}}\rho_e^{\text{out}},\bigotimes_{e\in E_{\text{target}}}\rho_e^{\text{target}}\right)\label{eq-network_protocol_fidelity_1}\\[0.2cm]
		&=\prod_{e\in E_{\text{target}}}F\left(\rho_e^{\text{out}},\rho_e^{\text{target}}\right).\label{eq-network_protocol_fidelity_2}
	\end{align}
	
	\begin{remark}
		Protocols for achieving the task of entanglement distribution as we have described it here are a special case of general quantum network protocols that have been described in Refs.~\cite{AML16,AK17,BA17,RKB+18,Pir19,Pir19b,DBWH19,TKRW19,BAKE20}. The figure of merit considered in all of these works is the trace distance to the target state instead of fidelity (however, Ref.~\cite{DBWH19} uses fidelity). Furthermore, the primary goal in these works is to determine the optimal rates at which the target state can be created. This involves determining the maximum number of edge-disjoint paths in the network between the desired non-adjacent target nodes, which is when theorems such as Menger's theorem (Theorem~\ref{thm-Menger}) and the theorem of Tutte and Nash-Williams (Theorem~\ref{thm-Tutte_NW}) are helpful, and allow for the problem to be formulated as a flow-optimization problem, which typically have formulations as linear programs; see Ref.~\cite{Brazil2015} for more information. We also mention that multi-user requests are dealt with explicitly in \cite{Pir19b,BAKE20,CERW20}.
		
		In this thesis, we are primarily concerned with developing explicit forms for the channel $\mathcal{L}_{\text{LOCC}}$ in \eqref{eq-network_protocol_fidelity} under practical conditions rather than with optimal rates, as we outline in detail in Section~\ref{sec-network_architecture} below. However, in Section~\ref{sec-practical_figures_merit}, we explain how the methods developed here can be used to determine optimal rates using the aforementioned flow-optimization techniques.~\defqed
	\end{remark}
	
	Let us now discuss one simple strategy for achieving the desired target graph $G_{\text{target}}$ and its associated state $\sigma_{G_{\text{target}}}$. For concreteness, let us suppose that input graph $G_{\text{in}}$ contains only two-element edges, so that the elementary links correspond to bipartite entanglement (see, e.g., the left panel of Figure~\ref{fig-network_transformation}). Then, a simple strategy is to first perform entanglement distillation of the elementary links (in order to increase the fidelity to the target state), then to execute joining protocols (i.e., protocols to create virtual links) along appropriate paths. The latter requires an algorithm for finding paths in the network that can accommodate multi-user requests. For example, in the right panel of Figure~\ref{fig-network_transformation}, the virtual link $\{v_1,v_3\}$ can be established along the path $(v_1,e_1,v_2,e_2,v_3)$, and the virtual link $\{v_2,v_4\}$ can be established along the path $(v_2,e_5,v_5,e_4,v_4)$. Once a path to create a particular virtual link has been identified, then the appropriate joining protocol is performed along that path in order to create the virtual link. If we let $e'\in E_{\text{target}}$ denote this new virtual link and we let $w$ denote the path, then $\rho_{e'}^{\text{out}}$ is simply the output of the corresponding joining channel, i.e.,
	\begin{equation}
		\rho_{e'}^{\text{out}}=\mathcal{L}_{w\to e'}\left(\bigotimes_{e\in w}\rho_e^{\text{in}}\right),
	\end{equation}
	where we have made use of the notation defined in Section~\ref{sec-network_q_state} for the joining channel, which is an LOCC channel. Then, from \eqref{eq-network_protocol_fidelity_2}, the task is then to calculate the fidelity between the state $\rho_{e'}^{\text{out}}$ and the corresponding target state $\rho_{e'}^{\text{target}}$. The target state is typically a pure entangled state, i.e., $\rho_{e'}^{\text{target}}=\ket{\psi}\bra{\psi}_{e'}$, which means that the task is to calculate
	\begin{equation}\label{eq-fidelity_after_joining}
		\bra{\psi}_{e'}\mathcal{L}_{w\to e'}\left(\bigotimes_{e\in w}\rho_e^{\text{in}}\right)\ket{\psi}_{e'}.
	\end{equation}

	As examples, let us consider calculating the fidelity in \eqref{eq-fidelity_after_joining} for the three LOCC joining channels defined in Section~\ref{sec-LOCC_channels}. In each case, we find that fidelity of the output state of the channel with respect to the target state can be written in terms of the fidelities of the bipartite input states with respect to the maximally entangled state. We provide the proofs in Appendix~\ref{app-proofs}.
	
	\begin{proposition}\label{lem-ent_swap_post_fidelity}
		For all $n\geq 1$, and for all states $\rho_{AR_1^1}^1,\rho_{R_1^2R_2^1}^2,\dotsc,\rho_{R_n^2B}^{n+1}$, the fidelity with respect to the maximally entangled state of the state after entanglement swapping of $\rho_{AR_1^1}^1,\rho_{R_1^2R_2^1}^2,\dotsc,\rho_{R_n^2B}^{n+1}$ is given by
		\begin{multline}\label{eq-ent_swap_post_fidelity}
			\bra{\Phi^+}_{AB}\mathcal{L}_{A\vec{R}_1\dotsb\vec{R}_nB\to AB}^{\text{ES};n}\left(\rho_{AR_1^1}^1\otimes\rho_{R_1^2R_2^1}^2\otimes\dotsb\otimes\rho_{R_n^2B}^{n+1}\right)\ket{\Phi^+}_{AB}\\=\sum_{\substack{x_1,\dotsc,x_n=0\\z_1,\dotsc,z_n=0}}^{d-1}\bra{\Phi_{z_1+\dotsb+z_n,x_1+\dotsb+x_n}}\rho_{AR_1^1}^1\ket{\Phi_{z_1+\dotsb+z_n,x_1+\dotsb+x_n}}\bra{\Phi_{z_1,x_1}}\rho_{R_1^2R_2^1}^2\ket{\Phi_{z_1,x_1}}\dotsb\bra{\Phi_{z_n,x_n}}\rho_{R_n^2B}^{n+1}\ket{\Phi_{z_n,x_n}}.
		\end{multline}
	\end{proposition}
	
	\begin{proof}
		See Appendix~\ref{app-ent_swap_post_fidelity_pf}.
	\end{proof}
	
	%Let us now prove an analogous result for the GHZ entanglement swapping channel defined in \eqref{eq-GHZ_ent_swap_channel}.
	
	\begin{proposition}\label{lem-ent_swap_GHZ_post_fidelity}
		For all $n\geq 1$, and for all states $\rho_{AR_1^1}^1,\rho_{R_1^2R_2^1}^2,\dotsc,\rho_{R_n^2B}^{n+1}$, the fidelity with respect to the $(n+2)$-party GHZ state of the state after the GHZ entanglement swapping of $\rho_{AR_1^1}^1,\rho_{R_1^2R_2^1}^2,\dotsc,\rho_{R_n^2B}^{n+1}$ is
		\begin{multline}
			\bra{\text{GHZ}_{n+2}}\mathcal{L}_{A\vec{R}_1\dotsb\vec{R}_nB\to AR_1^1\dotsb R_n^1B}^{\text{GHZ};n}\left(\rho_{AR_1^1}^1\otimes\rho_{R_1^2R_2^1}^2\otimes\dotsb\otimes\rho_{R_n^2B}^{n+1}\right)\ket{\text{GHZ}_{n+2}}\\=\sum_{z_1,\dotsc,z_n=0}^1\bra{\Phi_{z_1+\dotsb+z_n,0}}\rho_{AR_1^1}\ket{\Phi_{z_1+\dotsb+z_n,0}}\bra{\Phi_{z_1,0}}\rho_{R_1^2R_2^1}^2\ket{\Phi_{z_1,0}}\dotsb\bra{\Phi_{z_n,0}}\rho_{R_n^2B}^{n+1}\ket{\Phi_{z_n,0}}.
		\end{multline}
	\end{proposition}
	
	\begin{proof}
		See Appendix~\ref{app-ent_swap_GHZ_post_fidelity_pf}.
	\end{proof}
	
	%Finally, let us prove the analogous result of the graph state distribution channel defined in \eqref{eq-graph_state_dist_channel}.
	
	\begin{proposition}\label{prop-graph_state_dist_post_fidelity}
		For all $n\geq 2$, all graphs $G$ with $n$ vertices, and all two-qubit states $\rho_{A_1R_1}^1$, $\rho_{A_2R_2}^2,\dotsc,\allowbreak\rho_{A_nR_n}^n$, the fidelity with respect to the graph state $\ket{G}$ of the state after the graph state distribution channel applied to $\rho_{A_1R_1}^1$, $\rho_{A_2R_2}^2,\dotsc,\allowbreak\rho_{A_nR_n}^n$ is
		\begin{multline}
			\bra{G}\mathcal{L}_{A_1^nR_1^n\to A_1^n}^{(G)}\left(\rho_{A_1R_1}^1\otimes\dotsb\otimes\rho_{A_nR_n}^n\right)\ket{G}\\=\sum_{\vec{x}\in\{0,1\}^n}\bra{\Phi_{z_1,x_1}}\rho_{A_1R_1}^1\ket{\Phi_{z_1,x_1}}\bra{\Phi_{z_2,x_2}}\rho_{A_2R_2}^2\ket{\Phi_{z_2,x_2}}\dotsb\bra{\Phi_{z_n,x_n}}\rho_{A_nR_n}^n\ket{\Phi_{z_n,x_n}},
		\end{multline}
		where the column vector $\vec{z}=(z_1,\dotsc,z_n)^{\t}$ is given by $\vec{z}=A(G)\vec{x}$, with $A(G)$ the adjacency matrix of $G$.
	\end{proposition}
	
	\begin{proof}
		See Appendix~\ref{app-graph_state_dist_post_fidelity_pf}.
	\end{proof}

\section{Practical network architecture}\label{sec-network_architecture}

	In Section~\ref{sec-ent_dist_general}, we presented the task of entanglement distribution in a quantum network in general terms. Essentially, the task corresponds to the transformation of a graph of elementary links into a graph with virtual links (shared entanglement between non-adjacent nodes in the network). We also introduced a figure of merit that quantifies (using fidelity) how close the quantum state of the transformed graph is to the quantum state of a target graph. Let us now point out three important facts that are relevant in practice but are not explicitly taken into account in the general formulation presented in Section~\ref{sec-ent_dist_general}:
	\begin{itemize}
		\item  The generation of elementary links, as well as entanglement swapping and entanglement distillation protocols, are all typically probabilistic.

		\item The quantum memories at the nodes of the network have limited coherence times.
		
		\item Every node has a limited number of quantum memories.
	\end{itemize}
	All of these practical limitations mean that the LOCC channel $\mathcal{L}_{\text{LOCC}}$ in \eqref{eq-ent_dist_protocol_graphs_trans} that defines the quantum network protocol for entanglement distribution cannot be too general. In this section, we consider these practical limitations in detail. We start with probabilistic elementary link generation in Section~\ref{sec-practical_elem_link_generation}, and we consider both ground-based and satellite-based models. Probabilistic elementary link generation implies that the subgraphs $G_{\text{in}}$ of the physical graph of the network, which are the inputs to the quantum network protocol, need to be thought of as random variables, and in Section~\ref{sec-random_subgraphs} we formalize this. Finally, in Section~\ref{sec-network_q_state_practical}, we define the quantum state of the network, as well as the joining and distillation channels, within the model of probabilistic elementary link generation. 

	%The upshot of the developments of this section is that the LOCC channel in \eqref{eq-ent_dist_protocol_graphs_trans} that defines a quantum network protocol for entanglement distribution has a natural formulation as a quantum decision process, because it allows for keeping track of all of the practical elements listed above, and this is the subject of Chapter~\ref{chap-network_QDP}.

\subsection{Elementary link generation}\label{sec-practical_elem_link_generation}

	Our basic model for probabilistic elementary link generation is illustrated in Figure~\ref{fig-net_architecture}, and it has been considered in Refs.~\cite{AJP+03,JKR+16,DKD18,KMSD19}. For every physical link in the network, there is a source station that prepares and distributes an entangled state to the corresponding nodes. In general, all of these source stations operate independently of each other, distributing entangled states as they are requested.
	
	\begin{figure}
		\centering
		\includegraphics[scale=0.65]{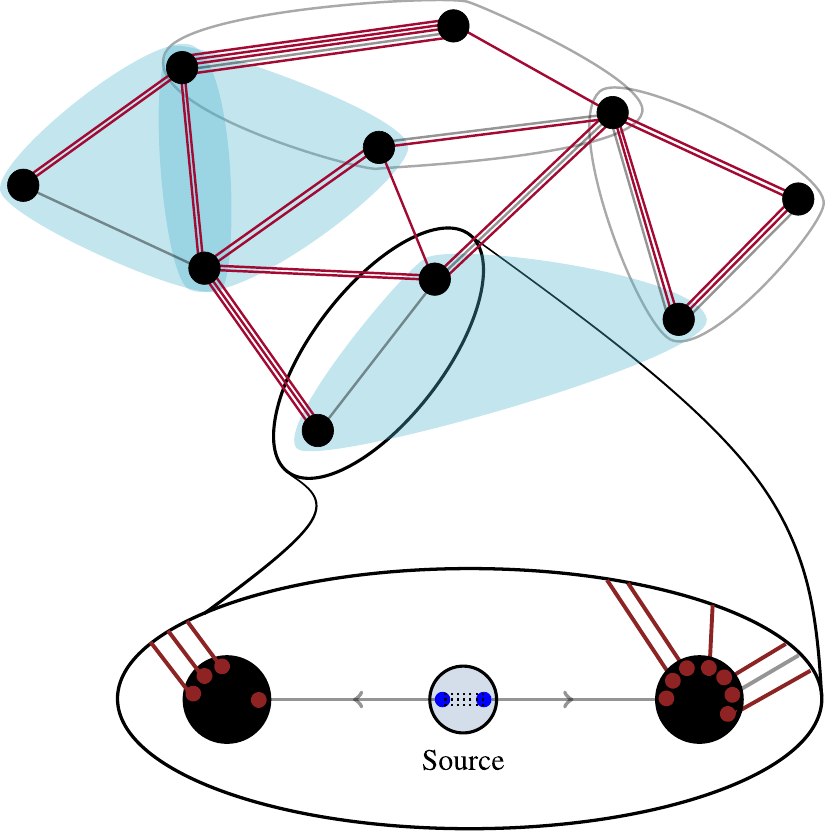}
		\caption{Our model for elementary link generation in a quantum network consists of source stations associated to every physical (elementary) link that distributes entangled states to the corresponding nodes \cite{AJP+03,JKR+16,DKD18,KMSD19}.}\label{fig-net_architecture}
	\end{figure}

	Let $G=(V,E)$ be the graph of physical links, and let $e\in E$ be an arbitrary (hyper)edge with $k\geq 2$ nodes $v_1,v_2,\dotsc,v_k$. The source corresponding to this physical link prepares a $k$-partite quantum state $\rho_e^S\equiv\rho_{A_e^{v_1}A_e^{v_2}\dotsb A_e^{v_k}}^S$. This quantum state is distributed to the nodes by sending each quantum system $A_e^{v_j}$ through a quantum channel $\mathcal{S}_{e,j}\equiv\mathcal{S}_{A_e^{v_j}}$, so that the state after transmission is
	\begin{equation}
		\rho_e^{S,\text{out}}\coloneqq\left(\mathcal{S}_{e,1}\otimes\mathcal{S}_{e,2}\otimes\dotsb\otimes\mathcal{S}_{e,k}\right)(\rho_e^S)=\mathcal{S}_e(\rho_e^S),
	\end{equation}
	where
	\begin{equation}
		\mathcal{S}_e\coloneqq\mathcal{S}_{e,1}\otimes\mathcal{S}_{e,2}\otimes\dotsb\otimes\mathcal{S}_{e,k}.
	\end{equation}
	
	After transmission from the source to the nodes, the nodes typically have to execute a \textit{heralding procedure}, which is an LOCC protocol executed by the nodes that confirms whether all of the nodes received their quantum systems; the reason for this will be clarified below, when we consider specific transmission channels. If the heralding procedure succeeds, then the nodes store their quantum systems in a quantum memory. Mathematically, the heralding procedure can be described by a quantum instrument $\{\mathcal{M}^0,\mathcal{M}^1\}$, which we recall means that $\mathcal{M}^0$ and $\mathcal{M}^1$ are completely positive trace non-increasing maps such that $\mathcal{M}^0+\mathcal{M}^1$ is trace preserving. The map $\mathcal{M}^0$ corresponds to failure of the heralding procedure, and the map $\mathcal{M}^1$ corresponds to success. The outcome of the heralding procedure can then be captured by the following transformation of the state $\rho_e^{S,\text{out}}$ to a classical-quantum state:
	\begin{equation}\label{eq-link_cq_state_initial}
		\rho_e^{S,\text{out}}\mapsto \ket{0}\bra{0}\otimes\mathcal{M}_e^0(\rho_e^{S,\text{out}})+\ket{1}\bra{1}\otimes\mathcal{M}_e^1(\rho_e^{S,\text{out}})=\ket{0}\bra{0}\otimes \widetilde{\sigma}_e(0)+\ket{1}\bra{1}\otimes\widetilde{\sigma}_e(1),
	\end{equation}
	where the classical register holds the binary outcome of the heralding procedure (`1' for success and `0' for failure) and the quantum register holds the quantum state of the nodes corresponding to the outcome. In particular,
	\begin{equation}\label{eq-initial_state_tilde_0}
		\widetilde{\sigma}_e(0)\coloneqq(\mathcal{M}_e^0\circ\mathcal{S}_e)(\rho_e^S)
	\end{equation}
	is the (unnormalized) quantum state corresponding to failure, and
	\begin{equation}\label{eq-initial_state_tilde_1}
		\widetilde{\sigma}_e(1)\coloneqq (\mathcal{M}_e^1\circ\mathcal{S}_e)(\rho_e^{S})
	\end{equation}
	is the (unnormalized) quantum state corresponding to success. The quantum states conditioned on success and failure, respectively, are defined to be
	\begin{equation}\label{eq-initial_link_states}
		\rho_e^0\coloneqq\frac{\widetilde{\sigma}_e(1)}{\Tr[\widetilde{\sigma}_e(1)]}, \quad \tau_e^{\varnothing}\coloneqq\frac{\widetilde{\sigma}_e(0)}{\Tr[\widetilde{\sigma}_e(0)]}.
	\end{equation}
	The superscript `0' in $\rho_e^0$ indicates that the quantum memories of the nodes are in their initial state immediately after success of the heradling procedure; we expand on this below. We let
	\begin{equation}\label{eq-elem_link_success_prob}
		p_e\coloneqq \Tr[\widetilde{\sigma}_e(1)]
	\end{equation}
	denote the overall probability of success of the transmission from the source and of the heralding procedure, and we call it the \textit{transmission-heralding success probability} from now on. We also let
	\begin{align}
		\widehat{\sigma}_e&\coloneqq\ket{0}\bra{0}\otimes(\mathcal{M}_e^0\circ\mathcal{S}_e)(\rho_e^S)+\ket{1}\bra{1}\otimes(\mathcal{M}_e^1\circ\mathcal{S}_e)(\rho_e^S)\\
		&=\ket{0}\bra{0}\otimes\widetilde{\sigma}_e(0)+\ket{1}\bra{1}\otimes\widetilde{\sigma}_e(1)\label{eq-elem_link_initial_cq_state}
	\end{align}
	denote the quantum state of the elementary link corresponding to the edge $e\in E$ immediately after transmission and heralding.
	
	Now, as mentioned above, once the heralding procedure succeeds, the nodes store their quantum systems in their local quantum memory. Quantum memories have been made using trapped ions \cite{SDS09}, Rydberg atoms \cite{ZMHZ10,HHH+10}, atom-cavity systems \cite{LSSH01,RR15}, NV centers in diamond \cite{DMD+13,WCT+14,NTD+16,DHR17,RGR+18,RYG+18}, individual rare-earth ions in crystals \cite{KLW+18}, and superconducting processors \cite{KLS18}. The quantum memories are in general imperfect, which means that the quantum systems decohere over time. We describe this decoherence by a quantum channel $\mathcal{N}_{e,j}$ acting on each quantum system $A_e^{v_j}$ of the elementary link, $j\in\{1,2,\dotsc,k\}$. The decoherence channel is applied at every time step in which the quantum system is in memory. The overall quantum channel acting on all of the quantum systems in the elementary link is
	\begin{equation}\label{eq-elem_link_decoherence_channel}
		\mathcal{N}_e\coloneqq\mathcal{N}_{e,1}\otimes\mathcal{N}_{e,2}\otimes\dotsb\otimes\mathcal{N}_{e,k}.
	\end{equation}
	The quantum state of the elementary link after $m$ time steps in the memories is therefore given by
	\begin{equation}\label{eq-elem_link_state_in_mem}
		\rho_e(m)\coloneqq\mathcal{N}_e^{\circ m}(\rho_e^0),
	\end{equation}
	where $\mathcal{N}_e^{\circ m}=\mathcal{N}_e\circ\mathcal{N}_e\circ\dotsb\circ\mathcal{N}_e$ ($m$ times). For a particular target/desired quantum state of the elementary link, which we assume to be a pure state $\psi_e=\ket{\psi}\bra{\psi}_e$, we let
	\begin{equation}\label{eq-fidelity_decay}
		f_m^e(\rho_e^0;\psi_e)\coloneqq \bra{\psi}_e\rho_e(m)\ket{\psi}_e=\bra{\psi}_e\mathcal{N}_e^{\circ m}(\rho_e^0)\ket{\psi}_e
	\end{equation}
	denote the fidelity of the state $\rho_e(m)$ with respect to the target state $\psi_e$. For brevity, we suppress the dependence of $f_m^e$ on the target state $\psi_e$ whenever it is understood or is unimportant. We also suppress, for brevity, the dependence of $f_m^e$ on the decoherence channels of the quantum memories.
	
	Let us now consider specific examples of transmission channels that are relevant in practice. In what follows, we suppress the dependence of quantum states, channels, etc., on the edge $e\in E$ of the physical graph when it is not important, with the understanding that all such expressions hold for an arbitrary edge.

\subsubsection{Ground-based transmission}\label{sec-network_setup_ground_based}

	The most common medium for quantum information transmission for communication purposes is photons traveling either through either free space or fiber-optic cables. These transmission media are modeled well by a bosonic pure-loss/attenuation channel $\mathcal{L}^{\eta}$ \cite{Serafini_book}, where $\eta\in(0,1]$ is the transmittance of the medium, which for fiber-optic or free-space transmission has the form $\eta=\e^{-\frac{L}{L_0}}$ \cite{SveltoBook,KJK_book,KGMS88_book}, where $L$ is the transmission distance and $L_0$ is the attenuation length of the fiber.
	
	Before the $k$ quantum systems corresponding to the source state $\rho^S$ are transmitted through the pure-loss channel, they are each encoded into $d$ bosonic modes with $d\geq 2$. A simple encoding is the following:
	\begin{equation}\label{eq-d_rail_encoding}
		\ket{0_d}\coloneqq\ket{1,0,0,\dotsc,0},\quad \ket{1_d}\coloneqq\ket{0,1,0,\dotsc,0},\quad\dotsc,\quad\ket{(d-1)_d}\coloneqq\ket{0,0,0,\dotsc,1}.
	\end{equation}
	In other words, using $d$ bosonic modes, we form a qudit quantum system by defining the standard basis elements of the associated Hilbert space by the states corresponding to a single photon in each of the $d$ modes. We let
	\begin{equation}
		\ket{\text{vac}}\coloneqq\ket{0,0,\dotsc,0}
	\end{equation}
	denote the vacuum state of the $d$ modes, which is the state containing no photons.
	
	In the context of photonic state transmission, the source state $\rho^S$ is typically of the form $\ket{\psi^S}\bra{\psi^S}$, where
	\begin{equation}\label{eq-source_state}
		\ket{\psi^S}=\sqrt{\smash[b]{p_0^S}}\ket{\text{vac}}+\sqrt{\smash[b]{p_1^S}}\ket{\psi_1^S}+\sqrt{\smash[b]{p_2^S}}\ket{\psi_2^S}+\dotsb,
	\end{equation}
	where $\ket{\psi_n^S}$ is a state vector with $n$ photons in total for each of the $k$ parties and the numbers $p_n^S\geq 0$ are probabilities, so that $\sum_{n=0}^{\infty} p_n^S=1$. For example, in the case $k=2$ and $d=2$, the following source state is generated from a parametric down-conversion process (see, e.g., Refs.~\cite{KB00,KGD+16}):
	\begin{align}
		\ket{\psi^S}&=\sum_{n=0}^{\infty}\frac{\sqrt{n+1}r^n}{\e^q}\ket{\psi_n^S},\\
		\ket{\psi_n^S}&=\frac{1}{\sqrt{n+1}}\sum_{m=0}^n (-1)^m \ket{n-m,m;m,n-m},
	\end{align}
	where $r$ and $q$ are parameters characterizing the process. One often considers a truncated version of this state as an approximation, so that \cite{KGD+16}
	\begin{multline}\label{eq-SPDC_state_2}
		\ket{\psi^S}=\sqrt{p_0}\ket{0,0;0,0}+\sqrt{\frac{p_1}{2}}(\ket{1,0;0,1}+\ket{0,1;1,0})\\+\sqrt{\frac{p_2}{3}}(\ket{2,0;0,2}+\ket{1,1;1,1}+\ket{0,2;2,0}),
	\end{multline}
	where $p_0+p_1+p_2=1$.
	
	Typically, the encoding into bosonic modes is not perfect, which means that a source state of the form \eqref{eq-source_state} is not ideal and that the desired state is given by one of the state vectors $\ket{\psi_j^S}$, and the other terms arise due to the naturally imperfect nature of the source. For example, for the state in \eqref{eq-SPDC_state_2}, the desired bipartite state is the maximally entangled state
	\begin{equation}\label{eq-photonic_source_ideal}
		\ket{\Psi^+}=\frac{1}{\sqrt{2}}(\ket{1,0;0,1}+\ket{0,1;1,0}).
	\end{equation}
	
	Once the source state is prepared, each mode is sent through the pure-loss channel. Letting
	\begin{equation}
		\mathcal{L}^{\eta,(d)}\coloneqq \underbrace{\mathcal{L}^{\eta}\otimes\mathcal{L}^{\eta}\otimes\dotsb\otimes\mathcal{L}^{\eta}}_{d\text{ times}}
	\end{equation}
	denote the quantum channel that acts on the $d$ modes of each of the $k$ systems, the overall quantum channel through which the source state $\rho^S$ is sent is
	\begin{equation}
		\mathcal{L}^{\vec{\eta},(k;d)}\coloneqq \underbrace{\mathcal{L}^{\eta_1,(d)}\otimes\mathcal{L}^{\eta_2,(d)}\otimes\dotsb\otimes\mathcal{L}^{\eta_k,(d)}}_{k\text{ times}},
	\end{equation}
	where $\vec{\eta}=(\eta_1,\eta_2,\dotsc,\eta_k)$ and $\eta_j$ is the transmittance of the medium to the $j^{\text{th}}$ node in the edge. The quantum state shared by the $k$ nodes after transmission from the source is then $\rho^{S,\text{out}}=\mathcal{L}_{\vec{\eta}}^{(k;d)}(\rho^S)$.
	
	Now, it is well known (see, e.g., Ref.~\cite{BH14}) that the action of the bosonic pure-loss channel on any linear operator $\sigma_d$ encoded in $d$ modes according to the encoding in \eqref{eq-d_rail_encoding} is equivalent to the output of an erasure channel \cite{BDS97,GBP97}. In general, a $d$-dimensional quantum erasure channel $\mathcal{E}_p^{(d)}$, with $p\in[0,1]$, is defined as follows. Consider the vector space $\mathbb{C}^d$ with orthonormal basis elements $\{\ket{0},\ket{1},\dotsc,\ket{d-1}\}$, and the vector space $\mathbb{C}^{d+1}$ with orthonormal basis elements $\{\ket{0},\ket{1},\dotsc,\ket{d-1},\ket{d}\}$. Then, for all linear operators $X\in\Lin(\mathbb{C}^d)$, $\mathcal{E}_p^{(d)}(X)=pX+(1-p)\ket{d}\bra{d}$. Note that the output is an element of $\Lin(\mathbb{C}^{d+1})$. In particular, note that the vector $\ket{d}$ is orthogonal to the input vector space $\mathbb{C}^d$.
	
	\begin{lemma}\label{lem-pure_loss_erasure}
		Let $d\geq 2$. For all linear operators $X$ acting on a $d$-dimensional Hilbert space defined by the basis elements in \eqref{eq-d_rail_encoding}, we have that
		\begin{equation}\label{eq-pure_loss_erasure}
			\mathcal{L}^{\eta,(d)}(X)=(\mathcal{L}^{\eta}\otimes\dotsb\otimes\mathcal{L}^{\eta})(\sigma_d)=\eta X+(1-\eta)\Tr[X]\ket{\text{vac}}\bra{\text{vac}}.
		\end{equation}
	\end{lemma}
	
	\begin{proof}
		To start, the bosonic pure-loss channel has the following Kraus representation \cite{FH08,SSS11}:
		\begin{equation}\label{eq-pure_loss_Kraus}
			\mathcal{L}^{\eta}(\rho)=\sum_{\ell=0}^{\infty} \frac{(1-\eta)^{\ell}}{\ell!}\sqrt{\eta}^{a^\dagger a} a^k\rho a^{k\dagger}\sqrt{\eta}^{a^\dagger a},
		\end{equation}
		where $a$ and $a^\dagger$ are the annihilation and creation operators of the bosonic mode, which are defined as $a\ket{n}=\sqrt{n}\ket{n-1}$ for all $n\geq 1$ (with $a\ket{0}=0$), and $a^\dagger\ket{n}=\sqrt{n+1}\ket{n+1}$ for all $n\geq 0$.
		
		Now, every linear operator $\sigma_d$ acting on a $d$-dimensional space that is encoded into $d$ bosonic modes as in \eqref{eq-d_rail_encoding} can be written as
		\begin{equation}
			X=\sum_{\ell,\ell'=0}^{d-1}\alpha_{\ell,\ell'}\ket{\ell_d}\bra{\ell'_d},
		\end{equation}
		for $\alpha_{\ell,\ell'}\in\mathbb{C}$. Using \eqref{eq-pure_loss_Kraus}, it is straightforward to show that
		\begin{align}
			\mathcal{L}^{\eta}(\ket{0}\bra{0})&=\ket{0}\bra{0},\\
			\mathcal{L}^{\eta}(\ket{0}\bra{1})&=\sqrt{\eta}\ket{0}\bra{1},\\
			\mathcal{L}^{\eta}(\ket{1}\bra{0})&=\sqrt{\eta}\ket{1}\bra{0},\\
			\mathcal{L}^{\eta}(\ket{1}\bra{1})&=(1-\eta)\ket{0}\bra{0}+\eta\ket{1}\bra{1}.
		\end{align}
		Using this, we find that
		\begin{align}
			(\mathcal{L}^{\eta}\otimes\dotsb\otimes\mathcal{L}^{\eta})(\ket{\ell_d}\bra{\ell'_d})=\left\{\begin{array}{l l} \eta\ket{\ell_d}\bra{\ell_d}+(1-\eta)\ket{\text{vac}}\bra{\text{vac}} & \text{if } \ell=\ell',\\ \eta\ket{\ell_d}\bra{\ell'_d} & \text{if }\ell\neq\ell'. \end{array}\right.
		\end{align}
		Therefore,
		\begin{align}
			(\mathcal{L}^{\eta}\otimes\dotsb\otimes\mathcal{L}^{\eta})(X)&=\eta\sum_{\ell,\ell'=0}^{d-1}\alpha_{\ell,\ell'}\ket{\ell_d}\bra{\ell'_d}+(1-\eta)\left(\sum_{\ell=0}^{d-1}\alpha_{\ell,\ell}\right)\ket{\text{vac}}\bra{\text{vac}}\\
			&=\eta X+(1-\eta)\Tr[X]\ket{\text{vac}}\bra{\text{vac}},
		\end{align}
		as required.
	\end{proof}
	
	After transmission from the source to the nodes, the heralding procedure typically involves doing measurements at the nodes to check whether all of the photons arrived. In the ideal case the quantum instrument $\{\mathcal{M}^0,\mathcal{M}^1\}$ for the heralding procedure corresponds simply to a measurement in the single-photon subspace defined by \eqref{eq-d_rail_encoding}. To be specific, let
	\begin{align}
		\Lambda^1&\coloneqq\Pi^{(d)}\coloneqq\ket{0_d}\bra{0_d}+\ket{1_d}\bra{1_d}+\dotsb+\ket{(d-1)_d}\bra{(d-1)_d},\label{eq-photonic_heralding_example_1}\\
		\Lambda^0&\coloneqq \mathbbm{1}_{\mathcal{H}_d}-\Lambda^0,\label{eq-photonic_heralding_example_2}
	\end{align}
	where $\Pi^{(d)}$ is the projection onto the $d$-dimensional single-photon subspace defined by \eqref{eq-d_rail_encoding}, and $\mathbbm{1}_{\mathcal{H}_d}$ is the identity operator of the full Hilbert space $\mathcal{H}_d$ of $d$ bosonic modes. Then, letting $\vec{x}\in\{0,1\}^k$ and defining
	\begin{equation}\label{eq-photonic_heralding_example_3}
		\Lambda^{\vec{x}}\coloneqq\Lambda^{x_1}\otimes\Lambda^{x_2}\otimes\dotsb\otimes\Lambda^{x_k},
	\end{equation}
	the maps $\mathcal{M}^0$ and $\mathcal{M}^1$ have the form
	\begin{align}
		\mathcal{M}^1(\cdot)&=\Lambda^{\vec{1}}(\cdot)\Lambda^{\vec{1}},\label{eq-photonic_heralding_1}\\
		\mathcal{M}^0(\cdot)&=\sum_{\substack{\vec{x}\in\{0,1\}^k\\\vec{x}\neq\vec{1}}}\Lambda^{\vec{x}}(\cdot)\Lambda^{\vec{x}}.\label{eq-photonic_heralding_0}
	\end{align}
	These maps correspond to perfect photon-number-resolving detectors. However, the detectors are typically noisy due to dark counts and other imperfections (see, e.g., Refs.~\cite{KGD+16}), so that in practice the maps $\mathcal{M}^0$ and $\mathcal{M}^1$ will not have the ideal forms presented in \eqref{eq-photonic_heralding_1} and \eqref{eq-photonic_heralding_0}.
	
	\begin{figure}
		\centering
		\includegraphics[width=0.95\textwidth]{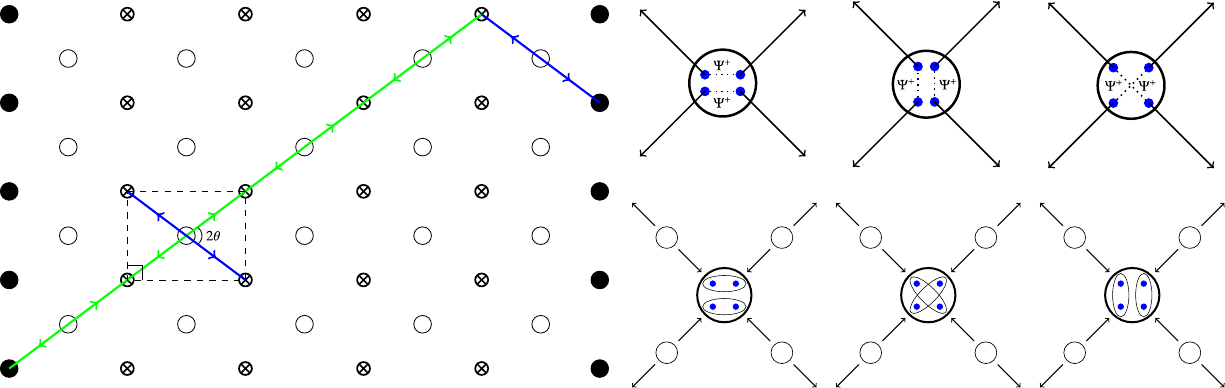}
		\caption{The ground-based quantum network architecture considered in Ref.~\cite{DKD18}. The solid black dots are the senders and the receivers, while the circular dots represent the source stations and the symbols ``\textbf{$\otimes$}'' represent measurement stations. As shown in the top three panels on the right, the source stations can send Bell states in one of three configurations. Similarly, as shown in the bottom three panels on the right, the measurement stations can perform a Bell-basis measurement in one of three configurations, for the purpose of entanglement swapping and routing.}\label{fig-grid_architecture}
	\end{figure}
	
	Also, in practice, for the photonic state transmission that we are considering here, the combination of heralding and storage of the qubit in quantum memory occurs as follows. Each node has locally an optical Bell measurement device. First, a memory-photon entangled state is generated, and then a Bell measurement is performed on the photon from the memory-photon pair and the incoming photon from the source. This strategy allows for direct knowledge about the arrival of the photon, which is then communicated to the neighboring node (see, e.g., Ref. \cite{DHR17}). At the same time, conditioned on the success of the Bell measurement, the state of the photonic qubit is transferred to the memory qubit. Linear-optical Bell measurements are limited to a success probability of 50\% \cite{LCS99,CL01,VY99}, although higher success probabilities are in principle possible using nonlinear elements or by increasing the number of photons \cite{KKS01,LSSH01,KKS02,Grice11,WHF+16}.
	
	If the source produces the ideal quantum state, such as the state in \eqref{eq-photonic_source_ideal}, and if the heralding procedure is also ideal, then using \eqref{eq-pure_loss_erasure} we obtain
	\begin{align}
		\widetilde{\sigma}(1)&=\eta_1\eta_2\Psi^+,\\
		\widetilde{\sigma}(0)&=\eta_1(1-\eta_2)\frac{\Pi^{(2)}}{2}\otimes\ket{\text{vac}}\bra{\text{vac}}+(1-\eta_1)\eta_2\ket{\text{vac}}\bra{\text{vac}}\otimes\frac{\Pi^{(2)}}{2}\nonumber\\
		&\qquad\qquad\qquad\qquad\qquad+(1-\eta_1)(1-\eta_2)\ket{\text{vac}}\bra{\text{vac}}\otimes\ket{\text{vac}}\bra{\text{vac}},
	\end{align}
	which means that the transmission-heralding success probability as defined in \eqref{eq-elem_link_success_prob} is simply $p=\Tr[\widetilde{\sigma}(1)]=\eta_1\eta_2$.
	
	In Figure~\ref{fig-grid_architecture}, we show the ground-based network architecture introduced in Ref.~\cite{DKD18}, in which the senders and receivers are at two ends of a network of source stations and measurement stations arranged in a grid-like fashion. The measurement stations perform a Bell-basis measurement as part of the entanglement swapping protocol presented in Example~\ref{ex-ent_swap}, and the measurements can be done in different configurations in order to allow for different paths to be used.
	
	\begin{remark}[Multiplexing]\label{rem-multiplexing}
		In practice, in order to increase the transmission-heralding success probability, multiplexing strategies are used. The term ``multiplexing'' here refers to the use of a single transmission channel to send multiple signals simultaneously, with the signals being encoded into distinct (i.e., orthogonal) frequency modes; see, e,g., Ref.~\cite{GKF+15}. If $M\geq 1$ distinct frequency modes are used, then the source state being transmitted is $(\rho^S)^{\otimes M}$. If $p$ denotes the probability that any single one of the signals is received and heralded successfully, then the probability that at least one of the $M$ signals is received and heralded successfully is $1-(1-p)^M$.
	\end{remark}

\subsubsection{Transmission from satellites}\label{sec-sat_architecture}

	Let us now consider the model of elementary link generation proposed in Ref.~\cite{KBD+19}, in which the entanglement sources are placed on satellites orbiting the earth.
	
	Satellites are one of the best methods for achieving global-scale quantum communication with current and near-term resources \cite{AJP+03,JH13,BAL17,Sim17,Cubesat2017,Nanobob2018}. Several proposals for satellite-based quantum networks have been made that use satellite-to-ground transmission, ground-to-satellite transmission, or both \cite{AJP+03,Bon09,ESH+12,BMH+13,BBM+15,TCT+16,Bedington2016nanosatellite,Cubesat2017,Nanobob2018,HMG18,HMG19,VLBKL19}. Recent experiments \cite{TCT+16,LYL+2017,YCL17,LCL+17,TCF+17,Ren17satteleport,LCH+18,CCD+18} (see also Ref.~\cite{LV19} for a review) between a handful of nodes has opened up the possibility of building a global-scale quantum internet using satellites.
	
	\begin{figure}
		\centering
		\includegraphics[scale=0.4]{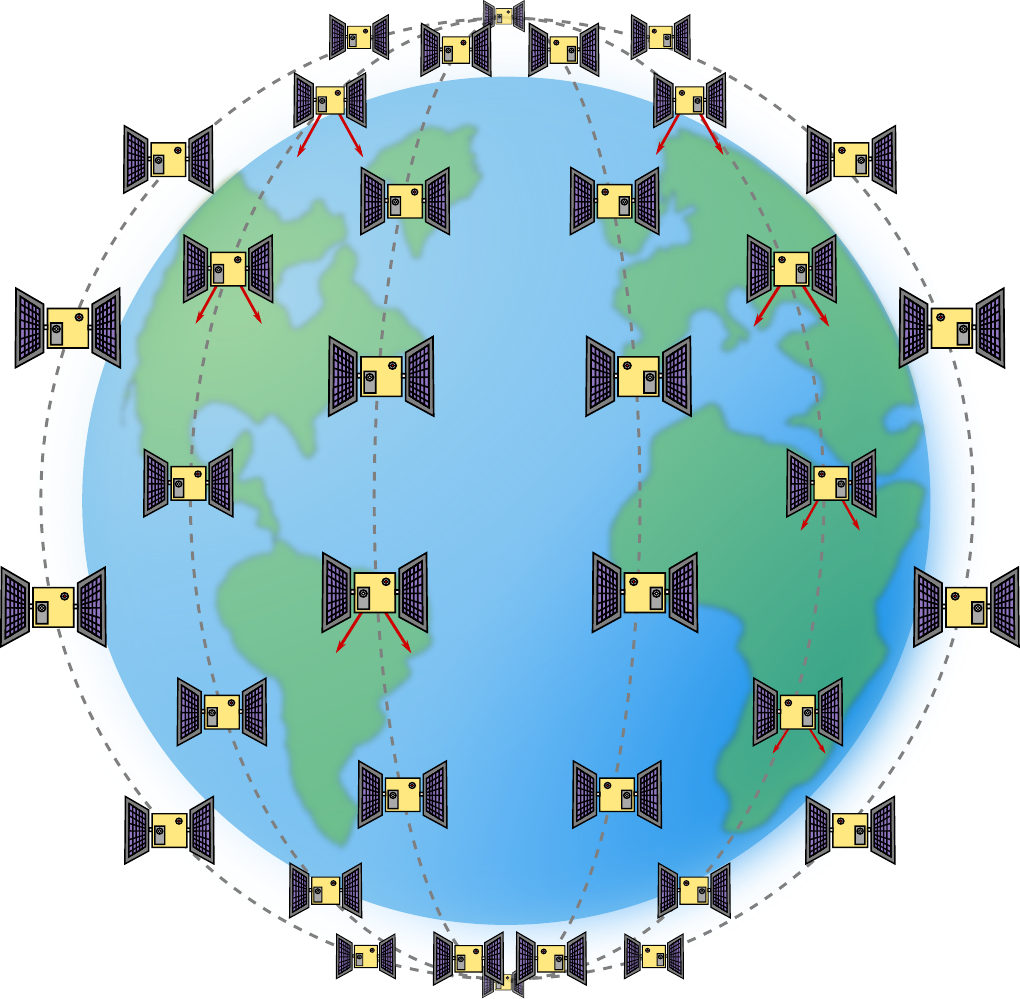}
		\caption{Depiction of the satellite constellation proposed in Ref.~\cite{KBD+19}, which consists of $N_R$ equally-spaced rings of satellites in polar orbits and $N_S$ satellites per ring. The satellites contain a source of bipartite entanglement for establishing elementary links between ground stations.}\label{fig-sat_architecture}
	\end{figure}

	The global-scale orbiting constellation of satellites proposed in Ref.~\cite{KBD+19} is illustrated in Figure~\ref{fig-sat_architecture}. There are $N_R$ equally spaced rings of satellites in polar orbits and $N_S$ equally spaced satellites in each ring, so that there are $N_RN_S$ satellites in total, all of which are at the same altitude. This type of satellite constellation falls into the general class of Walker star constellations \cite{Walker70}, and it is similar to the Iridium communications-satellite constellation \cite{Leopold92,PRFT99}. Prior works have examined various other types of satellite constellations \cite{Luders61,Walker70,LA98,LFP98}, and the recent Starlink constellation \cite{Handley18} is being used to provide a global satellite-based (classical) internet service.
	
	As mentioned, the satellites act as source stations that transmit pairs of entangled photons to ground stations for the purpose of establishing elementary links. Note that the satellites could alternatively be used as quantum repeaters \cite{LKB20,GSH+20}, which would require ground-to-satellite transmission \cite{BMH+13}. The photon sources on the satellites produce polarization-entangled photon pairs. State-of-the-art sources of entangled photons are capable of producing polarization-entangled photons on a chip with a fidelity up to 0.97~\cite{HKK+13,matsuda2012monolithically,Kang:16,KRR+17}.
	
	When modeling photon transmission from satellites to ground stations, we must take into account background photons. Here, we analyze the scenario in which a source generates an entangled photon pair and distributes the individual photons to two parties, Alice ($A$) and Bob ($B$). We allow the distributed photons to mix with background photons from an uncorrelated thermal source. Also, as before, we use the bosonic encoding defined in \eqref{eq-d_rail_encoding}, but we stick to $d=2$, i.e., qubit source states and thus bipartite elementary links. In this scenario, it is common for the two modes to represent the polarization degrees of freedom of the photons, so that
	\begin{equation}
		\ket{H}\equiv\ket{0_2}=\ket{1,0},\quad \ket{V}\equiv\ket{1_2}=\ket{0,1}
	\end{equation}
	represent the state of one horizontally and vertically polarized photon, respectively.
	
	Now, consider a tensor product of thermal states for the horizontal and vertical polarization modes:
	\begin{equation}
    	\Theta^{\overline{n}_H}\otimes\Theta^{\overline{n}_V}=\Bigg(\sum_{n=0}^\infty\left(\frac{\overline{n}_H^n}{(\overline{n}_H+1)^{n+1}}\right) \dyad{n}\Bigg)\otimes\left(\sum_{n=0}^\infty\left(\frac{\overline{n}_V^n}{(\overline{n}_V+1)^{n+1}}\right)\dyad{n}\right),
	\end{equation}
	where $\overline{n}_k$ is the average number of photons in the thermal state for the polarization mode $k$. We assume this state comes from an incoherent source with no polarization preference (e.g., the sun), so that $\overline{n}_H=\overline{n}_V=\overline{n}/2$ for some $\overline{n}\geq 0$. Furthermore, we assume some (non-polarization) filtering procedure, which reduces the number of background thermal photons, such that $\overline{n}\ll 1$. We then rewrite the above state to first order in the small parameter $\overline{n}$:
	\begin{align}
    	\Theta^{\frac{\overline{n}}{2}}\otimes\Theta^{\frac{\overline{n}}{2}}&\approx\left(\left(1-\frac{\overline{n}}{2}\right)\dyad{0} +\frac{\overline{n}}{2}\dyad{1}\right)\otimes\left(\left(1-\frac{\overline{n}}{2}\right)\dyad{0} +\frac{\overline{n}}{2}\dyad{1}\right) \\
    	&\approx(1-\overline{n})\dyad{\text{vac}}+\frac{\overline{n}}{2}\left(\dyad{H}+\dyad{V}\right).\label{eq-therm_expand}
	\end{align}
	We thus define our approximate thermal background state as
	\begin{equation}\label{eq-approx_th_state}
		\widetilde{\Theta}^{\overline{n}}\coloneqq(1-\overline{n})\dyad{\text{vac}}+\frac{\overline{n}}{2}\left(\dyad{H}+\dyad{V}\right).
	\end{equation}
	The transmission channel from the source to the ground is then approximately
	\begin{equation}\label{eq-noisy_transmission_channel}
		\mathcal{L}^{\eta_{\text{sg}},\overline{n}}(\rho_{A_1A_2})\coloneqq\Tr_{E_1E_2}\left[\left(U_{A_1E_1}^{\eta_{\text{sg}}}\otimes U^{\eta_{\text{sg}}}_{A_2E_2}\right)\left(\rho_{A_1A_2}\otimes\widetilde{\Theta}_{E_1E_2}^{\overline{n}}\right)\left(U^{\eta_{\text{sg}}}_{A_1E_1}\otimes U^{\eta_{\text{sg}}}_{A_2E_2}\right)^\dagger\right],
	\end{equation}
	where $U^{\eta_{\text{sg}}}$ is the beamsplitter unitary (see, e.g., Ref.~\cite{Serafini_book}), and $A_1$ and $A_2$ refer to the horizontal and vertical polarization modes, respectively, of the dual-rail quantum system being transmitted; similarly for $E_1$ and $E_2$. Note that for $\overline{n}=0$, the transformation in \eqref{eq-noisy_transmission_channel} reduces to the one in \eqref{eq-pure_loss_erasure} with $d=2$. 
	
	The transmittance $\eta_{\text{sg}}$ generally depends on atmospheric conditions (such as turbulence and weather conditions) and on orbital parameters (such as altitude and zenith angle) \cite{Vogel2017weather,Vogel2019atmlinks,LKB+19}. In general, we can decompose $\eta_{\text{sg}}$ as 
	\begin{equation}
		\eta_{\text{sg}}=\eta_{\text{fs}}\eta_{\text{atm}}\label{eq-eta_sg}
	\end{equation}
	where $\eta_{\text{fs}}$ is the free-space transmittance and $\eta_{\text{atm}}$ is the atmospheric transmittance. Free-space loss occurs due to diffraction (i.e., beam broadening) over the channel and due to the use of finite-sized apertures at the receiving end. These effects cause $\eta_{\text{fs}}$ to scale as the inverse squared distance in the far-field regime. Atmospheric loss occurs due to absorption and scattering in the atmosphere and scales exponentially with distance as a result of the Beer-Lambert law \cite{BH08,Andrews2005randomedia,KJK_book}. However, since atmospheric absorption is relevant only in a layer of thickness 10--20~km above the earth's surface \cite{KJK_book}, free-space diffraction is the main source of loss in satellite-to-ground transmission. In order to characterize the free-space and atmospheric transmittances with simple analytic expressions, we ignore turbulence-induced effects in the lower atmosphere, such as beam profile distortion, beam broadening (prominent for ground-to-satellite transmission \cite{KJK_book,BMH+13}), and beam wandering (see, e.g., Ref. \cite{Vogel2019atmlinks}). Note that turbulence effects can be corrected using classical adaptive optics \cite{KJK_book}. We also ignore the inhomogeneous density profile of the atmosphere, which can lead to path elongation effects at large zenith angles. A comprehensive analysis of loss can be found in Refs.~\cite{Andrews2005randomedia,Vogel2019atmlinks}.
	
	Consider the spatial mode for an optical beam traveling a distance $L$ between the sender and receiver, with a circular receiving aperture of radius $r$. Then, the free-space transmittance $\eta_{\text{fs}}$ is given by \cite{SveltoBook}
	\begin{equation}\label{eq-fs_transmittance}
		\eta_{\text{fs}}(L)=1-\exp\left(-\frac{2r^2}{w(L)^2}\right),  
	\end{equation}  
	where
	\begin{equation}
		w(L)\coloneqq w_{0}\sqrt{1+\left(\frac{L}{L_{R}}\right)^2}
	\end{equation}
	is the beam waist at a distance $L$ from the focal region ($L=0$), $L_{R}\coloneqq\pi w_{0}^2\lambda^{-1}$ is the Rayleigh range, $\lambda$ is the wavelength of the optical mode, and $w_0$ is the initial beam-waist radius.

	In order to characterize $\eta_{\text{atm}}$, we model the atmosphere as a homogeneous absorptive layer of finite thickness. Uniformity of the atmospheric layer then implies uniform absorption (at a given wavelength), such that $\eta_{\text{atm}}$ depends only on the optical path traversed through the atmosphere. Under these assumptions, and using the Beer-Lambert law \cite{BH08}, for a zenith angle of $\zeta$ we have that
	\begin{equation}\label{eq-atmospheric_transmittance_zenith}
		\eta_{\text{atm}}(L,h)=\left\{\begin{array}{l l} (\eta_{\text{atm}}^{\text{zen}})^{\sec\zeta} & \text{if } -\frac{\pi}{2}<\zeta<\frac{\pi}{2},\\[0.2cm] 0 & \text{if } |\zeta|\geq\frac{\pi}{2}, \end{array}\right.
	\end{equation}
	with $\eta_{\text{atm}}^{\text{zen}}$ the transmittance at zenith ($\zeta=0$). For $|\zeta|>\frac{\pi}{2}$, we set $\eta_{\text{atm}}=0$, because the satellite is over the horizon and thus out of sight. The zenith angle $\zeta$ is given by
	\begin{equation}\label{eq-sec_zen}
		\cos\zeta=\frac{h}{L}-\frac{1}{2}\frac{L^2-h^2}{R_{\oplus}L}
	\end{equation}
	for a circular orbit of altitude $h$, with $R_{\oplus}\approx 6378$~km being the earth's radius.
	
	Note that the model of atmospheric transmittance given by Eq.~\eqref{eq-atmospheric_transmittance_zenith} and Eq.~\eqref{eq-sec_zen} is quite accurate for small zenith angles \cite{KJK_book}. However, for satellite-to-ground transmission at or near the horizon (i.e., for $\zeta=\pm\pi/2$), more exact methods relying on the standard atmospheric model must be used \cite{Vogel2019atmlinks}. In practice, it makes sense to set $\eta_{\text{atm}}=0$ at large zenith angles, because the loss will typically be too high for the link to be practically useful.
	
	For a source state $\rho_{AB}^S$, with $A\equiv A_1A_2$ and $B\equiv B_1B_2$, the quantum state shared by Alice and Bob after transmission of the state $\rho_{AB}^S$ from the satellite to the ground stations is
	\begin{equation}\label{eq-transmission_noisy_output}
		\rho_{AB}^{S,\text{out}}=\left(\mathcal{L}_A^{\eta_{\text{sg}}^{(1)},\overline{n}_1}\otimes\mathcal{L}_B^{\eta_{\text{sg}}^{(2)},\overline{n}_2}\right)(\rho_{AB}^S),
	\end{equation}	
	where $\eta_{\text{sg}}^{(1)}$ and $\eta_{\text{sg}}^{(2)}$ are the transmittances to the ground stations and $\overline{n}_1$ and $\overline{n}_2$ are the corresponding thermal background noise parameters. In Chapter~\ref{chap-sats}, we look at a specific example of a source state $\rho_{AB}^S$, and thus provide an explicit form for the state $\rho_{AB}^{S,\text{out}}$. We also consider the heralding procedure defined by \eqref{eq-photonic_heralding_example_1}--\eqref{eq-photonic_heralding_0}, and thus provide explicit forms for the states $\rho^0$ and $\tau^{\varnothing}$ corresponding to success and failure, respectively, of the heralding procedure.

\subsection{Random subgraphs}\label{sec-random_subgraphs}

	Having established that elementary link generation is in practice probabilistic, we see that, given a graph $G=(V,E)$ corresponding to the physical (elementary) links of a quantum network, the subgraphs of active elementary links are in general random, and thus vary from time to time. In particular, then, the graph $G_{\text{in}}$ in \eqref{eq-ent_dist_protocol_graphs_trans} that is the input to an entanglement distribution protocol is in general random. For every $e\in E$, we therefore define the binary random variable $X_e$, called the \textit{elementary link status random variable}, such that
	\begin{equation}
		p_e\coloneqq\Pr[X_e=1]\quad\Longrightarrow\quad\Pr[X_e=0]=1-p_e,
	\end{equation}
	where $p_e\in[0,1]$ is the transmission-heralding success probability for the elementary link corresponding to the edge $e\in E$, as defined in \eqref{eq-elem_link_success_prob}. If $G=(V,E,c)$ is a multigraph and there are $ c(e)$ parallel edges connecting the vertices in $e$ (recall the discussion in Section~\ref{sec-graph_theory}), then we use $X_{e^j}$ to denote the random variable for the $j^{\text{th}}$ parallel edge, and we define
	\begin{equation}
		N_e\coloneqq\sum_{j=1}^{ c(e)} X_{e^j}
	\end{equation}
	to be the random variable for the total number of parallel edges connecting the vertices in $e$. We let $p_{e^j}\coloneqq\Pr[X_{e^j}=1]$ be the probability that the $j^{\text{th}}$ parallel edge is active.
	
	We define all of the random variables $X_{e^j}$, $e\in E$, $1\leq j\leq c(e)$, to be mutually independent, so that
	\begin{equation}
		\Pr[X_{e^j}=\alpha,\,X_{{e'}^{j'}}=\beta]=\Pr[X_{e^j}=\alpha]\cdot\Pr[X_{{e'}^{j'}}=\beta]
	\end{equation}
	for all distinct combinations of $(e,j)$ and $(e',j')$ and for all $\alpha,\beta\in\{0,1\}$.
	
	Let $\vec{p}\coloneqq (p_{e^j}:e\in E,\,1\leq j\leq c(e))$, and let $\Xi_G=\{0,1\}^{|G|}$ denote the set of all \textit{configurations} of the graph $G$, where we recall that $|G|$ denotes the size of the graph, which is by definition the total number of edges in $G$. By definition, every $\vec{x}=(x^{e^j}:e\in E,\,1\leq j\leq c(e))\in\Xi_G$, $x^{e^j}\in\{0,1\}$, tells us whether the elementary link corresponding to the edge $e$ is active ($x^{e^j}=1$) or inactive ($x^{e^j}=0$). Then, we can think of the random variables $X_{e^j}$ as functions $X_{e^j}:\Xi_G\to\{0,1\}$ such that
	\begin{equation}
		X_{e^j}(\vec{x})=x^{e^j}.
	\end{equation}
	
	We define $G_{\vec{p}}$ to be a random variable such that, for all $\vec{x}\in\Xi_G$, $G_{\vec{p}}(\vec{x})=(V(\vec{x}),E(\vec{x}),c)$ is the subgraph of $G$ consisting only of the edges corresponding to active elementary links, meaning that
	\begin{equation}
		\{e^j:x^{e^j}=1,\,e\in E,\,1\leq j\leq c(e)\}
	\end{equation}
	is the set of edges and $V(\vec{x})\subseteq V$ is the corresponding set of vertices. We can think of $G_{\vec{p}}$ as being a ``random graph'', in the sense that every random sample from $\Xi_G$ gives us a subgraph of $G$. Then, because all of the random variables $X_{e^j}$ are defined to be mutually independent, the probability of a particular subgraph $G_{\vec{p}}(\vec{x})$, with $\vec{x}\in\Xi_G$, is simply the product of the probabilities of the individual elementary links, i.e.,
	\begin{align}
		\Pr[G_{\vec{p}}(\vec{x})]&=\Pr[X_{e^j}=x^{e^j}:e\in E,\,1\leq j\leq c(e)]\\
		&=\prod_{e\in E}\prod_{j=1}^{ c(e)}\Pr[X_{e^j}=x^{e^j}]\\
		&=\prod_{e\in E}\prod_{j=1}^{ c(e)}p_{e^j}^{x^{e^j}}(1-p_{e^j})^{1-x^{e^j}}.\label{eq-graph_config_prob}
	\end{align}
	
	For every subset $E'\subseteq E$, we define the random variable
	\begin{equation}
		X_{E'}\coloneqq \prod_{e\in E'}\prod_{j=1}^{ c(e)}X_{e^j},
	\end{equation}
	which is simply the product of the status random variables for the elements in $E'$, so that
	\begin{equation}
		\Pr[X_{E'}=1]=\mathbb{E}[X_{E'}]=\prod_{e\in E'}\prod_{j=1}^{ c(e)}p_{e^j}
	\end{equation}
	is the probability that all of the elementary links specified by $E'$ are active. Observe then that
	\begin{equation}
		\Pr[G_{\vec{p}}(\vec{x})]=\Pr[X_{E(\vec{x})}=1]\cdot\Pr[X_{E\setminus E(\vec{x})}=0].
	\end{equation}
	
	Let
	\begin{equation}\label{eq-num_active_elem_links}
		L(G_{\vec{p}})\coloneqq\sum_{e\in E}\sum_{j=1}^{c(e)}X_{e^j}
	\end{equation}
	be the random variable for the number of active elementary links in $G$, so that for all configurations $\vec{x}\in\Xi_G$ of the graph,
	\begin{equation}\label{eq-graph_num_active_edges}
		L(G_{\vec{p}})(\vec{x})=\sum_{e\in E}\sum_{j=1}^{ c(e)}X_{e^j}(\vec{x})=\sum_{e\in E}\sum_{j=1}^{ c(e)}x^{e^j}=|G_{\vec{p}}(\vec{x})|
	\end{equation}
	is the size of the subgraph $G_{\vec{p}}(\vec{x})$. Note that
	\begin{equation}
		\mathbb{E}[L(G_{\vec{p}})]=\sum_{e\in E}\sum_{j=1}^{ c(e)}p_{e^j}.
	\end{equation}
	
	Recall the definition of the size of the largest connected component of the graph $G$ from \eqref{eq-largest_connected_component}. Now, for the random variable $G_{\vec{p}}$, we define $S^{\max}(G_{\vec{p}})$ to be the random variable for the largest connected component of $G_{\vec{p}}$, so that 
	\begin{equation}\label{eq-random_subgraph_largest_cluster_RV}
		S^{\max}(G_{\vec{p}})(\vec{x})\coloneqq S^{\max}(G_{\vec{p}}(\vec{x}))=\max\{|C|:C\in\sfrac{G_{\vec{p}}(\vec{x})}{\leftrightarrow}\}
	\end{equation}
	for all $\vec{x}\in\Xi_G$. Recalling that the connected components of a graph partition the graph, we immediately have the following simple fact.
	
	\begin{lemma}\label{lem-largest_cluster_UB}
		For all graphs $G=(V,E,c)$ and all vectors $\vec{p}=(p_{e^j}:e\in E,\,1\leq j\leq c(e))$ of probabilities, the size of the largest connected component of $G_{\vec{p}}$ is bounded from above by the size of $G_{\vec{p}}$, i.e., 
		\begin{equation}
			S^{\max}(G_{\vec{p}})\leq L(G_{\vec{p}}).
		\end{equation}
	\end{lemma}
	
	\begin{proof}
		Let $\vec{x}\in\Xi_G$ be an arbitrary configuration of $G$. Then,
		\begin{equation}
			\max\{|C|:C\in\sfrac{G_{\vec{p}}(\vec{x})}{\leftrightarrow}\}\leq\sum_{C\in\sfrac{G_{\vec{p}}(\vec{x})}{\leftrightarrow}}|C|,
		\end{equation}
		because the size of a connected component is always a natural number. Then, because the connected components partition the graph, the sum of the sizes of the connected components is equal simply to the size of the graph, i.e.,
		\begin{equation}
			\sum_{C\in\sfrac{G_{\vec{p}}(\vec{x})}{\leftrightarrow}} |C|=|G_{\vec{p}}(\vec{x})|=L(G_{\vec{p}})(\vec{x}),
		\end{equation}
		where the last equality holds due to \eqref{eq-graph_num_active_edges}. We thus have that $S^{\max}(G_{\vec{p}})(\vec{x})\leq L(G_{\vec{p}})(\vec{x})$ for all $\vec{x}\in\Xi_G$, which means that $S^{\max}(G_{\vec{p}})\leq L(G_{\vec{p}})$, as required.
	\end{proof}
	
	\begin{remark}
		Although related, the notion of ``random graph'' that we consider here is different from the usual notion of a random graph (see, e.g., Ref.~\cite[Chapter~13]{BM08_book}), in which only the number of vertices (and sometimes the number of edges) is fixed, but the vertices can be connected arbitrarily. Here, by starting with the underlying physical graph $G$, the topology is fixed, so that vertices can only be connected based on the topology of $G$.~\defqed
	\end{remark}
	
	%\begin{remark}\label{rem-edge_capacity_flow}
	%	In the context of flow problems in graph theory, the quantity $c(e)$ is referred to as the \textit{capacity} of $e$, and the quantity $N_e$ is referred to as the \textit{flow} along $e$.~\defqed
	%\end{remark}

\subsection{Quantum states and channels revisited}\label{sec-network_q_state_practical}

	Given the practical elements of elementary link generation that we have described so far, let us now revisit the development in Section~\ref{sec-network_q_state} on quantum states and channels in a quantum network.
	
	Let $G=(V,E,c)$ be a graph corresponding to the physical (elementary) links of a quantum network. Let $e^j$, $e\in E$, $1\leq j\leq c(e)$, be an arbitrary edge in the graph. Recall the classical-quantum state $\widehat{\sigma}_{e^j}$ of the elementary link corresponding to the edge $e^j$ defined in \eqref{eq-elem_link_initial_cq_state}. From this, we have that the overall quantum state of the network, based on the definition in \eqref{eq-network_total_state}, is defined as
	\begin{equation}\label{eq-network_initial_cq_state}
		\widehat{\sigma}_{G_{\vec{p}}}\coloneqq\bigotimes_{e\in E}\bigotimes_{j=1}^{ c(e)}\widehat{\sigma}_{e^j}=\sum_{\vec{x}\in\Xi_G}\ket{\vec{x}}\bra{\vec{x}}_{X_G}\otimes \widetilde{\sigma}_{G_{\vec{p}}}(\vec{x}).
	\end{equation}
	Here, $X_G=\{X_{e^j}:e\in E,\,1\leq j\leq c(e)\}$ is the collection of all link status random variables and
	\begin{equation}
		\widetilde{\sigma}_{G_{\vec{p}}}(\vec{x})\coloneqq\bigotimes_{e\in E}\bigotimes_{j=1}^{ c(e)}\widetilde{\sigma}_{e^j}(x^{e^j}),
	\end{equation}
	where we recall that $x^{e^j}\in\{0,1\}$ and
	\begin{align}
		\widetilde{\sigma}_{e^j}(0)&=(\mathcal{M}_{e^j}^0\circ\mathcal{S}_{e^j})(\rho_{e^j}^S),\\
		\widetilde{\sigma}_{e^j}(1)&=(\mathcal{M}_{e^j}^1\circ\mathcal{S}_{e^j})(\rho_{e^j}^S).
	\end{align}
	(Recall \eqref{eq-initial_state_tilde_0} and \eqref{eq-initial_state_tilde_1}.) Observe that
	\begin{equation}
		\Tr[\widetilde{\sigma}_{e^j}(1)]=p_{e^j}=\Pr[X_{e^j}=1],
	\end{equation}
	and thus the probability of a configuration $\vec{x}\in\Xi_G$, as in \eqref{eq-graph_config_prob}, is given by
	\begin{equation}
		\Tr[\widetilde{\sigma}_{G_{\vec{p}}}]=\Pr[G_{\vec{p}}(\vec{x})]=\prod_{e\in E}\prod_{j=1}^{ c(e)}\Tr[\widetilde{\sigma}_{e^j}(x^{e^j})].
	\end{equation}
	
	The \textit{expected quantum state of the network} is
	\begin{equation}
		\sigma_{G_{\vec{p}}}\coloneqq\Tr_{X_G}[\widehat{\sigma}_{G_{\vec{p}}}]=\sum_{\vec{x}\in\Xi_G}\widetilde{\sigma}_{G_{\vec{p}}}(\vec{x}),
	\end{equation}
	and it is simply the average of the quantum states of the configurations $\vec{x}\in\Xi_G$.
	
	In addition to probabilistic elementary link generation, we also have probabilistic joining protocols and probabilistic entanglement distillation protocols in practice. We can describe these mathematically using quantum instrument channels. In particular, following the notation in Section~\ref{sec-network_q_state}, given a path $w$ of active elementary links in the network, the joining channel $\mathcal{L}_{w\to e'}$ that forms the new virtual link $e'$ is given in the probabilistic setting by
	\begin{equation}\label{eq-q_instr_channel_swapping}
		\mathcal{L}_{w\to e'}(\cdot)=\ket{0}\bra{0}\otimes\mathcal{L}_{w\to e'}^0(\cdot)+\ket{1}\bra{1}\otimes\mathcal{L}_{w\to e'}^1(\cdot),
	\end{equation}
	where $\mathcal{L}_{w\to e'}^0$ and $\mathcal{L}_{w\to e'}^1$ are completely positive trace non-increasing LOCC maps such that $\mathcal{L}_{w\to e'}^0+\mathcal{L}_{w\to e'}^1$ is a trace-preserving map, and thus an LOCC quantum channel. Specifically, $\mathcal{L}_{w\to e'}^0$ corresponds to failure of the joining protocol and $\mathcal{L}_{w\to e'}^1$ corresponds to success of the joining protocol. Given an input state $\rho_w$ corresponding to the given path $w$, the success probability of the joining protocol is $\Tr\left[\mathcal{L}_{w\to e'}^1(\rho_w)\right]$.
	
	Probabilistic entanglement distillation protocols have an analogous mathematical description. Given an element $e\in E$ with $ c(e)=n$ parallel edges, every probabilistic entanglement distillation protocol has the form
	\begin{equation}
		\mathcal{D}_{e^1\dotsb e^n\to e^1\dotsb e^{n'}}^{e}(\cdot)=\ket{0}\bra{0}\otimes\mathcal{D}_{e^1\dotsb e^n\to e^1\dotsb e^{n'}}^{e;0}(\cdot)+\ket{1}\bra{1}\otimes\mathcal{D}_{e^1\dotsb e^n\to e^1\dotsb e^{n'}}^{e;1}(\cdot),
	\end{equation}
	where $\mathcal{D}_{e^1\dotsb e^n\to e^1\dotsb e^{n'}}^{e;0}$ and $\mathcal{D}_{e^1\dotsb e^n\to e^1\dotsb e^{n'}}^{e;1}$ are completely positive trace non-increasing LOCC maps such that $\mathcal{D}_{e^1\dotsb e^n\to e^1\dotsb e^{n'}}^{e;0}+\mathcal{D}_{e^1\dotsb e^n\to e^1\dotsb e^{n'}}^{e;1}$ is a trace-preserving map, and thus an LOCC quantum channel. Specifically, $\mathcal{D}_{e^1\dotsb e^n\to e^1\dotsb e^{n'}}^{e;0}$ corresponds to failure of the protocol and $\mathcal{D}_{e^1\dotsb e^n\to e^1\dotsb e^{n'}}^{e;1}$ corresponds to success of the protocol.
	
	We have already seen an example of a probabilistic entanglement distillation protocol in Example~\ref{ex-ent_distill}. Let us now look particular examples of probabilistic versions of the joining protocols defined in Examples~\ref{ex-ent_swap}, \ref{ex-GHZ_ent_swap}, and \ref{ex-graph_state_dist}.
	
	\begin{example}[Probabilistic joining protocols]
		The joining protocols defined in Examples~\ref{ex-ent_swap}, \ref{ex-GHZ_ent_swap}, and \ref{ex-graph_state_dist} all involve measurements. The main reason these protocols do not have unit success probability in practice is that these measurements cannot be implemented perfectly, meaning that they have a non-unit probability of executing the ideal measurement, which leads to the overall protocol having a non-unit success probability\footnote{Sometimes, the gate operations involved in the protocol also cannot be implemented perfectly with unit probability, thus contributing to the overall probabilistic nature of the protocol.}. In this example, we consider a simple way to model probabilistic measurements. The actual mathematical model of the measurement depends on the specifics of the particular implementation being considered, but ultimately the overall quantum channel for the protocol can be written as a quantum instrument channel of the form in \eqref{eq-q_instr_channel_swapping}.
		
		\begin{enumerate}
			\item Let us first consider the entanglement swapping channel defined in \eqref{eq-ent_swap_channel}:
				\begin{align}
					\mathcal{L}_{A\vec{R}_1\vec{R}_2\dotsb\vec{R}_nB\to AB}^{\text{ES};n}\left(\rho_{A\vec{R}_1\vec{R}_2\dotsb\vec{R}_nB}\right)&\coloneqq\sum_{\vec{z},\vec{x}\in[d]^n}\Tr_{\vec{R}_1\vec{R}_2\dotsb\vec{R}_n}\left[M_{\vec{R}_1\vec{R}_2\dotsb\vec{R}_n}^{\vec{z},\vec{x}}W_B^{\vec{z},\vec{x}}\left(\rho_{A\vec{R}_1\vec{R}_2\dotsb\vec{R}_nB}\right)\left(W_B^{\vec{z},\vec{x}}\right)^\dagger\right],\\[0.2cm]
					M_{\vec{R}_1\vec{R}_2\dotsb\vec{R}_n}^{\vec{z},\vec{x}}&\coloneqq \Phi_{\vec{R}_1}^{z_1,x_1}\otimes\Phi_{\vec{R}_2}^{z_2,x_2}\otimes\dotsb\otimes\Phi_{\vec{R}_n}^{z_n,x_n},\\[0.2cm]
					W_B^{\vec{z},\vec{x}}&\coloneqq Z_B^{z_1+\dotsb+z_n}X_B^{x_1+\dotsb+x_n}.
				\end{align}
				A simple way to make the protocol probabilistic is to modify the operators $M_{\vec{R}_1\dotsb\vec{R}_n}^{\vec{z},\vec{x}}$ as follows:
				\begin{equation}
					M_{\vec{R}_1\dotsb\vec{R}_n}^{\vec{z},\vec{x}}\to\widetilde{M}_{\vec{R}_1\dotsb\vec{R}_n}^{\vec{z},\vec{x};\vec{\alpha}}\coloneqq\Lambda_{\vec{R}_1}^{z_1,x_1,\alpha_1}\otimes\dotsb\otimes\Lambda_{\vec{R}_n}^{z_n,x_n,\alpha_n},
				\end{equation}
				where $\{\Lambda_{\vec{R}_j}^{z_j,x_j,\alpha_j}\}_{z_j,x_j,\alpha_j\in\{0,1\}}$, $1\leq j\leq n$, are POVMs such that
				\begin{align}
					\Lambda_{\vec{R}_j}^{z_j,x_j,1}&=q_j\Phi_{\vec{R}_j}^{z_j,x_j},\\[0.2cm]
					\sum_{z_j,x_j\in\{0,1\}}\Lambda_{\vec{R}_j}^{z_j,x_j,0}&=(1-q_j)\mathbbm{1}_{\vec{R}_j}.
				\end{align}
				The values $q_j\in[0,1]$ represent the success probability of the Bell-basis measurement at the $j^{\text{th}}$ intermediate node. We then define the LOCC instrument channel for the probabilistic entanglement swapping protocol as follows:
				\begin{multline}
					\widetilde{\mathcal{L}}_{A\vec{R}_1\dotsb\vec{R}_nB\to AB}^{\text{ES};n}\left(\rho_{A\vec{R}_1\dotsb\vec{R}_nB}\right)\coloneqq\ket{0}\bra{0}\otimes\widetilde{\mathcal{L}}_{A\vec{R}_1\dotsb\vec{R}_nB\to AB}^{\text{ES};n;0}\left(\rho_{A\vec{R}_1\dotsb\vec{R}_nB}\right)\\[0.2cm]+\ket{1}\bra{1}\otimes\widetilde{\mathcal{L}}_{A\vec{R}_1\dotsb\vec{R}_nB\to AB}^{\text{ES};n;1}\left(\rho_{A\vec{R}_1\dotsb\vec{R}_nB}\right),
				\end{multline}
				where
				\begin{align}
					\widetilde{\mathcal{L}}_{A\vec{R}_1\dotsb\vec{R}_nB\to AB}^{\text{ES};n;1}\left(\rho_{A\vec{R}_1\dotsb\vec{R}_nB}\right)&\coloneqq\sum_{\vec{z},\vec{x}\in\{0,1\}^n}\Tr_{\vec{R}_1\dotsb\vec{R}_n}\left[\widetilde{M}_{\vec{R}_1\dotsb\vec{R}_n}^{\vec{z},\vec{x};\vec{1}}W_B^{\vec{z},\vec{x}}\left(\rho_{A\vec{R}_1\dotsb\vec{R}_nB}\right)\left(W_B^{\vec{z},\vec{x}}\right)^\dagger\right]\\
					&=q_1\dotsb q_n\mathcal{L}_{A\vec{R}_1\dotsb\vec{R}_nB\to AB}^{\text{ES};n}\left(\rho_{A\vec{R}_1\dotsb\vec{R}_nB}\right)
				\end{align}
				and
				\begin{equation}
					\widetilde{\mathcal{L}}_{A\vec{R}_1\dotsb\vec{R}_n\to AB}^{\text{ES};n;0}\left(\rho_{A\vec{R}_1\dotsb\vec{R}_nB}\right)\coloneqq\sum_{\substack{\vec{z},\vec{x},\vec{\alpha}\in\{0,1\}^n\\\vec{\alpha}\neq\vec{1}}}\Tr_{\vec{R}_1\dotsb\vec{R}_n}\left[\widetilde{M}_{\vec{R}_1\dotsb\vec{R}_n}^{\vec{z},\vec{x};\vec{\alpha}}W_B^{\vec{z},\vec{x}}\left(\rho_{A\vec{R}_1\dotsb\vec{R}_nB}\right)\left(W_B^{\vec{z},\vec{x}}\right)^\dagger\right].
				\end{equation}
				Then, the success probability of the protocol is
				\begin{equation}
					\Tr\left[\widetilde{\mathcal{L}}_{A\vec{R}_1\dotsb\vec{R}_nB\to AB}^{\text{ES};n;1}\left(\rho_{A\vec{R}_1\dotsb\vec{R}_nB}\right)\right]=q_1\dotsb q_n
				\end{equation}
				for all states $\rho_{A\vec{R}_1\dotsb\vec{R}_nB}$.
				
			\item The GHZ entanglement swapping channel defined in \eqref{eq-GHZ_ent_swap_channel} can be made probabilistic in a similar manner. Recall first that
				\begin{equation}\label{eq-GHZ_ent_swap_channel_2}
					\mathcal{L}_{A\vec{R}_1\dotsb\vec{R}_nB\to AR_1^1\dotsb R_n^1B}^{\text{GHZ};n}\left(\rho_{A\vec{R}_1\dotsb\vec{R}_nB}\right)\coloneqq\sum_{\vec{x}\in\{0,1\}^n} P_{\vec{R}_1\dotsb\vec{R}_nB}^{\vec{x}}\left(\rho_{A\vec{R}_1\dotsb\vec{R}_nB}\right)P_{\vec{R}_1\dotsb\vec{R}_nB}^{\vec{x}~\dagger},
				\end{equation}
				where
				\begin{equation}
					P_{\vec{R}_1\dotsb\vec{R}_nB}^{\vec{x}}\coloneqq K_{\vec{R}_1}^{x_1}\otimes K_{\vec{R}_2}^{x_2}X_{R_2^1}^{x_1}\otimes\dotsb\otimes K_{\vec{R}_n}^{x_n}X_{R_n^1}^{x_{n-1}}\otimes X_B^{x_n}
				\end{equation}
				for all $\vec{x}\in\{0,1\}^n$ and
				\begin{equation}
					K_{\vec{R}_j}^{x_j}=\bra{x}_{R_j^2}\text{CNOT}_{\vec{R}_j}
				\end{equation}
				for all $1\leq j\leq n$. We can write \eqref{eq-GHZ_ent_swap_channel_2} as follows:
				\begin{multline}
					\mathcal{L}_{A\vec{R}_1\dotsb\vec{R}_nB\to AR_1^1\dotsb R_n^1B}^{\text{GHZ};n}\left(\rho_{A\vec{R}_1\dotsb\vec{R}_n}\right)\\=\sum_{\vec{x}\in\{0,1\}^n}\Tr_{R_1^2\dotsb R_n^2}\left[\ket{\vec{x}}\bra{\vec{x}}_{R_1^2\dotsb R_n^2}C_{\vec{R}_1\dotsb\vec{R}_n}X_{R_2^1\dotsb R_n^1B}^{\vec{x}}\left(\rho_{A\vec{R}_1\dotsb\vec{R}_nB}\right)X_{R_2^1\dotsb R_n^1B}C_{\vec{R}_1\dotsb\vec{R}_n}\right],
				\end{multline}
				where
				\begin{align}
					C_{\vec{R}_1\dotsb\vec{R}_n}&\coloneqq\text{CNOT}_{\vec{R}_1}\otimes\dotsb\otimes\text{CNOT}_{\vec{R}_n},\\
					X_{R_2^1\dotsb R_n^1B}^{\vec{x}}&\coloneqq X_{R_2^1}^{x_1}\otimes\dotsb\otimes X_{B}^{x_n}.
				\end{align}
				Then, to make the protocol probabilistic, we can make the following simple modification:
				\begin{equation}
					\ket{\vec{x}}\bra{\vec{x}}_{R_1^2\dotsb R_n^2}\to\Lambda_{R_1^2\dotsb R_n^2}^{\vec{x},\vec{\alpha}}\coloneqq\Lambda_{R_1^2}^{x_1,\alpha_1}\otimes\dotsb\otimes\Lambda_{R_n^2}^{x_n,\alpha_n},
				\end{equation}
				where $\{\Lambda_{R_j^2}^{x_j,\alpha_j}\}_{x_j,\alpha_j\in\{0,1\}}$, $1\leq j\leq n$, are POVMs such that
				\begin{align}
					\Lambda_{R_j^2}^{x_k,1}&=q_j\ket{x_j}\bra{x_j}_{R_j^2},\\[0.2cm]
					\sum_{x_j\in\{0,1\}}\Lambda_{R_j^2}^{x_j,0}&=(1-q_j)\mathbbm{1}_{R_j^2}.
				\end{align}
				The values $q_j\in[0,1]$ represent the success probability of the standard-basis measurement at the $j^{\text{th}}$ intermediate node. Then, we define the LOCC quantum instrument channel for the GHZ entanglement swapping protocol as follows:
				\begin{multline}
					\widetilde{\mathcal{L}}_{A\vec{R}_1\dotsb\vec{R}_nB\to AR_1^1\dotsb R_n^1B}^{\text{GHZ};n}\left(\rho_{A\vec{R}_1\dotsb\vec{R}_nB}\right)\coloneqq\ket{0}\bra{0}\otimes\widetilde{\mathcal{L}}_{A\vec{R}_1\dotsb\vec{R}_nB\to AR_1^1\dotsb R_n^1B}^{\text{GHZ};n;0}\left(\rho_{A\vec{R}_1\dotsb\vec{R}_nB}\right)\\[0.2cm]+\ket{1}\bra{1}\otimes\widetilde{\mathcal{L}}_{A\vec{R}_1\dotsb\vec{R}_nB\to AR_1^1\dotsb R_n^1B}^{\text{GHZ};n;1}\left(\rho_{A\vec{R}_1\dotsb\vec{R}_nB}\right),
				\end{multline}
				where
				\begin{align}
					&\widetilde{\mathcal{L}}_{A\vec{R}_1\dotsb\vec{R}_nB\to AR_1^1\dotsb R_n^1B}^{\text{GHZ};n;1}\left(\rho_{A\vec{R}_1\dotsb\vec{R}_nB}\right)\nonumber\\
					&\qquad\qquad=\sum_{\vec{x}\in\{0,1\}^n}\Tr_{R_1^2\dotsb R_n^2}\left[\Lambda_{R_1^2\dotsb R_n^2}^{\vec{x},\vec{1}}C_{\vec{R}_1\dotsb\vec{R}_n}X_{R_2^1\dotsb R_n^1B}^{\vec{x}}\left(\rho_{A\vec{R}_1\dotsb\vec{R}_nB}\right)X_{R_2^1\dotsb R_n^1B}^{\vec{x}}C_{\vec{R}_1\dotsb\vec{R}_n}\right]\\
					&\qquad\qquad=q_1\dotsb q_n \mathcal{L}_{A\vec{R}_1\dotsb\vec{R}_nB\to AR_1^1\dotsb AR_n^1B}^{\text{GHZ};n}\left(\rho_{A\vec{R}_1\dotsb\vec{R}_nB}\right)
				\end{align}
				and
				\begin{multline}
					\widetilde{\mathcal{L}}_{A\vec{R}_1\dotsb\vec{R}_nB\to AR_1^1\dotsb R_n^1B}^{\text{GHZ};n;0}\left(\rho_{A\vec{R}_1\dotsb\vec{R}_nB}\right)\\\coloneqq\sum_{\substack{\vec{x},\vec{\alpha}\in\{0,1\}^n\\\vec{\alpha}\neq\vec{1}}}\Tr_{R_1^1\dotsb R_n^1}\left[\Lambda_{R_1^1\dotsb R_n^1}^{\vec{x},\vec{\alpha}}C_{\vec{R}_1\dotsb\vec{R}_n}X_{R_2^1\dotsb R_n^1B}^{\vec{x}}\left(\rho_{A\vec{R}_1\dotsb\vec{R}_nB}\right)X_{R_2^1\dots R_n^1B}^{\vec{x}}C_{\vec{R}_1\dotsb\vec{R}_n}\right].
				\end{multline}
				Then, the success probability of the protocol is
				\begin{equation}
					\Tr\left[\widetilde{\mathcal{L}}_{A\vec{R}_1\dotsb\vec{R}_nB\to AR_1^1\dotsb R_n^1B}^{\text{GHZ};n;1}\left(\rho_{A\vec{R}_1\dotsb\vec{R}_nB}\right)\right]=q_1\dotsb q_n
				\end{equation}
				for all states $\rho_{A\vec{R}_1\dotsb\vec{R}_nB}$.
				
			\item Finally, let us consider the graph state distribution channel defined in \eqref{eq-graph_state_dist_channel}:
				\begin{align}
					\mathcal{L}_{A_1^nR_1^n\to A_1^n}^{(G)}(\rho_{A_1^nR_1^n})&=\sum_{\vec{x}\in\{0,1\}^n} \left(Z_{A_1^n}^{\vec{x}}\otimes\bra{G^{\vec{x}}}_{R_1^n}\right)\left(\rho_{A_1^nR_1^n}\right)\left(Z_{A_1^n}^{\vec{x}}\otimes\ket{G^{\vec{x}}}_{R_1^n}\right)\\
					&=\sum_{\vec{x}\in\{0,1\}^n}\Tr_{R_1^n}\left[\left(Z_{A_1^n}^{\vec{x}}\otimes\ket{G^{\vec{x}}}\bra{G^{\vec{x}}}_{R_1^n}\right)\left(\rho_{A_1^nR_1^n}\right)\left(Z_{A_1^n}^{\vec{x}}\otimes\mathbbm{1}_{R_1^n}\right)\right].
				\end{align}
				In order to make the protocol probabilistic, we can make the following modification:
				\begin{equation}
					\ket{G^{\vec{x}}}\bra{G^{\vec{x}}}_{R_1^n}\to \Lambda_{R_1^n}^{\vec{x},\alpha},
				\end{equation}
				where $\{\Lambda_{R_1^n}^{\vec{x},\alpha}\}_{\vec{x}\in\{0,1\}^n,\alpha\in\{0,1\}}$ is a POVM such that
				\begin{align}
					\Lambda_{R_1^n}^{\vec{x},1}&=q\ket{G^{\vec{x}}}\bra{G^{\vec{x}}}_{R_1^n},\\[0.2cm]
					\sum_{\vec{x}\in\{0,1\}}\Lambda_{R_1^n}^{\vec{x},0}&=(1-q)\mathbbm{1}_{R_1^n}.
				\end{align}
				The value $q\in[0,1]$ represents the success probability of the measurement given by the POVM $\{\ket{G^{\vec{x}}}\bra{G^{\vec{x}}}_{R_1^n}\}_{\vec{x}\in\{0,1\}^n}$. Then, we define the LOCC quantum instrument channel for the graph state distribution protocol as follows:
				\begin{equation}
					\widetilde{\mathcal{L}}_{A_1^nR_1^n\to A_1^n}^{(G)}(\rho_{A_1^nR_1^n})\coloneqq\ket{0}\bra{0}\otimes\widetilde{\mathcal{L}}_{A_1^nR_1^n\to A_1^n}^{(G);0}(\rho_{A_1^nR_1^n})+\ket{1}\bra{1}\otimes\widetilde{\mathcal{L}}_{A_1^nR_1^n\to A_1^n}^{(G);1}(\rho_{A_1^nR_1^n}),
				\end{equation}
				where
				\begin{align}
					\widetilde{\mathcal{L}}_{A_1^nR_1^n\to A_1^n}^{(G);1}(\rho_{A_1^nR_1^n})&=\sum_{\vec{x}\in\{0,1\}^n}\Tr_{R_1^n}\left[\left(Z_{A_1^n}^{\vec{x}}\otimes\Lambda_{R_1^n}^{\vec{x},1}\right)\left(\rho_{A_1^nR_1^n}\right)\left(Z_{A_1^n}^{\vec{x}}\otimes\mathbbm{1}_{R_1^n}\right)\right]\\
					&=q\mathcal{L}_{A_1^nR_1^n\to A_1^n}^{(G)}(\rho_{A_1^nR_1^n})
				\end{align}
				and
				\begin{equation}
					\widetilde{\mathcal{L}}_{A_1^nR_1^n\to A_1^n}^{(G);0}(\rho_{A_1^nR_1^n})=\sum_{\vec{x}\in\{0,1\}^n}\Tr_{R_1^n}\left[\left(Z_{A_1^n}^{\vec{x}}\otimes\Lambda_{R_1^n}^{\vec{x},0}\right)\left(\rho_{A_1^nR_1^n}\right)\left(Z_{A_1^n}^{\vec{x}}\otimes\mathbbm{1}_{R_1^n}\right)\right].
				\end{equation}
				Then, the success probability of the protocol is
				\begin{equation}
					\Tr\left[\widetilde{\mathcal{L}}_{A_1^nR_1^n\to A_1^n}^{(G);1}(\rho_{A_1^nR_1^n})\right]=q
				\end{equation}
				for all states $\rho_{A_1^nR_1^n}$.~\defqed
		\end{enumerate}
	\end{example}

\section{Summary}	
	
	In this chapter, we provided an introduction to quantum networks. Specifically, we provided precise definitions of what a quantum network is and how to model it mathematically (see the introduction to the chapter and Section~\ref{sec-network_q_state}). In Section~\ref{sec-ent_dist_general}, we defined the task of entanglement distribution in a quantum network. The task of entanglement distribution can be thought of as taking the graph of physical (elementary) links and transforming it into a new graph that contains the elementary links as well as a desired set of \textit{virtual links}, i.e., links that are obtained by performing so-called ``joining protocols'' on elementary links, examples of which we provide in Section~\ref{sec-LOCC_channels}. Then, in Section~\ref{sec-network_architecture}, we take our first step toward practical schemes for entanglement distribution. We start by presenting ground-based and satellite-based elementary link generation schemes in Section~\ref{sec-practical_elem_link_generation}, and from these developments it becomes clear that elementary link generation is in general probabilistic. We also consider quantum memories with finite coherence times and how they should be modeled. With these practical considerations in mind, we then refine our initial mathematical model of quantum networks and entanglement distribution in Sections~\ref{sec-random_subgraphs} and \ref{sec-network_q_state_practical}.
	
	This chapter, along with the previous chapter on quantum decision processes, sets the stage for the next chapter, in which we combine the developments of this chapter and the previous chapter in order to develop an explicit quantum network protocol using quantum decision processes.

\chapter{QUANTUM NETWORK PROTOCOLS VIA DECISION PROCESSES}\label{chap-network_QDP}

	In the previous chapter, we discussed the basic theoretical framework of quantum networks---their description as graphs, the task of entanglement distribution, and practical considerations that need to be taken into account. In this chapter, we bring together this general theory of quantum networks from the previous chapter and we combine it with the general theory of quantum decision processes from Chapter~\ref{chap-QDP} in order to develop explicit protocols for entanglement distribution in a quantum network.

	Our task of interest in quantum networks is entanglement distribution, as described in Section~\ref{sec-ent_dist_general}. Given a graph $G=(V,E,c)$ describing the physical (elementary) links of the network, the goal is to obtain a network corresponding to a target graph $G_{\text{target}}=(V,E_{\text{target}},c_{\text{target}})$, which contains not only elementary links but also virtual links, i.e., entanglement shared by nodes that are not physically connected. A simple achievable strategy for this, as described in Section~\ref{sec-figures_of_merit_general} and discussed in Refs.~\cite{AK17,BA17} (see also Ref.~\cite{KMSD19}), is to generate the appropriate target state in the elementary links, and then to apply the appropriate joining protocols to create virtual links.
	
	In Section~\ref{sec-network_architecture}, we discussed important practical aspects of entanglement distribution that need to be taken into account when developing such quantum network protocols, one of which is that elementary link generation is probabilistic. This fact led to the conclusion that the initial quantum state of every elementary link, as defined in \eqref{eq-elem_link_initial_cq_state}, is the following classical-quantum state:
	\begin{equation}\label{eq-elem_link_initial_cq_state_2}
		\widehat{\sigma}_{e^j}=\ket{0}\bra{0}\otimes\widetilde{\sigma}_{e^j}(0)+\ket{1}\bra{1}\otimes\widetilde{\sigma}_{e^j}(1),
	\end{equation}
	for all $e\in E$ and $1\leq j\leq c(e)$, where
	\begin{equation}
		\widetilde{\sigma}_{e^j}(x)=(\mathcal{M}_{e^j}^x\circ\mathcal{S}_{e^j})(\rho_{e^j}^S)\quad\forall~x\in\{0,1\},
	\end{equation}
	and $\rho_{e^j}^S$, $\mathcal{S}_{e^j}$, and $\mathcal{M}_{e^j}^x$ are the source state, transmission channel, and heralding map, respectively, for the elementary link given by the edge $e^j$. The classical register tells us whether or not the elementary link is active (i.e., whether the transmission and heralding succeeded), and the quantum register contains the corresponding quantum state of the elementary link. From this, we found that the initial quantum state of the entire network, as defined in \eqref{eq-network_initial_cq_state}, also has a classical-quantum form:
	\begin{equation}\label{eq-network_initial_cq_state_2}
		\widehat{\sigma}_{G_{\vec{p}}}\coloneqq\sum_{\vec{x}\in\Xi_G}\ket{\vec{x}}\bra{\vec{x}}_{X_G}\otimes\widetilde{\sigma}_{G_{\vec{p}}}(\vec{x}).
	\end{equation}
	The classical register now gives us the configuration $\vec{x}\in\Xi_G$ of the network, i.e., it tells us which of the elementary links are active, and the quantum register contains the quantum states after transmission from the source and heralding.
	
	Now, given the initial quantum state of the network in \eqref{eq-network_initial_cq_state_2}, how should the protocol described above for achieving the target network given by $G_{\text{target}}$ proceed? If the configuration obtained after the initial step, along with the fidelities of the links, does not correspond to the target network---for example, some of the required elementary links attempts might have failed---then it makes sense to try the source distribution again for the failed elementary links. For the ones that succeeded, it might make sense to keep the quantum states in memory rather than discard the states, request new ones from the sources, and risk some of these new attempts failing. From these considerations, some questions naturally arise:
	\begin{enumerate}
		\item What is the (optimal) sequence of actions that should be performed for every elementary link, as a function of time, in order to achieve the target network with the desired probability and fidelity?
		
		\item How long does it take to achieve the target network with the desired probability and fidelity?
	\end{enumerate}
	These questions fall under the domain of quantum decision processes, so it makes sense to try to view quantum network protocols from the lens of quantum decision processes. A basic first step in this direction is to determine the agent(s) and the environment(s).
	
	\begin{figure}
		\centering
		\includegraphics[width=0.95\textwidth]{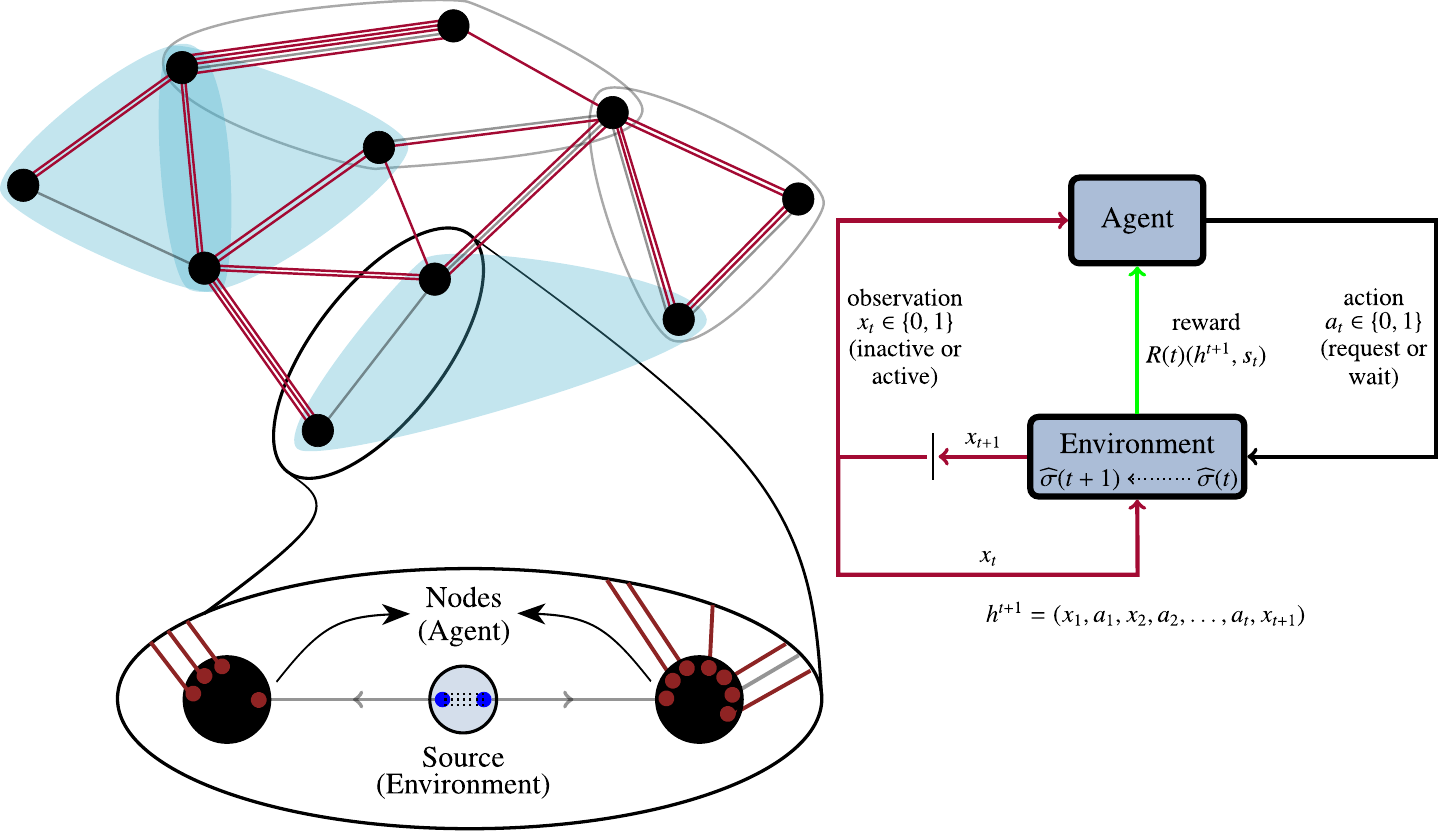}
		\caption{Our approach to quantum network protocols via quantum decision processes. Given a graph $G=(V,E,c)$ corresponding to the physical (elementary) links of a quantum network, we define an (independent) agent for every element $e\in E$, and the associated environment is defined to be the quantum systems distributed to the nodes of $e$ by the corresponding source station. The agents are classical, and they have two possible actions, either to ``wait'' (i.e., keep the entangled state currently in quantum memory) or to ``request'' (i.e., discard the entangled state currently in quantum memory and request a new one from the source). The set of observations of each agent is the status of elementary link, i.e., whether or not it is active. We can also allow the agents to perform entanglement distillation, which corresponds to a separate action; see Appendix~\ref{sec-network_QDP_distillation}.}\label{fig-net_agent_env_0}
	\end{figure}
	
	In Figure~\ref{fig-net_agent_env_0}, we illustrate our approach to quantum network protocols via quantum decision processes. We associate an (independent) agent to every elementary link (more precisely, to every element of $E$), and its environment is the collection of quantum systems distributed to the corresponding nodes by the corresponding source station. We define the agent to be classical here, such that at every time step it can choose one of two possible actions: ``wait'' (i.e., keep the entangled state currently in quantum memory), or ``request'' (i.e., discard the entangled state currently in quantum memory and request a new one from the source). Based on this action, the agent receives one of two possible observations: either the link is active or it is inactive. We also allow for the agent to execute an additional action, ``distill'' (i.e., perform entanglement distillation on a subset of parallel elementary links), with the observations again being whether or not the corresponding parallel links are active (i.e., whether or not the entanglement distillation succeeded). With these basic elements in place, the sequence of actions that should be performed as part of the network protocol is nothing more than the policy of the agent(s), meaning that the policy itself defines the protocol. Indeed, let us recall Lemma~\ref{lem-QDP_LOCC}, which tells us that every quantum decision process corresponds to an LOCC quantum channel. Furthermore, as we show in Section~\ref{sec-indep_agents}, each of the actions of the agents corresponds to an LOCC channel. This means that every policy of the elementary link quantum decision process outlined above is a particular case of the LOCC channel $\mathcal{L}_{\text{LOCC}}$ in \eqref{eq-ent_dist_protocol_graphs_trans} that defines the quantum network protocol for entanglement distribution. Furthermore, the quantum decision process provides us with an explicit decomposition of the LOCC channel into discrete time steps.
	
	The developments of the previous paragraph, summarized in Figure~\ref{fig-net_agent_env_0}, mean that our overall approach to quantum network protocols is the following. Starting with the graph $G=(V,E,c)$ of physical elementary links, all of the elementary links independently undergo $T$-step policies $\pi_T^{e^j}$, with $e\in E$, $1\leq j\leq c(e)$, and $T\in[1,\infty)$. After $T$ time steps, an algorithm finds paths for creating the virtual links specified by the target graph $G_{\text{target}}$ and the corresponding joining protocols are performed. If entire target network cannot be achieved in $T$ time steps, then a decision is made to either conclude the protocol with the current configuration or to continue for another $T$ time steps under the same policies. We summarize this approach in Figure~\ref{fig-QDP_protocol}.

	\begin{figure}
		\centering
		\begin{specbox*}{0.9\textwidth}{Quantum network protocol via quantum decision processes}
			\begin{description}[style=nextline]
				\item[Given]
					\begin{itemize}[leftmargin=*,topsep=0pt,before=\leavevmode\vspace{-1.5em}]
						\item Graph $G=(V,E,c)$ of physical (elementary) links.
						\item Source states $\rho_{e^j}^S$, transmission channels $\mathcal{S}_{e^j}$, and heralding instruments $\{\mathcal{M}_{e^j}^0,\mathcal{M}_{e^j}^1\}$ for every $e\in E$ and $1\leq j\leq c(e)$. 
						\item Time $T\in[1,\infty)$.
						\item $T$-step policies $\pi_T^{e^j}$ for every $e\in E$ and $1\leq j\leq c(e)$.
					\end{itemize}
					
				\item[Required] A target graph $G_{\text{target}}=(V,E_{\text{target}},c_{\text{target}})$, with corresponding target states $\rho_{e^j}^{\text{target}}$ and fidelity thresholds $f_{e^j}$ for every $e\in E_{\text{target}}$ and $1\leq j\leq c_{\text{target}}(e)$.
				
				\item[Protocol]
					\begin{enumerate}[leftmargin=*,topsep=0pt,before=\leavevmode\vspace{-1.5em}]
						\item Every elementary link corresponding to $e^j$, $e\in E$ and $1\leq j\leq c(e)$, follows the policy $\pi_T^{e^j}$ for $T$ time steps.
						\item Given the configuration of active elementary links at the end of $T$ time steps, an algorithm \cite{SMI+17,PKT+19,CRDW19} determines paths for forming the virtual links specified by $E_{\text{target}}$ and $c_{\text{target}}$, and the corresponding joining protocols are performed.
						\item If the required target graph and fidelity thresholds are reached, then STOP; otherwise, decide:
						\begin{enumerate}
							\item STOP; or
							\item CONTINUE: repeat Steps 1--3.
						\end{enumerate}
					\end{enumerate}
			\end{description}
		\end{specbox*}
		\vspace{-0.5cm}
		\caption{Outline of our quantum network protocol based on quantum decision processes. Every elementary link in the network undergoes a quantum decision process specified by a $T$-step policy, with $T$ finite. At the end of the $T$ time steps, the appropriate paths in the network are found and the corresponding joining protocols are performed in order to achieve the network corresponding to the target graph $G_{\text{target}}$.}\label{fig-QDP_protocol}
	\end{figure}

	Let $\vec{\pi}_T=\{\pi_T^{e_j}:e\in E,\,1\leq j\leq c(e)\}$ denote the collection of $T$-step policies for all of the elementary links. Then, the quantum state of the network at time steps $1\leq t\leq T$ is 
	\begin{equation}\label{eq-network_cq_state_QDP}
		\widehat{\sigma}_G^{\vec{\pi}_T}(t)=\bigotimes_{e\in E}\bigotimes_{j=1}^{ c(e)}\widehat{\sigma}_{e^j}^{\pi_T^{e^j}}(t)
	\end{equation}
	where $\widehat{\sigma}_{e^j}^{\pi^{e^j}}(1)=\widehat{\sigma}_{e^j}$. For $2\leq t\leq T$, every state $\widehat{\sigma}_{e^j}^{\pi^{e^j}}(t)$ is given by the general expression in \eqref{eq-QDP_cq_state}, which is defined in terms of the transition maps and measurement operators of the environment. In Section~\ref{sec-indep_agents}, we explicitly define these environment elements, along with the reward functions. Then, in Section~\ref{sec-practical_figures_merit}, we define figures of merit for evaluating policies that are important in practical settings.
	
	\begin{remark}\label{rem-QDP_network_extensions}
		The tensor product structure in \eqref{eq-network_cq_state_QDP} holds because all of the agents for the elementary links are independent. This means that all of the agents have knowledge only of the status of their own elementary link. We stick to this setting throughout this thesis. We can use the quantum decision processes for the elementary links as building blocks for quantum decision processes for groups of elementary links. Furthermore, we can use quantum decision processes to develop quantum network protocols in which the agents can cooperate, so that they have knowledge of the network in a certain local neighborhood of their elementary link.  We discuss these possibilities as a direction for future work in Appendix~\ref{sec-future_work}.
		
		Physically, the agents of the quantum decision processes should be thought of as (classical) devices, and the policies as algorithms that are executed by the devices. Cooperating agents then correspond to devices that can communicate (classically) with each other. We emphasize that the agents are not necessarily the end users of the network making decisions according to the policies in real time. The end users should be thought of simply as providing the target graph $G_{\text{target}}$ in Figure~\ref{fig-QDP_protocol} based on their desired application.
		
		The protocol outlined in Figure~\ref{fig-QDP_protocol} is a relatively simple one in which the quantum decision process is used simply for the elementary links. In principle, it is possible to incorporate routing and path-finding algorithms into the decision process framework, so that using the algorithms in Refs.~\cite{SMI+17,PKT+19,CRDW19}, as done for the algorithm described in Figure~\ref{fig-QDP_protocol}, is not required. This also leads to the possibility for performing reinforcement learning of path-finding and routing algorithms. We discuss these possibilities, and other possibilities for developing more sophisticated protocols using quantum decision processes, in Appendix~\ref{sec-future_work}.~\defqed
	\end{remark}
	
	Now, in order to achieve the target graph $G_{\text{target}}$ at time $T+1$, the physical graph must have a particular configuration by time $T$---specifically, a certain subset of elementary links must be active---and the corresponding elementary links must have a certain fidelity to the appropriate target states, so that after the joining protocols the virtual links in $G_{\text{target}}$ meet the desired fidelity thresholds. These required elementary links (which may not be unique) correspond to the following classical-quantum state:
	\begin{equation}
		\bigotimes_{e\in E'}\bigotimes_{j=1}^{c(e)}\ket{1}\bra{1}_{X_{e^j}}\otimes \psi_{e^j}^{\text{target}},
	\end{equation}
	where $E'\subseteq E$ corresponds to the elementary links that are required to be active, and $\psi_{e^j}^{\text{target}}$ are the corresponding target pure states. Then, the fidelity between this state and the one in \eqref{eq-network_cq_state_QDP} is
	\begin{align}
		F\left(\bigotimes_{e\in E'}\bigotimes_{j=1}^{c(e)}\widehat{\sigma}_{e^j}^{\pi_e^j}(T),\bigotimes_{e\in E'}\bigotimes_{j=1}^{c(e)}\ket{1}\bra{1}_{X_{e^j}}\otimes \psi_{e^j}^{\text{target}}\right)&=\prod_{e\in E'}\prod_{j=1}^{c(e)}F\left(\widehat{\sigma}_{e^j}^{\pi^{e^j}}(T),\ket{1}\bra{1}_{X_{e_j}}\otimes\psi_{e_j}^{\text{target}}\right)\\
		&=\prod_{e\in E'}\prod_{j=1}^{c(e)}\Tr\left[\left(\ket{1}\bra{1}_{X_{e_j}}\otimes\psi_{e^j}^{\text{target}}\right)\widehat{\sigma}_{e^j}^{\pi^{e^j}}(T)\right].\label{eq-network_QDP_total_fid_pol}
	\end{align}
	So the task is to optimize the quantity in \eqref{eq-network_QDP_total_fid_pol} with respect to policies $\pi^{e^j}$ for the elementary links in $G$. The quantum decision process for elementary links that we develop in Section~\ref{sec-indep_agents} is such that expected reward for every elementary link is precisely the function in each term of the product in \eqref{eq-network_QDP_total_fid_pol}. This fact allows us to use the methods from Chapter~\ref{chap-QDP} to determine optimal policies, and we discuss this in Section~\ref{sec-network_QDP_pol_opt}.

\section{Quantum decision process for elementary link generation}\label{sec-indep_agents}

	We now start with the formal development of the quantum decision process for elementary link generation outlined in Figure~\ref{fig-net_agent_env_0}, and it is based on the model for elementary link generation outlined in Section~\ref{sec-practical_elem_link_generation}. Roughly speaking, the decision process for every elementary link is such that, at each time step, the agent (which we define to be all of the nodes in the elementary link as a collective entity) either requests entanglement from a source station (which we define to be the collection of quantum systems distributed to the nodes by the source stations) or keeps the quantum state currently stored in memory. For every time step that the quantum state is held in memory, the decoherence channel defined in \eqref{eq-elem_link_decoherence_channel} is applied to each of the quantum systems comprising the quantum state of the elementary link. This process goes on for a given time $T<\infty$, after which a reward is given. This is the basic quantum decision process that we develop in this section. We defer the discussion of quantum decision processes for elementary link generation with entanglement distillation to Appendix~\ref{sec-future_work}.
	
%\subsection{Without distillation}\label{sec-network_QDP_no_purif}

	The quantum decision process for elementary link generation is defined as follows.
	
	\begin{definition}[QDP for elementary link generation]\label{def-network_QDP_elem_link}
		Let $G=(V,E,c)$ be the graph corresponding to the physical (elementary) links of a quantum network, and let $e^j$, $e\in E$, $1\leq j\leq c(e)$, be arbitrary. As shown in Figure~\ref{fig-net_agent_env_0}, we define a quantum decision process for $e^j$ by defining the agent for $e^j$ to be collectively the nodes belonging to $e^j$, and we define its environment to be the quantum systems distributed by the source station to the nodes of $e^j$. Then, the other elements of the quantum decision process are defined as follows.
		\begin{itemize}
			\item We denote the quantum systems of the environment collectively by $E^{e^j}$, and we let $E_t^{e^j}$ denote these quantum systems at time $t\geq 0$. The state of the environment at time $t=0$ is the source state $\rho_{E_0^{e^j}}^{S}$.
			
			\item We let $\mathcal{X}=\{0,1\}$ tell us whether or not the elementary link is active at a particular time. In particular, then, we define random variables $X_{e^j}(t)$ for all $t\geq 1$ as follows:
				\begin{itemize}
					\item $X_{e^j}(t)=0$: elementary link is inactive (transmission and heralding not successful);
					\item $X_{e^j}(t)=1$: elementary link is active (transmission and heralding successful).
				\end{itemize}
				We let $\mathcal{A}=\{0,1\}$ be the set of possible actions of the agent, and we define corresponding random variables $A_{e^j}(t)$ for all $t\geq 1$ as follows:
				\begin{itemize}
					\item $A_{e^j}(t)=0$: wait/keep the entangled state;
					\item $A_{e^j}(t)=1$: discard the entangled state and request a new entangled state.
				\end{itemize}
				We let $\Omega(t)$ denote all histories up to time $t$, where every element $h^t\in\Omega(t)$ is a sequence of the form
				\begin{equation}
					h^t=(x_1,a_1,x_2,a_2,\dotsc,a_{t-1},x_t),
				\end{equation}
				with $x_j\in\mathcal{X}$ for all $1\leq j\leq t$ and $a_j\in\mathcal{A}$ for all $1\leq j\leq t-1$. The corresponding random variable for the history is
				\begin{equation}
					H_{e^j}(t)\coloneqq (X_{e^j}(1),A_{e^j}(1),X_{e^j}(2),A_{e^j}(2),\dotsc,A_{e^j}(t-1),X_{e^j}(t)).
				\end{equation}
				
				The random variables $X_{e^j}(t)$, $A_{e^j}(t)$, and $H_{e^j}(t)$ are mutually independent by definition for all~$e^j$.
				
			\item The transition maps are defined to be time independent, and we denote them by $\mathcal{T}_{e^j}^{x_t,a_t,x_{t+1}}\equiv \mathcal{T}_{E_t^{e^j}\to E_{t+1}^{e^j}}^{x_t,a_t,x_{t+1}}$ for all $x_t,a_t,x_{t+1}\in\{0,1\}$ and all $t\geq 1$, where
				\begin{align}
					\mathcal{T}_{e^j}^{x_t,1,1}(\sigma)&\coloneqq \Tr[\sigma](\mathcal{M}_{e^j}^1\circ\mathcal{S}_{e^j})(\rho_{e^j}^S)\quad\forall~x_t\in\{0,1\},\label{eq-network_QDP_trans_1}\\
					\mathcal{T}_{e^j}^{x_t,1,0}(\sigma)&\coloneqq \Tr[\sigma](\mathcal{M}_{e^j}^0\circ\mathcal{S}_{e^j})(\rho_{e^j}^S)\quad\forall~ x_t\in\{0,1\},\label{eq-network_QDP_trans_2}\\
					\mathcal{T}_{e^j}^{1,0,1}(\sigma)&\coloneqq \mathcal{N}_{e^j}(\sigma),\label{eq-network_QDP_trans_3}\\
					\mathcal{T}_{e^j}^{0,0,0}(\sigma)&\coloneqq \sigma\label{eq-network_QDP_trans_4}
				\end{align}
				for all linear operators $\sigma$, where we recall the definitions of the source transmission channel $\mathcal{S}_{e^j}$, the heralding quantum instrument $\{\mathcal{M}_{e^j}^0,\mathcal{M}_{e^j}^1\}$, and the decoherence channel $\mathcal{N}_{e^j}$ from Section~\ref{sec-practical_elem_link_generation}. Superscript combinations not defined above are equal to the zero map by definition, i.e., $\mathcal{T}_{e^j}^{0,0,1}=0$ and $\mathcal{T}_{e^j}^{1,0,0}=0$. The maps $\mathcal{T}_{e^j}^{0;x_1}\equiv\mathcal{T}_{E_0^{e^j}\to E_1^{e^j}}$, $x_1\in\{0,1\}$, are defined to be
				\begin{align}
					\mathcal{T}_{e^j}^{0;0}&\coloneqq\mathcal{M}_{e^j}^0\circ\mathcal{S}_{e^j},\label{eq-network_QDP_trans_5}\\
					\mathcal{T}_{e^j}^{0;1}&\coloneqq\mathcal{M}_{e^j}^1\circ\mathcal{S}_{e^j}.\label{eq-network_QDP_trans_6}
				\end{align}
				
			\item Given a pure target state $\psi_{e^j}^{\text{target}}=\ket{\psi^{\text{target}}}\bra{\psi^{\text{target}}}_{e^j}$, the reward at time $t\geq 1$ is defined as follows:
				\begin{align}
					\mathcal{R}_{e^j}^{t;h^{t+1},1}(\cdot)&=\psi_{e^j}^{\text{target}}(\cdot)\psi_{e^j}^{\text{target}},\\
					\mathcal{R}_{e^j}^{t;h^{t+1},0}(\cdot)&=(\mathbbm{1}_{e^j}-\psi_{e^j}^{\text{target}})(\cdot)(\mathbbm{1}_{e^j}-\psi_{e^j}^{\text{target}}),
				\end{align}
				for all $h^{t+1}\in\Omega(t)$, and the functions $R_{e^j}(t):\Omega(t+1)\times\{0,1\}\to\mathbb{R}$ are defined as follows:
				\begin{align}
					R_{e^j}(t)(h^{t+1},0)&=0,\\
					R_{e^j}(t)(h^{t+1},1)&=\delta_{x_{t+1},1},
				\end{align}
				for all histories $h^{t+1}=(x_1,a_1,\dotsc,x_t,a_t,x_{t+1})\in\Omega(t+1)$.
				
			\item A $T$-step policy for the agent is a sequence of the form $\pi_T^{e^j}=(d_1^{e^j},d_2^{e^j},\dotsc,d_T^{e^j})$, where the decision functions $d_t^{e^j}:\Omega(t)\times\mathcal{A}\to[0,1]$ are defined to be 
				\begin{equation}
					d_t^{e^j}(h^t)(a_t)\coloneqq\Pr[A_{e^j}(t)=a_t|H_{e^j}(t)=h^t]
				\end{equation}
				for all $1\leq t\leq T$, all histories $h^t\in\Omega(t)$, and all $a_t\in\mathcal{A}$.~\defqed
		\end{itemize}
	\end{definition}
	\smallskip
	\begin{remark}\label{rem-elem_link_QDP}
		Let us make some remarks about our definition of the quantum decision process for elementary link generation.
		\begin{itemize}
			\item The quantum decision process uses discrete time steps to model the generation of elementary links in a quantum network. Physically, each time step is equivalent to the classical communication time between the nodes for the purpose of heralding. A discrete-time model for quantum networks is also used in Ref.~\cite{CRDW19}.
			
			\item Our definition of the transition maps is consistent with our description of the decision process given at the beginning of Section~\ref{sec-indep_agents}: if the action is to wait, and the elementary link is currently active, then we apply the decoherence channel $\mathcal{N}_e$ to the quantum state of the elementary link; if the action is to request, then the current quantum state of the elementary link is discarded and the source transmission and heralding is performed again. If the elementary link is currently not active and the action is to wait, then the quantum state stays as it is.
			
			\item  Note that $X_e(1)=X_e$, where $X_e$ is the random variable defined in Section~\ref{sec-random_subgraphs}. When we have a multigraph $G=(V,E,c)$ describing the physical (elementary) links in the network, then recall that $c(e)$ is the number of parallel edges between the nodes in $e\in E$. Therefore, $c(e)$ is the maximum number of entangled states that can be shared by the nodes of the edge per time step (see Figure~\ref{fig-net_agent_env_0}). We define
				\begin{equation}\label{eq-network_QDP_num_active_parallel_links}
					N_e(t)\coloneqq\sum_{j=1}^{c(e)}X_{e^j}(t)
				\end{equation}
				to be the \textit{number of active parallel elementary links} at time $t$, where $X_{e^j}(t)$ is the status of the $j^{\text{th}}$ parallel elementary link of corresponding to $e\in E$ at time $t$. In general, $N_e(t)$ is a Poisson-binomial random variable (see, e.g., Ref.~\cite{CL97}).
				
				In the context of flow problems in graphs, $c(e)$ represents the capacity of $e\in E$, and $N_e(t)$ represents the flow along $e$ at time $t$. The task of finding the maximum number of edge-disjoint paths in a network (as outlined in Section~\ref{sec-graph_theory}) can be phrased as a flow problem, in which case the expected flows $\mathbb{E}[N_e(t)]_{\pi}$ for all $e\in E$ can be used to determine the rates at which virtual links can be created in the network; see Refs.~\cite{AK17,BA17,BAKE20,TKRW19,CERW20}.
			
			\item The probability distributions of the random variables $X_{e^j}(t)$ and $A_{e^j}(t)$, and more generally the probability distribution of the histories $H_{e^j}(t)$, depend on the particular policy under consideration. We use the notation $\Pr[H_{e^j}(t)=h^t]_{\pi_t^{e^j}}$ to denote the probability of a history $h^t\in\Omega(t)$ according to a policy $\pi_t^{e^j}$.
			
				Also, we note that the set of all histories is $\Omega(t)=\{0,1\}^{2t-1}$ for all $t\geq 1$, because both the set $\mathcal{A}$ of actions and the set $\mathcal{X}$ of observations are equal to $\{0,1\}$.
			
			\item Using \eqref{eq-QDP_transition_prob_gen} and the definition of the transition maps, we have the following values for the transition probabilities for all $t\geq 1$ and for any history $h^t=(x_1,a_1,\dotsc,a_{t-1},x_t)$:
				\begin{align}
					\Pr[X_{e^j}(t+1)=0|X_{e^j}(t)=x_t,A_{e^j}(t)=1]_{\pi^{e^j}}&=1-p_{e^j},\label{eq-link_trans_prob_1}\\
					\Pr[X_{e^j}(t+1)=1|X_{e^j}(t)=x_t,A_{e^j}(t)=1]_{\pi^{e^j}}&=p_{e^j},\label{eq-link_trans_prob_2}\\
					\Pr[X_{e^j}(t+1)=x_{t+1}|X_{e^j}(t)=x_t,A_{e^j}(t)=0]_{\pi^{e^j}}&=\delta_{x_t,x_{t+1}}\quad\forall~x_{t+1}\in\{0,1\}.\label{eq-link_trans_prob_4}
				\end{align}
				Observe that the transition probabilities are time independent.
				
			\item Note that in this decision process the agent is classical. We thus denote the policy of the agent in terms of decision functions, as in the classical case, which give us the probability that a particular action is taken conditioned on the history.~\defqed
		\end{itemize}
	\end{remark}
	
	\begin{definition}[Memory time random variable]\label{def-network_QDP_mem_time_RV}
		Let $G=(V,E,c)$ be the graph corresponding to the physical (elementary) links of a quantum network, and let $e^j$, $e\in E$, $1\leq j\leq c(e)$, be arbitrary. For every policy $\pi$ for the elementary link corresponding to $e^j$, we define the random variable $M_{e^j}^{\pi}(t)$ to be the amount of time that the quantum state of the elementary link corresponding to ${e^j}$ is held in memory at time $t$ when following the policy $\pi$. It is defined by the recursion relation
		\begin{equation}\label{eq-mem_time_def1}
			M_{e^j}^{\pi}(t)=\left\{\begin{array}{l l} M_{e^j}^{\pi}(t-1)+X_{e^j}(t) & \text{if }A_{e^j}(t-1)=0,\\ X_{e^j}(t)-1 & \text{if }A_{e^j}(t-1)=1, \end{array}\right.
		\end{equation}
		where $M_{e^j}^{\pi}(0)\equiv -1$. Alternatively, an explicit expression for $M_{e^j}^{\pi}(t)$ is the following:
		\begin{align}
			M_{e^j}^{\pi}(t)&=A_{e^j}(0)(X_{e^j}(1)+X_{e^j}(2)+\dotsb+X_{e^j}(t)-1)\overline{A_{e^j}(1)}\,\overline{A_{e^j}(2)}\dotsb\overline{A_{e^j}(t-1)}\nonumber\\
			&\quad +A_{e^j}(1)(X_{e^j}(2)+X_{e^j}(3)+\dotsb+X_{e^j}(t)-1)\overline{A_{e^j}(2)}\,\overline{A_{e^j}(3)}\dotsb\overline{A_{e^j}(t-1)}\nonumber\\
			&\quad +A_{e^j}(2)(X_{e^j}(3)+X_{e^j}(4)+\dotsb+X_{e^j}(t)-1)\overline{A_{e^j}(3)}\,\overline{A_{e^j}(4)}\dotsb\overline{A_{e^j}(t-1)}\nonumber\\
			&\quad +\dotsb\nonumber\\
			&\quad +A_{e^j}(t-1)(X_{e^j}(t)-1)\\
			&=\sum_{j=1}^t A_{e^j}(j-1)\left(\sum_{\ell=j}^t X_{e^j}(\ell)-1\right)\prod_{k=j}^{t-1}\overline{A_{e^j}(k)}, \label{eq-mem_time_def2}
		\end{align}
		where $A_{e^j}(0)\equiv 1$ and $\overline{A_{e^j}(k)}\coloneqq 1-A_e(k)$ for all $k\geq 1$.~\defqed
	\end{definition}
	
	Intuitively, the quantity $M_{e^j}^{\pi}(t)$ is the number of consecutive time steps up to the $t^{\text{th}}$ time step that the action ``wait'' is performed since the most recent ``request'' action. The value $M_{e^j}^{\pi}(t)=-1$ can be thought of as the resting state of the quantum memory, when it is not loaded.

	Now, from Proposition~\ref{prop-QDP_cq_state}, we have that the classical-quantum state of an elementary link corresponding to an edge $e^j$ of a graph $G=(V,E,c)$ of a quantum network, with $e\in E$ and $1\leq j\leq c(e)$, is
	\begin{equation}\label{eq-network_QDP_cq_state}
		\widehat{\sigma}_{H_t^{e^j}E_t^{e^j}}^{\pi}(t)=\sum_{h^t\in\Omega(t)}\ket{h^t}\bra{h^t}_{H_t^{e^j}}\otimes\widetilde{\sigma}_{E_t^{e^j}}^{\pi}(t;h^t),
	\end{equation}
	where $\pi=(d_1,d_2,\dotsc,d_{t-1})$ is an arbitrary policy and
	\begin{equation}\label{eq-network_QDP_cond_state_unnormalized}
		\widetilde{\sigma}_{E_t^{e^j}}^{\pi}(t;h^t)=\left(\prod_{j=1}^{t-1}d_j(h_j^t)(a_j)\right)\left(\mathcal{T}_{E_{t-1}^{e^j}\to E_t^{e^j}}^{x_{t-1},a_{t-1},x_t}\circ\dotsb\circ\mathcal{T}_{E_2^{e^j}\to E_3^{e^j}}^{x_2,a_2,x_3}\circ\mathcal{T}_{E_1^{e^j}\to E_2^{e^j}}^{x_1,a_1,x_2}\circ\mathcal{T}_{E_0^{e^j}\to E_1^{e^j}}^{0;x_1}\right)(\rho_{E_0^{e^j}}^S).
	\end{equation}
	From Remark~\ref{rem-QDP_cond_states}, the conditional states are defined to be
	\begin{equation}\label{eq-network_QDP_cond_state}
		\sigma_{E_t^{e^j}}^{\pi}(t|h^t)=\frac{\widetilde{\sigma}_{E_t^{e^j}}^{\pi}(t;h^t)}{\Tr[\widetilde{\sigma}_{E_t^{e^j}}^{\pi}(t;h^t)]},
	\end{equation}
	and the probability of a history $h^t\in\Omega(t)$ under a policy $\pi^{e^j}$ is given by
	\begin{equation}\label{eq-network_QDP_history_prob}
		\Pr[H_{e^j}(t)=h^t]_{\pi}=\Tr[\widetilde{\sigma}_{E_t^{e^j}}^{\pi}(t;h^t)].
	\end{equation}
	From now on, for ease of notation, we make frequent use of the abbreviations
	\begin{equation}
		\widehat{\sigma}_{e^j}^{\pi}(t)\equiv\widehat{\sigma}_{H_t^{e^j}E_t^{e^j}}^{\pi}(t),\quad \widetilde{\sigma}_{e^j}^{\pi}(t;h^t)\equiv\widetilde{\sigma}_{E_t^{e^j}}^{\pi}(t;h^t),\quad \sigma_{e^j}^{\pi}(t|h^t)\equiv\sigma_{E_t^{e^j}}^{\pi}(t|h^t).
	\end{equation}
	In other words, we simply write ``$e^j$'' in the subscript instead of ``$E_t^{e^j}$'' whenever the time-dependence of the quantum system of the environment is understood or unimportant.
	
	Using the definition of the transition maps in \eqref{eq-network_QDP_trans_1}--\eqref{eq-network_QDP_trans_6}, along the memory time random variable from Definition~\ref{def-network_QDP_mem_time_RV}, we can derive an explicit expression for the conditional quantum state $\sigma_{e^j}^{\pi^e}(t|h^t)$.
	
	\begin{theorem}[Quantum state of an elementary link]\label{prop-link_quantum_state}
		Let $G=(V,E,c)$ be the graph corresponding to the physical (elementary) links of a quantum network, and let $e^j$, $e\in E$, $1\leq j\leq c(e)$, be arbitrary. For all $t\geq 1$ and histories $h^t=(x_1,a_1,\dotsc,a_{t-1},x_t)\in\Omega(t)$, and for all policies $\pi$, we have
		\begin{equation}\label{eq-link_state_unnormalized}
			\sigma_{e^j}^{\pi}(t|h^t)=x_t\,\rho_{e^j}\!\left(M_{e^j}^{\pi}(t)(h^t)\right)+(1-x_t)\tau_{e^j}^{\varnothing},
		\end{equation}
		where from \eqref{eq-elem_link_state_in_mem} we recall that
		\begin{equation}
			\rho_{e^j}(m)=\mathcal{N}_{e^j}^{\circ m}(\rho_{e^j}^0),
		\end{equation}
		and from \eqref{eq-initial_link_states} and \eqref{eq-elem_link_success_prob} we recall that
		\begin{equation}\label{eq-initial_link_states_2}
			\rho_{e^j}^0=\frac{(\mathcal{M}_{e^j}^1\circ\mathcal{S}_{e^j})(\rho_{e^j}^S)}{p_{e^j}},\quad \tau_{e^j}^{\varnothing}=\frac{(\mathcal{M}_{e^j}^0\circ\mathcal{S}_{e^j})(\rho_{e^j}^S)}{p_{e^j}},\quad p_{e^j}=\Tr[(\mathcal{M}_{e^j}^1\circ\mathcal{S}_{e^j})(\rho_{e^j}^S)].
		\end{equation}
		Furthermore,
		\begin{equation}\label{eq-hist_prob_general}
			\Pr[H_{e^j}(t)=h^t]_{\pi}=\left(\prod_{j=1}^{t-1} d_j(h_j^t)(a_j)\right) p_{e^j}^{N_{e^j}^{\text{succ}}(t)(h^t)}(1-p_{e^j})^{N_{e^j}^{\text{req}}(t)(h^t)-N_{e^j}^{\text{succ}}(t)(h^t)}
		\end{equation}
		for all histories $h^t$, where
		\begin{equation}
			N_{e^j}^{\text{req}}(t)\coloneqq\sum_{j=1}^t A_{e^j}(j-1),\quad N_{e^j}^{\text{succ}}(t)\coloneqq \sum_{j=1}^t A_{e^j}(j-1)X_{e^j}(j)
		\end{equation}
		are the number of link requests and the number of successful link requests, respectively, up to time~$t$.
	\end{theorem}
	
	\begin{proof}
		First, let us observe that the statement of the proposition is true for $t=1$, because by \eqref{eq-network_QDP_trans_5}, \eqref{eq-network_QDP_trans_6}, and \eqref{eq-network_QDP_cond_state_unnormalized} we can write
		\begin{equation}
			\widetilde{\sigma}_{e^j}^{\pi}(1;x_1)=x_1\widetilde{\rho}_{e^j}^0+(1-x_1)\widetilde{\tau}_{e^j}^{\varnothing},
		\end{equation}
		where $\widetilde{\rho}_{e^j}^0\coloneqq(\mathcal{M}_{e^j}^1\circ\mathcal{S}_{e^j})(\rho_{e^j}^S)$ and $\widetilde{\tau}_{e^j}^{\varnothing}\coloneqq(\mathcal{M}_{e^j}^0\circ\mathcal{S}_{e^j})(\rho_{e^j}^S)$. Then, indeed, we have $M_{e^j}^{\pi}(1)=0$ according to the definition in \eqref{eq-mem_time_def1}, as required, if $x_1=1$. Furthermore,
		\begin{equation}
			\Tr[\widetilde{\sigma}_{e^j}(1;x_1)]=x_1p_{e^j}+(1-x_1)(1-p_{e^j})=p_{e^j}^{x_1}(1-p_{e^j})^{1-x_1},
		\end{equation}
		so that
		\begin{align}
			\sigma_{e^j}^{\pi}(1|x_1)&=\frac{x_1\widetilde{\rho}_{e^j}^0+(1-x_1)\widetilde{\tau}_{e^j}^{\varnothing}}{p_{e^j}^{x_1}(1-p_{e^j})^{1-x_1}}\\
			&=\left\{\begin{array}{l l} \rho_{e^j}^0 & \text{if }x_1=1,\\ \tau_{e^j}^{\varnothing} & \text{if }x_1=0, \end{array}\right.\\
			&=x_1\rho_{e^j}^0+(1-x_1)\tau_{e^j}^{\varnothing}
		\end{align}
		where we recall the definitions of $\rho_{e^j}^0$ and $\tau_{e^j}^{\varnothing}$ from \eqref{eq-initial_link_states_2}.
		
		Now, for $t\geq 2$, we use \eqref{eq-network_QDP_cond_state_unnormalized}. Based on the definition of the transition maps, for every time step $j>1$ in which the action ``wait'' (i.e., $A_{e^j}(j)=0$) is performed and the elementary link is active (i.e., $X_{e^j}(j)=1$), the elementary link stays active at time step $j+1$, and thus by definition the memory time must be incremented by one, which is consistent with the definition of the memory time $M_{e^j}^{\pi}(t)$ given in \eqref{eq-mem_time_def1}, and the quantum state of the elementary link goes from $\rho_{e^j}(M_{e^j}^{\pi}(t))$ to $\rho_{e^j}(M_{e^j}^{\pi}(t)+1)$. If instead the elementary link is active at time $j$ and the action ``request'' is performed (i.e., $A_{e^j}(j)=1$), then the quantum state of the elementary link is discarded and is replaced either by the state $\rho_{e^j}^0$ (if $X_{e^j}(j+1)=1$) with probability $p_{e^j}$ or by the state $\tau_{e^j}^{\varnothing}$ (if $X_{e^j}(j+1)=0$) with probability $1-p_{e^j}$. In the former case, the memory time must be reset to zero, consistent with \eqref{eq-mem_time_def1}, and in the latter case, the memory time is $-1$, also consistent with \eqref{eq-mem_time_def1}.
		
		Furthermore, by definition of the transition maps, every time the action ``request'' is performed, we obtain a factor of $p_{e^j}$ (if the request succeeds) or $1-p_{e^j}$ (if the request fails). If the action ``wait'' is performed, then we obtain no additional multiplicative factors. The quantity $N_{e^j}^{\text{succ}}(t-1)$ is, by definition, equal to the number of requests that succeeded in $t-1$ time steps. Therefore, overall, we obtain a factor $p_{e^j}^{N_{e^j}^{\text{succ}}(t-1)}$ at the $(t-1)^{\text{st}}$ time step for the number of successful requests. The number of failed requests in $t-1$ time steps is given by
		\begin{align}
			\sum_{j=1}^{t-1} A_{e^j}(j-1)(1-X_{e^j}(j))&=\sum_{j=1}^{t-1} A_{e^j}(j-1)-\sum_{j=1}^{t-1} A_{e^j}(j-1)X_{e^j}(j)\\
			&=N_{e^j}^{\text{req}}(t-1)-N_{e^j}^{\text{succ}}(t-1),
		\end{align}
		so that we obtain an overall factor of $(1-p_{e^j})^{N_{e^j}^{\text{req}}(t-1)-N_{e^j}^{\text{succ}}(t-1)}$ at the $(t-1)^{\text{st}}$ time step for the failed requests. Also, the memory time at the $(t-1)^{\text{st}}$ time step is $M_{e^j}^{\pi}(t-1)(h_{t-1}^t)$, and then because the quantum state is either $\rho_{e^j}(M_{e^j}^{\pi}(t-1)(h_{t-1}^t))$ or $\tau_{e^j}^{\varnothing}$, we obtain
		\begin{align}
			\widetilde{\sigma}_{e^j}^{\pi}(t;h^t)&=\left(\prod_{j=1}^{t-1}d_j(h_j^t)(a_j)\right)p_{e^j}^{N_{e^j}^{\text{succ}}(t-1)(h_{t-1}^t)}(1-p_{e^j})^{N_{e^j}^{\text{req}}(t-1)(h_{t-1}^t)-N_{e^j}^{\text{succ}}(t-1)(h_{t-1}^t)}\nonumber\\
			&\qquad\qquad\times\left(x_{t-1}\mathcal{T}_{e^j}^{1,a_{t-1},x_t}(\rho(M_{e^j}^{\pi}(t-1)(h_{t-1}^t)))+(1-x_{t-1})\mathcal{T}_{e^j}^{0,a_{t-1},x_t}(\tau_{e^j}^{\varnothing})\right)\\
			&=\left(\prod_{j=1}^{t-1}d_j(h_j^t)(a_j)\right)p_{e^j}^{N_{e^j}^{\text{succ}}(t-1)(h_{t-1}^t)}(1-p_{e^j})^{N_{e^j}^{\text{req}}(t-1)(h_{t-1}^t)-N_{e^j}^{\text{succ}}(t-1)(h_{t-1}^t)}\nonumber\\
			&\qquad\qquad\times p_{e^j}^{a_{t-1}x_t}(1-p_{e^j})^{a_{t-1}(1-x_t)}(x_t\rho_{e^j}(M_{e^j}^{\pi}(t)(h^t))+(1-x_t)\tau_{e^j}^{\varnothing}) \\
			&=\left(\left(\prod_{j=1}^{t}d_j(h_j^t)(a_j)\right)p_{e^j}^{N_{e^j}^{\text{succ}}(t)(h^t)}(1-p_{e^j})^{N_{e^j}^{\text{req}}(t)(h^t)-N_{e^j}^{\text{succ}}(t)(h^t)}\right)(x_t\rho_{e^j}(M_{e^j}^{\pi}(t)(h^t))+(1-x_t)\tau_{e^j}^{\varnothing}).
		\end{align}
		Then, because $\Pr[H_{e^j}(t)=h^t]_{\pi}=\Tr[\widetilde{\sigma}_{e^j}^{\pi}(t;h^t)]$, we have
		\begin{equation}
			\Pr[H_{e^j}(t)=h^t]_{\pi}=\left(\prod_{j=1}^{t}d_j(h_j^t)(a_j)\right)p_{e^j}^{N_{e^j}^{\text{succ}}(t)(h^t)}(1-p_{e^j})^{N_{e^j}^{\text{req}}(t)(h^t)-N_{e^j}^{\text{succ}}(t)(h^t)},
		\end{equation}
		as required. Finally,
		\begin{equation}
			\sigma_{e^j}^{\pi}(t|h^t)=\frac{\widetilde{\sigma}_{e^j}^{\pi}(t;h^t)}{\Tr[\widetilde{\sigma}_{e^j}^{\pi}(t;h^t)]}=x_t\,\rho_{e^j}\!\left(M_{e^j}^{\pi}(t)(h^t)\right)+(1-x_t)\tau_{e^j}^{\varnothing},
		\end{equation}
		which completes the proof.
	\end{proof}
	
	Using Theorem~\ref{prop-link_quantum_state}, we immediately obtain an expression for the expected quantum state of an elementary link for all times $t\geq 1$, which is defined according to the general definition in \eqref{eq-QDP_exp_state}.
	
	\begin{corollary}[Expected quantum state of an elementary link]\label{cor-link_avg_q_state}
		Let $G=(V,E,c)$ be the graph corresponding to the physical (elementary) links of a quantum network, and let $e^j$, $e\in E$, $1\leq j\leq c(e)$, be arbitrary. For all $t\geq 1$ and for all policies $\pi$, the expected quantum state of the elementary link corresponding to $e^j$ is
		\begin{equation}\label{eq-link_avg_q_state}
			\sigma_{e^j}^{\pi}(t)=(1-\Pr[X_{e^j}(t)=1]_{\pi})\tau_{e^j}^{\varnothing}+\sum_{m}\Pr[X_{e^j}(t)=1,M_{e^j}^{\pi}(t)=m]_{\pi}\,\rho_{e^j}(m),
		\end{equation}
		where the sum is over all possible values of the memory time, which in general depends on the policy $\pi$.
	\end{corollary}
	
	\begin{proof}
		Using the result of Theorem~\ref{prop-link_quantum_state}, along with \eqref{eq-QDP_exp_state}, the expected quantum state of the link at time $t\geq 1$ is given by
		\begin{align}
			\sigma_{e^j}^{\pi}(t)&=\Tr_{H_t^{e^j}}[\widehat{\sigma}_{H_t^{e^j}E_t^{e^j}}^{\pi}(t)]\\
			&=\sum_{h^t\in\Omega(t)}\widetilde{\sigma}_{E_t^{e^j}}^{\pi}(t;h^t)\\
			&=\sum_{h^t\in\Omega(t)}\Pr[H_{e^j}(t)=h^t]_{\pi}\left(X_{e^j}(t)(h^t)\rho_{e^j}\!\left(M_{e^j}^{\pi}(t)(h^t)\right)+(1-X_{e^j}(t)(h^t))\tau_{e^j}^{\varnothing}\right)\\
			&=\sum_{h^t\in\Omega(t):x_t=0}\Pr[H_{e^j}(t)=h^t]_{\pi}\,\tau_{e^j}^{\varnothing}+\sum_{h^t\in\Omega(t):x_t=1}\Pr[H_{e^j}(t)=h^t]_{\pi}\,\rho_{e^j}\!\left(M_{e^j}^{\pi}(t)(h^t)\right)\\
			&=(1-\Pr[X_{e^j}(t)=1]_{\pi})\tau_{e^j}^{\varnothing}+\sum_{m}\Pr[X_{e^j}(t)=1,M_{e^j}^{\pi}(t)=m]_{\pi}\,\rho_{e^j}(m),
		\end{align}
		where to obtain the last equality we rearranged the sum over the set $\{h^t\in\Omega(t):x_t=1\}$ so that the sum is over the possible values of the memory time $m$, which in general depends on the policy $\pi$. This completes the proof.
	\end{proof}

	\begin{theorem}[Elementary link success probability]\label{thm-link_prob}
		Let $G=(V,E,c)$ be the graph corresponding to the physical (elementary) links of a quantum network, and let $e^j$, $e\in E$, $1\leq j\leq c(e)$, be arbitrary. For all $t\geq 1$ and for all policies $\pi$,
		\begin{equation}\label{eq-link_value_prob_via_state}
			\Pr[X_{e^j}(t)=1]_{\pi}=\Tr\left[\ket{1}\bra{1}_{X_t^{e^j}}\widehat{\sigma}_{H_t^{e^j}E_t^{e^j}}^{\pi}(t)\right]=\mathbb{E}[X_{e^j}(t)]_{\pi},
		\end{equation}
		where we recall the definition of the classical-quantum state $\widehat{\sigma}_{H_t^{e^j}E_t^{e^j}}^{\pi}(t)$ in \eqref{eq-network_QDP_cq_state}.
	\end{theorem}
	
	\begin{proof}
		To see the first equality in \eqref{eq-link_value_prob_via_state}, observe that
		\begin{equation}
			\Tr\left[\ket{1}\bra{1}_{X_t^{e^j}}\widehat{\sigma}_{H_t^{e^j}E_t^{e^j}}^{\pi}(t)\right]=\sum_{h^t\in\Omega(t):X_{e^j}(t)(h^t)=1}\Pr[H_{e^j}(t)=h^t]_{\pi}.
		\end{equation}
		The expression on the right-hand side of this equation is equal to $\Pr[X_{e^j}(t)=1]_{\pi}$ by definition of the random variable $X_{e^j}(t)$. The second equality in \eqref{eq-link_value_prob_via_state} holds because $X_{e^j}(t)$ is a binary/Bernoulli random variable.
	\end{proof}
	
	The expected quantum state of the elementary link at time $t\geq 1$, given that the elementary link is active at time $t$, is defined to be
	\begin{align}
		\sigma_{e^j}^{\pi}(t|X_{e^j}(t)=1)&\coloneqq \frac{\Tr_{H_t^{e^j}}\left[\ket{1}\bra{1}_{X_t^{e^j}}\widehat{\sigma}_{H_t^{e^j}E_t^{e^j}}^{\pi}(t)\right]}{\Pr[X_{e^j}(t)=1]_{\pi}}\label{eq-link_avg_q_state_conditional_def}\\[0.2cm]
		&=\sum_{m}\Pr[M_{e^j}^{\pi}(t)=m|X_{e^j}(t)=1]_{\pi}\,\rho_{e^j}(m).\label{eq-link_avg_q_state_conditional}
	\end{align}

	Observe that the expressions in \eqref{eq-link_avg_q_state} and \eqref{eq-link_avg_q_state_conditional} hold for any policy of the agent. Given a particular policy, determining the expected quantum state means determining the joint probability distribution of the random variables $X_{e^j}(t)$ and $M_{e^j}(t)$, i.e., determining the quantities $\Pr[X_{e^j}(t)=1,M_{e^j}^{\pi}(t)=m]_{\pi}$ for all possible values of $m$. The probability distribution of $X_{e^j}(t)$ can then be obtained via marginalization, i.e., via $\Pr[X_{e^j}(t)=1]_{\pi}=\sum_m \Pr[X_{e^j}(t)=1,M_{e^j}^{\pi}(t)=m]_{\pi}$, where the sum is over all possible values of the memory random variable $M_{e^j}^{\pi}(t)$ (which can depend on the policy), or it can be obtained using \eqref{eq-link_avg_q_state_conditional}.
	
	Let us now prove a simple expression for the expected reward after $t$ time steps of the decision process.
	
	\begin{theorem}\label{thm-network_QDP_exp_reward_no_distill}
		Let $G=(V,E,c)$ be the graph corresponding to the physical (elementary) links of a quantum network, and let $e^j$, $e\in E$, $1\leq j\leq c(e)$, be arbitrary. For all $t\geq 1$ and for all policies $\pi$,
		\begin{equation}\label{eq-opt_pol_obj_func}
			\mathbb{E}[R_{e^j}(t)]_{\pi}=\Tr\left[\left(\ket{1}\bra{1}_{X_{t+1}^{e^j}}\otimes\psi^{\text{target}}_{e^j}\right)\widehat{\sigma}_{e^j}^{\pi}(t+1)\right],
		\end{equation}
		where $\psi_{e^j}^{\text{target}}$ is a pure target state for the elementary link corresponding to ${e^j}$; see Definition~\ref{def-network_QDP_elem_link}.
	\end{theorem}
	
	\begin{proof}
		From \eqref{eq-QDP_reward_map_tilde}, we have that
		\begin{equation}\label{eq-network_QDP_reward_map_tilde}
			\widetilde{\mathcal{R}}_{E_{t+1}^{e^j}}^{t;h^{t+1}}(\cdot)=\sum_{s\in\{0,1\}}R_{e^j}(t)(h^{t+1},s)\mathcal{R}_{E_{t+1}^{e^j}}^{t;h^{t+1}}=\delta_{x_{t+1},1}\psi_{e^j}^{\text{target}}(\cdot)\psi_{e^j}^{\text{target}}.
		\end{equation}
		Therefore, the map $\widehat{\mathcal{R}}_{H_{t+1}^{e^j}E_{t+1}^{e^j}\to E_{t+1}^{e^j}}(t)$ defined in \eqref{eq-QDP_reward_map_hat} is given by
		\begin{equation}
			\widehat{\mathcal{R}}_{H_{t+1}^{e^j}E_{t+1}^{e^j}\to E_{t+1}^{e^j}}(t)\left(\ket{h^{t+1}}\bra{h^{t+1}}_{H_{t+1}^{e^j}}\otimes\sigma_{E_{t+1}^{e^j}}\right)=\delta_{x_{t+1},1}\psi_{e^j}^{\text{target}}(\sigma_{E_{t+1}^{e^j}})\psi_{e^j}^{\text{target}}
		\end{equation}
		for all $h^{t+1}\in\Omega(t+1)$ and all states $\sigma_{E_{t+1}^{e^j}}$. We can equivalently write this as
		\begin{equation}
			\widehat{\mathcal{R}}_{H_{t+1}^{e^j}E_{t+1}^{e^j}\to E_{t+1}^{e^j}}(t)\left(\rho_{H_{t+1}^{e^j}E_{t+1}^{e^j}}\right)=\left(\bra{1}_{X_{t+1}^{e^j}}\otimes\psi_{E_{t+1}^{e^j}}^{\text{target}}\right)\left(\rho_{H_{t+1}^{e^j}E_{t+1}^{e^j}}\right)\left(\ket{1}_{X_{t+1}^{e^j}}\otimes\psi_{E_{t+1}^{e^j}}^{\text{target}}\right)
		\end{equation}
		for an arbitrary state $\rho_{H_{t+1}^{e^j}E_{t+1}^{e^j}}$. Therefore, using \eqref{eq-QDP_exp_reward_formula_2}, we obtain
		\begin{align}
			\mathbb{E}[R_{e^j}(t)]_{\pi}&=\Tr\left[\left(\bra{1}_{X_{t+1}^{e^j}}\otimes\psi_{E_{t+1}^{e^j}}^{\text{target}}\right)\left(\widehat{\sigma}_{H_{t+1}^{e^j}E_{t+1}^{e^j}}^{\pi}(t+1)\right)\left(\ket{1}_{X_{t+1}^{e^j}}\otimes\psi_{E_{t+1}^{e^j}}^{\text{target}}\right)\right]\\
			&=\Tr\left[\left(\ket{1}\bra{1}_{X_{t+1}^{e^j}}\otimes\psi_{E_{t+1}^{e^j}}^{\text{target}}\right)\widehat{\sigma}_{H_{t+1}^{e^j}E_{t+1}^{e^j}}(t+1)\right],
		\end{align}
		as required.
	\end{proof}
	
	Before moving on, we mention that the entirety of Chapter~\ref{chap-mem_cutoff} is devoted to the study of a particular policy for the quantum decision process defined in Definition~\ref{def-network_QDP_elem_link}, namely the ``memory-cutoff'' policy. For this policy, we derive explicit forms for all of the results derived above.

\section{Figures of merit}\label{sec-practical_figures_merit}
	
	Let us now consider figures of merit for evaluating policies in the context of entanglement distribution. This is in preparation for policy optimization, which we consider in Section~\ref{sec-network_QDP_pol_opt} below. The reward scheme given in Definition~\ref{def-network_QDP_elem_link} already provides us with a figure of merit (the expected reward) for evaluating policies, and in this section we justify the particular form of that reward scheme by showing that the expected reward is related to the fidelity of the quantum state of the elementary link to the target state. However, in practical settings, we might be interested in more than just the fidelities of the elementary and virtual links to the target states. Indeed, due to the probabilistic nature of elementary link generation, in the context of entanglement distribution as outlined in Section~\ref{sec-ent_dist_general}, we must regard the input and output graphs $G_{\text{in}}$ and $G_{\text{out}}$ in \eqref{eq-ent_dist_protocol_graphs_trans} as random variables, based on the formulation in Section~\ref{sec-random_subgraphs}. We must therefore take into consideration the probability that the desired subgraphs are obtained. In addition, because memory coherence times are finite, the figures of merit that we define (particularly the fidelity) should take the memory times into account. In this section, we provide such ``practical'' figures of merit.

\subsection{Link statuses}\label{sec-network_QDP_link_statuses}

	The figures of merit that we define in this section are based on obtaining the target graph of an entanglement distribution protocol with high probability. Before defining the figures of merit, let us briefly review the content of Section~\ref{sec-random_subgraphs}. 

	To start, let $G=(V,E,c)$ be the graph corresponding to the physical (elementary) links of a quantum network. Recall from Section~\ref{sec-random_subgraphs} that the set $\Xi_G$ gives us the configurations of the graph based on probabilities $\vec{p}=(p_{e^j}:e\in E,\,1\leq j\leq c(e))$ assigned to the individual elementary links. Then, we define $G_{\vec{p}}$ to be a random variable such that $G_{\vec{p}}(\vec{x})=(V(\vec{x}),E(\vec{x}),c)$, $\vec{x}\in\Xi_G$, is the subgraph of $G$ that contains only those edges corresponding to active elementary links, as specified by $\vec{x}$, so that the probability of obtaining the subgraph $G_{\vec{p}}(\vec{x})$ is
	\begin{equation}
		\Pr[G_{\vec{p}}(\vec{x})]=\prod_{e\in E}\prod_{j=1}^{c(e)}\Pr[X_{e^j}=x^{e^j}],
	\end{equation}
	where $X_{e^j}$ is the random variable for the edge $e^j$ such that $\Pr[X_{e^j}=x]=p_{e^j}^x(1-p_{e^j})^{1-x}$ for all $x\in\{0,1\}$.
	
	Now, a quantum decision process for the elementary links, as given by Definition~\ref{def-network_QDP_elem_link}, provides time-dependent elementary link probabilities via the random variables $X_{e^j}(t)$. From this, we have that the random graphs $G_{\vec{p}}$ are also time dependent. Specifically, for a collection $\vec{\pi}=\left(\pi^{e^j}:e\in E,\,1\leq j\leq c(e)\right)$ of policies for every elementary link, we define 
	\begin{equation}\label{eq-network_QDP_subgraph_RV_policy}
		G_{\vec{\pi}}(t)\coloneqq G_{\vec{p}_{\vec{\pi}}(t)},\quad\vec{p}_{\vec{\pi}}(t)\coloneqq\left(\Pr\left[X_{e^j}(t)=1\right]_{\pi^{e_j}}:e\in E,\,1\leq j\leq c(e)\right)
	\end{equation}
	to be a sequence of random variables for all $t\geq 1$ such that
	\begin{equation}\label{eq-network_QDP_subgraph_probability}
		\Pr[G_{\vec{\pi}}(t)(\vec{x})]=\prod_{e\in E}\prod_{j=1}^{c(e)}\Pr[X_{e^j}(t)=x^{e^j})]_{\pi^{e^j}}
	\end{equation}
	for all configurations $\vec{x}\in\Xi_G$. Note that the product form for this probability is due to the fact that all of the agents are independent, which means that the random variables $X_{e^j}(t)$ are mutually independent for all $e^j$.
	
	The expression in \eqref{eq-network_QDP_subgraph_probability} can be written in a simpler way by making the following definition.
	
	\begin{definition}[Collective elementary link status]\label{def-network_QDP_collective_elem_link_status}
		Let $G=(V,E,c)$ be the graph corresponding to the physical (elementary) links of a quantum network. For every subset $E'\subseteq E$, we define
		\begin{equation}\label{eq-network_QDP_collective_link_status}
			X_{E'}(t)\coloneqq \prod_{e\in E'}\prod_{j=1}^{ c(e)}X_{e^j}(t)
		\end{equation}
		to be the \textit{collective elementary link status}, so that
		\begin{equation}\label{eq-network_QDP_exp_link_status_collective}
			\Pr[X_{E'}(t)=1]_{\vec{\pi}}=\mathbb{E}[X_{E'}(t)]_{\vec{\pi}}=\prod_{e\in E'}\prod_{j=1}^{ c(e)}\Pr[X_{e^j}(t)=1]_{\pi^{e^j}}
		\end{equation}
		for every collection $\vec{\pi}=\left(\pi^{e^j}:e\in E',\,1\leq j\leq c(e)\right)$ of elementary link policies.~\defqed
	\end{definition}
	
	Using \eqref{eq-network_QDP_exp_link_status_collective}, the expression in \eqref{eq-network_QDP_subgraph_probability} can be rewritten as
	\begin{equation}\label{eq-network_QDP_subgraph_probability_alt}
		\Pr[G_{\vec{\pi}}(t)(\vec{x})]=\Pr[X_{E(\vec{x})}(t)=1]_{\vec{\pi}}\cdot\Pr[X_{E\setminus E(\vec{x})}(t)=0]_{\vec{\pi}}.
	\end{equation}
	
	Now, suppose that we have a desired target graph of active elementary links, and this target graph corresponds to the configuration $\vec{x}_{\text{target}}$. In order to have a high probability of obtaining this configuration, we thus have to maximize the quantity in \eqref{eq-network_QDP_subgraph_probability_alt} with respect to policies. On the other hand, it might be the case that we are concerned more about establishing the elementary links corresponding to the subset $E(\vec{x}_{\text{target}})$ than about obtaining the exact configuration $\vec{x}_{\text{target}}$; i.e., we might not necessarily care about the status of the elementary links that are not in $E(\vec{x}_{\text{target}})$. The relevant probability is then obtained by marginalizing the probability in \eqref{eq-network_QDP_subgraph_probability_alt} over the elements of $E\setminus E(\vec{x}_{\text{target}})$, and doing so results in $\Pr[X_{E(\vec{x}_{\text{target}})}(t)=1]_{\vec{\pi}}$.  Therefore, it makes sense in many cases to simply maximize the quantity $\Pr[X_{E'}(t)=1]_{\vec{\pi}}=\mathbb{E}[X_{E'}(t)]_{\vec{\pi}}$, for some desired subset $E'$ of elementary links, rather than the product in \eqref{eq-network_QDP_subgraph_probability_alt}.
	
	In addition to the status of elementary links, we also need to consider the status of virtual links. For simplicity, in the definition below we consider a graph $G$ consisting only of two-element edges.
	
	\begin{definition}[Virtual link status]\label{def-network_QDP_virtual_link_status}
		Let $G=(V,E,c)$ be the graph corresponding to the physical (elementary) links of a quantum network, and let $E$ be a set of two-element subsets of $V$. Given a pair $v_1,v_n\in V$ of distinct non-adjacent vertices and a path $w=(v_1,e_1,v_2,e_2,\dotsc,e_{n-1},v_n)$ between them for some $n\geq 2$, the \textit{virtual link status}, i.e., the status of the virtual link given by the edge $\{v_1,v_n\}$, is defined to be the random variable
		\begin{equation}
			X_{\{v_1,v_n\};w}(t)=X_{e_1}(t)Y_{e_1,e_2}X_{e_2}(t)Y_{e_2,e_3}\dotsb X_{e_{n-2}}(t)Y_{e_{n-2},e_{n-1}}X_{e_{n-1}}(t),
		\end{equation}
		where $Y_{e_i,e_j}$ is the binary random variable for the joining measurement that connects the elementary links given by $e_i$ and $e_j$, such that
		\begin{equation}
			\Pr[Y_{e_i,e_j}=1]=q_{e_i,e_j},
		\end{equation}
		for probabilities $q_{e_i,e_j}\in[0,1]$.~\defqed
	\end{definition}

\subsection{Fidelity}\label{sec-network_QDP_fidelity}

	We now look at the fidelity of elementary and virtual links.
	
	\begin{definition}[Elementary link fidelity]\label{def-network_QDP_elem_link_fidelity}
		Let $G=(V,E,c)$ be the graph corresponding to the physical (elementary) links of a quantum network, let $e^j$, $e\in E$, $1\leq j\leq c(e)$, be arbitrary, and consider the quantum decision process for $e^j$ given by Definition~\ref{def-network_QDP_elem_link}. For an arbitrary policy $\pi$ and a target pure state $\psi$, we define the following two random variables:
		\begin{equation}\label{eq-network_QDP_elem_link_fidelity}
			\widetilde{F}_{e^j}^{\pi}(t;\psi)\coloneqq X_{e^j}(t)f_{M_{e^j}^{\pi}(t)}^{e^j}(\rho_{e^j}^0;\psi),\quad F_{e^j}^{\pi}(t;\psi)\coloneqq\frac{\widetilde{F}_{e^j}^{\pi}(t;\psi)}{\Pr[X_{e^j}(t)=1]_{\pi}},
		\end{equation}
		where $M_{e^j}^{\pi}(t)$ is the memory time random variable defined in Definition~\ref{def-network_QDP_mem_time_RV}, $\rho_{e^j}^0=\frac{1}{p_{e^j}}(\mathcal{M}_{e^j}^1\circ\mathcal{S}_{e^j})(\rho_{e^j}^S)$, and $f_m^{e^j}(\rho_{e^j}^0;\psi)$ is defined in \eqref{eq-fidelity_decay}; see Section~\ref{sec-practical_elem_link_generation}.~\defqed
	\end{definition}
		
	We discuss the distinction between the random variables $\widetilde{F}$ and $F$ in detail in Section~\ref{sec-network_QDP_pol_opt} in the context of policy optimization. For now, it suffices to note that, intuitively, $F$ can be thought of as the fidelity of the elementary link \textit{given} that the elementary link is active. As for $\widetilde{F}$, it turns out that it is related to the reward of the quantum decision process defined in Definition~\ref{def-network_QDP_elem_link}.
	
	\begin{proposition}\label{prop-network_QDP_exp_reward_fid_tilde}
		Let $G=(V,E,c)$ be the graph corresponding to the physical (elementary) links of a quantum network, let $e^j$, $e\in E$, $1\leq j\leq c(e)$ be arbitrary, and consider the quantum decision process for $e^j$ given by Definition~\ref{def-network_QDP_elem_link}. For an arbitrary policy $\pi$ and a target pure state $\psi$,
		\begin{equation}
			\mathbb{E}[R_{e^j}(t)]_{\pi}=\mathbb{E}[\widetilde{F}_{e^j}^{\pi}(t+1;\psi)]=\Tr\left[\left(\ket{1}\bra{1}_{X_{t+1}^{e^j}}\otimes\psi_{e^j}\right)\widehat{\sigma}_{e^j}^{\pi}(t+1)\right].
		\end{equation}
	\end{proposition}
	
	\begin{proof}
		First of all, by the definition of expectation, we have
		\begin{equation}\label{eq-avg_Ftilde}
			\mathbb{E}[\widetilde{F}_{e^j}^{\pi}(t+1;\psi)]=\sum_{m} f_m^{e^j}(\rho_{e^j}^0;\psi)\Pr[X_{e^j}(t+1)=1,M_{e^j}^{\pi}(t+1)=m]_{\pi},
		\end{equation}
		where the sum is over all possible values of the random variable $M_{e^j}^{\pi}(t)$, which depends on the policy $\pi$. Then, by Theorem~\ref{thm-network_QDP_exp_reward_no_distill} and Theorem~\ref{prop-link_quantum_state},
		\begin{align}
			\mathbb{E}[R_{e^j}(t)]_{\pi}&=\Tr\left[\left(\ket{1}\bra{1}_{X_{t+1}^{e^j}}\otimes\psi_{e^j}\right)\widehat{\sigma}_{e^j}^{\pi}(t+1)\right]\\
			&=\sum_{h^{t+1}:x_{t+1}=1}\Pr[H_{e^j}(t+1)=h^{t+1}]_{\pi}\bra{\psi}\rho(M_{e^j}^{\pi}(t+1)(h^{t+1}))\ket{\psi}\\
			&=\sum_{h^{t+1}:x_{t+1}=1}\Pr[H_{e^j}(t+1)=h^{t+1}]_{\pi}f_{M_{e^j}^{\pi}(t+1)(h^{t+1})}^{e^j}(\rho_{e^j}^0;\psi)\\
			&=\sum_m f_m^{e^j}(\rho_{e^j}^0;\psi)\Pr[X_{e^j}(t+1)=1,M_{e^j}^{\pi}(t+1)=m]_{\pi},
		\end{align}
		where the last equality holds because the sum with respect to the set $\{h^{t+1}:x_{t+1}=1\}$ can be rearranged into a sum with respect to the possible values of the memory time $M_{e^j}^{\pi}(t+1)$ when the elementary link is active. This completes the proof.
	\end{proof}

	Let us now consider the fidelity of virtual links. For simplicity, we consider the situation in which the graph $G$ contains only two-element edges.
	
	Let $G=(V,E)$ be the graph corresponding to the physical (elementary) links of a quantum network, such that $E$ is a set of two-element subsets of $V$. Consider a pair $v_1,v_{n+1}\in V$ of distinct non-adjacent vertices and a path $w=(v_1,e_1,v_2,e_2,\dotsc,e_{n},v_{n+1})$ between them for some $n\geq 2$, and let $\mathcal{L}_{w\to e'}$ denote the quantum channel corresponding to a joining protocol that creates the virtual link corresponding to the edge $e'\coloneqq\{v_1,v_{n+1}\}$, as in \eqref{eq-q_instr_channel_swapping}. Then, for policies $\pi^{e_i}$ for the elementary links along the path, the expected quantum state at the output of the joining protocol at time $t\geq 1$ is given by
	\begin{multline}
		\mathcal{L}_{w\to e'}\left(\bigotimes_{i=1}^{n} \sigma_{e_i}^{\pi^{e_i}}(t|X_{e_i}(t)=1)\right)=\ket{0}\bra{0}\otimes\mathcal{L}_{w\to e'}^{0}\left(\bigotimes_{i=1}^{n} \sigma_{e_i}^{\pi^{e_i}}(t|X_{e_i}(t)=1)\right)\\+\ket{1}\bra{1}\otimes\mathcal{L}_{w\to e'}^1\left(\bigotimes_{i=1}^{n} \sigma_{e_i}^{\pi^{e_i}}(t|X_{e_i}(t)=1)\right),
	\end{multline}
	where the states $\sigma_{e_i}^{\pi^{e_i}}(t|X_{e_i}=1)$ are defined in \eqref{eq-link_avg_q_state_conditional_def}. If $\ket{\psi^{\text{target}}}\bra{\psi^{\text{target}}}_{e'}$ is a target pure state for the virtual link given by the edge $e'$, then the fidelity of the output state with respect to the target state, conditioned on success of the protocol, is
	\begin{equation}\label{eq-network_QDP_virtual_link_fid_cond}
		\frac{1}{p_{\text{succ}}}\bra{\psi^{\text{target}}}_{e'}\mathcal{L}_{w\to e'}^1\left(\bigotimes_{i=1}^{n} \sigma_{e_i}^{\pi^{e_i}}(t|X_{e_i}(t)=1)\right)\ket{\psi^{\text{target}}}_{e'},
	\end{equation}
	where
	\begin{equation}
		p_{\text{succ}}=\Tr\left[\mathcal{L}_{w\to e'}^1\left(\bigotimes_{i=1}^{n} \sigma_{e_i}^{\pi^{e_i}}(t|X_{e_i}(t)=1)\right)\right].
	\end{equation}

\subsection{Waiting time}\label{sec-network_QDP_waiting_time}

	An important question when dealing with probabilistic elementary link generation is the \textit{expected waiting time}, which is a figure of merit that indicates how long it takes (on average) to establish an elementary or virtual link. This figure of merit has been considered in prior work in the context of a linear chain of quantum repeaters \cite{CJKK07,BPv11,SSv19,VK19,BCE19}. Here, we define the waiting time for elementary and virtual links in the context of quantum network protocols using quantum decision processes.
	
	When defining the waiting times, we imagine a scenario in which elementary link generation is continuously occurring in the network \cite{CRDW19} and that an end-user request for entanglement occurs at a time $t_{\text{req}}\geq 0$. The waiting time is then the number of time steps from time $t_{\text{req}}$ onward that it takes to establish the entanglement.
	
	\begin{definition}[Elementary link waiting time]\label{def-network_QDP_elem_link_waiting_time}
		Let $G=(V,E,c)$ be the graph corresponding to the physical (elementary) links of a quantum network, let $e\in E$ be arbitrary, and consider the quantum decision process for $e^j$ given by Definition~\ref{def-network_QDP_elem_link}, where $1\leq j\leq c(e)$. For all $t_{\text{req}}\geq 0$, the waiting time for the elementary link corresponding to the edge $e^j$ is defined to be
		\begin{equation}
			W_{e^j}(t_{\text{req}})\coloneqq\sum_{t=t_{\text{req}}+1}^{\infty} t X_{e^j}(t)\prod_{i=t_{\text{req}}+1}^{t-1}(1-X_{e^j}(i)).
		\end{equation}
		Then, the expected waiting time is
		\begin{equation}\label{eq-waiting_time_prob_late_request}
			\mathbb{E}[W_{e^j}(t_{\text{req}})]_{\pi}=\sum_{t=t_{\text{req}}+1}^{\infty} t\Pr[X_{e_j}(t_{\text{req}}+1)=0,X_{e_j}(t_{\text{req}}+2)=0,\dotsc,X_{e^j}(t_{\text{req}}+t)=1]_{\pi},
		\end{equation}
		where $\pi$ is an arbitrary policy for the elementary link corresponding to the edge $e^j$.~\defqed
	\end{definition}
	
	Using the collective elementary link status defined in Definition~\ref{def-network_QDP_collective_elem_link_status}, we make the following definition for the waiting time for a collection of elementary links.
	
	\begin{definition}[Collective elementary link waiting time]\label{def-collective_elem_link_waiting_time}
		Let $G=(V,E,c)$ be the graph corresponding to the physical (elementary) links of a quantum network, and let $t_{\text{req}}\geq 0$ be arbitrary. For every subset $E'\subseteq E$, the waiting time for the elementary links corresponding to the elements of $E'$ is defined to be
		\begin{equation}
			W_{E'}(t_{\text{req}})\coloneqq\sum_{t=t_{\text{req}}+1}^{\infty} tX_{E'}(t)\prod_{i=t_{\text{req}}+1}^{t-1}(1-X_{E'}(i)).~\defqedspec
		\end{equation}
	\end{definition}

	In other words, the collective elementary link waiting time is the time it takes for all of the elementary links given by $E'$ to be simultaneously active, and its expected value is
	\begin{equation}
		\mathbb{E}[W_{E'}(t_{\text{req}})]_{\pi}=\sum_{t=t_{\text{req}}+1}^{\infty} t\Pr[X_{E'}(t_{\text{req}}+1)=0,X_{E'}(t_{\text{req}}+2)=0,\dotsc,X_{E'}(t_{\text{req}}+t)=1]_{\vec{\pi}},
	\end{equation}
	where $\vec{\pi}=\left(\pi^{e^j}:e\in E',\,1\leq j\leq c(e)\right)$ is an arbitrary collection of policies for the elementary links corresponding to $E'$.
	
	Let us now consider virtual links.

	\begin{definition}[Virtual link waiting time]\label{def-network_QDP_virtual_waiting_time}
		Let $G=(V,E,c)$ be the graph corresponding to the physical (elementary) links of a quantum network, and let $t_{\text{req}}\geq 0$ be arbitrary. Given a pair $v_1,v_n\in V$ of distinct non-adjacent vertices and a path $w=(v_1,e_1,v_2,e_2,\dotsc,e_{n-1},v_n)$ between them for some $n\geq 2$, the \textit{virtual link waiting time} along this path is defined to be the amount of time it takes to establish the virtual link given by the edge $\{v_1,v_n\}$: 
		\begin{equation}
			W_{\{v_1,v_n\};w}(t_{\text{req}})\coloneqq W_{E_w}(t_{\text{req}})\sum_{t=t_{\text{req}}+1}^{\infty}t Y_{w}(1-Y_{w})^{t-1},
		\end{equation}
		where $E_w=\{e_1,e_2,\dotsc,e_{n-1}\}$ is the set of edges corresponding to the path $w$, $W_{E_w}(t_{\text{req}})$ is the collective elementary link waiting time from Definition~\ref{def-collective_elem_link_waiting_time}, and $Y_{E_w}$ is a binary random variable for the success of the joining protocol along the path $w$, so that $Y_w=1$ corresponds to success of the joining protocol and $Y_w=0$ to failure. We define $Y_w$ and $W_{E_w}$ to be independent random variables.~\defqed
	\end{definition}
	
	The formula for the virtual link waiting time in Definition~\ref{def-network_QDP_virtual_waiting_time} is based on the formula in Ref.~\cite{CJKK07}. It corresponds to the simple strategy of waiting for all of the elementary links along the path $w$ to be established, then performing the measurements for the joining protocol. Note that this strategy is consistent with our overall quantum network protocol in Figure~\ref{fig-QDP_protocol}. The expected value of the virtual link waiting time in this case is
	\begin{align}
		\mathbb{E}[W_{\{v_1,v_n\};w}(t_{\text{req}})]_{\vec{\pi}}&=\mathbb{E}[W_{E_w}(t_{\text{req}})]_{\pi}\sum_{t=t_{\text{req}}+1}^{\infty}t\mathbb{E}[Y_w(1-Y_w)^{t-1}]\\
		&=\mathbb{E}[W_{E_w}(t_{\text{req}})]_{\pi}\left(\frac{(1-q)^{t_{\text{req}}}(1+qt_{\text{req}})}{q}\right),
	\end{align}
	where $q\coloneqq\Pr[Y_w=1]$ and $\vec{\pi}=\left(\pi^{e^j}:e\in E_w,\,1\leq j\leq c(e)\right)$ is an arbitrary collection of policies for the elementary links corresponding to $E_w$ If $t_{\text{req}}=0$, then
	\begin{equation}
		\mathbb{E}[W_{\{v_1,v_n\};w}(0)]_{\vec{\pi}}=\frac{\mathbb{E}[W_{E_w}(0)]_{\pi}}{q}.
	\end{equation}

\begin{comment}
	\begin{remark}
		...This gives an upper bound on the average waiting time for an end-to-end link (i.e., it is the worst case); improved waiting time bounds using special policy for multiple links under this strategy; also give analytical formula for the average fidelity of each link corresponding to this waiting time. In this case (i.e., the special policy), look at the final end-to-end fidelity. Placing a minimum requirement for the end-to-end fidelity, find out what the elementary link fidelity has to be, and thus how long one must wait before starting the link generation process
		
		[In general, we can do joining measurement in other orders (e.g., immediately after neighboring elementary links are active), which would give better (i.e., lower) average waiting time...see \cite{SSv19} for linear chains....]
		
		...\cite{BCE19} for numerical computations of waiting time in linear chains...
	\end{remark}

%	[...in the context of the decision processes, investigate how many time steps minimum are needed to attain a certain fidelity...this would be related to waiting time...]
\end{comment}

\subsection{Rates}\label{sec-network_QDP_rates}

	We now define two types of rates.
	
	\begin{definition}[Success rate of an elementary link]\label{def-network_QDP_elem_link_succ_rate}
		Let $G=(V,E,c)$ be the graph corresponding to the physical (elementary) links of a quantum network, let $e\in E$ be arbitrary, and consider the quantum decision process for $e^j$ given by Definition~\ref{def-network_QDP_elem_link}, with $1\leq j\leq c(e)$. For all $t\geq 1$, we define the \textit{success rate} up to time $t$ of $e^j$ as
		\begin{equation}
			S_{e^j}(t)\coloneqq\frac{\displaystyle\sum_{i=1}^t A_{e^j}(i-1)X_{e^j}(i)}{\displaystyle\sum_{i=1}^t A_{e^j}(i-1)},
		\end{equation}
		which is simply the ratio of the number of successful transmissions when a request is made to the total number of requests made within time $t$. We let $A(0)\equiv 1$.~\defqed
	\end{definition}
		
	The success rate can also be thought as the number of successful requests per channel use up to time $t$. Indeed, notice that the quantity $\sum_{i=1}^t A_{e^j}(i-1)$ in the denominator of $S_{e^j}(t)$ is the number of uses of the transmission channel in $t$ time steps.
	
%	[...here, we have made use of the fact that the action random variables $A(t)$ take only two values for all times $t\geq 1$. In the case that more than two actions are allowed, as discussed in Remark~\ref{rem-multiple_actions}, the definition would have to be suitably modified...]
	
	\begin{definition}[Elementary link activity rate]\label{def-network_QDP_elem_link_act_rate}
		Let $G=(V,E,c)$ be the graph corresponding to the physical (elementary) links of a quantum network, and let $e\in E$ be arbitrary. Then, the \textit{elementary link activity rate} up to time $t$ corresponding to $e$ is
		\begin{equation}
			r_e(t)\coloneqq\frac{1}{t}\sum_{i=1}^t N_e(i),
		\end{equation}
		where $N_e(i)$ is defined to be the number of parallel elementary links corresponding to $e$ at time $i$, given by \eqref{eq-network_QDP_num_active_parallel_links}.~\defqed
	\end{definition}
	
	We can define an activity rate for virtual links as well. In this case, we must look at the number of edge-disjoint paths (or Steiner trees, as the situation warrants) between the nodes given by $e$. See Section~\ref{sec-graph_theory} and Ref.~\cite[Section~5]{BA17} for more information on finding edge-disjoint paths in (hyper)graphs.
	
	\begin{definition}[Virtual link activity rate]
		Let $G=(V,E,c)$ be the graph corresponding to the physical (elementary) links of a quantum network, and let $e'\coloneqq\{v_1,\dotsc,v_k\}\notin E$ be a collection of distinct nodes corresponding to a virtual link for some $k\geq 2$. Then, the \textit{virtual link activity rate} up to time $t$ corresponding to $e'$ is
		\begin{equation}
			r_{e'}(t)\coloneqq\frac{1}{t}\sum_{i=1}^t N_{e'}(i),
		\end{equation}
		where
		\begin{equation}
			N_{e'}(i)\coloneqq\sum_{j=1}^{N_{e'}^{\max}} X_{e';w_j}(i)
		\end{equation}
		is the number of edge-disjoint paths (or Steiner trees) corresponding to $e'$, with $w_j$ denoting the paths and $X_{e';w_j}(i)$ the virtual link status as given in Definition~\ref{def-network_QDP_virtual_link_status}. $N_{e'}^{\max}$ is defined to be the maximum number of edge-disjoint paths (or Steiner trees) in the physical graph $G$.
	\end{definition}
	
	It is often of interest to determine the quantity
	\begin{equation}
		\liminf_{t\to\infty}\mathbb{E}[r_e(t)]_{\vec{\pi}}=\liminf_{t\to\infty}\frac{1}{t}\sum_{i=1}^t\mathbb{E}[N_e(i)]_{\vec{\pi}}
	\end{equation}
	(see, e.g., Ref.~\cite{DPW20}), where $\vec{\pi}$ is a collection of policies and $e$ refers to either an elementary link or virtual link. In Chapter~\ref{chap-mem_cutoff}, we provide a closed-form expression for this quantity in the case of elementary links under the memory-cutoff policy.

\subsubsection*{Key rates for QKD}

	We are also interested in secret key rates for quantum key distribution (QKD); see Section~\ref{sec-QKD} for a brief overview of QKD. In order to determine these rates, we need to keep track of the quantum state of the relevant elementary links as a function of time and the policy being used. The following discussion and formulas for the secret key rates are based on Ref.~\cite{GKF+15}.
	
	Suppose that $K$ is a function that gives the number of secret key bits per entangled state shared by the nodes of either an elementary link or virtual link. ($K$ is, for example, $K_{\text{BB84}}$, $K_{\text{6-state}}$, or $K_{\text{DI}}$, as defined in Section~\ref{sec-QKD}.) Then, suppose that $G=(V,E,c)$ is the graph corresponding to the physical (elementary) links of a quantum network.
	\begin{itemize}
		\item\textit{Elementary links.} Let $e^j$, with $e\in E$ and $1\leq j\leq c(e)$, be an arbitrary edge corresponding to an elementary link in the network. Then, at each time step, once the elementary link has been established (i.e., both transmission and heralding succeed), the nodes of the elementary link perform the required measurements for the QKD protocol being considered. Then, the secret key rate (in units of secret key bits per second) is
			\begin{equation}\label{eq-network_QKD_key_rate}
				\widetilde{K}_{e^j}= p_{e^j}\nu_{e^j}^{\text{rep}}K,
			\end{equation}
			where $p_{e^j}$ is given by \eqref{eq-elem_link_success_prob} and $\nu_{e^j}^{\text{rep}}$ is the repetition rate (in units of time steps per second). The quantity $K$ is calculated based on the quantum state $\rho_{e^j}^0$ (defined in \eqref{eq-initial_link_states}) that is shared by the nodes immediately upon successful transmission and heralding. The factor of $p_{e^j}$ is present because it is the average number of entangled states shared successfully by the nodes of the elementary link corresponding to $e^j$ per time step. The repetition rate $\nu_{e^j}^{\text{rep}}$ corresponds to the duration of each time step, and it can be thought of as being the refresh rate of the measurement devices; see, e.g., Ref.~\cite{PRML14}. In Chapter~\ref{chap-sats}, we provide an example calculation of the key rate in \eqref{eq-network_QKD_key_rate} in the context of elementary link generation with satellites.
			
		\item\textit{Virtual links.} Consider a collection $e'\coloneqq\{v_1,\dotsc,v_k\}\notin E$ of distinct nodes corresponding to a virtual link for some $k\geq 2$, and let $w$ be a path in the physical graph leading to the virtual link given by $e'$. Suppose that all of the elementary links in the path $w$, denoted by $E_w$, follow independent policies in $\vec{\pi}=\left(\pi^e:e\in E_w\right)$ up to some time $t\geq 1$. Then, a joining protocol along $w$ is performed to establish the virtual link. Conditioned on success of the joining protocol, the quantum state of the virtual link is given by \eqref{eq-network_QDP_virtual_link_fid_cond}, namely,
			\begin{equation}\label{eq-network_QDP_virtual_link_output_cond}
				\frac{1}{p_{\text{succ}}}\mathcal{L}_{w\to e'}^1\left(\bigotimes_{e\in E_w} \sigma_e^{\pi^e}(t|X_{e}(t)=1)\right),
			\end{equation}
			where
			\begin{equation}
				p_{\text{succ}}=\Tr\left[\mathcal{L}_{w\to e'}^1\left(\bigotimes_{e\in E_w}\sigma_e^{\pi^e}(t|X_e(t)=1)\right)\right]
			\end{equation}
			is the success probability of the joining protocol. Then, the secret key rate (in units of secret key bits per second) for the virtual link along the path $w$ is
			\begin{equation}
				\widetilde{K}_{e';w}(t)=p_{\text{succ}}\left(\prod_{e\in E_w}\Pr[X_{e}(t)=1]_{\pi^e}\right)\nu_{e'}^{\text{rep}}K.
			\end{equation}
			Here, $K$ is calculated using the state in \eqref{eq-network_QDP_virtual_link_output_cond}. The repetition rate $\nu_{e'}^{\text{rep}}$ in this case is a function of the time $t$ as well as the end-to-end classical communication time required for executing the joining protocol.
	\end{itemize}

\subsection{Cluster size}\label{sec-network_QDP_cluster_size}

	The last figure of merit that we consider is the cluster size \cite{DKD18,KMSD19}. In Section~\ref{sec-graph_theory}, we defined the quantity $S^{\max}(G)$, which is the size of the largest connected component, or cluster, of a graph $G$. As explained in Section~\ref{sec-random_subgraphs}, in our model of probabilistic elementary link generation, the physical graph $G=(V,E,c)$ of elementary links in a quantum network becomes a random variable $G_{\vec{p}}$. Then, as explained in Section~\ref{sec-network_QDP_link_statuses}, in the context of quantum decision processes this random variable becomes the time- and policy-dependent random variable $G_{\vec{\pi}}(t)$ defined in \eqref{eq-network_QDP_subgraph_RV_policy}. The figure of merit that we consider in this section is the expected largest cluster size of the subgraphs $G_{\vec{\pi}}(t)(\vec{x})$ of the physical graph $G$ corresponding to configurations $\vec{x}\in\Xi_G$.
	
	\begin{definition}[Expected largest cluster size of elementary links]\label{def-exp_largest_cluster_G}
		Let $G=(V,E,c)$ be the graph corresponding to the physical (elementary) links of a quantum network, and let $\vec{\pi}=\left(\pi^{e^j}:e\in E,\,1\leq j\leq c(e)\right)$ be a collection of policies for all of the elementary links in the network. Then, we define
		\begin{equation}\label{eq-exp_largest_cluster_G}
			s_G^{\max}(t;\vec{\pi})\coloneqq\frac{1}{|G|}\mathbb{E}[S^{\max}(G_{\vec{\pi}}(t))]
		\end{equation}
		to be the (normalized) expected value of the random variable $S^{\max}(G_{\vec{\pi}}(t))$ defined in \eqref{eq-random_subgraph_largest_cluster_RV}.~\defqed
	\end{definition}
	
	By definition of expectation, and using \eqref{eq-random_subgraph_largest_cluster_RV}, we have that
	\begin{align}
		s_G^{\max}(t;\vec{\pi})&=\frac{1}{|G|}\sum_{\vec{x}\in\Xi_G}S^{\max}(G_{\vec{\pi}}(t)(\vec{x}))\Pr[G_{\vec{\pi}}(t)(\vec{x})]\\
		&=\frac{1}{|G|}\sum_{\vec{x}\in\Xi_G}S^{\max}(G_{\vec{\pi}}(t)(\vec{x}))\prod_{e\in E}\prod_{j=1}^{c(e)}\Pr[X_{e^j}(t)=x^{e^j}]_{\pi^{e^j}},
	\end{align}
	where the last equality is due to \eqref{eq-network_QDP_subgraph_probability}.

	The expected largest cluster size is a figure of merit that is explicitly topology dependent, and it can be used to evaluate not only policies but also the network topology itself, as done in Ref.~\cite{DKD18}. The latter is due to the phenomenon of \textit{percolation} (see Ref.~\cite{Grim99_book} for an introduction) that occurs for graphs that have lattice topologies. For such graphs, there exist critical probabilities beyond which a large cluster is guaranteed to exist. Such critical probabilities thus quantify the robustness of the network to transmission losses. We illustrate this concept explicitly in Section~\ref{sec-mem_cutoff_cluster_size} in the context of the memory-cutoff policy.

\section{Policy optimization}\label{sec-network_QDP_pol_opt}

	In this section, we consider policy optimization for the quantum decision process for elementary link generation defined in Definition~\ref{def-network_QDP_elem_link}. Let us recall from that definition that a policy for an elementary link corresponding to an edge $e^j$ of a physical graph $G=(V,E,c)$, with $e\in E$ and $1\leq j\leq c(e)$, is a sequence of the form $\pi=(d_1^{e^j},d_2^{e^j},\dotsc)$, where the $d_t^{e^j}$ are decision functions defined as
	\begin{equation}
		d_t^{e^j}(h^t)(a_t)=\Pr[A_{e^j}(t)=a_t|H_{e^j}(t)=h^t]
	\end{equation}
	for all $t\in[1,\infty)$, $a_t\in\mathcal{A}$, and $h^t\in\Omega(t)$. In other words, the decision functions give the probability distributions over actions conditioned on histories.
	
	In the previous section, we laid out five figures of merit with respect to which one could optimize policies. Although it is possible in principle to optimize policies with respect to any one of those figures of merit, in this section we stick to the standard optimization problem in quantum decision processes (as laid out in Section~\ref{sec-QDP_pol_opt}), in which the objective function is the expected reward at the horizon time. From the discussion at the beginning of this chapter, this figure of merit corresponds precisely to the fidelity figure of merit for entanglement distribution protocols, as defined in Section~\ref{sec-figures_of_merit_general}. Specifically, let $G=(V,E,c)$ be the graph corresponding to the physical (elementary) links of a quantum network. Recall from \eqref{eq-network_initial_cq_state_2} that the overall quantum state of the network after $t\geq 1$ time steps is
	\begin{equation}\label{eq-network_initial_cq_state_3}
		\widehat{\sigma}_{G}^{\vec{\pi}}(t)=\bigotimes_{e\in E}\bigotimes_{j=1}^{c(e)}\widetilde{\sigma}_{e^j}^{\pi^{e^j}}(t),
	\end{equation}
	where the states $\widehat{\sigma}_{e^j}^{\pi^{e^j}}(t)$ have the form given by \eqref{eq-network_QDP_cq_state} and \eqref{eq-network_QDP_cond_state_unnormalized}, and $\vec{\pi}=\left(\pi^{e^j}:e\in E,\,1\leq j\leq c(e)\right)$ is a collection of policies for the elementary links. From the discussion at the beginning of this chapter, we would like the physical graph $G$ to have a subset $E'\subseteq E$ in which all of the corresponding elementary links are active, so that the classical-quantum state for this subset is
	\begin{equation}\label{eq-network_target_cq_state}
		\bigotimes_{e\in E'}\bigotimes_{j=1}^{c(e)} \ket{1}\bra{1}_{X^{e^j}}\otimes\psi_{e^j}^{\text{target}},
	\end{equation}
	where $\psi_{e^j}^{\text{target}}$ are pure target states. Then, from \eqref{eq-network_QDP_total_fid_pol}, we have that
	\begin{align}
		&F\left(\bigotimes_{e\in E'}\bigotimes_{j=1}^{c(e)}\widehat{\sigma}_{e^j}^{\pi_e^j}(t+1),\bigotimes_{e\in E'}\bigotimes_{j=1}^{c(e)}\ket{1}\bra{1}_{X_{t+1}^{e^j}}\otimes \psi_{e^j}^{\text{target}}\right)\nonumber\\
		&\qquad\qquad\qquad\qquad=\prod_{e\in E'}\prod_{j=1}^{c(e)}\Tr\left[\left(\ket{1}\bra{1}_{X_{t+1}^{e_j}}\otimes\psi_{e^j}^{\text{target}}\right)\widehat{\sigma}_{e^j}^{\pi^{e^j}}(t+1)\right]\\
		&\qquad\qquad\qquad\qquad=\prod_{e\in E'}\prod_{j=1}^{c(e)}\mathbb{E}[R_{e^j}(t)]_{\pi^{e^j}}\\
		&\qquad\qquad\qquad\qquad=\prod_{e\in E'}\prod_{j=1}^{c(e)}\mathbb{E}\left[\widetilde{F}_{e^j}^{\pi^{e^j}}(t+1;\psi_{e^j}^{\text{target}})\right],\label{eq-network_QDP_total_fid_pol_2}
	\end{align}
	where to obtain the second-last line we used Theorem~\ref{thm-network_QDP_exp_reward_no_distill} and to obtain the last line we used Proposition~\ref{prop-network_QDP_exp_reward_fid_tilde}. The task is therefore to optimize the quantity in \eqref{eq-network_QDP_total_fid_pol_2} with respect to the policies $\pi^{e^j}$ of the elementary links. We have assumed throughout this chapter that the agents act independently, which means that every policy $\pi^{e^j}$ is an independent variable. Our strategy for optimizing the objective function in \eqref{eq-network_QDP_total_fid_pol_2} is to optimize each term in the product individually, so in what follows we consider optimizing the policy for just one elementary link.
	
	\begin{remark}
		As explained in Remark~\ref{rem-QDP_network_extensions}, let us emphasize again that the policies being discussed here are not necessarily executed by the end-users themselves. The policies should be thought of as describing how the \textit{devices} located at the nodes of the end users should operate, given certain performance/resource requirements by the end users for their desired application. The \textit{training} of the devices proceeds via an optimization algorithm which could be backward or forward recursion (which we discuss below), or it could be a reinforcement learning algorithm (which we discuss in Appendix~\ref{sec-future_work}).~\defqed
	\end{remark}

\subsection{Backward recursion}\label{sec-network_QDP_backward_recursion}
	
	Recall that the agent in the definition of our quantum decision process in Definition~\ref{def-network_QDP_elem_link} is classical. Therefore, from the discussion in Section~\ref{sec-QDP_classical_v_quantum_agents}, we immediately obtain the following result.
	
	\begin{theorem}[Optimal finite-horizon policy for elementary link generation]\label{thm-opt_policy}
		Let $G=(V,E,c)$ be the graph corresponding to the physical (elementary) links of a quantum network, and let $e^j$, with $e\in E$ and $1\leq j\leq c(e)$, correspond to an arbitrary elementary link in the network. Then, for all $t\geq 1$, and for all pure states $\psi$,
		\begin{equation}
			\max_{\pi}\mathbb{E}\left[\widetilde{F}_{e^j}^{\pi}(t+1;\psi)\right]=\sum_{x_1=0}^1\max_{a_1\in\{0,1\}}w_{2}(x_1,a_1),
		\end{equation}
		where $\pi=(d_1,d_2,\dotsc,d_t)$ and
		\begin{align}
			w_{j}(h^{j-1},a_{j-1})&=\sum_{x_{j}=0}^1\max_{a_{j}\in\{0,1\}}w_{j+1}(h^{j-1},a_{j-1},x_j,a_j) \quad\forall~2\leq j\leq t,\\[0.2cm]
			w_{t+1}(h^{t},a_{t})&=\bra{\psi}\widetilde{\sigma}_{e^j}^{(\mathsf{E})}(t+1;h^t,a_t,1)\ket{\psi},\label{eq-network_QDP_pol_opt_final}
		\end{align}
		and 
		\begin{equation}
			\widetilde{\sigma}_{e^j}^{(\mathsf{E})}(t+1;h^t,a_t,1)\coloneqq\left(\mathcal{T}_{e^j}^{x_t,a_t,x_{t+1}}\circ\dotsb\circ\mathcal{T}_{e^j}^{x_2,a_2,x_3}\circ\mathcal{T}_{e^j}^{x_1,a_1,x_2}\circ\mathcal{T}_{e^j}^{0;x_1}\right)(\rho_{e^j}^S).
		\end{equation}
		Furthermore, the optimal policy is deterministic and given by
		\begin{equation}
			d_j^*(h^j)=\argmax_{a_j\in\{0,1\}}w_{j+1}(h^j,a_j)\quad \forall~1\leq j\leq t.
		\end{equation}
	\end{theorem}
	
	\begin{proof}
		The result follows immediately from the discussion in Section~\ref{sec-QDP_classical_v_quantum_agents}, because the agent is classical, specifically from \eqref{eq-QDP_opt_pol_classical_agent_1}--\eqref{eq-QDP_opt_pol_classical_agent_5}. Also, from \eqref{eq-network_QDP_reward_map_tilde}, recall that
		\begin{equation}
			\widetilde{\mathcal{R}}_{E_{t+1}^{e^j}}^{t;h^t,a_t,x_{t+1}}(\cdot)=\delta_{x_{t+1},1}\psi_{E_{t+1}^{e^j}}(\cdot)\psi_{E_{t+1}^{e^j}},
		\end{equation}
		which leads to the expression in \eqref{eq-network_QDP_pol_opt_final}.
	\end{proof}
	
	As mentioned at the beginning of Section~\ref{sec-QDP_forward_recursion}, the backward recursion algorithm is exponentially slow in the horizon time because of the fact that the number of histories grows exponentially in time. In this particular case, the number of histories up to time $t$ is $|\Omega(t)|=|\{0,1\}^{2t-1}|=2^{2t-1}$. Therefore, the functions $w_j$ in the statement of Theorem~\ref{thm-opt_policy}, which are used to determine the optimal actions, have an exponentially increasing number of values. Therefore, in practice, it is necessary to make use of other, more efficient algorithms, such as those based on forward recursion, or more generally reinforcement learning algorithms. We discuss these possibilities in Appendix~\ref{sec-future_work}. In Appendix~\ref{app-QDP_SDP}, we derive upper bounds on general quantum decision processes.

\subsection{Forward recursion}\label{sec-network_QDP_forward_recursion}

	Let us now consider the forward recursion policy, which we discussed in general terms in Section~\ref{sec-QDP_forward_recursion}. Intuitively, in this policy, the agent deterministically picks the action $a_t$ at time $t$ that maximizes the quantity $\mathbb{E}[\widetilde{F}^{\pi}(t+1)]$ at the next time step. We let $\pi^{\text{FR}}=(d_1^{\text{FR}},d_2^{\text{FR}},\dotsc)$ denote this policy.
	
	Let $G=(V,E,c)$ be the graph corresponding to the physical (elementary) links of a quantum network, and let $e^j$, with $e\in E$ and $1\leq j\leq c(e)$, correspond to an arbitrary elementary link in the network. Let $\pi=(d_1,d_2,\dotsc,d_{t-1})$ be an arbitrary policy of the elementary link up to time $t-1$. Under this policy, the classical-quantum state of the elementary link is (recall \eqref{eq-network_QDP_cq_state})
	\begin{equation}
		\widehat{\sigma}_{e^j}^{\pi}(t)=\sum_{h^t\in\Omega(t)}\ket{h^t}\bra{h^t}_{H_t^{e^j}}\otimes\widetilde{\sigma}_{E_t^{e^j}}^{\pi}(t;h^t).
	\end{equation}
	Now, for an arbitrary decision function $d_t$ corresponding to the decision at time $t$, we obtain the following classical-quantum state of the elementary link at time $t+1$:
	\begin{align}
		\widehat{\sigma}_{e^j}^{(\pi,d_t)}(t+1)&=\sum_{h^{t+1}\in\Omega(t+1)}\ket{h^{t+1}}\bra{h^{t+1}}_{H_{t+1}^{e^j}}\otimes\widetilde{\sigma}_{E_{t+1}^{e^j}}^{(\pi,d_t)}(t+1;h^{t+1})\\
		&=\sum_{\substack{h^t\in\Omega(t)\\a_t\in\mathcal{A}\\x_{t+1}\in\mathcal{X}}}\ket{h^t,a_t,x_{t+1}}\bra{h^t,a_t,x_{t+1}}_{H_{t+1}^{e^j}}\otimes d_t(h^t)(a_t)\mathcal{T}_{e^j}^{x_t,a_t,x_{t+1}}\left(\widetilde{\sigma}_{e_j}^{\pi}(t;h^t)\right).
	\end{align}
	Then,
	\begin{equation}
		\mathbb{E}\left[\widetilde{F}_{e^j}^{(\pi,d_t)}(t+1;\psi_{e^j})\right]=\Tr\left[\left(\ket{1}\bra{1}_{X_{t+1}^{e^j}}\otimes\psi_{e^j}\right)\widehat{\sigma}_{e^j}^{(\pi,d_t)}(t+1)\right].
	\end{equation}
	Using \eqref{eq-link_state_unnormalized} we have two possibilities for every history $h^t$ (assuming both actions are taken deterministically):
	\begin{align}
		a_t=0\Rightarrow \mathbb{E}\left[\widetilde{F}_{e^j}^{(\pi,d_t)}(t+1;\psi_{e^j})\right]&=\bra{\psi}\mathcal{T}_{e^j}^{X_{e^j}(t)(h^t),0,1}\left(\widetilde{\sigma}_{e^j}^{\pi}(t;h^t)\right)\ket{\psi}\nonumber\\[0.2cm]
		&=\Pr[H(t)=h^t]_{\pi}X_{e^j}(t)(h^t)\bra{\psi}\mathcal{N}_{e^j}\left(\rho_{e^j}\left(M_{e^j}^{\pi}(t)(h^t)\right)\right)\ket{\psi}\nonumber\\[0.2cm]
		&=\Pr[H(t)=h^t]_{\pi}X_{e^j}(t)(h^t)f_{M_{e^j}^{\pi}(t)(h^t)+1}^{e^j}(\rho_{e^j}^0),\\[0.3cm]
		a_t=1\Rightarrow \mathbb{E}\left[\widetilde{F}_{e^j}^{(\pi,d_t)}(t+1;\psi_{e^j})\right]&=\Pr[H(t)=h^t]_{\pi}p_{e^j} f_0^{e^j}(\rho_{e^j}^0)
	\end{align}
	So the task is to determine which of the two quantities, $X_{e^j}(t)(h^t)f_{M_{e^j}^{\pi}(t)(h^t)+1}^{e^j}(\rho_{e^j}^0)$ and $p_{e^j}f_0^{e^j}(\rho_{e^j}^0)$, is higher. If the elementary link is not active at time $t$, meaning that $X_{e^j}(t)(h^t)=0$, then requesting a link, i.e., selecting $a_t=1$, is gives a higher value than selecting $a_t=0$ (because the latter leads to a value of $\mathbb{E}[\widetilde{F}_{e^j}^{\pi}(t+1;\psi_{e^j})]=0$ for all $p_{e^j}>0$). On the other hand, if the elementary link is active at time $t$, then the task is to compare
	\begin{equation}\label{eq-forward_greedy_condition}
		f_{M_{e^j}^{\pi}(t)(h^t)+1}^{e^j}(\rho_{e^j}^0)\,(a_t=0)\quad\text{and}\quad p_{e^j}f_0^{e^j}(\rho_{e^j}^0)\,(a_t=1)
	\end{equation}
	for all histories $h^t\in\Omega(t)$. Which of these two quantities is higher (and thus which action is taken) depends on the success probability $p_{e^j}\in(0,1)$, the noise model of the quantum memory, and on the target pure state $\psi_{e^j}$. We conclude that the decision functions $d_t^{\text{FR}}$ for the forward recursion policy are given by
	\begin{equation}\label{eq-network_QDP_forward_recursion_policy}
		d_t^{\text{FR}}(h^t)=\left\{\begin{array}{l l} 1 & \text{if }x_t=0,\\[0.3cm] 0 & \text{if }x_t=1\text{ and }f_{M_{e^j}^{\pi}(t)(h^t)+1}^{e^j}(\rho_{e^j}^0)>p_{e^j} f_0^{e^j}(\rho_{e^j}^0),\\[0.3cm] 1 & \text{if } x_t=1\text{ and }f_{M_{e^j}^{\pi}(t)(h^t)+1}^{e^j}(\rho_{e^j}^0)\leq p_{e^j} f_0^{e^j}(\rho_{e^j}^0), \end{array}\right.
	\end{equation}
	for all $t\geq 1$, where $\pi=(d_1,d_2,\dotsc,d_{t-1})$ is an arbitrary policy providing the actions of the agent up to time $t-1$.

	\begin{remark}\label{rem-pol_opt_obj_func}
		We end this chapter with a remark about our choice of the objective function for policy optimization.
		
		We have justified our use of the fidelity random variable $\widetilde{F}$ as the objective function for policy optimization because it arises naturally in the context of entanglement distribution as defined in Section~\ref{sec-ent_dist_general}. However, it is possible to justify this choice in a slightly different way, which is worth pointing out. We do this by considering two alternatives to the function $\widetilde{F}$.
		
		The first alternative is to use the expectation value of the link status random variable $X$, as defined in Definition~\ref{def-network_QDP_elem_link}. If we use $\mathbb{E}[X(t)]$ as the objective function for policy optimization, then it is clear that the policy consisting of the action ``request'' at every time step before the link is established, and the action ``wait'' at every time step after the link is established, is optimal, in the sense that it achieves the highest value of $\mathbb{E}[X(t)]$ for all $t\geq 1$. In other words, the optimal policy in this case is simply the policy $\pi=(d_1,d_2,\dotsc)$ given by
		\begin{equation}
			d_t(h^t)=\left\{\begin{array}{l l} 1 & \text{if }x_t=0, \\ 0 & \text{if }x_t=1, \end{array}\right.
		\end{equation}
		for all $t\geq 1$ and all histories $h^t\in\Omega(t)$. A higher value of $\mathbb{E}[X(t)]$ comes, of course, at the cost of a lower fidelity, because each ``wait'' action decreases the fidelity of the quantum state stored in memory.
		
		The second alternative to $\widetilde{F}$ is the fidelity random variable $F$ defined alongside $\widetilde{F}$ in \eqref{eq-network_QDP_elem_link_fidelity}. Then, if we attempt to maximize $\mathbb{E}[F(t)]$ with respect to policies, then it is clear that the action ``request'' at every time step is optimal, because then the quantity $\mathbb{E}[F(t)]$ is equal to the initial fidelity $f_0$ at all time steps, which is the highest that can be obtained (without entanglement distillation). In other words, the optimal policy in this case is given by $\pi=(d_1,d_2,\dotsc)$, where $d_t(h^t)=1$ for all $t\geq 1$ and all histories $h^t\in\Omega(t)$. A higher value of the fidelity comes at the cost of a lower expected link value, because the probability that the link is active stays at $p$ for all times under this policy, i.e., $\Pr[X(t)=1]=p$ for all $t\geq 1$ if at every time step the agent requests a link.
		
		The quantity $\mathbb{E}[\widetilde{F}(t)]=\mathbb{E}[X(t)f_{M(t)}(\rho^0)]$ by definition incorporates the trade-off between the link value and the link fidelity, and can lead to non-trivial policies.~\defqed
	\end{remark}

\section{Summary}
	
	The developments of this chapter constitute one of the core contributions of this thesis: a protocol for entanglement distribution in a quantum network using quantum decision processes, summarized in Figure~\ref{fig-QDP_protocol}. The protocol is a relatively simple one in which a quantum decision process is used to model the evolution of the elementary links of the network for a certain horizon time $T$, before the required entanglement distillation and joining protocols are performed in order to transform the network of elementary links into a network of virtual links. We defined the quantum decision process for elementary links in Section~\ref{sec-indep_agents}, and then we provided an expression for the expected quantum state of an elementary link along with other important facts. In Section~\ref{sec-practical_figures_merit}, we defined figures of merit that are relevant for evaluating policies for elementary link generation. These figures of merit take into account the practical aspects of entanglement distribution as presented in Section~\ref{sec-network_architecture} of the previous chapter, such as probabilistic elementary link generation and quantum memories with finite coherence times. Then, in Section~\ref{sec-network_QDP_pol_opt}, we applied the general backward and forward recursion algorithms presented in Chapter~\ref{chap-QDP} to the case of elementary link generation. These algorithms provide us with an optimal sequence of actions that should be performed by the agents in a quantum network for the purpose of achieving entanglement distribution using the protocol in Figure~\ref{fig-QDP_protocol}.
	
	The most important aspect of the developments of this chapter is that, with quantum decision processes, we have a systematic method for understanding and developing quantum network protocols, and we are able to incorporate crucial practical aspects of entanglement distribution that are not typically taken into account in information-theoretic treatments of quantum network protocols, such as probabilistic elementary link generation and quantum memories with finite coherence times. Furthermore, such information-theoretic works are typically only concerned with optimal rates for entanglement distribution and not with explicit (practical) protocols to attain them. The work of this chapter, and of this thesis as a whole, is meant to provide a step toward achieving such optimal practical protocols.

%\addtocontents{toc}{\protect\newpage}
\chapter{MEMORY-CUTOFF POLICIES}\label{chap-mem_cutoff}

	In the previous chapter, we developed a quantum network protocol using a quantum decision process to model the generation of elementary links; see Figure~\ref{fig-net_agent_env_0} and Figure~\ref{fig-QDP_protocol} for a summary. The upshot of the development of the previous chapter is that we have a systematic method for determining an optimal sequence of actions that should be performed in a quantum network for the purpose of entanglement distribution. In order to make these developments more concrete, and also to connect to prior work, in this chapter we look at a specific example of a policy.
	
	A natural policy to consider, and one that has been considered extensively previously \cite{CJKK07,SRM+07,SSM+07,BPv11,JKR+16,DHR17,RYG+18,SSv19,KMSD19,SJM19,LCE20}, is the following. An elementary link is requested at every time step until it is established, and once it is established it is held in quantum memory for some pre-specified amount $t^{\star}$ of time (usually called the ``memory cutoff'' and not necessarily equal to the memory coherence time) before the quantum state of the elementary link is discarded and requested again. The cutoff $t^{\star}$ can be any value in the set $\mathbb{N}_0\cup\{\infty\}$, where $\mathbb{N}_0=\{0,1,2,\dotsc\}$. There are two extreme cases of this policy: when $t^{\star}=0$, a request is made at \textit{every} time step regardless of whether the previous request succeeded; if $t^{\star}=\infty$, then an elementary link request is made at every time step until the elementary link is established, and once it is established the corresponding entangled state remains in memory indefinitely---no further request is ever made. In this chapter we provide a complete analysis of this policy for all values $t^{\star}\in\mathbb{N}_0\cup\{\infty\}$. The following are our primary goals in this chapter.
	\begin{itemize}
		\item Determine the expected quantum state of an elementary link under the memory-cutoff policy, as given by the general expression in \eqref{eq-link_avg_q_state}. We do this in Section~\ref{sec-mem_cutoff_exp_q_state} in the finite-horizon ($t<\infty$) and infinite-horizon ($t\to\infty$) cases.
		\item Calculate the figures of merit defined in Section~\ref{sec-practical_figures_merit} for elementary links under the memory-cutoff policy. We do this in Section~\ref{sec-mem_cutoff_figures_of_merit}.
	\end{itemize}
	
	Throughout this chapter, we deal exclusively with elementary links. Therefore, for convenience and ease of notation, throughout this chapter we suppress the dependence of functions, random variables, probability distributions, and policies on the edge $e^j$ of the physical graph corresponding to an elementary link, unless it is required, with the understanding that the expressions hold for an arbitrary elementary link.

\section{Definition and basic properties}
	
	The memory-cutoff is a deterministic policy that is defined as follows.
	
	\begin{definition}[Memory-cutoff policy]
		Given a $t^{\star}\in\mathbb{N}_0\cup\{\infty\}$, the \textit{$t^{\star}$ memory-cutoff policy} is the sequence $\pi^{t^{\star}}=(d_1^{t^{\star}},d_2^{t^{\star}},\dots)$, where
		\begin{align}
			\text{if }t^{\star}\in\mathbb{N}_0\,\,:\,\,d_t^{t^{\star}}(h^t)&\coloneqq\left\{\begin{array}{l l} 0 & \text{if }M^{t^{\star}}(t)(h^t)<t^{\star}, \\ 1 & \text{if }M^{t^{\star}}(t)(h^t)=t^{\star}, \end{array}\right.\\[0.3cm]
			\text{if }t^{\star}=\infty\,\,:\,\,d_t^{\infty}(h^t)&\coloneqq\left\{\begin{array}{l l} 0 & \text{if } M^{\infty}(t)(h^t)\geq 0, \\ 1 & \text{if }M^{\infty}(t)(h^t)=-1, \end{array}\right.
		\end{align}
		where the memory time random variable $M^{t^{\star}}(t)$ is defined as
		\begin{equation}\label{eq-mem_time_cutoff_policy}
			M^{t^{\star}}(t)=\left\{\begin{array}{l l} \displaystyle \left(\sum_{j=1}^t X(j)-1\right)\text{mod}(t^{\star}+1) & \text{if }t^{\star}\in\mathbb{N}_0, \\[0.3cm] \displaystyle \sum_{j=1}^t X(j)-1 & \text{if }t^{\star}=\infty, \end{array}\right. 
		\end{equation}
		and $X(t)$, $t\geq 1$, are the elementary link status random variables defined in Definition~\ref{def-network_QDP_elem_link}.~\defqed
	\end{definition}
	
	\begin{remark}
		We start by making the following initial remarks about the definition of the memory-cutoff policy.
		\begin{itemize}
			\item Our definition of the memory time in \eqref{eq-mem_time_cutoff_policy} for $t^{\star}\in\mathbb{N}_0$ is different from the general definition provided in \eqref{eq-mem_time_def2}, and we make this modification for convenience. With this formula, the memory time is always in $\{0,1,\dotsc,t^{\star}\}$ when $t^{\star}\in\mathbb{N}_0$. Also note that, with this formula, we get a memory value of $-1\text{mod}(t^{\star}+1)=t^{\star}$ even when the memory is not loaded. The advantage of this is that, if  $M^{t^{\star}}(t)<t^{\star}$, then $X(t)=1$. For $t^{\star}=\infty$, the definition of $M^{\infty}(t)$ is consistent with \eqref{eq-mem_time_def2}. In particular, the possible values of $M^{\infty}(t)$ are $-1,0,1,\dotsc,t-1$.
			
			\item Note that for $t^{\star}=0$, we have $M^0(t)(h^t)=0$ for all $t\geq 1$, which means that the decision function is simply
				\begin{equation}\label{eq-mem_cutoff_zero}
					d_t^0(h^t)=1\quad\forall~t\geq 1,\,h^t\in\{0,1\}^{2t-1}.
				\end{equation}
			
				Also, observe that for $t^{\star}=\infty$, we can write the decision function $d_t^{\infty}$ in the following simpler form:
				\begin{equation}\label{eq-mem_cutoff_infty_dt_simpler}
					d_t^{\infty}(h^t)=\left\{\begin{array}{l l} 0 & \text{if }X(t)(h^t)=1, \\ 1 & \text{if } X(t)(h^t)=0. \end{array}\right.
				\end{equation}
				In other words, for the $t^{\star}=\infty$ memory-cutoff policy, it suffices to look at the status of the elementary link at the current time in order to determine the next action.
			
			\item We denote probability distributions of histories $H(t)=(X(1),A(1),\dotsc,A(t-1),X(t))$ under the $t^{\star}$ memory-cutoff policy by $\Pr[H(t)=h^t]_{t^{\star}}$, and similarly for marginal distributions.
			
			\item We can use \eqref{eq-link_trans_prob_4} to conclude that
				\begin{equation}
					\Pr[X(t+1)=x_{t+1}|H(t)=h^t,A(t)=0]_{t^{\star}}=\left\{\begin{array}{l l} 0 & \text{if } x_{t+1}=0,\\ 1 & \text{if } x_{t+1}=1.\end{array}\right.
				\end{equation}
				The transition probabilities given in \eqref{eq-link_trans_prob_1}--\eqref{eq-link_trans_prob_4} therefore reduce to the following for the $t^{\star}$ memory-cutoff policy for all $h^t\in\Omega(t)$ and all $t^{\star}\in\mathbb{N}_0\cup\{\infty\}$:
				\begin{align}
					\Pr[X(t+1)=0|H(t)=h^t,A(t)=1]_{t^{\star}}&=1-p,\label{eq-cutoff_trans_gen_prob1}\\
					\Pr[X(t+1)=1|H(t)=h^t,A(t)=1]_{t^{\star}}&=p,\label{eq-cutoff_trans_gen_prob2}\\
					\Pr[X(t+1)=0|H(t)=h^t,A(t)=0]_{t^{\star}}&=0,\label{eq-cutoff_trans_gen_prob3}\\
					\Pr[X(t+1)=1|H(t)=h^t,A(t)=0]_{t^{\star}}&=1.\label{eq-cutoff_trans_gen_prob4}
				\end{align}
				The following conditional probabilities then hold for all $t^{\star}\in\mathbb{N}_0$:
				\begin{align}
					&\Pr[X(t+1)=1,M^{t^{\star}}(t+1)=0|X(t)=0,M^{t^{\star}}(t)=m]_{t^{\star}}=p,\,\, 0\leq m\leq t^{\star},\label{eq-cutoff_trans_prob1}\\
					&\Pr[X(t+1)=1,M^{t^{\star}}(t+1)=0|X(t)=1,M^{t^{\star}}(t)=t^{\star}]_{t^{\star}}=p,\label{eq-cutoff_trans_prob2}\\
					&\Pr[X(t+1)=0,M^{t^{\star}}(t+1)=t^{\star}|X(t)=0,M^{t^{\star}}(t)=m]_{t^{\star}}=1-p,\,\, 0\leq m\leq t^{\star},\label{eq-cutoff_trans_prob3}\\
					&\Pr[X(t+1)=0,M^{t^{\star}}(t+1)=t^{\star}|X(t)=1,M^{t^{\star}}(t)=t^{\star}]_{t^{\star}}=1-p,\label{eq-cutoff_trans_prob4}\\
					&\Pr[X(t+1)=1,M^{t^{\star}}(t+1)=m+1|X(t)=1,M^{t^{\star}}(t)=m]_{t^{\star}}=1,\,\, 0\leq m\leq t^{\star}-1,\label{eq-cutoff_trans_prob5}
				\end{align}
				with all other possible conditional probabilities equal to zero. Since these transition probabilities are time-independent, and since the pair $(X(t+1),M^{t^{\star}}(t+1))$ depends only on $(X(t),M^{t^{\star}}(t))$, we have that $((X(t),M^{t^{\star}}(t)):t\geq 1)$ is a stationary/time-homogeneous Markov process. As such, the conditional probabilities can be organized into the transition matrix $T(t^{\star})$, $t^{\star}\in\mathbb{N}_0$, defined as follows:
				\begin{multline}\label{eq-cutoff_trans_prob6}
					\left(T(t^{\star})\right)_{\substack{x,m\\x',m'}}\coloneqq \Pr[X(t+1)=x,M^{t^{\star}}(t+1)=m|X(t)=x',M^{t^{\star}}(t)=m']_{t^{\star}},\\x,x'\in\{0,1\},~m,m'\in\{0,1,\dotsc,t^{\star}\}.
				\end{multline}
				
				For $t^{\star}=\infty$, observe that the action at time $t\geq 1$ depends only the current value of the elementary link, not on its entire history. In particular,
				\begin{equation}
					A(t)=1-X(t) \quad\text{(when }t^{\star}=\infty).
				\end{equation}
				Indeed, if $X(t)=0$, then by definition of the $t^{\star}=\infty$ memory-cutoff policy a request is made, so that $A(t)=1$, as required. If $X(t)=1$, then the elementary link is kept in memory, meaning that $A(t)=0$. The transition probabilities in \eqref{eq-cutoff_trans_gen_prob1}--\eqref{eq-cutoff_trans_gen_prob4} can therefore be simplified to the following when $t^{\star}=\infty$:
				\begin{align}
					\Pr[X(t+1)=0|X(t)=0]_{\infty}&=1-p,\\
					\Pr[X(t+1)=1|X(t)=0]_{\infty}&=p,\\
					\Pr[X(t+1)=0|X(t)=1]_{\infty}&=0,\\
					\Pr[X(t+1)=1|X(t)=1]_{\infty}&=1.
				\end{align}
				These transition probabilities are time-independent and Markovian, so they can be organized into the transition matrix $T(\infty)$ defined as follows:
				\begin{equation}\label{eq-cutoff_trans_infty}
					\left(T(\infty)\right)_{\substack{x\\x'}}\coloneqq \Pr[X(t+1)=x|X(t)=x']_{\infty},\quad x,x'\in\{0,1\}.
				\end{equation}
				
				In Section~\ref{sec-mem_cutoff_exp_q_state_infinite_horizon} (see Remark~\ref{rem-mem_cutoff_stationary_dist} therein), we determine the stationary (i.e., $t\to\infty$) distribution of the Markov processes defined by the transition matrix $T(t^{\star})$ defined in \eqref{eq-cutoff_trans_prob6}.~\defqed
		\end{itemize}
	\end{remark}

\section{Expected quantum state of an elementary link}\label{sec-mem_cutoff_exp_q_state}
	
	Recall from \eqref{eq-link_avg_q_state} that the expected quantum state of an elementary link of a quantum network undergoing a policy $\pi$ is given by 
	\begin{equation}\label{eq-link_avg_q_state_mem_cutoff}
			\sigma^{\pi}(t)=(1-\Pr[X(t)=1]_{\pi})\tau^{\varnothing}+\sum_{m}\Pr[X(t)=1,M^{\pi}(t)=m]_{\pi}\,\rho(m),
		\end{equation}
		where $\tau^{\varnothing}$ and $\rho(m)$ are defined in \eqref{eq-initial_link_states} and \eqref{eq-elem_link_state_in_mem}, respectively. We also have the conditional state in \eqref{eq-link_avg_q_state_conditional}:
		\begin{equation}\label{eq-link_avg_q_state_conditional_mem_cutoff}
			\sigma^{\pi}(t|X(t)=1)=\sum_m \Pr[M_e^{\pi}(t)=m|X(t)=1]_{\pi}\rho(m).
		\end{equation}
		We see that these states can be found by calculating the probability distributions given by $\Pr[X(t)=1,M^{\pi}(t)=m]_{\pi}$ and $\Pr[X(t)=1]_{\pi}$. Let us now do exactly this for the memory-cutoff policy. Specifically, we provide analytic expressions for the expected quantum state for all $t^{\star}\in\mathbb{N}_0\cup\{\infty\}$ in both the finite-horizon ($t<\infty$) setting and the infinite-horizon ($t\to\infty$) setting. We defer all of the proofs to Appendix~\ref{app-mem_cutoff_details}.

\subsection{Finite-horizon setting}\label{sec-mem_cutoff_exp_q_state_finite_horizon}

	In the finite-horizon setting, we obtain the following result for the joint probability distribution of the link value $X(t)$ and memory time $M^{t^{\star}}(t)$ random variables.
		
	\begin{theorem}\label{thm-mem_status_pr}
		For all $t\geq 1$, $t^{\star}\in\mathbb{N}_0\cup\{\infty\}$, and $p\in[0,1]$,
		\begin{equation}\label{eq-mem_status_pr_a}
			\Pr[M^{t^{\star}}(t)=m,X(t)=1]_{t^{\star}}=p(1-p)^{t-(m+1)},\quad t\leq t^{\star}+1,~0\leq m\leq t-1,
		\end{equation}
		and
		\begin{multline}\label{eq-mem_status_pr}
			\Pr[M^{t^{\star}}(t)=m,X(t)=1]_{t^{\star}}=\sum_{x=0}^{\floor{\frac{t-1}{t^{\star}+1}}} \binom{t-(m+1)-xt^{\star}}{x}\boldsymbol{1}_{t-(m+1)-x(t^{\star}+1)\geq 0}p^{x+1} (1-p)^{t-(m+1)-x(t^{\star}+1)}, \\ t>t^{\star}+1,~0\leq m\leq t^{\star}.
		\end{multline}
	\end{theorem}
	
	\begin{proof}
		See Appendix~\ref{sec-mem_status_pr_pf}.
	\end{proof}
	
	As an immediate corollary of Theorem~\ref{thm-mem_status_pr}, we obtain the probability distribution of the link value random variable $X(t)$.
	
	\begin{corollary}\label{cor-link_status_Pr1}
		For all $t\geq 1$, $t^{\star}\in\mathbb{N}_0\cup\{\infty\}$, and $p\in[0,1]$, the probability that an elementary link of a quantum network undergoing the $t^{\star}$ memory-cutoff policy is active at time $t$ is
		\begin{equation}\label{eq-link_status_Pr1}
			\Pr[X(t)=1]_{t^{\star}}=\left\{\begin{array}{l l} 1-(1-p)^t, & t\leq t^{\star}+1,\\[0.5cm] \displaystyle \sum_{x=0}^{\floor{\frac{t-1}{t^{\star}+1}}}\sum_{k=1}^{t^{\star}+1}\binom{t-k-xt^{\star}}{x}\boldsymbol{1}_{t-k-x(t^{\star}+1)\geq 0}p^{x+1}(1-p)^{t-k-(t^{\star}+1)x}, & t>t^{\star}+1. \end{array}\right.
		\end{equation}
	\end{corollary}
	
	\begin{proof}
		This follows immediately from the fact that $\Pr[X(t)=1]_{t^{\star}}=\sum_{m=0}^{t-1}\Pr[X(t)=1,M^{t^{\star}}(t)=m]_{t^{\star}}$ for $t\leq t^{\star}+1$ and that $\Pr[X(t)=1]_{t^{\star}}=\sum_{m=0}^{t^{\star}} \Pr[X(t)=1,M^{t^{\star}}(t)=m]$ for $t>t^{\star}+1$.
	\end{proof}
	
	\begin{figure}
		\centering
		\includegraphics[scale=1]{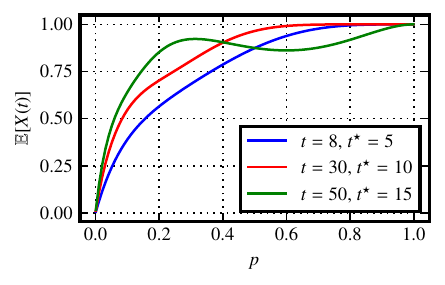}\quad
		\includegraphics[scale=1]{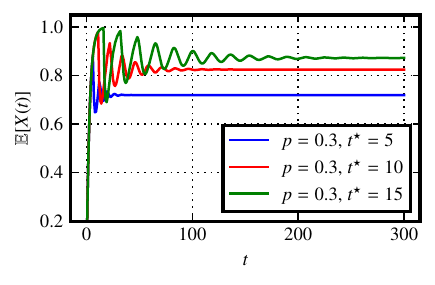}
		\caption{(Left) The expected elementary link value, given by \eqref{eq-link_status_Pr1}, as a function of the transmission-heralding probability $p$ for various values of $t$ and $t^{\star}$. (Right) The expected elementary link value, given by \eqref{eq-link_status_Pr1}, as a function of $t$ for fixed values of $p$ and $t^{\star}$.}\label{fig-avg_link_value_example}
	\end{figure}
	
	See Figure~\ref{fig-avg_link_value_example} for plots of $\mathbb{E}[X(t)]_{t^{\star}}=\Pr[X(t)=1]_{t^{\star}}$ as a function of the time steps $t$ and as a function of the transmission-heralding success probability $p$.
	
	Combining Theorem~\ref{thm-mem_status_pr} and Corollary~\ref{cor-link_status_Pr1} immediately leads to our desired expressions for the quantum states in \eqref{eq-link_avg_q_state_mem_cutoff} and \eqref{eq-link_avg_q_state_conditional_mem_cutoff}.
	
	\begin{corollary}
		For all $t\geq 1$, $t^{\star}\in\mathbb{N}_0\cup\{\infty\}$, and $p\in[0,1]$, the expected quantum state of an elementary link undergoing the $t^{\star}$ memory-cutoff policy is
		\begin{equation}\label{eq-avg_q_state_link_overall}
			\sigma^{t^{\star}}(t)=\left\{\begin{array}{l l} \displaystyle (1-p)^t\tau^{\varnothing}+\sum_{m=0}^{t-1}p(1-p)^{t-(m+1)}\rho(m), & t\leq t^{\star}+1,\\[0.5cm] \displaystyle (1-\Pr[X(t)=1]_{t^{\star}})\tau^{\varnothing}+\sum_{m=0}^{t^{\star}}\Pr[M^{t^{\star}}(t)=m,X(t)=1]_{t^{\star}}\rho(m) & t>t^{\star}+1, \end{array} \right.
		\end{equation}
		where for $t>t^{\star}+1$ the expressions for $\Pr[X(t)=1]_{t^{\star}}$ and $\Pr[M^{t^{\star}}(t)=m,X(t)=1]_{t^{\star}}$ are given in \eqref{eq-link_status_Pr1} and \eqref{eq-mem_status_pr}, respectively. From \eqref{eq-link_avg_q_state_conditional}, we also have
		\begin{equation}\label{eq-avg_q_state_link_active}
			\sigma^{t^{\star}}(t|X(t)=1)=\left\{\begin{array}{l l} \displaystyle \sum_{m=0}^{t-1}\frac{p(1-p)^{t-(m+1)}}{1-(1-p)^t}\rho(m), & t\leq t^{\star}+1,\\[0.5cm] \displaystyle \sum_{m=0}^{t^{\star}}\frac{\Pr[M^{t^{\star}}(t)=m,X(t)=1]_{t^{\star}}}{\Pr[X(t)=1]_{t^{\star}}}\rho(m), & t>t^{\star}+1, \end{array}\right.
		\end{equation}
		where again for $t>t^{\star}+1$ the expressions for $\Pr[X(t)=1]_{t^{\star}}$ and $\Pr[M^{t^{\star}}(t)=m,X(t)=1]_{t^{\star}}$ are given in \eqref{eq-link_status_Pr1} and \eqref{eq-mem_status_pr}, respectively.
	\end{corollary}

	Let us also determine the probabilities $\Pr[M^{t^{\star}}(t)=m,X(t)=0]_{t^{\star}}$.
	
	\begin{proposition}\label{prop-mem_time_prob_x0}
		For all $t\geq 1$, $t^{\star}\in\mathbb{N}_0$, $p\in[0,1]$, and $m\in\{0,1,\dotsc,t^{\star}\}$,
		\begin{equation}\label{eq-mem_time_prob_x0}
			\Pr[M^{t^{\star}}(t)=m,X(t)=0]_{t^{\star}}=\left\{\begin{array}{l l} \delta_{m,t^{\star}}(1-p)^t,\quad t\leq t^{\star}+1,\\[0.5cm] \displaystyle \delta_{m,t^{\star}}\sum_{x=0}^{\floor{\frac{t-1}{t^{\star}+1}}}\binom{t-1-xt^{\star}}{x}p^x(1-p)^{t-(t^{\star}+1)x},\quad t>t^{\star}+1. \end{array}\right.
		\end{equation}
		For $t^{\star}=\infty$,
		\begin{equation}
			\Pr[M^{\infty}(t)=m,X(t)=0]_{\infty}=\delta_{m,-1}(1-p)^t
		\end{equation}
		for all $m\in\{-1,0,1,\dotsc,t-1\}$.
	\end{proposition}
	
	\begin{proof}
		See Appendix~\ref{sec-mem_time_prob_x0_pf}.
	\end{proof}

\subsection{Infinite-horizon setting}\label{sec-mem_cutoff_exp_q_state_infinite_horizon}
	
	Let us now consider the $t\to\infty$, or infinite-horizon behavior of an elementary link under the memory-cutoff policy.
	
	\begin{theorem}\label{thm-link_status_avg_inf}
		For all $t^{\star}\in\mathbb{N}_0\cup\{\infty\}$ and $p\in[0,1]$, the expected elementary link status of an elementary link in a quantum network undergoing the $t^{\star}$ memory-cutoff policy is
		\begin{equation}\label{eq-link_status_avg_inf}
			\lim_{t\to\infty}\mathbb{E}[X(t)]_{t^{\star}}=\frac{(t^{\star}+1)p}{1+t^{\star}p}.
		\end{equation}
	\end{theorem}
	
	\begin{proof}
		See Appendix~\ref{app-link_status_avg_inf_pf}.
	\end{proof}
	
	Note that if $t^{\star}=\infty$, then
	\begin{equation}
		\lim_{t^{\star}\to\infty}\lim_{t\to\infty}\mathbb{E}[X(t)]_{t^{\star}}=\lim_{t^{\star}\to\infty}\frac{(t^{\star}+1)p}{1+t^{\star}p}=1,
	\end{equation}
	which is what we expect, because if $t^{\star}=\infty$, then the elementary link, once established, never has to be discarded.
	
	\begin{theorem}\label{thm-avg_mem_status_inf}
		For all $t^{\star}\in\mathbb{N}_0$, $p\in[0,1]$, and $m\in\{0,1,\dotsc,t^{\star}\}$,
		\begin{align}
			\lim_{t\to\infty}\Pr[M^{t^{\star}}(t)=m,X(t)=1]_{t^{\star}}&=\frac{p}{1+t^{\star}p},\label{eq-avg_mem_status_inf}\\
			\lim_{t\to\infty}\Pr[M^{t^{\star}}(t)=m,X(t)=0]_{t^{\star}}&=\frac{1-p}{1+t^{\star}p}\delta_{m,t^{\star}}.\label{eq-mem_time_prob_infty_x0}
		\end{align}
	\end{theorem}
	
	\begin{proof}
		See Appendix~\ref{app-avg_mem_status_inf_pf}.
	\end{proof}
	
	As an immediate consequence of Theorem~\ref{thm-link_status_avg_inf} and Theorem~\ref{thm-avg_mem_status_inf}, we obtain the following.
	
	\begin{corollary}
		For all $t^{\star}\in\mathbb{N}_0$ and $p\in[0,1]$, the expected quantum states of an elementary link undergoing the $t^{\star}$ memory-cutoff policy are, in the infinite-horizon setting,
		\begin{align}
			&\lim_{t\to\infty}\sigma^{t^{\star}}(t)=\frac{1-p}{1+t^{\star}p}\tau^{\varnothing}+\frac{p}{1+t^{\star}p}\sum_{m=0}^{t^{\star}}\rho(m),\label{eq-avg_q_state_overall_tInfty}\\
			&\lim_{t\to\infty}\sigma^{t^{\star}}(t|X(t)=1)=\frac{1}{t^{\star}+1}\sum_{m=0}^{t^{\star}}\rho(m).\label{eq-avg_cond_q_state_overall_tInfty}
		\end{align}
	\end{corollary}
	
	\begin{proof}
		The result is immediate from \eqref{eq-avg_q_state_link_overall} and \eqref{eq-avg_q_state_link_active} along with Theorems~\ref{thm-link_status_avg_inf} and \ref{thm-avg_mem_status_inf}. In particular, for \eqref{eq-avg_cond_q_state_overall_tInfty}, we make use of the fact that
		\begin{align}
			\lim_{t\to\infty}\Pr[M^{t^{\star}}(t)=m|X(t)=1]&=\lim_{t\to\infty}\frac{\Pr[M^{t^{\star}}(t)=m,X(1)=1]}{\Pr[X(t)=1]}\\
			&=\frac{\lim_{t\to\infty}\Pr[M^{t^{\star}}(t)=m,X(t)=1]}{\lim_{t\to\infty}\Pr[X(t)=1]}\\
			&=\frac{1}{t^{\star}+1},\label{eq-avg_mem_status_inf_cond}
		\end{align}
		which holds for all $t^{\star}\in\mathbb{N}_0$.
	\end{proof}
	
	\begin{remark}\label{rem-mem_cutoff_stationary_dist}
		The expressions in \eqref{eq-avg_mem_status_inf} and \eqref{eq-mem_time_prob_infty_x0} constitute a stationary distribution for the Markov process $((X(t),M^{t^{\star}}(t)):t\geq 1)$ with transition matrix $T(t^{\star})$, $t^{\star}\in\mathbb{N}_0$, defined in \eqref{eq-cutoff_trans_prob6}. Indeed, let $\vec{p}(t^{\star})$ be a column vector with elements given by
		\begin{equation}
			p_{x,m}(t^{\star})\coloneqq\lim_{t\to\infty}\Pr[M^{t^{\star}}(t)=m,X(t)=x]_{t^{\star}},
		\end{equation}
		for all $m\in\{0,1,\dotsc,t^{\star}\}$ and $x\in\{0,1\}$. Then, it is straightforward to show that 
		\begin{equation}
			T(t^{\star})\vec{p}(t^{\star})=\vec{p}(t^{\star}).
		\end{equation}
		In other words, $\vec{p}(t^{\star})$ is a stationary probability distribution.~\defqed
	\end{remark}

\section{Quantum state of the network}

	Let $G=(V,E,c)$ be the graph corresponding to the physical (elementary) links of a quantum network, and let $\vec{\pi}=\left(\pi^{e^j}:e\in E,\,1\leq j\leq c(e)\right)$ be a collection of policies for the elementary links of the network. From \eqref{eq-network_QDP_cq_state}, we have that the classical-quantum state of the elementary link corresponding to $e^j$ is
	\begin{equation}
		\widehat{\sigma}_{H_t^{e^j}E_t^{e^j}}^{\pi^{e^j}}(t)=\sum_{h^t\in\Omega(t)}\ket{h^t}\bra{h^t}_{H_t^{e^j}}\otimes\sigma_{E_t^{e^j}}^{\pi^{e^j}}(t;h^t)
	\end{equation}
	for all $t\geq 1$. The overall quantum state of the network is then
	\begin{equation}
		\bigotimes_{e\in E}\bigotimes_{j=1}^{c(e)}\widehat{\sigma}_{e^j}^{\pi^{e^j}}(t),
	\end{equation}
	where the $\pi^{e^j}$ are (independent) policies for the elementary links.
	
	Now, if the elementary link corresponding to the edge $e^j$ undergoes the $t_{e^j}^{\star}$ memory-cutoff policy for some $t_{e^j}^{\star}\in\mathbb{N}_0\cup\{\infty\}$, then we denote the quantum state of the network at time $t\geq 1$ by
	\begin{equation}
		\bigotimes_{e\in E}\bigotimes_{j=1}^{c(e)}\widehat{\sigma}_{e^j}^{t_{e^j}^{\star}}(t).
	\end{equation}
	By tracing out the history registers in the classical-quantum states $\widehat{\sigma}_{e^j}^{t_{e^j}^{\star}}(t)$, we obtain the expected quantum state of the network at time $t\geq 1$:
	\begin{equation}\label{eq-avg_q_state_overall_sublinks}
		\bigotimes_{e\in E}\bigotimes_{j=1}^{c(e)} \sigma_{e^j}^{t_{e^j}^{\star}}(t),
	\end{equation}
	where each $\sigma_{e^j}^{t_{e^j}^{\star}}(t)$ is given by the expression in \eqref{eq-avg_q_state_link_overall}. In the limit $t\to\infty$, and when $t_{e^j}^{\star}\in\mathbb{N}_0$ for all $e\in E$ and $1\leq j\leq c(e)$, we use the expression in \eqref{eq-avg_q_state_overall_tInfty} to obtain
	\begin{equation}\label{eq-avg_q_state_overall_sublinks_tInfty}
		\bigotimes_{e\in E}\bigotimes_{j=1}^{c(e)} \left(\frac{1-p_{e^j}}{1+t_{e^j}^{\star}p_{e^j}}\tau_{e^j}^{\varnothing}+\frac{p_{e^j}}{1+t_{e^j}^{\star}p_{e^j}}\sum_{m=0}^{t_{e^j}^{\star}} \rho_{e^j}(m)\right)\quad (t\to\infty),
	\end{equation}
	where $\{p_{e^j}:e\in E,\,1\leq j\leq c(e)\}$ is the set of success probabilities for the elementary links, as defined in \eqref{eq-elem_link_success_prob}. Similarly, the quantum state of the network at time $t\geq 1$, conditioned on all of the elementary links being active at time $t$, is given by
	\begin{equation}\label{eq-avg_cond_q_state_overall_sublinks}
		\bigotimes_{e\in E}\bigotimes_{j=1}^{c(e)}\sigma_{e^j}^{t_{e^j}^{\star}}(t|X_{e^j}(t)=1).
	\end{equation}
	When all of the cutoffs are finite, then from \eqref{eq-avg_cond_q_state_overall_tInfty} we obtain the following in the limit $t\to\infty$:
	\begin{equation}
		\bigotimes_{e\in E}\bigotimes_{j=1}^{c(e)}\left(\frac{1}{t_{e^j}^{\star}+1}\sum_{m=0}^{t_{e^j}^{\star}}\rho_{e^j}(m)\right).
	\end{equation}

\section{Figures of merit}\label{sec-mem_cutoff_figures_of_merit}
	
	We now evaluate the figures of merit defined in Section~\ref{sec-practical_figures_merit} for the memory-cutoff policy.

\subsection{Link statuses}

	We start by recalling the collective elementary link status random variable defined in \eqref{eq-network_QDP_collective_link_status}:
	\begin{equation}
		X_{E'}(t)=\prod_{e\in E'}\prod_{j=1}^{c(e)}X_{e^j}(t),
	\end{equation}
	where $E'\subseteq E$ and $G=(V,E,c)$ is a graph corresponding to the physical (elementary) links of a quantum network. Then, because all of the random variables are independent, for an arbitrary collection $\vec{\pi}=\left(\pi^{e^j}:e\in E,\,1\leq j\leq c(e)\right)$ of elementary link polices, we have that
	\begin{equation}
		\mathbb{E}[X_{E'}(t)]_{\vec{\pi}}=\prod_{e\in E'}\prod_{j=1}^{c(e)}\mathbb{E}[X_{e^j}(t)]_{\pi^{e^j}}.
	\end{equation}
	
	\begin{figure}
		\centering
		\includegraphics[width=0.48\textwidth]{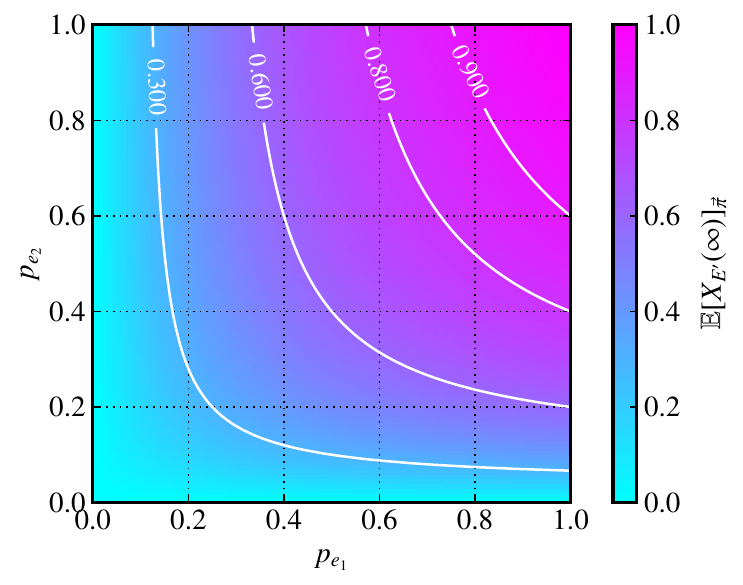}~
		\includegraphics[width=0.48\textwidth]{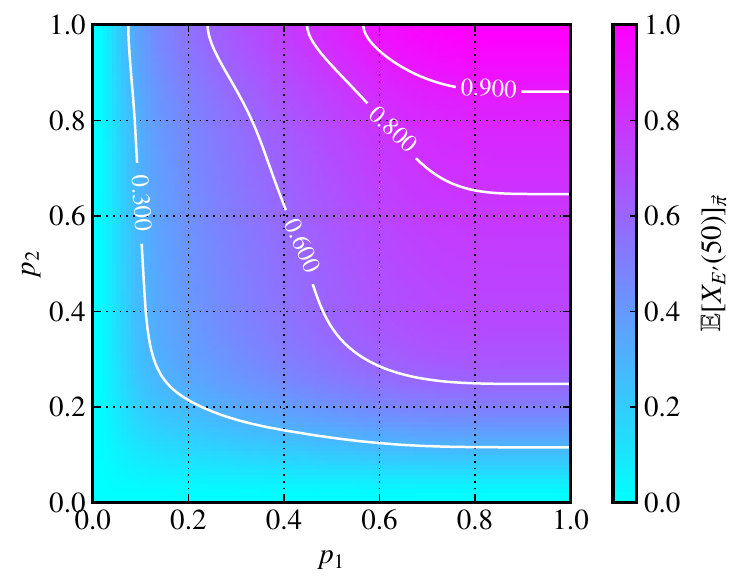}
		\caption{The expected collective elementary link status $\mathbb{E}[X_{E'}(t)]_{\vec{\pi}}$ of a collection of elementary links, all undergoing the memory-cutoff policy. (Left) When $E'=\{e_1,e_2\}$, with $c(e_1)=c(e_2)=1$, in the limit $t\to\infty$. One elementary link has success probability $p_{e^1}$ and cutoff $t_{e_1}^{\star}=5$, and the other elementary link has success probability $p_{e_2}$ and cutoff $t_{e_2}^{\star}=2$. We use the notation $\mathbb{E}[X_{E'}(\infty)]_{\vec{\pi}}\equiv\lim_{t\to\infty}\mathbb{E}[X_{E'}(t)]_{\vec{\pi}}$. (Right) When $E'=\{e_1,e_2,e_3,e_4\}$, with $c(e_1)=c(e_2)=c(e_3)=c(e_4)=1$, after $t=50$ time steps. Two of the elementary links have success probability $p_1$ with cutoffs $5,15$, and the other two links have success probability $p_2$ with cutoffs $10,20$.}\label{fig-avg_link_prod_example}
	\end{figure}
	
	Now, let us suppose that the elementary link given by the edge $e^j$ follows the $t_{e^j}^{\star}$ memory-cutoff policy, for all $e\in E'$ and $1\leq j\leq c(e)$. Then, by Corollary~\ref{cor-link_status_Pr1}, we obtain an analytic expression for expected collective elementary link status for all $t\geq 1$. In the limit $t\to\infty$, using Theorem~\ref{thm-link_status_avg_inf} we obtain the following simple expression:
	\begin{equation}
		\lim_{t\to\infty}\mathbb{E}[X_{E'}(t)]=\prod_{e\in E'}\prod_{j=1}^{c(e)} \frac{(t_{e^j}^{\star}+1)p_{e^j}}{1+t_{e^j}^{\star}p_{e^j}},
	\end{equation}
	where $\{p_{e^j}:e\in E,\,1\leq j\leq c(e)\}$ is the set of transmission-heralding success probabilities for the elementary links, as defined in \eqref{eq-elem_link_success_prob}. See Figure~\ref{fig-avg_link_prod_example} for plots of expected collective elementary link status.

	Let $G=(V,E,c)$ be the graph corresponding to the physical (elementary) links of a quantum network. Recall from \eqref{eq-network_QDP_num_active_parallel_links} that, given an $e\in E$, the number of active parallel elementary links at time $t$ corresponding to $e$ is
	\begin{equation}
		N_e(t)=\sum_{j=1}^{c(e)}X_{e^j}(t).
	\end{equation}
	If we assume that all of the parallel elementary links given by $e^j$ undergo the $t_{e^j}^{\star}$ memory-cutoff policy, then using Corollary~\ref{cor-link_status_Pr1} we obtain an analytic expression for $\mathbb{E}[N_e(t)]_{\vec{\pi}}=\sum_{j=1}^{c(e)}\mathbb{E}[X_{e^j}(t)]_{t_{e^j}^{\star}}$ for all $t\geq 1$. In the limit $t\to\infty$, we obtain the following simple expression:
	\begin{equation}\label{eq-avg_num_sublinks_infty}
		\lim_{t\to\infty}\mathbb{E}[N_e(t)]_{\vec{\pi}}=\sum_{j=1}^{c(e)}\frac{(t_{e^j}^{\star}+1)p_{e^j}}{1+t_{e^j}^{\star}p_{e^j}},
	\end{equation}
	where, as before, $\{p_{e^j}:e\in E,\,1\leq j\leq c(e)\}$ is the set of transmission-heralding success probabilities for the elementary links, as defined in \eqref{eq-elem_link_success_prob}. If $p_{e^j}=p_e$ and $t_{e^j}^{\star}=t^{\star}$ for all $1\leq j\leq c(e)$, then
	\begin{equation}
		\lim_{t\to\infty}\mathbb{E}[N_e(t)]_{\vec{\pi}}=c(e)\frac{(t_e^{\star}+1)p_e}{1+t_e^{\star}p_e}.
	\end{equation}
	As mentioned in Remark~\ref{rem-elem_link_QDP}, in the context of flow problems in graphs the quantity $c(e)$ can be thought of as the capacity of the element $e\in E$, and it represents the maximum number of entangled states that can be shared by the nodes of $e$ per unit time. Then, $N_e(t)$ can be thought of as the flow along $e$, and $\mathbb{E}[N_e(t)]$ is the expected flow. The task of finding the maximum number of edge-disjoint paths in a network (as outlined in Section~\ref{sec-graph_theory}) can be phrased as a flow problem, in which case the expected flows $\mathbb{E}[N_e(t)]_{\pi}$ for all $e\in E$ can be used to determine the rates at which virtual links can be created in the network; see Refs.~\cite{AK17,BA17,BAKE20,TKRW19,CERW20}.

\subsection{Fidelity}

	Let us now consider the elementary link fidelity random variables defined in Section~\ref{sec-network_QDP_fidelity}:
	\begin{equation}
		\widetilde{F}^{\pi}(t;\psi)=X(t)f_{M^{\pi}(t)}(\rho^0;\psi),\quad F^{\pi}(t;\psi)=\frac{\widetilde{F}^{\pi}(t;\psi)}{\Pr[X(t)=1]_{\pi}},
	\end{equation}
	where $\pi$ is an arbitrary policy, $\rho^0$ is the state of the elementary link after successful transmission and heralding, as defined in \eqref{eq-initial_link_states}, and $\psi$ is a target pure state. From \eqref{eq-avg_Ftilde}, we have that
	\begin{align}
		\mathbb{E}[\widetilde{F}^{\pi}(t;\psi)]&=\sum_m f_m(\rho^0;\psi)\Pr[M^{\pi}(t)=m,X(t)=1]_{\pi},\\
		\mathbb{E}[F^{\pi}(t;\psi)]&=\sum_m f_m(\rho^0;\psi)\Pr[M^{\pi}(t)=m|X(t)=1]_{\pi},
	\end{align}
	where $f_m(\rho^0;\psi)$ is given by \eqref{eq-fidelity_decay} and the sum is over all possible values of $M^{\pi}(t)$, which in general depends on the policy $\pi$.
	
	Now, if an elementary link undergoes the $t^{\star}$ memory-cutoff policy, then from Theorem~\ref{thm-mem_status_pr} and Corollary~\ref{cor-link_status_Pr1}, we immediately obtain the following analytic expressions for the expectation values of these quantities under the memory cutoff policy for all $t\geq 1$:
	\begin{align}
		\mathbb{E}[\widetilde{F}^{t^{\star}}(t;\psi)]&=\left\{\begin{array}{l l} \displaystyle \sum_{m=0}^{t-1}f_m(\rho^0;\psi)p(1-p)^{t-(m+1)} & t\leq t^{\star}+1,\\[0.5cm] \displaystyle \sum_{m=0}^{t^{\star}} f_m(\rho^0;\psi) \Pr[M^{t^{\star}}(t)=m,X(t)=1]_{t^{\star}} & t>t^{\star}+1, \end{array}\right.\label{eq-avg_fid_tilde}\\[0.5cm]
		\mathbb{E}[F^{t^{\star}}(t;\psi)]&=\left\{\begin{array}{l l}\displaystyle \sum_{m=0}^{t-1}f_m(\rho^0;\psi)\frac{p(1-p)^{t-(m+1)}}{1-(1-p)^t} & t\leq t^{\star}+1,\\[0.5cm] \displaystyle \sum_{m=0}^{t^{\star}} f_m(\rho^0;\psi)\frac{\Pr[M^{t^{\star}}(t)=m,X(t)=1]_{t^{\star}}}{\Pr[X(t)=1]_{t^{\star}}} & t>t^{\star}+1, \end{array}\right.\label{eq-avg_fid}
	\end{align}
	where in \eqref{eq-avg_fid_tilde} and \eqref{eq-avg_fid} the expression for $\Pr[M^{t^{\star}}(t)=m,X(t)=1]_{t^{\star}}$ for $t>t^{\star}+1$ is given in \eqref{eq-mem_status_pr}, and the expression for $\Pr[X(t)=1]_{t^{\star}}$ for $t>t^{\star}+1$ is given in \eqref{eq-link_status_Pr1}.
	
	In the limit $t\to\infty$, using Theorem~\ref{thm-link_status_avg_inf} and Theorem~\ref{thm-avg_mem_status_inf}, we obtain the following:
	\begin{align}
		\lim_{t\to\infty}\mathbb{E}[\widetilde{F}^{t^{\star}}(t;\psi)]&=\frac{p}{1+t^{\star}p}\sum_{m=0}^{t^{\star}}f_m(\rho^0;\psi),\quad t^{\star}\in\mathbb{N}_0,\label{eq-avg_fid_tilde_tInfty}\\
		\lim_{t\to\infty}\mathbb{E}[F^{t^{\star}}(t;\psi)]&=\frac{1}{t^{\star}+1}\sum_{m=0}^{t^{\star}}f_m(\rho^0;\psi),\quad t^{\star}\in\mathbb{N}_0.\label{eq-avg_fid_tInfty}
	\end{align}

\subsection{Waiting time}\label{sec-waiting_time}

	Let us now consider the waiting time for elementary links undergoing the memory-cutoff policy. We considered the waiting time for general policies in Section~\ref{sec-network_QDP_waiting_time}, and we defined the waiting time for a single elementary link and a collection of elementary links in Definition~\ref{def-network_QDP_elem_link_waiting_time} and Definition~\ref{def-collective_elem_link_waiting_time}, respectively.
	
	As described in Section~\ref{sec-network_QDP_waiting_time}, we consider the scenario in which the elementary link generation process is persistent, even if no end-user request is made. In other words, we consider an ``always-on''/continuous elementary link generation procedure that is ready to go whenever end-user entanglement is requested, rather than have the entire process begin only when end-user entanglement is requested \cite{CRDW19}. In this scenario, let us first consider a single elementary link, and let us the expression in \eqref{eq-waiting_time_prob_late_request} to find an analytic expression for the expected waiting time.
	
	\begin{theorem}\label{thm-avg_waiting_time_req}
		Let $G=(V,E,c)$ be the graph corresponding to the physical (elementary) links of a quantum network, let $e^j$, with $e\in E$ and $1\leq j\leq c(e)$, be an arbitrary edge of the graph, and let the elementary link corresponding to $e^j$ have transmission-heralding success probability $p\in[0,1]$. For all $t^{\star}\in\mathbb{N}_0$ and $t_{\text{req}}\geq 0$, the expected waiting time for the elementary link corresponding to $e^j$, when undergoing the $t^{\star}$ memory-cutoff policy, is
		\begin{equation}\label{eq-avg_waiting_time_req}
			\mathbb{E}[W_{e^j}(t_{\text{req}})]_{t^{\star}}=\frac{\Pr[M^{t^{\star}}(t_{\text{req}}+1)=t^{\star},X(t_{\text{req}}+1)=0]_{t^{\star}}}{p(1-p)}.
		\end{equation}
		For $t^{\star}=\infty$,
		\begin{equation}
			\mathbb{E}[W_{e^j}(t_{\text{req}})]_{\infty}=\frac{\Pr[X(t_{\text{req}}+1)=0]_{\infty}}{p(1-p)}=\frac{(1-p)^{t_{\text{req}}}}{p}.
		\end{equation}
	\end{theorem}
	
	\begin{remark}
		As a check, let us first observe the following:
		\begin{itemize}
			\item If $t_{\text{req}}=0$, then because $\Pr[M^{t^{\star}}(1)=t^{\star},X(1)=0]_{t^{\star}}=1-p$ for all $t^{\star}\in\mathbb{N}_0$ (see Proposition~\ref{prop-mem_time_prob_x0}), we obtain $\mathbb{E}[W(0)]_{t^{\star}}=\frac{1}{p}$, as expected. We get the same result for $t^{\star}=\infty$.
			
			\item If $t^{\star}=0$, then we get $\Pr[M^{t^{\star}}(t_{\text{req}}+1)=0,X(t_{\text{req}}+1)=0]_{0}=1-p$ for all $t_{\text{req}}\geq 0$ (see Proposition~\ref{prop-mem_time_prob_x0}), which means that $\mathbb{E}[W_{e^j}(t_{\text{req}})]_0=\frac{1}{p}$ for all $t_{\text{req}}\geq 0$. This makes sense, because in the $t^{\star}=0$ policy the elementary link is never held in memory. \defqed
		\end{itemize}
	\end{remark}
	
	\begin{proof}[Proof of Theorem~\ref{thm-avg_waiting_time_req}]
		Using \eqref{eq-waiting_time_prob_late_request}, we have
		\begin{align}
			&\Pr[W_{e^j}(t_{\text{req}})=t]_{t^{\star}}\nonumber\\
			&\,\,=\Pr[X(t_{\text{req}}+1)=0,\dotsc,X(t_{\text{req}}+t)=1]_{t^{\star}}\\
			&\,\,=\sum_{m_1,\dotsc,m_t=0}^{t^{\star}}\Pr[X(t_{\text{req}}+1)=0,M^{t^{\star}}(t_{\text{req}}+1)=m_1,\dotsc,X(t_{\text{req}}+t)=1,M^{t^{\star}}(t_{\text{req}}+t)=m_t]_{t^{\star}}.
		\end{align}
		Using the transition matrix $T(t^{\star})$ defined in \eqref{eq-cutoff_trans_prob1}--\eqref{eq-cutoff_trans_prob6}, we obtain
		\begin{multline}
			\Pr[W_{e^j}(t_{\text{req}})=t]_{t^{\star}}\\=\sum_{m_1,\dotsc,m_t=0}^{t^{\star}} (T(t^{\star}))_{\substack{1,m_t\\0,m_{t-1}}}\dotsb (T(t^{\star}))_{\substack{0,m_3\\0,m_2}}(T(t^{\star}))_{\substack{0,m_2\\0,m_1}}\Pr[M^{t^{\star}}(t_{\text{req}}+1)=m_1,X(t_{\text{req}}+1)=0]_{t^{\star}}.
		\end{multline}
		Using \eqref{eq-mem_time_prob_x0}, along with \eqref{eq-cutoff_trans_prob1}--\eqref{eq-cutoff_trans_prob6}, we have that
		\begin{equation}
			\Pr[W_{e^j}(t_{\text{req}})=t]_{t^{\star}}=\Pr[M^{t^{\star}}(t_{\text{req}}+1)=t^{\star},X(t_{\text{req}}+1)=0]_{t^{\star}}p(1-p)^{t-2},
		\end{equation}
		for all $t\geq 1$. The result then follows.
		
		For $t^{\star}=\infty$, using the transition matrix $T(\infty)$ defined in \eqref{eq-cutoff_trans_infty} leads to
		\begin{multline}
			\Pr[X(t_{\text{req}}+1)=0,\dotsc,X(t_{\text{req}}+t)=1]_{\infty}\\ = (T(\infty))_{\substack{1\\0}}(T(\infty))_{\substack{0\\0}}\dotsb(T(\infty))_{\substack{0\\0}}\Pr[X(t_{\text{req}}+1)=0]_{\infty}.
		\end{multline}
		Then, from \eqref{eq-link_status_Pr1}, we have that $\Pr[X(t_{\text{req}}+1)=0]=(1-p)^{t_{\text{req}}+1}$, so that
		\begin{equation}
			\Pr[W_{e^j}(t_{\text{req}})=t]_{\infty}=p(1-p)^{t-2}(1-p)^{t_{\text{req}}+1}
		\end{equation}
		for all $t\geq 1$. The result then follows.
	\end{proof}
	
	\begin{figure}
		\centering
		\includegraphics[scale=1]{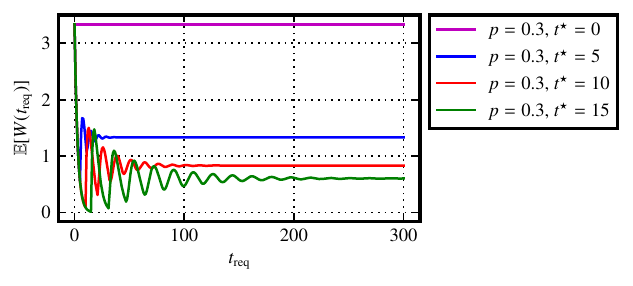}
		\caption{The expected waiting time for a single elementary link, given by \eqref{eq-avg_waiting_time_req}, as a function of the request time $t_{\text{req}}$. We let $p=0.3$, and we take various values for the cutoff $t^{\star}$.}\label{fig-avg_waiting_time_example}
	\end{figure}
	
	In the limit $t_{\text{req}}\to\infty$, we obtain using \eqref{eq-mem_time_prob_infty_x0},
	\begin{equation}\label{eq-wait_time_treqInfty}
		\lim_{t_{\text{req}}\to\infty}\mathbb{E}[W_{e^j}(t_{\text{req}})]_{t^{\star}}=\frac{1}{p(1+t^{\star}p)},\quad t^{\star}\in\mathbb{N}_0.
	\end{equation}
	See Figure~\ref{fig-avg_waiting_time_example} for plots of the expected waiting time, given by \eqref{eq-avg_waiting_time_req}, as a function of the request time $t_{\text{req}}$ for various values of $t^{\star}$. As long as $t^{\star}$ is strictly greater than zero, the waiting time is strictly less than $\frac{1}{p}$, despite the oscillatory behavior for small values of $t_{\text{req}}$. In the limit $t_{\text{req}}\to\infty$, we see that the waiting time is monotonically decreasing with increasing $t^{\star}$, which is also apparent from \eqref{eq-wait_time_treqInfty}.
	
	Let us now consider a collection of elementary links undergoing the memory-cutoff policy. Prior work \cite{Prax13,BPv11,SSv19} has established formulas for the expected waiting time in the case of a linear chain of elementary links, all of which have the same transmission-heralding success probability as well as the same cutoff. Here, let us consider an arbitrary collection of elementary links, and we let all of them undergo the $t^{\star}=\infty$ memory-cutoff policy. The proof of the expected waiting time in this scenario is provided in \cite[Appendix~A]{KMSD19}, but here we present a different proof method.
	
	\begin{theorem}\label{thm-mem_cutoff_collective_waiting_time_tstarInf}
		Let $G=(V,E,c)$ be the graph corresponding to the physical (elementary) links of a quantum network, and let $E'\subseteq E$ be arbitrary and such that $|E'|=M\geq 1$. If all of the $M$ elementary links undergo the $t^{\star}=\infty$ memory-cutoff policy, and if all of the elementary links have the same transmission-heralding success probability $p\in(0,1)$, then
		\begin{equation}\label{eq-exp_waiting_time_tInfty}
			\mathbb{E}[W_{E'}(t_{\text{req}})]_{\infty}=\sum_{k=1}^M \binom{M}{k}(-1)^{k+1}\left(1+\frac{(1-p_k)^{t_{\text{req}}+1}}{p_k}\right),\quad p_k\coloneqq 1-(1-p)^k.
		\end{equation}
	\end{theorem}
	
	\begin{proof}
		See Appendix~\ref{app-mem_cutoff_collective_waiting_time_tstarInf_pf}.
	\end{proof}

\subsection{Rates}

	Let us now consider the rate quantities defined in Section~\ref{sec-network_QDP_rates}, starting with the elementary link success rate in Definition~\ref{def-network_QDP_elem_link_succ_rate}.
	
	\begin{theorem}\label{thm-avg_sucess_rate}
		Let $G=(V,E,c)$ be the graph corresponding to the physical (elementary) links of a quantum network, let $e^j$, with $e\in E$ and $1\leq j\leq c(e)$, be an arbitrary edge of the graph, and let the elementary link corresponding to $e^j$ have transmission-heralding success probability $p\in[0,1]$. For all $t^{\star}\in\mathbb{N}_0\cup\{\infty\}$ and $t\geq 1$, the expected success rate for the elementary link corresponding to $e^j$, when undergoing the $t^{\star}$ memory-cutoff policy, is
		\begin{equation}
			\mathbb{E}[S_{e^j}(t)]_{t^{\star}}=\sum_{j=0}^{t-1}\frac{1}{j+1}p(1-p)^j,\quad t\leq t^{\star}+1.
		\end{equation}
		For $t>t^{\star}+1$,
		\begin{multline}
			\mathbb{E}[S_{e^j}(t)]_{t^{\star}}=\sum_{x=0}^{\floor{\frac{t-1}{t^{\star}+1}}}\left(\frac{x}{t-t^{\star}x}\binom{t-1-xt^{\star}}{x}p^x(1-p)^{t-(t^{\star}+1)x}\right.\\\left.+\sum_{k=1}^{t^{\star}+1}\frac{x+1}{t-k-t^{\star}x+1}\binom{t-k-xt^{\star}}{x}p^{x+1}(1-p)^{t-k-(t^{\star}+1)x}\boldsymbol{1}_{t-k-(t^{\star}+1)x\geq 0}\right).
		\end{multline}
	\end{theorem}
	
	\begin{proof}
		See Appendix~\ref{sec-avg_sucess_rate_pf}.
	\end{proof}
	
	Let us now consider the elementary link activity rate in Definition~\ref{def-network_QDP_elem_link_act_rate}.
	
	\begin{theorem}
		Let $G=(V,E,c)$ be the graph corresponding to the physical (elementary) links of a quantum network, and let $e\in E$ be arbitrary. Let the elementary links corresponding to the edges $e^j$, with $1\leq j\leq c(e)$, have transmission-heralding success probabilities $p_{e^j}\in[0,1]$, and suppose that every such elementary link undergoes the $t_{e^j}^{\star}$ memory-cutoff policy, with $t_{e^j}^{\star}\in\mathbb{N}_0\cup\{\infty\}$. Then, the expected rate $\mathbb{E}[r_e(t)]_{\vec{\pi}}$ of elementary link generation in the limit $t\to\infty$ is as follows:
		\begin{equation}\label{eq-link_act_rate_infty}
			\lim_{t\to\infty}\mathbb{E}[r_e(t)]_{\vec{\pi}}=\lim_{t\to\infty}\frac{1}{t}\sum_{i=1}^t\mathbb{E}[N_e(i)]_{\vec{\pi}}=\sum_{j=1}^{c(e)} \frac{(t_{e^j}^{\star}+1)p_{e^j}}{1+t_{e^j}^{\star}p_{e^j}},
		\end{equation}
		where $\vec{\pi}$ denotes the collection of memory-cutoff policies for the elementary links.
	\end{theorem}
	
	\begin{proof}
		The expected rate $\mathbb{E}[r_e(t)]_{\vec{\pi}}$ is, by defintion, the Ces\'{a}ro mean of the sequence $(\mathbb{E}[N_e(i)]_{\vec{\pi}})_{i=1}^{t}$. It is known that the limit of Ces\'{a}ro means of a sequence is equal to the limit of the sequence itself; see, e.g., \cite{Knopp_book}. Therefore, because $\lim_{i\to\infty}\mathbb{E}[N_e(i)]_{\vec{\pi}}$ exists and is given by \eqref{eq-avg_num_sublinks_infty}, we obtain the desired result.
	\end{proof}

\subsection{Cluster size}\label{sec-mem_cutoff_cluster_size}

	The last figure of merit that we look at is the cluster size, which we defined in Section~\ref{sec-network_QDP_cluster_size}. In particular, from Definition~\ref{def-exp_largest_cluster_G}, recall that
	\begin{equation}
		s_G^{\max}(t;\vec{\pi})=\frac{1}{|G|}\mathbb{E}[S^{\max}(G_{\vec{\pi}}(t))],
	\end{equation}
	where $G=(V,E,c)$ is the graph corresponding to the physical (elementary) links of a quantum network, and $\vec{\pi}=\left(\pi^{e^j}:e\in E,\,1\leq j\leq c(e)\right)$ is a collection of policies for the elementary links in the network.
	
	\begin{figure}
		\centering
		\includegraphics[scale=1]{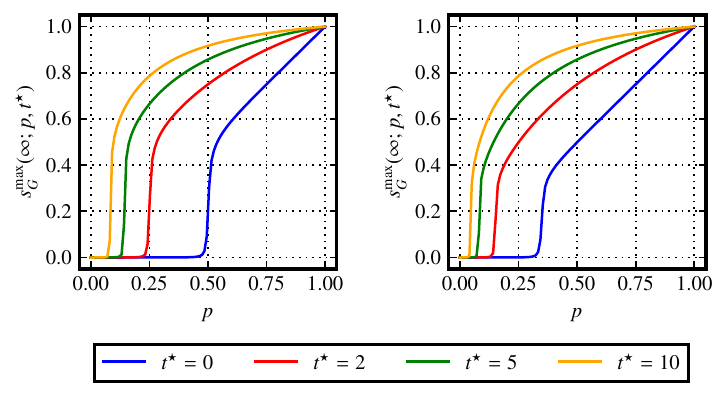}
		\caption{The (normalized) expected largest cluster size in the limit $t\to\infty$ for networks corresponding to a square lattice (left) and triangular lattice (right) with 500 edges. Every elementary link in the network has a transmission-heralding success probability of $p$ and undergoes the $t^{\star}$ memory-cutoff policy.}\label{fig-bond_perc_mem_cutoff}
	\end{figure}
	
	Now, in order to illustrate the usefulness of this figure of merit, let us consider the simple scenario considered in Ref.~\cite{KMSD19}, in which there is only a single parallel elementary link corresponding to every $e\in E$, which means that $|G|=|E|$. Also, let us suppose that all of the elementary links have the same transmission-heralding success probability $p\in[0,1]$, and that all of the elementary links undergo the same $t^{\star}$ memory-cutoff policy, with $t^{\star}\in\mathbb{N}_0\cup\{\infty\}$. We let $s_G^{\max}(t;p,t^{\star})\equiv s_G^{\max}(t;\vec{\pi})$ denote the (normalized) expected largest cluster size in this scenario. In Figure~\ref{fig-bond_perc_mem_cutoff}, we plot this quantity in the limit $t\to\infty$ when the graph $G$ is the square lattice and the triangular lattice of size $|G|=500$. In both cases, for every value of $t^{\star}$ there is a value of the transmission-heralding success probability $p$, call it $p_{\text{crit}}(G;t^{\star})$, below which the expected largest cluster size is effectively zero, and beyond which the expected largest cluster size increases monotonically with $p$. The critical probability $p_{\text{crit}}(G;t^{\star})$ is a clear dividing point between these two regimes. Furthermore, we see that the critical probability decreases with increasing $t^{\star}$.
	
	The phenomenon observed in Figure~\ref{fig-bond_perc_mem_cutoff} is called \textit{bond percolation}, and it has been deeply studied; see, e.g., Ref.~\cite{Grim99_book}. In bond percolation theory, one considers a graph $G$ in which every edge of the graph is present with \textit{bond probability} $p_{\text{bond}}\in[0,1]$ and absent with probability $1-p_{\text{bond}}$. It is then well known that for graphs $G$ that are regular lattices (such as the square and triangular lattices), there exists a probability $p_{\text{crit}}(G)$, known as the \textit{critical bond probability}, such that as the size $|G|$ of the graph increases the probability of having a large connected component vanishes for all $p_{\text{bond}}<p_{\text{crit}}$ and approaches one for all $p_{\text{bond}}\geq p_{\text{crit}}$. For the square lattice, $p_{\text{crit}}(G)=\frac{1}{2}$ \cite{Kesten80}; for the triangular lattice, $p_{\text{crit}}(G)=2\sin(\frac{\pi}{18})\approx 0.34730$ \cite{SE64,Wierman81}.
	
	The scenario that we consider in Figure~\ref{fig-bond_perc_mem_cutoff} corresponds to the standard bond percolation scenario described in the previous paragraph when $t^{\star}=0$, because in this case $p_{\text{bond}}=\Pr[X(t)=1]_0=p$ (see Corollary~\ref{cor-link_status_Pr1}). For $t^{\star}>0$, the scenario that we consider is equivalent to the bond percolation scenario with $p_{\text{bond}}=\Pr[X(t)=1]_{t^{\star}}$. In particular, then, in the limit $t\to\infty$, by Theorem~\ref{thm-link_status_avg_inf} we have that $p_{\text{bond}}=\frac{(t^{\star}+1)p}{1+t^{\star}p}$. From this, we immediately obtain an expression for the critical probability $p_{\text{crit}}(G;t^{\star})$:
	\begin{equation}\label{eq-mem_cutoff_critical_prob_tStar_tInf}
		p_{\text{crit}}(G;t^{\star})=\frac{p_{\text{crit}}(G)}{1+t^{\star}(1-p_{\text{crit}}(G))},
	\end{equation}
	for all $t^{\star}\in\mathbb{N}_0\cup\{\infty\}$.
	
	As pointed out in Ref.~\cite{DKD18}, the critical probability $p_{\text{crit}}(G;t^{\star})$ can be interpreted as measure of the robustness of a quantum network to transmission losses. In particular, it gives us the value of the transmission-heralding success probability that needs to be achieved in order to have a large cluster of active elementary links as the size of the network increases. Notably, because the critical probability is different for the square and triangular lattices, the critical probability can be used as a method for evaluating the topology of the network---networks corresponding to lower critical probabilities are more robust than networks with higher critical probabilities, because for the former a large cluster can be achieved with a lower transmission-heralding success probability. We also see in Figure~\ref{fig-bond_perc_mem_cutoff} that the critical probability decreases with the cutoff $t^{\star}$, and it is made clear by the expression in \eqref{eq-mem_cutoff_critical_prob_tStar_tInf}. Based on our interpretation of the critical probability, the advantage of using a higher cutoff becomes clear, because higher cutoffs decrease the critical probability.	
	
	One can also consider \textit{inhomogeneous bond percolation}, which is when the edges of the graph have different bond probabilities. In this scenario, instead of a single critical probability we obtain a so-called critical surface that separates a highly connected network from a disconnected network. Examples of this in the context of quantum networks have been considered in Ref.~\cite{DKD18}.
	
	In Lemma~\ref{lem-largest_cluster_UB}, we showed that the largest cluster size of a network is bounded from above by the total number of active elementary links in the network. Let $G=(V,E,c)$ be the graph corresponding to the physical (elementary) links of a quantum network. From \eqref{eq-num_active_elem_links}, we have that
	\begin{equation}
		L(G_{\vec{\pi}}(t))=\sum_{e\in E}\sum_{j=1}^{c(e)} X_{e^j}(t)
	\end{equation}
	is the random variable for the number of active elementary links in the network at time $t\geq 1$, where $\vec{\pi}=\left(\pi^{e^j}:e\in E,\,1\leq j\leq c(e)\right)$ is the collection of policies for the elementary links. Then, the expected number of active elementary links is simply
	\begin{equation}
		\mathbb{E}[L(G_{\vec{\pi}}(t))]=\sum_{e\in E}\sum_{j=1}^{c(e)}\mathbb{E}[X_{e^j}(t)]_{\pi^{e^j}}.
	\end{equation}
	Now, let the elementary links corresponding to the edges $e^j$, with $1\leq j\leq c(e)$, have transmission-heralding success probabilities $p_{e^j}\in[0,1]$, and suppose that every such elementary link undergoes the $t_{e^j}^{\star}$ memory-cutoff policy, with $t_{e^j}^{\star}\in\mathbb{N}_0\cup\{\infty\}$. Then, using Corollary~\ref{cor-link_status_Pr1}, we obtain an analytic expression for the expected number of active elementary links in the network for all $t\geq 1$. In the case $t\to\infty$, we can use Theorem~\ref{thm-link_status_avg_inf} to obtain
	\begin{equation}
		\lim_{t\to\infty}\mathbb{E}[L(G_{\vec{\pi}}(t))]=\sum_{e\in E}\sum_{j=1}^{c(e)}\frac{(t_{e^j}^{\star}+1)p_{e^j}}{1+t_{e^j}^{\star}p_{e^j}}.
	\end{equation}

%	\begin{figure}
%		\centering
%		\includegraphics[width=0.48\textwidth]{Figures/avg_num_links_uneven.pdf}~
%		\includegraphics[width=0.48\textwidth]{Figures/avg_num_links_uneven2.pdf}
%		\caption{The expected total number $\mathbb{E}[L_M(t)]$ of active elementary links when $M$ edges in total are being considered. We use the notation $\mathbb{E}[L_M(\infty)]\equiv\lim_{t\to\infty}\mathbb{E}[L_M(t)]$. (Left) $M=2$ edges in the limit $t\to\infty$ with $N^{\max}=1$ parallel link for each edge. One link has success probability $p_1$ and cutoff $t_1^{\star}=5$, and the other link has success probability $p_2$ and cutoff $t_2^{\star}=2$. (Right) $M=4$ edges after $t=50$ time steps, with $N^{\max}=1$ parallel link for each edge. Two of the links have success probability $p_1$ with cutoffs $5,15$, and the other two links have success probability $p_2$ with cutoffs $10,20$.}\label{eq-fig_avg_num_links_example}
%	\end{figure}

\section{Summary}
	
	In this chapter, we considered an explicit example of a policy for the quantum decision process for elementary link generation presented in Chapter~\ref{chap-network_QDP}, which we called the memory-cutoff policy. This is an intuitive policy that has been considered extensively in prior work, and in this chapter we have explicitly cast it within the framework of quantum decision processes. The policy is defined as follows: if the elementary link is not currently active (meaning that the corresponding nodes do not share an entangled state), then the agent requests an entangled state from its associated source station. If the elementary is active (meaning that the nodes share an entangled state and it is stored in their respective quantum memories), then they keep the entangled state stored in the quantum memory if it has been stored for less than $t^{\star}$ time steps, and they discard the entangled state and request a new one if it has been stored for exactly $t^{\star}$ time steps. Here, $t^{\star}\in\mathbb{N}_0\cup\{\infty\}$ is the memory cutoff. For this policy, we provided analytic expressions for the expected quantum state of an elementary link in both the finite-horizon case ($t<\infty$) as well as the infinite-horizon case ($t\to\infty$). We also provided analytic expressions for the figures of merit defined in Section~\ref{sec-practical_figures_merit}. The memory-cutoff policy is useful not only because of its simple mathematical form, but also because it is practical, especially for near-term quantum memories that do not have very high coherence times.

\chapter{ELEMENTARY LINK GENERATION WITH SATELLITES}\label{chap-sats}

	In this chapter, we put together all of the pieces developed in Chapters~\ref{chap-QDP}--\ref{chap-mem_cutoff} and present an example of elementary link generation using the satellite-based network architecture from Ref.~\cite{KBD+19} that we presented in Section~\ref{sec-sat_architecture}.
	
	Although many ground-based quantum networking schemes have been developed, experimental demonstrations performed so far have been limited \cite{Hump+18,Kalb+17,Ritter+12} and do not scale to the distances needed to realize a global-scale quantum internet. On the other hand, satellites have been recognized as one of the best methods for achieving global-scale quantum communication with current and near-term resources \cite{AJP+03,JH13,BAL17,Sim17,Cubesat2017,Nanobob2018}. As we show in Section~\ref{sec-elem_link_sats}, using satellites is advantageous due to the fact that the majority of the optical path traversed by an entangled photon pair is in free space, resulting in lower loss compared to ground-based entanglement distribution over atmospheric or fiber-optic links. Satellites can also be used to implement long-distance QKD with untrusted nodes, which is missing from most current (ground-based) implementations of long-distance QKD due to the lack of a quantum repeater. A satellite-based approach also allows for the possibility to use quantum strategies for tasks such as establishing a robust and secure international time scale via a quantum network of clocks \cite{KKBJ+14}, extending the baseline of telescopes for improved astronomical imaging \cite{GJC12,1KBDL18,2KBDL18}, and exploring fundamental physics \cite{RJ+12,Bruschi+14}.

\section{Quantum state of an elementary link}\label{sec-elem_link_sats}

	In Section~\ref{sec-sat_architecture}, we determined that the (approximate) channel for satellite-to-ground transmission of a dual-rail photon is given by \eqref{eq-noisy_transmission_channel}:
	\begin{equation}
		\mathcal{L}_{A_1A_2}^{\eta_{\text{sg}},\overline{n}}(\rho_{A_1A_2})\coloneqq\Tr_{E_1E_2}\left[\left(U_{A_1E_1}^{\eta_{\text{sg}}}\otimes U^{\eta_{\text{sg}}}_{A_2E_2}\right)\left(\rho_{A_1A_2}\otimes\widetilde{\Theta}_{E_1E_2}^{\overline{n}}\right)\left(U^{\eta_{\text{sg}}}_{A_1E_1}\otimes U^{\eta_{\text{sg}}}_{A_2E_2}\right)^\dagger\right],
	\end{equation}
	where $\eta_{\text{sg}},\overline{n}\in[0,1]$, with $\overline{n}$ being the average number of background photons and $\eta_{\text{sg}}$ being the transmittance of the satellite-to-ground medium. In particular, if the satellite is at altitude $h$ and the path length from the satellite to the ground station is $L$, then
	\begin{equation}
		\eta_{\text{sg}}(L,h)=\eta_{\text{fs}}(L)\eta_{\text{atm}}(L,h),
	\end{equation}
	where
	\begin{equation}\label{eq-fs_transmittance_2}
		\eta_{\text{fs}}(L)=1-\exp\left(-\frac{2r^2}{w(L)^2}\right),\quad w(L)\coloneqq w_{0}\sqrt{1+\left(\frac{L}{L_{R}}\right)^2},\quad L_{R}\coloneqq\pi w_{0}^2\lambda^{-1},
	\end{equation}
	and
	\begin{equation}\label{eq-atmospheric_transmittance_zenith_2}
		\eta_{\text{atm}}(L,h)=\left\{\begin{array}{l l} (\eta_{\text{atm}}^{\text{zen}})^{\sec\zeta}, & \text{if } -\frac{\pi}{2}<\zeta<\frac{\pi}{2},\\[0.2cm] 0, & \text{if } |\zeta|\geq\frac{\pi}{2}, \end{array}\right.
	\end{equation}
	with $\eta_{\text{atm}}^{\text{zen}}$ the transmittance at zenith ($\zeta=0$). In general, the zenith angle $\zeta$ is given by
	\begin{equation}
		\cos\zeta=\frac{h}{L}-\frac{1}{2}\frac{L^2-h^2}{R_{\oplus}L}
	\end{equation}
	for a circular orbit of altitude $h$, with $R_{\oplus}\approx 6378$~km being the earth's radius. The following parameters thus characterize the total transmittance from satellite to ground: the initial beam waist $w_0$, the receiving aperture radius $r$, the wavelength $\lambda$ of the satellite-to-ground signals, and the atmospheric transmittance $\eta_{\text{atm}}^{\text{zen}}$ at zenith. See Table~\ref{table-parameters} for the values that we use for these parameters in this chapter.
	
	\begin{table}
		\centering
		\caption{Parameters used in the modeling of loss from satellites to ground stations; see \eqref{eq-fs_transmittance_2} and \eqref{eq-atmospheric_transmittance_zenith_2}.}\label{table-parameters}
		\begin{tabular}{|>{\centering\arraybackslash}m{1.3cm}  >{\centering\arraybackslash}m{3.5cm}  >{\centering\arraybackslash}m{1.4cm} |}
			\hline
			Parameter & Definition & Value \\\hline\hline
			$r$ & Receiving aperture radius & 0.75 m \\[0.2cm] 
			$w_0$ & Initial beam waist & 2.5 cm \\[0.2cm]
			$\lambda$ & Wavelength of satellite-to-ground signals & 810 nm \\[0.6cm]
			$\eta_{\text{atm}}^{\text{zen}}$ & Atmospheric transmittance at zenith & 0.5 at 810 nm \cite{BMH+13} \\ \hline
		\end{tabular}
	\end{table}
	
	Using the values in Table~\ref{table-parameters}, we plot in the right panel of Figure~\ref{fig-atmosphere_geometry} the total transmittance as a function of the ground distance $d$ between two ground stations with a satellite at the midpoint, as depicted in the left panel of Figure~\ref{fig-atmosphere_geometry} (top). We observe that for larger ground separations the total transmittance $\eta_{\text{sg}}^2$ is actually larger for a higher altitude than for a lower altitude; for example, beyond approximately $d=1600$~km the transmittance for $h=1000$~km is larger than for $h=500$~km. We also observe that there are altitudes at which the transmittance is maximal. Intuitively, beyond the maximum point, the atmospheric contribution to the loss is less dominant, while below the maximum (i.e., for lower altitudes) the atmosphere is the dominant source of loss. This feature is unique for optical transmission from satellite to ground. Furthermore, we can compare the transmittances in the left-hand plot of Figure~\ref{fig-atmosphere_geometry} with the transmittance $\e^{-d/L_0}$ for ground-based transmission with fiber-optics, with $L_0\approx 22~\text{km}$ \cite{KGMS88_book}. For example, with satellites, we have a transmittance of approximately $10^{-7}$ for $d=1500~\text{km}$ and $h=4000~\text{km}$, while for the same distance we have $\e^{-1500/22}\approx 10^{-30}$ using fiber-optic transmission.
	
	\begin{figure}
		\centering
		\includegraphics[width=\textwidth]{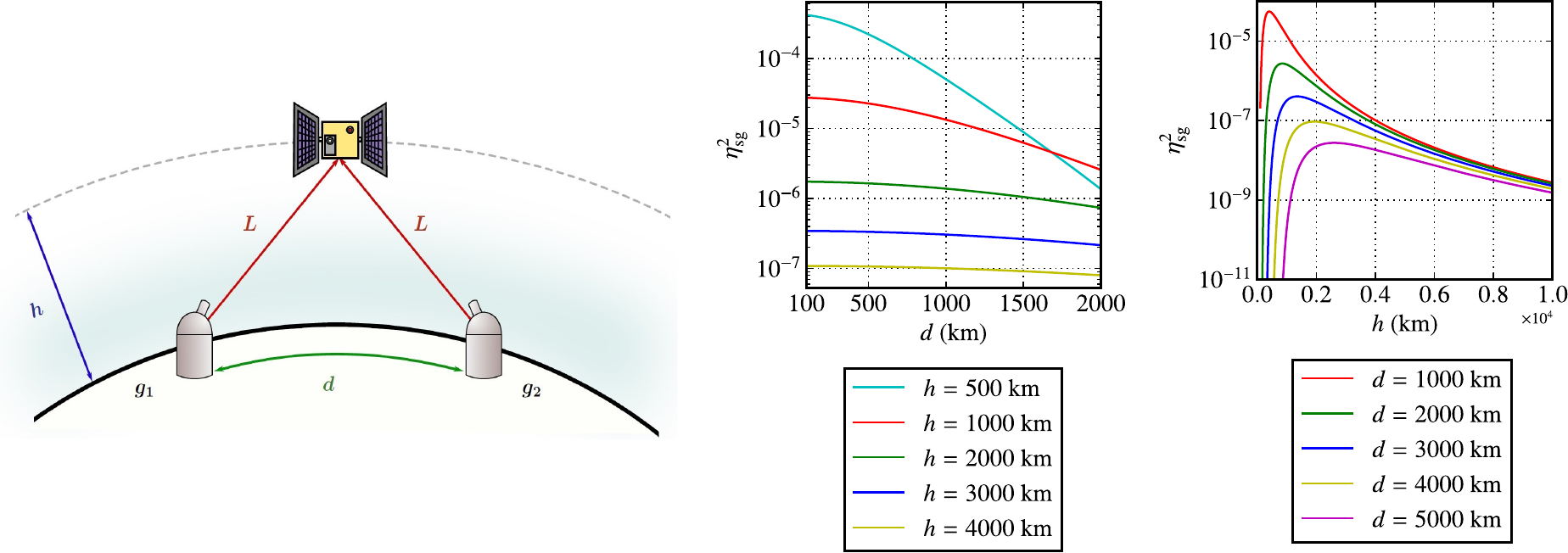}
		\caption{Optical satellite-to-ground transmission. (Left) Two ground stations $g_1$ and $g_2$ are separated by a distance $d$ with a satellite at altitude $h$ at the midpoint. Both ground stations are the same distance $L$ away from the satellite, so that the total transmittance for two-qubit entanglement transmission (one qubit to each ground station) is $\eta_{\text{sg}}^2$, where $\eta_{\text{sg}}=\eta_{\text{fs}}\eta_{\text{atm}}$, with $\eta_{\text{fs}}$ given by \eqref{eq-fs_transmittance_2} and $\eta_{\text{atm}}$ given by \eqref{eq-atmospheric_transmittance_zenith_2}. (Right) Plots of the total transmittance $\eta_{\text{sg}}^2$ as a function of $d$ and $h$.}\label{fig-atmosphere_geometry}
	\end{figure}
	
	Now, consider two ground stations, one corresponding to Alice and one corresponding to Bob. Given a state $\rho_{AB}^S$ produced by the source on the satellite, the state after transmission of the system $A$ to Alice and the system $B$ to Bob is given by \eqref{eq-transmission_noisy_output}:
	\begin{equation}
		\rho_{AB}^{S,\text{out}}=\left(\mathcal{L}_A^{\eta_{\text{sg}}^{(1)},\overline{n}_1}\otimes\mathcal{L}_B^{\eta_{\text{sg}}^{(2)},\overline{n}_2}\right)(\rho_{AB}^S),
	\end{equation}	
	where $\eta_{\text{sg}}^{(1)}$ and $\eta_{\text{sg}}^{(2)}$ are the transmittances to the ground stations and $\overline{n}_1$ and $\overline{n}_2$ are the corresponding thermal background noise parameters.
	
	After transmission, we assume a heralding procedure defined by post-selecting on coincident events using (perfect) photon-number-resolving detectors. One can justify this assumption because, in the high-loss and low-noise regimes ($\eta_{\text{sg}}^{(1)},\eta_{\text{sg}}^{(2)},\overline{n}\ll 1$), the probability of four-photon and three-photon occurrences is negligible compared to two-photon events. Therefore, upon successful heralding, the (unnormalized) quantum state shared by Alice and Bob is
	\begin{equation}\label{eq-initial_state_tilde_1_sats}
		\widetilde{\sigma}_{AB}(1)\coloneqq\Pi_{AB}\left(\mathcal{L}_A^{\eta_{\text{sg}}^{(1)},\overline{n}_1}\otimes\mathcal{L}_B^{\eta_{\text{sg}}^{(2)},\overline{n}_2}\right)(\rho_{AB}^S)\Pi_{AB},
	\end{equation}
	where
	\begin{equation}
		\Pi_{AB}\coloneqq(\dyad{H}_A+\dyad{V}_A)\otimes(\dyad{H}_B+\dyad{V}_B)
	\end{equation}
	is the projection onto the two-photon-coincidence subspace. Note that the projection $\Pi_{AB}$ is exactly the projection $\Lambda^1\otimes\Lambda^1$, with $\Lambda^1$ defined in \eqref{eq-photonic_heralding_example_1}. Then, the transmission-heralding success probability is, as per the definition in \eqref{eq-elem_link_success_prob},
	\begin{equation}
		p\coloneqq \Tr[\widetilde{\sigma}_{AB}(1)]=\Tr\left[\Pi_{AB}\left(\mathcal{L}_A^{\eta_{\text{sg}}^{(1)},\overline{n}_1}\otimes\mathcal{L}_B^{\eta_{\text{sg}}^{(2)},\overline{n}_2}\right)(\rho_{AB}^S)\right].
	\end{equation}
	We refer to the discussion in Section~\ref{sec-network_setup_ground_based} for an explanation of how the heralding procedure described here mathematically is conducted in practice.
	
	Now, let us take the source state $\rho_{AB}^S$ to be the following:
	\begin{equation}\label{eq-sats_source_state}
		\rho_{AB}^S=f_S\Phi_{AB}^{+}+\left(\frac{1-f_S}{3}\right)(\Phi_{AB}^-+\Psi_{AB}^++\Psi_{AB}^-),
	\end{equation}
	where $f_S\in[0,1]$ and 
	\begin{align}
		\Phi_{AB}^{\pm}&\coloneqq\ket{\Phi^{\pm}}\bra{\Phi^{\pm}}_{AB},\quad \ket{\Phi^{\pm}}\coloneqq\frac{1}{\sqrt{2}}(\ket{H,H}\pm\ket{V,V}),\\
		\Psi_{AB}^{\pm}&\coloneqq\ket{\Psi^{\pm}}\bra{\Psi^{\pm}}_{AB},\quad \ket{\Psi^{\pm}}\coloneqq\frac{1}{\sqrt{2}}(\ket{H,V}\pm\ket{V,H}).
	\end{align}	
	Using \eqref{eq-log_neg_Bell_diag_spec}, we have that the source state $\rho_{AB}^S$ is entangled if and only if $f_S>\frac{1}{2}$.
	
	Using \eqref{eq-sats_source_state}, we obtain an explicit form for the (unnormalized) state $\widetilde{\sigma}_{AB}(1)$ in \eqref{eq-initial_state_tilde_1_sats}.
	
	\begin{proposition}\label{prop-noisy_transmission_output_Bell}
		Let $\eta_{\text{sg}}^{(1)},\eta_{\text{sg}}^{(2)},\overline{n}_1,\overline{n}_2\in[0,1]$, and consider the source state $\rho_{AB}^S$ given by \eqref{eq-sats_source_state}. Then, after successful heralding, the (unnormalized) state $\widetilde{\sigma}_{AB}(1)$ given by \eqref{eq-initial_state_tilde_1_sats} is equal to
		\begin{align}
			\widetilde{\sigma}_{AB}(1)&=\Pi_{AB}\left(\mathcal{L}_A^{\eta_{\text{sg}}^{(1)},\overline{n}_1}\otimes\mathcal{L}_B^{\eta_{\text{sg}}^{(2)},\overline{n}_2}\right)(\rho_{AB}^S)\Pi_{AB}\nonumber\\
			&=\frac{1}{2}\left(f_S(a+b)+\left(\frac{1-f_S}{3}\right)(a+2c-b)\right)\Phi_{AB}^+\nonumber\\
			&\quad +\frac{1}{2}\left(f_S(a-b)+\left(\frac{1-f_S}{3}\right)(a+2c+b)\right)\Phi_{AB}^-\nonumber\\
			&\quad +\frac{1}{2}\left(f_Sc+\left(\frac{1-f_S}{3}\right)(2a+c)\right)\Psi_{AB}^+\nonumber\\
			&\quad +\frac{1}{2}\left(f_Sc+\left(\frac{1-f_S}{3}\right)(2a+c)\right)\Psi_{AB}^-,\label{eq-rho_0_sats}
		\end{align}
		where
		\begin{equation}\label{eq-initial_state_params}
			a\coloneqq x_1x_2+y_1y_2,\quad b\coloneqq z_1z_2,\quad c\coloneqq x_1y_2+y_1x_2,
		\end{equation}
		and
		\begin{align}
			x_i&\coloneqq (1-\overline{n}_i)\eta_{\text{sg}}^{(i)}+\frac{\overline{n}_i}{2}\left(\left(1-2\eta_{\text{sg}}^{(i)}\right)^2+\left(\eta_{\text{sg}}^{(i)}\right)^2\right),\\
			y_i&\coloneqq \frac{\overline{n}_i}{2}\left(1-\eta_{\text{sg}}^{(i)}\right)^2,\\
			z_i&\coloneqq (1-\overline{n}_i)\eta_{\text{sg}}^{(i)}-\overline{n}_i\eta_{\text{sg}}^{(i)}\left(1-2\eta_{\text{sg}}^{(i)}\right),
		\end{align}
		for $i\in\{1,2\}$.
	\end{proposition}

	\begin{proof}
		See Appendix~\ref{sec-noisy_transmission_output_Bell_pf}.
	\end{proof}

	From \eqref{eq-rho_0_sats}, we have that the transmission-heralding success probability is given by
	\begin{equation}\label{eq-trans_succ_prob_sats}
		p=\Tr[\widetilde{\sigma}_{AB}(1)]=a+c=(x_1+y_1)(x_2+y_2),
	\end{equation}
	so that the quantum state shared by Alice and Bob conditioned on successful heralding is, as per the definition in \eqref{eq-initial_link_states},
	\begin{equation}\label{eq-rho_0_norm_sats}
		\rho_{AB}^0=\frac{\widetilde{\sigma}_{AB}(1)}{p}.
	\end{equation}

\subsection{Basic figures of merit}

	Let us now evaluate the quality of entanglement transmission from a satellite to two ground stations. To do this, we make use of two of the figures of merit defined in Section~\ref{sec-practical_figures_merit}. In particular, we want to determine how these figures of merit vary as a function of ground station separation distance and satellite altitude. For illustrative purposes, and for simplicity, we focus primarily on the simple scenario depicted in the left-most panel of Figure~\ref{fig-atmosphere_geometry}, in which a satellite passes over the midpoint between two ground stations, and is also in the same plane as the ground stations. In this case, the satellite is an equal distance away from both ground stations, so that $\eta_{\text{sg}}^{(1)}=\eta_{\text{sg}}^{(2)}$. We also let $\overline{n}_1=\overline{n}_2$. This means that $x_1=x_2\equiv x$, $y_1=y_2\equiv y$ and $z_1=z_2\equiv z$, so that
	\begin{equation}\label{eq-sat_transmission_symmetric}
		a=x^2+y^2,\quad b=z^2,\quad c=2xy\quad (\eta_{\text{sg}}^{(1)}=\eta_{\text{sg}}^{(2)}=\eta_{\text{sg}}\text{ and } \overline{n}_1=\overline{n}_2=\overline{n}).
	\end{equation}
	In this scenario, given a distance $d$ between the ground stations and an altitude $h$ for the satellite, by simple geometry the distance $L$ between the satellite and either ground station is given by
	\begin{equation}\label{eq-link_distance_symmetric}
		L=\sqrt{4R_{\oplus}(R_{\oplus}+h)\sin^2\left(\frac{d}{4R_{\oplus}}\right)+h^2},
	\end{equation}
	where $R_{\oplus}$ is the radius of the earth.

	Now, one basic figure of merit from Section~\ref{sec-practical_figures_merit} is the expected initial link status $\mathbb{E}[X_{AB}(1)]$, which is equal simply to the transmission-heralding success probability $p$ in \eqref{eq-trans_succ_prob_sats}. Due to the altitude of the satellites, there typically has to be multiplexing of the signals (see Remark~\ref{rem-multiplexing}) in order to maintain a high probability of both ground stations receiving the entangled state. In Figure~\ref{fig-initial_fid_prob}, we plot the success probability with multiplexing, which is given by $1-(1-p)^M$, where $M$ is the number of distinct frequency modes used for multiplexing.
	
	\begin{figure}
		\centering
		\includegraphics[width=0.9\textwidth]{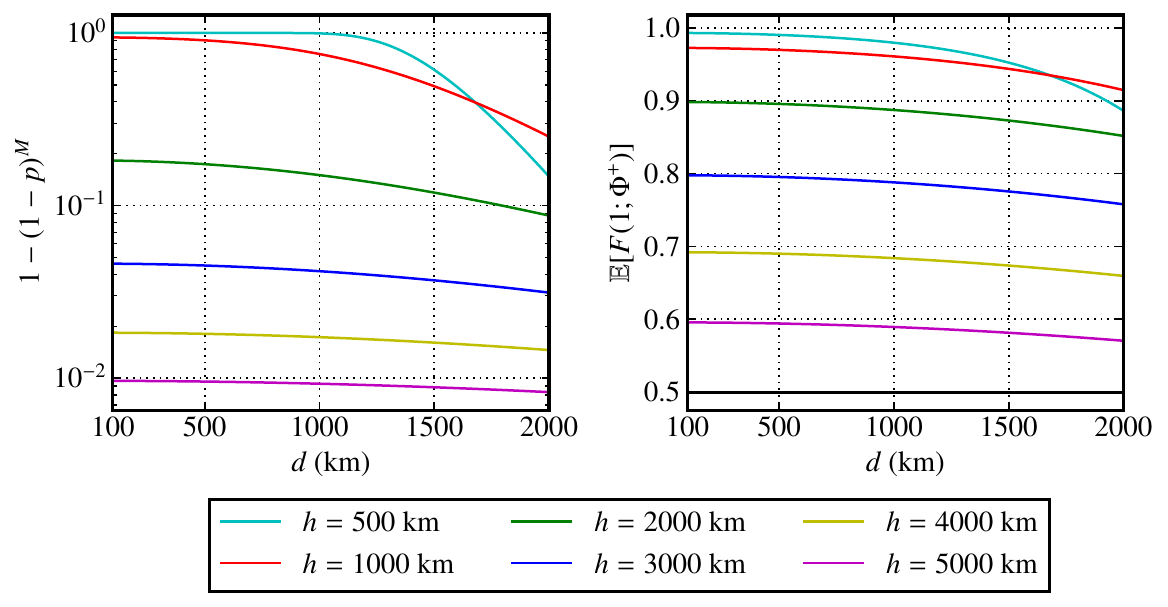}\\[1cm]
		\includegraphics[width=0.9\textwidth]{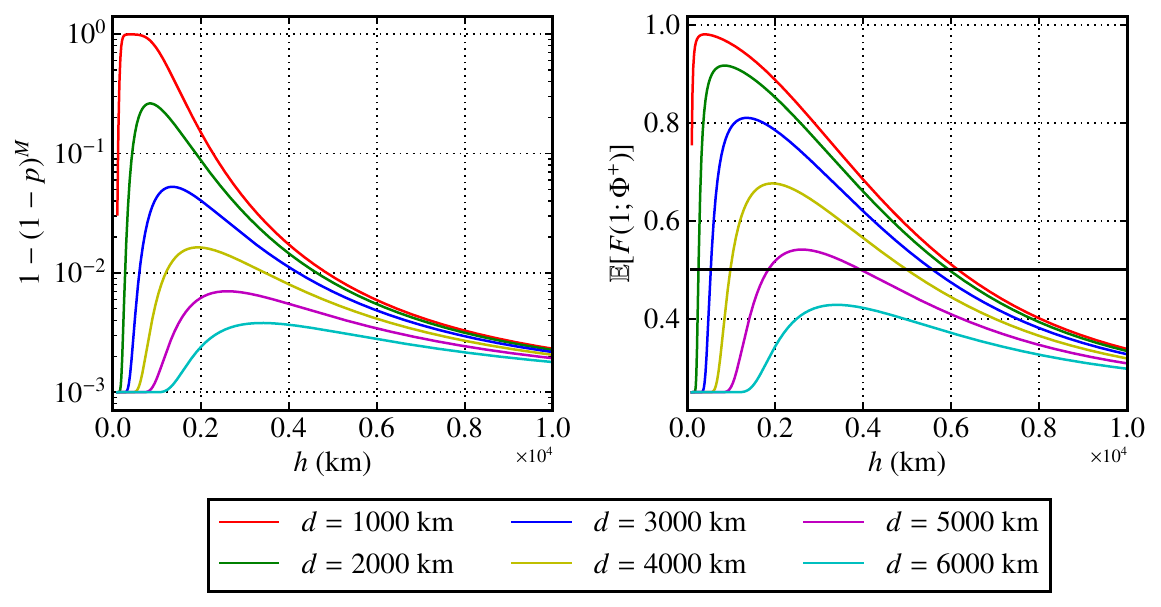}
		\caption{Plots of the transmission-heralding success probability as well as the initial fidelity of the quantum state $\rho_{AB}^0$ conditioned on successful heralding for the situation depicted in the left-most panel of Figure~\ref{fig-atmosphere_geometry}, in which $\eta_{\text{sg}}^{(1)}=\eta_{\text{sg}}^{(2)}=\eta_{\text{sg}}$ and $\overline{n}_1=\overline{n}_2=\overline{n}$. Indicated is the threshold fidelity of $\frac{1}{2}$ beyond which the state $\rho_{AB}^0$ is entangled (see Proposition~\ref{prop-initial_state_entanglement}). The success probability is shown in a multiplexing setting with $M=10^5$ (see Remark~\ref{rem-multiplexing}). Also, we have let $\overline{n}=10^{-4}$ and $f_S=1$.}\label{fig-initial_fid_prob}
	\end{figure}
	
	We also plot in Figure~\ref{fig-initial_fid_prob} the fidelity figure of merit $\mathbb{E}[F(1;\Phi^+)]$, which is given by
	\begin{align}
		\mathbb{E}[F(1;\Phi^+)]&=\bra{\Phi^+}\rho_{AB}^0\ket{\Phi^+}=\frac{1}{p}\mathbb{E}[\widetilde{F}(1;\Phi^+)],\label{eq-initial_fid_sats}\\
		\mathbb{E}[\widetilde{F}(1;\Phi^+)]&=\bra{\Phi^+}\widetilde{\sigma}_{AB}(1)\ket{\Phi^+}=\frac{1}{2}f_S(a+b)+\frac{1}{2}\left(\frac{1-f_S}{3}\right)(a+2c-b),\label{eq-initial_fid_tilde}
	\end{align}
	with $a,b,c$ given by \eqref{eq-initial_state_params} in general and by \eqref{eq-sat_transmission_symmetric} in the special case depicted in the left-most panel of Figure~\ref{fig-atmosphere_geometry}.
	
	The fidelity of $\rho_{AB}^0$ with respect to $\Phi_{AB}^+$ is related in a simple way to the entanglement of $\rho_{AB}^0$. In particular, as we now show, $\rho_{AB}^0$ is entangled if and only if its fidelity with respect to $\Phi_{AB}^+$ is strictly greater than $\frac{1}{2}$, and this leads to constraints on the loss and noise parameters of the satellite-to-ground transmission. 
	
	\begin{proposition}\label{prop-initial_state_entanglement}
		The quantum state $\rho_{AB}^0$ after successful satellite-to-ground transmission, as defined in \eqref{eq-rho_0_norm_sats}, is entangled if and only if the fidelity of the source state in \eqref{eq-sats_source_state} satisfies $f_S>\frac{1}{2}$, and
		\begin{equation}\label{eq-initial_state_ent_cond_pf}
			2(f_S-1)a+(4f_S-1)b-(1+2f_S)c>0,
		\end{equation}
		with $a,b,c$ given by \eqref{eq-initial_state_params} in general and by \eqref{eq-sat_transmission_symmetric} in the special case depicted in the left-most panel of Figure~\ref{fig-atmosphere_geometry}.
	\end{proposition}
	
	\begin{proof}
		Observe that the state $\rho_{AB}^0$ is a Bell-diagonal state of the form
		\begin{equation}\label{eq-initial_state_entanglement_pf1}
			\rho_{AB}^0=(\alpha+\beta)\Phi_{AB}^++(\alpha-\beta)\Phi_{AB}^-+\gamma\Psi_{AB}^++\gamma\Psi_{AB}^-,
		\end{equation}
		where $\alpha,\beta,\gamma\geq 0$ (when $f_S>\frac{1}{2}$). Indeed, the coefficient of $\Phi_{AB}^+$ in \eqref{eq-rho_0_sats} can be written as
		\begin{equation}
			\frac{1}{2}f_Sa+\frac{1}{2}\left(\frac{1-f_S}{3}\right)(a+2c)+\frac{1}{2}f_Sb-\frac{1}{2}\left(\frac{1-f_S}{3}\right)b,
		\end{equation}
		and the coefficient of $\Phi_{AB}^-$ in \eqref{eq-rho_0_sats} can be written as
		\begin{equation}
			\frac{1}{2}f_Sa+\frac{1}{2}\left(\frac{1-f_S}{3}\right)(a+2c)-\left(\frac{1}{2}f_Sb-\frac{1}{2}\left(\frac{1-f_S}{3}\right)b\right).
		\end{equation}
		We can thus make the following identifications:
		\begin{align}
			\alpha&\equiv\frac{1}{a+c}\left(\frac{1}{2}f_Sa+\frac{1}{2}\left(\frac{1-f_S}{3}\right)(a+2c)\right),\label{eq-rho0_alpha}\\
			\beta&\equiv \frac{1}{a+c}\left(\frac{1}{2}f_Sb-\frac{1}{2}\left(\frac{1-f_S}{3}\right)b\right),\label{eq-rho0_beta}\\
			\gamma&\equiv \frac{1}{2}f_Sc+\frac{1}{2}\left(\frac{1-f_S}{3}\right)(2a+c).\label{eq-rho0_gamma}
		\end{align}
		Now, using \eqref{eq-log_neg_Bell_diag_spec}, we have that $\rho_{AB}^0$ is entangled if and only if $\bra{\Phi^+}\rho_{AB}^0\ket{\Phi^+}>\frac{1}{2}$. Then, from \eqref{eq-initial_fid_sats}, we have that
		\begin{equation}
			\bra{\Phi^+}\rho_0\ket{\Phi^+}=\frac{1}{2}f_S\frac{a+b}{a+c}+\frac{1}{2}\left(\frac{1-f_S}{3}\right)\frac{a+2c-b}{a+c},
		\end{equation}
		so we require
		\begin{equation}
			\frac{1}{2}f_S\frac{a+b}{a+c}+\frac{1}{2}\left(\frac{1-f_S}{3}\right)\frac{a+2c-b}{a+c}>\frac{1}{2}.
		\end{equation}
		Simplifying this leads to
		\begin{equation}
			2(f_S-1)a+(4f_S-1)b-(1+2f_S)c>0,
		\end{equation}
		as required.
	\end{proof}
	
	\begin{figure}
		\centering
		\includegraphics[width=0.48\textwidth]{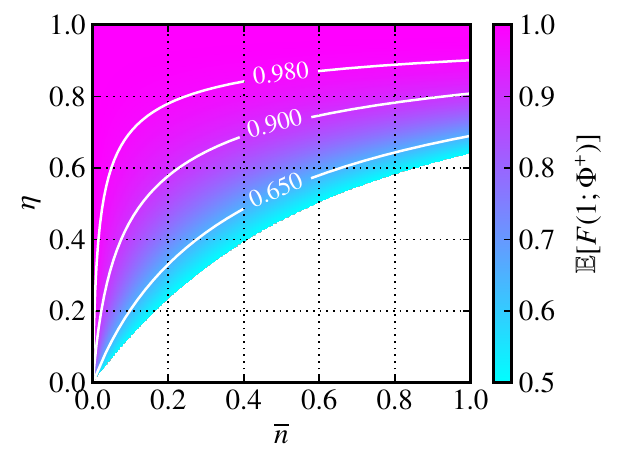}\quad
		\includegraphics[width=0.48\textwidth]{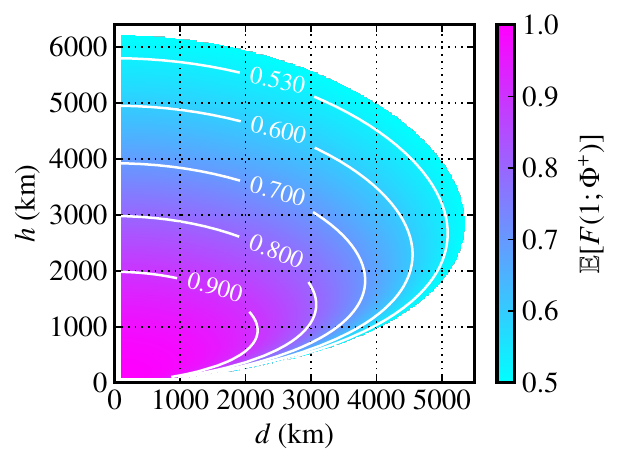}
		\caption{Plots of the entanglement region for the state $\rho_{AB}^0$ obtained after successful satellite-to-ground transmission for the scenario depicted in the left-most panel of Figure~\ref{fig-atmosphere_geometry}. For both plots, we assume $f_S=1$. For the right-hand plot, we take $\overline{n}_1=\overline{n}_2=10^{-4}$.}\label{fig-ent_sg}
	\end{figure}

	Now, for the scenario depicted in the left-most panel of Figure~\ref{fig-atmosphere_geometry}, we have that $x_1=x_2=x$, $y_1=y_2=y$, and $z_1=z_2=z$, so that from \eqref{eq-sat_transmission_symmetric} we have $a=x^2+y^2$, $b=z^2$, and $c=2xy$. Substituting this into \eqref{eq-initial_state_ent_cond_pf} leads to $2(f_S-1)(x^2+y^2)+(4f_S-1)z^2-2(1+2f_S)xy>0$ as the condition for $\rho_{AB}^0$ to be entangled. We plot this condition in Figure~\ref{fig-ent_sg}. The inequality gives us the colored regions, and the values within the regions are obtained by evaluating the fidelity according to \eqref{eq-initial_fid_sats}.

\subsection{Key rates for QKD}
	
	Let us also consider key rates for quantum key distribution (QKD) between Alice and Bob, who are at the ends of the elementary link whose quantum state is $\rho_{AB}^0$ (conditioned on successful transmission and heralding), as given by \eqref{eq-rho_0_norm_sats}. (See Section~\ref{sec-QKD} for a brief overview of QKD.) Recalling from the proof of Proposition~\ref{prop-initial_state_entanglement} that $\rho_{AB}^0$ is a quantum state of the form
	\begin{equation}
		\rho_{AB}^0=(\alpha+\beta)\Phi_{AB}^++(\alpha-\beta)\Phi_{AB}^-+\gamma\Psi_{AB}^++\gamma\Psi_{AB}^-,
	\end{equation}
	with $\alpha,\beta,\gamma$ defined in \eqref{eq-rho0_alpha}--\eqref{eq-rho0_gamma}, it is easy to show using \eqref{eq-rhoAB_Qx}--\eqref{eq-rhoAB_Qz} that
	\begin{align}
		Q_{\text{BB84}}^{(d,h)}&=\frac{1}{2}(Q_x+Q_z)=\frac{3}{4}-\frac{1}{2}\beta-\alpha,\label{eq-QBER_BB84_sats}\\
		Q_{\text{6-state}}^{(d,h)}&=\frac{1}{3}(Q_x+Q_y+Q_z)=\frac{2}{3}(1-(\alpha+\beta)).\label{eq-QBER_6state_sats}
	\end{align}
	For the device-independent protocol, we assume that the correlation is such that $Q_{\text{DI}}^{(d,h)}=Q_{\text{6-state}}^{(d,h)}$ and $S^{(d,h)}=2\sqrt{2}(1-2Q_{\text{DI}}^{(d,h)})$; see Section~\ref{sec-QKD} and the caption of Figure~\ref{fig-key_rates} for more information. Then, assuming that $M$ signals per second are transmitted from the satellite, using \eqref{eq-network_QKD_key_rate} we have
	\begin{align}
		\widetilde{K}_{\text{BB84}}(d,h)&=M(a+c)K_{\text{BB84}}(Q_{\text{BB84}}^{(d,h)})=M(a+c)(1-2h_2(Q_{\text{BB84}}^{(d,h)})),\label{eq-key_rate_BB84_sats}\\[0.3cm]
		\widetilde{K}_{\text{6-state}}(d,h)&=M(a+c)K_{\text{6-state}}(Q_{\text{6-state}}^{(d,h)})\\
		&=M(a+c)\left(1+\left(1-\frac{3Q_{\text{6-state}}^{(d,h)}}{2}\right)\log_2\left(1-\frac{3Q_{\text{6-state}}^{(d,h)}}{2}\right)+\frac{3Q_{\text{6-state}}^{(d,h)}}{2}\log_2\left(\frac{Q_{\text{6-state}}^{(d,h)}}{2}\right)\right),\label{eq-key_rate_6state_sats}\\[0.2cm]
		\widetilde{K}_{\text{DI}}(d,h)&=M(a+c)K_{\text{DI}}(Q_{\text{DI}}^{(d,h)},S^{(d,h)})\\
		&=1-h_2(Q_{\text{DI}}^{(d,h)})-h_2\left(\frac{1+\sqrt{(S^{(d,h)}/2)^2-1}}{2}\right),\label{eq-key_rate_DI_sats}
	\end{align}
	where we recall the definitions of the functions $K_{\text{BB84}}$, $K_{\text{6-state}}$, and $K_{\text{DI}}$ in \eqref{eq-key_rate_BB84}, \eqref{eq-key_rate_6state}, and \eqref{eq-key_rate_DIQKD}, respectively. Also, recall that $p=a+c$ is the transmission-heralding success probability, so that $M(a+c)$ is the average number of signals received successfully per second. We plot these secret key rates in Figure~\ref{fig-key_rates_sats}.
	
	\begin{figure}
		\centering
		\includegraphics[width=0.48\textwidth]{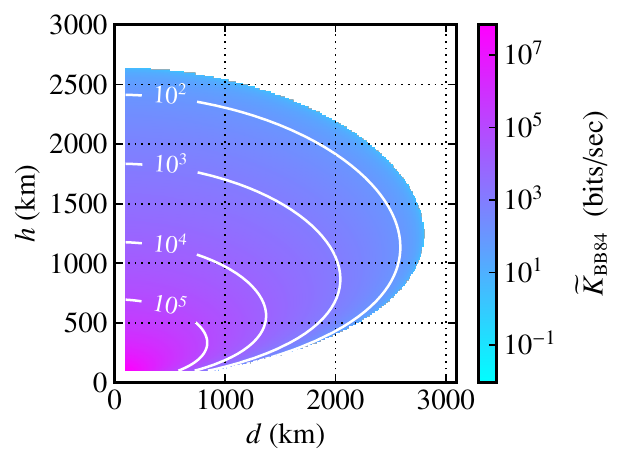}\quad
		\includegraphics[width=0.48\textwidth]{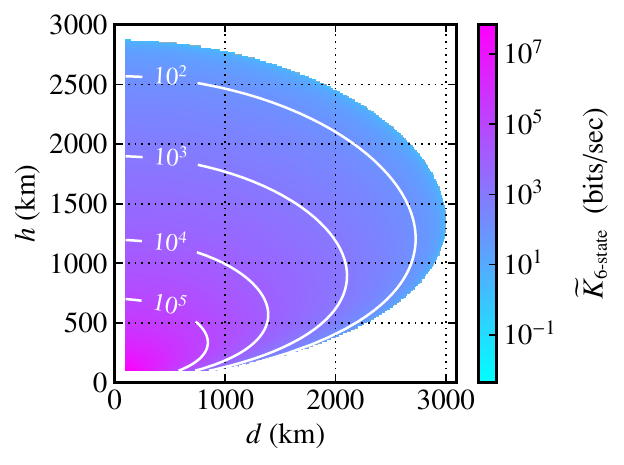}\\
		\includegraphics[width=0.48\textwidth]{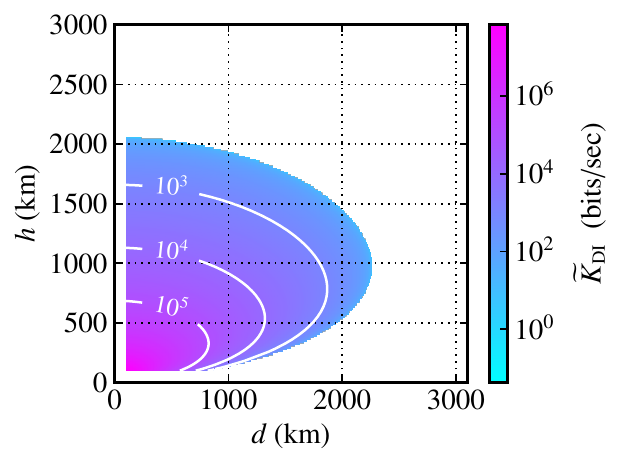}
		\caption{Asymptotic secret key rates for the BB84, six-state, and device-independent (DI) quantum key distribution protocols (see Section~\ref{sec-QKD}) for the scenario depicted in Figure~\ref{fig-atmosphere_geometry}. When calculating the error rates in \eqref{eq-QBER_BB84_sats} and \eqref{eq-QBER_6state_sats}, we take $f_S=1$. To calculate the key rates in \eqref{eq-key_rate_BB84_sats}, \eqref{eq-key_rate_6state_sats}, and \eqref{eq-key_rate_DI_sats}, we have taken $M=10^9$. See the caption of Figure~\ref{fig-key_rates} for more information on how the secret key rate for the DI protocol is calculated.}\label{fig-key_rates_sats}
	\end{figure}
	
	In Figure~\ref{fig-key_rates_sats}, notice that the region of non-zero secret key rate is largest for the six-state protocol, with the region for the BB84 protocol being smaller and the region for the DI protocol being even smaller. This is due to the fact that the error threshold for the DI protocol is the smallest among the three protocols, with the error threshold for the BB84 protocol slightly larger, and the error threshold for the 6-state protocol the largest; see Figure~\ref{fig-key_rates}.

\section{Quantum memory model}\label{sec-quantum_memory_sats}

	Having examined the quantum state immediately after successful transmission and heralding, let us now consider a particular model of decoherence for the quantum memories in which the transmitted qubits are stored. For illustrative purposes, we consider a simple amplitude damping decoherence model for the quantum memories. The amplitude damping channel $\mathcal{A}_{\gamma}$ is a qubit channel, with $\gamma\in[0,1]$, such that \cite{NC00_book}
	\begin{equation}
		\rho=\begin{pmatrix} 1-x & \alpha \\ \conj{\alpha} & x\end{pmatrix} \longmapsto \mathcal{A}_{\gamma}(\rho)=\begin{pmatrix} 1-x(1-\gamma) & \alpha\sqrt{1-\gamma} \\ \conj{\alpha}\sqrt{1-\gamma} & x(1-\gamma) \end{pmatrix}.
	\end{equation}
	Note that for $\gamma=0$ we recover the noiseless (identity) channel. We can relate $\gamma$ to the coherence time of the quantum memory, which we denote by $t_{\text{coh}}$, as follows \cite[Section~3.4.3]{Preskill20}:
	\begin{equation}\label{eq-AD_gamma_coh}
		\gamma\coloneqq 1-\e^{-\frac{1}{t_{\text{coh}}}}.
	\end{equation}
	Note that infinite coherence time corresponds to an ideal quantum memory, meaning that the quantum channel is noiseless. Indeed, by relating the noise parameter $\gamma$ to the coherence time as in \eqref{eq-AD_gamma_coh}, we have that $t_{\text{coh}}=\infty\Rightarrow\gamma=0$.
	
	For $m\geq 1$ applications of the amplitude damping channel, it is straightforward to show that
	\begin{equation}
		\mathcal{A}_{\gamma}^{\circ m}(\rho)=\begin{pmatrix} 1-(1-\gamma)^m x & \alpha(1-\gamma)^{\frac{m}{2}} \\ \conj{\alpha}(1-\gamma)^{\frac{m}{2}} & (1-\gamma)^m x \end{pmatrix} = \begin{pmatrix} 1-\e^{-\frac{m}{t_{\text{coh}}}}x & \alpha\e^{-\frac{m}{2t_{\text{coh}}}} \\ \conj{\alpha}\e^{-\frac{m}{2t_{\text{coh}}}} & \e^{-\frac{m}{t_{\text{coh}}}} x \end{pmatrix}=\begin{pmatrix} 1-\lambda_m x & \alpha\lambda_m \\ \conj{\alpha}\lambda_m & \lambda_m x\end{pmatrix},
	\end{equation}
	where in the last equality we let
	\begin{equation}
		\lambda_m\coloneqq\e^{-\frac{m}{t_{\text{coh}}}}.
	\end{equation}
	In particular, we have
	\begin{align}
		\mathcal{A}_{\gamma}^{\circ m}(\ket{0}\bra{0})=\ket{0}\bra{0},\quad \mathcal{A}_{\gamma}^{\circ m}(\ket{0}\bra{1})&=\sqrt{\lambda_m}\ket{0}\bra{1},\quad \mathcal{A}_{\gamma}^{\circ m}(\ket{1}\bra{0})=\sqrt{\lambda_m}\ket{1}\bra{0},\\
		\mathcal{A}_{\gamma}^{\circ m}(\ket{1}\bra{1})&=(1-\lambda_m)\ket{0}\bra{0}+\lambda_m\ket{1}\bra{1},
	\end{align}
	Using this, it is straightforward to show that, for all $m\geq 1$,
	\begin{align}
		\rho_{AB}(m)&\coloneqq(\mathcal{A}_{\gamma}^{\circ m}\otimes\mathcal{A}_{\gamma}^{\circ m})(\rho_{AB}^0)\\
		&=\left(\alpha\lambda_m^2+\left(\beta-\frac{1}{2}\right)\lambda_m+\frac{1}{2}\right)\Phi_{AB}^++\left(\alpha\lambda_m^2+\left(-\beta-\frac{1}{2}\right)\lambda_m+\frac{1}{2}\right)\Phi_{AB}^-\\
		&\quad +\lambda_m\left(\frac{1}{2}-\alpha\lambda_m\right)\Psi_{AB}^++\lambda_m\left(\frac{1}{2}-\alpha\lambda_m\right)\Psi_{AB}^-\\
		&\quad +\frac{1}{2}(1-\lambda_m)\left(\ket{\Phi^+}\bra{\Phi^-}_{AB}+\ket{\Phi^-}\bra{\Phi^+}_{AB}\right),\label{eq-elem_link_m_sats}
	\end{align}
	where $\alpha$ and $\beta$ are given by \eqref{eq-rho0_alpha} and \eqref{eq-rho0_beta}, respectively. Note that we have assumed that the memories corresponding to systems $A$ and $B$ have the same coherence time. It follows that
	\begin{equation}\label{eq-fid_decay_sats}
		f_m(\rho_{AB}^0;\Phi^+)\coloneqq\bra{\Phi^+}(\mathcal{A}_{\gamma}^{\circ m}\otimes\mathcal{A}_{\gamma}^{\circ m})(\rho_{AB}^0)\ket{\Phi^+}=\alpha\lambda_m^2+\left(\beta-\frac{1}{2}\right)\lambda_m+\frac{1}{2}.
	\end{equation}
	Note that $f_m(\rho_{AB}^0;\Phi^+)\leq f_0(\rho_{AB}^0;\Phi^+)$ for all $m\in\mathbb{N}_0$.

\section{Policies}
	
	We now consider policies for elementary links. Specifically, we consider the quantum decision process for elementary links given in Definition~\ref{def-network_QDP_elem_link}. Recall that the main elements needed for evaluating policies are the quantum state $\rho_{AB}^0$ immediately after successful transmission and heralding, which we determined in \eqref{eq-rho_0_norm_sats}, and the fidelity ``decay'' function $f_m(\rho_{AB}^0;\Phi^+)$, which we determined in \eqref{eq-fid_decay_sats} under the amplitude damping decoherence model. With these elements in place, let us start by evaluating the memory-cutoff policy.

\subsection{Memory-cutoff policy}

	The memory-cutoff policy, to which we devoted the entirety of Chapter~\ref{chap-mem_cutoff}, is defined by the following decision function:
	\begin{equation}
		d_t^{t^{\star}}(h^t)=\left\{\begin{array}{l l}0 & \text{if }M^{t^{\star}}(t)(h^t)<t^{\star}, \\ 1 & \text{if }M^{t^{\star}}(t)(h^t)=t^{\star}, \end{array}\right.
	\end{equation}
	for all $t\geq 1$ and all histories $h^t\in\{0,1\}^{2t-1}$, where $t^{\star}\in\mathbb{N}_0$ is the cutoff and $M^{t^{\star}}(t)$ is the memory time random variable, whose definition is in \eqref{eq-mem_time_cutoff_policy}. Intuitively, in the memory-cutoff policy, the entangled state, once successfully stored in quantum memories at the nodes, is kept there for $t^{\star}$ time steps, at which point it is discarded and a new one is requested from the source.
	
	From \eqref{eq-avg_q_state_link_active}, the expected quantum state $\sigma_{AB}^{t^{\star}}(t|X(t)=1)$ at time $t\geq 1$ of an elementary link obtained via satellite-to-ground distribution, conditioned on the elementary link being active at time $t$, is given by
	\begin{equation}
		\sigma_{AB}^{t^{\star}}(t|X(t)=1)=\left\{\begin{array}{l l} \displaystyle \sum_{m=0}^{t-1}\frac{p(1-p)^{t-(m+1)}}{1-(1-p)^t}\rho_{AB}(m), & t\leq t^{\star}+1,\\[0.5cm] \displaystyle \sum_{m=0}^{t^{\star}}\frac{\Pr[M^{t^{\star}}(t)=m,X(t)=1]_{t^{\star}}}{\Pr[X(t)=1]_{t^{\star}}}\rho_{AB}(m), & t>t^{\star}+1, \end{array}\right.
	\end{equation}
	where the expression for $\rho_{AB}(m)$ is given by \eqref{eq-elem_link_m_sats}. 
	
	For the elementary link fidelity, from \eqref{eq-avg_fid_tilde} and \eqref{eq-avg_fid}, we have that
	\begin{align}
		\mathbb{E}[\widetilde{F}^{t^{\star}}(t;\Phi^+)]&=\left\{\begin{array}{l l} \displaystyle \sum_{m=0}^{t-1}f_m(\rho_{AB}^0;\Phi^+)p(1-p)^{t-(m+1)} & t\leq t^{\star}+1,\\[0.5cm] \displaystyle \sum_{m=0}^{t^{\star}} f_m(\rho_{AB}^0;\Phi^+) \Pr[M^{t^{\star}}(t)=m,X(t)=1]_{t^{\star}} & t>t^{\star}+1, \end{array}\right.\label{eq-avg_fid_tilde_sats}\\[0.5cm]
		\mathbb{E}[F^{t^{\star}}(t;\Phi^+)]&=\left\{\begin{array}{l l}\displaystyle \sum_{m=0}^{t-1}f_m(\rho^0;\psi)\frac{p(1-p)^{t-(m+1)}}{1-(1-p)^t} & t\leq t^{\star}+1,\\[0.5cm] \displaystyle \sum_{m=0}^{t^{\star}} f_m(\rho_{AB}^0;\Phi^+)\frac{\Pr[M^{t^{\star}}(t)=m,X(t)=1]_{t^{\star}}}{\Pr[X(t)=1]_{t^{\star}}} & t>t^{\star}+1, \end{array}\right.
	\end{align}
	for all $t\geq 1$, where the expression for $f_m(\rho_{AB}^0;\Phi^+)$ is given in \eqref{eq-fid_decay_sats}. In the limit $t\to\infty$ and for $t^{\star}\in\mathbb{N}_0$, from \eqref{eq-avg_fid_tilde_tInfty} and \eqref{eq-avg_fid_tInfty}, we have 
	\begin{align}
		\lim_{t\to\infty}\mathbb{E}[\widetilde{F}^{t^{\star}}(t;\Phi^+)]&=\frac{p}{1+t^{\star}p}\sum_{m=0}^{t^{\star}}\left(\alpha\lambda_m^2+\left(\beta-\frac{1}{2}\right)\lambda_m+\frac{1}{2}\right),\\
		\lim_{t\to\infty}\mathbb{E}[F^{t^{\star}}(t;\Phi^+)]&=\frac{1}{t^{\star}+1}\sum_{m=0}^{t^{\star}}\left(\alpha\lambda_m^2+\left(\beta-\frac{1}{2}\right)\lambda_m+\frac{1}{2}\right)
	\end{align}
	Then, using the fact that $\lambda_m=\e^{-\frac{m}{t_{\text{coh}}}}$, it is straightforward to show that
	\begin{equation}
		\sum_{m=0}^{t^{\star}}\lambda_m=\e^{-\frac{t^{\star}}{2t_{\text{coh}}}}\frac{\sinh\left(\frac{1+t^{\star}}{2t_{\text{coh}}}\right)}{\sinh\left(\frac{1}{2t_{\text{coh}}}\right)},\quad \sum_{m=0}^{t^{\star}}\lambda_m^2=\e^{-\frac{t^{\star}}{t_{\text{coh}}}}\frac{\sinh\left(\frac{1+t^{\star}}{t_{\text{coh}}}\right)}{\sinh\left(\frac{1}{t_{\text{coh}}}\right)}.
	\end{equation}
	Therefore,
	\begin{align}
		\lim_{t\to\infty}\mathbb{E}[\widetilde{F}^{t^{\star}}(t;\Phi^+)]&=\frac{\alpha p \e^{-\frac{t^{\star}}{t_{\text{coh}}}}}{1+t^{\star}p}\frac{\sinh\left(\frac{1+t^{\star}}{t_{\text{coh}}}\right)}{\sinh\left(\frac{1}{t_{\text{coh}}}\right)}+\frac{p \e^{-\frac{t^{\star}}{2t_{\text{coh}}}}}{1+t^{\star}p}\left(\beta-\frac{1}{2}\right)\frac{\sinh\left(\frac{1+t^{\star}}{2t_{\text{coh}}}\right)}{\sinh\left(\frac{1}{2t_{\text{coh}}}\right)}+\frac{1}{2}\frac{(t^{\star}+1)p}{1+t^{\star}p},\\
		\lim_{t\to\infty}\mathbb{E}[F^{t^{\star}}(t;\Phi^+)]&=\frac{\alpha\e^{-\frac{t^{\star}}{t_{\text{coh}}}}}{t^{\star}+1}\frac{\sinh\left(\frac{1+t^{\star}}{t_{\text{coh}}}\right)}{\sinh\left(\frac{1}{t_{\text{coh}}}\right)}+\frac{\e^{-\frac{t^{\star}}{2t_{\text{coh}}}}}{t^{\star}+1}\left(\beta-\frac{1}{2}\right)\frac{\sinh\left(\frac{1+t^{\star}}{2t_{\text{coh}}}\right)}{\sinh\left(\frac{1}{2t_{\text{coh}}}\right)}+\frac{1}{2}.
	\end{align}

	For $t^{\star}=\infty$, from \eqref{eq-avg_fid_tilde_sats}, we obtain
	\begin{equation}
		\mathbb{E}[\widetilde{F}^{\infty}(t;\Phi^+)]=\sum_{m=0}^{t-1}\left(\alpha\lambda_m^2+\left(\beta-\frac{1}{2}\right)\lambda_m+\frac{1}{2}\right)p(1-p)^{t-1-m}
	\end{equation}
	for all $t\geq 1$. Evaluating the sums leads to
	\begin{multline}\label{eq-fid_tilde_tstarInf}
		\mathbb{E}[\widetilde{F}^{\infty}(t;\Phi^+)]=\frac{\alpha p\e^{\frac{2}{t_{\text{coh}}}}\left(\e^{-\frac{2t}{t_{\text{coh}}}}-(1-p)^t\right)}{1-\e^{\frac{2}{t_{\text{coh}}}}(1-p)}+\left(\beta-\frac{1}{2}\right)\frac{p\e^{\frac{1}{t_{\text{coh}}}}\left(\e^{-\frac{t}{t_{\text{coh}}}}-(1-p)^t\right)}{1-\e^{\frac{1}{t_{\text{coh}}}}(1-p)}\\+\frac{1}{2}\left(1-(1-p)^t\right).
	\end{multline}
	Then, for all $p\in(0,1]$, we obtain
	\begin{equation}
		\lim_{t\to\infty}\mathbb{E}[\widetilde{F}^{\infty}(t;\Phi^+)]=\frac{1}{2}.
	\end{equation}
	
	Let us now focus primarily on the $t^{\star}=\infty$ memory-cutoff policy by considering an example. Consider the situation depicted in the left-most panel of Figure~\ref{fig-atmosphere_geometry}, in which we have two ground stations separated by a distance $d$ and a satellite at altitude $h$ that passes over the midpoint between the two ground stations and is also in the same plane as the ground stations. Now, given that the ground stations are separated by a distance $d$, it takes time at least $\frac{2d}{c}$ to perform the heralding procedure, as this is the round-trip communication time between the ground stations ($c$ is the speed of light). We thus take the duration of each time step in the quantum decision process for the elementary link to be $\frac{2d}{c}$. If the coherence time of the quantum memories is $x$ seconds, then $t_{\text{coh}}=\frac{xc}{2d}$ time steps. In Figure~\ref{fig-tstarInf_sats}, we plot the quantities $\mathbb{E}[\widetilde{F}^{\infty}(t;\Phi^+)]$ (solid lines), $\mathbb{E}[F^{\infty}(t;\Phi^+)]$ (dashed lines), and $\mathbb{E}[X(t)]_{\infty}$ (dotted lines) for the $t^{\star}=\infty$ memory-cutoff policy under this scenario.
	
	\begin{figure}
		\centering
		\includegraphics[width=0.9\textwidth]{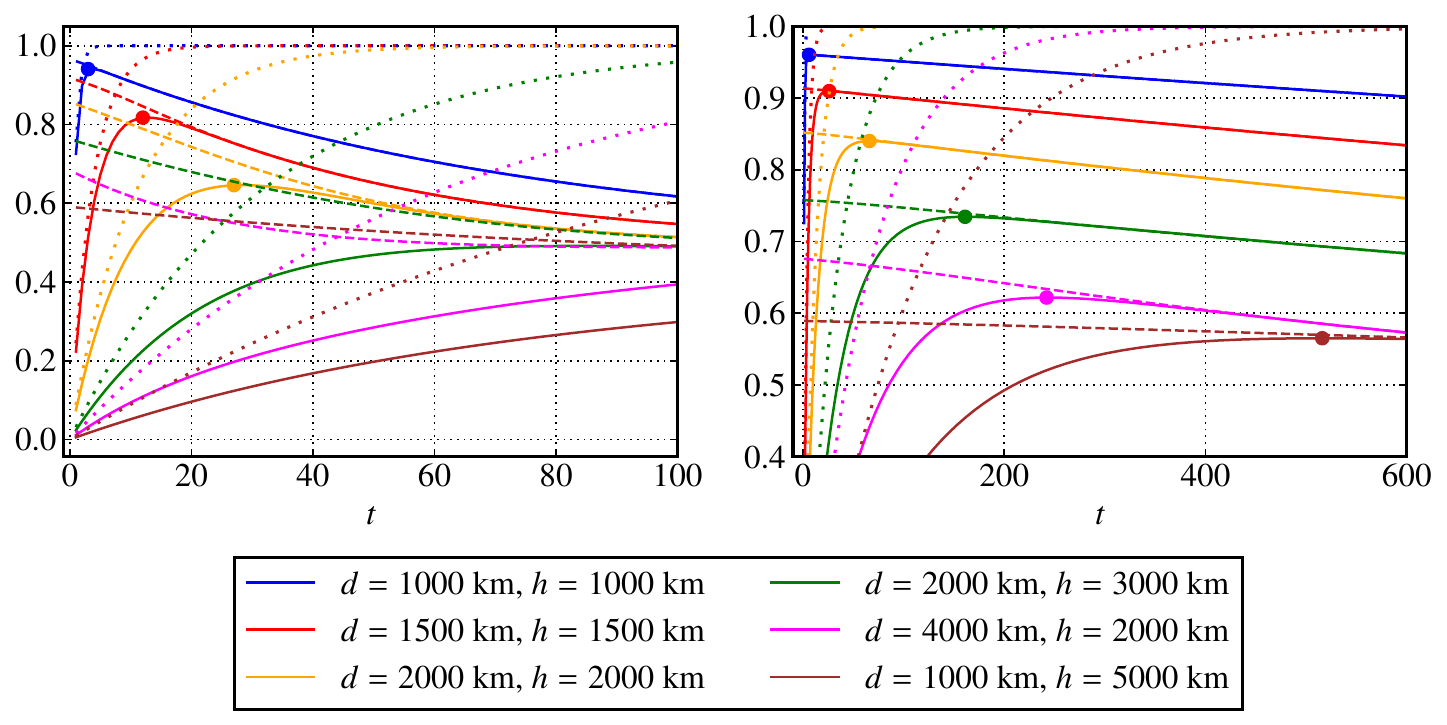}
		\caption{The $t^{\star}=\infty$ memory-cutoff policy for satellite-to-ground elementary link generation for various ground distances $d$ and satellite altitudes $h$, according to the situation depicted in the left-most panel of Figure~\ref{fig-atmosphere_geometry}. The solid lines are $\mathbb{E}[\widetilde{F}^{\infty}(t;\Phi^+)]$ (as given by \eqref{eq-fid_tilde_tstarInf}), the dashed lines are $\mathbb{E}[F^{\infty}(t;\Phi^+)]$, and the dotted lines are $\mathbb{E}[X(t)]=\Pr[X(t)=1]=1-(1-p)^t$ (see \eqref{cor-link_status_Pr1}), where $p=1-(1-(a+c))^M$, with $a$ and $c$ given by \eqref{eq-sat_transmission_symmetric} and $M=10^5$. We let $f_S=1$ be the fidelity of the source, we let $\overline{n}_1=\overline{n}_2=10^{-4}$ be the average number of background photons, and we take the memory coherence times to be 1 s (left) and 60 s (right). The dots are placed at the maxima of the curves for $\mathbb{E}[\widetilde{F}^{\infty}(t;\Phi^+)]$.}\label{fig-tstarInf_sats}
	\end{figure}
	
	We can see clearly in Figure~\ref{fig-tstarInf_sats} the advantage of using the random variable $\widetilde{F}$ as the figure of merit for evaluating policies. In particular, we can see the trade-off between the quantities $\widetilde{F}$, $F$, and $X$. On the one hand, the fidelity $\mathbb{E}[F^{\infty}(t;\Phi^+)]$ is always highest at time $t=1$, but at this point the elementary link activity probability $\mathbb{E}[X(t)]_{\infty}$ is simply $p$. Since we want not only a high fidelity for the elementary link but also a high probability that the elementary link is active, by optimizing $\widetilde{F}$ it is possible to achieve a higher link activity probability at the expense of a slightly lower fidelity. Specifically, in Figure~\ref{fig-tstarInf_sats}, we see that for every configuration of $d$ and $h$ there exists a time step $t_{\text{crit}}\geq 1$ at which $\widetilde{F}$ is maximal. At this point, the link activity probability is $1-(1-p)^{t_{\text{crit}}}$, which in many cases is dramatically greater than $p$, while the fidelity $\mathbb{E}[F^{\infty}(t_{\text{crit}};\Phi^+)]$ is only slightly lower than the fidelity at time $t=1$. Therefore, by waiting until time $t_{\text{crit}}$, it is possible to obtain an elementary link that is almost deterministically active, while incurring only a slight decrease in the fidelity. The time $t_{\text{crit}}$, obtained by optimizing the quantity $\mathbb{E}[\widetilde{F}^{\infty}(t;\Phi^+)]$ with respect to time $t$ and can be found using the formula in \eqref{eq-fid_tilde_tstarInf}, can be viewed as the optimal horizon time $T$ that should be chosen for the quantum network protocol presented in Figure~\ref{fig-QDP_protocol}. We refer to Ref.~\cite{CRDW19} for an argument similar to the one presented here, except that in Ref.~\cite{CRDW19} the time $t_{\text{crit}}$ is obtained by considering a desired value of the fidelity $\mathbb{E}[F^{\infty}(t;\Phi^+)]$ rather than by optimizing $\mathbb{E}[\widetilde{F}^{\infty}(t;\Phi^+)]$ with respect to $t$, which is what we do here.

\subsection{Forward recursion policy}

	We defined the forward recursion policy in Section~\ref{sec-network_QDP_forward_recursion} as the policy such that the action at time $t$ is equal to the one that maximizes the quantity $\mathbb{E}[\widetilde{F}^{\pi}(t+1;\psi)]$ at the next time step. Mathematically, this policy is given by the following decision function:
	\begin{equation}\label{eq-forward_greedy_decision_func}
		d_t^{\text{FR}}(h^t)=\left\{\begin{array}{l l} 1 & \text{if }x_t=0, \\ 0 & \text{if }x_t=1\text{ and } f_{M^{\pi}(t)(h^t)+1}(\rho^0)>pf_0(\rho^0), \\ 1 & \text{if }x_t=1\text{ and } f_{M^{\pi}(t)(h^t)+1}(\rho^0)\leq pf_0(\rho^0), \end{array}\right.
	\end{equation}
	for all $t\geq 1$ and all histories $h^t\in\{0,1\}^{2t-1}$. The forward recursion policy is useful because it provides an efficient method for obtaining lower bounds on the maximum expected reward for a quantum decision process; see Section~\ref{sec-QDP_forward_recursion}.
	
	Observe that if $p=1$, then the second condition in \eqref{eq-forward_greedy_decision_func} is always false, because of the fact that $f_m(\rho_{AB}^0)\leq f_0(\rho_{AB}^0)$ for all $m\in\mathbb{N}$ (see \eqref{eq-fid_decay_sats}). Therefore, when $p=1$, we have that $d_t^{\text{FR}}=d_t^0$, i.e., the forward recursion policy is equal to the $t^{\star}=0$ memory-cutoff policy; see \eqref{eq-mem_cutoff_zero}. We now show that the forward recursion policy reduces to a memory-cutoff policy even when $p<1$.
		
	\begin{proposition}\label{prop-forward_recursion_sats}
		Consider satellite-to-ground bipartite elementary link generation, as developed in Section~\ref{sec-elem_link_sats}, with $\overline{n}_1=\overline{n}_2=0$ and $f_S=1$, and let $p\in(0,1)$ be the transmission-heralding success probability, as given by \eqref{eq-trans_succ_prob_sats}. Let $t_{\text{coh}}$ be the coherence time of the quantum memories, as defined in Section~\ref{sec-quantum_memory_sats}. Then, for all $t\geq 1$,
		\begin{equation}
			d_t^{\text{FR}}=\left\{\begin{array}{l l} d_t^{\infty} & \text{if }p\leq\frac{1}{2}, \\[0.2cm] d_t^{t^{\star}} & \text{if }p>\frac{1}{2},\end{array}\right.
		\end{equation}
		where
		\begin{equation}\label{eq-forward_greedy_cutoff_sats_spec}
			t^{\star}=\left\lceil-\frac{t_{\text{coh}}}{2}\ln(2p-1)-1\right\rceil.
		\end{equation}
		In other words, if $p\leq\frac{1}{2}$, then the forward recursion policy is equal to the $t^{\star}=\infty$ memory-cutoff policy; if $p>\frac{1}{2}$, then the forward recursion policy is equal to the $t^{\star}$ memory-cutoff policy, with $t^{\star}$ given by \eqref{eq-forward_greedy_cutoff_sats_spec}.
	\end{proposition}

	\begin{proof}
		Let $m\equiv M(t)(h^t)$. Then, for the state $\rho_{AB}^0$ given by \eqref{eq-initial_link_states}, using \eqref{eq-fid_decay_sats} the second condition in \eqref{eq-forward_greedy_decision_func} translates to 
		\begin{align}
			&\alpha\lambda_{m+1}^2+\left(\beta-\frac{1}{2}\right)\lambda_{m+1}+\frac{1}{2}>p(\alpha+\beta)\label{eq-forward_greedy_condition_sats}\\
			\Rightarrow & p<\frac{\alpha\lambda_{m+1}^2}{\alpha+\beta}+\frac{\left(\beta-\frac{1}{2}\right)\lambda_{m+1}}{\alpha+\beta}+\frac{1}{2(\alpha+\beta)}.\label{eq-forward_greedy_condition_sats2}
		\end{align}
		In the case $\overline{n}_1=\overline{n}_2=0$ and $f_S=1$, we have that $\alpha=\beta=\frac{1}{2}$, so that the inequality in \eqref{eq-forward_greedy_condition_sats2} becomes
		\begin{equation}\label{eq-forward_greedy_condition_sats3}
			p<\frac{1}{2}\left(\e^{-\frac{2(m+1)}{t_{\text{coh}}}}+1\right),
		\end{equation}
		Now, this inequality is satisfied for all $m\in\mathbb{N}_0$ if and only if $p\leq\frac{1}{2}$. In other words, if $p\leq\frac{1}{2}$, then for all possible memory times the action is to wait if the link is currently active, meaning that the decision function in \eqref{eq-forward_greedy_decision_func} becomes
		\begin{equation}
			d_t^{\text{FR}}(h^t)=\left\{\begin{array}{l l} 1 & \text{if }x_t=0, \\ 0 & \text{if }x_t=1, \end{array}\right.
		\end{equation}
		which is precisely the decision function $d_t^{\infty}$ for the $t^{\star}=\infty$ memory-cutoff policy; see \eqref{eq-mem_cutoff_infty_dt_simpler}.
		
		For $p\in\left(\frac{1}{2},1\right)$, whether the inequality in \eqref{eq-forward_greedy_condition_sats3} is satisfied or not depends on the memory time $m$. Consider the largest value of $m$ for which the inequality is satisfied, and denote that value by $m^*$. Since the action is to wait, at the next time step the memory value will be $m^*+1$, which by definition will not satisfy the inequality in \eqref{eq-forward_greedy_condition_sats}. This means that, for all memory times strictly less than $m^*+1$, the forward recursion policy dictates that the agent wait if the elementary link is currently active. As soon as the memory time is equal to $m^*+1$, then the forward recursion policy dictates that the agent request a new elementary link. This means that $m^*+1$ is a memory cutoff value. In particular, by rearranging the inequality in \eqref{eq-forward_greedy_condition_sats3}, we obtain
		\begin{equation}
			m<-\frac{t_{\text{coh}}}{2}\ln(2p-1)-1,
		\end{equation}
		which means that
		\begin{equation}
			m^*=\left\lfloor -\frac{t_{\text{coh}}}{2}\ln(2p-1)-1 \right\rfloor,
		\end{equation}
		and
		\begin{equation}
			t^{\star}=1+m^*=\left\lceil -\frac{t_{\text{coh}}}{2}\ln(2p-1)-1 \right\rceil,
		\end{equation}
		as required.
	\end{proof}
	
	Observe that the cutoff in \eqref{eq-forward_greedy_cutoff_sats_spec} is equal to zero for all $p\geq\frac{1}{2}\left(1+\e^{-\frac{2}{t_{\text{coh}}}}\right)$. This means that $p=1$ is not the only transmission-heralding success probability for which the forward recursion policy is equal to the $t^{\star}=0$ memory-cutoff policy. Intuitively, for $\frac{1}{2}\left(1+\e^{-\frac{2}{t_{\text{coh}}}}\right)\leq p\leq 1$, the transmission-heralding success probability is high enough that it is not necessary to store the quantum state in memory---for the purposes of maximizing the expected value of $\widetilde{F}$, it suffices to request a new quantum state at every time step. At the other extreme, for $0\leq p\leq\frac{1}{2}$, the probability is too low to keep requesting---for the purposes of maximizing the expected value of $\widetilde{F}$, it is better to keep the quantum state in memory indefinitely.
	
	\begin{remark}
		The result of Proposition~\ref{prop-forward_recursion_sats} goes beyond elementary link generation with satellites, because we assumed that $\overline{n}_1=\overline{n}_2=0$ and $f_S=1$. As a result of these assumptions, the result of Proposition~\ref{prop-forward_recursion_sats} applies to every elementary link generation scenario (such as ground-based elementary link generation as described in Section~\ref{sec-network_setup_ground_based}) in which the transmission channel is a pure-loss channel, the heralding procedure is described by \eqref{eq-photonic_heralding_example_1}--\eqref{eq-photonic_heralding_0}, the source state is equal to the target state, and the quantum memories are modeled as in Section~\ref{sec-quantum_memory_sats}.
	\end{remark}

%	[...evaluate for some examples for \eqref{eq-forward_greedy_cutoff_sats_spec}, which holds for $\overline{n}_1=\overline{n}_2=0$ and $f_S=1$: pick a distance $d$ and altitude $h$...this fixes $p=\eta_1\eta_2$...then, based on the distance, take the duration of each time step to be the classical communication time between the ground stations...then obtain $t_{\text{coh}}$ based on this by picking a current coherence time...]

\subsection{Backward recursion policy}

	Finally, to end this chapter, let us consider the backward recursion policy, which we know to be optimal from Theorem~\ref{thm-opt_policy}. However, from the discussion at the end of Section~\ref{sec-network_QDP_backward_recursion}, we know that the backward recursion algorithm presented is exponentially slow in the horizon time, which limits its usefulness in practice. The most common thing to do in practice is to use reinforcement learning algorithms to come up with lower bounds on the optimal reward. In this section, therefore, we simply provide results from the backward recursion algorithm for small horizon times, just as a proof of concept. 
	
	First, in Figure~\ref{fig-opt_pols_general}, we plot optimal values of $\mathbb{E}[\widetilde{F}^{\pi}(t+1;\Phi^+)]$ for a single elementary link for various values of the transmission-heralding success probability $p$ and memory coherence time $t_{\text{coh}}$. We assume that $\overline{n}_1=\overline{n}_2=0$ (no background photons) and $f_S=1$ (source state is $\rho^S=\Phi^+$). For large values of $p$, we find that the optimal value is reached very quickly, within a couple of time steps. For values of $p$ close to one (such as $p=0.9$), we see that for certain values of the coherence time $t_{\text{coh}}$ the initial value of $\mathbb{E}[\widetilde{F}^{\pi}(t+1);\Phi^+]$ (i.e., the value for $t=0$) is optimal. In this case, the optimal policy is to request an elementary link at every time step.
	
	\begin{figure}
		\centering
		\includegraphics[width=0.9\textwidth]{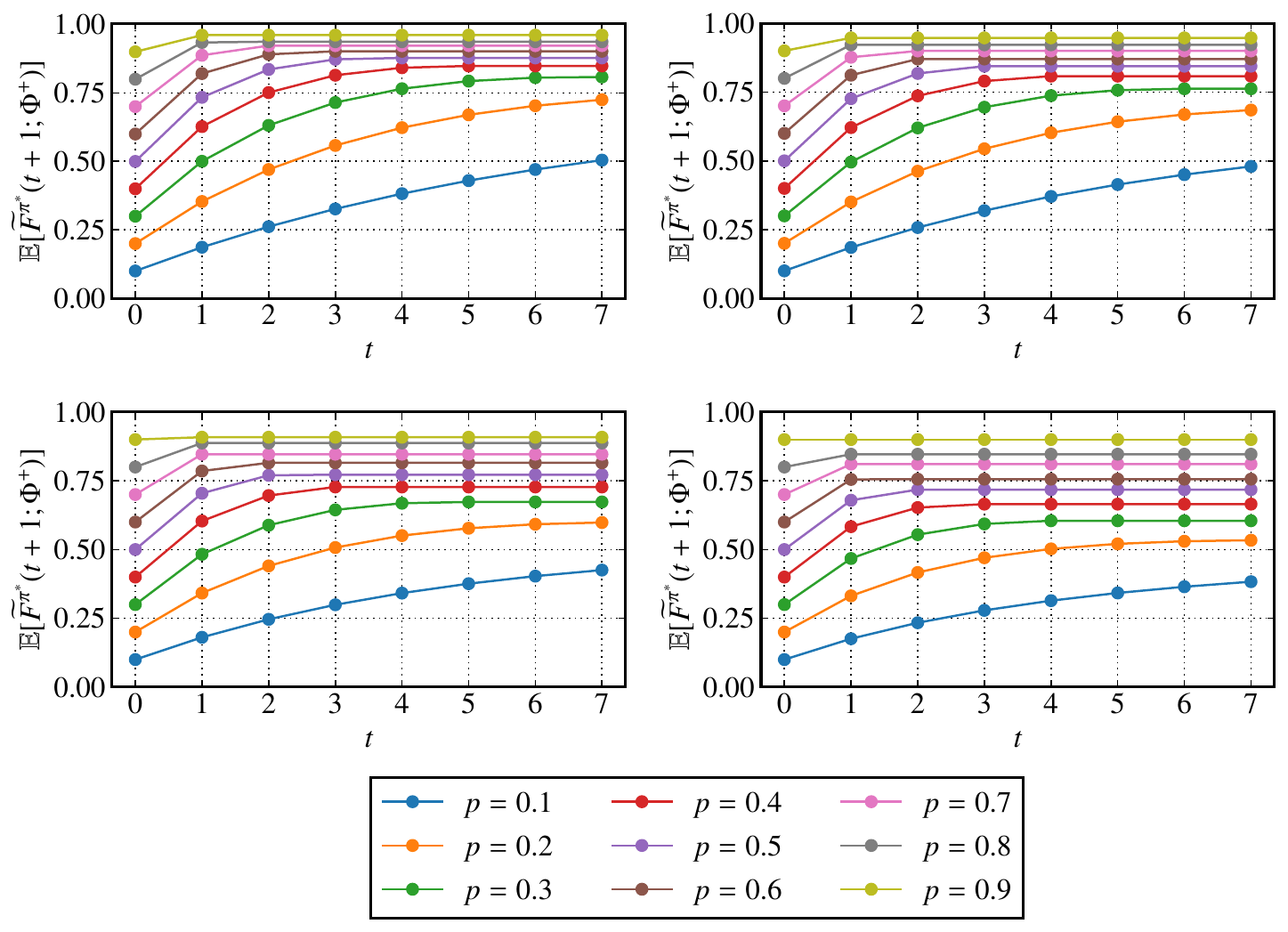}
		\caption{Optimal values of $\mathbb{E}[\widetilde{F}^{\pi}(t+1;\Phi^+)]$ for a single elementary link distributed by a satellite to two ground stations. We take various values for the transmission-heralding success probability $p$, and we assume that $\overline{n}_1=\overline{n}_2=0$ and $f_S=1$. The panels correspond to different values of the memory coherence time: $t_{\text{coh}}=30$ (top left), $t_{\text{coh}}=20$ (top right), $t_{\text{coh}}=10$ (bottom left), and $t_{\text{coh}}=6$ (bottom right).}\label{fig-opt_pols_general}
	\end{figure}
	
	In Figure~\ref{fig-opt_pol_sats}, we similarly plot optimal values of $\mathbb{E}[\widetilde{F}^{\pi}(t+1;\Phi^+)]$ for a single elementary link, except now we plot them as a function of the ground station distance $d$ and the satellite altitude $h$ as per the situation depicted in the left-most panel of Figure~\ref{fig-atmosphere_geometry}. We also plot the expected link values $\mathbb{E}[X(t+1)]_{\pi}$ and the expected fidelities $\mathbb{E}[F^{\pi}(t+1;\Phi^+)]$ associated with the optimal policies. As before, we assume that $f_S=1$, but unlike before we assume that $\overline{n}_1=\overline{n}_2=10^{-4}$, and we consider multiplexing with $M=10^5$ distinct frequency modes per transmission. We assume a coherence time of 1 s throughout. For small distance-altitude pairs, we find that the optimal value is reached within five time steps. For these cases, it is worth pointing out that the optimal value of $\mathbb{E}[\widetilde{F}^{\pi}(t+1;\Phi^+)]$ corresponds to a link activity value $\mathbb{E}[X(t+1)]_{\pi}$ of nearly one, while the fidelity (although it drops, as expected) does not drop significantly, meaning that the elementary link can still be useful for performing entanglement distillation of parallel elementary links or for creating virtual links. It is also interesting to point out that for a ground distance separation of $d=2000$ km, the optimal values for satellite altitude $h=1000$ km is higher than for $h=500$ km. This result can be traced back to the top-left panel of Figure~\ref{fig-initial_fid_prob}, in which we see that the transmission-heralding success probability curves for $h=500$ km and $h=1000$ km cross over at around 1700 km, so that $h=1000$ km has a higher probability than $h=500$ km when $d=2000$ km.
	
	\begin{figure}
		\centering
		\includegraphics[width=0.9\textwidth]{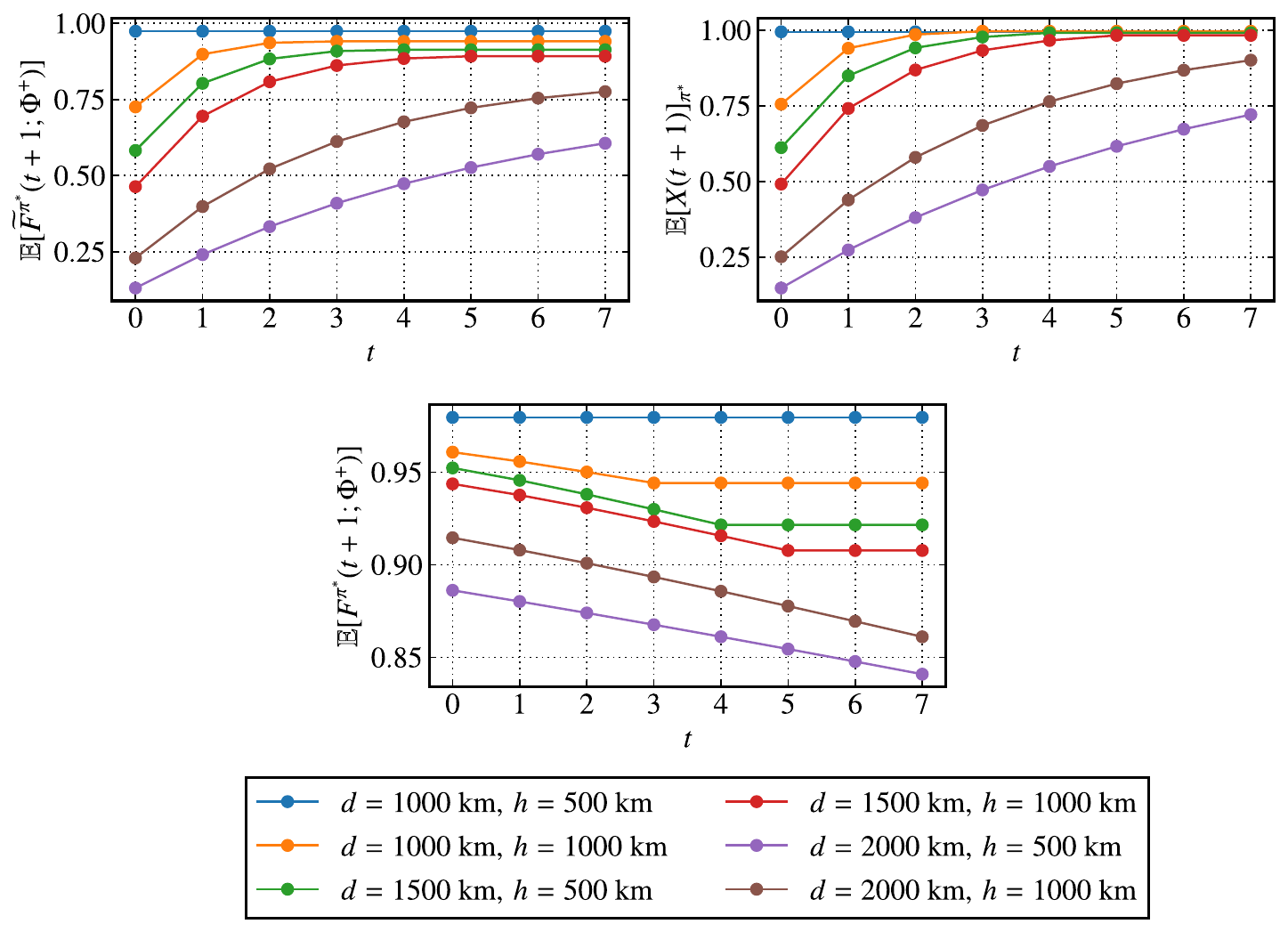}
		\caption{Optimal values of $\mathbb{E}[\widetilde{F}^{\pi}(t+1;\Phi^+)]$, along with the associated link values $\mathbb{E}[X(t+1)]_{\pi}$ and fidelities $\mathbb{E}[F^{\pi}(t+1;\Phi^+)]$, for a single elementary link distributed by a satellite to two ground stations, according to the symmetric situation depicted in the left-most panel of Figure~\ref{fig-atmosphere_geometry}. We assume that $f_S=1$ and that $\overline{n}_1=\overline{n}_2=10^{-4}$, and we assume that the quantum memories have a coherence time of 1 s. We also assume multiplexing with $M=10^5$ distinct frequency modes per transmission.}\label{fig-opt_pol_sats}
	\end{figure}

%	...observation: if the initial probability $p$ is high enough, namely, if
%	\begin{equation}
%		p\geq \frac{\alpha\lambda_1^2}{\alpha+\beta}+\frac{\left(\beta-\frac{1}{2}\right)\lambda_1}{\alpha+\beta}+\frac{1}{2(\alpha+\beta)},
%	\end{equation}
%	then the optimal policy is to just request at each time step...otherwise, ....

\section{Summary}

	In this chapter, we applied all of the concepts developed in this thesis to the example of satellite-to-ground elementary link generation, based on the satellite-to-ground transmission model presented in Ref.~\cite{KBD+19}. We used the explicit satellite-to-ground transmission channel developed in Section~\ref{sec-sat_architecture} to determine the quantum state of the elementary link after successful transmission and heralding. We evaluated the fidelity of this quantum state as a function of ground station separation distance and satellite altitude, and we provided estimates of secret key rates for quantum key distribution. Then, using a particular noise model for the quantum memories at the ground stations, we examined policies for the elementary link generation quantum decision process developed in Chapter~\ref{chap-network_QDP}. We looked first at the memory-cutoff policy and then at the forward recursion policy. Interestingly, for the latter, we found that the forward recursion policy is simply the $t^{\star}$ memory-cutoff policy, with the cutoff $t^{\star}$ depending explicitly on the coherence time of the quantum memories and on the transmission-heralding success probability. Finally, we applied the backward recursion algorithm to determine optimal policies for small horizon times.

\chapter{SUMMARY AND OUTLOOK}\label{chap-summary}

	The central topic of this work is the theory of quantum networks---specifically, how to describe them and how to develop protocols for entanglement distribution in practical scenarios with near-term quantum technologies. The goal in this area of research is to develop protocols that can handle multiple-user requests, work for any given network topology, and can adapt to changes in topology and attacks to the network infrastructure, with the ultimate goal being the realization of the quantum internet. The main message of this work is that reinforcement learning can be used to discover such protocols.

	Before embarking on this ambitious endeavor of using reinforcement learning to discover entanglement distribution protocols in quantum networks, several important elements need to be put into place. First, because reinforcement learning algorithms are defined for decision processes, we require a mathematical framework for general quantum decision processes, because both the agent and the environment in a quantum network protocol can, in principle, be quantum mechanical systems. We also need a systematic language for describing quantum networks and entanglement distribution protocols in general. Finally, by combining the previous two elements, we need a theoretical framework for entanglement distribution protocols in a quantum network based on quantum decision processes. This thesis provides several initial steps in this direction, and thus initial steps towards the goal of reinforcement learning-based quantum network protocols. The following are the main contributions of this thesis.
	\begin{enumerate}
		\item A theory of quantum decision processes in which both the agent and the environment are described by quantum systems (Chapter~\ref{chap-QDP}). This development extends prior work \cite{BBA14,YY18} on quantum partially observable Markov decision processes. Due to the wide applicability of quantum decision processes in quantum information processing tasks (as described in Section~\ref{sec-QDP_q_info_tasks}), the material of this chapter is expected to be of independent interest outside of the quantum network context.
		
		\item A mathematical model for describing quantum networks, as well as quantum states and channels in a quantum network; a theoretical framework for entanglement distribution in a quantum network based on graph transformations, as in Refs.~\cite{SMI+17,CRDW19}; and practical ground- and satellite-based network architectures that take into account some of the limitations of near-term quantum technologies (Chapter~\ref{chap-network_setup}).
		
		\item An explicit entanglement distribution protocol using quantum decision processes to model elementary link generation under practical settings, and corresponding figures of merit (Chapter~\ref{chap-network_QDP}).
	\end{enumerate}
	As a proof of concept for these theoretical contributions, in Chapters~\ref{chap-mem_cutoff} and \ref{chap-sats} we considered explicit examples of policies and elementary link generation scenarios. In Chapter~\ref{chap-mem_cutoff}, we considered the so-called memory-cutoff policy for elementary link generation. Then, in Chapter~\ref{chap-sats}, we applied all of the concepts developed in this thesis to the example of satellite-to-ground elementary link generation.
	
	The majority of the main results of this thesis are analytical in nature, and while some of these results might not be directly applicable to real quantum networks, the true value of the developments of this thesis is in the tools and techniques that have been provided that pave the way for developing more sophisticated quantum network protocols using quantum decision processes, and to computationally obtain optimal policies using reinforcement learning algorithms. Therefore, moving forward, here are some general directions to be explored.
	\begin{itemize}
		\item Incorporate more sophisticated operations into the quantum decision process presented in Definition~\ref{def-network_QDP_elem_link}, such as entanglement distillation and entanglement swapping protocols.
		
		\item Examine quantum decision processes with multiple cooperating agents, which would involve extending the developments of Chapter~\ref{chap-QDP} to multiple agents. Doing so would then also allow for quantum decision processes to be used for developing routing protocols for quantum networks in a manner analogous to the classical case; see, e.g., Refs.~\cite{Bell58,BL93,LB93,Mammeri19,YLX+19}. We also refer to Ref.~\cite{CRDW19} for routing strategies that are closely aligned with the techniques presented in this thesis.
		
		\item Use reinforcement learning algorithms to perform policy optimization, as opposed to the backward and forward recursion algorithms used in the examples considered in this thesis. See Ref.~\cite{Sut18_book} for an introduction to reinforcement learning algorithms.
	\end{itemize}
	We provide details about some of these directions for future work in Appendix~\ref{sec-future_work}.
	
	Quantum networks for entanglement distribution provide us with some of the best opportunities for using near-term quantum technologies, because entanglement is an important and central resource for many quantum information processing tasks, such as quantum teleportation, quantum key distribution, and distributed quantum computing. Devising efficient methods for distributing entanglement to spatially separated parties is therefore an important task. This thesis provides a general framework for devising such protocols with the practical limitations of near-term quantum technologies taken into account, and it represents a small step towards the ultimate goal of a global-scale quantum internet.

%	The reason for this is that decision processes form the theoretical backbone of reinforcement learning, which is a type of machine learning that has been used precisely to find optimal solutions to problems have large parameter spaces, such as the parameter space of quantum network protocols. Furthermore, because reinforcement learning algorithms allow for an agent to discover an optimal policy without explicit prior knowledge of its environment, the framework of quantum decision processes can be used to discover optimal device-independent quantum network protocols.

\end{mainmatter}

%\begin{backmatter}
%\chapter{Summary}
%\fancyhead[L]{\nouppercase \leftmark}
%\end{backmatter}
\cleardoublepage\phantomsection
\begin{appendices}
\cleardoublepage\phantomsection
%\appendixpage
%\noappendicestocpagenum
%\addappheadtotoc
\pagestyle{appendix}
\titleformat{\chapter}[display]{\large\bfseries}{APPENDIX~\thechapter.\vspace{-0.5cm}}{0.5cm}{}
\renewcommand{\appendixname}{APPENDIX}
\addtocontents{toc}{\protect\setcounter{tocdepth}{0}}

\chapter{UPPER BOUNDS FOR QUANTUM DECISION PROCESSES}\label{app-QDP_SDP}

	%\section{SDP upper bounds}

	In this appendix, we derive five upper bounds on the maximum expected reward of quantum decision processes.
	
	%We derive five upper bounds. We start by deriving two upper bounds that are in general just as inefficient to compute (in terms of the horizon time) as the exact optimal reward. However, the techniques used to obtain these bounds are useful for derving the last three upper bounds, which are in general more efficient to compute with respect to the horizon time.
	
	\begin{theorem}[QDP Upper Bound 1]\label{thm-QDP_UB1}
		Let $\mathsf{Q}=(\mathsf{E},\mathsf{A})$ be a quantum decision process, with $\mathsf{A}=(T,\pi_T)$, $T<\infty$, and $\pi_T$ an arbitrary $T$-step policy. Then,
		\begin{equation}
			F(\mathsf{E},(T,\pi_T))\leq F_{\text{UB1}}(\mathsf{E},T),
		\end{equation}
		where
		\begin{multline}
			F_{\text{UB1}}(\mathsf{E},T)\\=|\Omega(T)|\max_{h^T\in\Omega(T)}\lambda_{\max}\left(\sum_{\substack{a_T\in\mathcal{A}\\x_{T+1}\in\mathcal{X}}}\left(\bigotimes_{t=1}^T M_{A_t}^{t;a_t}\right)\Tr\left[\widetilde{\mathcal{R}}_{E_{T+1}}^{T;h^T,a_T,x_{T+1}}\left(\widetilde{\sigma}_{E_{T+1}}^{(\mathsf{E})}(T+1;h^T,a_T,x_{T+1})\right)\right]\right).
		\end{multline}
	\end{theorem}
	
	\begin{proof}
		Let us start by rewriting the expression for $F(\mathsf{E},(T,\pi_T))$ in \eqref{eq-QDP_policy_eval_pf1}:
		\begin{align}
			&F(\mathsf{E},(T,\pi_T))\nonumber\\
			&\quad=\sum_{h^{T+1}\in\Omega(T+1)}\left(\prod_{t=1}^T\Tr\left[M_{A_t}^{t;a_t}\rho_{A_t}^{h_t^{T+1}}\right]\right)\Tr\left[\widetilde{\mathcal{R}}_{E_{T+1}}^{T;h^{T+1}}\left(\widetilde{\sigma}_{E_{T+1}}^{(\mathsf{E})}(T+1;h^{T+1})\right)\right]\\
			&\quad=\sum_{h^{T+1}\in\Omega(T+1)}\Tr\left[\left(M_{A_1}^{t;a_1}\otimes\dotsb\otimes M_{A_T}^{T;a_T}\right)\left(\rho_{A_1}^{h_1^{T+1}}\otimes\dotsb\otimes\rho_{A_T}^{h_T^{T+1}}\right)\otimes\widetilde{\mathcal{R}}_{E_{T+1}}^{T;h^{T+1}}\left(\widetilde{\sigma}_{E_{T+1}}^{(\mathsf{E})}(T+1;h^{T+1})\right)\right].
		\end{align}
		Now, let us define the following operators:
		\begin{align}
			\widehat{\rho}_{H_TA_1^T}(T)&\coloneqq \frac{1}{|\Omega(T)|}\sum_{h^T\in\Omega(T)}\ket{h^T}\bra{h^T}_{H_T}\otimes(\rho_{A_1}^{h_1^T}\otimes\dotsb\otimes\rho_{A_T}^{h^T})\\
			\widehat{M}_{\overline{A}_1^TA_1^T}(T)&\coloneqq \sum_{a_1,\dotsc,a_T\in\mathcal{A}}\ket{a_1,\dotsc,a_T}\bra{a_1,\dotsc,a_T}_{\overline{A}_1\dotsb\overline{A}_T}\otimes(M_{A_1}^{1;a_1}\otimes\dotsb\otimes M_{A_T}^{T:a_T}),\\
			\widehat{\sigma}_{H_{T+1}E_{T+1}}^{(\mathsf{E})}(T)&\coloneqq\sum_{h^{T+1}\in\Omega(T+1)}\ket{h^{T+1}}\bra{h^{T+1}}_{H_{T+1}}\otimes\widetilde{\sigma}_{E_{T+1}}^{(\mathsf{E})}(T+1;h^{T+1}).
		\end{align}
		Then, we have that
		\begin{align}
			&F(\mathsf{E},(T,\pi_T))\nonumber\\
			&\quad=|\Omega(T)|\Tr\left[\widehat{\rho}_{H_TA_1^T}(T)\widehat{M}_{\overline{A}_1^T A_1^T}(T)\widehat{\mathcal{R}}_{H_{T+1}E_{T+1}\to E_{T+1}}^{(\mathsf{E})}(T)\left(\widehat{\sigma}_{H_{T+1}E_{T+1}}^{(\mathsf{E})}(T+1)\right)\right]\\
			&\quad=|\Omega(T)|\Tr\left[\widehat{\rho}_{H_TA_1^T}(T)\Tr_{\overline{A}_TX_{T+1}E_{T+1}}\left[\widehat{M}_{\overline{A}_1^T A_1^T}(T)\widehat{\mathcal{R}}_{H_{T+1}E_{T+1}\to E_{T+1}}^{(\mathsf{E})}(T)\left(\widehat{\sigma}_{H_{T+1}E_{T+1}}^{(\mathsf{E})}(T+1)\right)\right]\right],
		\end{align}
		where we recall the definition of the map $\widehat{\mathcal{R}}_{H_{T+1}E_{T+1}\to E_{T+1}}^{(\mathsf{E})}$ in \eqref{eq-QDP_reward_map_hat}. Now, because $\widehat{\rho}_{H_TA_1^T}(T)$ is a density operator by definition, using \eqref{eq-max_EV_SDP} we obtain
		\begin{equation}
			F(\mathsf{E},(T,\pi_T))\leq|\Omega(T)|\lambda_{\max}\left(\Tr_{\overline{A}_TX_{T+1}E_{T+1}}\left[\widehat{M}_{\overline{A}_1^TA_1^T}(T)\widehat{\mathcal{R}}_{H_{T+1}E_{T+1}\to E_{T+1}}^{(\mathsf{E})}(T)\left(\widehat{\sigma}_{H_{T+1}E_{T+1}}^{(\mathsf{E})}(T+1)\right)\right]\right).
		\end{equation}
		Then,
		\begin{align}
			&\Tr_{\overline{A}_TX_{T+1}E_{T+1}}\left[\widehat{M}_{\overline{A}_1^TA_1^T}(T)\widehat{\mathcal{R}}_{H_{T+1}E_{T+1}\to E_{T+1}}^{(\mathsf{E})}(T)\left(\widehat{\sigma}_{H_{T+1}E_{T+1}}^{(\mathsf{E})}(T+1)\right)\right]\nonumber\\
			&\quad=\sum_{h^{T+1}\in\Omega(T+1)}\ket{h^T}\bra{h^T}_{H_T}\otimes(M_{A_1}^{1;a_1}\otimes\dotsb\otimes M_{A_T}^{T;a_T})\Tr\left[\widetilde{\mathcal{R}}_{E_{T+1}}^{T;h^{T+1}}\left(\widetilde{\sigma}_{E_{T+1}}^{(\mathsf{E})}(T+1;h^T,a_T,x_{T+1})\right)\right]\\
			&\quad=\sum_{h^T\in\Omega(T)}\ket{h^T}\bra{h^T}_{H_T}\otimes\sum_{\substack{a_T\in\mathcal{A}\\x_{T+1}\in\mathcal{X}}}\left(\bigotimes_{t=1}^T M_{A_t}^{t;a_t}\right)\Tr\left[\widetilde{\mathcal{R}}_{E_{T+1}}^{T;h^T,a_T,x_{T+1}}\left(\widetilde{\sigma}_{E_{T+1}}^{(\mathsf{E})}(T+1;h^T,a_T,x_{T+1})\right)\right].
		\end{align}
		The operator on the last line is block diagonal, which means that its largest eigenvalue can be found by searching over the largest eigenvalues of the individual blocks. In other words,
		\begin{multline}
			\lambda_{\max}\left(\Tr_{\overline{A}_TX_{T+1}E_{T+1}}\left[\widehat{M}_{\overline{A}_1^TA_1^T}(T)\widehat{\mathcal{R}}_{H_{T+1}E_{T+1}}(T)\left(\widehat{\sigma}_{H_{T+1}E_{T+1}}^{(\mathsf{E})}(T+1)\right)\right]\right)\\=\max_{h^T\in\Omega(T)}\lambda_{\max}\left(\sum_{\substack{a_T\in\mathcal{A}\\x_{T+1}\in\mathcal{X}}}\left(\bigotimes_{t=1}^T M_{A_t}^{t;a_t}\right)\Tr\left[\widetilde{\mathcal{R}}_{E_{T+1}}^{T;h^T,a_T,x_{T+1}}\left(\widetilde{\sigma}_{E_{T+1}}^{(\mathsf{E})}(T+1;h^T,a_T,x_{T+1})\right)\right]\right),
		\end{multline}
		which gives us the desired result.
	\end{proof}
	
	The upper bound derived in Theorem~\ref{thm-QDP_UB1} involves a search over all histories up to time $T$, which makes calculating it exponentially slow. Furthermore, the bound is in general very loose, particularly because of the prefactor of $|\Omega(T)|$, which we know increases exponentially in $T$. The following upper bound can in general be tighter, although it is still exponentially slow to compute in terms of the horizon time.
	
	\begin{theorem}[QDP Upper Bound 2]\label{thm-QDP_UB2}
		Let $\mathsf{Q}=(\mathsf{E},\mathsf{A})$ be a quantum decision process, with $\mathsf{A}=(T,\pi_T)$, $T<\infty$, and $\pi_T$ an arbitrary $T$-step policy. Then,
		\begin{equation}
			F(\mathsf{E},(T,\pi_T))\leq F_{\text{UB2}}(\mathsf{E},T),
		\end{equation}
		where $F_{\text{UB2}}(\mathsf{E},T)$ is the optimal solution to the following semi-definite program:
		\begin{equation}\label{eq-QDP_UB2_SDP}
			\begin{array}{l l} 
				\text{maximize} & \displaystyle \sum_{\substack{h^T\in\Omega(T)\\a_T\in\mathcal{A}\\x_{T+1}\in\mathcal{X}}} \Tr\left[\left(\widetilde{\mathcal{R}}_{E_{T+1}}^{T;h^T,a_T,x_{T+1}}\circ\widetilde{\mathcal{P}}_{E_0\to E_{T+1}}^{(\mathsf{E}),T;h^T,a_T,x_{T+1}}\right)\left(Z_{E_0}^{T;h^T,a_T}\right)\right] \\[1.5cm] \text{subject to} & \displaystyle Z_{E_0}^{t;h^t,a_t}\geq 0,\quad \sum_{a_t\in\mathcal{A}}Z_{E_0}^{t;h^t,a_t}=Z_{E_0}^{t-1;h^{t-1},a_{t-1}}\quad\forall~T\geq t\geq 2,\,h^t\in\Omega(t),\, a_t\in\mathcal{A},\\[0.7cm] & \displaystyle Z_{E_0}^{1;x_1,a_1}\geq 0,\quad \sum_{a_1\in\mathcal{A}}Z_{E_0}^{1;x_1,a_1}=\sigma_{E_0}\quad\forall x_1\in\mathcal{X},\, a_1\in\mathcal{A}.
			\end{array}
		\end{equation}
	\end{theorem}
	
	\begin{proof}
		We start by recalling that
		\begin{equation}
			F(\mathsf{E},(T,\pi_T))=\sum_{h^{T+1}\in\Omega(T+1)}\Tr\left[\widetilde{\mathcal{R}}_{E_{T+1}}^{T;h^{T+1}}\left(\widetilde{\sigma}_{E_{T+1}}^{(\mathsf{E},(T,\pi_T))}(T+1;h^{T+1})\right)\right],
		\end{equation}
		and that
		\begin{align}
			\widetilde{\sigma}_{E_{T+1}}^{(\mathsf{E},(T,\pi_T))}(T+1;h^{T+1})&=\left(\prod_{t=1}^T\Tr\left[M_{A_t}^{t;a_t}\rho_{A_t}^{h_t^{T+1}}\right]\right)\widetilde{\sigma}_{E_{T+1}}^{(\mathsf{E})}(T+1;h^{t+1})\\
			&=\left(\prod_{t=1}^T\Tr\left[M_{A_t}^{t;a_t}\rho_{A_t}^{h_t^{T+1}}\right]\right)\widetilde{\mathcal{P}}_{E_0\to E_{T+1}}^{(\mathsf{E}),T;h^{T+1}}(\sigma_{E_0}).
		\end{align}
		Now, let us define the following operators:
		\begin{equation}\label{eq-QDP_UB2_pf1}
			Z_{E_0}^{t;h^t,a_t}\coloneqq\left(\prod_{j=1}^t\Tr\left[M_{A_j}^{j;a_j}\rho_{A_j}^{h_j^t}\right]\right)\sigma_{E_0}\quad\forall~1\leq t\leq T.
		\end{equation}
		These operators are positive semi-definite by definition. Furthermore, because $\{M_{A_t}^{t;a_t}\}_{a_t\in\mathcal{A}}$ is POVM for all $1\leq t\leq T$, so that $\sum_{a_t\in\mathcal{A}}M_{A_t}^{t;a_t}=\mathbbm{1}_{A_t}$, we obtain
		\begin{equation}
			\sum_{a_t\in\mathcal{A}}Z_{E_0}^{t;h^t,a_t}=\left(\sum_{a_t\in\mathcal{A}}\Tr\left[M_{A_t}^{t;a_t}\rho_{A_t}^{h^t}\right]\right)\left(\prod_{j=1}^{t-1}\Tr\left[M_{A_j}^{j;a_j}\rho_{A_j}^{h^t_j}\right]\right)\sigma_{E_0}=Z_{E_0}^{t-1;h^t_{t-1},a_{t-1}},
		\end{equation}
		for all $2\leq t\leq T$. For $t=1$, we have
		\begin{equation}
			\sum_{a_1\in\mathcal{A}} Z_{E_0}^{1;x_1,a_1}=\left(\sum_{a_1\in\mathcal{A}_1}\Tr\left[M_{A_1}^{1;a_1}\rho_{A_1}^{x_1}\right]\right)\sigma_{E_0}=\sigma_{E_0}.
		\end{equation}
		The operators $Z_{E_0}^{t;h^t,a_t}$ as defined in \eqref{eq-QDP_UB2_pf1} are therefore a feasible choice for the variables in the SDP in \eqref{eq-QDP_UB2_SDP}. Furthermore, because
		\begin{equation}
			\widetilde{\sigma}_{E_{T+1}}^{(\mathsf{E},(T,\pi_T))}(T+1;h^{T+1})=\widetilde{\mathcal{P}}_{E_0\to E_{T+1}}^{(\mathsf{E}),T;h^{T+1}}\left(Z_{E_0}^{T;h_T^{T+1},a_T}\right),
		\end{equation}
		we have that
		\begin{align}
			F(\mathsf{E},(T,\pi_T))&=\sum_{h^{T+1}\in\Omega(T+1)}\Tr\left[\left(\widetilde{\mathcal{R}}_{E_{T+1}}^{T;h^{T+1}}\circ\widetilde{\mathcal{P}}_{E_0\to E_{T+1}}^{(\mathsf{E}),T;h^{T+1}}\right)\left(Z_{E_0}^{T;h^{T+1}_T,a_T}\right)\right]\\
			&=\sum_{\substack{h^T\in\Omega(T)\\a_T\in\mathcal{A}\\x_{T+1}\in\mathcal{X}}}\Tr\left[\left(\widetilde{\mathcal{R}}_{E_{T+1}}^{T;h^T,a_T,x_{T+1}}\circ\widetilde{\mathcal{P}}_{E_0\to E_{T+1}}^{(\mathsf{E}),T;h^T,a_T,x_{T+1}}\right)\left(Z_{E_0}^{T;h^T,a_T}\right)\right].
		\end{align}
		We can thus conclude that the optimal value of the SDP in \eqref{eq-QDP_UB2_SDP} is bounded from below by $F(\mathsf{E},(T,\pi_T))$, as required.
	\end{proof}
	
	Observe that the number of variables in the SDP in \eqref{eq-QDP_UB2_SDP} that defines the upper bound $F_{\text{UB2}}$ is
	\begin{equation}
		\sum_{t=1}^T |\Omega(t)||\mathcal{A}|=\sum_{t=1}^T (|\mathcal{X}||\mathcal{A}|)^t=\frac{|\mathcal{X}||\mathcal{A}|\left(1-\left(|\mathcal{X}||\mathcal{A}|\right)^T\right)}{1-|\mathcal{X}||\mathcal{A}|},
	\end{equation}
	where each term in the sum arises because there are $|\Omega(t)||\mathcal{A}|$ operators of the form $Z^{t;h^t,a_t}$. Therefore, although $F_{\text{UB2}}$ can be computed using a semi-definite program, it is not very efficient due to the extremely high number of variables. Similarly, the upper bound $F_{\text{UB1}}$ is not very efficient, because it involves a search over all histories in $\Omega(T)$. Calculating these two upper bounds is therefore essentially just as time-efficient as computing the exact reward, which is to say not very efficient.
	
	In order to obtain efficiently computable upper bounds, we now use the technique described in \cite[Section~4.3]{VW16}. This technique, along with some simplifying assumptions on the transition and reward maps of the environment, leads to simple, efficient efficient upper bounds that are computable using semi-definite programs (SDPs).
	
	\begin{figure}
		\centering
		\includegraphics[scale=1.25]{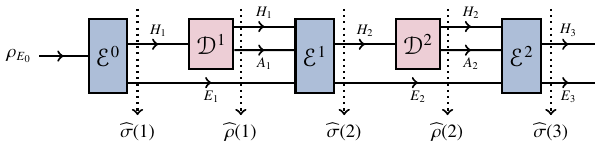}
		\caption{A quantum decision process up to horizon time $T=2$. In addition to the classical-quantum states $\widehat{\sigma}(t)$ that we have already defined in \eqref{eq-QDP_cq_state_0}, we define the classical-quantum states $\widehat{\rho}(t)$ immediately after the decision channels.}\label{fig-QDP_alt}
	\end{figure}
	
	In Figure~\ref{fig-QDP_alt}, we have shown a QDP with horizon time $T=2$. As pointed out in \cite[Section~4.3]{VW16}, the following constraints should be satisfied by the states $\widehat{\sigma}(t)$ and $\widehat{\rho}(t)$ before and after the decision channel $\mathcal{D}^t$ for all $1\leq t\leq T$:
	\begin{equation}\label{eq-QDP_intermediate_state_constraints}
		\Tr_{H_tA_t}\left[\widehat{\rho}_{H_tA_tE_t}(t)\right]=\Tr_{H_t}\left[\widehat{\sigma}_{H_tE_t}(t)\right]\quad\forall~1\leq t\leq T.
	\end{equation}
	These constraints arise from the simple fact that the reduced state of system $E_t$ obtained from $\widehat{\rho}_{H_tA_tE_t}(t)$ and $\widehat{\sigma}_{H_tE_t}(t)$ should be unchanged because the decision channel $\mathcal{D}^t$ does not act on it, as we can clearly see in Figure~\ref{fig-QDP_alt}. From these constraints, we obtain the following result.
	
	\begin{theorem}[QDP Upper Bounds 3 \& 4]
		Let $\mathsf{Q}=(\mathsf{E},\mathsf{A})$ be a quantum decision process, with $\mathsf{A}=(T,\pi_T)$, $T<\infty$, and $\pi_T$ an arbitrary $T$-step policy. Let the environment $\mathsf{E}$ be such that the reward maps and the transition maps satisfy
		\begin{align}
			\widetilde{\mathcal{R}}_{E_{t+1}}^{t;h^t,a_t,x_{t+1}}&=\widetilde{\mathcal{R}}_{E_{t+1}}^{t;a_t,x_{t+1}}\quad \forall~1\leq t\leq T,\,h^t\in\Omega(t),\label{eq-reward_map_assumption}\\
			\mathcal{T}_{E_t\to E_{t+1}}^{t;h^t,a_t,x_{t+1}}&=\mathcal{T}_{E_t\to E_{t+1}}^{t;a_t,x_{t+1}}\quad\forall~1\leq t\leq T,\,h^t\in\Omega(t).\label{eq-transition_map_assumption}
		\end{align}
		Then,
		\begin{equation}\label{eq-QDP_UB3_ineq}
			F(\mathsf{E},(T,\pi_T))\leq F_{\text{UB3}}(\mathsf{E},T),
		\end{equation}
		where $F_{\text{UB3}}(\mathsf{E},T)$ is the solution to the following optimization problem:
		\begin{equation}\label{eq-QDP_UB3_SDP}
			\begin{array}{l l}
				\text{maximize} & \displaystyle \sum_{\substack{a_T\in\mathcal{A}\\x_{T+1}\in\mathcal{X}}}\Tr\left[\left(\mathcal{M}_{A_T\to\varnothing}^{T;a_T}\otimes\widetilde{\mathcal{R}}_{E_{T+1}}^{T;a_T,x_{T+1}}\circ\mathcal{T}_{E_T\to E_{T+1}}^{T;a_T,x_{T+1}}\right)\left(Y_{A_TE_T}^T\right)\right]\\[1.2cm]
				\text{subject to} & \displaystyle Y_{A_tE_t}^t\geq 0 \quad\forall~T\geq t\geq 1,\\[0.5cm]
				& \displaystyle \Tr_{A_t}\left[Y_{A_tE_t}^t\right]=\sum_{\substack{a_{t-1}\in\mathcal{A}\\x_t\in\mathcal{X}}}\left(\mathcal{M}_{A_{t-1}\to\varnothing}^{t-1;a_{t-1}}\otimes\mathcal{T}_{E_{t-1}\to E_t}^{t-1;a_{t-1},x_t}\right)\left(Y_{A_{t-1}E_{t-1}}^{t-1}\right)\quad\forall~T\geq t\geq 2,\\[1cm]
				& \displaystyle \Tr_{A_1}\left[Y_{A_1E_1}^1\right]=\sum_{x_1\in\mathcal{X}}\mathcal{T}_{E_0\to E_1}^{1;x_1}(\sigma_{E_0}),
			\end{array}
		\end{equation}
		where we recall the definition of $\mathcal{M}_{A_j\to\varnothing}^{j;a_j}$ in \eqref{eq-QDP_agent_meas_map}. We also have
		\begin{equation}\label{eq-QDP_UB4_ineq}
			F(\mathsf{E},(T,\pi_T))\leq F_{\text{UB4}}(\mathsf{E},T),
		\end{equation}
		where $F_{\text{UB4}}(\mathsf{E},T)$ is the solution to the following optimization problem:
		\begin{equation}\label{eq-QDP_UB4_SDP}
			\begin{array}{l l}
				\text{maximize} & \displaystyle \sum_{\substack{a_T\in\mathcal{A}\\x_{T+1}\in\mathcal{X}}}\Tr\left[\left(\widetilde{\mathcal{R}}_{E_{T+1}}^{T;a_T,x_{T+1}}\circ\mathcal{T}_{E_T\to E_{T+1}}^{T;a_T,x_{T+1}}\right)\left(V_{E_T}^{T;a_T}\right)\right] \\ [1.2cm]
				\text{subject to} & \displaystyle V_{E_t}^{t;a_t}\geq 0\quad\forall~T\geq t\geq 1,\\[0.5cm]
				& \displaystyle \sum_{a_t\in\mathcal{A}_t}V_{E_t}^{t;a_t}=\sum_{\substack{a_{t-1}\in\mathcal{A}\\x_t\in\mathcal{X}}}\mathcal{T}_{E_{t-1}\to E_t}^{t-1;a_{t-1},x_t}\left(V_{E_{t-1}}^{t-1;a_{t-1}}\right)\quad\forall~T\geq t\geq 2,\\[1cm]
				& \displaystyle \sum_{a_1\in\mathcal{A}}V_{E_1}^{1;a_1}=\sum_{x_1\in\mathcal{X}}\mathcal{T}_{E_0\to E_1}^{0;x_1}(\sigma_{E_0}).
			\end{array}
		\end{equation}
	\end{theorem}
	
	\begin{proof}
		We start by proving the inequality in \eqref{eq-QDP_UB3_ineq}. First, notice that for $2\leq t\leq T$, the condition in \eqref{eq-QDP_intermediate_state_constraints} is equivalent to
		\begin{equation}\label{eq-QDP_intermediate_state_constraints_alt}
			\Tr_{H_tA_t}\left[\widehat{\rho}_{H_tA_tE_t}(t)\right]=\Tr_{H_t}\left[\mathcal{E}_{H_{t-1}A_{t-1}E_{t-1}\to H_tE_t}^{t-1}\left(\widehat{\rho}_{H_{t-1}A_{t-1}E_{t-1}}(t-1)\right)\right].
		\end{equation}
		The operators $\widehat{\rho}_{H_tA_tE_t}(t)$ have the form
		\begin{equation}\label{eq-QDP_UB34_pf2}
			\widehat{\rho}_{H_tA_tE_t}(t)=\sum_{h^t\in\Omega(t)}\ket{h^t}\bra{h^t}_{H_t}\otimes\rho_{A_t}^{h^t}\otimes\widetilde{\sigma}_{E_t}^{(\mathsf{E},\mathsf{A})}(t;h^t).
		\end{equation}
		Now, let
		\begin{equation}\label{eq-QDP_UB34_pf1}
			S_{A_tE_t}^{t;h^t}=\rho_{A_t}^{h^t}\otimes\widetilde{\sigma}_{E_t}^{(\mathsf{E},\mathsf{A})}(t;h^t).
		\end{equation}
		for all $1\leq t\leq T$ and $h^t\in\Omega(t)$. Using this, we obtain
		\begin{multline}
			\mathcal{E}_{H_{t-1}A_{t-1}E_{t-1}\to H_tE_t}^{t-1}\left(\widehat{\rho}_{H_{t-1}A_{t-1}E_{t-1}}(t-1)\right)\\=\sum_{\substack{h^{t-1}\in\Omega(t-1)\\a_{t-1}\in\mathcal{A}\\x_t\in\mathcal{X}}}\ket{h^{t-1},a_{t-1},x_t}\bra{h^{t-1},a_{t-1},x_t}_{H_t}\otimes\left(\mathcal{M}_{A_{t-1}\to\varnothing}^{t-1;a_{t-1}}\otimes\mathcal{T}_{E_{t-1}\to E_t}^{t-1;h^{t-1},a_{t-1},x_t}\right)\left(S_{A_{t-1}E_{t-1}}^{t-1;h^{t-1}}\right).
		\end{multline}
		Therefore, the condition in \eqref{eq-QDP_intermediate_state_constraints_alt} implies that
		\begin{align}
			\sum_{h^t\in\Omega(t)}\Tr_{A_t}\left[S_{A_tE_t}^{t;h^t}\right]&=\sum_{\substack{h^{t-1}\in\Omega(t-1)\\a_{t-1}\in\mathcal{A}\\x_t\in\mathcal{X}}}\left(\mathcal{M}_{A_{t-1}\to\varnothing}^{t-1;a_{t-1}}\otimes\mathcal{T}_{E_{t-1}\to E_t}^{t-1;h^{t-1},a_{t-1},x_t}\right)\left(S_{A_{t-1}E_{t-1}}^{t-1;h^{t-1}}\right)\label{eq-QDP_UB34_pf4}\\
			&=\sum_{\substack{a_{t-1}\in\mathcal{A}\\x_t\in\mathcal{X}}}\left(\mathcal{M}_{A_{t-1}\to\varnothing}^{t-1;a_{t-1}}\otimes\mathcal{T}_{E_{t-1}\to E_t}^{t-1;a_{t-1},x_t}\right)\left(\sum_{h^{t-1}\in\Omega(t-1)}S_{A_{t-1}E_{t-1}}^{t-1;h^{t-1}}\right)
		\end{align}
		for all $2\leq t\leq T$, where in the second line we have used the assumption that $\mathcal{T}_{E_{t-1}\to E_t}^{t-1;h^{t-1},a_{t-1},x_t}=\mathcal{T}_{E_{t-1}\to E_t}^{t-1;a_{t-1},x_t}$. For $t=1$, we obtain
		\begin{equation}
			\sum_{x_1\in\mathcal{X}}\Tr_{A_1}\left[S_{A_1E_1}^{1;x_1}\right]=\sum_{x_1\in\mathcal{X}}\mathcal{T}_{E_0\to E_1}^{0;x_1}(\sigma_{E_0}).
		\end{equation}
		Using the fact that
		\begin{equation}
			\widetilde{\sigma}_{E_{T+1}}^{(\mathsf{E},\mathsf{A})}(T+1;h^T,a_T,x_{T+1})=\left(\mathcal{M}_{A_T\to\varnothing}^{T;a_T}\otimes\mathcal{T}_{E_T\to E_{T+1}}^{T;a_T,x_{T+1}}\right)\left(S_{A_TE_T}^{T;h^T}\right),
		\end{equation}
		along with the assumption that $\widetilde{\mathcal{R}}_{E_{T+1}}^{T;h^T,a_T,x_{T+1}}=\widetilde{\mathcal{R}}_{E_{T+1}}^{T;a_T,x_{T+1}}$, we obtain
		\begin{equation}
			F(\mathsf{E},\mathsf{A})=\sum_{\substack{a_T\in\mathcal{A}\\x_{T+1}\in\mathcal{X}}}\Tr\left[\left(\mathcal{M}_{A_T\to\varnothing}^{T;a_T}\otimes\widetilde{\mathcal{R}}_{E_{T+1}}^{T;a_T,x_{T+1}}\circ\mathcal{T}_{E_T\to E_{T+1}}^{T;a_T,x_{T+1}}\right)\left(\sum_{h^T\in\Omega(T)}S_{A_TE_T}^{T;h^T}\right)\right].
		\end{equation}
		Now, observe that in both the objective function and the constraints, only the sums $\sum_{h^t\in\Omega(t)}S_{A_tE_t}^{t;h^t}$ are involved. This means that, for the purpose of optimization, it suffices to consider the variables
		\begin{equation}\label{eq-QDP_UB34_pf6}
			Y_{A_tE_t}^t\coloneqq\sum_{h^t\in\Omega(t)}S_{A_tE_t}^{t;h^t}.
		\end{equation}
		Therefore, we see that the choice in \eqref{eq-QDP_UB34_pf1} leads to the feasible choice in \eqref{eq-QDP_UB34_pf6} for the variables of the SDP in \eqref{eq-QDP_UB3_SDP} that defines $F_{\text{UB3}}(\mathsf{E},T)$. We can thus conclude that $F(\mathsf{E},\mathsf{A})\leq F_{\text{UB3}}(\mathsf{E},T)$, as required.
		
		The proof of \eqref{eq-QDP_UB4_ineq} is similar. For this, we consider the following classical quantum states:
		\begin{align}
			\widehat{\tau}_{H_t\overline{A}_tE_t}(t)&\coloneqq\mathcal{M}_{A_t\to\overline{A}_t}^t(\widehat{\rho}_{H_tA_tE_t}(t))\\
			&=\sum_{h^t\in\Omega(t)}\sum_{a_t\in\mathcal{A}}\ket{h^t}\bra{h^t}_{H_t}\otimes\ket{a_t}\bra{a_t}_{\overline{A}_t}\otimes\Tr\left[M_{A_t}^{t;a_t}\rho_{A_t}^{h^t}\right]\widetilde{\sigma}_{E_t}^{(\mathsf{E},\mathsf{A})}(t;h^t),
		\end{align}
		where $\mathcal{M}_{A_t\to\overline{A}_t}^t(\cdot)\coloneqq\sum_{a_t\in\mathcal{A}}\ket{a_t}\bra{a_t}_{\overline{A}_t}\Tr\left[M_{A_t}^{t;a_t}\rho_{A_t}^{h^t}\right]$, and for the second line we have used \eqref{eq-QDP_UB34_pf2}. Then, because the channel $\mathcal{M}_{A_t\to\overline{A}}^t$ does not act on the system $E_t$, we obtain the following condition analogous to the one in \eqref{eq-QDP_intermediate_state_constraints}:
		\begin{equation}\label{eq-QDP_intermediate_state_constraints_alt2}
			\Tr_{H_t}\left[\widehat{\sigma}_{H_t E_t}(t)\right]=\Tr_{H_t\overline{A}_tE_t}\left[\widehat{\tau}_{H_t\overline{A}_tE_t}(t)\right],
		\end{equation}
		for all $1\leq t\leq T$.

		Now, let
		\begin{equation}\label{eq-QDP_UB34_pf3}
			W_{E_t}^{t;h^t,a_t}\coloneqq \Tr\left[M_{A_t}^{t;a_t}\rho_{A_t}^{h^t}\right]\widetilde{\sigma}_{E_t}^{(\mathsf{E},\mathsf{A})}(t;h^t),
		\end{equation}
		for all $1\leq t\leq T$, all $h^t\in\Omega(t)$, and all $a_t\in\mathcal{A}$. Then, because $\{M_{A_t}^{t;a_t}\}_{a_t\in\mathcal{A}}$ is a POVM, so that $\sum_{a_t\in\mathcal{A}}M_{A_t}^{t;a_t}=\mathbbm{1}_{A_t}$, we have
		\begin{align}
			\sum_{a_t\in\mathcal{A}}W_{E_t}^{t;h^t,a_t}&=\widetilde{\sigma}_{E_t}^{(\mathsf{E},\mathsf{A})}(t;h^t)\\
			&=\Tr\left[M_{A_{t-1}}^{t-1;a_{t-1}}\rho_{A_{t-1}}^{h^t_{t-1}}\right]\mathcal{T}_{E_{t-1}\to E_t}^{t-1;h_{t-1}^t,a_{t-1},x_t}\left(\widetilde{\sigma}_{E_{t-1}}^{(\mathsf{E},\mathsf{A})}(t-1;h_{t-1}^t)\right)\\
			&=\mathcal{T}_{E_{t-1}\to E_t}^{t-1;h_{t-1}^t,a_{t-1},x_t}\left(W_{E_{t-1}}^{t-1;h_{t-1}^t,a_{t-1}}\right),
		\end{align}
		which holds for all $2\leq t\leq T$. This condition, along with the one in \eqref{eq-QDP_intermediate_state_constraints_alt2}, leads to
		\begin{align}
			\sum_{h^t\in\Omega(t)}\sum_{a_t\in\mathcal{A}}W_{E_t}^{t;h^t,a_t}&=\sum_{h^{t-1}\in\Omega(t-1)}\sum_{\substack{a_{t-1}\in\mathcal{A}\\x_t\in\mathcal{X}}}\mathcal{T}_{E_{t-1}\to E_t}^{t-1;a_{t-1},x_t}\left(W_{E_{t-1}}^{t-1;h^{t-1},a_{t-1}}\right)\\
			&=\sum_{\substack{a_{t-1}\in\mathcal{A}\\x_t\in\mathcal{X}}}\mathcal{T}_{E_{t-1}\to E_t}^{t-1;a_{t-1},x_t}\left(\sum_{h^{t-1}\in\Omega(t-1)}W_{E_{t-1}}^{t-1;h^{t-1},a_{t-1}}\right),
		\end{align}
		for all $2\leq t\leq T$, where we have used the assumption $\mathcal{T}_{E_{t-1}\to E_t}^{t-1;h^{t-1},a_{t-1},x_t}=\mathcal{T}_{E_{t-1}\to E_t}^{t-1;a_{t-1},x_t}$. For $t=1$, we obtain
		\begin{equation}
			\sum_{x_1\in\mathcal{X}}\sum_{a_1\in\mathcal{A}}W_{E_1}^{1;x_1,a_1}=\sum_{x_1\in\mathcal{X}}\mathcal{T}_{E_0\to E_1}^{0;x_1}(\sigma_{E_0}).
		\end{equation}
		Then, using the fact that
		\begin{equation}
			\widetilde{\sigma}_{E_{T+1}}^{(\mathsf{E},\mathsf{A})}(T+1;h^T,a_T,x_{T+1})=\mathcal{T}_{E_T\to E_{T+1}}^{T;a_T,x_{T+1}}\left(W_{E_T}^{T;h^T,a_T}\right),
		\end{equation}
		along with the assumption that $\widetilde{\mathcal{R}}_{E_{T+1}}^{T;h^T,a_T,x_{T+1}}=\widetilde{\mathcal{R}}_{E_{T+1}}^{T;a_T,x_{T+1}}$, we obtain
		\begin{align}
			F(\mathsf{E},\mathsf{A})&=\sum_{h^T\in\Omega(T)}\sum_{\substack{a_T\in\mathcal{A}\\x_{T+1}\in\mathcal{X}}}\Tr\left[\left(\widetilde{\mathcal{R}}_{E_{T+1}}^{T;a_T,x_{T+1}}\circ\mathcal{T}_{E_T\to E_{T+1}}^{T;a_T,x_{T+1}}\right)\left(W_{E_T}^{t;h^T}\right)\right]\\
			&=\sum_{\substack{a_T\in\mathcal{A}\\x_{T+1}\in\mathcal{X}}}\Tr\left[\left(\widetilde{\mathcal{R}}_{E_{T+1}}^{T;a_T,x_{T+1}}\circ\mathcal{T}_{E_T\to E_{T+1}}^{T;a_T,x_{T+1}}\right)\left(\sum_{h^T\in\Omega(T)}W_{E_T}^{T;h^T,a_T}\right)\right].
		\end{align}
		Now, observe that in both the objective function and in the constraints, only the sums $\sum_{h^t\in\Omega(t)}W_{E_t}^{t;h^t,a_t}$ are involved. This means that, for the purpose of optimization, it suffices to consider the variables
		\begin{equation}\label{eq-QDP_UB34_pf5}
			V_{E_t}^{t;a_t}\coloneqq\sum_{h^t\in\Omega(t)} W_{E_t}^{t;h^t,a_t}.
		\end{equation}
		Therefore, we see that the choice in \eqref{eq-QDP_UB34_pf3} leads to the feasible choice in \eqref{eq-QDP_UB34_pf5} for the variables of the SDP in \eqref{eq-QDP_UB4_SDP} that defines $F_{\text{UB4}}(\mathsf{E},T)$. We can thus conclude that $F(\mathsf{E},\mathsf{A})\leq F_{\text{UB4}}(\mathsf{E},T)$, as required.
	\end{proof}
	
	Observe that there are only $T$ variables in the SDP in \eqref{eq-QDP_UB3_SDP}, and there are $T|\mathcal{A}|$ variables in the SDP in \eqref{eq-QDP_UB4_SDP}. The SDPs in \eqref{eq-QDP_UB3_SDP} and \eqref{eq-QDP_UB4_SDP} thus provide upper bounds that are efficient to compute relative to the first two upper bounds considered in this section; however, we note that \eqref{eq-QDP_UB3_SDP} and \eqref{eq-QDP_UB4_SDP} are upper bounds in the case that the reward maps and transition maps satisfy \eqref{eq-reward_map_assumption} and \eqref{eq-transition_map_assumption}, respectively. 
	
	Let us derive one more upper bound that applies when the assumptions in \eqref{eq-reward_map_assumption} and \eqref{eq-transition_map_assumption} are relaxed slightly.
	
	\begin{theorem}[QDP Upper Bound 5]
		Let $\mathsf{Q}=(\mathsf{E},\mathsf{A})$ be an arbitrary quantum decision process, with $\mathsf{A}=(T,\pi_T)$, $T<\infty$, and $\pi_T$ an arbitrary $T$-step policy. Let the environment $\mathsf{E}$ be such that the reward maps and the transition maps satisfy
		\begin{align}
			\widetilde{\mathcal{R}}_{E_{t+1}}^{t;h^t,a_t,x_{t+1}}&=\widetilde{\mathcal{R}}_{E_{t+1}}^{t;x_t,a_t,x_{t+1}}\quad\forall~1\leq t\leq T,\,h^t\in\Omega(t),\\
			\mathcal{T}_{E_t\to E_{t+1}}^{t;h^t,a_t,x_{t+1}}&=\mathcal{T}_{E_t\to E_{t+1}}^{t;x_t,a_t,x_{t+1}}\quad\forall~1\leq t\leq T,\, h^t\in\Omega(t).
		\end{align}
		Then,
		\begin{equation}
			F(\mathsf{E},(T,\pi_T))\leq F_{\text{UB5}}(\mathsf{E},T),
		\end{equation}
		where $F_{\text{UB5}}(\mathsf{E},T)$ is the solution to the following semi-definite program:
		\begin{equation}\label{eq-QDP_UB5_SDP}
			\begin{array}{l l}
				\text{maximize} & \displaystyle \sum_{\substack{x_T,x_{T+1}\in\mathcal{X}\\a_T\in\mathcal{A}}}\Tr\left[\left(\mathcal{M}_{A_T\to\varnothing}^{T;a_T}\otimes\widetilde{\mathcal{R}}_{E_{T+1}}^{T;a_T,x_{T+1}}\circ\mathcal{T}_{E_T\to E_{T+1}}^{T;x_T,a_T,x_{T+1}}\right)\left(K_{A_TE_T}^{T;x_T}\right)\right]\\[1.3cm]
				\text{subject to} & \displaystyle K_{A_tE_t}^{t;x_t}\geq 0 \quad\forall~T\geq t\geq 1,\,x_t\in\mathcal{X},\\[0.5cm]
				& \displaystyle \sum_{x_t\in\mathcal{X}}\Tr_{A_t}\left[K_{A_tE_t}^{t;x_t}\right]=\sum_{\substack{x_{t-1},x_t\in\mathcal{X}\\a_{t-1}\in\mathcal{A}}}\left(\mathcal{M}_{A_{t-1}\to\varnothing}^{t-1;a_{t-1}}\otimes\mathcal{T}_{E_{t-1}\to E_t}^{t-1;x_{t-1},a_{t-1},x_t}\right)\left(K_{A_{t-1}E_{t-1}}^{t-1;x_{t-1}}\right)\quad\forall~T\geq t\geq 2,\\[1.1cm]
				& \displaystyle \sum_{x_1\in\mathcal{X}}\Tr_{A_1}\left[K_{A_1E_1}^{1;x_1}\right]=\sum_{x_1\in\mathcal{X}}\mathcal{T}_{E_0\to E_1}^{1;x_1}(\sigma_{E_0}).
			\end{array}
		\end{equation}
	\end{theorem}
	
	\begin{remark}
		Observe that the SDP in \eqref{eq-QDP_UB5_SDP} has $T|\mathcal{X}|$ variables.~\defqed
	\end{remark}
	
	\begin{proof}
		The proof is similar to the proof of \eqref{eq-QDP_UB3_ineq}. Let us recall \eqref{eq-QDP_UB34_pf4} from that proof:
		\begin{equation}\label{eq-QDP_UB5_pf0}
			\sum_{h^t\in\Omega(t)}\Tr_{A_t}\left[S_{A_tE_t}^{t;h^t}\right]=\sum_{\substack{h^{t-1}\in\Omega(t-1)\\a_{t-1}\in\mathcal{A}\\x_t\in\mathcal{X}}}\left(\mathcal{M}_{A_{t-1}\to\varnothing}^{t-1;a_{t-1}}\otimes\mathcal{T}_{E_{t-1}\to E_t}^{t-1;h^{t-1},a_{t-1},x_t}\right)\left(S_{A_{t-1}E_{t-1}}^{t-1;h^{t-1}}\right),
		\end{equation}
		where
		\begin{equation}\label{eq-QDP_UB5_pf1}
			S_{A_tE_t}^{t;h^t}=\rho_{A_t}^{h^t}\otimes\widetilde{\sigma}_{E_t}^{(\mathsf{E},\mathsf{A})}(t;h^t)
		\end{equation}
		for all $1\leq t\leq T$ and all $h^t\in \Omega(t)$. Now, by assumption, $\mathcal{T}_{E_{t-1}\to E_t}^{t-1;h^{t-1},a_{t-1},x_t}=\mathcal{T}_{E_{t-1}\to E_t}^{t-1;x_{t-1},a_{t-1},x_t}$, which means that we can write the \eqref{eq-QDP_UB5_pf0} as
		\begin{multline}
			\sum_{x_t\in\mathcal{X}}\Tr_{A_t}\left[\sum_{\substack{h^{t-1}\in\Omega(t-1)\\a_{t-1}\in\mathcal{A}}}S_{A_tE_t}^{t;h^{t-1},a_{t-1},x_t}\right]\\=\sum_{\substack{x_{t-1},x_t\in\mathcal{X}\\a_{t-1}\in\mathcal{A}}}\left(\mathcal{M}_{A_{t-1}\to\varnothing}^{t-1;a_{t-1}}\otimes\mathcal{T}_{E_{t-1}\to E_t}^{t-1;x_{t-1},a_{t-1},x_t}\right)\left(\sum_{\substack{h^{t-2}\in\Omega(t-1)\\a_{t-2}\in\mathcal{A}}}S_{A_{t-1}E_{t-1}}^{t-1;h^{t-2},a_{t-2},x_{t-1}}\right),
		\end{multline}
		which holds for all $2\leq t\leq T$. For $t=1$, we have
		\begin{equation}
			\sum_{x_1\in\mathcal{X}}\Tr_{A_1}\left[S_{A_1E_1}^{1;x_1}\right]=\sum_{x_1\in\mathcal{X}}\mathcal{T}_{E_0\to E_1}^{0;x_1}(\sigma_{E_0}).
		\end{equation}
		Furthermore, using the assumption that $\widetilde{\mathcal{R}}_{E_{t+1}}^{t;h^t,a_t,x_{t+1}}=\widetilde{\mathcal{R}}_{E_{t+1}}^{t;x_t,a_t,x_{t+1}}$, we have that
		\begin{equation}
			F(\mathsf{E},\mathsf{A})=\sum_{\substack{x_T,x_{T+1}\in\mathcal{X}\\a_T\in\mathcal{A}}}\Tr\left[\left(\mathcal{M}_{A_T\to\varnothing}^{T;a_T}\otimes\widetilde{\mathcal{R}}_{E_{T+1}}^{T;x_T,a_T,x_{T+1}}\circ\mathcal{T}_{E_T\to E_{T+1}}^{T;x_T,a_T,x_{T+1}}\right)\left(\sum_{\substack{h^{T-1}\in\Omega(T-1)\\a_{T-1}\in\mathcal{A}}}S_{A_TE_T}^{T;h^{T-1},a_{T-1},x_{T-1}}\right)\right].
		\end{equation}
		Now, in both the objective function and in the constraints, only the sums $\sum_{\substack{h^{t-1}\in\Omega(t-1)\\a_{t-1}\in\mathcal{A}}}S_{A_tE_t}^{t;h^{t-1},a_{t-1},x_t}$ are involved. This means that, for the purpose of optimization, it suffices to consider the variables
		\begin{equation}\label{eq-QDP_UB5_pf2}
			K_{A_tE_t}^{t;x_t}\coloneqq\sum_{\substack{h^{t-1}\in\Omega(t-1)\\a_{t-1}\in\mathcal{A}}} S_{A_tE_t}^{t;h^{t-1},a_{t-1},x_t}.
		\end{equation}
		Therefore, the choice in \eqref{eq-QDP_UB5_pf1} leads to the feasible choice in \eqref{eq-QDP_UB5_pf2} for the variables of the SDP in \eqref{eq-QDP_UB5_SDP} that defines $F_{\text{UB5}}(\mathsf{E},T)$. We can thus conclude that $F(\mathsf{E},\mathsf{A})\leq F_{\text{UB5}}(\mathsf{E},T)$, as required.
	\end{proof}
	
	%In Chapter~\ref{chap-sats}, we apply all five of the upper bounds derived in this section to an example of entanglement distribution using satellites.

\chapter{DETAILS OF THE MEMORY-CUTOFF POLICY}\label{app-mem_cutoff_details}

	In this appendix, we go through the details of the memory-cutoff policy.
	
	Let us start by considering what the histories $h^t$ look like through an particular example. Consider an elementary link for which $t^{\star}=3$, and let us consider the values of the elementary link up to time $t=10$. Given that each elementary link request succeeds with probability $p$ and fails with probability $1-p$, in Table~\ref{table-link_cutoff_example} we write down the probability for each sequence of elementary link values according to the formula in \eqref{eq-hist_prob_general}. Note that we only include those histories that have non-zero probability (indeed, some sequences $h^t=(x_1,a_1,\dotsc,a_{t-1},x_t)\in\{0,1\}^{2t-1}=\Omega(t)$ will have zero probability under the memory-cutoff policy). We also include in the table the memory times $M^{t^{\star}}(t)$, which are calculated using the formula in \eqref{eq-mem_time_cutoff_policy}. Since the memory-cutoff policy is deterministic, it suffices to keep track only of the elementary link values $(x_1,\dotsc,x_t)$ and not of the action values, because the action values are given deterministically by the elementary link values. For the elementary link value sequences, we define two quantities that are helpful for obtaining analytic formulas for the figures of merit defined in Section~\ref{sec-practical_figures_merit}. The first quantity is $Y_1^{t^{\star}}(t)$, which we define to be the number of full blocks of ones (having length $t^{\star}+1$) in elementary link value sequences up to time $t-1$. The values that $Y_1^{t^{\star}}(t)$ can take are $0,1,\dotsc,\floor{\frac{t-1}{t^{\star}+1}}$ if $t^{\star}<\infty$, and 0 if $t^{\star}=\infty$. We also define the quantity $Y_2^{t^{\star}}(t)$ to be the number of trailing ones in elementary link value sequences up to time $t$. The values that $Y_2^{t^{\star}}(t)$ can take are $0,1,\dotsc,t^{\star}+1$ if $t^{\star}<\infty$, and $0,1,\dotsc,t$ if $t^{\star}=\infty$.
		
	\begin{table}
		\centering
		\caption{Elementary link value sequences $(x_1,x_2,\dotsc,x_{10})$ for an elementary link with $t^{\star}=3$ up to time $t=10$. The quantity $Y_1^{t^{\star}}(t)$ is the number of full blocks of ones in link value sequence up to time $t-1$, and $Y_2^{t^{\star}}(t)$ is the number of trailing ones in elementary link value sequence up to time $t$. $M^{t^{\star}}(t)$ is the memory time at time $t$, given by the formula in \eqref{eq-mem_time_cutoff_policy}.}\label{table-link_cutoff_example}
		\begin{tabular}{|m{0.25cm} m{0.25cm} m{0.25cm} m{0.25cm} m{0.25cm} m{0.25cm} m{0.25cm} m{0.25cm} m{0.25cm} m{0.4cm}||c|c|c|c|c}
			\hline $x_1$ & $x_2$ & $x_3$ & $x_4$ & $x_5$ & $x_6$ & $x_7$ & $x_8$ & $x_9$ & $x_{10}$ & $Y_1^{t^{\star}}(t)(h^t)$ & $Y_2^{t^{\star}}(t)(h^t)$ & $\Pr[H(t)=h^t]_{t^{\star}}$ & $M^{t^{\star}}(t)(h^t)$ \\ \hline\hline
			0 & 0 & 0 & 0 & 0 & 0 & 0 & 0 & 0 & 0 & 0 & 0 & $(1-p)^{10}$ & 3  \\ \hline
			1 & 1 & 1 & 1 & 0 & 0 & 0 & 0 & 0 & 0 & 1 & 0 & $p(1-p)^6$ & 3  \\
			0 & 1 & 1 & 1 & 1 & 0 & 0 & 0 & 0 & 0 & 1 & 0 & $p(1-p)^6$ & 3  \\
			0 & 0 & 1 & 1 & 1 & 1 & 0 & 0 & 0 & 0 & 1 & 0 & $p(1-p)^6$ & 3  \\
			0 & 0 & 0 & 1 & 1 & 1 & 1 & 0 & 0 & 0 & 1 & 0 & $p(1-p)^6$ & 3  \\
			0 & 0 & 0 & 0 & 1 & 1 & 1 & 1 & 0 & 0 & 1 & 0 & $p(1-p)^6$ & 3  \\
			0 & 0 & 0 & 0 & 0 & 1 & 1 & 1 & 1 & 0 & 1 & 0 & $p(1-p)^6$ & 3  \\ \hline
			1 & 1 & 1 & 1 & 1 & 1 & 1 & 1 & 0 & 0 & 2 & 0 & $p^2(1-p)^2$ & 3  \\
			1 & 1 & 1 & 1 & 0 & 1 & 1 & 1 & 1 & 0 & 2 & 0 & $p^2(1-p)^2$ & 3  \\
			0 & 1 & 1 & 1 & 1 & 1 & 1 & 1 & 1 & 0 & 2 & 0 & $p^2(1-p)^2$ & 3  \\ \hline
			0 & 0 & 0 & 0 & 0 & 0 & 0 & 0 & 0 & 1 & 0 & 1 & $p(1-p)^{9}$ & 0  \\
			0 & 0 & 0 & 0 & 0 & 0 & 0 & 0 & 1 & 1 & 0 & 2 & $p(1-p)^{8}$ & 1  \\
			0 & 0 & 0 & 0 & 0 & 0 & 0 & 1 & 1 & 1 & 0 & 3 & $p(1-p)^{7}$ & 2  \\
			0 & 0 & 0 & 0 & 0 & 0 & 1 & 1 & 1 & 1 & 0 & 4 & $p(1-p)^{6}$ & 3  \\ \hline
			1 & 1 & 1 & 1 & 0 & 0 & 0 & 0 & 0 & 1 & 1 & 1 & $p^2(1-p)^5$ & 0  \\
			0 & 1 & 1 & 1 & 1 & 0 & 0 & 0 & 0 & 1 & 1 & 1 & $p^2(1-p)^5$ & 0  \\
			0 & 0 & 1 & 1 & 1 & 1 & 0 & 0 & 0 & 1 & 1 & 1 & $p^2(1-p)^5$ & 0  \\
			0 & 0 & 0 & 1 & 1 & 1 & 1 & 0 & 0 & 1 & 1 & 1 & $p^2(1-p)^5$ & 0 \\
			0 & 0 & 0 & 0 & 1 & 1 & 1 & 1 & 0 & 1 & 1 & 1 & $p^2(1-p)^5$ & 0  \\
			0 & 0 & 0 & 0 & 0 & 1 & 1 & 1 & 1 & 1 & 1 & 1 & $p^2(1-p)^5$ & 0  \\ \hline
			1 & 1 & 1 & 1 & 0 & 0 & 0 & 0 & 1 & 1 & 1 & 2 & $p^2(1-p)^4$ & 1  \\
			0 & 1 & 1 & 1 & 1 & 0 & 0 & 0 & 1 & 1 & 1 & 2 & $p^2(1-p)^4$ & 1  \\
			0 & 0 & 1 & 1 & 1 & 1 & 0 & 0 & 1 & 1 & 1 & 2 & $p^2(1-p)^4$ & 1  \\
			0 & 0 & 0 & 1 & 1 & 1 & 1 & 0 & 1 & 1 & 1 & 2 & $p^2(1-p)^4$ & 1  \\
			0 & 0 & 0 & 0 & 1 & 1 & 1 & 1 & 1 & 1 & 1 & 2 & $p^2(1-p)^4$ & 1  \\ \hline
			1 & 1 & 1 & 1 & 0 & 0 & 0 & 1 & 1 & 1 & 1 & 3 & $p^2(1-p)^3$ & 2  \\
			0 & 1 & 1 & 1 & 1 & 0 & 0 & 1 & 1 & 1 & 1 & 3 & $p^2(1-p)^3$ & 2  \\
			0 & 0 & 1 & 1 & 1 & 1 & 0 & 1 & 1 & 1 & 1 & 3 & $p^2(1-p)^3$ & 2  \\
			0 & 0 & 0 & 1 & 1 & 1 & 1 & 1 & 1 & 1 & 1 & 3 & $p^2(1-p)^3$ & 2  \\ \hline
			1 & 1 & 1 & 1 & 0 & 0 & 1 & 1 & 1 & 1 & 1 & 4 & $p^2(1-p)^2$ & 3  \\
			0 & 1 & 1 & 1 & 1 & 0 & 1 & 1 & 1 & 1 & 1 & 4 & $p^2(1-p)^2$ & 3  \\
			0 & 0 & 1 & 1 & 1 & 1 & 1 & 1 & 1 & 1 & 1 & 4 & $p^2(1-p)^2$ & 3  \\ \hline
			1 & 1 & 1 & 1 & 1 & 1 & 1 & 1 & 0 & 1 & 2 & 1 & $p^3(1-p)$ & 0  \\
			1 & 1 & 1 & 1 & 0 & 1 & 1 & 1 & 1 & 1 & 2 & 1 & $p^3(1-p)$ & 0  \\
			0 & 1 & 1 & 1 & 1 & 1 & 1 & 1 & 1 & 1 & 2 & 1 & $p^3(1-p)$ & 0 \\ \hline
			1 & 1 & 1 & 1 & 1 & 1 & 1 & 1 & 1 & 1 & 2 & 2 & $p^3$ & 1 \\ \hline
		\end{tabular}
	\end{table}
	
	Using the random variables $Y_1^{t^{\star}}(t)$ and $Y_2^{t^{\star}}(t)$, along with the general formula in \eqref{eq-hist_prob_general}, we obtain the following formula for the probability of histories with non-zero probability.
	
	\begin{proposition}\label{prop-time_seq_prob}
		For all $t\geq 1$, $t^{\star}\in[0,\infty)$, $p\in[0,1]$, and histories $h^t=(x_1,a_1,x_2,\allowbreak a_2,\dotsc,a_{t-1},x_t)$ with non-zero probability,
		\begin{multline}\label{eq-time_seq_prob}
			\Pr[H(t)=h^t]_{t^{\star}}=p^{Y_1^{t^{\star}}(t)(h^t)}(1-p)^{t-(t^{\star}+1)Y_1^{t^{\star}}(t)(h^t)}\delta_{Y_2^{t^{\star}}(t)(h^t),0}\\+(1-\delta_{Y_2^{t^{\star}}(t)(h^t),0})p^{Y_1^{t^{\star}}(t)(h^t)+1}(1-p)^{t-Y_2^{t^{\star}}(t)(h^t)-(t^{\star}+1)Y_1^{t^{\star}}(t)(h^t)},
		\end{multline}
		where $Y_1^{t^{\star}}(t)(h^t)$ is defined to be the number of full blocks of ones of length $t^{\star}+1$ up to time $t-1$ in the sequence $(x_1,x_2,\dotsc,x_t)$ of elementary link values, and $Y_2^{t^{\star}}(t)(h^t)$ is defined to be the number of trailing ones in the sequence $(x_1,x_2,\dotsc,x_t)$. For $t^{\star}=\infty$,
		\begin{equation}\label{eq-time_seq_prob_infty}
			\Pr[H(t)=h^t]_{t^{\star}}=(1-p)^t\delta_{Y_2^{t^{\star}}(t)(h^t),0}+(1-\delta_{Y_2^{t^{\star}}(t)(h^t),0})p(1-p)^{t-Y_2^{t^{\star}}(t)(h^t)}.
		\end{equation}
	\end{proposition}
	
	\begin{proof}
		The result in \eqref{eq-time_seq_prob} follows immediately from the formula in \eqref{eq-hist_prob_general} by observing that $N_{\text{succ}}(t)=Y_1^{t^{\star}}(t)+1-\delta_{Y_2^{t^{\star}}(t),0}$ and $N_{\text{req}}(t)=t-(t^{\star}+1)Y_1^{t^{\star}}(t)-Y_2^{t^{\star}}(t)$. For $t^{\star}=\infty$, we only ever have trailing ones in the elementary link value sequences, so that $Y_1(t)^{\infty}(h^t)=0$ for all $t\geq 1$ and all histories $h^t$. The result in \eqref{eq-time_seq_prob_infty} then follows.
	\end{proof}
	
	Next, let us count the number of elementary link value sequences with non-zero probability. Using Table~\ref{table-link_cutoff_example} as a guide, we obtain the following.
	
	\begin{lemma}\label{prop-num_time_seq}
		For all $t\geq 1$ and all $t^{\star}\in[0,\infty]$, let $\Xi(t;t^{\star})$ denote the set of elementary link value sequences for the $t^{\star}$ memory cutoff policy that have non-zero probability. Then, the number of elements in the set $\Xi(t;t^{\star})$ is
		\begin{equation}\label{eq-num_time_seq}
			\abs{\Xi(t;t^{\star})}=\sum_{x=0}^{\floor{\frac{t-1}{t^{\star}+1}}}\sum_{k=0}^{t^{\star}+1}\left(\binom{t-1-xt^{\star}}{x}\delta_{k,0}+(1-\delta_{k,0})\binom{t-k-xt^{\star}}{x}\boldsymbol{1}_{t-k-x(t^{\star}+1)\geq 0}\right).
		\end{equation}
		For $t^{\star}=\infty$, $\abs{\Xi(t;\infty)}=1+t$.
	\end{lemma}
	
	\begin{proof}
		We start by counting the number of elementary link value sequences when the number of trailing ones is equal to zero, i.e., when $k\equiv Y_2^{t^{\star}}(t)(h^t)=0$. If we also let the number $x\equiv Y_1^{t^{\star}}(t)(h^t)$ of full blocks of ones in time $t-1$ be equal to one, then there are $t^{\star}+1$ ones and $t-t^{\star}-2$ zeros up to time $t-1$. The total number of elementary link value sequences is then equal to the number of ways that the single block of ones can be moved around in the elementary link value sequence up to time $t-1$. This quantity is equivalent to the number of permutations of $t-1-t^{\star}$ objects with $t-t^{\star}-2$ of them being identical (these are the zeros), which is given by
		\begin{equation}
			\frac{(t-1-t^{\star})!}{(t-2-t^{\star})!(t-1-t^{\star}-t+t^{\star}+2)!}=\frac{(t-1-t^{\star})!}{(t-t^{\star}-2)!(1)!}=\binom{t-1-t^{\star}}{1}.
		\end{equation}
		We thus have the $x=0$ and $k=0$ term in the sum in \eqref{eq-num_time_seq}. If we stick to $k=0$ but now consider more than one full block of ones in time $t-1$ (i.e., let $x\equiv Y_1^{t^{\star}}(t)(h^t)\geq 1$), then the number of elementary link value sequences is given by a similar argument as before: it is equal to the number of ways of permuting $t-1-xt^{\star}$ objects, with $x$ of them being identical (the blocks of ones) and the remaining $t-1-x(t^{\star}+1)$ objects also identical (the number of zeros), i.e., $\binom{t-1-xt^{\star}}{x}$. The total number of elementary link value sequences with zero trailing ones is therefore
		\begin{equation}\label{eq-num_time_seq_pf1}
			\sum_{x=0}^{\floor{\frac{t-1}{t^{\star}+1}}}\binom{t-1-xt^{\star}}{x}.
		\end{equation}
		
		Let us now consider the case $k\equiv Y_2(t)(h^t)>0$. Then, the number of time slots in which full blocks of ones can be shuffled around is $t-k$. If there are $x$ blocks of ones in time $t-k$, then by the same arguments as before, the number of such elementary link value sequences is given by the number of ways of permuting $t-k-xt^{\star}$ objects, with $x$ of them being identical (the full blocks of ones) and the remaining $t-k-x(t^{\star}+1)$ of them also identical (these are the zeros up to time $t-k$). In other words, the number of link value sequences with $k>0$ and $x\geq 0$ is
		\begin{equation}\label{eq-num_time_seq_pf2}
			\binom{t-k-xt^{\star}}{x}\boldsymbol{1}_{t-k-x(t^{\star}+1)\geq 0}.
		\end{equation}
		We must put the indicator function $\boldsymbol{1}_{t-k-x(t^{\star}+1)\geq 0}$ in order to ensure that the binomial coefficient makes sense. This also means that, depending on the time $t$, not all values of $k$ between 0 and $t^{\star}+1$ can be considered in the total number of elementary link value sequences (simply because it might not be possible to fit all possible values of trailing ones and full blocks of ones within that amount of time). By combining \eqref{eq-num_time_seq_pf1} and \eqref{eq-num_time_seq_pf2}, we obtain the desired result.
		
		In the case $t^{\star}=\infty$, because there are never any full blocks of ones and only trailing ones, we have $t$ link value sequences, each containing $k$ trailing ones, where $1\leq k\leq t$. We also have an elementary link value sequence consisting of all zeros, giving a total of $t+1$ link value sequences.
	\end{proof}
	\smallskip
	\begin{remark}
		Note that when $t^{\star}=0$, we get
		\begin{align}
			\abs{\Xi(t;0)}&=\sum_{x=0}^{t-1}\sum_{k=0}^1\left(\binom{t-1}{x}\delta_{k,0}+(1-\delta_{k,0})\binom{t-k}{x}\boldsymbol{1}_{t-k-x\geq 0}\right)\\
			&=\sum_{x=0}^{t-1}\binom{t-1}{x}+\sum_{x=0}^{t-1}\binom{t-1}{x}\underbrace{\boldsymbol{1}_{t-1-x\geq 0}}_{1~\forall x}\\
			&=2^{t-1}+2^{t-1}\\
			&=2^t.
		\end{align}
		In other words, when $t^{\star}=0$, \textit{all} $t$-bit strings are valid link value sequences. %This makes sense, because when $t^{\star}=0$ the full blocks of ones only have a length of 1.
		
		For $t\leq t^{\star}+1$, no full blocks of ones in time $t-1$ are possible, so we get
		\begin{align}
			\abs{\Xi(t;t^{\star})}&=\sum_{k=0}^{t^{\star}+1}\left(\binom{t-1}{0}\delta_{k,0}+(1-\delta_{k,0})\binom{t-k}{0}\boldsymbol{1}_{t-k\geq 0}\right)\\
			&=\binom{t-1}{0}+\sum_{k=1}^t \binom{t-k}{0}\\
			&=1+t.
		\end{align}
		This coincides with the result for $t^{\star}=\infty$, because when $t^{\star}=\infty$ the condition $t\leq t^{\star}+1$ is satisfied for all $t\geq 1$.~\defqed
	\end{remark}

\section{Proof of Theorem~\ref{thm-mem_status_pr}}\label{sec-mem_status_pr_pf}

	For $t\leq t^{\star}+1$, because no full blocks of ones up to time $t-1$ are possible, the possible values for the memory time are $0,1,\dotsc,t-1$. Furthermore, for each value of $m\in\{0,1,\dotsc,t-1\}$, there is only one elementary link value sequence for which $M^{t^{\star}}(t)=m$, and this sequence has $Y_2^{t^{\star}}(t)=m+1$ trailing ones and thus probability $p(1-p)^{t-1-m}$ by Proposition~\ref{prop-time_seq_prob}.
		
	For $t>t^{\star}+1$, we proceed similarly by considering the number $Y_1^{t^{\star}}(t)$ of full blocks of ones in time $t-1$ and the number $Y_2^{t^{\star}}(t)$ of trailing ones in elementary link value sequences $(x_1,x_2,\dotsc,x_t)$ such that $x_t=1$. Since we must have $x_t=1$, we require $Y_2(t)\geq 1$. Now, in order to have a memory time of $M^{t^{\star}}(t)=m$, we can have elementary link value sequences consisting of any number $x=Y_1^{t^{\star}}(t)$ of full blocks of ones ranging from 0 to $\floor{\frac{t-1}{t^{\star}+1}}$ as long as $Y_2^{t^{\star}}(t)=m+1$. (Note that at the end of each full block of ones the memory time is equal to $t^{\star}$.) The number of such elementary link value sequences is
	\begin{equation}
		\binom{t-(m+1)-xt^{\star}}{x}\boldsymbol{1}_{t-(m+1)-x(t^{\star}+1)\geq 0},
	\end{equation}
	as given by \eqref{eq-num_time_seq_pf2}, and the probability of each such link value sequence is $p^{x+1}(1-p)^{t-(m+1)-x(t^{\star}+1)}$. By summing over all $0\leq x\leq\floor{\frac{t-1}{t^{\star}+1}}$, we obtain the desired result.

\section{Proof of Proposition~\ref{prop-mem_time_prob_x0}}\label{sec-mem_time_prob_x0_pf}

	For finite $t^{\star}$, when $t\leq t^{\star}+1$, there is only one elementary link value sequence ending with a zero, and that is the sequence consisting of all zeros, which has probability $(1-p)^t$. Furthermore, since the value of the memory for this sequence is equal to $t^{\star}$, only the case $M^{t^{\star}}(t)=t^{\star}$ has non-zero probability. When $t>t^{\star}+1$, we can again have non-zero probability only for $M^{t^{\star}}(t)=t^{\star}$. In this case, because every link value sequence has to end with a zero, we must have $Y_2^{t^{\star}}(t)=0$. Therefore, using \eqref{eq-time_seq_prob}, along with \eqref{eq-num_time_seq_pf1}, we obtain the desired result.
		
	For $t^{\star}=\infty$, only the link value sequence consisting of all zeros ends with a zero, and in this case we have $M^{\infty}(t)=-1$. The result then follows.

\section{Proof of Theorem~\ref{thm-link_status_avg_inf}}\label{app-link_status_avg_inf_pf}

	Since we consider the limit $t\to\infty$, it suffices to consider the expression for $\Pr[X(t)=1]_{t^{\star}}$ in \eqref{eq-link_status_Pr1} for $t>t^{\star}+1$. Also due to the $t\to\infty$ limit, we can disregard the indicator function in \eqref{eq-link_status_Pr1}, so that
	\begin{equation}\label{eq-link_status_avg_inf_pf0}
		\lim_{t\to\infty}\mathbb{E}[X(t)]_{t^{\star}}=\lim_{t\to\infty}\sum_{x=0}^{\floor{\frac{t-1}{t^{\star}+1}}}\sum_{k=1}^{t^{\star}+1}\binom{t-k-xt^{\star}}{x}p^{x+1}(1-p)^{t-k-(t^{\star}+1)x}.
	\end{equation}
	Next, consider the binomial expansion of $(1-p)^{t-k-(t^{\star}+1)x}$:
	\begin{equation}
		(1-p)^{t-k-(t^{\star}+1)x}=\sum_{j=0}^{\infty}\binom{t-k-(t^{\star}+1)x}{j}(-1)^j p^j.
	\end{equation}
	Substituting this into \eqref{eq-link_status_avg_inf_pf0} gives us
	\begin{align}
		\lim_{t\to\infty}\mathbb{E}[X(t)]_{t^{\star}}&=p\lim_{t\to\infty}\sum_{x,j=0}^{\infty}\sum_{k=1}^{t^{\star}+1}\binom{t-k-t^{\star}x}{x}\binom{t-k-(t^{\star}+1)x}{j}(-1)^j p^{x+j}\\
		&=p\lim_{t\to\infty}\sum_{\ell=0}^{\infty}\sum_{j=0}^{\ell}\sum_{k=1}^{t^{\star}+1}\binom{t-k-t^{\star}j}{j}\binom{t-k-(t^{\star}+1)j}{\ell-j}(-1)^{\ell-j} p^{\ell}.\label{eq-link_status_avg_inf_pf4}
	\end{align}
	Now, for brevity, let $a\equiv t-k$, and let us focus on the sum
	\begin{equation}\label{eq-link_status_avg_inf_pf1}
		\sum_{j=0}^{\ell}(-1)^{\ell-j}\binom{a-t^{\star}j}{j}\binom{a-t^{\star}j-j}{\ell-j}.
	\end{equation}
	We start by expanding the binomial coefficients to get
	\begin{align}
		\binom{a-t^{\star}j}{j}\binom{a-t^{\star}j-j}{\ell-j}&=\frac{(a-t^{\star}j)!}{j!(\ell-j)!(a-t^{\star}j-\ell)!}\\
		&=\frac{1}{j!(\ell-j)!}\prod_{s=0}^{\ell-1}(a-t^{\star}j-s)\\
		&=\frac{1}{\ell!}\binom{\ell}{j}\prod_{s=0}^{\ell-1}(a-t^{\star}j-s).
	\end{align}
	Next, we have
	\begin{equation}
		\prod_{s=0}^{\ell-1}(a-t^{\star}j-s)=\sum_{n=0}^{\ell}(-1)^{\ell-n}\begin{bmatrix}\ell\\n\end{bmatrix}(a-t^{\star}j)^n,
	\end{equation}
	where $\begin{bmatrix}\ell\\n\end{bmatrix}$ is the (unsigned) Stirling number of the first kind\footnote{This number is defined to be the number of permutations of $\ell$ elements with $n$ disjoint cycles.}. Performing the binomial expansion of $(a-t^{\star}j)^n$, the sum in \eqref{eq-link_status_avg_inf_pf1} becomes
	\begin{equation}\label{eq-link_status_avg_inf_pf3}
		\sum_{j=0}^{\ell}\sum_{n=0}^{\ell}\sum_{i=0}^n (-1)^{\ell-j}\frac{1}{\ell!}\binom{\ell}{j}\begin{bmatrix}\ell\\n\end{bmatrix}\binom{n}{i}(-1)^i(t^{\star})^i j^i a^{n-i}.
	\end{equation}
	Now, it holds that
	\begin{equation}\label{eq-link_status_avg_inf_pf2}
		\sum_{j=0}^{\ell}(-1)^{\ell-j}\frac{1}{\ell!}\binom{\ell}{j}j^i=(-1)^{2\ell}\begin{Bmatrix}i\\\ell\end{Bmatrix},
	\end{equation}
	where $\begin{Bmatrix}i\\\ell\end{Bmatrix}$ is the Stirling number of the second kind\footnote{This number is defined to be the number of ways to partition a set of $i$ objects into $\ell$ non-empty subsets.}. For $i<\ell$, it holds that $\begin{Bmatrix}i\\\ell\end{Bmatrix}=0$, and $\begin{Bmatrix}\ell\\\ell\end{Bmatrix}=1$. Since $i$ ranges from 0 to $n$, and $n$ itself ranges from 0 to $\ell$, the sum in \eqref{eq-link_status_avg_inf_pf2} is zero except for when $i=\ell$. The sum in \eqref{eq-link_status_avg_inf_pf2} is therefore effectively equal to $(-1)^{2\ell}\delta_{i,\ell}$. Substituting this into \eqref{eq-link_status_avg_inf_pf3} leads to
	\begin{equation}
		\sum_{n=0}^{\ell}\sum_{i=0}^n(-1)^{2\ell}\delta_{i,\ell}\begin{bmatrix}\ell\\n\end{bmatrix}\binom{n}{i}(-1)^i(t^{\star})^i a^{n-i}=(-1)^{\ell}(t^{\star})^{\ell},
	\end{equation}
	where we have used the fact that $\begin{bmatrix}\ell\\\ell\end{bmatrix}=1$. Altogether, we have shown that
	\begin{equation}\label{eq-special_sum}
		\sum_{j=0}^{\ell}(-1)^{\ell-j}\binom{a-t^{\star}j}{j}\binom{a-t^{\star}j-j}{\ell-j}=(-1)^{\ell}(t^{\star})^{\ell}
	\end{equation}
	for all $\ell\geq 0$. The sum is independent of $a=t-k$. Substituting this result into \eqref{eq-link_status_avg_inf_pf4}, and using the fact that
	\begin{equation}\label{eq-homographic_func}
		\sum_{\ell=0}^{\infty} (-1)^{\ell} x^{\ell}=\frac{1}{1+x},\quad x\neq -1,
	\end{equation}
	we get
	\begin{equation}
		\lim_{t\to\infty}\mathbb{E}[X(t)]_{t^{\star}}=p\sum_{\ell=0}^{\infty}\sum_{k=1}^{t^{\star}+1}(-1)^{\ell}(t^{\star}p)^{\ell}=p(t^{\star}+1)\sum_{\ell=0}^{\infty} (-1)^{\ell} (t^{\star}p)^{\ell}=\frac{(t^{\star}+1)p}{1+t^{\star}p},
	\end{equation}
	as required.

\section{Proof of Theorem~\ref{thm-avg_mem_status_inf}}\label{app-avg_mem_status_inf_pf}

	The proof is very similar to the proof of Theorem~\ref{thm-link_status_avg_inf}. Using the result of Theorem~\ref{thm-mem_status_pr}, in the limit $t\to\infty$ we have
	\begin{equation}
		\lim_{t\to\infty}\Pr[M^{t^{\star}}(t)=m,X(t)=1]_{t^{\star}}=\lim_{t\to\infty}\sum_{x=0}^{\infty} \binom{t-(m+1)-xt^{\star}}{x}p^{x+1}(1-p)^{t-(m+1)-x(t^{\star}+1)}.
	\end{equation}
	Using the binomial expansion of $(1-p)^{t-(m+1)-x(t^{\star}+1)}$, exactly as in the proof of Theorem~\ref{thm-link_status_avg_inf}, we can write
	\begin{align}
		&\lim_{t\to\infty}\Pr[M^{t^{\star}}(t)=m,X(t)=1]_{t^{\star}}\nonumber\\
		&\qquad=\lim_{t\to\infty}\sum_{x=0}^{\infty}\sum_{j=0}^{\infty} p\binom{t-(m+1)-xt^{\star}}{x}\binom{t-(m+1)-(t^{\star}+1)x}{j}(-1)^jp^{x+j}\\
		&\qquad=\lim_{t\to\infty}\sum_{\ell=0}^{\infty}\sum_{j=0}^{\ell}p\binom{t-(m+1)-jt^{\star}}{j}\binom{t-(m+1)-(t^{\star}+1)j}{\ell-j}(-1)^{\ell-j}p^{\ell}.
	\end{align}
	Then, using \eqref{eq-special_sum}, we have that
	\begin{equation}
		\sum_{j=0}^{\ell}(-1)^{\ell-j}\binom{t-(m+1)-jt^{\star}}{j}\binom{t-(m+1)-(t^{\star}+1)j}{\ell-j} =(-1)^{\ell}(t^{\star})^{\ell}
	\end{equation}
	for all $t\geq 1$ and all $m\in\{0,1,\dotsc,t^{\star}\}$. Finally, using \eqref{eq-homographic_func}, we obtain
	\begin{equation}
		\lim_{t\to\infty}\Pr[M^{t^{\star}}(t)=m,X(t)=1]_{t^{\star}}=p\sum_{\ell=0}^{\infty} (-1)^{\ell}(t^{\star}p)^{\ell}=\frac{p}{1+t^{\star}p},
	\end{equation}
	as required.
	
	The proof of \eqref{eq-mem_time_prob_infty_x0} is similar, and it involves making use of the result of Proposition~\ref{prop-mem_time_prob_x0}.

\section{Proof of Theorem~\ref{thm-mem_cutoff_collective_waiting_time_tstarInf}}\label{app-mem_cutoff_collective_waiting_time_tstarInf_pf}

	By definition,
	\begin{equation}
		\Pr[W_{E'}(t_{\text{req}})=t]_{\infty}=\Pr[X_{E'}(t_{\text{req}}+1)=0,\dotsc,X_{E'}(t_{\text{req}}+t)=1]_{\infty}.
	\end{equation}
	Note that
	\begin{equation}\label{eq-waiting_time_tstarInf_pf1}
		\Pr[W_{E'}(t_{\text{req}})=1]_{\infty}=\Pr[X_{E'}(t_{\text{req}}+1)=1]_{\infty}=(1-(1-p)^{t_{\text{req}}+1})^M=p_{t_{\text{req}}+1}^M,
	\end{equation}
	which holds because all of the elementary links are generated independently and because they all have the same success probability.
	
	Now, for $t\geq 2$, our first goal is to prove that
	\begin{equation}\label{eq-waiting_time_tstarInf_pf4}
		\Pr[W_{E'}(t_{\text{req}})=t]_{\infty}=(1-(1-p_{t_{\text{req}}+1})(1-p)^{t-1})^M-(1-(1-p_{t_{\text{req}}+1})(1-p)^{t-2})^M.
	\end{equation}
	In order to prove this, let us for the moment take $t_{\text{req}}=0$. Then, $X_{E'}(1)=0$ means that at least one of the $M$ elementary links is not active in the first time step, and the same for all subsequent time steps except for the $t^{\text{th}}$ time step, in which all of the $M$ elementary links are active. Then, because $t^{\star}=\infty$, the links that are active in the first time step always remain active. This means that we can evaluate $\Pr[W_{E'}(0)=t]_{\infty}$ by counting the number of elementary links that are inactive at each time step. For example, for $t=2$, we obtain
	\begin{align}
		\Pr[X_{E'}(1)=0,X_{E'}(2)=1]_{\infty}&=\sum_{k_1=1}^M\binom{M}{k_1}\underbrace{(1-p)^{k_1}}_{\substack{k_1\text{inactive links}\\\text{in the first}\\\text{time step}}} \underbrace{p^{M-k_1}}_{\substack{M-k_1\text{ active}\\\text{links in the}\\\text{first time step}}} \underbrace{p^{k_1}}_{\substack{\text{remaining }k_1\\\text{ inactive links}\\\text{succeed in the}\\\text{second time step}}}\\
		&=p^M\sum_{k_1=1}^M\binom{M}{k_1}(1-p)^{k_1}.
	\end{align}
	Similarly, for $t=3$, we find that
	\begin{align}
		&\Pr[X_{E'}(1)=0,X_{E'}(2)=0,X_{E'}(3)=1]_{\infty}\nonumber\\
		&\qquad\qquad\qquad\qquad=\sum_{k_1=1}^M\binom{M}{k_1}(1-p)^{k_1}p^{M-k_1}\sum_{k_2=1}^{k_1}\binom{k_1}{k_2}(1-p)^{k_2}p^{k_1-k_2}p^{k_2}\\
		&\qquad\qquad\qquad\qquad=p^M\sum_{k_1=1}^M\sum_{k_2=1}^{k_1}\binom{M}{k_1}\binom{k_1}{k_2}(1-p)^{k_1}(1-p)^{k_2}.
	\end{align}
	In general, then, for all $t\geq 2$,
	\begin{align}
		&\Pr[W_{E'}(0)=t]_{\infty}\nonumber=\Pr[X_{E'}(1)=0,\dotsc,X_{E'}(t)=1]_{\infty}\\
		&\quad=p^M\sum_{k_1=1}^M\sum_{k_2=1}^{k_1}\sum_{k_3=1}^{k_2}\dotsb\sum_{k_{t-1}=1}^{k_{t-2}}\binom{M}{k_1}\binom{k_1}{k_2}\binom{k_2}{k_3}\dotsb\binom{k_{t-2}}{k_{t-1}}(1-p)^{k_1}(1-p)^{k_2}\dotsb(1-p)^{k_{t-1}}\\
		&\quad=\sum_{k_1=1}^M\binom{M}{k_1}(1-p)^{k_1}p^{M-k_1}\underbrace{p^{k_1}\sum_{k_2=1}^{k_1}\sum_{k_3=1}^{k_2}\dotsb\sum_{k_{t-1}=1}^{k_{t-2}}\binom{k_1}{k_2}\binom{k_2}{k_3}\dotsb\binom{k_{t-2}}{k_{t-1}}(1-p)^{k_2}\dotsb(1-p)^{k_{t-1}}}_{\Pr[W_{k_1}^{(\infty)}(0)=t-1]}\\
		&\quad=\sum_{k_1=1}^M \binom{M}{k_1}(1-p)^{k_1}p^{M-k_1}\Pr[W_{k_1}(0)=t-1]_{\infty} \label{eq-waiting_time_tstarInf_pf3}
	\end{align}
	Using this, we can immediately prove the following result by induction on $t$:
	\begin{equation}\label{eq-waiting_time_tstarInf_pf2}
		\Pr[W_{E'}(0)=t]_{\infty}=(1-(1-p)^t)^M-(1-(1-p)^{t-1})^M.
	\end{equation}
	Indeed, from \eqref{eq-waiting_time_tstarInf_pf1}, we immediately have that this result holds for $t=1$. Similarly, using the fact that
	\begin{equation}
		\sum_{k_1=1}^M\binom{M}{k_1}(1-p)^{k_1}=-1+(2-p)^M=\frac{1}{p^M}\left((1-(1-p)^2)^M-(1-(1-p))^M\right),
	\end{equation}
	we see that \eqref{eq-waiting_time_tstarInf_pf2} holds for $t=2$ as well. Now, assuming that \eqref{eq-waiting_time_tstarInf_pf2} holds for all $t\geq 2$, using \eqref{eq-waiting_time_tstarInf_pf3} we find that
	\begin{align}
		\Pr[W_{E'}(0)=t+1]_{\infty}&=\sum_{k_1=1}^M \binom{M}{k_1}(1-p)^{k_1}p^{M-k_1}\Pr[W_{k_1}(0)=t]_{\infty}\\
		&=\sum_{k_1=1}^M\binom{M}{k_1}(1-p)^{k_1}p^{M-k_1}\left((1-(1-p)^t)^{k_1}-(1-(1-p)^{t-1})^{k_1}\right)\\
		&=(1-(1-p)^{t+1})^M-(1-(1-p)^t)^M,
	\end{align}
	as required. Therefore, \eqref{eq-waiting_time_tstarInf_pf2} holds for all $t\geq 1$.
	
	We are now in a position to prove \eqref{eq-waiting_time_tstarInf_pf4}. Recall that for the $t^{\star}=\infty$ policy, $\Pr[X(t)=1]_{\infty}=1-(1-p)^t=p_t$. Therefore, at time step $t_{\text{req}}+1$, the probability that $k_1\geq 1$ elementary links are inactive is $(1-p_{t_{\text{req}}+1})^{k_1}$ and the probability that $M-k_1$ elementary links are active is $p_{t_{\text{req}}+1}^{M-k_1}$. In the subsequent time steps, each inactive elementary link from the previous time step is active with probability $p$ and inactive with probability $1-p$. Therefore,
	\begin{align}
		\Pr[W_{E'}(t_{\text{req}})=t]_{\infty}&=\sum_{k_1=1}^M\binom{M}{k_1}(1-p_{t_{\text{req}}+1})^{k_1}p_{t_{\text{req}}+1}^{M-k_1}\sum_{k_2=1}^{k_1}\binom{k_1}{k_2}(1-p)^{k_2}p^{k_2-k_1}\dotsb\nonumber\\
		&\qquad\qquad\qquad\qquad\qquad\qquad\dotsb\sum_{k_{t-1}=1}^{k_{t-2}}\binom{k_{t-2}}{k_{t-1}}(1-p)^{k_{t-1}}p^{k_{t-2}-k_{t-1}}p^{k_{t-1}}\\
		&=\sum_{k_1=1}^M\binom{M}{k_1}(1-p_{t_{\text{req}}+1})^{k_1}p_{t_{\text{req}}+1}^{M-k_1}p^{k_1}\sum_{k_2=1}^{k_1}\dotsb\nonumber\\
		&\qquad\qquad\qquad\qquad\qquad\dotsb\sum_{k_{t-1}=1}^{k_{t-2}}\binom{k_1}{k_2}\dotsb\binom{k_{t-2}}{k_{t-1}}(1-p)^{k_2}\dotsb(1-p)^{k_{t-1}}\\
		&=\sum_{k_1=1}^M\binom{M}{k_1}(1-p_{t_{\text{req}}+1})^{k_1}p_{t_{\text{req}}+1}^{M-k_1}\Pr[W_{k_1}(0)=t-1]_{\infty}\\
		&=(1-(1-p_{t_{\text{req}}+1})(1-p)^{t-1})^M-(1-(1-p_{t_{\text{req}}+1})(1-p)^{t-2})^M,
	\end{align}
	which is precisely \eqref{eq-waiting_time_tstarInf_pf4}.
	
	Now, for brevity, let 
	\begin{equation}
		\widetilde{q}\equiv 1-p_{t_{\text{req}}+1},\quad q\equiv 1-p.
	\end{equation}
	Then,
	\begin{align}
		\Pr[W_{E'}(t_{\text{req}})=t]_{\infty}&=(1-\widetilde{q}q^{t-1})^M-(1-\widetilde{q}q^{t-2})^M\\
		&=\sum_{k=0}^M\binom{M}{k}(-1)^k(\widetilde{q}q^{t-1})^k-\sum_{k=0}^M\binom{M}{k}(-1)^k(\widetilde{q}q^{t-2})^k\\
		&=\sum_{k=1}^M\binom{M}{k}(-1)^k\widetilde{q}^k(q^{t-1})^k(1-q^{-k}).
	\end{align}
	Then, using the fact that
	\begin{equation}
		\sum_{t=2}^\infty t (q^k)^{t-1}=\frac{q^k(2-q^k)}{(1-q^k)^2},
	\end{equation}
	we obtain
	\begin{align}
		\mathbb{E}[W_{E'}(t_{\text{req}})]_{\infty}&=\sum_{t=1}^\infty t \Pr[W_{E'}(t_{\text{req}})=t]_{\infty}\\
		&=(1-\widetilde{q})^M+\sum_{k=1}^M\binom{M}{k}(-1)^k\widetilde{q}^k\left(\frac{q^k(2-q^k)}{(1-q^k)^2}\right)(1-q^{-k})\\
		&=(1-\widetilde{q})^M+\sum_{k=1}^M\binom{M}{k}(-1)^{k+1}\widetilde{q}^k\left(1+\frac{1}{1-q^k}\right)\\
		&=(1-\widetilde{q})^M+\sum_{k=1}^M\binom{M}{k}(-1)^{k+1}\widetilde{q}^k\left(1+\frac{1}{p_k}\right)\\
		&=(1-\widetilde{q})^M+\sum_{k=1}^M\binom{M}{k}(-1)^{k+1}\widetilde{q}^k+\sum_{k=1}^M\binom{M}{k}(-1)^{k+1}\frac{(1-p_k)^{t_{\text{req}}+1}}{p_k}\\
		&=1+\sum_{k=1}^M\binom{M}{k}(-1)^{k+1}\frac{(1-p_k)^{t_{\text{req}}+1}}{p_k},
	\end{align}
	where in the second-last line we used the fact that $\widetilde{q}^k=(1-p_k)^{t_{\text{req}}+1}$. Finally, using the fact that $1=\sum_{k=1}^M\binom{M}{k}(-1)^{k+1}$, we obtain
	\begin{equation}
		\mathbb{E}[W_{E'}(t_{\text{req}})]_{\infty}=\sum_{k=1}^M\binom{M}{k}(-1)^{k+1}\left(1+\frac{(1-p_k)^{t_{\text{req}}+1}}{p_k}\right),
	\end{equation}
	as required.

\section{Proof of Theorem~\ref{thm-avg_sucess_rate}}\label{sec-avg_sucess_rate_pf}
	
	We start with the observation that, for all histories $h^t$, the number of successful requests can be written in terms of the number $Y_1^{t^{\star}}(t)(h^t)$ of blocks of ones of length $t^{\star}+1$ and the number $Y_2^{t^{\star}}(t)(h^t)$ of trailing ones in the link value sequence corresponding to $h^t$ as
	\begin{equation}
		Y_1^{t^{\star}}(t)(h^t)+1-\delta_{Y_2^{t^{\star}}(t)(h^t),0}.
	\end{equation}
	Similarly, the total number of failed requests is
	\begin{equation}
		t-Y_2^{t^{\star}}(t)(h^t)-(t^{\star}+1)Y_1^{t^{\star}}(t)(h^t).
	\end{equation}
	Therefore,
	\begin{align}
		S_{e^j}(t)(h^t)&=\frac{Y_1^{t^{\star}}(t)(h^t)+1-\delta_{Y_2^{t^{\star}}(t)(h^t),0}}{t-Y_2^{t^{\star}}(t)(h^t)-(t^{\star}+1)Y_1^{t^{\star}}(t)(h^t)+Y_1^{t^{\star}}(t)(h^t)+1-\delta_{Y_2^{t^{\star}}(t)(h^t),0}}\\
		&=\frac{Y_1^{t^{\star}}(t)(h^t)+1-\delta_{Y_2^{t^{\star}}(t)(h^t),0}}{t-Y_2^{t^{\star}}(t)(h^t)-t^{\star}Y_1^{t^{\star}}(t)(h^t)+1-\delta_{Y_2^{t^{\star}}(t)(h^t),0}}.\label{eq-avg_success_rate_pf1}
	\end{align}
	Now, for $t\leq t^{\star}+1$, we always have $Y_1^{t^{\star}}(t)(h^t)=0$ for all histories $h^t$, and the link value sequence can consist only of a positive number of trailing ones not exceeding $t$. Thus, from Proposition~\ref{prop-time_seq_prob}, the probability of any such history is $p(1-p)^{t-Y_2^{t^{\star}}(t)(h^t)}$. Using \eqref{eq-avg_success_rate_pf1} then leads to
	\begin{align}
		\mathbb{E}[S_{e^j}(t)]_{t^{\star}}&=\sum_{h^t\in\Omega(t)}S_{e^j}(t)(h^t)\Pr[H(t)=h^t]_{t^{\star}}\\
		&=\sum_{k=1}^t \frac{1}{t-k+1}p(1-p)^{t-k}\\
		&=\sum_{j=0}^{t-1}\frac{1}{j+1}p(1-p)^j
	\end{align}
	for $t\leq t^{\star}+1$, as required, where the last equality follows by a change of summation variable.
	
	For $t>t^{\star}+1$, we use \eqref{eq-avg_success_rate_pf1} again, keeping in mind this time that the number of trailing ones can be equal to zero, to get
	\begin{align}
		\mathbb{E}[S_{e^j}(t)]_{t^{\star}}&=\sum_{h^t\in\Omega(t)}S_{e^j}(t)(h^t)\Pr[H(t)=h^t]_{t^{\star}}\\
		&=\sum_{h^t:Y_2^{t^{\star}}(t)(h^t)=0}S_{e^j}(t)(h^t)\Pr[H(t)=h^t]_{t^{\star}}+\sum_{h^t:Y_2^{t^{\star}}(t)(h^t)\geq 1}S_{e^j}(t)(h^t)\Pr[H(t)=h^t]_{t^{\star}}\\
		&=\sum_{x=0}^{\floor{\frac{t-1}{t^{\star}+1}}}\left(\frac{x}{t-t^{\star}x}\Pr[H(t)=h^t:Y_1^{t^{\star}}(t)(h^t)=x,Y_2^{t^{\star}}(t)(h^t)=0]_{t^{\star}}\right.\nonumber\\
		&\qquad\quad\left.+\sum_{k=1}^{t^{\star}+1}\frac{x+1}{t-k-t^{\star}x+1}\Pr[H(t)=h^t:Y_1^{t^{\star}}(t)(h^t)=x,Y_2^{t^{\star}}(t)(h^t)=k]_{t^{\star}}\right).
	\end{align}
	Using Proposition~\ref{prop-time_seq_prob}, we arrive at the desired result.
	
	\begin{figure}
		\centering
		\includegraphics[scale=1]{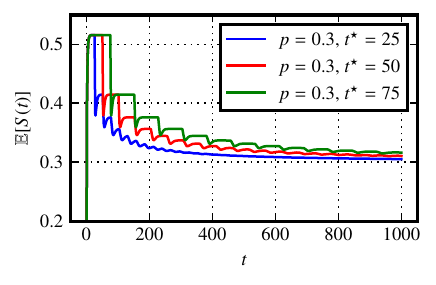}
		\caption{The expected success rate, as given by the expressions in Theorem~\ref{thm-avg_sucess_rate}, for an elementary link with $p=0.3$ and various cutoffs.}\label{fig-avg_success_rate_example}
	\end{figure}
	
	See Figure~\ref{fig-avg_success_rate_example} for a plot of the expected rate $\mathbb{E}[S(t)]$ as a function of time for various values of the cutoff. We find that the rate has essentially the shape of a decaying square wave, which is clearer for larger values of the cutoff. In particular, the ``plateaus'' in the curves have a period of $t^{\star}+1$ time steps. Let us now consider the values of these pleateaus. The largest plateau can be found by considering the case $t^{\star}=\infty$, because in this case the condition $t\leq t^{\star}+1$ is satisfied for all $t\geq 1$, and it is when this condition is true that the largest plateau occurs. Using Theorem~\ref{thm-avg_sucess_rate} with $t^{\star}=\infty$, we find that the value of the largest plateau approaches
	\begin{equation}\label{eq-avg_succ_rate_tInfty_1}
		\lim_{t\to\infty}\mathbb{E}[S_{e^j}(t)]_{\infty}=\lim_{t\to\infty}\sum_{j=0}^{t-1}\frac{1}{j+1}p(1-p)^j=-\frac{p\ln p}{1-p},
	\end{equation}
	for all $p\in(0,1)$. In the case $t^{\star}\in[0,\infty)$, as we see in Figure~\ref{fig-avg_success_rate_example}, there are multiple plateaus, with each plateau lasting for a period of $t^{\star}+1$ time steps, as mentioned earlier. The values of these pleateaus depend on the number $x\geq 0$ of full blocks of ones in the link value sequence. Specifically, the values of the plateaus approach
	\begin{multline}
		\lim_{t\to\infty}\sum_{k=1}^{t-(t^{\star}+1)x}\frac{x+1}{t-k-t^{\star}x+1}\binom{t-k-t^{\star}x}{x}p^{x+1}(1-p)^{t-k-(t^{\star}+1)x}\\=\lim_{t\to\infty}\sum_{j=(t^{\star}+1)x}^{t-1}\frac{x+1}{j-t^{\star}x+1}\binom{j-t^{\star}x}{x}p^{x+1}(1-p)^{j-(t^{\star}+1)x}=p\cdot{~}_2F_1(1,1,2+x,1-p),
	\end{multline}
	for all $x\geq 0$, where $_2F_1(a,b,c,z)$ is the hypergeometric function. Then, using the fact that $\lim_{x\to\infty} {~}_2F_1(1,1,2+x,1-p)=1$ \cite{CAJ17}, we conclude that the plateaus approach the value of $p$, i.e.,
	\begin{equation}\label{eq-avg_succ_rate_tInfty_2}
		\lim_{t\to\infty}\mathbb{E}[S_{e^j}(t)]_{t^{\star}}=p,\quad t^{\star}\in[0,\infty).
	\end{equation}

\chapter{PROOFS}\label{app-proofs}

	In this appendix, we present some of the longer proofs of the facts that are presented in this thesis.

\section{Proof of Proposition~\ref{lem-ent_swap_post_fidelity}}\label{app-ent_swap_post_fidelity_pf}

	For simplicity, and because we are mainly interested in qubits in this thesis, we show the proof only for $d=2$. The steps for $d>2$ are analogous but cumbersome.
		
	Let $\rho_{A\vec{R}_1\dotsb\vec{R}_nB}$ be an arbitrary state. Then, we have that
	\begin{multline}
		\bra{\Phi^+}_{AB}\mathcal{L}_{A\vec{R}_1\dotsb\vec{R}_nB\to AB}^{\text{ES}_n}\left(\rho_{A\vec{R}_1\dotsb\vec{R}_nB}\right)\ket{\Phi^+}_{AB}\\=\sum_{\vec{x},\vec{z}\in\{0,1\}^n}\left(\bra{\Phi_{a,b}}_{AB}\otimes\bra{\Phi_{z_1,x_1}}_{R_1^1R_1^2}\otimes\dotsb\otimes\bra{\Phi_{z_n,x_n}}_{R_n^1R_n^2}\right)\left(\rho_{A\vec{R}_1\dotsb\vec{R}_nB}\right)\\\left(\ket{\Phi_{a,b}}_{AB}\otimes\ket{\Phi_{z_1,x_1}}_{R_1^1R_1^2}\otimes\dotsb\otimes\ket{\Phi_{z_n,x_n}}_{R_n^1R_n^2}\right),\label{eq-ent_swap_output_fid_pf1}
	\end{multline}
	where
	\begin{equation}
		a\coloneqq z_1+\dotsb+z_1,\quad b\coloneqq x_1+\dotsb+x_n.
	\end{equation}
	Now, we use the fact that
	\begin{align}
		\ket{\Phi_{z,x}}&=\frac{1}{\sqrt{2}}\sum_{i=0}^1 (-1)^{iz}\ket{i,i+x},\\
		\ket{i,j}&=\frac{1}{\sqrt{2}}\sum_{x,z=0}^1 (-1)^{iz}\delta_{j,i+x}\ket{\Phi_{z,x}},\label{eq-comp_to_Bell}
	\end{align}
	for all $i,j,x,z\in\{0,1\}$. Then,
	\begin{align}
		& \ket{\Phi_{a,b}}_{AB}\otimes\ket{\Phi_{z_1,x_1}}_{R_1^1R_1^2}\otimes\ket{\Phi_{z_2,x_2}}_{R_2^1R_2^2}\otimes\dotsb\otimes\ket{\Phi_{z_n,x_n}}_{R_n^1R_n^2}\\
		&\quad =\frac{1}{\sqrt{2^{n+1}}}\sum_{j,i_1,i_2,\dotsc,i_n=0}^1(-1)^{jb+i_1z_1+i_2z_2+\dotsb+i_nz_n}\ket{j,j+a}_{AB}\ket{i_1,i_1+x_1}_{R_1^1R_1^2}\dotsb\ket{i_n,i_n+x_n}_{R_n^1R_n^2}\\
		&\quad =\frac{1}{\sqrt{2^{n+1}}}\sum_{j,i_1,i_2,\dotsc,i_n=0}^1(-1)^{jb+i_1z_1+i_2z_2+\dotsb+i_nz_n}\ket{j,i_1}_{AR_1^1}\ket{i_1+x_1,i_2}_{R_1^2R_2^1}\dotsb\ket{i_n+x_n,j+a}_{R_n^2B}\\
		&\quad=\frac{1}{2^{n+1}}\sum_{j,i_1,\dotsc,i_n,a',b'=0}^1\sum_{\substack{x_1',\dotsc,x_n'=0\\z_1',\dotsc,z_n'=0}}^1(-1)^{jb+i_1z_1+\dotsb i_nz_n}(-1)^{jb'}\delta_{i_1,j+a'}\ket{\Phi_{b',a'}}_{AR_1^1}\nonumber\\
		&\qquad\qquad\qquad\qquad\qquad\qquad\qquad\qquad\qquad\qquad\qquad(-1)^{(i_1+x_1)z_1'}\delta_{i_2,i_1+x_1+x_1'}\ket{\Phi_{z_1',x_1'}}_{R_1^2R_2^1}\dotsb\nonumber\\
		&\qquad\qquad\qquad\qquad\qquad\qquad\qquad\qquad\qquad\qquad\qquad(-1)^{(i_n+x_n)z_n'}\delta_{j+a,i_n+x_n+x_n'}\ket{\Phi_{z_n',x_n'}}_{R_n^2B}
	\end{align}
	After some simplification, and using the fact that
	\begin{equation}\label{eq-bit_string_sum_spec}
		\sum_{\vec{\gamma}\in\{0,1\}^n}(-1)^{\vec{\gamma}^{\t}\vec{x}}=2^n\delta_{\vec{x},\vec{0}}
	\end{equation}
	for all $\vec{x}\in\{0,1\}^n$, we obtain
	\begin{multline}
		\ket{\Phi_{a,b}}_{AB}\otimes\ket{\Phi_{z_1,x_1}}_{R_1^1R_1^2}\otimes\ket{\Phi_{z_2,x_2}}_{R_2^1R_2^2}\otimes\dotsb\otimes\ket{\Phi_{z_n,x_n}}_{R_n^1R_n^2}\\=\frac{1}{2^n}\sum_{\vec{x}',\vec{z}'\in\{0,1\}^n}(-1)^{(x_1'+\dotsb+x_n')(z_1+z_1')+(x_1+x_2'+\dotsb+x_n')(z_2+z_2')+\dotsb+(x_1+\dotsb+x_{n-1}+x_n')(z_n+z_n')}(-1)^{x_1z_1'+x_2z_2'+\dotsb+x_nz_n'}\\\ket{\Phi_{z_1'+\dotsb+z_n',x_1'+\dotsb+x_n'}}_{AR_1}\otimes\ket{\Phi_{z_1',x_1'}}_{R_1^2R_2^1}\otimes\dotsb\otimes\ket{\Phi_{z_n',x_n'}}_{R_n^2B}.
	\end{multline}
	Substituting this into \eqref{eq-ent_swap_output_fid_pf1}, and after much simplifying and repeated use of \eqref{eq-bit_string_sum_spec}, we obtain
	\begin{multline}
		\bra{\Phi^+}_{AB}\mathcal{L}_{A\vec{R}_1\dotsb\vec{R}_nB\to AB}^{\text{ES}_n}\left(\rho_{A\vec{R}_1\dotsb\vec{R}_nB}\right)\ket{\Phi^+}_{AB}\\=\sum_{\vec{x},\vec{z}\in\{0,1\}^n}\left(\bra{\Phi_{z_1+\dotsb z_n,x_1+\dotsb+x_n}}_{AR_1^1}\otimes\bra{\Phi_{z_1,x_1}}_{R_1^2R_2^1}\otimes\dotsb\otimes\bra{\Phi_{z_n,x_n}}_{R_n^2B}\right)\left(\rho_{A\vec{R}_1\dotsb\vec{R}_nB}\right)\\\left(\ket{\Phi_{z_1+\dotsb z_n,x_1+\dotsb+x_n}}_{AR_1^1}\otimes\ket{\Phi_{z_1,x_1}}_{R_1^2R_2^1}\otimes\dotsb\otimes\ket{\Phi_{z_n,x_n}}_{R_n^2B}\right),
	\end{multline}
	which holds for all states $\rho_{A\vec{R}_1\dotsb\vec{R}_nB}$. It therefore holds for the tensor product state in the statement of the proposition. This concludes the proof.

\section{Proof of Proposition~\ref{lem-ent_swap_GHZ_post_fidelity}}\label{app-ent_swap_GHZ_post_fidelity_pf}
	
	Let $\rho_{A\vec{R}_1\dotsb\vec{R}_nB}$ be an arbitrary state. We then have
	\begin{multline}\label{eq-ent_swap_GHZ_post_fid_pf1}
		\bra{\text{GHZ}_{n+2}}\mathcal{L}_{A\vec{R}_1\dotsb\vec{R}_nB\to AR_1^1\dotsb R_n^1B}^{\text{GHZ};n}\left(\rho_{A\vec{R}_1\dotsb\vec{R}_nB}\right)\ket{\text{GHZ}_{n+2}}\\=\frac{1}{2}\sum_{x,x'=0}^1\sum_{\vec{x}\in\{0,1\}^n}\bra{x,x,\dotsc,x}L^{x_n}_n\dotsb L^{x_2}_2 L^{x_1}_1\left(\rho_{A\vec{R}_1\dotsb\vec{R}_nB}\right)L^{x_1\dagger}_{1}L^{x_2\dagger}_{2}\dotsb L^{x_n\dagger}_{n}\ket{x',x',\dotsc,x'},
	\end{multline}
	where
	\begin{align}
		L^{x_j}_{j}&\coloneqq\bra{x_j}_{R_j^2}\text{CNOT}_{\vec{R}_j}X_{R_{j+1}^1}^{x_j}\\
		&=\bra{x_j}_{R_j^2}\left(\sum_{x'=0}^1\ket{x'}\bra{x'}_{R_j^1}\otimes X_{R_j^2}^{x'}\right)(\mathbbm{1}_{\vec{R}_j}\otimes X_{R_{j+1}^1}^{x_j})\\
		&=\sum_{x'=0}^1\ket{x'}\bra{x'}_{R_j^1}\otimes\bra{x_j+x'}_{R_j^2}\otimes X_{R_{j+1}^1}^{x_j}.
	\end{align}
	Then,
	\begin{multline}
		L^{x_n}_{n}\dotsb L^{x_2}_2 K^{x_1}_1=\sum_{x_1',\dotsc,x_n'=0}^1\ket{x_1',\dotsc,x_n'}\bra{x_1',x_2'+x_1,x_3'+x_2,\dotsc,x_n'+x_{n-1}}_{R_1^1R_2^1\dotsb R_n^1}\\\otimes\bra{x_1+x_1',x_2+x_2',\dotsc,x_n+x_n'}_{R_1^2R_2^2\dotsb R_n^2}\otimes X_B^{x_n},
	\end{multline}
	so that, using \eqref{eq-comp_to_Bell},
	\begin{align}
		&\bra{x,x,\dotsc,x}_{AR_1^1R_2^1\dotsb R_n^1B}L^{x_n}_n\dotsb L^{x_2}_2 L^{x_1}_1\nonumber\\
		&\quad=\bra{x}_A\bra{x,x+x_1,x+x_2,\dotsc,x+x_{n-1}}_{R_1^1R_2^1\dotsb R_n^1}\bra{x_1+x,x_2+x,\dotsc,x_n+x}_{R_1^2R_2^2\dotsb R_n^2}\bra{x+x_n}_B\\
		&\quad=\bra{x,x}_{AR_1^1}\bra{x+x_1,x+x_1}_{R_1^2R_2^1}\bra{x+x_2,x+x_2}_{R_2^2R_3^1}\dotsb\bra{x+x_n,x+x_n}_{R_n^2B}\\
		&\quad=\frac{1}{\sqrt{2^{n+1}}}\sum_{\vec{z}\in\{0,1\}^n}^1(-1)^{z_1x}(-1)^{z_2(x+x_1)}\dotsb(-1)^{z_{n+1}(x+x_n)}\bra{\Phi_{z_1,0}}_{AR_1^1}\bra{\Phi_{z_2,0}}_{R_1^2R_2^1}\dotsb\bra{\Phi_{z_{n+1},0}}_{R_n^2B}.
	\end{align}
	Substituting this into \eqref{eq-ent_swap_GHZ_post_fid_pf1}, simplifying, and making use of \eqref{eq-bit_string_sum_spec}, we obtain
	\begin{multline}
		\bra{\text{GHZ}_{n+2}}\mathcal{L}_{A\vec{R}_1\dotsb\vec{R}_nB\to AR_1^1\dotsb R_n^1B}^{\text{GHZ};n}\left(\rho_{A\vec{R}_1\dotsb\vec{R}_nB}\right)\ket{\text{GHZ}_{n+2}}\\=\sum_{z_2,\dotsc,z_{n+1}=0}^1\bra{\Phi_{z_2+\dotsb+z_{n+1},0}}\bra{\Phi_{z_2,0}}\dotsb\bra{\Phi_{z_{n+1},0}}\left(\rho_{A\vec{R}_1\dotsb\vec{R}_nB}\right)\ket{\Phi_{z_2+\dotsb+z_{n+1},0}}\ket{\Phi_{z_2,0}}\dotsb\ket{\Phi_{z_{n+1},0}}.
	\end{multline}
	This holds for all states $\rho_{A\vec{R}_1\dotsb\vec{R}_nB}$, so it holds for the tensor product state in the statement of the proposition, thus completing the proof.

\section{Proof of Proposition~\ref{prop-graph_state_dist_post_fidelity}}\label{app-graph_state_dist_post_fidelity_pf}

	Let $\rho_{A_1^nB_1^n}$ be an arbitrary $2n$-qubit state. Then, by definition of the channel $\mathcal{L}^{(G)}$, we have that
	\begin{equation}
		\bra{G}\mathcal{L}^{(G)}\left(\rho_{A_1^nR_1^n}\right)\ket{G}=\sum_{\vec{\gamma}\in\{0,1\}^n}\left(\bra{G^{\vec{\gamma}}}_{A_1^n}\otimes\bra{G^{\vec{\gamma}}}_{R_1^n}\right)\left(\rho_{A_1^nR_1^n}\right)\left(\ket{G^{\vec{\gamma}}}_{A_1^n}\otimes\ket{G^{\vec{\gamma}}}_{R_1^n}\right),
	\end{equation}
	where we recall the definition of $\ket{G^{\vec{\gamma}}}$ in \eqref{eq-graph_state_x}. Now,
	\begin{equation}
		\ket{G^{\vec{\gamma}}}_{A_1^n}\otimes\ket{G^{\vec{\gamma}}}_{R_1^n}=\frac{1}{2^n}\sum_{\vec{\alpha},\vec{\beta}\in\{0,1\}^n}(-1)^{\gamma_1(\alpha_1+\beta_1)+\dotsb\gamma_n(\alpha_n+\beta_n)}(-1)^{\frac{1}{2}\vec{\alpha}^{\t}A(G)\vec{\alpha}+\frac{1}{2}\vec{\beta}^{\t}A(G)\vec{\beta}}\ket{\vec{\alpha}}_{A_1^n}\otimes\ket{\vec{\beta}}_{R_1^n},
	\end{equation}
	and, for all $\vec{\alpha},\vec{\beta}\in\{0,1\}^n$,
	\begin{equation}
		\ket{\vec{\alpha}}_{A_1^n}\otimes\ket{\vec{\beta}}_{R_1^n}=\frac{1}{\sqrt{2^n}}\sum_{\vec{x},\vec{z}\in\{0,1\}^n}(-1)^{\alpha_1z_1+\dotsb+\alpha_nz_n}\delta_{\beta_1,\alpha_1+x_1}\dotsb\delta_{\beta_n,\alpha_n+x_n}\ket{\Phi_{z_1,x_1}}_{A_1R_1}\otimes\dotsb\otimes\ket{\Phi_{z_n,x_n}}_{A_nR^n},
	\end{equation}
	where we have used \eqref{eq-comp_to_Bell}. Then,
	\begin{multline}
		\ket{G^{\vec{\gamma}}}_{A_1^n}\otimes\ket{G^{\vec{\gamma}}}_{R_1^n}\\=\frac{1}{(2^n)^{\frac{3}{2}}}\sum_{\vec{\alpha},\vec{x},\vec{z}\in\{0,1\}^n}(-1)^{\vec{\gamma}^{\t}\vec{x}+\vec{\alpha}^{\t}\vec{z}}(-1)^{\frac{1}{2}\vec{\alpha}^{\t}A(G)\vec{\alpha}+\frac{1}{2}(\vec{\alpha}+\vec{x})^{\t}A(G)(\vec{\alpha}+\vec{x})}\ket{\Phi_{x_1,z_1}}_{A_1R_1}\otimes\dotsb\otimes\ket{\Phi_{z_n,x_n}}_{A_nR_n}.
	\end{multline}
	Now, because $A(G)$ is a symmetric matrix, we have that $\vec{\alpha}^{\t}A(G)\vec{x}=\vec{x}^{\t}A(G)\vec{\alpha}$. We thus obtain
	\begin{equation}
		(-1)^{\frac{1}{2}\vec{\alpha}^{\t}A(G)\vec{\alpha}+\frac{1}{2}(\vec{\alpha}+\vec{x})^{\t}A(G)(\vec{\alpha}+\vec{x})}=(-1)^{\vec{\alpha}^{\t}A(G)\vec{x}+\frac{1}{2}\vec{x}^{\t}A(G)\vec{x}},
	\end{equation}
	so that
	\begin{equation}
		\ket{G^{\vec{\gamma}}}_{A_1^n}\otimes\ket{G^{\vec{\gamma}}}_{R_1^n}=\frac{1}{(2^n)^{\frac{3}{2}}}\sum_{\vec{\alpha},\vec{x},\vec{z}\in\{0,1\}^n}(-1)^{\vec{\gamma}^{\t}\vec{x}+\vec{\alpha}^{\t}\vec{z}}(-1)^{\frac{1}{2}\vec{x}^{\t}A(G)\vec{x}+\vec{\alpha}^{\t}A(G)\vec{x}}\ket{\Phi_{z_1,x_1}}_{A_1R_1}\otimes\dotsb\otimes\ket{\Phi_{z_n,x_n}}_{A_nR_n}.
	\end{equation}
	Therefore, using \eqref{eq-bit_string_sum_spec}, we find that
	\begin{multline}
		\sum_{\vec{\gamma}\in\{0,1\}^n}\left(\bra{G^{\vec{\gamma}}}_{A_1^n}\otimes\bra{G^{\vec{\gamma}}}_{R_1^n}\right)\left(\rho_{A_1^nR_1^n}\right)\left(\ket{G^{\vec{\gamma}}}_{A_1^n}\otimes\ket{G^{\vec{\gamma}}}_{R_1^n}\right)\\=\frac{1}{(2^n)^2}\sum_{\vec{\alpha},\vec{\alpha}',\vec{z},\vec{z}',\vec{x}\in\{0,1\}^n}(-1)^{\vec{\alpha}^{\t}(A(G)\vec{x}+\vec{z})+\vec{\alpha}^{\prime\t}(A(G)\vec{x}+\vec{z}')}\left(\bra{\Phi_{z_1,x_1}}_{A_1R_1}\otimes\dotsb\otimes\bra{\Phi_{z_n,x_n}}_{A_nR_n}\right)\left(\rho_{A_1^nR_1^n}\right)\\\left(\ket{\Phi_{z_1',x_1'}}_{A_1R_1}\otimes\dotsb\otimes\ket{\Phi_{z_n',x_n}}_{A_nR_n}\right).
	\end{multline}
	Using \eqref{eq-bit_string_sum_spec} two more times in the summation with respect to $\vec{\alpha}$ and $\vec{\alpha}'$ finally leads to 
	\begin{multline}
		\sum_{\vec{\gamma}\in\{0,1\}^n}\left(\bra{G^{\vec{\gamma}}}_{A_1^n}\otimes\bra{G^{\vec{\gamma}}}_{R_1^n}\right)\left(\rho_{A_1^nR_1^n}\right)\left(\ket{G^{\vec{\gamma}}}_{A_1^n}\otimes\ket{G^{\vec{\gamma}}}_{R_1^n}\right)\\=\sum_{\vec{x}\in\{0,1\}^n}\left(\bra{\Phi_{z_1,x_1}}_{A_1R_1}\otimes\dotsb\otimes\bra{\Phi_{z_n,x_n}}_{A_nR_n}\right)\left(\rho_{A_1^nR_1^n}\right)\left(\ket{\Phi_{z_1,x_1}}_{A_1R_1}\otimes\dotsb\otimes\ket{\Phi_{z_n,x_n}}_{A_nR_n}\right),
	\end{multline}
	where $\vec{z}=A(G)\vec{x}$. Since this holds for all states $\rho_{A_1^nR_1^n}$, it holds for the tensor product state in the statement of the proposition, which completes the proof.

\section{Proof of Proposition~\ref{prop-noisy_transmission_output_Bell}}\label{sec-noisy_transmission_output_Bell_pf}

	The statement of Proposition~\ref{prop-noisy_transmission_output_Bell} follows from the following expressions, which we prove here:
	\begin{align}
		\Pi_{AB}(\mathcal{L}_A^{\eta_1,\overline{n}_1}\otimes\mathcal{L}_B^{\eta_2,\overline{n}_2})(\Phi_{AB}^{\pm})\Pi_{AB}&=\frac{1}{2}(x_1x_2+y_1y_2\pm z_1z_2)\Phi_{AB}^++\frac{1}{2}(x_1x_2+y_1y_2\mp z_1z_2)\Phi_{AB}^-\nonumber\\
		&\quad +\frac{1}{2}(x_1y_2+y_1x_2)\Psi_{AB}^++\frac{1}{2}(x_1y_2+y_1x_2)\Psi_{AB}^-,\label{eq-noisy_transmission_output_Bell12}\\
		\Pi_{AB}(\mathcal{L}_A^{\eta_1,\overline{n}_1}\otimes\mathcal{L}_B^{\eta_2,\overline{n}_2})(\Psi_{AB}^{\pm})\Pi_{AB}&=\frac{1}{2}(x_1y_2+y_1x_2)\Phi_{AB}^++\frac{1}{2}(x_1y_2+y_1x_2)\Phi_{AB}^-\nonumber\\
		&\quad +\frac{1}{2}(x_1x_2+y_1y_2\pm z_1z_2)\Psi_{AB}^+\nonumber\\
		&\quad +\frac{1}{2}(x_1x_2+y_1y_2\mp z_1z_2)\Psi_{AB}^-, \label{eq-noisy_transmission_output_Bell34}
	\end{align}
	for all $\eta_1,\eta_2,\overline{n}_1,\overline{n}_2\in[0,1]$.
	
	First, consider the system $A\equiv A_1A_2$, which is acted upon by the channel $\mathcal{L}^{\eta_1,\overline{n}_1}$. The beamsplitter unitary $U^{\eta_1}$ acts on the creation and annihilation operators of the system and environment modes as follows \cite{Serafini_book}:
	\begin{align}
		\hat{a}&\mapsto U^{\eta_1}\hat{a}U^{\eta_1\dagger}=\sqrt{\eta_1}\hat{a}+\sqrt{1-\eta_1}\hat{e},\\
		\hat{e}&\mapsto U^{\eta_1}\hat{e}U^{\eta_1\dagger}=\sqrt{1-\eta_1}\hat{a}-\sqrt{\eta_1}\hat{e}.
	\end{align}
	Then, it is straightforward to show that
	\begin{align}
		U_{A_1E_1}^{\eta_1}\ket{0,1}_{A_1E_1}&=\sqrt{1-\eta_1}\ket{1,0}_{A_1E_1}-\sqrt{\eta_1}\ket{0,1}_{A_1E_1},\\
		U_{A_1E_1}^{\eta_1}\ket{1,0}_{A_1E_1}&=\sqrt{\eta_1}\ket{1,0}_{A_1E_1}+\sqrt{1-\eta_1}\ket{0,1}_{A_1E_1},\\
		U_{A_1E_1}^{\eta_1}\ket{1,1}_{A_1E_1}&=\sqrt{2\eta_1(1-\eta_1)}\ket{2,0}_{A_1E_1}+(1-2\eta_1)\ket{1,1}_{A_1E_1}-\sqrt{2\eta_1(1-\eta_1)}\ket{0,2}_{A_1E_1}.
	\end{align}
	From this, we readily obtain
	\begin{align}
		\Tr_{E_1}[U_{A_1E_1}^{\eta_1}\ket{1,1}\bra{1,1}_{A_1E_1}U_{A_1E_1}^{\eta_1\dagger}]&=2\eta_1(1-\eta_1)\ket{2}\bra{2}_{A_1}+(1-2\eta_1)^2\ket{1}\bra{1}_{A_1}\nonumber\\
		&\qquad\qquad\qquad\qquad\qquad\qquad+2\eta_1(1-\eta_1)\ket{0}\bra{0}_{A_1},\label{eq-pf1}\\
		\Tr_{E_1}[U_{A_1E_1}^{\eta_1}\ket{0,1}\bra{1,1}_{A_1E_1}U_{A_1E_1}^{\eta_1\dagger}]&=\sqrt{2\eta_1}(1-\eta_1)\ket{1}\bra{2}_{A_1}-\sqrt{\eta_1}(1-2\eta_1)\ket{0}\bra{1}_{A_1},\label{eq-pf2}\\
		\Tr_{E_1}[U_{A_1E_1}^{\eta_1}\ket{1,0}\bra{1,1}_{A_1E_1}U_{A_1E_1}^{\eta_1\dagger}]&=\sqrt{2(1-\eta_1)}\eta_1\ket{1}\bra{2}_{A_1}+\sqrt{1-\eta_1}(1-2\eta_1)\ket{0}\bra{1}_{A_1},\label{eq-pf3}\\
		\Tr_{E_1}[U_{A_1E_1}^{\eta_1}\ket{0,0}\bra{1,1}_{A_1E_1}U_{A_1E_1}^{\eta_1\dagger}]&=\sqrt{2\eta_1(1-\eta_1)}\ket{0}\bra{2}_{A_1},\label{eq-pf4}\\
		\Tr_{E_1}[U_{A_1E_1}^{\eta_1}\ket{0,0}\bra{1,0}_{A_1E_1}U_{A_1E_1}^{\eta_1\dagger}]&=\sqrt{\eta_1}\ket{0}\bra{1}_{A_1},\label{eq-pf5}\\
		\Tr_{E_1}[U_{A_1E_1}^{\eta_1}\ket{1,0}\bra{0,0}_{A_1E_1}U_{A_1E_1}^{\eta_1\dagger}]&=\sqrt{\eta_1}\ket{1}\bra{0}_{A_1},\label{eq-pf6}\\
		\Tr_{E_1}[U_{A_1E_1}^{\eta_1}\ket{1,0}\bra{1,0}_{A_1E_1}U_{A_1E_1}^{\eta_1\dagger}]&=\eta_1\ket{1}\bra{1}_{A_1}+(1-\eta_1)\ket{0}\bra{0}_{A_1},\label{eq-pf7}\\
		\Tr_{E_1}[U_{A_1E_1}^{\eta_1}\ket{0,1}\bra{0,1}_{A_1E_1}U_{A_1E_1}^{\eta_1\dagger}]&=(1-\eta_1)\ket{1}\bra{1}_{A_1}+\eta_1\ket{0}\bra{0}_{A_1}.\label{eq-pf8}
	\end{align}
	Analogous expressions hold for $\Tr_{E_2}[U_{A_2E_2}^{\eta_1}\ket{i,j}\bra{k,\ell}_{A_2E_2}U_{A_2E_2}^{\eta_1\dagger}]$ for $i,j,k,\ell\in\{0,1\}$.
	
	Now, writing an arbitrary linear operator $X_{A_1A_2}$ as
	\begin{equation}
		X_{A_1A_2}=\alpha\ket{0,1}\bra{0,1}+\beta\ket{0,1}\bra{1,0}+\beta'\ket{1,0}\bra{0,1}+\gamma\ket{1,0}\bra{1,0},
	\end{equation}
	with $\alpha,\beta,\beta',\gamma\in\mathbb{C}$, we have that
	\begin{align}
		X_{A_1A_2}\otimes\widetilde{\Theta}_{E_1E_2}^{\overline{n}_1}&=\alpha(1-\overline{n}_1)\ket{0,0,1,0}\bra{0,0,1,0}_{A_1E_1A_2E_2}+\alpha\frac{\overline{n}_1}{2}\ket{0,0,1,1}\bra{0,0,1,1}_{A_1E_1A_2E_2}\\
		&\quad+\alpha\frac{\overline{n}_1}{2}\ket{0,1,1,0}\bra{0,1,1,0}_{A_1E_1A_2E_2}+\beta(1-\overline{n}_1)\ket{0,0,1,0}\bra{1,0,0,0}_{A_1E_1A_2E_2}\\
		&\quad+\beta\frac{\overline{n}_1}{2}\ket{0,0,1,1}\bra{1,0,0,1}_{A_1E_1A_2E_2}+\beta\frac{\overline{n}_1}{2}\ket{0,1,1,0}\bra{1,1,0,0}_{A_1E_1A_2E_2}\\
		&\quad+\beta'(1-\overline{n}_1)\ket{1,0,0,0}\bra{0,0,1,0}_{A_1E_1A_2E_2}+\beta'\frac{\overline{n}_1}{2}\ket{1,0,0,1}\bra{0,0,1,1}_{A_1E_1A_2E_2}\\
		&\quad+\beta'\frac{\overline{n}_1}{2}\ket{1,1,0,0}\bra{0,1,1,0}_{A_1E_1A_2E_2}+\gamma(1-\overline{n}_1)\ket{1,0,0,0}\bra{1,0,0,0}_{A_1E_1A_2E_2}\\
		&\quad +\gamma\frac{\overline{n}_1}{2}\ket{1,0,0,1}\bra{1,0,0,1}_{A_1E_1A_2E_2}+\gamma\frac{\overline{n}_1}{2}\ket{1,1,0,0}\bra{1,1,0,0}_{A_1E_1A_2E_2},
	\end{align}
	so that, using \eqref{eq-pf1}--\eqref{eq-pf8}, we obtain
	\begin{align}
		\mathcal{L}_{A_1A_2}^{\eta_1,\overline{n}_1}(X_{A_1A_2})&=\Tr_{E_1E_2}[(U_{A_1E_1}^{\eta_1}\otimes U^{\eta_1}_{A_2E_2})(X_{A_1A_2}\otimes\widetilde{\Theta}_{E_1E_2}^{\overline{n}_1})(U^{\eta_1}_{A_1E_1}\otimes U^{\eta_1}_{A_2E_2})^\dagger]\\[0.2cm]
		&=\alpha(1-\overline{n}_1)\ket{0}\bra{0}_{A_1}\otimes(\eta_1\ket{1}\bra{1}_{A_2}+(1-\eta_1)\ket{0}\bra{0}_{A_2})\\
		&\quad+\alpha\frac{\overline{n}_1}{2}\ket{0}\bra{0}_{A_1}\otimes(2\eta_1(1-\eta_1)\ket{2}\bra{2}_{A_2}+(1-2\eta_1)^2\ket{1}\bra{1}_{A_2}\nonumber\\
		&\qquad\qquad\qquad\qquad\qquad\qquad\qquad\qquad\qquad\qquad+2\eta_1(1-\eta_1)\ket{0}\bra{0}_{A_2})\\
		&\quad+\alpha\frac{\overline{n}_1}{2}((1-\eta_1)\ket{1}\bra{1}_{A_1}+\eta_1\ket{0}\bra{0}_{A_1})\otimes(\eta_1\ket{1}\bra{1}_{A_2}+(1-\eta_1)\ket{0}\bra{0}_{A_2})\\
		&\quad+\beta(1-\overline{n}_1)\sqrt{\eta_1}\ket{0}\bra{1}_{A_1}\otimes\sqrt{\eta_1}\ket{1}\bra{0}_{A_2}\\
		&\quad+\beta\frac{\overline{n}_1}{2}\sqrt{\eta_1}\ket{0}\bra{1}_{A_1}\otimes(\sqrt{2\eta_1}(1-\eta_1)\ket{2}\bra{1}_{A_2}-\sqrt{\eta_1}(1-2\eta_1)\ket{1}\bra{0}_{A_2})\\
		&\quad+\beta\frac{\overline{n}_1}{2}(\sqrt{2\eta_1}(1-\eta_1)\ket{1}\bra{2}_{A_1}-\sqrt{\eta_1}(1-2\eta_1)\ket{0}\bra{1}_{A_1})\otimes\sqrt{\eta_1}\ket{1}\bra{0}_{A_2}\\
		&\quad +\beta'(1-\overline{n}_1)\sqrt{\eta_1}\ket{1}\bra{0}_{A_1}\otimes\sqrt{\eta_1}\ket{0}\bra{1}_{A_2}\\
		&\quad+\beta'\frac{\overline{n}_1}{2}\sqrt{\eta_1}\ket{1}\bra{0}_{A_1}\otimes(\sqrt{2\eta_1}(1-\eta_1)\ket{1}\bra{2}_{A_2}-\sqrt{\eta_1}(1-2\eta_1)\ket{0}\bra{1}_{A_2})\\
		&\quad+\beta'\frac{\overline{n}_1}{2}(\sqrt{2\eta_1}(1-\eta_1)\ket{2}\bra{1}_{A_1}-\sqrt{\eta_1}(1-2\eta_1)\ket{1}\bra{0}_{A_1})\otimes\sqrt{\eta_1}\ket{0}\bra{1}_{A_2}\\
		&\quad+\gamma(1-\overline{n}_1)(\eta_1\ket{1}\bra{1}_{A_1}+(1-\eta_1)\ket{0}\bra{0}_{A_1})\otimes\ket{0}\bra{0}_{A_2}\\
		&\quad+\gamma\frac{\overline{n}_1}{2}(\eta_1\ket{1}\bra{1}_{A_1}+(1-\eta_1)\ket{0}\bra{0}_{A_1})\otimes((1-\eta_1)\ket{1}\bra{1}_{A_2}+\eta_1\ket{0}\bra{0}_{A_2})\\
		&\quad+\gamma\frac{\overline{n}_1}{2}(2\eta_1(1-\eta_1)\ket{2}\bra{2}_{A_1}+(1-2\eta_1)^2\ket{1}\bra{1}_{A_1}\nonumber\\
		&\qquad\qquad\qquad\qquad\qquad\qquad\qquad\qquad +2\eta_1(1-\eta_1)\ket{0}\bra{0}_{A_1})\otimes\ket{0}\bra{0}_{A_2}.
	\end{align}
	Then, with the appropriate choices of $\alpha,\beta,\beta',\gamma$, we obtain
	\begin{align}
		\mathcal{L}_{A_1A_2}^{\eta_1,\overline{n}_1}(\ket{0,1}\bra{0,1}_{A_1A_2})&=w_1\ket{0,0}\bra{0,0}_{A_1A_2}+x_1\ket{0,1}\bra{0,1}_{A_1A_2}+y_1\ket{1,0}\bra{1,0}_{A_1A_2}\nonumber\\
		&\qquad\qquad\qquad\qquad\qquad+u_1\ket{1,1}\bra{1,1}_{A_1A_2}+2u_1\ket{0,2}\bra{0,2}_{A_1A_2},\label{eq-pf9}\\
		\mathcal{L}_{A_1A_2}^{\eta_1,\overline{n}_1}(\ket{0,1}\bra{1,0}_{A_1A_2})&=z_1\ket{0,1}\bra{1,0}_{A_1A_2}+\sqrt{2}u_1\ket{0,2}\bra{1,1}_{A_1A_2}+\sqrt{2}u_1\ket{1,1}\bra{2,0}_{A_1A_2},\label{eq-pf10}\\
		\mathcal{L}_{A_1A_2}^{\eta_1,\overline{n}_1}(\ket{1,0}\bra{0,1}_{A_1A_2})&=z_1\ket{1,0}\bra{0,1}_{A_1A_2}+\sqrt{2}u_1\ket{1,1}\bra{0,2}_{A_1A_2}+\sqrt{2}u_1\ket{2,0}\bra{1,1}_{A_1A_2},\label{eq-pf11}\\
		\mathcal{L}_{A_1A_2}^{\eta_1,\overline{n}_1}(\ket{1,0}\bra{1,0}_{A_1A_2})&=w_1\ket{0,0}\bra{0,0}_{A_1A_2}+y_1\ket{0,1}\bra{0,1}_{A_1A_2}+x_1\ket{1,0}\bra{1,0}_{A_1A_2}\nonumber\\
		&\qquad\qquad\qquad\qquad\qquad+u_1\ket{1,1}\bra{1,1}_{A_1A_2}+2u_1\ket{2,0}\bra{2,0}_{A_1A_2},\label{eq-pf12}
	\end{align}
	where
	\begin{align}
		w_1&\coloneqq (1-\overline{n}_1)(1-\eta_1)+\frac{3}{2}\overline{n}_1\eta_1(1-\eta_1),\\
		x_1&\coloneqq (1-\overline{n}_1)\eta_1+\frac{\overline{n}_1}{2}((1-2\eta_1)^2+\eta_1^2),\\
		y_1&\coloneqq \frac{\overline{n}_1}{2}(1-\eta_1)^2,\\
		u_1&\coloneqq \frac{\overline{n}_1}{2}\eta_1(1-\eta_1),\\
		z_1&\coloneqq (1-\overline{n}_1)\eta_1-\overline{n}_1\eta_1(1-2\eta_1).
	\end{align}
	Analogous expressions hold for $\mathcal{L}_{B_1B_2}^{\eta_2,\overline{n}_2}(\ket{i,j}\bra{k,\ell}_{B_1B_2})$ for $i,j,k,\ell\in\{0,1\}$.

	Finally, since
	\begin{equation}
		\ket{\Phi^{\pm}}_{A_1A_2B_1B_2}\coloneqq\frac{1}{\sqrt{2}}(\ket{0,1,0,1}_{A_1A_2B_1B_2}\pm\ket{1,0,1,0}_{A_1A_2B_1B_2}),
	\end{equation}
	we find that
	\begin{align}
		(\mathcal{L}_A^{\eta_1,\overline{n}_1}\otimes\mathcal{L}_B^{\eta_2,\overline{n}_2})(\Phi_{AB}^{\pm})&=\frac{1}{2}\left(\mathcal{L}_{A_1A_2}^{\eta_1,\overline{n}_1}(\ket{0,1}\bra{0,1}_{A_1A_2})\otimes\mathcal{L}_{B_1B_2}^{\eta_2,\overline{n}_2}(\ket{0,1}\bra{0,1}_{B_1B_2})\right.\\
		&\qquad\left. \pm\mathcal{L}_{A_1A_2}^{\eta_1,\overline{n}_1}(\ket{0,1}\bra{1,0}_{A_1A_2})\otimes\mathcal{L}_{B_1B_2}^{\eta_2,\overline{n}_2}(\ket{0,1}\bra{1,0}_{B_1B_2})\right.\\
		&\qquad\left. \pm\mathcal{L}_{A_1A_2}^{\eta_1,\overline{n}_1}(\ket{1,0}\bra{0,1}_{A_1A_2})\otimes\mathcal{L}_{B_1B_2}^{\eta_2,\overline{n}_2}(\ket{1,0}\bra{0,1}_{B_1B_2})\right.\\
		&\qquad\left. +\mathcal{L}_{A_1A_2}^{\eta_1,\overline{n}_1}(\ket{1,0}\bra{1,0}_{A_1A_2})\otimes\mathcal{L}_{B_1B_2}^{\eta_2,\overline{n}_2}(\ket{1,0}\bra{1,0}_{B_1B_2})\right).
	\end{align}
	
	Then, using \eqref{eq-pf9}-\eqref{eq-pf12}, we obtain
	\begin{align}
		\Pi_{AB}(\mathcal{L}_A^{\eta_1,\overline{n}_1}\otimes\mathcal{L}_B^{\eta_2,\overline{n}_2})(\Phi_{AB}^{\pm})\Pi_{AB}&=\frac{1}{2}\left((x_1\ket{H}\bra{H}_A+y_1\ket{V}\bra{V}_A)\otimes(x_2\ket{H}\bra{H}_B+y_2\ket{V}\bra{V}_B)\right.\\
		&\quad\left. \pm z_1\ket{H}\bra{V}_A\otimes z_2\ket{H}\bra{V}_B\right.\\
		&\quad\left. \pm z_1\ket{V}\bra{H}_A\otimes z_2\ket{V}\bra{H}_B\right.\\
		&\quad\left. +(y_1\ket{H}\bra{H}_A+x_1\ket{V}\bra{V}_B)\otimes(y_2\ket{H}\bra{H}_B+x_2\ket{V}\bra{V}_B)\right)\\
		&=\frac{1}{2}(x_1x_2+y_1y_2)(\ket{H,H}\bra{H,H}_{AB}+\ket{V,V}\bra{V,V}_{AB})\\
		&\quad\pm\frac{1}{2}z_1z_2(\ket{H,H}\bra{V,V}_{AB}+\ket{V,V}\bra{H,H}_{AB})\\
		&\quad+\frac{1}{2}(x_1y_2+y_1x_2)(\ket{H,V}\bra{H,V}_{AB}+\ket{V,H}\bra{V,H}_{AB})\\
		&=\frac{1}{2}(x_1x_2+y_1y_2\pm z_1z_2)\Phi_{AB}^++\frac{1}{2}(x_1x_2+y_1y_2\mp z_1z_2)\Phi_{AB}^-\\
		&\quad + \frac{1}{2}(x_1y_2+y_1x_2)\Psi_{AB}^++\frac{1}{2}(x_1y_2+y_1x_2)\Psi_{AB}^-,
	\end{align}
	where the last equality is precisely \eqref{eq-noisy_transmission_output_Bell12}, as required. The calculation to obtain \eqref{eq-noisy_transmission_output_Bell34} is analogous.

\chapter{DIRECTIONS FOR FUTURE WORK}\label{sec-future_work}

	In this appendix, we outline some directions for future work. We state ideas and also outline some concrete steps that can be taken.

\section{Elementary link QDP with distillation}\label{sec-network_QDP_distillation}

	In Chapter~\ref{chap-network_QDP}, we defined a quantum decision process for elementary link generation without including entanglement distillation protocols. In order to develop more sophisticated entanglement distribution protocols using quantum decision processes, it is crucial that entanglement distillation be included in the quantum decision process for elementary link generation. Let us describe in detail how this can be done.
	
	The setting of elementary link generation with entanglement distillation applies specifically to the case that the network is described by a multigraph $G=(V,E,c)$, where the function $c:E\to\mathbb{N}$ tells us how many parallel edges correspond to the elements of $E$. We denote the parallel edges corresponding to $e\in E$ by $e^1,e^2,\dotsc,e^{ c(e)}$ (see Section~\ref{sec-graph_theory}). This means that up to $ c(e)$ of the parallel elementary links can be distilled at once. As before, we associate an agent to every $e\in E$, but now the environment of each agent consists of all of the quantum systems for all of the $ c(e)$ parallel elementary links associated with $e$. The agents are allowed to perform the same actions as before on each of their parallel elementary links, but now they can also distill some number $j\leq c(e)$ of the parallel elementary links.
	
	Throughout this development, we leave the actual distillation protocol arbitrary and simply assume, as described in Section~\ref{sec-network_q_state_practical}, that it can be described by some LOCC quantum instrument channel. See Refs.~\cite{BBP96,DAR96,BDSW96} for examples of bipartite entanglement distillation protocols, and Refs.~\cite{DCT99,DC00,DAB03,ADB05,HDB05,MB05} for examples of multipartite entanglement distillation protocols. See also Refs.~\cite{ZPZ01,YKI01,PSBZ01,DVDV03,BM05,MMO+07,CB08,NFB14,RST+18,KAJ19}. Upper bounds on the fidelity that can be achieved after an entanglement distillation protocol, in the non-asymptotic setting, can be calculated using a semi-definite program (SDP), as shown in Ref.~\cite{Rains01}. For practical entanglement distillation schemes, which typically only consist of one round of local operations and classical communication and also have non-unit success probability, SDP upper bounds have been provided in Ref.~\cite{RST+18}. In Ref.~\cite{WMDB19}, the authors use reinforcement learning to discover protocols for entanglement distillation. See Refs.~\cite{DBC99,HKBD07,RPL09,SJM19} for an analysis of quantum repeater protocols with entanglement distillation.
	
	\begin{definition}[QDP for elementary link generation with distillation]\label{def-QDP_elem_link_distill}
		Let $G=(V,E,c)$ be the multigraph corresponding to the physical (elementary) links of a quantum network, and let $e\in E$ consists of $k$ nodes, $k\geq 2$. We define a quantum decision process for $e$ by defining the agent for $e$ to be collectively the nodes belonging to $e$, and we define its environment to be the collection of all $k c(e)$ quantum systems distributed by the source station to the nodes of $e$. Then, the other elements of the quantum decision process are defined as follows.
		\begin{itemize}
			\item We denote the quantum systems of the environment collectively by $E^{e}$, and we let $E_t^e$ denote these quantum systems at time $t\geq 0$. We use $E_t^{e^j}$, $1\leq j\leq c(e)$, to refer to the quantum systems of the $j^{\text{th}}$ parallel elementary link of $e$ at time $t\geq 0$. The state of the environment at time $t=0$ is the source state $\rho_{E_0^e}^S\coloneqq\bigotimes_{j=1}^{ c(e)}\rho_{E_0^{e^j}}^{S}$, where $\rho_{E_0^{e^j}}^S$ is the source state for the $j^{\text{th}}$ parallel elementary link. We use the abbreviation $\rho_{e^j}^S\equiv\rho_{E_0^{e^j}}^S$ throughout.
			
			\item We let $\mathcal{X}_e=\{0,1\}^{ c(e)}$ tell us whether or not the parallel elementary links are active. In other words, every element $\vec{x}=(x^1,x^2,\dotsc,x^{ c(e)})\in\mathcal{X}_e$ is such that $x^j=0$ indicates that the $j^{\text{th}}$ parallel elementary link is inactive and $x^j=1$ indicates that the $j^{\text{th}}$ parallel elementary link is active. We then define random variables $X_{e^j}(t)$, $t\geq 1$ taking values in $\{0,1\}$ as follows:
				\begin{itemize}
					\item $X_{e^j}(t)=0$: the $j^{\text{th}}$ parallel elementary link is inactive;
					\item $X_{e^j}(t)=1$: the $j^{\text{th}}$ parallel elementary link is active.
				\end{itemize}
				We let
				\begin{multline}
					\mathcal{A}_e=\left\{(0;\vec{a}):\vec{a}\in\{0,1\}^{ c(e)}\right\}\\\cup\left\{(1;j,j',s_j,\vec{a}):2\leq j\leq c(e),\, j'<j,\, s_j\subseteq[ c(e)],\, \vec{a}\in\{0,1\}^{ c(e)-j}\right\}
				\end{multline}
				be the set of actions of the agent. An action of the form $(0;\vec{a})$ indicates that no distillation is to be performed, and the values in $\vec{a}=(a^1,a^2,\dotsc,a^{ c(e)})$ indicate whether the parallel elementary link should be kept (``wait'') or discarded and reattempted (``request''); i.e., $a^j=0$ indicates that the $j^{\text{th}}$ parallel elementary link should be kept (``wait'') and $a^j=1$ indicates that the $j^{\text{th}}$ parallel elementary link should be discarded and reattempted (``request''). An action of the form $(1;j,j',s_j,\vec{a})$ indicates that a distillation protocol from $j$ to $j'$ parallel elementary links should be performed on the $j$ parallel elementary links specified by the set $s_j=\{i_1,i_2,\dotsc,i_j\}\subseteq[ c(e)]$, and on the remaining $ c(e)-j$ parallel elementary links the action is given by the elements of $\vec{a}=(a^i:i\in[ c(e)]\setminus s_j)\in\{0,1\}^{ c(e)-j}$ should be performed. Based on these definitions, we define random variables $A(t)$, $t\geq 1$ that take values in $\mathcal{A}_e$.
				
				A history $H_e(t)$ is of the form
				\begin{equation}
					H_e(t)=(\vec{X}_e(1),A_e(1),\vec{X}_e(2),A_e(2),\dotsc,A_e(t-1),\vec{X}_e(t-1)),
				\end{equation}
				where $\vec{X}_e(t)=(X_{e^1}(t),\dotsc,X_{e^{c(e)}}(t))$. The set of all histories is $\Omega_e(t)$ and every $h^t\in\Omega_e(t)$ is of the form $h^t=(\vec{x}_1,a_1,\vec{x}_2,a_2,\dotsc,a_{t-1},\vec{x}_t)$ for all $t\geq 1$, with $h^1=\vec{x}_1$.
			
			\item The transition maps $\mathcal{T}_e^{0;\vec{x}_1}$ are defined as follows:
				\begin{equation}
					\mathcal{T}_e^{0;\vec{x}_1}\coloneqq\bigotimes_{j=1}^{ c(e)}\left(\mathcal{M}_{e^j}^{x_1^j}\circ\mathcal{S}_{e^j}\right)
				\end{equation}
				for all $\vec{x}_1=(x_1^1,x_1^2,\dotsc,x_1^{ c(e)})\in\{0,1\}^{ c(e)}$, where $\mathcal{S}_{e^j}$ is the source transmission channel for the $j^{\text{th}}$ parallel elementary link and $\{\mathcal{M}_{e^j}^0,\mathcal{M}_{e^j}^1\}$ is the heralding quantum instrument for the $j^{\text{th}}$ parallel elementary link.
			
				For $t\geq 1$, we denote the transition maps by $\mathcal{T}_{e;e^1,\dotsc,e^{ c(e)}}^{\vec{x}_t,a_t,\vec{x}_{t+1}}$, with $\vec{x}_t,\vec{x}_{t+1}\in\mathcal{X}_e$ and $a_t\in\mathcal{A}_e$. Using the transition maps for individual parallel elementary links as defined in Definition~\ref{def-network_QDP_elem_link}, for actions of the form $(0,\vec{a}_t)$ we let
				\begin{equation}\label{eq-network_QDP_w_distill_trans1}
					\mathcal{T}_{e;e^1,\dotsc,e^{ c(e)}}^{\vec{x}_t,(0;\vec{a}_t),\vec{x}_{t+1}}\coloneqq\bigotimes_{i=1}^{ c(e)}\mathcal{T}_{e^i}^{x_t^i,a_t^i,x_{t+1}^i}\quad \forall~\vec{a}_t\in\{0,1\}^{ c(e)},
				\end{equation}
				where the maps $\mathcal{T}_{e^i}^{x_t^i,a_t^i,x_{t+1}^i}$ are defined in Definition~\ref{def-network_QDP_elem_link}. In other words, if no distillation is to be performed, then the individual parallel elementary links evolve independently, exactly as they do in the quantum decision process without distillation. For actions of the form $(1;j,j',s_j,\vec{a}_t)$, with $s_j=\{i_1,\dotsc,i_j:1\leq i_1<\dotsb<i_j\leq c(e)\}$, we let
				\begin{equation}\label{eq-network_QDP_w_distill_trans2}
					\begin{aligned}
					&\mathcal{T}_{e;e^1,\dotsc,e^{ c(e)}}^{\vec{x}_t,(1;j,j',s_j,\vec{a}_t),\vec{x}_{t+1}}\\&\quad\coloneqq \left\{\begin{array}{l l} \displaystyle \bigotimes_{i\in s_j}\mathcal{T}_{e^i}^{x_t^i,0,x_{t+1}^i}\bigotimes_{i\in[ c(e)]\setminus s_j}\mathcal{T}_{e^i}^{x_t^i,a_t^i,x_{t+1}^i} & \text{if }x_t^{i_1}\dotsb x_t^{i_j}=0, \\[1cm]
					
					\displaystyle \mathcal{D}^{e;1}_{e^{i_1}\dotsb e^{i_j}\to e^{i_1}\dotsb e^{i_{j'}}}\bigotimes_{\ell=j'+1}^j \widetilde{\mathcal{S}}_{e^{i_{\ell}}}^{x_{t+1}^{i_{\ell}}}(\rho_{e^{i_{\ell}}}^S)\bigotimes_{i\in[ c(e)]\setminus s_j}\mathcal{T}_{e^i}^{x_t^i,a_t^i,x_{t+1}^i} & \text{if } x_t^{i_1}\dotsb x_t^{i_j}=1,\\ & x_{t+1}^{i_1}\dotsb x_{t+1}^{i_{j'}}=1, \\[1cm]
					
					\displaystyle \mathcal{D}^{e;0}_{e^{i_1}\dotsb e^{i_j}\to e^{i_1}\dotsb e^{i_{j'}}}\bigotimes_{\ell=j'+1}^j \widetilde{\mathcal{S}}_{e^{i_{\ell}}}^{x_{t+1}^{i_{\ell}}}(\rho_{e^{i_{\ell}}}^S)\bigotimes_{i\in[ c(e)]\setminus s_j}\mathcal{T}_{e^i}^{x_t^i,a_t^i,x_{t+1}^i} & \text{if } x_t^{i_1}\dotsb x_t^{i_j}=1,\\ & x_{t+1}^{i_\ell}=0 \\ & \forall~1\leq\ell\leq j', \end{array}\right.
					\end{aligned}
				\end{equation}
				where
				\begin{equation}
					\widetilde{\mathcal{S}}_{e^j}^{x^j}\coloneqq\mathcal{M}_{e^j}^{x^j}\circ\mathcal{S}_{e^j}
				\end{equation}
				for all $1\leq j\leq c(e)$. In words, if at least one of the parallel links in the set $s_j$ is inactive (so the first condition is satisfied), then no distillation is performed and all of the $j$ links in $s_j$ are evolved according to the ``wait'' transition map for the individual links. If if all of the $j$ links in the set $s_j$ are active and the distillation succeeds (so the second condition is satisfied), then the distillation map corresponding to success is applied to the $j$ links in $s_j$, on the remaining $j-j'$ links we request a new link, and on the links not in $s_j$ we apply the actions specified by $\vec{a}_t$ individually. Finally, if all of the links in $s_j$ are active but the distillation fails (so the third condition is satisfied), then the distillation map corresponding to failure is applied to the $s_j$ links, on the remaining $j-j'$ links we request a new link, and on the links not in $s_j$ we apply the actions specified by $\vec{a}_t$ individually.
		
			\item Given a target $k c(e)$-partite pure quantum state $\psi_e^{\text{target}}=\ket{\psi^{\text{target}}}\bra{\psi^{\text{target}}}_e$, the reward at time $t\geq 1$ is defined as follows:
				\begin{align}
					\mathcal{R}_e^{t;h^{t+1},1}(\cdot)&=\psi_e^{\text{target}}(\cdot)\psi_e^{\text{target}},\\
					\mathcal{R}_e^{t;h^{t+1},0}(\cdot)&=(\mathbbm{1}_e-\psi_e^{\text{target}})(\cdot)(\mathbbm{1}_e-\psi_e^{\text{target}}),
				\end{align}
				for all $h^{t+1}\in\Omega_e(t+1)$, and the functions $R_e(t):\Omega_e(t+1)\times\{0,1\}\to\mathbb{R}$ are defined as follows:
				\begin{align}
					R_e(t)(h^{t+1},0)&=0,\\
					R_e(t)(h^{t+1},1)&=\delta_{\vec{x}_{t+1},\vec{1}},
				\end{align}
				for all $h^{t+1}\in\Omega_e(t+1)$.
			
			\item A $T$-step policy for the agent is a sequence of the form $\pi_T^e=(d_1^e,d_2^e,\dotsc,d_T^e)$, where the decision functions $d_t^e:\Omega_e(t)\times\mathcal{A}_e\to[0,1]$ are defined to be
				\begin{equation}
					d_t^e(h^t)(a_t)\coloneqq\Pr[A_e(t)=a_t|H(t)=h^t]
				\end{equation}
				for all $1\leq t\leq T$, all histories $h^t\in\Omega_e(t)$, and all $a_t\in\mathcal{A}_e$.~\defqed
		\end{itemize}
	\end{definition}
	\smallskip
	\begin{remark}
		A couple of remarks about Definition~\ref{def-QDP_elem_link_distill} are in order.
		\begin{itemize}
			\item First, let us count the number of actions in the set $\mathcal{A}_e$. When no distillation is performed, there are $2^{ c(e)}$ actions, corresponding to waiting or requesting for every individual parallel elementary link individually. When distillation is performed, the agent can choose the number $2\leq j\leq c(e)$ of parallel elementary links to be distilled, and the for remaining $ c(e)-j$ parallel links it chooses an element of $\{0,1\}^{ c(e)-j}$. Therefore, the total number of actions is
				\begin{equation}
					|\mathcal{A}_e|=\underbrace{2^{ c(e)}}_{\substack{\text{no}\\\text{distillation}}}+\sum_{j=2}^{ c(e)}\underbrace{\binom{ c(e)}{j}2^{ c(e)-j}}_{\substack{\text{distillation of}\\2\leq j\leq c(e)\\\text{parallel}\\\text{elementary}\\\text{links}}}(j-1) = 2^{ c(e)+1}+( c(e)-3)3^{ c(e)-1}.
				\end{equation}
				
			\item The reward scheme given in Definition~\ref{def-QDP_elem_link_distill} provides a non-zero reward at time $t$ if and only if all of the parallel elementary links at time $t+1$ are active. One simple alternative to this is to define the reward scheme by looking at the statuses of particular subsets of parallel elementary links. To do this, let $\psi_{e^j}^{j;\text{target}}$ be a $k$-partite target state for the $j^{\text{th}}$ parallel elementary link, for $1\leq j\leq c(e)$. Also, let
				\begin{equation}\label{eq-network_QDP_w_distill_reward_alt1}
					\Lambda_{e^j}^{1}\coloneqq\psi_{e^j}^{j;\text{target}},\quad \Lambda_{e^j}^{0}\coloneqq\mathbbm{1}_{e^j}-\psi_{e^j}^{j;\text{target}},
				\end{equation}
				and for $\vec{s}\coloneqq(s^1,s^2,\dotsc,s^{ c(e)})\in\{0,1\}^{ c(e)}$, let
				\begin{equation}\label{eq-network_QDP_w_distill_reward_alt2}
					\Lambda_e^{\vec{s}}\coloneqq\Lambda_{e^1}^{s^1}\otimes\Lambda_{e^2}^{s^2}\otimes\dotsb\otimes\Lambda_{e^{ c(e)}}^{s^{ c(e)}}.
				\end{equation}
				Then, we define the reward maps as
				\begin{equation}\label{eq-network_QDP_w_distill_reward_alt3}
					\mathcal{R}_e^{t;h^{t+1},\vec{s}_t}(\cdot)\coloneqq\Lambda_e^{\vec{s}_t}(\cdot)\Lambda_e^{\vec{s}_t}
				\end{equation}
				for all $t\geq 1$, all $h^{t+1}\in\Omega_e(t+1)$, and all $\vec{s}_t\in\{0,1\}^{ c(e)}$. We define the associated functions $R_e(t):\Omega_e(t+1)\times\{0,1\}^{ c(e)}\to\mathbb{R}$ as follows:
				\begin{equation}\label{eq-network_QDP_w_distill_reward_alt4}
					R_e(t)(h^{t+1},\vec{s}_t)\coloneqq \frac{1}{ c(e)}\sum_{j=1}^{ c(e)} \delta_{s_t^j,1}\delta_{x_{t+1}^j,1}.
				\end{equation}			
				The reward is thus the average of the rewards for the individual parallel elementary links.~\defqed
		\end{itemize}
	\end{remark}
	
	In order to make Definition~\ref{def-QDP_elem_link_distill} more concrete, particularly in terms of the transition maps, let us look at examples with two and three parallel elementary links.
	
	\begin{example}
		In the case of two parallel elementary links, the elements in Definition~\ref{def-QDP_elem_link_distill} are given as follows.
		\begin{align}
			\mathcal{X}_e&=\{(0,0),(0,1),(1,0),(1,1)\},\\
			\mathcal{A}_e&=\{(0;(0,0)),(0;(0,1)),(0;(1,0)),(0;(1,1)),(1;2,1,\{1,2\},\varnothing)\}.
		\end{align}
		We see that with two parallel elementary links there are there are five possible actions. Then, for all $\vec{x}_t,\vec{x}_{t+1}\in\mathcal{X}_e$, the transition maps in \eqref{eq-network_QDP_w_distill_trans1} corresponding to no distillation are:
		\begin{align}
			\mathcal{T}_{e;e^1,e^2}^{\vec{x}_t,(0;(0,0)),\vec{x}_{t+1}}&=\mathcal{T}_{e^1}^{x_t^1,0,x_{t+1}^1}\otimes\mathcal{T}_{e^2}^{x_t^2,0,x_{t+1}^2},\\
			\mathcal{T}_{e;e^1,e^2}^{\vec{x}_t,(0;(0,1)),\vec{x}_{t+1}}&=\mathcal{T}_{e^1}^{x_t^1,0,x_{t+1}^1}\otimes\mathcal{T}_{e^2}^{x_t^2,1,x_{t+1}^2},\\
			\mathcal{T}_{e;e^1,e^2}^{\vec{x}_t,(0;(1,0)),\vec{x}_{t+1}}&=\mathcal{T}_{e^1}^{x_t^1,1,x_{t+1}^1}\otimes\mathcal{T}_{e^2}^{x_t^2,0,x_{t+1}^2},\\
			\mathcal{T}_{e;e^1,e^2}^{\vec{x}_t,(0;(1,1)),\vec{x}_{t+1}}&=\mathcal{T}_{e^1}^{x_t^1,1,x_{t+1}^1}\otimes\mathcal{T}_{e^2}^{x_t^2,1,x_{t+1}^2}.
		\end{align}
		For the action corresponding to distillation of the two parallel elementary links, the transition maps in \eqref{eq-network_QDP_w_distill_trans2} are:
		\begin{equation}
			\mathcal{T}_{e;e^1,e^2}^{\vec{x}_t,(1;2,1,\{1,2\},\varnothing),\vec{x}_{t+1}}=\mathcal{T}_{e^1}^{x_t^1,0,x_{t+1}^1}\otimes\mathcal{T}_{e^2}^{x_t^2,0,x_{t+1}^2}
		\end{equation}
		for all $\vec{x}_t\neq(1,1)$ and all $\vec{x}_{t+1}\in\mathcal{X}_e$. When $\vec{x}_t=(1,1)$, we have
		\begin{align}
			\mathcal{T}_{e;e^1,e^2}^{(1,1),(1;2,1,\{1,2\},\varnothing),(0,0)}&=\mathcal{D}_{e^1e^2\to e^1}^{e;0}\otimes \widetilde{\mathcal{S}}_{e^2}^{0}(\rho_{e^2}^S),\\
			\mathcal{T}_{e;e^1,e^2}^{(1,1),(1;2,1,\{1,2\},\varnothing),(0,1)}&=\mathcal{D}_{e^1e^2\to e^1}^{e;0}\otimes\widetilde{\mathcal{S}}_{e^2}^{1}(\rho_{e^2}^S),\\
			\mathcal{T}_{e;e^1,e^2}^{(1,1),(1;2,1,\{1,2\},\varnothing),(1,0)}&=\mathcal{D}_{e^1e^2\to e^1}^{e;1}\otimes\widetilde{\mathcal{S}}_{e^2}^{0}(\rho_{e^2}^S),\\
			\mathcal{T}_{e;e^1,e^2}^{(1,1),(1;2,1,\{1,2\},\varnothing),(1,1)}&=\mathcal{D}_{e^1e^2\to e^1}^{e;1}\otimes\widetilde{\mathcal{S}}_{e^2}^{1}(\rho_{e^2}^S).~\defqedspec
		\end{align}
	\end{example}
	
	\begin{example}
		In the case of three parallel elementary links, the elements in Definition~\ref{def-QDP_elem_link_distill} are given as follows.
		\begin{align}
			\mathcal{X}_e&=\{0,1\}^3,\\
			\mathcal{A}_e&=\{(0;\vec{a}):\vec{a}\in\{0,1\}^3\}\cup\{(1;2,1,\{1,2\},0),(1;2,1,\{1,2\},1),\\
			&\qquad\qquad\qquad\qquad\qquad\,\, (1;2,1,\{1,3\},0),(1;2,1,\{1,3\},1),\\
			&\qquad\qquad\qquad\qquad\qquad\,\, (1;2,1,\{2,3\},0),(1;2,1,\{2,3\},1),\\
			&\qquad\qquad\qquad\qquad\qquad\,\, (1;3,1,\{1,2,3\},\varnothing),(1;3,2,\{1,2,3\},\varnothing)\}.
		\end{align}
		Then, the transition maps for the distillation actions are
		\begin{align}
			\mathcal{T}_{e;e^1,e^2,e^3}^{\vec{x}_t,(1;2,1,\{1,2\},a),\vec{x}_{t+1}}&=\mathcal{T}_{e^1}^{x_t^1,0,x_{t+1}^1}\otimes\mathcal{T}_{e^2}^{x_t^2,0,x_{t+1}^2}\otimes\mathcal{T}_{e^3}^{x_t^3,a,x_{t+1}^3}\quad\forall~\vec{x}_t\neq(1,1,x_t^3),\,\,x_t^3,a\in\{0,1\},\\[0.5cm]
			\mathcal{T}_{e;e^1,e^2,e^3}^{\vec{x}_t,(1;2,1,\{1,3\},a),\vec{x}_{t+1}}&=\mathcal{T}_{e^1}^{x_t^1,0,x_{t+1}^1}\otimes\mathcal{T}_{e^2}^{x_t^2,a,x_{t+1}^2}\otimes\mathcal{T}_{e^3}^{x_t^3,0,x_{t+1}^3}\quad\forall~\vec{x}_t\neq(1,x_t^2,1),\,\,x_t^2,a\in\{0,1\},\\[0.5cm]
			\mathcal{T}_{e;e^1,e^2,e^3}^{\vec{x}_t,(1;2,1,\{2,3\},a),\vec{x}_{t+1}}&=\mathcal{T}_{e^1}^{x_t^1,a,x_{t+1}^1}\otimes\mathcal{T}_{e^2}^{x_t^2,0,x_{t+1}^2}\otimes\mathcal{T}_{e^3}^{x_t^3,0,x_{t+1}^3}\quad\forall~\vec{x}_t\neq(x_t^1,1,1),\,\,x_t^1,a\in\{0,1\},\\[0.5cm]
			\mathcal{T}_{e;e^1,e^2,e^3}^{\vec{x}_t,(1;3,1,\{1,2,3\},a),\vec{x}_{t+1}}&=\mathcal{T}_{e^1}^{x_t^1,0,x_{t+1}^1}\otimes\mathcal{T}_{e^2}^{x_t^2,0,x_{t+1}^2}\otimes\mathcal{T}_{e^3}^{x_t^3,0,x_{t+1}^3}\nonumber\\
			&=\mathcal{T}_{e;e^1,e^2,e^3}^{\vec{x}_t,(1;3,2,\{1,2,3\},a),\vec{x}_{t+1}}\quad\forall~\vec{x}_t\neq(1,1,1),
		\end{align}
		and
		\begin{align}
			\mathcal{T}_{e;e^1,e^2,e^3}^{(1,1,x_t^3),(1;2,1,\{1,2\},a),\vec{x}_{t+1}}&=\mathcal{D}_{e^1e^2\to e^1}^{e;x_{t+1}^1}\otimes\widetilde{\mathcal{S}}_{e^2}^{x_{t+1}^2}(\rho_{e^2}^S)\otimes\mathcal{T}_{e^3}^{x_t^3,a,x_{t+1}^3}\quad\forall~\vec{x}_{t+1}\in\{0,1\}^3,\,\,x_t^3,a\in\{0,1\},\\[0.5cm]
			\mathcal{T}_{e;e^1,e^2,e^3}^{(1,x_t^2,1),(1;2,1,\{1,3\},a),\vec{x}_{t+1}}&=\mathcal{D}_{e^1e^3\to e^1}^{e;x_{t+1}^1}\otimes\mathcal{T}_{e^2}^{x_t^2,a,x_{t+1}^2}\otimes\widetilde{\mathcal{S}}_{e^3}^{x_{t+1}^3}(\rho_{e_3}^S)\quad\forall~\vec{x}_{t+1}\in\{0,1\}^3,\,\,x_t^2,a\in\{0,1\},\\[0.5cm]
			\mathcal{T}_{e;e^1,e^2,e^3}^{(x_t^1,1,1),(1;2,1,\{2,3\},a),\vec{x}_{t+1}}&=\mathcal{T}_{e^1}^{x_t^1,a,x_{t+1}^1}\otimes\mathcal{D}_{e^2e^3\to e^2}^{e;x_{t+1}^2}\otimes\widetilde{\mathcal{S}}_{e^3}^{x_{t+1}^3}(\rho_{e^3}^S)\quad\forall~\vec{x}_{t+1}\in\{0,1\}^3,\,\,x_t^1,a\in\{0,1\},\\[0.5cm]
			\mathcal{T}_{e;e^1,e^2,e^3}^{(1,1,1),(1;3,1,\{1,2,3\},\varnothing),\vec{x}_{t+1}}&=\mathcal{D}_{e^1e^2e^3\to e^1}^{e;x_{t+1}^1}\otimes\widetilde{\mathcal{S}}_{e^2}^{x_{t+1}^2}(\rho_{e^2}^S)\otimes\widetilde{\mathcal{S}}_{e^3}^{x_{t+1}^3}(\rho_{e^3}^S)\quad\forall~\vec{x}_{t+1}\in\{0,1\}^3,\\[0.5cm]
			\mathcal{T}_{e;e^1,e^2,e^3}^{(1,1,1),(1;3,2,\{1,2,3\},\varnothing),\vec{x}_{t+1}}&=\left\{\begin{array}{l l}\displaystyle \mathcal{D}^{e;1}_{e^1e^2e^3\to e^1e^2}\otimes\widetilde{\mathcal{S}}_{e^3}^{x_{t+1}^3}(\rho_{e^3}^S) & \text{if } x_{t+1}^1=x_{t+1}^2=1, \\[0.35cm] \displaystyle \mathcal{D}_{e^1e^2e^3\to e^1e^2}^{e;0}\otimes\widetilde{\mathcal{S}}_{e^3}^{x_{t+1}^3}(\rho_{e^3}^S) & \text{if } x_{t+1}^1=x_{t+1}^2=0. ~\defqedspec \end{array}\right.
		\end{align}
	\end{example}
	
	The expressions in \eqref{eq-network_QDP_cq_state}--\eqref{eq-network_QDP_history_prob} for the quantum decision process without entanglement distillation have analogous forms here when the quantum decision process incorporates entanglement distillation protocols. Specifically, for an arbitrary policy $\pi$
	\begin{equation}
		\widehat{\sigma}_{H_t^eE_t^e}(t)=\sum_{h^t\in\Omega_e(t)}\ket{h^t}\bra{h^t}_{H_t^e}\otimes\widetilde{\sigma}_{E_t^e}^{\pi}(t;h^t),
	\end{equation}
	and for an arbitrary history $h^t=(\vec{x}_1,a_1,\vec{x}_2,a_2,\dotsc,a_{t-1},\vec{x}_t)\in\Omega_e(t)$,
	\begin{align}
		\widetilde{\sigma}_{E_t^e}^{\pi}(t;h^t)&=\left(\prod_{j=1}^{t-1}d_j(h_j^t)(a_j)\right)\left(\mathcal{T}_{E_{t-1}^e\to E_t^e}^{\vec{x}_{t-1},a_{t-1},\vec{x}_1}\circ\dotsb\circ\mathcal{T}_{E_2^e\to E_3^e}^{\vec{x}_2,a_2,\vec{x}_3}\circ\mathcal{T}_{E_1^e\to E_2^e}^{\vec{x}_1,a_1,\vec{x}_2}\circ\mathcal{T}_{E_0^e\to E_1^e}^{0;\vec{x}_1}\right)(\rho_{E_0^e}^S),\\
		\sigma_{E_t^e}^{\pi}(t|h^t)&=\frac{\widetilde{\sigma}_{E_t^e}^{\pi}(t;h^t)}{\Tr[\widetilde{\sigma}_{E_t^e}(t;h^t)]},\\
		\Pr[H_e(t)=h^t]&=\Tr[\widetilde{\sigma}_{E_t^e}^{\pi}(t;h^t)].
	\end{align}

	Also, following the steps of the proof of Theorem~\ref{thm-network_QDP_exp_reward_no_distill}, we find the expected reward of the quantum decision process defined in Definition~\ref{def-QDP_elem_link_distill} is
	\begin{equation}
		\mathbb{E}[R_e(t)]_{\pi}=\Tr\left[\left(\ket{\vec{1}}\bra{\vec{1}}_{X_{t+1}^e}\otimes\psi_e^{\text{target}}\right)\widehat{\sigma}_e^{\pi}(t+1)\right]
	\end{equation}
	for all $e\in E$, $t\geq 1$, and all policies $\pi$, where $\psi_e^{\text{target}}$ is a $k c(e)$-partite target state for all of the parallel elementary links corresponding to $e$.
	
	For the alternative reward scheme defined in \eqref{eq-network_QDP_w_distill_reward_alt1}--\eqref{eq-network_QDP_w_distill_reward_alt4}, we have the following result.
	
	\begin{proposition}
		Let $G=(V,E,c)$ be the graph corresponding to the physical (elementary) links of a quantum network, and let $e\in E$ be arbitrary. Consider the reward scheme defined in \eqref{eq-network_QDP_w_distill_reward_alt1}--\eqref{eq-network_QDP_w_distill_reward_alt4} for the quantum decision process for $e$ as defined in Definition~\ref{def-QDP_elem_link_distill}. Then, for all $t\geq 1$ and for all policies $\pi$,
		\begin{equation}
			\mathbb{E}[R_e(t)]_{\pi}=\frac{1}{ c(e)}\sum_{j=1}^{ c(e)}\Tr\left[\left(\ket{1}\bra{1}_{X_{t+1}^{e^j}}\otimes\psi_{e^j}^{j;\text{target}}\right)\widehat{\sigma}_e^{\pi}(t+1)\right].
		\end{equation}
	\end{proposition}
	
	\begin{proof}
		Using \eqref{eq-network_QDP_w_distill_reward_alt1}--\eqref{eq-network_QDP_w_distill_reward_alt4}, we have
		\begin{align}
			\widetilde{\mathcal{R}}_e^{t;h^{t+1}}(\cdot)&\coloneqq\sum_{\vec{s}_t\in\{0,1\}^{ c(e)}}R_e(t)(h^{t+1},\vec{s}_t)\mathcal{R}_e^{t;h^{t+1},\vec{s}_t}(\cdot)\\
			&=\frac{1}{ c(e)}\sum_{\vec{s}_t\in\{0,1\}^{ c(e)}}\sum_{j=1}^{ c(e)}\delta_{s_t^j,1}\delta_{x_{t+1}^j,1}\Lambda_e^{\vec{s}^t}(\cdot)\Lambda_{e}^{\vec{s}_t}\\
			&=\frac{1}{ c(e)}\sum_{j=1}^{ c(e)}\sum_{s_t^1,\dotsc,s_t^{ c(e)}=0}^1\delta_{s_t^j,1}\delta_{x_{t+1}^j,1}\left(\Lambda_{e^1}^{s_t^1}\otimes\dotsb\otimes\Lambda_{e^{ c(e)}}^{s_t^{ c(e)}}\right)(\cdot)\left(\Lambda_{e^1}^{s_t^1}\otimes\dotsb\otimes\Lambda_{e_{ c(e)}}^{s_t^e}\right).
		\end{align}
		Then, using the general expression in \eqref{eq-QDP_exp_reward_formula_1}, we obtain
		\begin{align}
			\mathbb{E}[R_e(t)]_{\pi}&=\sum_{h^{t+1}\in\Omega_e(t+1)}\Tr\left[\widetilde{\mathcal{R}}_e^{t;h^{t+1}}\left(\widetilde{\sigma}_e(t+1;h^{t+1})\right)\right]\\
			&=\sum_{h^{t+1}\in\Omega_e(t+1)}\frac{1}{ c(e)}\sum_{j=1}^{ c(e)}\sum_{s_t^1,\dotsc,s_t^{ c(e)}=0}^1\delta_{s_t^j,1}\delta_{x_{t+1}^j,1}\Tr\left[\left(\Lambda_{e^1}^{s_t^1}\otimes\dotsb\otimes\Lambda_{e^{ c(e)}}^{s_t^{ c(e)}}\right)\widetilde{\sigma}_e(t+1;h^{t+1})\right.\nonumber\\
			&\qquad\qquad\qquad\qquad\qquad\qquad\qquad\qquad\qquad\qquad\qquad\left.\left(\Lambda_{e^1}^{s_t^1}\otimes\dotsb\otimes\Lambda_{e^{ c(e)}}^{s_t^{ c(e)}}\right)\right]\\
			&=\sum_{h^{t+1}\in\Omega_e(t+1)}\frac{1}{ c(e)}\sum_{j=1}^{ c(e)}\sum_{s_t^1,\dotsc,s_t^{ c(e)}=0}^1\delta_{s_t^j,1}\delta_{x_{t+1}^j,1}\Tr\left[\left(\Lambda_{e^1}^{s_t^1}\otimes\dotsb\otimes\Lambda_{e^{ c(e)}}^{s_t^{ c(e)}}\right)\widetilde{\sigma}_e(t+1;h^{t+1})\right]
		\end{align}
		where the last line holds because of the cyclicity of the trace and because every $\Lambda_{e^j}^{s_t^j}$ is a projection, meaning that $\Lambda_{e^j}^{s_t^j}\Lambda_{e_j}^{s_t^j}=\Lambda_{e^j}^{s_j^j}$. Then, because
		\begin{equation}
			\sum_{s_t^j=0}^1\Lambda_{e^j}^{s_t^j}=\mathbbm{1}_{e^j},\quad\Lambda_{e^j}^1=\psi_{e^j}^{j;\text{target}},
		\end{equation}
		we obtain
		\begin{equation}
			\mathbb{E}[R_e(t)]_{\pi}=\frac{1}{ c(e)}\sum_{j=1}^{ c(e)}\sum_{h^{t+1}\in\Omega_e(t+1)}\delta_{x_{t+1}^j,1}\Tr\left[\psi_{e^j}^{j;\text{target}}\widetilde{\sigma}_e(t+1;h^{t+1})\right].
		\end{equation}
		Finally, because
		\begin{equation}
			\sum_{h^{t+1}\in\Omega_e(t+1)}\delta_{x_{t+1}^j,1}\Tr\left[\psi_{e^j}^{j;\text{target}}\widetilde{\sigma}_e(t+1;h^{t+1})\right]=\Tr\left[\left(\ket{1}\bra{1}_{X_{t+1}^{e^j}}\otimes\psi_{e^j}^{j;\text{target}}\right)\widehat{\sigma}_e(t+1)\right]
		\end{equation}
		for all $1\leq j\leq c(e)$, we obtain the desired result.
	\end{proof}

	\begin{remark}
		As we can see, by including entanglement distillation in the quantum decision process, analytical expressions and derivations quickly become cumbersome and unmanageable. However, it is straightforward to deal with the quantum decision process computationally by explicitly programming the behavior of the transition maps and then applying reinforcement learning algorithms to discover optimal policies.
	\end{remark}

\section{QDP with entanglement swapping}

	In addition to entanglement distillation, we can define a quantum decision process that incorporates entanglement swapping. The following is a definition of a simple such quantum decision process, in which we consider only one parallel elementary link per element $e\in E$ of the physical graph $G=(V,E,c)$.

	\begin{definition}[QDP for two elementary links with entanglement swapping]\label{def-QDP_ent_swap}
		Let $G=(V,E,c)$ be the multigraph corresponding to the physical (elementary) links of a quantum network, and let $e_1,e_2\in E$ be such that $c(e_1)=c(e_2)=1$, and both $e_1$ and $e_2$ correspond to bipartite elementary links. We define a quantum decision process for $e_1,e_2$ by defining the agent to be collectively the nodes belonging to $e_1$ and $e_2$, and we define its environment to be the four quantum systems (two systems for $e_1$ and two for $e_2$) that are distributed to the nodes by source stations. Then, the other elements of the quantum decision process are defined as follows.
		
		\begin{itemize}
			\item We denote the quantum systems corresponding to $e_1$ by $E_t^{e_1}$ for all times $t\geq 0$, the quantum systems corresponding to $e_2$ by $E_t^{e_2}$ for all times $t\geq 0$, and we let $E_t^{e_1,e_2}\equiv E_t^{e_1}E_t^{e_2}$ denote both collections of quantum systems at time $t\geq 0$. The state of the environment at time $t=0$ is $\rho_{E_0^{e_1,e_2}}^0=\rho_{E_0^{e_1}}^S\otimes\rho_{E_0^{e_2}}^S$, where $\rho_{E_0^{e_j}}^S$ is the source state for the $j^{\text{th}}$ elementary link, $j\in\{1,2\}$.
			
			\item We let $\mathcal{X}_{e_1,e_2}=\{0,1\}^3$ tell us whether or not the two elementary links corresponding to $e_1$ and $e_2$ are active. In other words, every string $\vec{x}=(x^1,x^2,x^3)\in\mathcal{X}_{e_1,e_2}$ is such that $x^j=0$ indicates that the $j^{\text{th}}$ elementary link is inactive and $x^j=1$ indicates that the $j^{\text{th}}$ elementary link is active, with $j\in\{1,2\}$. For $j=3$, $x^3=0$ indicates that the virtual link arising from joining the elementary links corresponding to $e_1$ and $e_2$ is inactive, while $x^3=1$ indicates that the virtual link is active. We then define random variables $X_{e_j}(t)$, $j\in\{1,2\}$ and $t\geq 1$, taking values in $\{0,1\}$ as follows:
				\begin{itemize}
					\item $X_{e_j}(t)=0$: the $j^{\text{th}}$ elementary link is inactive;
					\item $X_{e_j}(t)=1$: the $j^{\text{th}}$ elementary link is active.
				\end{itemize}
				Letting $e'$ denote the edge corresponding to the virtual link obtained by joining the elementary links corresponding to $e_1$ and $e_2$, we define $X_{e'}(t)$, $t\geq 1$, such that
				\begin{itemize}
					\item $X_{e'}(t)=0$: the virtual link $e'$ is inactive;
					\item $X_{e'}(t)=1$: the virtual link $e'$ is active.
				\end{itemize}

				We let 
				\begin{equation}
					\mathcal{A}_{e_1,e_2}=\{0,1\}^2\cup\{\Join\},
				\end{equation}
				where the symbol ``$\Join$'' indicates that the joining operation corresponding to the quantum channel $\mathcal{L}_{v_1e_1v_2e_2v_3\to e'}$ as defined in \eqref{eq-q_instr_channel_swapping} should be performed, where $e'$ is the new edge arising from the joining operation. We then define random variables $A_{e_1,e_2}(t)$, $t\geq 1$, that take values in $\mathcal{A}_{e_1,e_2}$.
				
				A history $H_{e_1,e_2}(t)$ is of the form
				\begin{equation}
					H_{e_1,e_2}(t)=(\vec{X}_{e_1,e_2}(1),A_{e_1,e_2}(1),\vec{X}_{e_1,e_2}(2),A_{e_1,e_2}(2),\dotsc A_{e_1,e_2}(t-1),\vec{X}_{e_1,e_2}(t)),
				\end{equation}
				where $\vec{X}_{e_1,e_2}(t)=(X_{e_1}(t),X_{e_2}(t),X_{e'}(t))$. The set of all histories is $\Omega_{e_1,e_2}(t)$, and every $h^t\in\Omega_{e_1,e_2}(t)$ is of the form $h^t=(\vec{x}_1,a_1,\vec{x}_2,a_2,\dotsc,a_{t-1},\vec{x}_t)$ for all $t\geq 1$, with $h^1=\vec{x}_1$.
			
			\item For $t=0$, the transition maps are defined as follows:
				\begin{equation}
					\mathcal{T}_{e_1,e_2}^{0;\vec{x}_1}\coloneqq\left(\mathcal{M}_{e_1}^{x_t^1}\circ\mathcal{S}_{e_1}\right)\otimes\left(\mathcal{M}_{e_2}^{x_1^2}\circ\mathcal{S}_{e_2}\right)
				\end{equation}
				for all $\vec{x}_1=(x_1^1,x_1^2,x_1^3)\in\{0,1\}^3$, where $\mathcal{S}_{e_j}$ is the source transmission channel for the $j^{\text{th}}$ elementary link and $\{\mathcal{M}_{e_j}^0,\mathcal{M}_{e_j}^1\}$ is the heralding quantum instrument for the $j^{\text{th}}$ elementary link, $j\in\{1,2\}$.
			
				For $t\geq 1$, we denote the transition maps by $\mathcal{T}_{e_1,e_2}^{\vec{x}_t,a_t,\vec{x}_{t+1}}$, with $\vec{x}_t,\vec{x}_{t+1}\in\mathcal{X}_{e_1,e_2}$ and $a_t\in\mathcal{A}_{e_1,e_2}$, and we define them as follows when $a_t\neq\Join$ and $a_t=(a_t^1,a_t^2)\in\{0,1\}^2$:
				\begin{align}
					\mathcal{T}_{e_1,e_2}^{\vec{x}_t,a_t,\vec{x}_{t+1}}=\left\{\begin{array}{l l}\mathcal{T}_{e_1}^{x_t^1,a_t^1,x_{t+1}^1}\otimes\mathcal{T}_{e_2}^{x_t^2,a_t^2,x_{t+1}^2}\otimes\mathcal{T}_{e'}^{x_t^3,0,x_{t+1}^3} & \text{if }x_{t}^3=x_{t+1}^3=1,\\[0.5cm]
					\mathcal{T}_{e_1}^{x_t^1,a_t^1,x_{t+1}^1}\otimes\mathcal{T}_{e_2}^{x_t^2,a_t^2,x_{t+1}^2} & \text{otherwise}, \end{array}\right. 
				\end{align}
				where the maps $\mathcal{T}_{e^i}^{x_t^i,a_t^i,x_{t+1}^i}$ are defined in Definition~\ref{def-network_QDP_elem_link}. In other words, if the action is not to join the elementary links, and if the virtual link is active at the $t^{\text{th}}$ time step, then the individual links evolve independently, exactly as they do in the quantum decision process for elementary links in Definition~\ref{def-network_QDP_elem_link}. If the action is to join the elementary links, then
				\begin{align}
					\mathcal{T}_{e_1,e_2}^{\vec{x}_t,\Join,\vec{x}_{t+1}}&=\left\{\begin{array}{l l} \mathcal{T}_{e_1}^{x_t^1,0,x_{t+1}^1}\otimes\mathcal{T}_{e_2}^{x_t^2,0,x_{t+1}^2}\otimes\mathcal{T}_{e'}^{x_t^3,0,x_{t+1}^3} & \text{if }x_t^1x_t^2=0\text{ and }x_t^3=x_{t+1}^3=1,\\[0.5cm] \mathcal{T}_{e_1}^{x_t^1,0,x_{t+1}^1}\otimes\mathcal{T}_{e_2}^{x_t^2,0,x_{t+1}^2} & \text{if }x_t^1x_t^2=0 \text{ and }x_t^3=x_{t+1}^3=0,\\[0.5cm]
					\widetilde{\mathcal{S}}^{x_{t+1}^1}(\rho_{e_1}^S)\otimes\widetilde{\mathcal{S}}^{x_{t+1}^2}(\rho_{e_2}^S)\otimes\mathcal{L}_{v_1e_1v_2e_2v_3\to e'}^{x_{t+1}^3} & \text{if }x_t^1=x_t^2=1, \end{array}\right.
				\end{align}
				where $\widetilde{\mathcal{S}}_{e_j}^x=\mathcal{M}_{e_j}^x\circ\mathcal{S}_{e_j}$, $j\in\{1,2\}$.
				
			\item Given a target bipartite state $\psi_{e'}^{\text{target}}=\ket{\psi^{\text{target}}}\bra{\psi^{\text{target}}}_{e'}$ for the virtual link $e'$, the reward at time $t\geq 1$ is defined as follows:
				\begin{align}
					\mathcal{R}_{e_1,e_2}^{t;h^{t+1},1}(\cdot)&=\left\{\begin{array}{l l} \id_{e_1}\otimes\id_{e_2} & \text{if } x_{t+1}^3=0, \\[0.5cm] \psi_{e'}^{\text{target}}(\cdot)\psi_{e'}^{\text{target}} & \text{if }x_{t+1}^3=1,\end{array}\right.\\[1cm]
					\mathcal{R}_{e_1,e_2}^{t;h^{t+1},0}(\cdot)&=\left\{\begin{array}{l l} 0 & \text{if }x_{t+1}^3=0, \\[0.5cm] (\mathbbm{1}_{e'}-\psi_{e'}^{\text{target}})(\cdot)(\mathbbm{1}_{e'}-\psi_{e'}^{\text{target}}) & \text{if }x_{t+1}^3=1, \end{array}\right.
				\end{align}
				for all $h^{t+1}\in\Omega_{e_1,e_2}(t+1)$, and the functions $R_{e_1,e_2}(t):\Omega_{e_1,e_2}(t)\times\{0,1\}\to\mathbb{R}$ are defined as follows:
				\begin{align}
					R_{e_1,e_2}(t)(h^{t+1},0)&=0,\\
					R_{e_1,e_2}(t)(h^{t+1},1)&=\delta_{x_{t+1}^3,1},
				\end{align}
				for all $h^{t+1}\in\Omega_{e_1,e_2}(t+1)$. In other words, the reward is based on the fidelity of the quantum state of the virtual link to the target state $\psi_{e'}^{\text{target}}$. Note that if the virtual link is not active at time $t+1$, then the reward is zero and the quantum systems of the elementary links evolve according to the identity channel.
			
			\item  A $T$-step policy for the agent is a sequence of the form $\pi_T^{e_1,e_2}=(d_1^{e_1,e_2},d_2^{e_1,e_2},\dotsc,d_T^{e_1,e_2})$, where the decision functions $d_t^{e_1,e_2}:\Omega_{e_1,e_2}\times\mathcal{A}_{e_1,e_2}\to[0,1]$ are defined to be
				\begin{equation}
					d_t^{e_1,e_2}(h^t)(a_t)=\Pr[A_{e_1,e_2}(t)=a_t|H(t)=h^t]
				\end{equation}
				for all $1\leq t\leq T$, all histories $h^t\in\Omega_{e_1,e_2}(t)$, and all $a_t\in\mathcal{A}_{e_1,e_2}$.~\defqed
		\end{itemize}
	\end{definition}
	
	As with the quantum decision process for entanglement distillation in Definition~\ref{def-QDP_elem_link_distill}, there are other ways of defining the reward scheme. We leave this investigation, as well as investigations of the properties of this quantum decision process, to future work. Also, as with the quantum decision process defined in Definition~\ref{def-QDP_elem_link_distill}, analytical expressions and derivations can be cumbersome and unmanageable for the quantum decision process defined in Definition~\ref{def-QDP_ent_swap}; however, it is straightforward to deal with the quantum decision process computationally by explicitly programming the behavior of the transition maps and then applying reinforcement learning algorithms to discover optimal policies.

\section{Cooperating agents}

	The quantum decision process with entanglement swapping in Definition~\ref{def-QDP_ent_swap} is a simple example of a quantum decision process with multiple cooperating agents. When we say that agents ``cooperate'', we mean that the agents are allowed to communicate with each other. When the agents are purely classical, this means that they can communicate arbitrary classical information to each other, which then means that the agents have access to each other's histories. Consequently, multiple cooperating agents can be combined into a single ``superagent'', which is essentially what was done in Definition~\ref{def-QDP_ent_swap}. In the context of quantum networks, agents who cooperate have more than the knowledge of their own nodes. If every agent cooperates with an agent corresponding to a neighboring elementary link, then we can extend the definitions of the quantum decision processes in Definitions~\ref{def-network_QDP_elem_link}, \ref{def-QDP_elem_link_distill}, and \ref{def-QDP_ent_swap} to take multiple agents into account, and with these definitions we can define an entanglement distribution protocol that is more sophisticated than the one we consider in Chapter~\ref{chap-network_QDP}, which is summarized in Figure~\ref{fig-QDP_protocol}. In particular, in such a protocol, the agents would have knowledge of the network in their local vicinity, and this would in principle improve waiting times and rates for entanglement distribution. Furthermore, the quantum state of the network would not be a simple tensor product of the quantum states corresponding to the individual agents, as we have in \eqref{eq-network_cq_state_QDP} when all the agents are independent. See Refs.~\cite{PKT+19,CRDW19} for a discussion of nodes with local and global knowledge of a quantum network in the context of routing.
	
	Let us now briefly discuss cooperating agents in quantum decision processes more generally, i.e., outside the quantum network context. Recall from Definition~\ref{def-QDP} that in general the agents of a quantum decision process can be quantum systems. It is possible to extend the definition to include multiple agents, which may or may not cooperate. Such an extension would constitute a quantum generalization of stochastic games \cite{Shap53,SV15,LS2017}, which themselves are multi-agent generalizations of (classical) Markov decision processes \cite{Tan93,HW98,BBS08,FV97_book}. When the agents are described quantum mechanically, there are different ways in which we can define the word ``cooperate''. For example, ``cooperate'' could mean that the agents start with some prior shared entanglement, and thereafter communicate only classically. Another possibility is that the agents are allowed to communicate quantum information to each other in addition to classical information. Exploring these possibilities is an interesting avenue for future work and is of interest even outside the context of quantum networks.

\section{Reinforcement learning algorithms for optimal policies}

	We have mentioned that the backward recursion algorithm for finding an optimal policy for finite-horizon quantum decision processes is exponentially slow in the horizon time; see Section~\ref{sec-QDP_pol_opt}. One common solution to this is to use reinforcement learning algorithms, which are often more efficient and can find optimal (generally sub-optimal) policies. Since the agents in the quantum decision processes defined in Definitions~\ref{def-network_QDP_elem_link}, \ref{def-QDP_elem_link_distill}, and \ref{def-QDP_ent_swap} are classical, it is possible to use standard reinforcement learning algorithms directly to find optimal policies; see Ref.~\cite{Sut18_book} for an introduction to reinforcement learning algorithms.

	Let us also point out that, in the backward recursion algorithm, we assume complete knowledge of the environment, in the sense that the transition maps and other elements of the environment are known to the agent. The problem of finding an optimal policy thus reduces essentially to a \textit{dynamic programming} problem, also sometimes called a \textit{planning} problem. However, in general, a reinforcement learning algorithm is based on the agent acting in the environment in real time, without necessarily having prior knowledge of the environment. As such, the agent must first learn about the environment (hence the ``learning'' in reinforcement learning), and reinforcement learning algorithms typically involve an exploration stage in which the agent attempts to first learn about the environment by employing different combinations of actions. The agent then uses its accumulated knowledge of the environment to develop an optimal policy.

\section{Quantum algorithms for decision processes and reinforcement learning}

	Instead of using standard (classical) reinforcement learning algorithms, as described in the previous section, another possibility that has gained attention recently is to develop a \textit{quantum algorithm}, i.e., an algorithm that would run on a quantum computer, for finding an optimal policy. In this section, we provide a brief review of work that has been done in this direction.
	
	An initial proposal for solving (classical) partially observable Markov decision processes with a quantum algorithm was made in Ref.~\cite{RMS04}. Then, in Ref.~\cite{DCLT08}, quantum algorithms analogous to the well-known temporal difference (TD) learning and Q-learning (classical) reinforcement learning algorithms were provided, which essentially make use of Grover's algorithm as a way of searching through the space of all possible actions of the agent in order to find the action sequence with the highest reward. Similar ideas have been used in subsequent work to develop quantum algorithms for reinforcement learning problems \cite{PDM+14,DTB16,DTB15,DLWT17,Corne18,WYLC20,HDW20} (see also Refs.~\cite{DB17,DTB17} for reviews). Finding an optimal policy for a given decision process can be thought of as special kind of dynamic programming problem. Quantum algorithms for solving such problems have been considered in Refs.~\cite{Ron19,SNSW20}.

%	[these algorithms can give lower bounds on the reward by explicitly constructing a policy...can use a (quantum) variational ansatz via controlled unitaries in the Markovian case, when the action depends only on the previous (classical) value of the environment...or classical variational ansatz via neural networks...]

%[take the SDP and combine it with existing quantum algorithms for solving SDPs; this gives one possible quantum algorithm?]

%	[...papers on algorithms for multi-agent reinforcement learning: \cite{Tan93,HW98,BBS08}...

\end{appendices}

{\pagestyle{plain}
%\printbibliography[title={BIBLIOGRAPHY},heading=bibintoc]
\bibliographystyle{unsrtnat}
\bibliography{refs}
\cleardoublepage\phantomsection}

\newpage 

\pagestyle{plain}
%{\noindent\large\textbf{Vita}}
\chapter*{VITA}
\addcontentsline{toc}{chapter}{VITA}

	Sumeet Khatri was born in May 1991. He attended high school in Toronto, Canada at Weston Collegiate Institute, where he graduated from the International Baccalaureate Programme in 2009. He then attended the University of Waterloo in Waterloo, Canada for undergraduate studies. In 2014, he earned his Bachelor of Science degree (with honors) in Mathematical Physics, with a specialization in Astrophysics and a minor in Pure Mathematics. In 2016, he earned his Master of Science degree in Physics, also at the University of Waterloo, under the supervision of Norbert L\"{u}tkenhaus. In January 2017, he started is Ph.D. in Physics at Louisiana State University under the supervision of Mark M. Wilde and Jonathan P. Dowling. He anticipates graduating in May 2021.

\end{document}